\documentclass[12pt,a4paper,twoside]{book}
\usepackage[]{epsfig,color}
\usepackage{amsmath,amssymb,amsfonts,amsfonts,amsthm}
\usepackage{latexsym}
\usepackage{wrapfig}
\usepackage{fancyhdr}
\usepackage{algorithm}
\usepackage{colordvi,alltt}
\usepackage{multirow}
\usepackage{tabularx}
\usepackage{fancybox}
\usepackage{calc}
\usepackage[avantgarde]{quotchap}    %
\usepackage{paralist} %
\usepackage{multibib}                %
\usepackage{float}                   %
\usepackage{makeidx}                 %
\usepackage{hyperref}
\usepackage{breakurl}
\usepackage{pstricks,pst-node,pst-text,pst-3d,pst-plot}
\usepackage{pst-xkey,pst-jtree}

\psset{unit=1cm,xunit=1cm,yunit=1cm}

\newcites{links}{Links}

\floatstyle{boxed}
\newfloat{Task}{b}{lta}[chapter]

\newboolean{FinalVersion} %
\setboolean{FinalVersion}{false}

\newboolean{ShowTODO}
\setboolean{ShowTODO}{true}
\newboolean{ShowSVN}
\setboolean{ShowSVN}{true}

\newcommand{\ignore}[1]{}

\ifthenelse{\boolean{FinalVersion}}{
  \setboolean{ShowTODO}{false}
  \setboolean{ShowSVN}{false}
}{}

\newcommand{\task}[2]{%
  \ifthenelse{\boolean{ShowTODO}}%
  {%
   \begin{Task}[H]%
     \caption{#1}%
     {#2}%
   \end{Task}%
  }%
  {%
  }%
}%

\newcommand{\todo}[1]{%
  \ifthenelse{\boolean{ShowTODO}}%
  {\begin{center}%
      \framebox{\parbox{0.8\textwidth}{#1}}%
   \end{center}}%
  {}%
}

\renewcommand{\chaptermark}[1]{\markboth{{\it \chaptername\ \thechapter.\ #1}}{}}

\newtheorem{theorem}{Theorem}[chapter]
\newtheorem{lemma}[theorem]{Lemma}

\newtheorem{observation}[theorem]{Observation}
\newtheorem{corollary}[theorem]{Corollary}

\newcommand{\umin}{u_{\rm min}}
\newcommand{\umax}{u_{\rm max}}
\newcommand{\vmin}{v_{\rm min}}
\newcommand{\vmax}{v_{\rm max}}

\newcommand{\calA}{\mathcal{A}}
\newcommand{\calC}{\mathcal{C}}

\newcommand{\calH}{\mathcal{H}}

\newcommand{\calS}{\mathcal{S}}

\newcommand{\parms}{\mathbb{P}}

\newcommand{\reals}{\mathbb{R}}
\newcommand{\dirs}{\mathbb{S}}

\newcommand{\rrr}{\reals^3}
\newcommand{\rr}{\reals^2}
\newcommand{\spheretwo}{\dirs^2}

\newcommand{\vecd}{\vec{d}}
\newcommand{\Cpp}{{C}{\tt ++}}
\newcommand{\sgal}{{\sc Sgal}}
\newcommand{\cgal}{{\sc Cgal}}
\newcommand{\leda}{{\sc Leda}}
\newcommand{\core}{{\sc Core}}

\newcommand{\gmp}{{\sc GMP}}

\newcommand{\dcel}{{\sc Dcel}}
\newcommand{\hds}{{\sc Hds}}
\newcommand{\lisp}{{\sc Lisp}}
\newcommand{\exacus}{{\sc Exacus}}
\newcommand{\boost}{{\sc Boost}}
\newcommand{\bgl}{{\sc BGL}}
\newcommand{\stl}{{\sc STL}}

\newcommand{\gf}{{\sc Geometry Factory}}
\newcommand{\cgm}{{\sc Cgm}}

\newcommand{\vrml}{{\sc VRML}}
\newcommand{\kdtree}{{\sc Kd}-tree}

\newcommand{\sgmo}{{\bf SGM}}
\newcommand{\cgmo}{{\bf CGM}}
\newcommand{\ngmo}{{\bf NGM}}
\newcommand{\ch}{{\bf CH}}
\newcommand{\Fuku}{{\bf Fuk}}

\newcommand{\bez}{B{\'e}zier}
\newcommand{\mobius}{M\"obius}

\newcommand{\etal}{{\it et~al.}}
\newcommand{\tildegen}{\protect\raisebox{-0.12cm}{\symbol{'176}}}

\newcommand{\mindia}[1]{\ensuremath{{\mathcal M}(#1)}}
\newcommand{\distance}[2]{\ensuremath{\rho(#1, #2)}}

\definecolor{orange}{rgb}{1,0.5,0}

\def\concept#1{\textsf{\it #1}}
\def\ccode#1{{\texttt{#1}}}
\newcommand{\arr}{\ccode{Arrangement\_2}}
\newcommand{\arrwh}{\ccode{Arrangement\_with\_history\_2}}
\newcommand{\aos}{\ccode{Arrangement\_on\_surface\_2}}
\newcommand{\aoswh}{\ccode{Arrangement\_on\_surface\_with\_history\_2}}
\newcommand{\cBso}{\ccode{Boolean\_set\_operations\_2}}
\newcommand{\cCgm}{\ccode{Cubical\_gaussian\_map\_3}}
\newcommand{\cEnvelope}{\ccode{Envelope\_3}}
\newcommand{\cSweepline}{\ccode{Sweep\_line\_2}}
\newcommand{\cZone}{\ccode{Arrangement\_zone\_2}}
\newcommand{\cObserver}{\ccode{Arr\_observer}}
\newcommand{\cPolyhedron}{\ccode{Polyhedron\_3}}
\newcommand{\cNefii}{\ccode{Nef\_2}}
\newcommand{\cNefiii}{\ccode{Nef\_3}}

\def\capStyle#1{{\small {\textsf{#1}}}}

\newcommand{\GeneralPolygonSetTraits}{\concept{GeneralPolygonSetTraits\_2}}
\newcommand{\ArrangementDirectionalXMonotoneTraits}{\concept{ArrangementDirectionalXMonotoneTraits\_2}}
\newcommand{\ArrangementBasicTraits}{\concept{ArrangementBasicTraits\_2}}
\newcommand{\ArrangementXMonotoneTraits}{\concept{ArrangementXMonotoneTraits\_2}}
\newcommand{\ArrangementTraits}{\concept{ArrangementTraits\_2}}

\newcommand{\ArrangementLandmarksBasicTraits}{\concept{ArrangementLandmarksBasicTraits\_2}}
\newcommand{\ArrangementLandmarksXMonotoneTraits}{\concept{ArrangementLandmarksXMonotoneTraits\_2}}
\newcommand{\ArrangementLandmarksTraits}{\concept{ArrangementLandmarksTraits\_2}}

\newcommand{\NoBoundaryTraits}{\concept{NoBoundaryTraits}}
\newcommand{\HasBoundaryTraits}{\concept{HasBoundaryTraits}}
\newcommand{\BoundedBoundaryTraits}{\concept{BoundedBoundaryTraits}}
\newcommand{\UnboundedBoundaryTraits}{\concept{UnboundedBoundaryTraits}}
\newcommand{\AllBoundaryTraits}{\concept{AllBoundaryTraits}}

\def\ArrangementDirectionalXMonotoneT{\psframebox{\small \ArrangementDirectionalXMonotoneTraits}}
\def\GeneralPolygonSetT{\psframebox{\small \GeneralPolygonSetTraits}}
\def\ArrangementBasicT{\psframebox{\small \ArrangementBasicTraits}}
\def\ArrangementXMonotoneT{\psframebox{\small \ArrangementXMonotoneTraits}}
\def\ArrangementT{\psframebox{\small \ArrangementTraits}}

\def\ArrangementLandmarksBasicT{\psframebox{\small \ArrangementLandmarksBasicTraits}}
\def\ArrangementLandmarksXMonotoneT{\psframebox{\small \ArrangementLandmarksXMonotoneTraits}}
\def\ArrangementLandmarksT{\psframebox{\small \ArrangementLandmarksTraits}}

\def\NoBoundaryT{\psframebox{\small \NoBoundaryTraits}}
\def\HasBoundaryT{\psframebox{\small \textcolor{gray}{\HasBoundaryTraits}}}
\def\BoundedBoundaryT{\psframebox{\small \BoundedBoundaryTraits}}
\def\UnboundedBoundaryT{\psframebox{\small \UnboundedBoundaryTraits}}
\def\AllBoundaryT{\psframebox{\small \AllBoundaryTraits}}

\newcommand{\Index}[1]{#1\index{#1}}
\makeindex

\makeatletter
\def\psarcellipse{\pst@object{psarcellipse}}
\def\psarcellipse@i(#1){\@ifnextchar(%
{\psarcellipse@ii(#1)}{\psarcellipse@ii(0,0)(#1)}}
\def\psarcellipse@ii(#1)(#2)#3#4{%
\psarcellipse@iii(#1)(#2){#3}{#4}}
\def\psarcellipse@iii(#1)(#2)#3#4{%
\begin@ClosedObj
\pst@getcoor{#1}\pst@tempa
\pst@@getcoor{#2}%
\addto@pscode{%
#3 #4
\pst@coor
\ifdim\psk@dimen\p@=\z@\else
\psk@dimen CLW mul
dup 4 -1 roll sub neg 3 1 roll sub
\fi
\pst@tempa
\tx@Ellipse
}%
\def\pst@linetype{0}%
\showpointsfalse
\end@ClosedObj}

\def\my_axis#1#2#3#4#5#6{%
  \rput{#1}(0,0){
    \psline[linewidth=2pt]{->}(0,0)(#2,0)
    \uput[0]{-#1}(#2,0){{\large $X$}}
  }
  \rput{#3}(0,0){
    \psline[linewidth=2pt]{->}(0,0)(#4,0)
    \uput[0]{-#3}(#4,0){{\large $Y$}}
  }
  \rput{#5}(0,0){
    \psline[linewidth=2pt]{->}(0,0)(#6,0)
    \uput[0]{-#5}(#6,0){{\large $Z$}}
  }
}
\makeatother

\makeatletter
\providecommand*{\toclevel@algorithm}{0}
\makeatother

	       {\VerbatimEnvironment
		 \begin{Sbox}\begin{minipage}{\textwidth}\begin{Verbatim}}%
	       {\end{Verbatim}\end{minipage}\end{Sbox}
		 \setlength{\fboxsep}{8pt}\fbox{\TheSbox}}

	       {\begin{Sbox}\begin{minipage}}%
	       {\end{minipage}\end{Sbox}\fbox{\TheSbox}}

\begin{document}
\begin{titlepage}
\begin{center}
  \psfig{figure=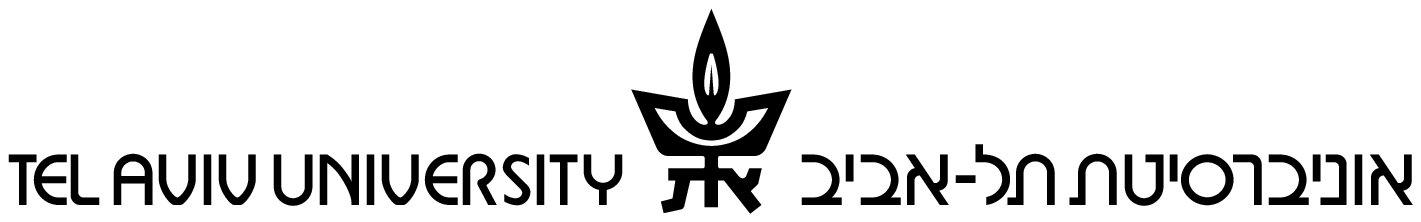,silent=}\\\vspace{0.2ex}
  {\large {\sc Raymond and Beverly Sackler}}\\\vspace{0.2ex}
  {\large {\sc Faculty of Exact Sciences}}\\\vspace{0.2ex}
  {\large {\sc The Blavatnik School of Computer Science}}\\\vfill

  {\Huge\bf Minkowski Sum Construction and}\\\vspace{0.3cm}
  {\Huge\bf other Applications of Arrangements}\\\vspace{0.3cm}
  {\Huge\bf of Geodesic Arcs on the Sphere}\\\vspace{3cm}

  {\large Thesis submitted for the degree of ``Doctor of Philosophy''}\\\vspace{0.5cm}
  {\large  by}\\\vspace{0.5cm}
  {\LARGE\bf Efraim Fogel}\\\vspace{3cm}
  {\large This work has been carried out under the supervision of}\\\vspace{0.2ex}
  {\large Prof. Dan Halperin}\\\vfill
  {\large Submitted to the Senate of Tel-Aviv University}\\\vspace{0.2ex}
  {\large October 2008}

  \author{}
  \date{}
\end{center}

\end{titlepage}
\thispagestyle{empty}
\frontmatter

\newpage
\subsection*{Acknowledgements}
This research project would not have been possible without the support
of many people. I wish to express my gratitude to my advisor,
Prof. Dan Halperin who was abundantly helpful and offered invaluable
assistance, support, insights, and guidance. This sentense is far
short of describing the magnitude of positive influence Danny had on
this research and this researcher.

\vspace{0.3ex}
\noindent I would like to thank Ron Wein and Ophir Setter
from Tel-Aviv University and Eric Berberich from Max-Planck-Insitut
f{\"u}r Informatik, for ongoing and fruitful collaboration. I would
also like to thank all other members of the applied computational
geometry group at the computer science school in Tel-Aviv University
who made the study duration festive.

\vspace{0.3ex}
\noindent I would like to thanks the members of the \cgal{} Editorial
Board in particular and all members of the \cgal{} developers'
community in general for sharing their wisdom through many useful
advises and rich discussions.

\vspace{0.3ex}
\noindent During the years I participated in the EU-funded projects
{\sc Ecg} (Effective Computational Geometry for Curves and Surfaces;
contract No. IST-2000-26473), {\sc Movie} (Motion Planning in Virtual
Environments, contract No. IST-2001-39250) and {\sc Acs} (Algorithms
for Complex Shapes; contract No. IST-006413) projects. I wish to
thank all the other participants in these projects, with whom I
enjoyed working.

\vspace{0.3ex}
\noindent Finally, I would like to express my love and gratitude to my
beloved families; for their understanding and endless love, through the
duration of my studies. 

\newpage ~~~ \newpage
\section*{Abstract}
We present two exact implementations of efficient output-sensitive
algorithms that compute {\em Minkowski sums} of two convex polyhedra
in $\rrr$. We do not assume general position. Namely, we handle
degenerate input, and produce exact results. We provide a tight bound
on the exact maximum complexity of Minkowski sums of polytopes in
$\rrr$ in terms of the number of facets of the summand polytopes.
The complexity of Minkowski sum structures is directly related to the
time consumption of our Minkowski sum constructions, as they are
output sensitive. We demonstrate the effectiveness of our Minkowski-sum
constructions through simple applications that exploit these operations
to detect collision between, and answer proximity queries about, two convex
polyhedra in $\rrr$.

The algorithms employ variants of a data structure that represents
arrangements embedded on two-dimensional parametric surfaces in $\rrr$,
and they make use of many operations applied to arrangements in these
representations. We have developed software components that support the
arrangement data-structure variants and the operations applied to them.
These software components are generic, as they can be instantiated with
any number type. However, our algorithms require only (exact) rational
arithmetic. These software components together with exact
rational-arithmetic enable a robust, efficient, and elegant
implementation of the Minkowski-sum constructions and the related
applications. These software components are provided through a
package of the Computational Geometry Algorithm Library
(\cgal)~\citelinks{cgal} called
\aos~\cite{cgal:wfzh-a2-07}. The code of \cgal{} in general, the \aos{}
package in particular, and all the rest of the code developed as part
of this thesis adhere to the {\em Generic Programming} paradigm and
follow the {\em Exact Geometric Computation} paradigm.

We also present exact implementations of other applications that exploit
arrangements of arcs of great circles, also known as  geodesic arcs,
embedded on the sphere. For example, we implemented robust polyhedra
central-projection and Boolean set-operations applied to point sets
embedded on the sphere bounded by geodesic arcs. We use them as basic
blocks in an exact implementation of an efficient algorithm that
partitions an assembly of polyhedra in $\rrr$ with two hands using
infinite translations. This application makes extensive use of
Minkowski-sum constructions and other operations on arrangements of
geodesic arcs embedded on the sphere. It distinctly shows the
importance of exact computation, as imprecise computation might result
with dismissal of valid partitioning-motions.

We have produced three movies that explain some of the concepts portrayed
in this thesis~\citelinks{movies}.\footnote{Throughout the thesis a
number in brackets (e.g.,~\citelinks{movies}) refers to the link list
starting on page~118, and an alphanumeric string in
brackets (e.g.,~\cite{ffhl-vemsa-02}) is a standard bibliographic reference.} 
The first movie~\cite{ffhl-vemsa-02} explains what Minkowski sums are and
demonstrates how they are used in various applications. The second
movie~\cite{fh-emscp_05} demonstrates the first method we have developed
to construct Minkowski-sums of convex polyhedra. The third
movie~\cite{fsh-agas-08} illustrates exact construction and
maintenance of arrangements induced by geodesic arcs and applications
that exploit such arrangements.

Additional information is available at the following web sites:\newline
\url{http://acg.cs.tau.ac.il/projects/internal-projects/gaussian-map-cubical}\newline
\url{http://acg.cs.tau.ac.il/projects/internal-projects/arrangements-on-surfaces},\newline
\texttt{http://acg.cs.tau.ac.il/projects/internal-projects/exact-complexity-of-}\newline
\texttt{minkowski-sums}, and \newline
\url{http://acg.cs.tau.ac.il/projects/internal-projects/arr-geodesic-sphere}

Auxiliary programs, source code, data sets, and documentation can be
downloaded from {\tt http://www.cs.tau.ac.il/\tildegen{}efif/Minkowski-sum}.

\tableofcontents
\listoffigures
\listoftables
\newpage ~~~ \newpage
\mainmatter

\begin{savequote}[10pc]
\sffamily
The Guide is definitive. Reality is frequently inaccurate.
\qauthor{Douglas Adams}
\end{savequote}
\chapter{Introduction}
\label{chap:intro}
Let $P$ and $Q$ be two closed convex polyhedra in $\mathbb{R}^d$.
The \Index{Minkowski sum} of $P$ and $Q$ is the convex polyhedron
$M = P \oplus Q = \{p + q\,|\,p \in P, q \in Q\}$. A polyhedron $P$
translated by a vector $t$ is denoted by $P^t$. {\em Collision Detection}
is a procedure that determines whether $P$ and $Q$ overlap. The
{\em Separation Distance} $\pi(P,Q)$ and {\em the Penetration Depth}
$\delta(P,Q)$ defined as
\begin{align*}
\pi(P,Q) & = \min\{\Vert t \Vert \,|\, P^t \cap Q \neq \emptyset, t \in \mathbb{R}^d\}\ ,\\
\delta(P,Q) & = \inf\{\Vert t \Vert \,|\, P^t \cap Q = \emptyset, t \in \mathbb{R}^d\}
\end{align*} 
are the minimum distances by which $P$ has to be translated so that $P$
and $Q$ intersect or become interior disjoint, respectively.
The problems of finding the distances above can also be posed given a
normalized vector $r$ that represents a direction, in which case the
minimum distance sought is in direction $r$. The
{\em Directional Penetration-Depth}, for example, is defined as
\begin{align*}
\delta_r(P,Q) = \inf\{\alpha\,|\, P^{\alpha \vec{r}} \cap Q = \emptyset, \alpha \in \mathbb{R}\}\ .
\end{align*}

\section{Main Contribution}
\label{sec:contribution}
We present two exact and robust implementations of efficient 
output-sensitive algorithms\index{algorithm!output sensitive} to
compute the Minkowski sum~\cite{ffhl-vemsa-02} of two convex polyhedra,
polytopes\index{polytope} for short, in
$\rrr$~\cite{fh-emscp_05,fh-eecms-07,fsh-eiaga-08,bfhks-apsca-09}. The
implementations are exact and robust, as they handle all degenerate cases,
and guarantee exact results. We demonstrate the effectiveness of our
Minkowski-sum computations through simple applications that exploit these
operations to detect collision, and compute the Euclidean separation-distance
between, and the directional penetration depth of, two polytopes in $\rrr$.

We compare our Minkowski-sum constructions with three other methods that
produce exact results we are aware of. One is a simple method that
computes the \Index{convex hull} of the pairwise sums of vertices of
two polytopes. The second is based on Nef polyhedra embedded on the
sphere, and the third is based on linear programming. We conducted
experiments with a broad family of polytopes
and compared the performance of our methods with the performance of
the other. The results reported in Table~\ref{tab:mink-time} clearly
shows that both our methods are significantly faster.

Each method we have developed uses a different variant of
Gaussian maps\index{Gaussian map}, also known as normal
diagrams\index{diagram!normal} or slope diagrams\index{diagram!slope}, to
maintain dual representations of polytopes. Each method employs a different
variant of a generic data-structure that represents
{\em arrangement}\index{arrangement} embedded on two-dimensional parametric
surfaces in $\rrr$ to maintain the dual representations. (Arrangements
embedded on two-dimensional parametric surfaces are subdivisions of the
surface, as induced by curves embedded on the surface. They play a central
role in this thesis; see Chapter~\ref{chap:aos} for a formal definition and
Figure~\ref{fig:arr} for an illustration.) Each method makes use of many
operations applied to arrangements in the corresponding representations.

We present a tight bound on the exact maximum complexity of Minkowski
sums of polytopes in $\rrr$~\cite{fhw-emcms-07}. In particular, we prove
that the maximum number of facets of the Minkowski sum of $k$ polytopes
with $m_1,m_2,\ldots,m_k$ facets respectively is bounded from above by
$\sum_{1 \leq i < j \leq k}(2m_i - 5)(2m_j - 5) + 
\sum_{1 \leq i \leq k}m_i + \binom{k}{2}$. Given $k$ positive integers
$m_1,m_2,\ldots,m_k$, we describe how to construct $k$ polytopes with
corresponding number of facets, such that the number of facets of their
Minkowski sum is exactly $\sum_{1 \leq i < j \leq k}(2m_i - 5)(2m_j - 5) + 
\sum_{1 \leq i \leq k}m_i + \binom{k}{2}$. When $k = 2$, for example, the
expression above reduces to $4m_1m_2 - 9m_1 - 9m_2 + 26$.

We also present an exact implementation of an efficient algorithm that
partitions an assembly of polyhedra in $\rrr$ with two hands using
infinite translations~\cite{fh-papit-08}. This application makes extensive
use of Minkowski-sum constructions and other operations on arrangements of
arcs of great circles, also known as \Index{geodesic} arcs, embedded on
the sphere, such as polyhedra central-projection and Boolean
set-operations of point sets embedded on the sphere bounded by geodesic
arcs. The assembly partitioning demonstrates the importance of exact
computation, as imprecise computation might result with dismissal of valid
partitioning-motions.

Both methods that construct Minkowski sums and the related applications
are implemented on top of the Computational Geometry Algorithms Library
(\cgal{})~\cite{ft-gpcl-07}, and are mainly based on the arrangement
package of the
library~\cite{fwh-cfpeg-04,wfzh-aptac-07,fhktww-a-07,bfhmw-scmtd-07,bfhmw-apsgf-09},
which supports arrangements embedded on two-dimensional parametric
surfaces and operations on them. We have redesigned, re-implemented, and
significantly extended the package exploiting advanced programming
techniques to yield a package that is easy to use, to extend, and to adapt
to a variety of applications.

The package contains a general framework for computing the zone of a
single curve embedded on a two-dimensional parametric surface and a
general framework for sweeping a set of such curves. The former
framework is used, for example, to insert curves one at a time into the
arrangement. The latter framework is used, for example, to compute the
arrangement induced by a collection of curves, and to compute the overlay
of two arrangements. Other operations, such as point location and vertical
ray shooting, are supported as well.

The design dictates the separation between the topological and geometric
aspects of the two-dimensional subdivision. This separation is
advantageous, as it allows users to employ the package with their own
representation of any special family of curves. They must however
supply a relevant component called
{\em geometry traits}\index{traits!geometry} class that handles the specific
family of curves  they are interested in, which mainly involves algebraic
computation. The separation is enabled by a modular design and conveniently
implemented within the {\em generic-programming}\index{generic programming}
paradigm. The separation is a key aspect of the package, as well as of
other central \cgal{} components, such as the various triangulation
packages~\cite{bdty-tcgal-00} and convex-hull algorithms
(see~\cite{cgal:eb-07} for more details), has been forced since its early
stages, and heightened by our new design.

The package comes with many geometric traits classes that handle all
kinds of curves organized in a structured hierarchy. In particular, we
mention the geometric traits class that handles continuous piecewise
linear curves, referred to as polylines\index{polyline} and the geometric
traits class that handles geodesic arcs embedded on the
sphere~\cite{fsh-agas-08,fsh-eiaga-08,bfhks-apsca-09}.
The former are of particular interest, as they can be used to approximate
more complex curves in the plane. At the same time they are easier to
deal with in comparison to higher-degree algebraic curves, as rational
arithmetic is sufficient to carry out exact computations on polylines.
The latter is broadly used by the assembly-partitioning application and
by the second method that constructs Minkowski sums.

The implementation of the package as well as the implementation our
Minkowski-sum constructions, collision detection, and assembly partitioning
applications handle degenerate input and produce exact results, as long as
the underlying \Index{number type} supports the arithmetic operations
$+$, $-$, $*$, and $/$ in unlimited precision over the
rationals,\footnote{Commonly referred to as a {\em field} number
type\index{number type!field}.} such as the rational number type
{\tt CGAL::Gmpq} based on \gmp{} --- Gnu's Multi Precision
library~\citelinks{gmp}.

\section{Background: Minkowski Sums}
\label{sec:intro:ms-bg}
Minkowski sums are closely related to proximity queries~\cite{cpq-lm-04}.
For example, the minimum separation distance between two polytopes $P$
and $Q$ is equal to the minimum distance between the origin and the
boundary of the Minkowski sum of $P$ and the reflection of $Q$ through
the origin~\cite{cc-dmtdb-86}. Computing Minkowski sums, collision
detection and proximity calculation constitute fundamental tasks in
computational geometry~\cite{hkl-r-04,cpq-lm-04,s-amp-04}. These
operations are ubiquitous in robotics, solid modeling, design
automation, manufacturing, assembly planning, virtual prototyping, and
many more domains; see,
e.g.,~\cite{blt-mosal-97,eor-cmsdm-92,KAUL91,l-rmp-91,vksm-sacmp}. 

The wide spectrum of ideas expressed in the massive amount of literature
published about the subject during the last three decades has inspired
the development of quite a few useful solutions. For an extensive overview
about the subject and a comprehensive list of packages see~\cite{cpq-lm-04}.
However, only recent advances in the implementation of
computational-geometry algorithms and data structures made our exact,
robust, and efficient implementation possible. The exact
geometric-computation paradigm~\cite{y-rgc-04} designed for 
implementing computational geometry algorithms particularly prevails
in \cgal{}. While in general the underlying arithmetic is highly time
consuming compared to machine floating-point arithmetic, major efficiency
is gained by computing predicates only to sufficient precision to evaluate
them correctly; see Section~\ref{ssec:intro:geometric-programming} and
e.g.,~\cite{pf-glese-06}.

Various methods to compute the Minkowski sum of two polyhedra in $\rrr$
have been proposed. The goal is typically to compute the boundary of the
sum provided in some representation. The combinatorial complexity
of the Minkowski sum of two polyhedra of $m$ and $n$ features respectively
can be as high as $\Theta(m^3n^3)$. One common approach to compute it is
to decompose each polyhedron into convex pieces, compute pairwise Minkowski
sums of pieces of the two, and finally the union of the pairwise sums.
Computing the exact Minkowski sum of non-convex polyhedra is naturally
expensive. Therefore, researchers have focused on computing an
approximation that satisfies some criteria, such as the algorithm
presented by Varadhan and Manocha~\cite{vm-amapm-06}. They guarantee a
two-sides Hausdorff distance bound on the approximation, and ensure that
it has the same number of connected components as the exact Minkowski sum.
Computing the Minkowski sum of two convex polyhedra remains a key
operation, and this is what we focus on. The combinatorial complexity of
the sum can be as high as $\Theta(mn)$ when both polyhedra are convex.
For the complexity of the intermediate case, where only one polyhedron
is convex, cf.~\cite{as-tmpcp-97,s-amp-04}.

Convex decomposition is not always possible, as in the presence of
non-convex curved objects. In these cases other techniques must be
applied, such as approximations using polynomial/rational curves in
2D~\cite{lke-prams-98}. Seong at al.~\cite{jmk-mstss-02} proposed an
algorithm to compute Minkowski sums of a subclass of objects; that is,
surfaces generated by slope-monotone closed curves. Flato and
Halperin~\cite{afh-pdecm-02} presented algorithms based on \cgal{} for
robust construction of planar Minkowski sums. While the citations in this
paragraph refer to computations of Minkowski sums of non-convex polyhedra,
and we concentrate on the convex cases, the latter is of particular
interest, as our method makes heavy use of the same software components,
in particular the \cgal{} arrangement package, which went through a few
phases of
improvements~\cite{fwh-cfpeg-04,wfzh-aptac-07,bfhmw-scmtd-07,bfhmw-apsgf-09}
since its usage in~\cite{afh-pdecm-02}; see
Section~\ref{sssec:intro:cgal:arr-history} for more details.

A particular accomplishment of the {\em \Index{kinetic framework}} in two
dimensions introduced by Guibas \etal~\cite{grs-kfcg-83} was the
definition of the {\em convolution} operation in two dimensions, a
superset of the Minkowski sum operation, and its exploitation in a
variety of algorithmic problems. Basch \etal\ extended its predecessor
concepts and presented an algorithm to compute the convolution in three
dimensions~\cite{bgrr-pttc-96}. An output-sensitive
algorithm\index{algorithm!output sensitive} for computing Minkowski
sums of polytopes was introduced in~\cite{gs-ccrs-87}.
Gritzmann and Sturmfels~\cite{gs-mapca-93} obtained a polynomial time
algorithm in the input and output sizes for computing Minkowski sums of
$k$ polytopes in $\mathbb{R}^d$ for a fixed dimension $d$, and
Fukuda~\cite{f-zcmac-04} provided an output-sensitive polynomial algorithm
for variable $d$ and $k$. Ghosh~\cite{p-ucfms-93} presented a unified
algorithm for computing 2D and 3D Minkowski sums of both convex and
non-convex polyhedra based on a {\em slope diagram}\index{diagram!slope}
representation. Computing the Minkowski sum amounts to computing the
slope diagrams of the two objects, {\em merging them} (see details in
Section~\ref{ssec:mscn:sgm:mink_sum},) and extracting the boundary of
the Minkowski sum from the merged diagram.
Wu \etal~\cite{wsd-iacms-03} introduced an improved version of Ghosh'
algorithm for convex polyhedra using vector operations. Bekker and
Roerdink~\cite{654217} provided a variation on the same idea. The
slope diagram of a 3D convex polyhedron can be represented as a 2D
object, essentially reducing the problem to a lower dimension. We
follow the same approach, but use exact computation.

We postpone a formal definition of arrangements to the next chapter.
In fact, we temporary put Minkowski sums and arrangements aside to
provide relevant programming background material.

\section{Background: Programming}
\label{sec:programming-bg}
\subsection{Generic Programming}
\label{ssec:intro:generic-programming}
Several definitions of the term {\em \Index{generic programming}}
have been proposed since it was first coined around the early
sixties, along with the introduction of the \lisp\index{LISP@\lisp}
programming language. Here we confine ourself to the classic notion
first described by Musser \etal~\cite{ms-gp-88}, who considered generic
programming as a discipline that consists of the gradual lifting of
concrete algorithms abstracting over details, while retaining the
algorithm semantics and efficiency. Within this context, several
approaches have been put into trial through the introduction of new
features in existing computer languages, or even new computer
languages all together. The software described in this thesis is
written in \Cpp, a programming language that is well equipped for
writing software according to the generic-programming paradigm
through the extensive use of class templates and function templates.

\subsubsection{Concepts and Models}
\label{sssec:intro:gp:concepts}
One crucial abstraction supported by all contemporary computer
languages is the subroutine (also known as procedure or function,
depending on the programming language). Another abstraction
supported by \Cpp{} is that of abstract data typing, where a new data
type is defined together with its basic operations. \Cpp{} also
supports \Index{object-oriented programming}, which emphasizes packaging
data and functionality together into units within a running program,
and is manifested in hierarchies of polymorphic data-types related
by inheritance. It allows referring to a value and manipulating it
without needing to specify its exact type. As a consequence, one can
write a single function that operates on a number of types within an
inheritance hierarchy. Generic programming\index{generic programming}
identifies a more powerful abstraction (perhaps less tangible than
other abstractions). It is a formal hierarchy of polymorphic abstract
requirements on data types referred to as {\em concepts}\index{concept},
and a set of classes that conform precisely to the specified
requirements, referred to as {\em models}\index{model}. Models that
describe behaviors are referred to as {\em \Index{traits}}
classes~\cite{m-tnutt-97}. Traits classes typically add a level of
indirection in template instantiation to avoid accreting parameters
to templates.

A generic algorithm has two parts: The actual instructions that
describe the steps of the algorithm, and a set of requirements that
specify which properties its argument types must satisfy. The
following \ccode{swap()} function is an example of the first part of a
generic algorithm.

\noindent\begin{tabularx}{\textwidth}{l@{}p{2ex}X}
  \hline
  & \multicolumn{2}{l}{\ccode{template <typename T> void swap(T \& a, T \& b) \{}}\\
  & & \ccode{T tmp = a; a = b; b = tmp;}\\
  & \multicolumn{2}{l}{\ccode{\}}}\\
  \hline
\end{tabularx}

\noindent
When the function call is compiled, it is instantiated with a data
type that must have an assignment operator. A data type that fulfils
this requirement is a model of a concept commonly called
\concept{Assignable}\index{Assignable@\concept{Assignable}}~\cite{a-gps-99}.
The \ccode{int} data type, for example, is a model of this concept, so
it can be used to instantiate the function
template~\cite{a-gps-99}~\citelinks{stl}.

A concept is a set of requirements divided into four categories, namely,
associated types, valid expressions, invariants, and complexity guarantees.
When a type meets all requirements of a concept, the type is considered a
{\em \Index{model}} of the \Index{concept}. When a concept extends the
requirements of another concept, the former is said to be a
{\em \Index{refinement}} of the latter.
\begin{description}
\item[Associated Types] are auxiliary types. For example, a type that
  represents a two-dimensional point, namely \ccode{Point\_2}, is required
  by the arrangement geometry traits concept; see
  Section~\ref{sec:aos:geometry-traits}.
\item[Valid Expressions] are \Cpp{} expressions that must compile
  successfully. For example, \ccode{p = q}, where \ccode{p}
  and \ccode{q} are both objects of type \ccode{Point\_2}. Valid
  expressions identify the set of operations a model of the concept
  must be able to perform.
\item[Invariants] are run-time characteristics such as time and space
  complexity bounds. In our context invariants typically take the form
  of preconditions and postconditions, which must always be satisfied.
  For example, a condition that requires that an input point $p$ lies
  on an input curve $c$ on invocation to a predicate that accepts both
  $p$ and $c$ as parameters. Having preconditions typically minimizes
  the concept, as the operations provided by a model must operate
  only on restricted arguments. Formally, removing preconditions from,
  and introducing postconditions to, a requirement set results with a
  refined concept.
\item[Complexity Guarantees] are maximum limits on the computing
  resources consumed by the various expressions.
\end{description}

\subsubsection{Traits Classes}
\label{sssec:intro:generic-programming:traits}
The name ``traits class'' comes from a standard \Cpp{} design
pattern~\cite{m-tnutt-97}, which provides a way of associating
information with a compile-time entity (typically a type). For example,
the standard class-template \ccode{std::iterator\_traits<T>} looks
roughly like this:

\noindent\begin{tabularx}{\textwidth}{l@{}p{2ex}X}
  \hline
  & \multicolumn{2}{l}{\ccode{template <typename Iterator> struct iterator\_traits \{}}\\
  & & \ccode{typedef ...\ iterator\_category;}\\
  & & \ccode{typedef ...\ value\_type;}\\
  & & \ccode{typedef ...\ difference\_type;}\\
  & & \ccode{typedef ...\ pointer;}\\
  & & \ccode{typedef ...\ reference;}\\
  & \multicolumn{2}{l}{\};}\\
  \hline
\end{tabularx}

\noindent
Iterators play an important role in \Index{generic programming}: A function
that operates on a range of objects usually accepts two iterators that
specify this range. The traits' \ccode{value\_type} specifies the type of
object the iterators are pointing at, while the \ccode{iterator\_category}
can be used to select more efficient algorithms depending on the
iterator's capabilities.

A key property of trait classes is that they are non-intrusive. Namely,
they allow us to associate information with arbitrary types, without
interfering with the internal representation of those types. Thus, it is
possible to define a traits class also for built-in types and types
defined in third-party libraries.

Within the context of \cgal, for example, a typical traits class is
required to define nested types of geometric objects and support
predicates involving objects of these types. Some algorithms also require
the provision of constructions by the traits class.

Let us continue with an easy geometric example. Consider a function that
accepts a set of points, given by the range
\ccode{[pts\_begin, pts\_end)},\footnote{This notation means that
\ccode{pts\_begin} points to the first point in the range, while
\ccode{pts\_end} points after the range ends (it is therefore called a
{\em past-the-end iterator}, and need not point to any valid point
object).} and computes the minimal axis-parallel rectangle that
contains all points in the range. It does so by locating the points
with extremal $x$ and $y$-coordinates, and then constructs the bounding
iso-rectangle accordingly:

\noindent\begin{tabularx}{\textwidth}{l@{}p{2ex}p{2ex}p{2ex}X}
  \hline
  & \multicolumn{4}{l}{\ccode{template <typename InputIterator, typename Traits>}}\\
  & \multicolumn{4}{l}{\ccode{typename Traits::Iso\_rectangle\_2}}\\
  & \multicolumn{4}{l}{\ccode{bounding\_rectangle(InputIterator pts\_begin, InputIterator pts\_end) \{}}\\
  & & \multicolumn{3}{l}{\ccode{Traits traits;}}\\
  & & \multicolumn{3}{l}{\ccode{InputIterator curr = pts\_begin;}}\\
  & & \multicolumn{3}{l}{\ccode{InputIterator left, right, top, bottom;}}\\
  & & \multicolumn{3}{l}{\ccode{left = right = top = bottom = curr++;}}\\
  & & \multicolumn{3}{l}{\ccode{while (curr++ != pts\_end) \{}}\\
  & & & \multicolumn{2}{l}{\ccode{if (traits.compare\_x(*curr, *left) == SMALLER) left = curr;}}\\
  & & & \multicolumn{2}{l}{\ccode{else if (traits.compare\_x(*curr, *right) == LARGER) right = curr;}}\\
  & & & \multicolumn{2}{l}{\ccode{if (traits.compare\_y(*curr, *bottom) == SMALLER) bottom = curr;}}\\
  & & & \multicolumn{2}{l}{\ccode{else if (traits.compare\_y(*curr, *top) == LARGER) top = curr;}}\\
  & & \multicolumn{3}{l}{\ccode{\}}}\\
  & & \multicolumn{3}{l}{\ccode{return traits.construct\_iso\_rectangle(*left, *right, *bottom, *top);}}\\
  & \multicolumn{4}{l}{\ccode{\}}}\\
  \hline
\end{tabularx}

The requirements that an instantiated traits class must satisfy in this case
are as follows: It has to defined the nested type \ccode{Iso\_rectangle\_2}
(and implicitly also a point type say \ccode{Point\_2}).
Moreover, it should supply two three-valued predicates\footnote{The
predicate return value is \ccode{SMALLER}, \ccode{EQUAL}, or
\ccode{LARGER}.} that compare two points by their $x$-coordinates
and by their $y$-coordinates, respectively. It should also support the
construction of an axis-parallel iso-rectangle from four points that
specify its extremal $x$ and $y$-coordinates. Note, however, that the
actual representation of points and rectangles (the coordinate system,
the number-type used to represent the coordinates, etc.)
and the implementation of the traits-class operations is entirely
decoupled from the function \ccode{bounding\_rectangle()} we have
introduced.

Consider an imaginary generic implementation of a data structure that
handles geometric arrangements embedded on two-dimensional parametric
surfaces in space called \aos. Its prototype is listed
below. This template class must be instantiated with a traits class
that in turn defines a type that represents a certain family of
curves, and some functions (or function
objects\index{function object}; see~\citelinks{function-object} for an exact
definition) that operate on curves of this family.

\noindent
\begin{tabularx}{\textwidth}{l@{}X}
  \hline
  & \ccode{template <typename Traits> class Arrangement\_on\_surface\_2 \{ ...\ \};}\\
  \hline
\end{tabularx}

\noindent
Traits classes that handle families of curves embedded on parametric
surfaces are intricate, as they model involved concepts. The precise
definitions of these concepts and their refinement hierarchy are
described in Section~\ref{sec:aos:geometry-traits}.

One important objective is to minimize the set of requirements the
traits concept imposes. A tight traits concept may save tremendously in
analysis and programming of classes that model the concept. Another
important reason for striving for the minimal set of requirements is
to avoid computing the same algebraic entity in different ways. The
importance of this is amplified in the context of computational
geometry, as a non tight model that consists of duplicate, but
slightly different, implementations of the same algebraic entity, can
lead to artificial degenerate conditions, which in turn can
drastically impair the performance.

Most traits classes in \cgal{} are parameterized by a model of the
\concept{Kernel} concept. The choice of a particular model determines,
among the other, the type of arithmetic used, as explained in the
following sections. One can easily switch between different models of the
\concept{Kernel} concept, but here lies a trap, as
Section~\ref{ssec:intro:geometric-programming} reveals. A kernel model
that supports exact arithmetic must be used to ensure robustness,
although inexact arithmetic could be used at a certain risk.

\subsubsection{Libraries}
\label{sssec:intro:generic-programming:lib}
Alexander Stepanov began exploring the potential of compile-time
polymorphism for revolutionizing software development in 1979. With
the help of several other researchers his work evolved into a prime
generic-programming library --- the Standard Template Library (\stl).
This library had became part of the \Cpp{} standard library in 1994,
approximately one year before early development of \cgal{} started; see
Section~\ref{sssec:intro:cgal:chronicles} for details about the evolution
of \cgal.

Through the years a few other generic-programming libraries emerged.
One notable library in the context of computational geometry is
\leda\index{LEDA@\leda} (Library of Efficient Data Types and
Algorithms), a library of combinatorial and geometric data types and
algorithms~\cite{mn-lpcgc-00}~\citelinks{leda}. Early development
of \leda{} started in 1988, ten years before the first public release
of \cgal{} became available. In some sense \leda{} is a predecessor
of \cgal{}, although the two libraries are headed in different
directions. While \leda{} is mostly a large collection of fundamental graph
related and general purpose data-structures and algorithms, \cgal{} is a
collection of large and complex data-structures and algorithms focusing on
geometry. 

A noticeable influence on \Index{generic programming} is conducted by the
\boost{} online community, which encourages the development of free
\Cpp{} software gathered in the \boost{} library
collection~\citelinks{boost}. It is a large set of
portable and high quality \Cpp{} libraries that work well with, and are
in the same spirit as, the \Cpp{} Standard Template Library. The
\boost{} Graph Library (\bgl)~\cite{sll-bgl-02}, which consists of
generic graph-algorithms, serves a particularly important role in our
context. An arrangement instance, for example, can be adapted as a
\bgl{} graph, and passed as input to generic algorithms already
implemented in the \bgl, such as the Dijkstra shortest path algorithm.
We use the \bgl{} to compute the strongly connected components of a
directed graph --- a phase in the assembly partitioning operation; see
Section~\ref{ssec:assem_plan:processing} for more details about the
application of this operation.

\subsection{Geometric Programming}
\label{ssec:intro:geometric-programming}
Implementing geometric algorithms and data structures is notoriously
difficult, much harder than may seem when just considering the
algorithm as described in a paper or a book.
In the traditional computational-geometry literature two assumptions
are usually made to simplify the design and analysis of geometric
algorithms. First, inputs are in ``general position''. That is,
degenerate input (e.g., three curves intersecting at a common point)
is precluded. Secondly, operations on real numbers yield accurate
results (the ``real {\sc Ram}'' model~\cite{ps-cgi-85}, which also
assumes that each basic operation takes constant time).
Unfortunately, these assumptions do not hold in practice, as
degenerate input is commonplace in practical applications and
numerical errors are inevitable. Thus, an algorithm implemented
without keeping this in mind may yield incorrect results, get into
an infinite loop, or just crash, while running on a degenerate, or
nearly degenerate, input (see~\cite{kmpsy-cerpg-08,s-rpigc-00} for
examples). These pitfalls have become well known, and have been the
subject of intensive research~\cite{s-rpigc-00,y-rgc-04}.

Indeed, the last decade has seen significant progress in the
development of software for computational geometry. The mission of
such a task, which Kettner and N{\"a}her~\cite{kn-tcgll-04} call
{\em geometric programming}, is to develop software that is correct,
efficient, flexible (namely adaptable and extensible\footnote{{\em
Adaptability\/} refers to the ability to incorporate existing user
code, and {\em extendibility\/} refers to the ability to enhance the
software with more code in the same style.}), and easy to use.

\subsubsection{Separation of Topology and Geometry}
\label{sssec:intro:geometric-programming:sep}
The use of the generic-programming\index{generic programming}
paradigm enables a convenient separation between the topology and the
geometry of data structures.\footnote{In this context, we sometimes say
{\em combinatorics} instead of topology, and say {\em algebra} or
{\em numerics} instead of geometry. We always mean the same thing: The
separation between the abstract, graph-like structure (the topology)
from the actual embedding on the surface (the geometry).} This is a key
aspect in the design of geometric software, and is put into practice,
for example, in the design of \cgal{} polyhedra, \cgal{} triangulations,
and our \cgal{} arrangements. This separation allows the convenient
abstraction of algorithms and data structures in combinatorial and
topological terms, regardless of the specific geometry of the objects at
hand and the algebra used to represent them. This abstraction is realized
through class and function templates that represent specific
data-structures and algorithmic frameworks, respectively. Consider again
our imaginary \aos{} class template from the previous section;
its improved prototype is listed below. It is instantiated with two
classes. The first, referred to as a {\em geometric traits} class,
defines the set of geometric-object types and the operations on objects
of these types. The second, defines the topological-object types and the
operations required to maintain the incident relations among objects of
these types.

\noindent
\begin{tabularx}{\textwidth}{l@{}X}
  \hline
  & \ccode{template <typename Geometry\_traits, typename Topology\_traits>}\\
  & \ccode{class Arrangement\_on\_surface\_2 \{ ... \};}\\
  \hline
\end{tabularx}

An immediate advantage of the separation between the topology and
the geometry of data structures is that users with limited expertise
in computational geometry can employ the data structure with their
own special type of objects. They must however supply the relevant
traits class, which mainly involve algebraic computations. A traits
class also encapsulates the number types used to represent
coordinates of geometric objects and to carry out algebraic
operations on them. It encapsulates the type of coordinate system
used (e.g., Cartesian, Homogeneous), and the geometric or algebraic
computation methods themselves. Naturally, a prospective user of the
package that develops a traits class would like to face as few
requirements as possible in terms of traits development.

Another advantage gained by the use of generic programming is the
convenient handling of numerical issues to expedite exact geometric
computation. We arrive at this conclusion at the end of the
next section. In a geometric algorithm each computational step is
either a construction step or a conditional step based on the result
of a predicate. The former produces a new geometric object such as
the intersection point of two segments. The latter typically
computes the sign of an expression used by the program control.
Different computational paths lead to results with different
combinatorial characteristics. Although numerical errors can
sometimes be tolerated and interpreted as small perturbations in the
input, they may lead to invalid combinatorial structures or
inconsistent state during a program execution. Thus, it suffices to
ensure that all predicates are evaluated correctly to eliminate
inconsistencies and guarantee combinatorially correct results.
This is easier said than done, but nowadays possible, as the next
section exposes.

\subsubsection{Exact Geometric Computation}
\label{sssec:intro:geometric-programming:egc}
The need for robust software implementations of computational-geometry
algorithms has driven many researchers over the last decade to develop
variants of the classic algorithms that are less susceptible to
degenerate inputs. The approaches taken to overcome the difficulties
in robustly implementing geometric algorithms can be roughly divided
into two categories: (i)~Exact computing, and (ii)~fixed-precision
approximation. In the latter approach the algorithms are modified so
that they can consistently cope with the limited precision of
computer arithmetic. In the former, which is the approach taken by
\cgal{} in general and the arrangement package in particular, ideal
computer arithmetic is emulated for the specific type of objects
being manipulated, and the code is prepared for successfully handling
degenerate input.

Advances in computer algebra enabled the development of efficient
software libraries that offer exact arithmetic manipulations on
unbounded integers, rational numbers (\gmp{} --- Gnu's multi-precision
library~\citelinks{gmp}), and algebraic numbers (the \core{}
library~\cite{klpy-clrngc-99}~\citelinks{core} and the numerical
facilities of \leda~\cite[Chapter~4]{mn-lpcgc-00}~\citelinks{leda}).
These exact-number types serve as fundamental building-blocks in the
robust implementation of many geometric applications in general
(see~\cite{y-rgc-04} for a review) and of those that employ
arrangements in particular.

Exact Geometric Computation (EGC), as summarized by
Yap~\cite{y-rgc-04}, simply amounts to ensuring that we never
err in predicate evaluations. EGC represents a significant
relaxation from the naive concept of numerical exactness. We only
need to compute to sufficient precision to make the correct
predicate evaluation. This has led to the development of several
techniques such as precision-driven computation, lazy evaluation,
adaptive computation, and floating-point filters, some of which are
implemented in \cgal, such as numerical filtering. Here, computation
is carried out using a number type that supports only inexact
arithmetic (e.g., double-precision floating-point arithmetic), while
applying a filter that checks whether the computation has reached a
stage of uncertainty, an event referred to as a {\em filter
failure} in the hacker's jargon. When a filter failure occurs, the
computation is re-done using exact arithmetic.

Switching between number types and exact computation techniques, and
choosing the appropriate components that best suit the application
needs, is conveniently enabled through the generic-programming
paradigm, as it typically requires only a minor code change reflected
in the instantiating of just a few data types.

\subsection{Computational Geometry Algorithms Library}
\label{ssec:intro:cgal}
\subsubsection{\protect{\cgal} Chronicles}
\label{sssec:intro:cgal:chronicles}
Several research groups in Europe started to develop small geometry
libraries on their own in the early 1990s. A consortium of several sites
in Europe and Israel was founded in 1995 to cultivate the labor
of these groups and gather their produce in a common library called
\cgal{} --- the Computational Geometry Algorithms Library. 
The goal was to promote the research in computational geometry and
translate the results into useful, reliable, and efficient programs for
industrial and academic
applications~\cite{o-dcgal-96,cgal:eb-07,kn-tcgll-04,fgkss-dccga-00},
the very same goal that governs \cgal{} development efforts to date.

An INRIA startup, \gf~\citelinks{gf}
was founded in January 2003. The company sells \cgal{} commercial
licenses, support for \cgal, and customized developments based on \cgal.

In November 2003, when Version~3.0 was released, \cgal{} became an Open
Source\linebreak
Project~\citelinks{cgal}, allowing new
contributions from various resources. Most parts of \cgal{} are now
distributed under the GNU Lesser General Public License (GNU LGPL).

\cgal{} has evolved through the years and is now representing the
state-of-art in implementing computational geometry software in many
areas. The implementations of the \cgal{} software modules described in
this thesis are complete and robust, as they handle all degenerate cases.
They rigorously adhere to the
generic-programming\index{generic programming} paradigm to
overcome problems encountered when effective computational geometry
software is implemented. A glimpse at the structure of \cgal{} is
given in the following subsection.

\subsubsection{\protect{\cgal} Content}
\label{sssec:intro:cgal:content}
\cgal{} is written in \Cpp{} according to the
generic-programming\index{generic programming} paradigm described above.
It has a common programming style, which is very similar to that of the
\stl. Its \Index{application programming-interface} (API) is homogeneous,
and allows for a convenient and
consistent interfacing with other software packages and applications.
The library consists of about 900,000 lines of code divided among
approximately 4,000 files. \cgal{} also comes with numerous examples
and demos. The manual comprises about 3,500 pages. There are
approximately 65 chapters arranged in 14 parts. The \aos{} package,
for example, consists of about 140,000 lines of code divided among
approximately 300 files, described in about 300 pages of a didactic manual.

One distinguished piece consists of the geometric
kernels~\cite{fgkss-dccga-00}. A geometric kernel consist
of types of constant size non-modifiable geometric primitive objects
(e.g., points, lines, triangles, circles, etc.) and operations on
objects of these types.

Another distinguished piece, referred to as the ``Support Library''
consists of non-geometric facilities, such as circulators, random
generators, and I/O support for interfacing \cgal{} with various
visualization tools (i.e., input and output streams). An important
contribution of this piece is the number-type support. This piece
also contains extensive debugging utilities that handle warnings
and errors that may result from unfulfilled conditions.

The rest of the library offers a collection of geometric data
structures and algorithms such as convex hull, polygons and
polyhedra and operations on them (Boolean operations, polygon
offsetting), 2D and 3D triangulations, Voronoi diagrams, surface
meshing and surface subdivision, search structures, geometric
optimization, interpolation, and kinetic data-structures. The 2D
arrangements and its related data-structures naturally fit in. These
data structures and algorithms are parameterized by traits classes
that define the interface between them and the primitives they use.
In many cases, a kernel can be used as a traits class, or at least
the subtypes of a kernel can be used as components of traits classes
for these data structures and algorithms.

\subsubsection{\cgal{} Arrangement Package History}
\label{sssec:intro:cgal:arr-history}
\cgal{} contains an elaborate and efficient implementation of a generic
data-structure that represents an arrangement induced by general types
of curves embedded on a two-dimensional parametric surface in $\rrr$,
but it has not been like this from the beginning. The first version of
the \cgal{} arrangement package supported only line segments,
circular arcs, and restricted types of parabolas embedded in the plane.

While the first version supported only limited types of curves, it was
originally designed with the vision of supporting general
curves~\cite{fhhne-dipmc-00,h-dipac-00,hh-tdaca-00}. This vision was
reflected, among the other, through the separation between the
topological and the geometric aspects enabled by the generic-programming
paradigm (see Section~\ref{sssec:intro:geometric-programming:sep}).
Most of the principles related to the topology, e.g., the use of a
{\em doubly-connected edge list} (\dcel\index{DCEL@\dcel}) to maintain
the incident relations between the arrangement cells (i.e., vertices,
halfedges, and faces) were conceived from the start.
However, the types of curves that induce arrangements (see
Section~\ref{sec:aos:geometry-traits}) gradually expanded. A couple
of years after the introduction of \cgal{} arrangement Wein extended its
implementation to support arcs of conics~\cite{w-hlfac-02}. The
arsenal of geometric traits continues to grow to date. 

Following the requirements that emerged from the ECG project\footnote{ECG
is a Shared-Cost RTD (FET Open) Project of the European Union devoted
to Effective Computational Geometry for curves and
surfaces~\citelinks{ecg}.}, together with Wein we improved and refined
the software design of the \cgal{} arrangement
package~\cite{fwh-cfpeg-04}. This new design formed a common platform
for a preliminary comparison between different approaches to handle
arcs of conics~\cite{ecg:fhw-ecsca-04}. Two years later in a joint
effort with Wein and Zukerman the whole package was
revamped~\cite{wfzh-aptac-05} leading to more compact, easier-to-use,
and efficient code.

\cgal{} Version~3.2 released in~2006 included an arrangement package that
supported only bounded curves in the plane. This forced users to clip
unbounded curves before inserting them into the arrangement; it was the
user responsibility to clip without loss of information. However, this
solution is generally inconvenient and outright insufficient for some
applications. For example, representing the minimization diagram
\index{diagram!minimization} defined by the lower envelope of unbounded
surfaces in $\rrr$~\cite{m-rgece-06} generally requires more than one
unbounded face, whereas an arrangement of bounded clipped curves contains
a single unbounded face.

\cgal{} Version~3.3 released a year later in~2007 already included an
arrangement package that handled unbounded planer curves. As a matter of
fact it included much more. Together with Berberich, Melhorn, and Wein we
observed the possibility to maximize code reuse by generalizing the
various algorithms applied to arrangements and their implementations so
that they could be employed on a large class of surfaces and curves
embedded on them~\cite{bfhmw-scmtd-07,bfhmw-apsgf-09}. Indeed, the
algorithms and their implementations were designed with the vision of
supporting general curves embedded on parametric surfaces. However, only
a few geometric traits-classes that supported unbounded curves in the plane
were included in Version~3.3.

A future version of \cgal{}, expected to be released in~2010, is planned
to include an arrangement package that constructs, maintains, and operates
on arrangements embedded on certain two-dimensional orientable parametric
surfaces. The package already exists as a prototypical \cgal{} package under
the new name \aos{} to better reflect its capabilities. For example, it
includes
(i)~a geometric traits that handles geodesic arcs embedded on the
sphere~\cite{fsh-agas-08,bfhks-apsca-09}, (ii)~a geometric traits that
handles intersections between quadric surfaces embedded on a
quadric~\cite{bfhmw-scmtd-07,bfhks-apsca-09}, and (iii)~a geometric traits
that handles intersections between arbitrary algebraic surfaces and a
parameterized \Index{Dupin cyclide} embedded on the Dupin
cyclide~\cite{bk-eatdc-08,bfhks-apsca-09}. The references to the
arrangement software in this thesis in general and in Chapter~\ref{chap:aos}
in particular pertain to this latest version.

The leap in arrangement technology expressed by the ability to construct
and maintain arrangements embedded on two-dimensional parametric surfaces
immediately affects other components in \cgal{}, such as the \cBso{}
package. Only little effort is now required to support Boolean
set-operations on point sets bounded by general curves embedded on
two-dimensional parametric surfaces.

\section{Thesis Outline and Related Publications}
\label{outline}
The rest of this thesis is organized as follows. In
Chapter~\ref{chap:aos} we give an overview of the \cgal{} arrangement
package, which provides the common infrastructure for all software
solutions described in this thesis. This chapter contains selected
sections from several papers and from a book chapter we have co-authored,
and it provides new material that has not been published yet. In
particular, the chapter contains a description of the architecture of
the arrangement generic data-structure, part of which also appeared in
the chapter {\em Arrangement} of the book {\em Effective Computational
Geometry for Curves and Surfaces}~\cite{fhktww-a-07}. The chapter contains
a description of advanced programming techniques applied in the context of
arrangement, parts of which appeared in a joint work with R.\ Wein and
B.\ Zukererman, and published in the journal {\em Computational
Geometry --- Theory and Application} (special issue on
\cgal)~\cite{wfzh-aptac-07}. Preliminary results were first introduced in
(i) the proceedings of the {\em $12^{th}$ Annual European Symposium on
Algorithms (ESA)}~\cite{fwh-cfpeg-04} and (ii) the {\em $1^{st}$
Workshop of Library-Centric Software Design}~\cite{wfzh-aptac-05}. The
chapter also presents arrangements induced by general curves and embedded
on parametric surfaces, and it offers a detailed description of a
particular type of arrangement, namely arrangements of geodesic arcs on
the sphere. The presentation of general arrangements embedded on
parametric surfaces is based on a joint work with E.\ Berberich,
K.\ Melhorn, and R.\ Wein. Preliminary results of this work were first
published at the  {\em $23^{rd}$ European Workshop on Computational Geometry
(EWCG)}~\cite{bfhw-scmtd-07}. Improved results appeared in the proceeding
of the {\em $15^{th}$ Annual European Symposium on Algorithms
(ESA)}~\cite{bfhmw-scmtd-07}, and mature results were presented in a
manuscript~\cite{bfhmw-apsgf-09}. The discussion about particular
arrangements embedded on spheres is based on joint work with
O.\ Setter. Preliminary results of this work were first published at
the {\em $24^{th}$ European Workshop on Computational Geometry
(EWCG)}~\cite{fsh-eiaga-08}. A movie rendering the results of this work was
presented at the {\em $24^{th}$ Annual Symposium on Computational Geometry
(SoCG)}. The proceeding of this symposium contains the related
extended abstract~\cite{fsh-agas-08}. Mature results were presented in a
manuscript~\cite{bfhks-apsca-09}.

In Chapter~\ref{chap:mink-sums-construction} we present two different
complete, yet efficient, implementations of
output-sensitive algorithms\index{algorithm!output sensitive} to compute
the Minkowski sum of two polytopes in $\rrr$. We describe how the input
polytopes in polyhedral-mesh representation both methods accept
are converted into the corresponding internal representations unique to
each method. We provide the theoretical concepts both methods rely on,
and detailed descriptions specific to each method. The first method was
also published in the journal {\em Computer Aided
Design}~\cite{fh-eecms-07}. A preliminary version of this paper
appeared in the proceedings of the {\em $8^{th}$ Workshop on Algorithm
Engineering and Experimentation (Alenex'06)}~\cite{fh-eecms-06}. The
second method appeared in the
publications~\cite{fsh-eiaga-08,fsh-agas-08,bfhks-apsca-09}
mentioned above. We compare our Minkowski-sum constructions with other
methods that produce exact results, and provide a summary of a
performance analysis of our methods.

Chapter~\ref{chap:mink-sum-complexity} provides a tight bound on the
exact maximum complexity of Minkowski sums of $k$ polytopes in $\rrr$ in
terms of the number of facets of the polytope summands. It is based
on collaborative work with C.\ Weibel. The results of this work were introduced
in the proceedings of the {\em $23^{rd}$ annual Symposium on Computational
Geometry (SoCG)}~\cite{fhw-emcms-07}, and were accepted for publication
in the journal {\em Discrete and Computational Geometry}~\cite{fhw-emcms}.
We use this opportunity to thank
Shakhar Smorodinsky for fruitful discussions conducted while we were
investigating the bound above. The chapter provides a
proof of the upper bound, and establishes the lower bound through a
construction procedure.

Chapter~\ref{chap:assem_plan} introduces an exact implementation of an
efficient algorithm to obtain a partitioning motion given an assembly
of polyhedra in $\rrr$ --- a solution to a problem in the domain of
assembly planning. This application uses several types of operations
on arrangements of geodesic arcs embedded on the sphere as basic
blocks. In this context the chapter introduces exact 
implementations of additional applications that exploit geodesic arcs
embedded on the sphere, such as polyhedra central-projection and
Boolean set-operations applied to point sets embedded on the sphere
bounded by geodesic arcs. Great parts of this chapter are extracts
from a paper recently published in the {\em $8^{th}$
International Workshop on Algorithmic Foundations of Robotics
(WAFR)}~\cite{fh-papit-08}. Specific background of assembly planning
is provided at the beginning of the chapter.

We refer the reader to some ongoing research and future prospects and
conclude in Chapter~\ref{chap:conclusion}.

The software access-information along with some further design details
are provided in the Appendix.

\begin{savequote}[10pc]
\sffamily
A common mistake that people make when trying to design something
completely foolproof is to underestimate the ingenuity of complete fools.
\qauthor{Douglas Adams}
\end{savequote}
\chapter{Arrangements on Surfaces}
\label{chap:aos}
\newcommand{\cA}{\mathcal{A}}
\newcommand{\cC}{\mathcal{C}}
\newcommand{\cS}{\mathcal{S}}

\hyphenation{two-dimen-sional}
\begin{figure}[!b]
  \centerline{
    \begin{tabular}{ccc}
      \multicolumn{1}{p{124pt}}{
        \pspicture[](0,0)(4,4)
        \psset{unit=1cm,linewidth=1.0pt}
        \pscircle[linewidth=1pt](1,2){1}
        \pscircle[linewidth=1pt](2,1){1}
        \pscircle[linewidth=1pt](3,2){1}
        \pscircle[linewidth=1pt](2,3){1}
        \pscircle[linewidth=1pt](1.293,1.293){1}
        \pscircle[linewidth=1pt](2.707,1.293){1}
        \pscircle[linewidth=1pt](2.707,2.707){1}
        \pscircle[linewidth=1pt](1.293,2.707){1}
        \pscircle*[linecolor=red,linewidth=1.5pt](1.02,1.02){4pt}
        \psarc[linecolor=blue,linewidth=2pt](1.293,2.707){1}{181}{225}
        \pscustom[linewidth=0,fillstyle=solid,fillcolor=green]{
          \psarc(1,2){1}{45}{90}
          \psarc(1.293,2.707){0.97}{45}{90}
          \psarc(2,3){0.97}{135}{180}
          \psarc(2.707,2.707){1}{135}{180}
        }
        \psline[linecolor=green](1.04,3.02)(1.97,3.4)
        \endpspicture
      } &
      \multicolumn{1}{p{186pt}}{
        \pspicture[](0,0)(6,4)
        \psset{unit=1cm,linewidth=1.0pt}
        \psline[linewidth=1pt](1,0)(3.5,4)
        \psline[linewidth=1pt](3,0)(5,4)
        \psline[linewidth=1pt](0,1.8)(6,2.2)
        \psline[linewidth=1pt](0,3)(4.5,0)
        \psline[linewidth=1pt](0,3.5)(6,1.5)
        \psline[linewidth=1pt](2,4)(6,1)
        \psline[linewidth=1pt](0,2.5)(4.5,4)
        \pspolygon[linewidth=0.5pt](0,0)(6,0)(6,4)(0,4)
        \pscircle*[linecolor=red,linewidth=1.5pt](1.5,3){4pt}
        \psline[linecolor=blue,linewidth=2pt](3.02,3.23)(4.17,2.37)
        \pspolygon*[linecolor=green](2.05,1.65)(3.373,0.77)(4.015,2.05)(2.23,1.935)
        \endpspicture
      } &
      \includegraphics[height=120pt,keepaspectratio]{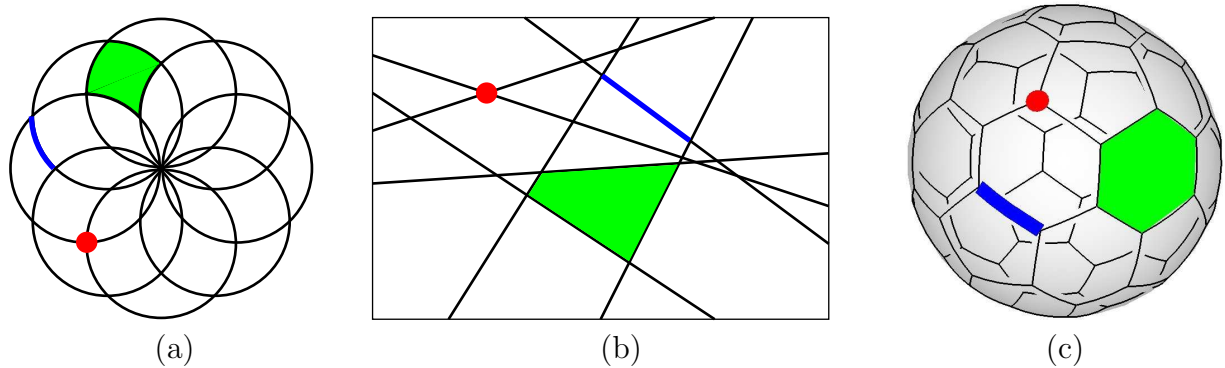}\\
      (a) & (b) & (c)
    \end{tabular}
  }
  \caption[Various types of arrangements]{\capStyle{(a) An arrangement
      of circles in the plane, (b) an arrangement of lines in the
      plane, and (c) an arrangement of geodesic arcs on the
      sphere.}}
  \label{fig:arr}
\end{figure}
Given a finite collection $\calC$ of geometric objects (such as
lines, planes, or spheres) the {\em \Index{arrangement}} $\calA(\calC)$
is the subdivision of the space where these objects reside into cells
as induced by the objects in $\calC$. In this thesis we deal only with
arrangements embedded on certain two-dimensional orientable parametric
surfaces in $\rrr$, i.e., planes, cylinders, spheres, tori, and surfaces
homeomorphic to them. In this case the objects in $\calC$ embedded on
the surface $S$ are curves that divide $S$ into a finite number of cells
of dimension $0$ ({\em vertices}), $1$ ({\em edges}) and $2$
({\em faces}). Figure~\ref{fig:arr} shows various types of arrangements
embedded on two-dimensional parametric surfaces.

The \aos{}\footnote{As a convention, \cgal{} prescribes the
suffix~{\tt \_2} for all data structures of planar objects and the
suffix~{\tt \_3} for all data structures of 3D objects. In the case of
arrangements on surfaces the suffix {\tt \_2} indicates the dimension
of the parameter space of the embedding surface.} package of
\cgal{}~\cite{cgal:wfzh-a2-07} is a
generic implementation of a complete software package that constructs
and maintains arrangements embedded on two-dimensional parametric
surfaces~\cite{fhktww-a-07,fwh-cfpeg-04,wfzh-aptac-07,bfhmw-scmtd-07,bfhmw-apsgf-09}
(see Section~\ref{sssec:intro:cgal:arr-history} for the history of the
package). As arrangements are ubiquitous in computational geometry and
have many theoretical and practical applications( see,
e.g.,~\cite{bkos-cgaa-00,as-aa-00,h-a-04}), many potential users in
academia and in industry can benefit from the \aos{} package.

As mentioned in Section~\ref{sssec:intro:geometric-programming:egc}
\cgal{} in general and the \aos{} package of \cgal{} in particular
follow the Exact Geometric Computation paradigm. The developed code
covers all cases to successfully handle degenerate input, while ideal 
computer arithmetic is emulated. While the \cgal{} package supports
arrangements induced by general algebraic curves of arbitrary degree,
the constructions and uses of arrangements described in this thesis
require only rational arithmetic. This is a key property that enables
efficient implementations of all the algorithms presented in the
thesis.

\cgal{} arrangements and their related components are the results of
ongoing collaborative research we were, and still are, deeply involved
with; see Section~\ref{sssec:intro:cgal:arr-history} for the evolution
of the arrangement package. We, in particular, had a significant
impact on specific topics within this research area as follows: We
came up with great parts of the geometry-traits concept hierarchy, and
the corresponding sets of minimal requirements (see
Section~\ref{sec:aos:geometry-traits}); we contributed the
geometry-traits module for polylines (see Section 
~\ref{ssec:aos:geometry-traits:models:polylines}) and the
geometry and topology traits modules for geodesic arcs embedded on the
sphere (see Section~\ref{sec:aos:geodesics}); we led the design of the
Boolean set-operation architecture (See
Section~\ref{ssec:aos:applications:bso}), and we provided solutions to
numerous issues encountered during the development of all these
components. This chapter provides a comprehensive overview of these
components, with a slight emphasis on asspects related to our work. 
It explains how arrangements induced by any type of curves
can be constructed, maintained, and used by other applications.
\section{Related Work}
\label{sec:aos:related-work}
Both closely related {\sc Mapc}~\cite{kcmk-eemapc-00}~\citelinks{mapc}
and {\sc Esolid}~\cite{ckfkm-esebe-04}~\citelinks{esolid} libraries
consist of an arrangement-construction module for
algebraic curves. However, these implementations make some
general-position assumptions on the input curves. The \leda{}
library~\cite{mn-lpcgc-00} includes geometric facilities that allow
the robust construction and maintenance of planar maps of line
segments that may contain degeneracies. However, the resulting
planar maps are represented as simple graphs that cannot fully
describe the topological structure of the arrangement. For example,
it is impossible to encode the containment relation between
disconnected components of the graph (i.e., to keep track of the
holes contained in a face; see
Section~\ref{ssec:aos:architecture:member-ops}). \leda-based
implementation of arrangements of conic curves and of cubic curves
were developed under the \exacus{}
project~\cite{behhk-exacus-05}~\citelinks{exacus}.

\cgal's arrangement package was the first complete software
implementation, designed for constructing arrangements of arbitrary
planar curves and supporting operations and queries on such
arrangements. The package was employed by many users to develop a
variety of applications in various domains. For example, it was used
to solve geometric optimization problems~\cite{r-vaasp-03,clkt-icncs-03},
to construct Minkowski sums\index{Minkowski sum}
efficiently~\cite{afh-pdecm-02,fh-eecms-07,fsh-agas-08},
to design snap-sounding algorithms~\cite{hp-isr-02}, to construct
envelopes of surfaces~\cite{m-rgece-06}. It was used in
\Index{motion planning}~\cite{hh-hymp-03,wbh-vvca-07},
\Index{assembly partitioning}~\cite{fh-papit-08},
\Index{cartography}~\cite{dhh-sesdtm-01}, and
several other applications~\citelinks{agilent, irit}.
The package was also used to compute arrangements of
quadrics~\cite{bhksw-eceic-05} by considering the planar arrangements
of their projected intersection curves. A better approach to compute such
arrangements~\cite{bfhmw-scmtd-07,bfhks-apsca-09} was introduced once the
package started to support arrangement on parametric surfaces, which also
enables the computation of arrangements embedded on Dupin
cyclides\index{Dupin cyclide}~\cite{bk-eatdc-08}. The torus, for example,
is a Dupin cyclide.

Sweeping the plane with a line is one of the most fundamental
algorithmic mechanisms in computational geometry. The \aos{} package
includes a generic implementation of an elaborate version of this
mechanism exploited by several higher-level operations supported by
the package.

The famous sweep-line\index{sweep line} algorithm of Bentley and
Ottmann~\cite{bo-arcgi-79} was originally formulated for sets of
non-vertical line segments, with the ``general position'' assumption
that no three segments intersect at a common point and no two segments
overlap. Many generalizations have been introduced ever since, such as
the ability to handle more general curves~\cite{sh-sac-89} and to deal
with degeneracies(see~\cite[Section~2.1]{bkos-cgaa-00}
and~\cite[Section~10.7]{mn-lpcgc-00} for a discussion about
degeneracies induced by line segments).

Effective algorithms for manipulating arrangements of curves have been
a topic of considerable interest in recent years with an emphasis on
exactness and efficiency of implementation~\cite{fhktww-a-07}.
Mehlhorn and Seel~\cite{ms-if-03} propose a general framework for
extending the sweep-line algorithm to handle unbounded curves; however,
their implementation can only handle lines in the plane. Arrangements on
spheres are covered by Andrade and Stolfi~\cite{as-eacs-01}, Halperin
and Shelton~\cite{hs-pssaa-98}, and recently Cazals and
Loriot~\cite{cl-ceacs-06}. Cazals and Loriot have developed a software
package that can sweep over a sphere and compute exact arrangements of
circles on it.

The \leda{} external package (LEP)
\ccode{SphereGeometry}~\citelinks{leda-SphereGeometry} handles geodesic
arcs on the sphere using an implicit representation, which enables the
use of exact rational arithmetic to handle objects of this type. The
package contains implementations of basic algorithms related to geodesic
arcs on a sphere, such as computing the spherical convex hull, the union
of two spherical polygons, and the width of a three-dimensional set of
points. It does not, however, support arrangements.

\section{Parametric Surfaces\index{surface!parametric}}
\label{sec:aos:parametric-surfaces}
A parametric surface $S$ is defined by a continuous function
$f_S: \parms \rightarrow \rrr$, where the domain
$\parms = U \times V$ is a rectangular two-dimensional parameter space
with bottom, top, left, and right boundaries, and the range
$S = f_S(\parms)$. $U$ and $V$ are open, half-open, or closed intervals
with endpoints in $\reals \cup \{-\infty, +\infty\}$.We use $\umin$, $\umax$,
$\vmin$, and $\vmax$ to denote the endpoints of $U$ and $V$, respectively. 
For example, the standard parameterization of the plane is
$f_S(u, v) = (u, v, 0)$, where $U = V =(-\infty, +\infty)$, and the unit
sphere is commonly parameterized as
$f_{S}(u, v) = (\cos u \cos v, \sin u \cos v, \sin v)$, where
$\parms = [-\pi, \pi] \times [-\frac{\pi}{2}, \frac{\pi}{2}]$.

A {\em \Index{contraction point}} $p \in S$ is a singular point, which
is the mapping of a whole boundary of the domain $\parms$. For example,
if the top boundary is contracted, we have
$\forall u \in U, f_S(u,\vmax) = p'$ for some fixed point $p' \in \rrr$.
An {\em \Index{identification curve}} $C \subset S$ is a continuous curve,
which is the mapping of opposite closed boundaries of the domain $\parms$.
If the left and right boundaries are identified, we have
$\forall v \in V, f_S(\umin,v) = f_S(\umax,v)$, (and similarly for the
bottom and top boundaries).
For example, consider the sphere as parameterized above. Its contraction
points are $(0, 0,\pm 1)$, as $f_S(u,-\frac{\pi}{2}) = (0, 0, -1)$ and
$f_S(u,\frac{\pi}{2}) = (0, 0, 1)$ for all $u$. Its identification curve
is $\{f_S(\pi,v)\,|\,-\frac{\pi}{2} \leq v \leq \frac{\pi}{2}\}$, as
$f_S(-\pi,v) = f_S(+\pi,v)$ for all $v$.

A {\em parameterizable curve} $\gamma$ is a continuous function
$\gamma: I \rightarrow \parms$ where $I$ is an open, half-open, or
closed interval with endpoints $0$ and $1$, and $\gamma$ is injective,
except for at a finite number of points. If $0 \not\in I$,
$\lim_{t \rightarrow 0+} \gamma(t)$ exists (in the closure of $\parms$)
and lies in an open side of the boundary. Similarly, if $1 \not\in I$,
$\lim_{t \rightarrow 1-} \gamma(t)$ exists and lies in an open side of
the boundary. A curve $C$ in $S$ is the image of a curve $\gamma$ in the
domain.

A curve is \emph{closed in the domain} if $\gamma(0) = \gamma(1)$; in
particular, $0 \in I$ and $1 \in I$. A curve is \emph{closed in the surface
$S$ (or simply closed)} if $f_S(\gamma(0)) = f_S(\gamma(1))$.
A curve $\gamma$ has two \emph{ends}, the 0-end $\langle \gamma,0 \rangle$
and the 1-end $\langle \gamma,1 \rangle$. If $d \in I$, the $d$-end has a
geometric interpretation. It is a point in $\parms$. If $d \not\in I$, the
$d$-end has no geometric interpretation. You may think of it as a point on
an open side of the domain or an initial or terminal segment of $\gamma$.
If $d \not\in I$, we say that the $d$-end of the curve is open.
Consider for example the equator curve on the sphere as parameterized
above. This curve is given by $\gamma(t) = (\pi(2t - 1),0)$, for
$t \in [0,1]$. The $0$-end of $\gamma$ is the point $(-\pi,0)$ in $\parms$
and a point on the equator of the sphere. It is closed on the sphere, but
non-closed in~$\parms$.  

A \emph{$u$-monotone curve} is the image of a curve $\gamma$, such that
if $t_1 < t_2$, then $u(\gamma(t_1)) < u(\gamma(t_2))$ for $t_1 < t_2$.
A \emph{vertical curve} is the image of a curve $\gamma$, such that
$u(\gamma(t)) = c$ for all $t \in I$ and some $c \in U$ and
$v(\gamma(t_1)) < v(\gamma(t_2))$ for $t_1 < t_2$. For instance, every
Meridian curve of a sphere parameterized as above is vertical.
A \emph{weakly $u$-monotone curve} is either vertical or 
$u$-monotone.\footnote{$u$-monotone curves refer to weakly $u$-monotone
curves hereafter.}

The \aos{} package handles inducing curves that are decomposable into
parameterizable weakly $u$-monotone curves as defined above. Any two
weakly $u$-monotone curves must intersect only a finite number of times
or overlap only in a finite number of sections, if at all. The curves
must be embedded on parameterizable surfaces as defined above. A curve can
be unbounded or reach the boundaries of the embedding surface.
A boundary may define a contraction point or an identification
curve.\footnote{We do not support surfaces that contain a contracted
identification curve.} We allow non-injectivity on the boundary, denoted
$\partial \parms$, and require bijectivity only in
$\parms \setminus \partial \parms$ (the interior of $\parms$). More
precisely, we require that $f_S(u_1,v_1) = f_S(u_2,v_2)$ and
$(u_1,v_1) \neq (u_2,v_2)$ imply $(u_1,v_1) \in \partial \parms$ and
$(u_2,v_2) \in \partial \parms$. Informally, we require that all geometric
operations defined in Section~\ref{sec:aos:geometry-traits} be applicable
on our curves.

Code reuse is maximized by generalizing the prevalent algorithms and their
implementations originally designed to operate on arrangements
embedded in the plane. The generalized code handles features embedded
in the modified surface
$\widetilde{S} : f_{\widetilde{S}} = f_S(u,v)\,|\,(u,v) \in \parms \setminus \partial \parms$ defined over the interior of the parameter space, where
identification curves, contraction points, and points at infinity are
removed. Specific code that handles unbounded features or features that
reach the boundaries is added to yield a complete implementation.

\section{The Arrangement Package Architecture}
\label{sec:aos:architecture}
The main class of the package, namely \aos{}, constructs and maintains
the embedding of a set of continuous weakly $u$-monotone curves that are
pairwise disjoint in their interiors on a two-dimensional parametric
surface in $\rrr$. It provides the necessary capabilities for maintaining
the embedded graph, while associating geometric data with the vertices,
edges, and faces of the graph. The embedded graph is represented using a
{\em doubly-connected edge list}
(\dcel\index{DCEL@\dcel})~\cite[Section~2.2]{bkos-cgaa-00}, which
maintains the incidence relations on its features~\cite{wfzh-aptac-07}.
Each edge of the subdivision is represented by two halfedges
with opposite orientation, and each halfedge is associated with the
face to its left. It is based on an implementation of a
halfedge data-structure \index{halfedge data-structure|see{HDS}}
(\hds\index{HDS@\hds})~\cite{cgal:k-hds-07} also used by the
polyhedral-surfaces package~\cite{k-ugpdd-99,cgal:k-ps-07}.

An important guideline in the design is to decouple the arrangement
representation from the various algorithms that operate on it. Thus,
the \aos{} class provides only a restricted set of methods for locally
modifying the arrangement; see
Section~\ref{ssec:aos:architecture:member-ops}. Non-trivial algorithms
that involve geometric operations are implemented as free (global)
functions that use the interface of the arrangement class;
see Section~\ref{sec:aos:facilities}. For example, the package offers
free functions for \emph{incremental} or {\em aggregated} insertion of
curves that may not necessarily be $u$-monotone, and the insertion
location of which are unknown {\em a priori}. Each input curve is
subdivides into several $u$-monotone subcurves before inserted using
one of the member methods listed in
Section~\ref{ssec:aos:architecture:member-ops}.

\subsection{The Data Structure}
\label{ssec:aos:architecture:data-structure}
The \ccode{Arrangement\_on\_surface\_2<GeometryTraits,TopologyTraits>}
class-template must be instantiated with two types as follows:
\begin{itemize}
\item
A geometry-traits\index{traits!geometry} class, which defines the
abstract interface between the arrangement data-structure and the
geometric primitives it uses. It is tailored to handle a specific
family of curves, and it encapsulates implementation details, such
as the number type used, the coordinate representation (i.e., Cartesian
or homogeneous), the algebraic computation methods, and auxiliary data
stored with the geometric objects, if present; see
Section~\ref{sec:aos:geometry-traits} for more details.
\item
A topology-traits\index{traits!topology} class, which adapts the
underlying \dcel{} to the embedding modified surface $\widetilde{S}$.
It determines whether the embedded surface is bounded, or otherwise
whether a boundary defines a contraction point or an identification
curve. If the inducing curves are confined to the modified
parameter space, the tasks of the topology-traits class are minimal.
However, in other cases it maintains the features that escape the
modified parameter space $\widetilde{\parms}$.
\end{itemize}

The underlying \dcel{} in turn associates a point with each vertex and
a $u$-monotone curve with each halfedge pair, where the geometric types
of the point and the $u$-monotone curve are defined by the
geometry-traits class. Users may extend the default \dcel{} data-structure,
in order to attach additional data to the \dcel{} records, as explained in
Section~\ref{ssec:aos:facilities:extension}. 

The \aoswh{} class-template represents an arrangement of general curves
embedded on a two-dimensional parametric surface, and maintains the
construction history of the arrangement. Input curves that induce the
arrangement are split into $u$-monotone subcurves that are pairwise
disjoint in their interior, and these subcurves are the embeddings of
the arrangement halfedges. While using the \aos{} class we lose track of
the connection between input curves and their final embeddings, in the
\aoswh{} data-structure each embedded $u$-monotone curve is extended to
store a pointer to the input curve associated with it, or a container of
curve pointers in case the embedded $u$-monotone curve is an overlapping
section of several input curves.

The \aoswh{} class is a simple {\em \Index{decorator}\/}\footnote{A
decorator attaches additional responsibilities to an object
dynamically~\cite{ghjv-dp-95}.} for\linebreak
\aos. It inherits from an
\aos{} class-template instantiated with a geometry-traits class that extends
the $u$-monotone curve type. It also stores the set of input curves, and
maintains a data structure that enables efficient traversal of all
halfedges induced by given input curves. The cross-pointers between input
curves and arrangement halfedges are maintained using an
{\em \Index{observer}}
(see Section~\ref{ssec:aos:facilities:notification}) that keeps track of
each change that involves an arrangement halfedge and updates its
underlying curve accordingly.

Users can traverse the original curves of each arrangement
halfedge, or iterate over all halfedges induced by a given input curve.
Tracing back the curve (or curves) that induced an arrangement edge is
essential in a variety of applications that use arrangements, such as
robot motion planning; see, e.g.,~\cite{hh-hymp-03}. It is possible, for
example, to efficiently remove a curve from the arrangement by deleting
all edges it induces.

Arrangements embedded in the plane are very common and, at least as far
as the arrangement package of \cgal{} is concerned, have a longer history
than their generalization for two-dimensional surfaces in $\rrr$. The
\arr{} class-template represents a planar subdivision. It maintains the
embedding of continuous weakly $u$-monotone curves in the $xy$ plane,
parameterized the natural way. That is, the two parameters $u$ and $v$ are
directly mapped to $x$ and $y$, respectively. Thus, $u$-monotonicity
implies $x$-monotonicity and vice versa. The \arr{} class-template
is parameterized with a geometry-traits class and with a \dcel{}
data-structure. It inherits from an \aos{} class-template instantiated with
the geometry traits template parameter and with a specific topology-traits
class suitable for the plane. The dedicated topology traits is instantiated
with the \dcel{}\index{DCEL@\dcel} template parameter. Similarly the
\arrwh{} class-template represents a planar subdivision, and maintains
the construction history of the arrangement.

The package offers various operations on arrangements stored in these
representations, such as point location, insertion of curves, removal of
curves, and overlay computation.

\subsection{Member Operations}
\label{ssec:aos:architecture:member-ops}
The interface of \aos{} consists of various methods that enable the
traversal of arrangement features. The class supplies iterators over
its vertices, halfedges, or faces. The classes \ccode{Vertex},
\ccode{Halfedge}, and \ccode{Face}, nested in the \aos{} class, supply
in turn methods for local traversals. For example, it is possible to
visit all halfedges incident to a specific vertex. Halfedges stored in
doubly-connected lists form chains. The chains define the inner and
outer connected components of the boundary
\index{components of the boundary|see{CCB}} (\Index{CCB}) of each face.
A bounded face in the \arr{} data structure has a single outer CCB
representing the outer boundary of the face, and may have several inner
CCBs representing its holes. However, a face in the general \aos{} data
structure may have several inner and outer CCBs; see
Section~\ref{sec:aos:geodesics}. Naturally, it is possible to traverse
all the halfedges along the inner and outer boundaries of a given face.

\begin{wrapfigure}[19]{r}{4.9cm}
  \vspace{-15pt}
  \setlength{\tabcolsep}{3pt}
  \begin{tabular}{cc}
    \epsfig{figure=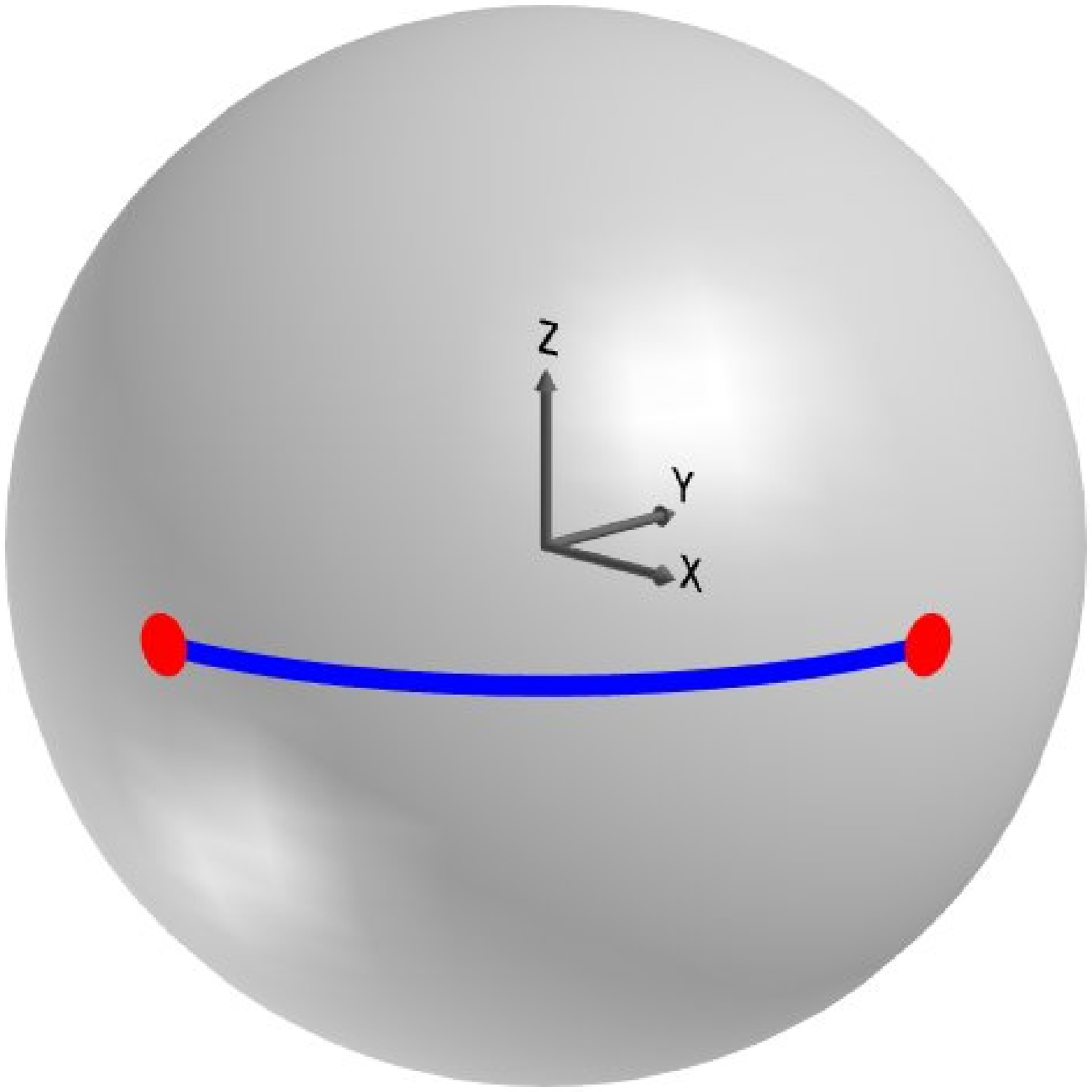,width=2.35cm,silent=} &
    \epsfig{figure=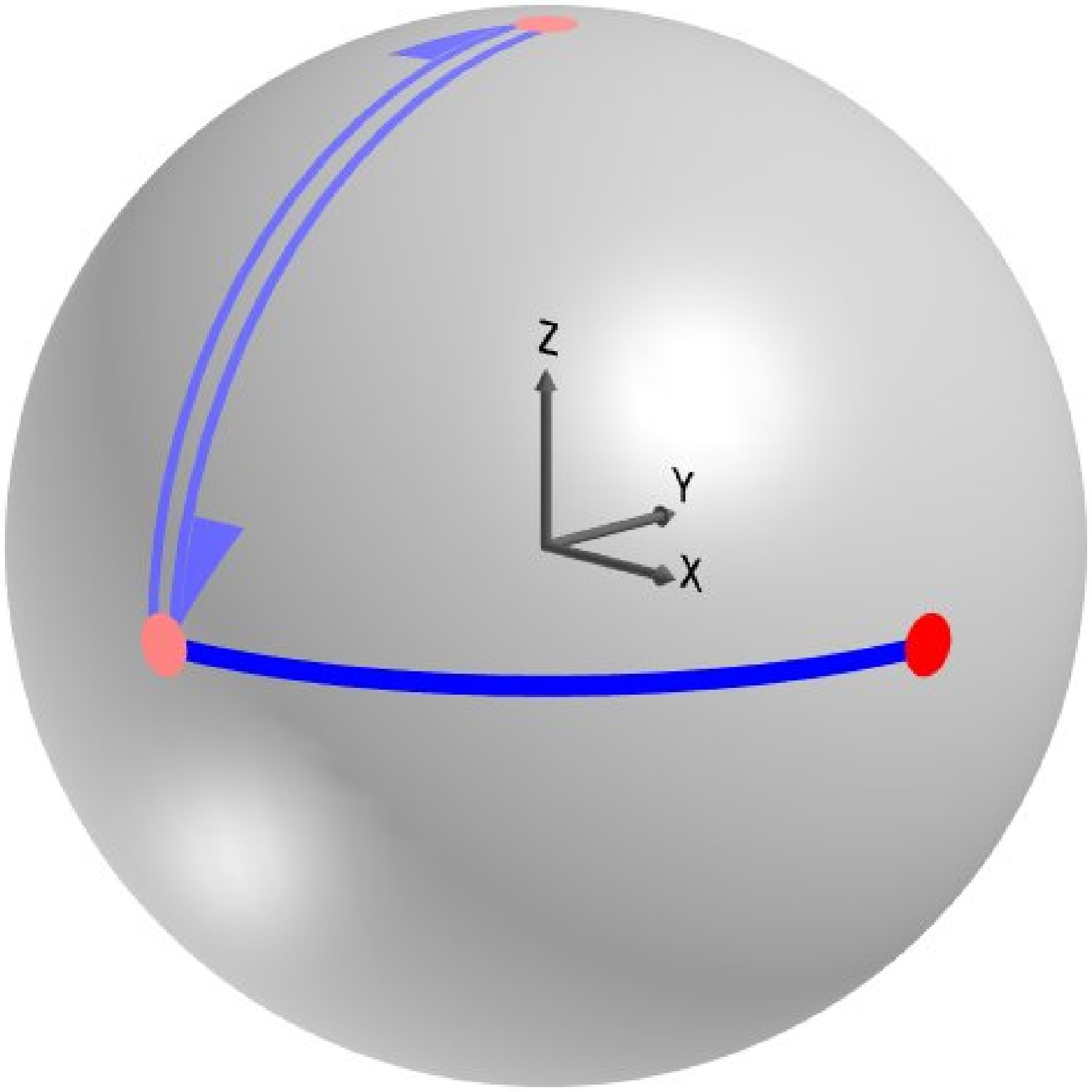,width=2.35cm,silent=}\\
    (a) & (b)\\
    \epsfig{figure=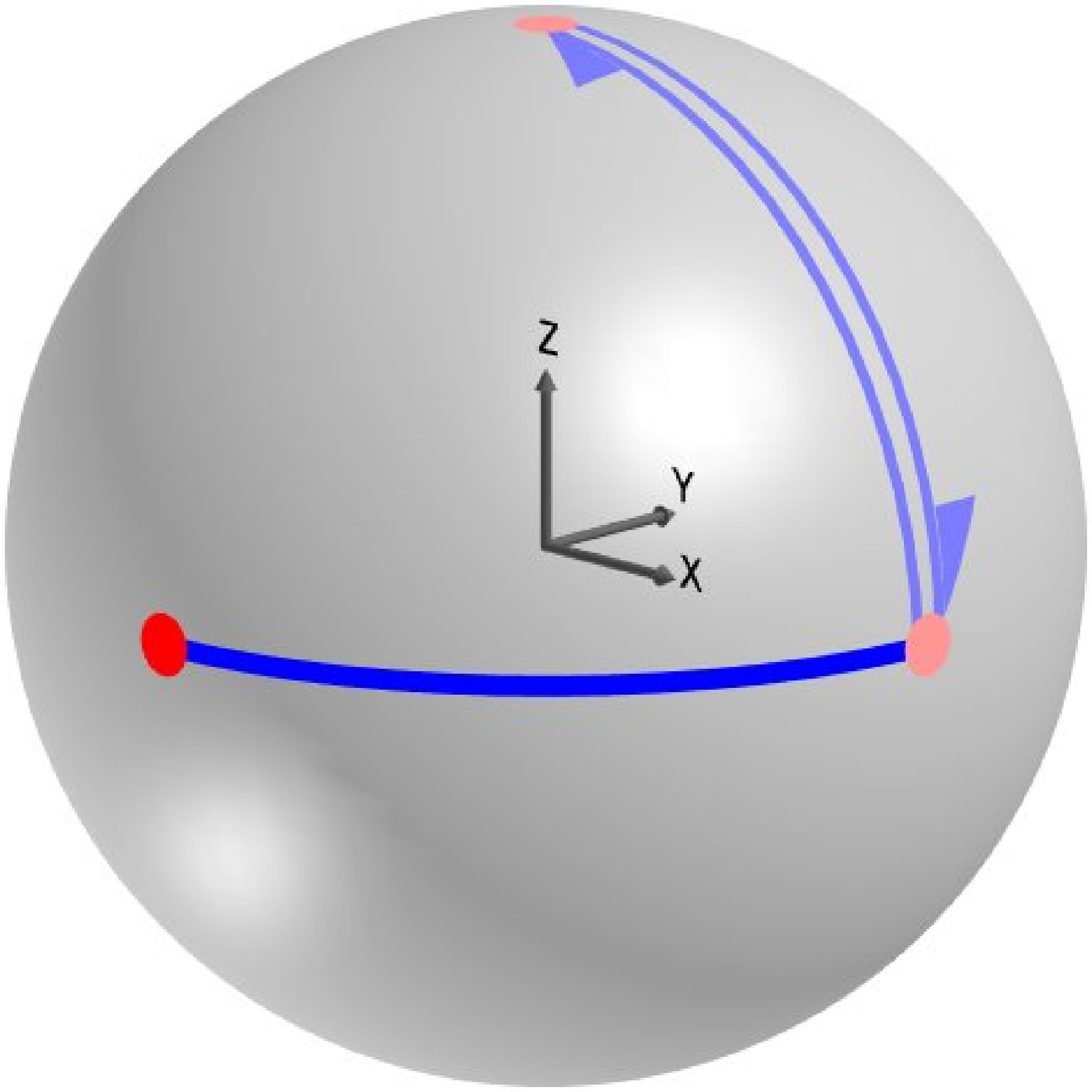,width=2.35cm,silent=} &
    \epsfig{figure=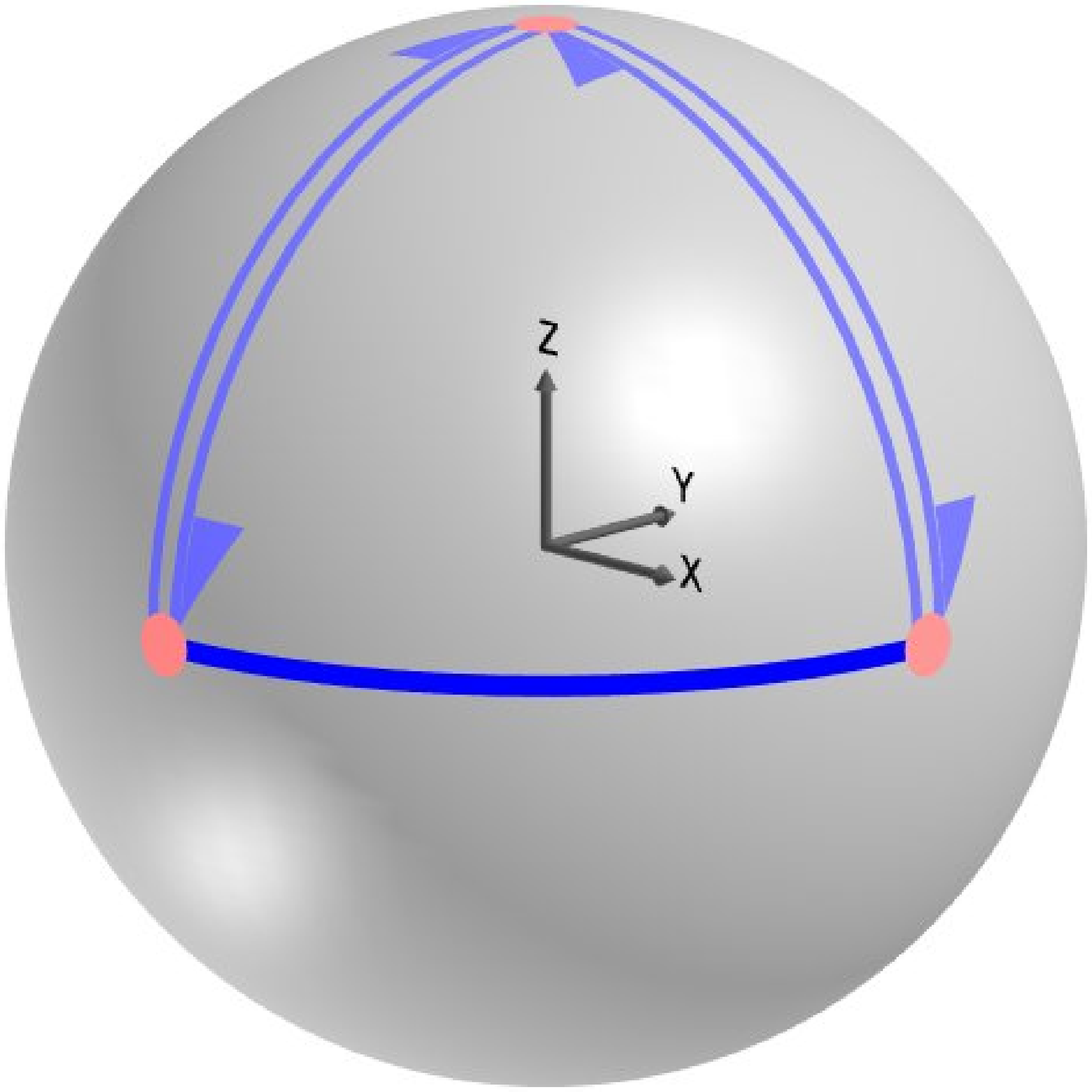,width=2.35cm,silent=}\\
    (c) & (d)
  \end{tabular}
  \caption[The arrangement immediate insertion methods]
  {\capStyle{The arrangement immediate insertion methods. The newly
      inserted curve is drawn in bright blue. Vertices created as a result
      of the insertion are drawn in bright red.}}
  \label{fig:insertion-fncs}
\end{wrapfigure}
\newcounter{aos:arch:cntr}
Alongside with the traversal methods, the arrangement class also
supports several methods that modify the arrangement. The
functions
\ccode{insert\_in\_face\_interior(C,f)} (Figure~\ref{fig:insertion-fncs} (a)),
\ccode{insert\_from\_left\_vertex(C,v)} (Figure~\ref{fig:insertion-fncs} (b)),
\ccode{insert\_from\_right\_vertex(C,v)}) (Figure~\ref{fig:insertion-fncs} (c)),
and
\ccode{insert\_at\_vertices(C,v1,v2)} (Figure~\ref{fig:insertion-fncs} (d))
create an edge that
corresponds to a $u$-monotone curve $C$, whose interior is disjoint
from existing edges and vertices. The choice of which one to use
depends on whether the curve endpoints are associated with existing
non-isolated arrangement vertices:
\setcounter{aos:arch:cntr}{1}(\roman{aos:arch:cntr})~If both curve endpoints
do not correspond to any existing vertex,
\ccode{insert\_in\_face\_interior()} is used to generate a new hole
inside an existing face.
\addtocounter{aos:arch:cntr}{1}(\roman{aos:arch:cntr})~If exactly one
endpoint corresponds to an existing \dcel{} vertex, one of the
functions \ccode{insert\_from\_left\_vertex()} or
\ccode{insert\_from\_right\_vertex()} is called, depending on which
endpoint is associated with an existing vertex. It forms an
``antenna'' emanating from an existing connected component.
\addtocounter{aos:arch:cntr}{1}(\roman{aos:arch:cntr})~Otherwise, both
endpoints correspond to existing vertices, and
\ccode{insert\_at\_vertices()} is called to connect these vertices
using a pair of twin halfedges. These functions hardly involve
any geometric operations, if at all.\footnote{Unless we
force checking preconditions. In this case the precondition
evaluation involves geometric computation.} They accept topologically
related parameters, and use them to operate directly on the \dcel{}
records, thus saving algebraic operations, which are especially
expensive when high-degree curves are involved.

Other modification methods included in the arrangement class enable
users to split an edge into two, to merge two edges incident to a
common vertex, and to remove an edge from the arrangement. It is
also possible to insert a point in the interior of a given face,
creating an isolated vertex that corresponds to this point, or to
remove an isolated vertex from the arrangement.

\subsection{Cell Extension}
\label{ssec:aos:facilities:extension}
As mentioned above the \aos{} is parameterized by a topological traits,
which in turn is parameterized by a \dcel{} class. Users may extend
the default \dcel{} data-structure, in order to attach additional
data to the \dcel{} records. The default \dcel{} model simply
associates a point with each \dcel{} vertex and a $u$-monotone curve
with each halfedge pair. Although it is possible to store auxiliary data
with the curves or points by extending their respective types (see
Section~\ref{ssec:aos:geometry-traits:extension}), it is sometimes
necessary to extend the vertex, halfedge, or face topological features
of the \dcel{}. Many times it is desired to associate extra data with
the arrangement faces only. For example, when an arrangement
represents the subdivision of a country into regions associated with
their population density. In this case there is no alternative other
than to extend the \dcel{} face, as there is no geometry-traits class
entity that corresponds to an arrangement face. A similar mechanism
for extending topological features with auxiliary attributes can be
found in other components of \cgal, such as the triangulation
packages~\cite{cgal:py-tds2-07} and the polyhedral-surfaces
package~\cite{cgal:k-ps-07}.

\section{The Arrangement Facilities}
\label{sec:aos:facilities}
\subsection{Sweep Line}
\label{ssec:aos:facilities:sweep-line}
The \aos{} package offers a generic implementation of the sweep-line
\index{sweep line} algorithm~\cite[Section~2.1]{bkos-cgaa-00} in form
of a class template called \cSweepline{}. It handles any set of
arbitrary $u$-monotone curves, and serves as the foundation of a family
of concrete operations, such as computing all intersection points induced
by a set of curves, constructing an arrangement of curves, aggregately
inserting a set of curves into an existing arrangement, and computing
the overlay of two arrangements. A concrete algorithm is realized
through a sweep-line visitor, a template parameter of \cSweepline{},
which follows the {\em \Index{visitor}}
design-pattern~\cite{ghjv-dp-95}, and models the concept
\concept{SweepLineVisitor\_2}. In this case, a visitor defines an
operation based on the sweep-line algorithm to be performed on an
arrangement without the need to change the arrangement structure.
Users may introduce their own sweep based algorithms by implementing an
appropriate visitor class.\footnote{The \boost{} Graph Library, for
example, uses visitors~\cite[Section~12.3]{sll-bgl-02} to support
user-defined extensions to its fundamental graph algorithms.}

Another parameter of the \cSweepline{} class-template is the
geometry-traits class, which must be instantiated with a model of the
\concept{ArrangementXMonotoneTraits\_2} concept; see
Section~\ref{sec:aos:geometry-traits} for the precise definition of
this concept. It defines the minimal set of geometric primitives,
among the other, required to perform the sweep-line algorithm briefly
described next.

An imaginary vertical curve is swept over the surface from left to right,
transforming the static two-dimensional problem into a dynamic
one-dimensional one. At each time during the sweep a subset of the input
$u$-monotone curves intersect this vertical line in a certain order. The
subset of curves and their order along the sweep line may change as the
line moves along the $u$-axis, implying a change in the topology of the
arrangement, only at a finite number of \emph{event points}, namely
intersection points of two curves and left endpoints or right endpoints
of arcs of curves. The event points, namely endpoints and all the
intersection points that have already been discovered, are stored in a
$uv$-lexicographic order in a dynamic event queue, named the
\emph{$U$-structure}. The ordered sequence of segments intersecting the
imaginary vertical line is stored in a dynamic structure called the
\emph{$V$-structure}. Both structures are maintained as balanced binary
trees that enable their efficient maintenance using an advanced
implementation of red-black trees~\cite{w-eirbt-05}.

During the sweep-line process the event objects in the $U$-structure
are sorted lexicographically, and the subcurve objects are stored in
the $V$-structure in the same order as the lexicographic order of
their intersection with the imaginary sweep-line. The \cSweepline{}
class performs only the operations required to maintain the $U$-structure
and the $V$-structure, while the visitor class is responsible for
producing the actual output of the algorithm. Whenever the sweep-line
class handles an event point $p$, it sends a notification to its visitor.
Using this information, the visitor can access the subcurves incident to
$p$ and the neighboring subcurves from above and below.

\subsection{Map Overlay}
\label{ssec:aos:facilities:overlay}
The map overlay\index{map overlay|see{overlay}} of two planar
subdivisions $\calS_1$ and $\calS_2$ is a planar subdivision $\calS$,
such that there is a face $f$ in $\calS$ if and only if there are
faces $f_1$ and $f_2$ in $\calS_1$ and $\calS_2$ respectively, such
that $f$ is a maximal connected subset of
$f_1 \cap f_2$~\cite[Section~2.3]{bkos-cgaa-00}. The \Index{overlay}
of two two-dimensional subdivisions embedded on a surface is defined
similarly.

The overlay of two given arrangements, conveniently referred to as
the ``blue'' and the ``red'' arrangements, is straightforwardly
implemented using a sweep-line visitor. A consolidated set of the
``blue'' and ``red'' curves is processed, while the imaginary vertical
line is swept over the surface. The $u$-monotone curve type is extended
with a color attribute (whose value is either \ccode{BLUE} or \ccode{RED});
see Section~\ref{ssec:aos:geometry-traits:extension}. Using the
extended type we filter out unnecessary computations. For example, we
ignore monochromatic intersections, and compute only red--blue
intersection points (or overlaps). This way the arrangement of a
consolidated set of ``blue'' and ``red'' curves is computed efficiently.

The overlay visitor needs to construct a \dcel{} that properly
represents the overlay of two input arrangements, the \dcel{}'s of
which are potentially independently extended (see
Section~\ref{ssec:aos:facilities:extension}). A face in the overlay
arrangement corresponds to overlapping regions of the blue and red
faces. An edge in the overlay arrangement is due to a blue edge, a
red edge, or an overlap of two differently colored edges. An overlay
vertex is due to a blue vertex, a red vertex, a coincidence of two
differently colored vertices, or an intersection of a blue and a
red curve. In each case, the data associated with the 
overlay \dcel{} feature should be computed from the red and blue
\dcel{} features that induce it. To this end, the overlay visitor is
parameterized by an overlay-traits module, which defines the merging
operations between various \dcel{} features, achieving maximum
genericity and flexibility for the users. The instantiated overlay
traits models the \concept{OverlayTraits} concept. The concept
requires the provision of ten functions that handle all possible cases
as follows:
\begin{compactenum}
\item A new vertex $v$ is induced by coinciding vertices $v_r$ and $v_b$.
\item A new vertex $v$ is induced by a vertex $v_r$ that lies on an edge $e_b$.
\item An analogous case of a vertex $v_b$ that lies on an edge $e_r$.
\item A new vertex $v$ is induced by a vertex $v_r$ that is contained in a 
  face $f_b$.
\item An analogous case of a vertex $v_b$ contained in a face $f_r$.
\item A new vertex $v$ is induced by the intersection of two edges $e_r$ and 
  $e_b$.
\item A new edge $e$ is induced by the overlap of two edges $e_r$ and $e_b$.
\item A new edge $e$ is induced by the an edge $e_r$ that is contained in a 
  face $f_b$.
\item An analogous case of an edge $e_b$ contained in a face $f_r$.
\item A new face $f$ is induced by the overlap of two faces $f_r$ and $f_b$.
\end{compactenum}
We apply the overlay operations in four different ways in this
thesis; see
Sections~\ref{ssec:mscn:sgm:mink_sum}, \ref{ssec:mscn:cgm:mink_sum},
\ref{ssec:assem_plan:pairwise-ms-projection}, and
\ref{ssec:assem_plan:motion-space-construction} for the different
applications. Each application requires the provision of a different
set of the ten functions above.

\subsection{Zone Construction}
\label{ssec:aos:facilities:zone}
The \Index{zone}~\cite{h-a-04} of a $u$-monotone curve $C$ in an
arrangement $\calA$ is the set of cells of $\calA(\calC)$ intersected
by the curve $C$.

The \aos{} package includes the \cZone{} class-template, which
computes the zone of an arrangement. Similar to the \cSweepline{}
template, the \cZone{} template is parameterized with a zone visitor,
a model of the concept \concept{ZoneVisitor\_2}, and it serves as
the foundation of a family of concrete operations, such as inserting
a single curve into an arrangement and determining whether a query
curve intersects with the curves of an arrangement.

The zone of a curve $C$ is computed by locating the left endpoint of
$C$ in the arrangement, and then ``walking'' along the curve towards
the right endpoint, keeping track of the vertices, edges, and faces
crossed on the way (see, for example,~\cite[Section~8.3]{bkos-cgaa-00}
for the computation of the zone of a line in an arrangement of lines).

It is sometimes necessary to compute the zone of a curve in an
arrangement without actually inserting the curve. In other situations,
the entire zone is not required, as in the case of a process that only
checks whether a query curve passes through an existing arrangement
vertex; if the answer is positive, the process can terminate as soon
as the vertex is located. While the sweep-line algorithm operates on a
set of input $u$-monotone curves and its visitors can just use the
notifications they receive to construct their output structures, the
zone-computation algorithm operates on an arrangement object and its
visitors may modify the same arrangement object as the computation
progresses. This makes the interaction of the main class with its
visitors slightly more intricate.

\subsection{Observers}
\label{ssec:aos:facilities:notification}
Some arrangement-based algorithms and applications should be bound
to a specific arrangement instance and receive notifications on
various topological changes this arrangement undergoes. This is not
just a convenience, but crucial to the usability of the package, as
it might be the only way for providing an algorithm with a certain
input, such as data that should be bound to the topological features
of the arrangement, and is available only during construction; see
Section~\ref{ssec:aos:facilities:extension} for an example.

The \aos{} package supports a notification mechanism, which follows
the {\em \Index{observer}} design-pattern~\cite{ghjv-dp-95}. In
this case of one-to-many dependency a set of observes depend on a
single arrangement, so that when the arrangement changes state, all
its dependents are notified and updated automatically. Using this
mechanism it is possible to attach any number of observer instances
to a specific arrangement, such that all attached observers get
notified on local and global changes the arrangement undergoes.

The \ccode{Arr\_observer<Arrangement>} class-template, parameterized
by an arrangement type, stores a pointer to an arrangement object,
and is capable of receiving notifications just before a structural
change occurs in the arrangement and immediately after such a change
takes place. Hence, each notification comprises of a pair of ``before''
and ``after'' functions (e.g., \ccode{before\_split\_face()} and
\ccode{after\_split\_face()}). The \cObserver{} class-template
serves as a base class for other observer classes and defines a set of
virtual notification functions, giving them all a default empty
implementation. The interface of the base class is designed to capture
all possible changes that arrangements can undergo, with a minimal set
of topological events.

The set of functions can be subdivided into three categories as follows:
\begin{enumerate}
\item
Notifiers of changes that affect the entire topological structure.
Such changes occur when the arrangement is cleared or when it is
assigned with the contents of another arrangement.
\item
Notifiers of a {\em local} change to the topological structure, such
as the creation of a new vertex or an edge, the splitting of an edge
or a face, the formation of a new hole inside a face, the removal of
an edge, etc.
\item\label{item:glob_notif}
Notifiers of a {\em global} change initiated by a free (global)
function, and called by the free function (e.g., incremental or
aggregate insert; see Section~\ref{sec:aos:architecture}). This
category consists of a single pair of notifiers, neither of them is
called by methods of the \aos{} class-template itself.
Issuing point-location queries (or any other queries for that matter)
between the calls to the ``before'' and ``after'' functions of this pair
is forbidden.\footnote{This constraint improves the efficiency of
the maintenance of auxiliary data structures for the relevant
point-location strategies, which have to update their data structures
according to the changes the arrangement undergoes (see 
Section~\ref{ssec:aos:facilities:point-location} for more
details). Since no point-location queries are issued between the
invocation of \ccode{before\_global\_change()} and
\ccode{after\_global\_change()}, it is not necessary to perform an 
update each time a local topological change occurs, and it is
possible to postpone the updates until after the global operation is
completed.}
\end{enumerate}
See~\cite{wf-ndcap-05} for a detailed specification of the arrangement
observer class sketched above.

Each arrangement object stores a list of pointers to \cObserver{}
objects, and whenever one of the structural changes listed in the
first two categories above is about to take place, the arrangement
object invokes the appropriate function of each of its observers. It
also does so immediately after the change has taken place. In addition,
a free function may choose to trigger a similar notification, which
falls under the third category above.

In case the new observer is attached to a non-empty arrangement,
its constructor may extract the relevant data from the non-empty
arrangement using various traversal methods offered by the public
interface of the \aos{} class, and update any internal data stored in
the observer. This is necessary, for example, in case of the
non-stateless point-location strategies, as shown in the next section.

\subsection{Point Location}
\label{ssec:aos:facilities:point-location}
Point location\index{point location} is defined as follows: Given a
point, find the arrangement cell that contains it. The \aos{} package
provides the means to answer this query. Typically, the result of the
point-location query is one of the arrangement faces, but in degenerate
situations the query point can lie on an edge, or it may coincide with
a vertex. Since the arrangement representation is decoupled from the
algorithms that operate on it, the \aos{} class does not support
point-location queries directly. Instead, the package provides a set of
classes that are capable of answering such queries, all are models of
the concept \concept{ArrangementPointLocation\_2}. Each model employs a
different algorithm or \emph{strategy} for answering queries. A model of
this concept must define the \ccode{locate()} function, which accepts an
input query point and returns an object representing the arrangement
cell that contains this point (a polymorphic \ccode{CGAL::Object}
instance that can either be a \ccode{Face\_handle}, a
\ccode{Halfedge\_handle}, or a \ccode{Vertex\_handle}).

The following models of the concept \concept{ArrangementPointLocation\_2}
are included in the\linebreak
\aos{} package.
\begin{itemize}
\item
\ccode{Arr\_naive\_point\_location} locates the query point
naively, by exhaustively scanning all arrangement cells. It is the
only strategy with unlimited support; see Section~\ref{sec:aos:geodesics}.
\item
\ccode{Arr\_walk\_along\_a\_line\_point\_location} simulates a
reverse traversal along an imaginary vertical ray emanating from the
query point toward infinity. It starts from the unbounded face of
the arrangement and moves downward toward the query point until it
locates the arrangement cell containing it.
\item
\ccode{Arr\_landmarks\_point\_location<Generator>} uses an auxiliary
generator class to create a set of ``landmark'' points, whose
location in the arrangement is known. Given a query point, it uses a
nearest-neighbor search structure (e.g., \kdtree) to find the
nearest landmark, and then traverses the straight-line segment
connecting this landmark to the query point.\footnote{The
``landmarks'' strategy, requires that the arrangement is
instantiated with a geometry-traits class that models the
\concept{ArrangementLandmarksTraits\_2} concept, which adds two
requirements to the basic \concept{ArrangementBasicTraits\_2}
concept: \setcounter{aos:arch:cntr}{1}(\roman{aos:arch:cntr}) Approximating
the coordinates of a given point $p$ using the double-precision
arithmetic, and \addtocounter{aos:arch:cntr}{1}(\roman{aos:arch:cntr})
constructing a $u$-monotone curve that connects two given points
$p$ and $q$, where $p$ represents a landmark point and $q$ is the
query point. Most traits classes included in the arrangement package
are models of this refined concept.} See~\cite{hh-esplp-08} for more
details.
\item
\ccode{Arr\_trapezoidal\_ric\_point\_location} implements Mulmuley's
point-location algorithm~\cite{m-fppa-90}, which is based on the
vertical decomposition of the arrangement into pseudo-trapezoids, and
maintains a history directed acyclic graph\index{graph!directed
  acyclic} (DAG) on top of the decomposition.
\end{itemize}
The last two strategies have query times that are shorter than the
query times of the first two. However, they require preprocessing and
consume more space, as they maintain auxiliary data structures. The
first two strategies do not require any extra data and operate directly
on the \dcel{} that represents the arrangement. For a complete survey
see~~\cite{hh-esplp-08}.

Each of the ``landmarks'' point-location class and the trapezoidal
point-location class uses an observer to receive notifications
whenever the arrangement is modified. For example, the default
generator employed by the ``landmarks'' strategy uses the arrangement
vertices as landmarks, so whenever a new vertex is created (by the
insertion of a new edge, by the splitting of an existing edge, or by
the insertion of an isolated point), it should be inserted into the
nearest-neighbor search structure maintained by the respective
landmark class. The usage of the notification mechanism makes it
possible to associate several point-location objects with the same
arrangement simultaneously.

The ``landmarks'' and the trapezoidal point-location strategies are
both characterized by very efficient query time at the cost of
time-consuming preprocessing. Naturally, these strategies exhibit
better overall performance when the number of arrangement updates is
relatively small compared to the number of issued queries. For a
report on extensive experiments with the various point-location
strategies see~\cite{hh-esplp-08}.

\section{Geometry-Traits Concepts}
\label{sec:aos:geometry-traits}
\begin{wrapfigure}[15]{r}{10.9cm}
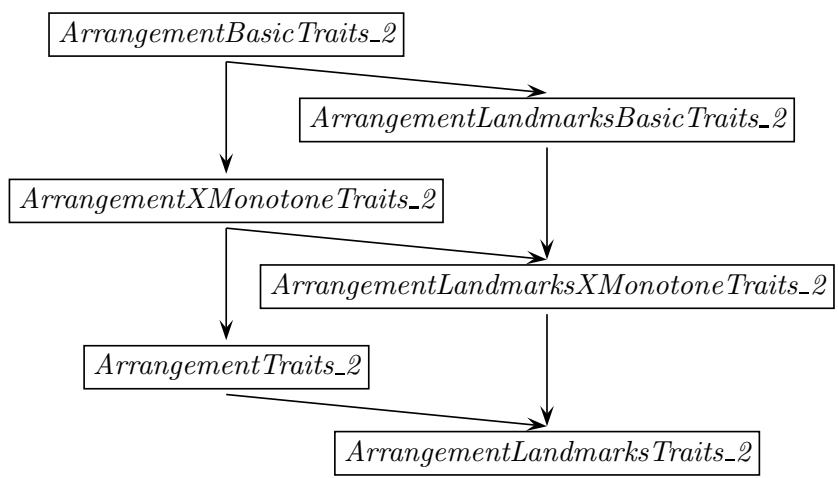

  \vspace{-12pt}
  \psset{treevshift=0,unit=1em,xunit=2em,yunit=1em,everytree={},
        etcratio=.75,triratio=.5}
     \jtree[everylabel=\sl,xunit=60pt,arrows=->]
     \! = {\ArrangementBasicT}
           <vert>[scaleby=0 3.6]{\ArrangementXMonotoneT}@axt !axt
          ^<right>[scaleby=2 1]{\ArrangementLandmarksBasicT}@lbt !lbt
           .
     \!axt = <vert>[scaleby=0 3.6]{\ArrangementT}@at !at
 	   .
     \!lbt = <vert>[scaleby=0 3.6]{\ArrangementLandmarksXMonotoneT}@lxt !lxt
           .
     \!lxt = <vert>[scaleby=0 3.6]{\ArrangementLandmarksT}@alt !alt
           .
     \ncline{axt:b}{lxt:t}
     \ncline{at:b}{alt:t}
     \endjtree
  \psset{treevshift=0,unit=1cm,xunit=1cm,yunit=1cm,everytree={},
         etcratio=.75,triratio=.5}
  \caption[Geometry-traits concept basic refinement
    hierarchy]{\capStyle{Refinement hierarchy of geometry-traits concepts.}}
  \label{fig:traits-hierarchy}
\end{wrapfigure}
The implementations of the various algorithms that construct and
manipulate arrangements are generic, as they are independent of the
type of curves they handle. All steps of the algorithms are enabled
by the minimal set of geometric primitives gathered in the
geometry-traits class, a model of a geometry-traits concept.
The geometry-traits concept is factored into a hierarchy of refined
concepts. The refinement hierarchy is defined according to the
identified minimal requirements imposed by different algorithms that
operate on arrangements, thus alleviating the production of traits
classes that handle complicated curves, and increasing the usability
of the algorithms. The requirements listed by the geometry-traits
concepts include only the utterly essential types and operations, and
fully specify all the preconditions that the input must satisfy, as
these may simplify the implementation of models of this concept even
further.

The following sections are dedicated to a detailed description of the
hierarchy. We list the minimal requirements of each layer in the
hierarchy, and provide formal definitions for the required operations.
The letters $x$ and $y$ are used in the code to refer to the two surface
parameters, as arrangements embedded in the $xy$-plane are more common
and familiar. The names of required nested types (e.g.,
\ccode{X\_monotone\_curve\_2}) and valid expressions (e.g.,
\ccode{compare\_x}) are faithful to the original source code. However,
we use the letters $u$ and $v$ in the formal definitions below to refer
to the two surface parameters, as these definitions apply to the general
case of arrangements embedded on surfaces. Let $\mathrm{cmp}_u()$ and
$\mathrm{cmp}_v()$ denote two predicates that accept two points and
compare them by their $u$-coordinates and by their $v$-coordinates
respectively. We use the following notation. For a point $p$, $(u_p,v_p)$,
denotes a pre-image, and for a curve $C$, $\gamma$ denotes a pre-image,
that is, $p = f_S(u_p,v_p)$ and $C(t) = f_S(\gamma(t))$ for all $t \in I$.

The basic concept \ArrangementBasicTraits{} requires the definition of
the types \ccode{Point\_2} and \ccode{X\_monotone\_curve\_2}. The latter
represents a $u$-monotone curve, and the former is the type of the
endpoints of the curves, representing a point on the surface. This
concept lists the minimal set of predicates on objects of these two types
sufficient to enable the operations provided by the \aos{}
class-template itself, namely the insertion of bounded $u$-monotone
curves that are interior disjoint from any vertex and edge in the
arrangement. All points and curves in the set below are required to
have an inverse pre-image in $\parms \setminus \partial \parms$. In
particular all curves are $u$-monotone.\\
\rule{\textwidth}{1pt}
\begin{compactdesc}
\item[\ccode{Compare\_x\_2}:]
  Compare two points by their $u$-coordinates.
\item[\ccode{Compare\_xy\_2}:]
  Compare two points lexicographically by their $u$ and then by their
  $v$-coordi\-nates.
\item[\ccode{Construct\_min\_2}:]
  Return the lexicographically smaller (left) endpoint of a given curve.
\item[\ccode{Construct\_max\_2}:]
  Return the lexicographically larger (right) endpoint of a given curve. 
\item[\ccode{Is\_vertical\_2}:]
  Determine whether a weakly $u$-monotone curve is vertical.
\item[\ccode{Compare\_y\_at\_x}:]
  Given a point $p$ and a curve $C$, such that the $u_p$
  lies in the $u$-range of $C$, determine whether $p$ is above, below,
  or lies on $C$. More precisely, if $C$ is vertical, determine
  whether $p$ lies on $C$, or above or below $C$. Otherwise,
  since $u(\gamma(0)) \leq u_p \leq u(\gamma(1))$ must hold and $C$ is
  $u$-monotone, there must be a unique $0 \leq t' \leq 1$, that
  satisfies $u(\gamma(t')) = u_p$. Return $\mathrm{cmp}_v(p, \gamma(t'))$.
\item[\ccode{Compare\_y\_at\_x\_right}:]
  Given two curves $C_1$ and $C_2$ that share a common left endpoint $p$,
  determine the relative position of the two curves immediately to the
  right of $p$. More precisely, return
  $\mathrm{cmp}_v(\gamma_1(\epsilon_1),\gamma_2(\epsilon_2))$, where
  $\epsilon_1,\epsilon_2 > 0$ are infinitesimally small.
\item[\ccode{Compare\_y\_at\_x\_left}:]
  Given two curves $C_1$ and $C_2$ that share a common right endpoint $p$,
  determine the relative position of the two curves immediately to the
  left of $p$. More precisely, return
  $\mathrm{cmp}_v(\gamma_1(1-\epsilon_1),\gamma_2(1-\epsilon_2))$, where
  $\epsilon_1,\epsilon_2 > 0$ are infinitesimally small. This is an
  {\em optional} requirement with ramifications in case it is not
  fulfilled; see Section~\ref{ssec:aos:geometry-traits:adaptor}.
\end{compactdesc}
\rule[5pt]{\textwidth}{1pt}
The set of predicates listed above is also sufficient for answering
point-location queries by the various point-location strategies,
with a small exception of the ``landmarks'' strategy, which requires
a traits class that models the refined concept
\ArrangementLandmarksTraits. This is described in
Section~\ref{ssec:aos:facilities:point-location}.

Constructing arrangements induced by $u$-monotone curves that may
intersect in their interior, requires an arrangement instantiated
with a traits class that models the concept
\ArrangementXMonotoneTraits{}. This concept refines the
basic arrangement-traits concept described above, as it requires
an additional method for computing intersections between $u$-monotone
curves, among the other. An intersection point between two curves is
also represented by the \ccode{Point\_2} type. The refined traits
concept also requires methods for splitting curves at these
intersection points to obtain pairs of interior disjoint subcurves
and merging pairs of subcurves. In summary, a model of the refined
concept must provide the additional operations bellow. All
points and curves in the set below are required to have an
inverse pre-image in $\parms \setminus \partial \parms$. In particular
all curves are $u$-monotone.\\
\rule{\textwidth}{1pt}
\begin{compactdesc}
\item[\ccode{Intersection\_2}:]
  Compute the intersections between two given curves $C_1$ and $C_2$.
\item[\ccode{Split\_2}:]
  Split a given curve $C$ at a given point $p$, which
  lies in the interior of $C$, into two interior disjoint subcurves.
\item[\ccode{Merge\_2}:]
  Merge two mergeable curves $C_1$ and $C_2$ into a single curve $C$.
\item[\ccode{Is\_mergeable\_2}:]
  Determine whether two curves $C_1$ and $C_2$ that share a common
  endpoint can be merged into a single continuous curve representable
  by the traits class.
\end{compactdesc}
\rule[5pt]{\textwidth}{1pt}

The further refined concept \ArrangementTraits{} enables the
construction of arrangements induced by {\em general} curves. A model
of the refined concept must define a third type that represents a
general (not necessarily $u$-monotone) curve, named \ccode{Curve\_2}.
It also has to supply a method that subdivides a given curve into
simple $u$-monotone subcurves, and possibly isolated points.\footnote{For
example, the curve $(x^2 + y^2)(x^2 + y^2 - 1) = 0$ is comprised of two
$u$-monotone circular arcs, which together form the unit circle, and a
singular isolated point at the origin.} We refer to the entire hierarchy
of refinements defined above as a single ``abstract'' concept called
\NoBoundaryTraits{}, as it represents concepts the models of which
handle curves that must have inverse pre-images in
$\parms \setminus \partial \parms$. We use this abstract concept to
simplify the description of the hierarchy defined below.

\begin{wrapfigure}[12]{r}{9.5cm}
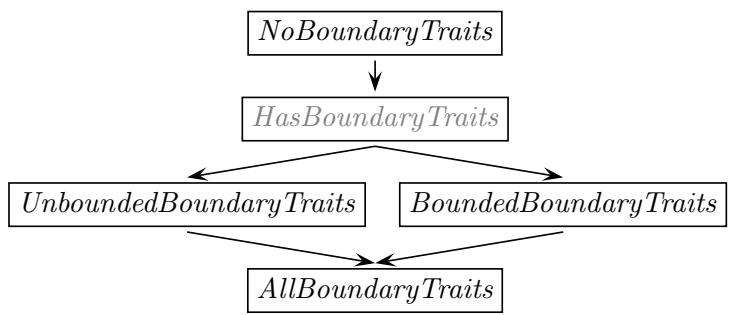

  \vspace{-12pt}
  \psset{treevshift=0,unit=1em,xunit=2em,yunit=1em,everytree={},
        etcratio=.75,triratio=.5}
     \jtree[everylabel=\sl,xunit=70pt,arrows=->]
     \! = {\NoBoundaryT}
           <vert>[scaleby=0 1]{\HasBoundaryT}@hbt !hbt
           .
     \!hbt = <right>[scaleby=1 1]{\BoundedBoundaryT}@bbt !bbt
 	  ^<left>[scaleby=1 1]{\UnboundedBoundaryT}@ubt !ubt
 	  .
     \!bbt = <left>[scaleby=1 1]{\AllBoundaryT}@abt !abt
 	  .
     \ncline{ubt:b}{abt:t}
     \endjtree
  \psset{treevshift=0,unit=1cm,xunit=1cm,yunit=1cm,everytree={},
         etcratio=.75,triratio=.5}
  \caption[Geometry-traits concept abstract refinement
    hierarchy]{\capStyle{Abstract refinement hierarchy of geometry-traits
      concepts for arrangement on surfaces.}}
  \label{fig:traits-abstract-hierarchy}
\end{wrapfigure}
The package introduces additional concepts, models of which are able
to handle unbounded curves or curves that reach the boundaries, the
endpoints of which coincide with contraction points or lie on
identification curves; see Figure~\ref{fig:traits-abstract-hierarchy}. The
``abstract'' \HasBoundaryTraits{} sub-hierarchy lists additional
predicates required to handle both curves that reach or approach the
boundaries of the parameter space. It has no models. The refined
\BoundedBoundaryTraits{} and \UnboundedBoundaryTraits{}
sub-hierarchies list additional predicates required to handle bounded
and unbounded curves, respectively. The geometry-traits class that
handles arcs of great circles models the \BoundedBoundaryTraits{}
concept, as the parameter space is bounded in all four
directions. Finally, the \AllBoundaryTraits{} sub-hierarchy refines
all the above. A model of this concept can handle unbounded curves
in some directions and bounded curves in others.

In the rest of this section all curves are required to be $u$-monotone.
The \HasBoundaryTraits{} concept requires the following additional
operations:\\
\rule{\textwidth}{1pt}
\begin{compactdesc}
\item[\ccode{Parameter\_space\_in\_x\_2}:]
  Given a curve $C$ and an index $d \in \{0,1\}$ that identifies one
  of its ends, determine the location of its pre-image in the domain
  $\parms$ along the $u$ dimension. More precisely, determine whether
  $u(\gamma(d))$ is equal to $\umin$, $\umax$, or falls in between. In
  case of an unbounded curve, determine whether
  $\lim_{t \rightarrow d} u(\gamma(t))$ is equal to $-\infty$ or $+\infty$.
\item[\ccode{Parameter\_space\_in\_y\_2}:]
  Given a curve $C$ and an index $d \in \{0,1\}$ that identifies one
  of its ends, determine the location of its pre-image in the domain
  $\parms$ along the $v$ dimension. More precisely, determine whether
  $v(\gamma(d))$ is equal to $\vmin$, $\vmax$, or falls in between. In
  case of an unbounded curve, determine whether
  $\lim_{t \rightarrow d} v(\gamma(t))$ is equal to $-\infty$ or $+\infty$.
\item[\ccode{Compare\_x\_near\_boundary\_2}:]
  There are two predicates:
  \begin{compactenum}
  \item
    Given a point $p$, the inverse of which
    is in $\parms \setminus \partial \parms$, a curve $C$, and an index
    $d \in \{0,1\}$ that identifies an end of $C$, compare the $u$
    coordinates of $p$ and a point along $C$ near its given end.
    More precisely, return $\mathrm{cmp}_u(p, \gamma(|d-\epsilon|))$,
    where $\epsilon > 0$ is infinitesimally small.
  \item
    Given two curves $C_1$ and $C_2$ and two corresponding indices
    $d_1, d_2 \in \{0,1\}$ that identify two ends of $C_1$ and $C_2$
    respectively, compare the $u$ coordinates of two points along $C_1$
    and $C_2$ respectively near their given ends. More precisely, return
    $\mathrm{cmp}_u(\gamma_1(|d_1-\epsilon_1|), \gamma_2(|d_2-\epsilon_2|))$,
    where $\epsilon_1, \epsilon_2 > 0$ are infinitesimally small.
  \end{compactenum}
  See Section~\ref{ssec:aos:geodesics:geometry-traits} for an example.
\item[\ccode{Compare\_y\_near\_boundary\_2}:]
  Given two curves $C_1$ and $C_2$, and a single index $d \in \{0,1\}$
  that identifies two ends of $C_1$ and $C_2$, compare the $v$
  coordinates of two points along $C_1$ and $C_2$ respectively near the
  given ends. More precisely, return
  $\mathrm{cmp}_v(\gamma_1(|d-\epsilon_1|), \gamma_2(|d-\epsilon_2|))$,
  where $\epsilon_1,\epsilon_2 > 0$ are infinitesimally small.
  See Section~\ref{ssec:aos:geodesics:geometry-traits} for an example.
\end{compactdesc}
\rule[5pt]{\textwidth}{1pt}

The \UnboundedBoundaryTraits{} concept requires the following additional
operations:\\
\rule{\textwidth}{1pt}
\begin{compactdesc}
\item[\ccode{Is\_bounded\_2}:]
  Given a curve $C$ and an index $d \in \{0,1\}$ that identifies an end
  of $C$, determine whether the curve end is bounded.
\end{compactdesc}
\rule[5pt]{\textwidth}{1pt}

The \BoundedBoundaryTraits{} concept requires the following additional
operations:\\
\rule{\textwidth}{1pt}
\begin{compactdesc}
\item[\ccode{Is\_on\_x\_identification\_2}:]
  This predicate applies only to a parameterization that has a vertical
  identification curve. Given a point $p$ (respectively a curve $C$),
  determine whether $p$ (respectively $C$) lies on the vertical and
  identified sides of the boundary. More precisely, determine whether
  $u_p \in \{\umin,\umax\}$. (Respectively, determine whether
  $u(\gamma(t)) \in \{\umin,\umax\}, \forall t \in [0,1]$.
\item[\ccode{Is\_on\_y\_identification\_2}:]
  This predicate applies only to a parameterization that has a horizontal
  identification curve. 
  Given a point $p$ (respectively a curve $C$), determine whether $p$
  (respectively $C$) lies on the horizontal and identified sides of the
  boundary. More precisely, determine whether $v_p \in \{\vmin,\vmax\}$ for
  all pre-images of $p$. (Respectively, determine whether
  $v(\gamma(t)) \in \{\vmin,\vmax\}, \forall t \in [0,1]$.)
\item[\ccode{Is\_on\_x\_contraction\_2}:]
  This predicate applies only to a parameterization that has a contracted
  vertical boundary. 
  determine whether $p$ coincides with a contraction point. More
  precisely, determine whether $u_p$ is equal to $\umin$ or $\umax$.
\item[\ccode{Is\_on\_y\_contraction\_2}:]
  This predicate applies only to a parameterization that has a contracted
  horizontal boundary. 
  determine whether $p$ coincides with a contraction point. More
  precisely, determine whether $v_p$ is equal to $\vmin$ or $\vmax$.
\item[\ccode{Compare\_x\_on\_identification\_2}:]
  This predicate applies only to a parameterization that has a horizontal
  identified sides of the boundary. 
  Given two points $p_1$ and $p_2$ that lie on the horizontal
  identification arc, compare their $u$-coordinates.
\item[\ccode{Compare\_y\_on\_identification\_2}:]
  This predicate applies only to a parameterization that has a vertical
  identified sides of the boundary.. 
  Given two points $p_1$ and $p_2$ that lie on the vertical
  identification arc, compare their $v$-coordinates.
\end{compactdesc}
\rule[5pt]{\textwidth}{1pt}

All traits-class operations are implemented as function
objects\index{object!function} ({\em
functors}\index{functor|see{object!function}}) according to \cgal's
guidelines. This allows extending the geometric types above, without
the need to redefine the methods that operate on them;
see~\cite{hhkps-aegk-07} for details on the extensible kernel. For a
detailed specification of the various concept
requirements see~\cite{wf-ndcap-05}.

\subsection{The Geometry-Traits Adaptor}
\label{ssec:aos:geometry-traits:adaptor}
The geometry-traits adaptor class-template implements geometric
operations that are not provided by a model of the geometry-traits
concept itself, using the operations supplied by a model of the
geometry-traits concept as basic blocks. It decreases the effort
required to develop geometry-traits models, and at the same time
increases the usability of the geometry-traits models, adapting them
for extended uses. A geometry-traits type is injected as a template
parameter into the adaptor class, which inherits from it,
centralizing all geometric operations. In cases where the efficiency
of methods is crucial, a developer has a way to override these methods
with optimized ones. 

For example, in order to determine whether a point $p$ is in the
$u$-range of a $u$-monotone curve $C$, the adaptor simply compares
$p$ to the endpoints of $C$. It checks whether $p$ lies to the right
of the left endpoint and to the left of the right endpoint.

In some cases, the geometry-traits adaptor class uses a
{\em tag-dispatching}\index{tag dispatching} mechanism to select the
appropriate implementation of a geometry-traits class operation. Tag
dispatching is a technique that uses function overloading to dispatch
a function {\em at compile time}, based on properties of the types of
the arguments the function
accepts~\citelinks{boost-generic-programming}. This mechanism enables
users to implement their traits class with a reduced or alternative
set of operations. The adaptor respects the tags listed below every
geometry-traits class must define.
\begin{compactdesc}
\item[\ccode{Has\_left\_category}:]
A Boolean tag that indicates whether the traits class provides the predicate
\ccode{compare\_y\_at\_x\_left}, which compares two
$u$-monotone curves to the \emph{left} of a common right endpoint.
This predicate is required only by some point-location strategies and
by the zone-computation algorithm.
While in some cases it is fairly easy for the traits-class implementer
to provide it, in other cases it can be rather difficult, or even quite
impossible. When this tag is false, the traits-class adaptor resorts
to a somewhat less efficient algorithm that uses (other) existing
traits-class predicates.
\item[\ccode{Has\_merge\_category}]
A Boolean tag that indicates whether a model of the
\concept{ArrangementXMonotoneTraits\_2} supports the merge of
$u$-monotone curves. If the tag is true, the traits class must
provide the two operations \ccode{merge\_2} and \ccode{is\_mergeable\_2}.
The merger operation is used to eliminate redundant features in the
arrangement. For example, if we have a T-shaped structure
formed by two line segments, and the vertical segment forming the
``leg'' is removed, then it is possible to merge the two horizontal
sub-segments. When the \emph{has-merge} tag is false, the adaptor
simply declares any pair of curves as non-mergeable. The only effect 
on the arrangement is that we cannot remove redundant vertices (of
degree two) following the deletion of edges.
\item[\ccode{Boundary\_category}:]
  A quadruple tag that categorizes the traits class according to the
  hierarchy described in Figure~\ref{fig:traits-abstract-hierarchy}.
  The adaptor provides empty implementations of the operations that are
  never invoked, yet required for smooth compilation.
\end{compactdesc}

\subsection{Geometry-Traits Models}
\label{ssec:aos:geometry-traits:models}
\begin{table}[!htp]
  \caption[Geometry-traits models]{\capStyle{Geometry-traits models}}
  \label{tab:geometry-traits-models}
  \centerline{
    \begin{tabular}{|l|l|l|l|l|}
      \hline
      \multicolumn{1}{|c|}{\textbf{Curve Family}} &
      \multicolumn{1}{c|}{\textbf{Degree}} &
      \multicolumn{1}{c|}{\textbf{Surface}} &
      \multicolumn{1}{c|}{\textbf{Boundness}} &
      \multicolumn{1}{c|}{\textbf{Arithmetic}}\\
      \hline
      \hline
      linear segment & 1 & plane & bounded & rational\\
      \hline
      linear segments, rays, lines & 1 & plane & unbounded & rational\\
      \hline
      piecewise linear curves & $\infty$ & plane & bounded & rational\\
      \hline
      circular arcs, linear segments & $\leq 2$ & plane & bounded & rational\\
      \hline
      algebraic curves & $\leq 2$ & plane & unbounded & algebraic\\
      \hline
      quadric projections & $\leq 4$ & plane & unbounded & algebraic\\
      \hline
      algebraic curves & $\leq 3$ & plane & unbounded & algebraic\\
      \hline
      algebraic curves & $\leq n$ & plane & unbounded & algebraic\\
      \hline
      planar \bez{} curves & $\leq n$ & plane & unbounded & algebraic\\
      \hline
      univariate polynomials & $\leq n$ & plane & unbounded & algebraic\\
      \hline
      geodesic arcs on sphere & $\leq 2$ & sphere & bounded & rational\\
      \hline
      quadric intersection arcs & $\leq 4$ & quadric & unbounded & algebraic\\
      \hline
      Dupin cyclide intersection arcs & $\leq n$ & Dupin cyclides &
      bounded & algebraic\\
      \hline
    \end{tabular}
  }
\end{table}
The large number of geometry-traits models already implemented enables
the construction and maintenance of arrangements induced by many
different types of curves. The package itself contains several models
of the geometry-traits concept. A few other models have been developed
by other groups of researchers. Models are distinguished not only by
the different families of curve they handle, but also by their
suitability for constructing and maintaining arrangements with
different characteristics. For example, there are two distinct models
that handle line segments~\cite{wfzh-aptac-07}. One caches information
in the curve records, while the other retains the minimal amount of
data. While operations on arrangements instantiated with the former
model consume more space, they are more efficient for dense
arrangements (namely, arrangements induced by curves with a large
number of intersections). Another model handles not only (bounded)
line-segments, but also rays and lines~\cite{bfhmw-scmtd-07,bfhmw-apsgf-09}.
There are traits models for non-linear curves, such as circular
arcs~\cite{cpt-eeccc-07}, conic
curves~\cite{w-hlfac-02,behhm-cbcab-02,ekptt-tpck-04}, cubic
curves~\cite{eksw-ceeccc-04}, and quartic
curves that are the projection of the intersection of two quadric
surfaces~\cite{bhksw-eceic-05}, and there are traits classes for arcs
of graphs of rational univariate polynomial
functions~\cite{lpt-cakra-08,wfzh-aptac-07}. There is even a traits
class that handles algebraic curves of arbitrary
degrees~\cite{ek-eeaaa-08}. There is also a traits class that handles
\bez{} curves~\cite{hw-execa-07}. There is a traits class for geodesic
arcs embedded on the sphere~\cite{fsh-agas-08,bfhks-apsca-09}, (see
Section~\ref{sec:aos:geodesics}), another one for intersections of
quadrics embedded on a quadric~\cite{bfhmw-scmtd-07,bfhks-apsca-09},
and another one for intersections of arbitrary algebraic surfaces with
a \Index{Dupin cyclide} embedded on the Dupin
cyclide~\cite{bk-eatdc-08,bfhks-apsca-09}. Finally, there is a model
that handles continuous piecewise linear curves, referred to as
polylines\index{polyline}, (see
Section~\ref{ssec:aos:geometry-traits:models:polylines}).

\subsection{Geometry-Traits Extension}
\label{ssec:aos:geometry-traits:extension}
Traits-class decorators\index{decorator} are used to extend the
geometric entities defined by the traits class with additional,
possibly non-geometric, data. An alternative way to achieve this is to
extend the geometric types of the kernel, as the kernel is fully
adaptable and extensible~\cite{hhkps-aegk-07}. However, this
indiscriminating extension may lead to an undue space-consumption, as
every geometric object is extended, regardless of its use. It also
requires nontrivial knowledge about the kernel structure and the
techniques to extend it.

There is a decorator that enables the extension of the (general) curve
and the $u$-monotone curve types with distinct types of data, and
there is a convenient one, where the data attached to the $u$-monotone
curve type is a set of objects, the type of which is attached to the
(general) curve type. This set usually contains a single data object,
unless the $u$-monotone curve corresponds to an overlapping section of
two curves or more. When a curve with a data field $d$ is split into
$u$-monotone subcurves, each subcurve is associated with a singleton
set $\{d\}$. When two $u$-monotone curves overlap, the decorator takes
the union of their data sets, and associates it with the resulting
overlapping subcurve.

\subsection{A Geometry-Traits Model that Handles Polylines}
\label{ssec:aos:geometry-traits:models:polylines}
Polylines\index{polyline} are of particular interest, as they can be
used to approximate more complex curves. At the same time handling
them is easier than handling higher-degree algebraic curves, as
rational arithmetic is sufficient to carry out exact computations
on polylines.

The geometry-traits model that handles polylines is a class-template
called\linebreak 
\ccode{Arr\_polyline\_traits\_2}. It must be instantiated with a
geometry-traits class that models the concept
\concept{ArrangementLinearTraits}. This concept refines the 
\concept{ArrangementTraits} concept, as it adds a variant --- it must
handle line segments. This variant cannot be enforced by the compiler,
but rather be verified at run time. A polyline curve is represented as
a vector of \ccode{SegmentTraits::X\_monotone\_curve\_2} objects
(namely segments). The polyline-traits class does not perform
any geometric operations directly. Instead, it solely relies on the
functionality of the instantiated segment-traits class. For example,
when we need to determine the position of a point with respect to a
$u$-monotone polyline, we use binary search to locate the relevant
segment that contains the point in its $u$-range, then we compute
the position of the point with respect to this segment. Thus,
operations on $u$-monotone polylines of size $m$ typically take
$O(\log m)$ time.

Users are free to choose the underlying segment-traits class based
on the number of expected intersection points (see discussion above
in Section~\ref{ssec:aos:geometry-traits:models}). Moreover, it is
possible to instantiate the polyline-traits class-template with a
traits class that handles segments with some additional data attached
to them (see Section~\ref{ssec:aos:geometry-traits:extension}). This
makes it possible to associate different data objects with the
different segments that compose a polyline.

\section{Arrangements of Geodesic Arcs on the Sphere}
\label{sec:aos:geodesics}
In this section we concentrate on the particular category of
arrangements embedded on surfaces, where the embedding space is the
sphere, and the inducing objects are geodesic arcs. There is an analogy
between this class of arrangements and the class of planar arrangements
induced by linear curves (i.e., segments, rays, and lines), as
properties of linear curves in the plane can be often, but not always,
adapted to geodesic arcs on the sphere.

An arrangement of geodesic arcs embedded on the sphere is defined as
an instance of the \aos{} class-template instantiated with appropriate
geometry- and topology-traits classes, namely
\ccode{Arr\_geodesic\_arc\_on\_sphere\_traits\_2} and\linebreak
\ccode{Arr\_spherical\_topology\_traits\_2}, respectively.
The geometry-traits class is tailored to handle geodesic arcs as
efficiently as possible, and defines the parameterization used:
$\parms = [-\pi + \alpha, \pi + \alpha] \times [-\frac{\pi}{2}, \frac{\pi}{2}]$,
$f_{S}(u, v) = (\cos u \cos v, \sin u \cos v, \sin v)$, where $\alpha$
is a variable that must be set at compile time, and is by default $0$.
This parameterization induces two contraction points
$p_s = (0, 0, -1)$ and $p_n = (0, 0, 1)$, referred to as the south and
north pole respectively, and an identification curve that coincides
with the opposite Prime (Greenwich) Meridian. We developed the
topology-traits class to support not only arrangements of geodesic arcs,
but any type of curves embedded on the sphere parameterized as above,
without compromising the performance of the operations gathered
in the traits class. We hope that this topology-traits class will come
in handy in the future for constructing and maintaining arrangements
induced by types other than geodesic arcs, such as general circular
arcs, which appear in arrangements induced by intersections of spheres
embedded on the sphere~\cite{cl-ceacs-06}. The topology-traits
class initializes the \dcel{} to have a single face, the embedding of which,
is the entire sphere. It is designed to retain the variant that this
face always contains the north pole during modifications the
arrangement may undergo. The topology-traits class is required, for example,
to inform its users that the top and bottom boundaries of the
parameter space are contracted and the left and right boundaries are
identified. It maintains a search structure of vertices that coincide
with the contraction points or lie on the identification arc.

\begin{wrapfigure}[7]{r}{3.5cm}
  \vspace{-5pt}
  \epsfig{figure=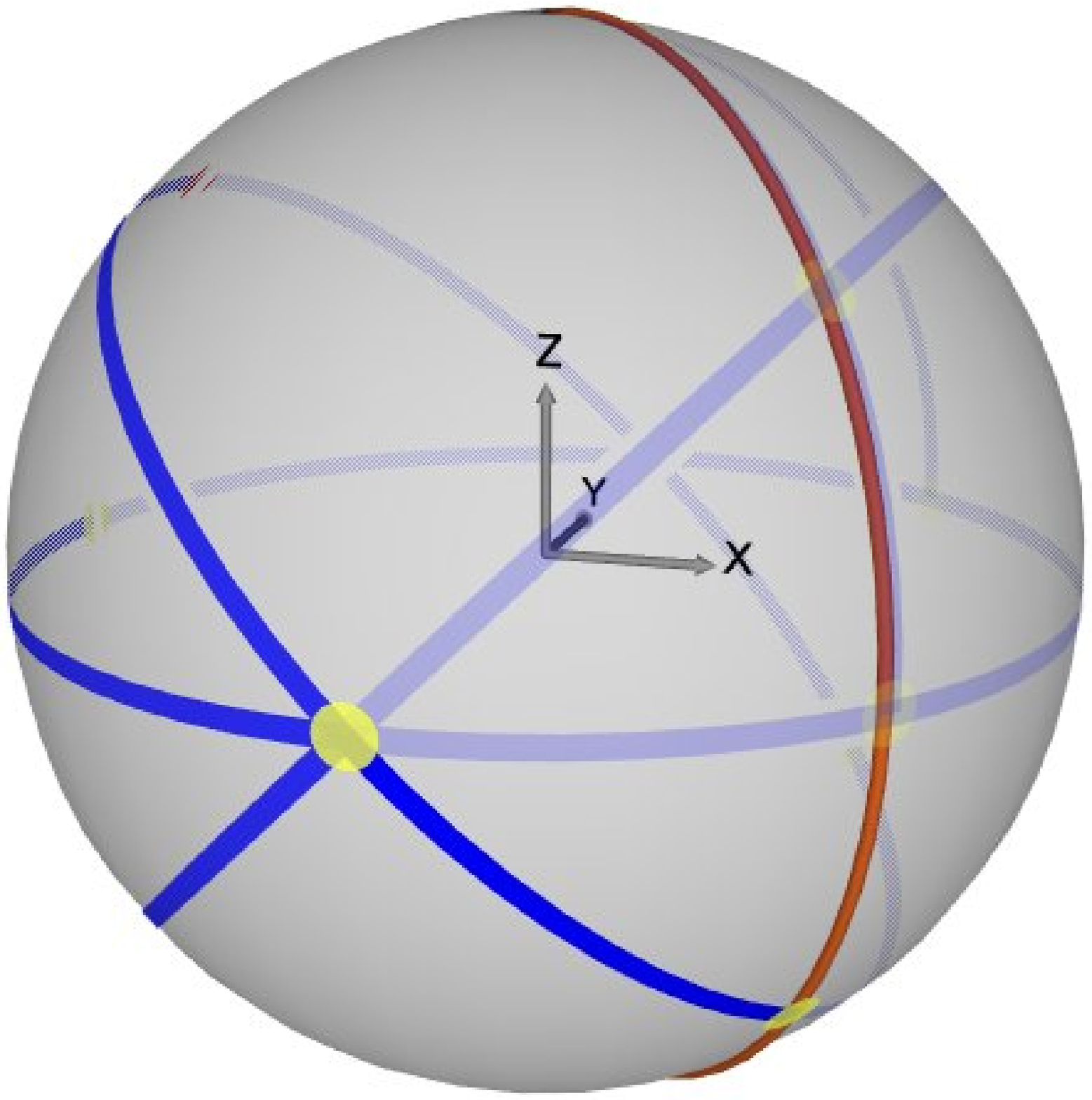,width=3.5cm,silent=}
\end{wrapfigure}
The figure to the right is a snapshot of a movie~\cite{fsh-agas-08}
that demonstrates, among the other, the sweep-line procedure carried
out on the sphere. The red vertical arc that connects the poles is the
imaginary sweep-line. The yellow vertex have been processed
already. The dark blue arcs are curves that have been processed
already and inserted into the arrangement. The light blue arcs are
curves that are to be processed as the sweep line advances.

\subsection{The Geometry-Traits Model}
\label{ssec:aos:geodesics:geometry-traits}
The geometry-traits class for geodesic arcs on the sphere is
parameterized with a geometric kernel~\cite{hhkps-aegk-07} that
encapsulates the number type used to represent coordinates of
geometric objects and to carry out algebraic operations on those
objects. The implementation handles all degeneracies, and is exact,
as long as the underlying number type is rational, even though the
embedding surface is a sphere. We are able to use high-performance
kernel models instantiated with exact rational number-types for the
implementation of this geometry-traits class, as exact rational
arithmetic suffices to carry out all necessary algebraic operations.
The ability to robustly construct arrangements of geodesic arcs on
the sphere, and robustly apply operations on them using only (exact)
rational arithmetic is a key property that enables an efficient
implementation.

\newcounter{aos:geodesics:cntr}
A point in our arrangement is defined to be an unnormalized vector in
$\rrr$, representing the place where the ray emanating from the origin
in the relevant direction pierces the sphere. An arc of a great circle is
represented by its two endpoints, and by the plane that contains the
endpoint vectord and goes through the origin. The orientation of the
plane and the source and target points determine which one of the two
great arcs is considered.

The point type is extended with an enumeration that indicates whether the
vector
\setcounter{aos:geodesics:cntr}{1}(\roman{aos:geodesics:cntr}) pierces the
south pole,
\addtocounter{aos:geodesics:cntr}{1}(\roman{aos:geodesics:cntr}) pierces
the north pole,
\addtocounter{aos:geodesics:cntr}{1}(\roman{aos:geodesics:cntr}) intersects
the identification arc, or
\addtocounter{aos:geodesics:cntr}{1}(\roman{aos:geodesics:cntr}) is in any
other direction. An arc of a great circle is extended with three Boolean
flags that indicate whether any one of the $x,y,z$ coordinates of the
normal of the plane that defines the arc vanishes. These flags are
used to minimize the number of invocations of the geometry-traits
operations, which has a drastic effect on the performance of arrangement
operations at the account of a slight increase in space consumption.
This representation enables an exact yet efficient implementation of
all geometric operations required by the geometry-traits concept using
exact rational arithmetic, as normalizing vectors (that represent
directions and plane normals) is completely avoided.

We describe in details four predicates, namely
\ccode{Compare\_x\_2}, \ccode{Compare\_xy\_2},\linebreak
\ccode{Compare\_x\_near\_boundary\_2}, and
\ccode{Compare\_y\_near\_boundary\_2};
see Section~\ref{sec:aos:geometry-traits} for the complete set of the
concept requirements. The former compares two points $p_1$ and $p_2$
by their $u$-coordinates. The concept admits the assumption that the
input points do not coincide with the contraction points and do not
lie on the identification arc. Recall that points are in fact
unnormalized vectors that represent directions in $\rrr$. We project
$p_1$ and $p_2$ onto the $xy$-plane to obtain two-dimensional
unnormalized vectors $\hat{p_1}$ and $\hat{p_2}$, respectively. We
compute the intersection between the identification arc and the
$xy$-plane to obtain a third two-dimensional unnormalized vector
$\hat{d}$. Finally, we test whether $\hat{d}$ is reached strictly
before
\begin{wrapfigure}[15]{r}{3.8cm}
  \centerline{
    \begin{tabular}{c}
      \epsfig{figure=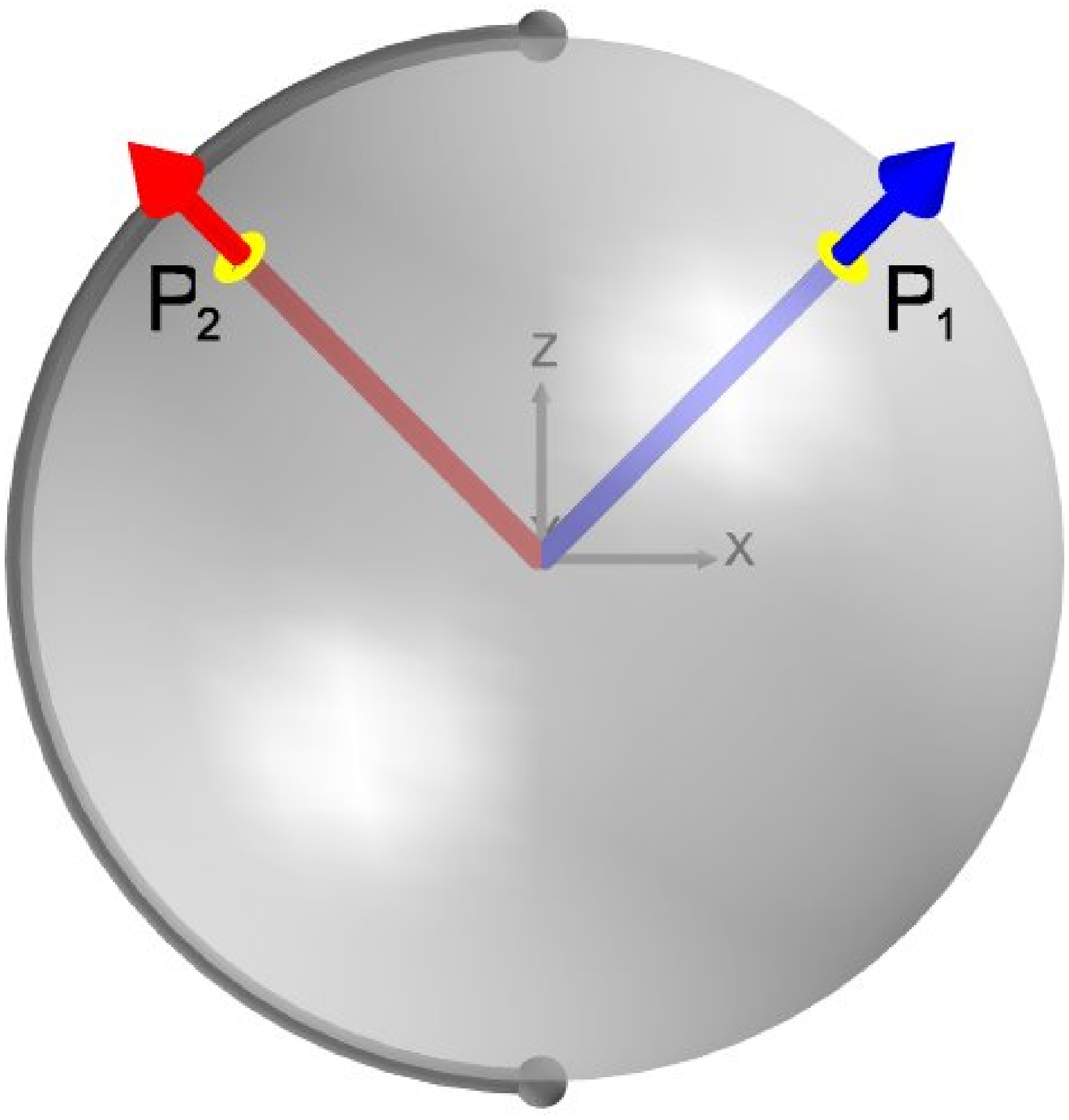,width=3.5cm,silent=}\\
      \pspicture[](-1.9,-1.9)(1.9,1.9)
      \psset{unit=1cm,linewidth=1pt}
      \pscircle[linewidth=1pt](0,0){1.5}
      \psaxes[labels=none]{->}(0,0)(1.5,1.5)
      \uput[0]{0}(1.5,0){$x$}
      \uput[90]{0}(0,1.5){$y$}
      \psline{->}(0,0)(-1.5,0)
      \uput[180]{0}(-1.5,0){$\hat{d}$}
      \psline{->}(0,0)(1.06,-1.06)
      \uput[-45]{0}(1.06,-1.06){$\hat{p_1}$}
      \psline{->}(0,0)(-1.06,-1.06)
      \uput[-135]{0}(-1.06,-1.06){$\hat{p_2}$}
      \endpspicture
    \end{tabular}
  }
\end{wrapfigure}
$\hat{p_2}$ is reached, while rotating counterclockwise
starting at $\hat{p_1}$. This geometric operation is supported by
every geometric kernel of \cgal{}. In the figure on the right $\hat{d}$
is reached strictly before $\hat{p_2}$ is reached. Therefore, the
$u$-coordinate of $p_1$ is larger than the $u$-coordinate of $p_2$.

The predicate \ccode{Compare\_xy\_2} compares two points $p_1$ and
$p_2$ lexicographically. It first applies \ccode{Compare\_x\_2} to
compare the $u$-coordinates of the two points. If the $u$-coordinates
are equal, it applies a predicate that compares the $v$-coordinates
of two points with identical $u$-coordinates, referred to as
\ccode{Compare\_y\_2}. This predicate first compares the signs of the
$z$-coordinates of the two unnormalized input vectors. If they differ,
it concludes that the point with the positive $z$-coordinate has a
$v$-coordinate that is larger than the $v$-coordinate of the point with
the negative $z$-coordinate. If the signs are identical, it compares the
squares of their normalized $z$-coordinates, essentially avoiding
the square-root operation. If the sign of the (unnormalized)
$z$-coordinates of both points is positive (resp. negative), the point
with the larger (resp. smaller) square of normalized $z$-coordinate has
a larger $v$-coordinate.

\begin{wrapfigure}[8]{r}{3.5cm}
  \vspace{-15pt}
  \centerline{
    \epsfig{figure=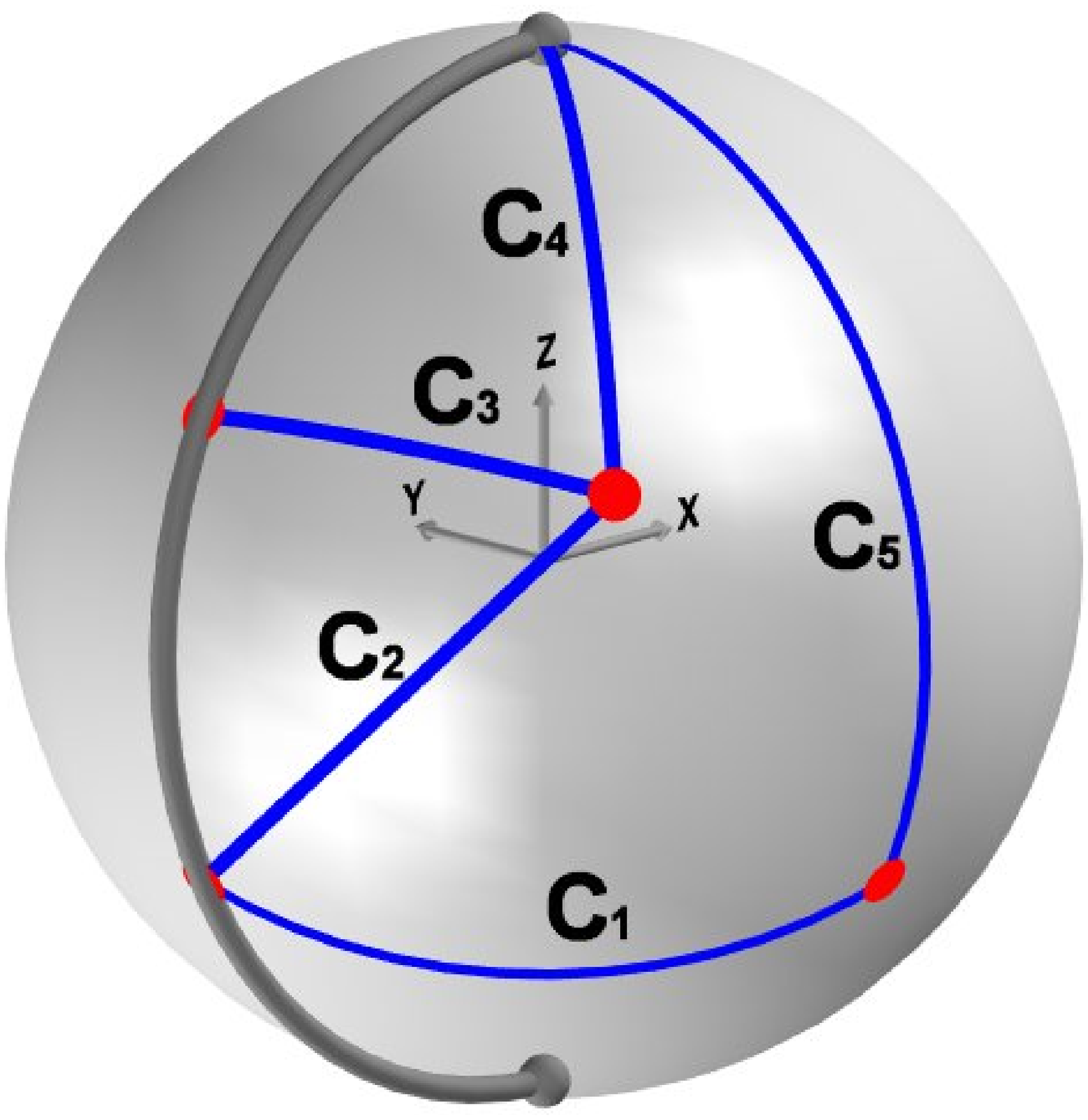,width=3.5cm,silent=}
  }
\end{wrapfigure}
The predicates above accept points, the pre-images of which,
lie in the interior of the parameter space. However, there is also a
need to lexicographically compare the ends of arcs, the
pre-images of which reach the boundary of the parameter space. The
predicate \ccode{Compare\_x\_near\_boundary\_2} accepts either
\begin{inparaenum}[(i)]
\item~a point, the pre-image of which lies in the interior of the
parameter space, and an arc end, or
\item~two arc ends.
\end{inparaenum}
Such an arc end is provided by a vertical arc and an index that
identifies one of the two ends of the arc, and must coincide with
one of the contraction points. The first variant compares the
$u$-coordinates of the input point and a point along the input arc
near its given end, whereas the second variant compares the
$u$-coordinates of two points along the input arcs near their
respective given ends. Recall, that the $u$-coordinates of all points
along a vertical arc are the same ($C_4$ and $C_5$ in the figure above).
Thus, we can compare the $u$-coordinate of an arbitrary point on a
vertical arc that lie inside the parameter space. For example, for the
second case, we compare the two vectors perpendicular to the normals to
the planes that define the vertical arcs, respectively, e.g., 
the $u$-coordinate of a point on the arc $C_4$ near its top end is
smaller than the $u$-coordinate of a point on the arc $C_5$ near its
top end, and in particular it is smaller than the $u$-coordinate of the
bottom end of $C_5$.
The \ccode{Compare\_y\_near\_boundary\_2} predicate compares the
$v$-coordinate of two arcs ends, the pre-images of which lie on the
same (left or right) identified side of the boundary of the parameter
space. We use the aforementioned \ccode{Compare\_y\_2} predicate to
compare the end points. If the points are equal, we compare the
normals to the plane that define the arcs. In our example, the left
end of $C_1$ is smaller than the left end of $C_2$, which is smaller
than the left end of $C_3$.

All the required geometric types listed in the traits concept are
maintained using only rational numbers. All required geometric
operations are implemented using only rational
arithmetic.\footnote{Points are represented as unnormalized vectors;
The coordinates of such points are converted into machine
floating-point only for rendering purposes.} Degeneracies, such as
overlapping arcs that occur during intersection computation, are
properly handled. The end result is a robust yet efficient implementation.

\section{Applications}
\label{sec:aos:applications}
Arrangement on surfaces have many applications this thesis falls short
to list. However, we do list a few samples we were involved (or
remotely involved) with the implementation of which, i.e., Regularized
Boolean Set-Operations, Envelopes of Surfaces, and Voronoi diagrams.
\Index{Minkowski sum} construction is covered in details in the following
chapter. The Boolean set-operation results, minimization diagrams,
maximization diagrams and Voronoi diagrams, (see
Section~\ref{ssec:aos:applications:env} for definitions), and
Minkowski-sums are all represented as arrangements, and as such can be
passed as input to consecutive operations on arrangements supported by
the \aos{} package and its derivatives.

\subsection{Regularized Boolean Set-Operations}
\label{ssec:aos:applications:bso}
Together with R.\ Wein and B.\ Zuckerman we have developed a package
that supports \Index{Boolean set-operations} on point sets bounded by
$u$-monotone curves embedded on two-dimensional parametric surfaces
in $\rrr$~\cite{cgal:fwzh-rbso2-07}. In particular, it contains the
implementation of {\em regularized} Boolean 
set-operations\index{Boolean set-operations!regularized}, intersection
predicates, and  point containment predicates. A regularized Boolean
set-operation $\mbox{op}^*$ can be obtained by first taking the
interior of the resulting point-set of an {\em ordinary} Boolean
set-operation $(P\ \mbox{op}\ Q)$ and then by taking the
closure~\cite{h-sm-04}. That is, 
$P\ \mbox{op}^*\ Q = \mbox{closure}(\mbox{interior} (P\ \mbox{op}\ Q))$.
Regularized Boolean set-operations appear in constructive solid
geometry\index{constructive solid geometry|see{CSG}} (\Index{CSG}),
because regular sets are closed under regularized Boolean set-operations,
and because regularization eliminates lower dimensional features, namely
isolated vertices and ``antennas'' (namely, dangling edges), thus
simplifying and restricting the representation to physically meaningful
solids. Ordinary Boolean set-operations, which distinguish between the
interior and the boundary of a polygon, are not implemented within this
package. However, we implemented a specialized ordinary union operation
as part of an assembly partitioning application; see
Chapter~\ref{chap:assem_plan}.

\begin{wrapfigure}{r}{7.6cm}
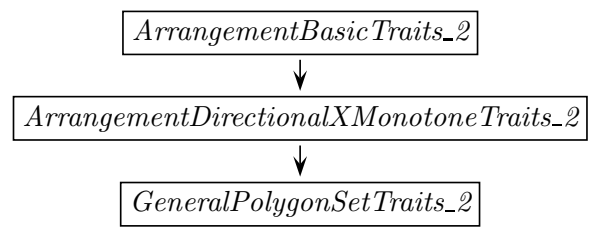

  \vspace{-12pt}
  \psset{treevshift=0,unit=1em,xunit=2em,yunit=1em,everytree={},
        etcratio=.75,triratio=.5}
     \jtree[everylabel=\sl,xunit=70pt,arrows=->]
     \! = {\ArrangementBasicT}
           <vert>[scaleby=0 1]{\ArrangementDirectionalXMonotoneT}@dxt !dxt
           .
     \!dxt = <vert>[scaleby=1 1]{\GeneralPolygonSetT}@gpst !gpst
 	   .
     \endjtree
  \psset{treevshift=0,unit=1cm,xunit=1cm,yunit=1cm,everytree={},
         etcratio=.75,triratio=.5}
  \caption[Geometry-traits concept refinement hierarchy for Boolean
    set-operations]{\capStyle{Refinement hierarchy of geometry traits
      concepts for Boolean set-operations.}}
  \label{fig:bso-traits-hierarchy}
\end{wrapfigure}
The operands and results of the regularized operations are general
polygons that may have holes. The boundaries of a general polygon and
of holes, if present, are general $u$-monotone curves. The \aos{}
class is employed to represent a point set embedded on a two-dimensional
parametric surface as an arrangement. A point set is typically
constructed from a single general polygon or a collection of interior
disjoint general polygons. The underlying arrangement must be
instantiated with a geometry traits that models the concept
\GeneralPolygonSetTraits. This concept refines the concept
\ArrangementDirectionalXMonotoneTraits, which in turns refines the
concept \ArrangementBasicTraits{}
(see Section~\ref{sec:aos:geometry-traits}).

The \ArrangementDirectionalXMonotoneTraits{} concept treats its
$u$-monotone curves as objects directed from one endpoint appointed
to be the source to the other endpoint appointed to be the target.
Thus, it requires few additional operations on $u$-monotone curves:\\
\rule{\textwidth}{1pt}
\begin{compactdesc}
\item[\ccode{Compare\_endpoints\_xy\_2}:]
  Given a $u$-monotone curve $C$, compare the source and target
  points of $C$ lexicographically.
\item[\ccode{Construct\_opposite\_2}:]
  Given a $u$-monotone curve $C$, construct the opposite curve of $C$
  (namely, swap the source and target endpoints of $C$).
\item[\ccode{Intersection\_2}:]
  Compute the intersections between two given curves $C_1$ and $C_2$.
\item[\ccode{Merge\_2}:]
  Merge two mergeable curves $C_1$ and $C_2$ into a single curve $C$.
\item[\ccode{Is\_mergeable\_2}:]
  Determine whether two curves $C_1$ and $C_2$ that share a common
  endpoint can be merged into a single continuous curve representable
  by the traits class.
\end{compactdesc}
\rule[5pt]{\textwidth}{1pt}
Most traits classes bundled in the \aos{} package and distributed with
\cgal, are models of the concept
\ArrangementDirectionalXMonotoneTraits.\footnote{The
\ccode{Arr\_polyline\_traits\_2} traits class is not a model of the
\ArrangementDirectionalXMonotoneTraits{} concept, as the
$u$-monotone curve it defines is always directed from left to
right. Thus, an opposite curve cannot be constructed.}

The \GeneralPolygonSetTraits{} concept requires its models to define
a type that represents a general polygon and another one that represents
general polygon with holes in addition to the \ccode{Point\_2} and
\ccode{X\_monotone\_curve\_2} types that must be defined by all models
of the generalized concept. It also requires the provision of several
operations that operate on these two types listed below.\\
\rule{\textwidth}{1pt}
\begin{compactdesc}
\item[\ccode{Construct\_polygon\_2}:]
  Given a sequence $\calC$ of $u$-monotone curves, construct a general
  polygon that has $\calC$ as its outer boundary.
\item[\ccode{Construct\_curves\_2}:]
  Given a general polygon $P$, obtain the sequence of $u$-monotone
  curves that comprise the boundary of $P$.
\item[\ccode{Construct\_general\_polygon\_with\_holes\_2}:]
  Given a general polygon $P$ and a (possibly empty) set of holes
  $\calH$, construct a general polygon with holes that has $P$ as its
  outer boundary and $\calH$ as its holes.
\item[\ccode{Construct\_outer\_boundary}:]
  Given a general polygon-with-holes $P$, obtain the general polygon
  that is its outer boundary.
\item[\ccode{Construct\_holes}:]
  Given a general polygon-with-holes $P$, obtain the holes of $P$ if any.
\item[\ccode{Is\_unbounded}:]
  Given a general polygon-with-holes $P$, determine whether it has an
  outer boundary.
\end{compactdesc}
\rule[5pt]{\textwidth}{1pt}

\subsection{Envelopes}
\label{ssec:aos:applications:env}
Lower envelopes of functions on parametric surfaces are defined in a way
similar to the standard definition of lower envelopes of bivariate
functions in space~\cite{h-a-04}. Let $S$ be a two-dimensional parametric
surface is $\rrr$. Given a set of bivariate functions
$F = \{f_1, \ldots, f_n\}$, where $f_i:S \to \reals$, their
{\em lower envelope}\index{envelope!lower} $\Psi(u, v)$ is defined to be
their point-wise minimum $\Psi(u, v) = \min_{1 \leq i \leq n} f_{i}(u, v)$.
The {\em \Index{minimization diagram}} \mindia{F} of the set $F$ is the
two-dimensional map obtained by the projection of the lower
envelope onto $S$. Upper envelopes and maximization diagrams are
defined analogously for the point-wise maximum of the functions.

\begin{figure}[!htp]
  \vspace{-10pt}
  \begin{tabular}{ccc}
    \epsfig{figure=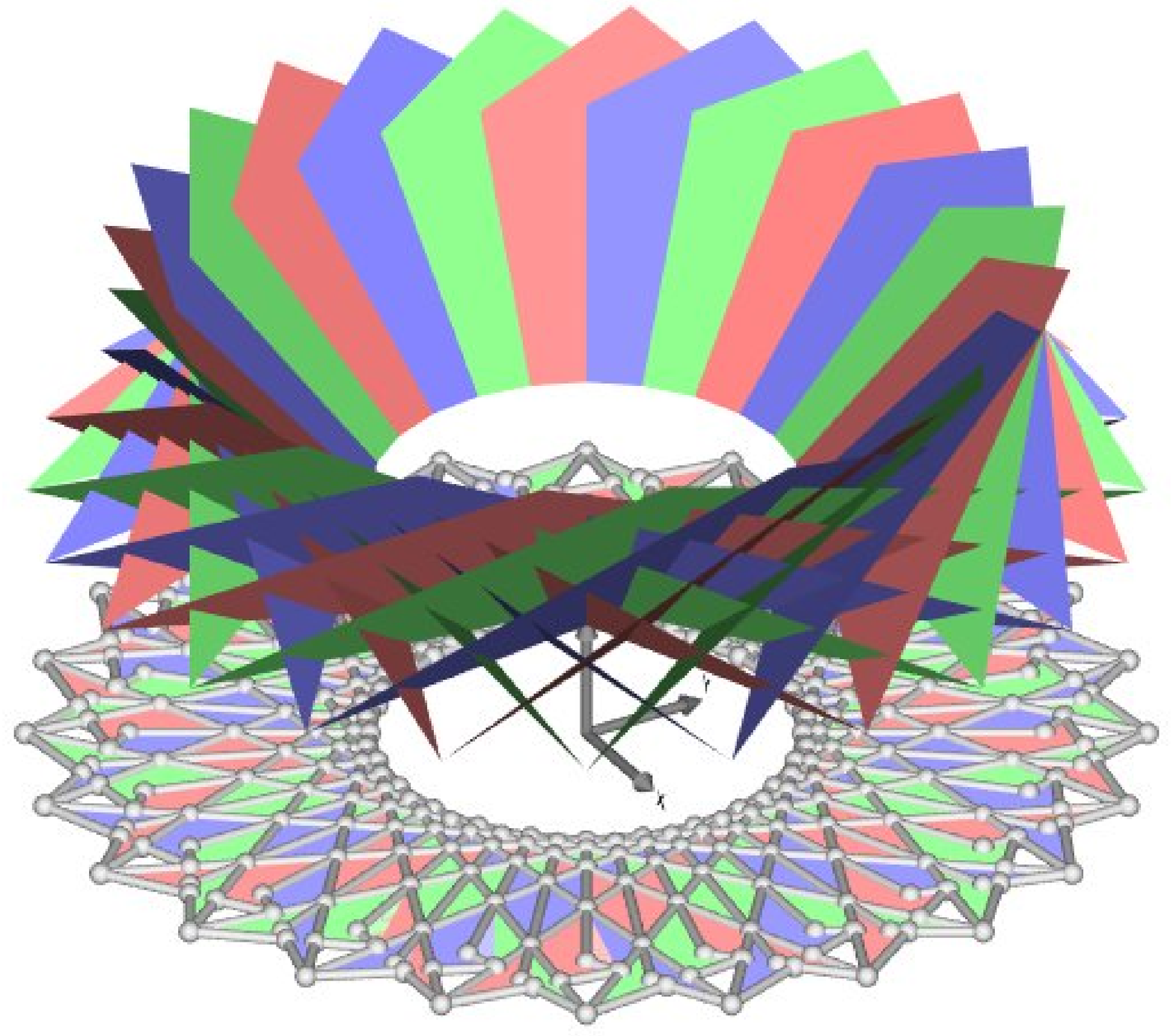,width=5cm,silent=} &
    \epsfig{figure=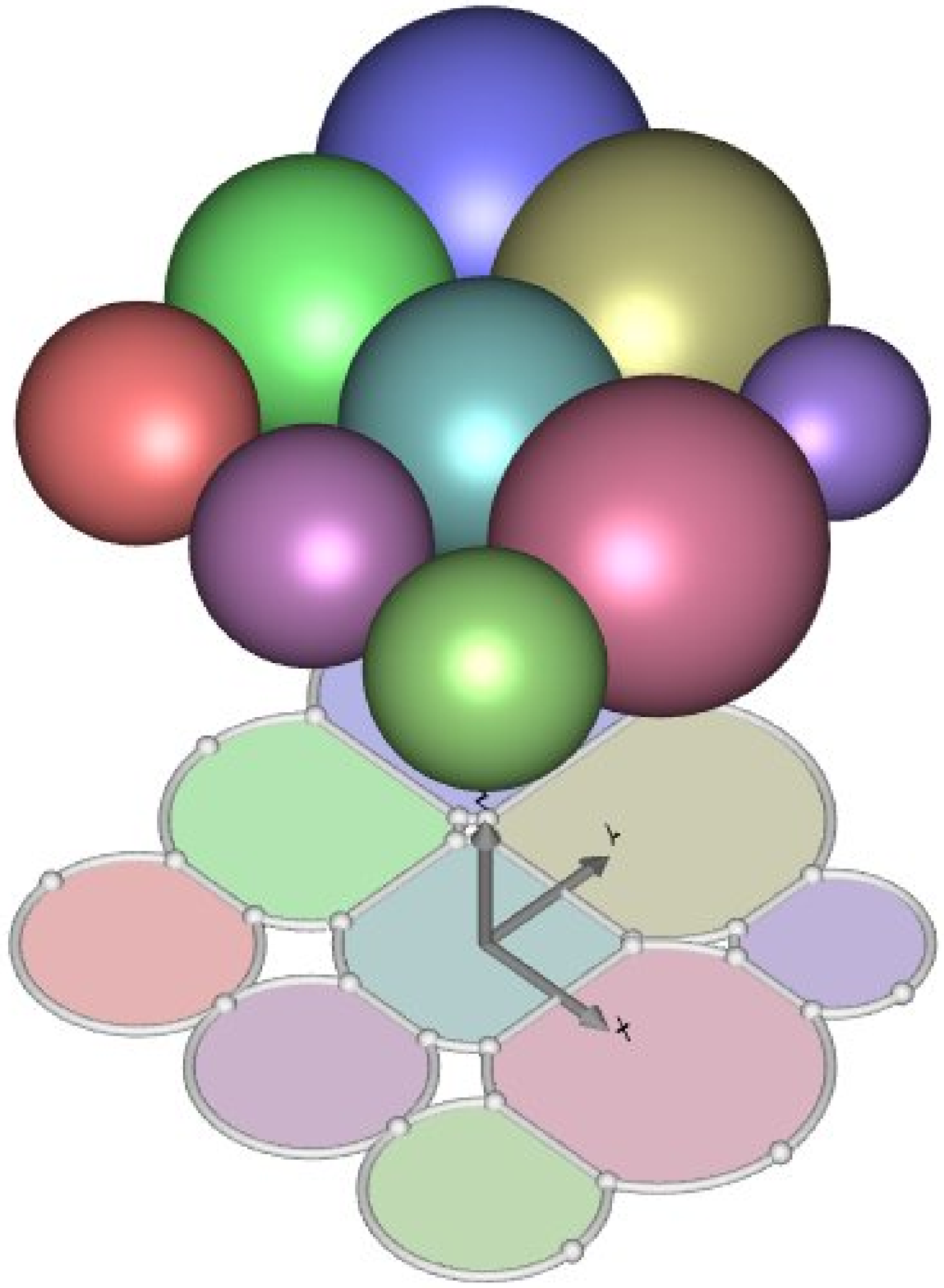,width=5cm,silent=} &
    \epsfig{figure=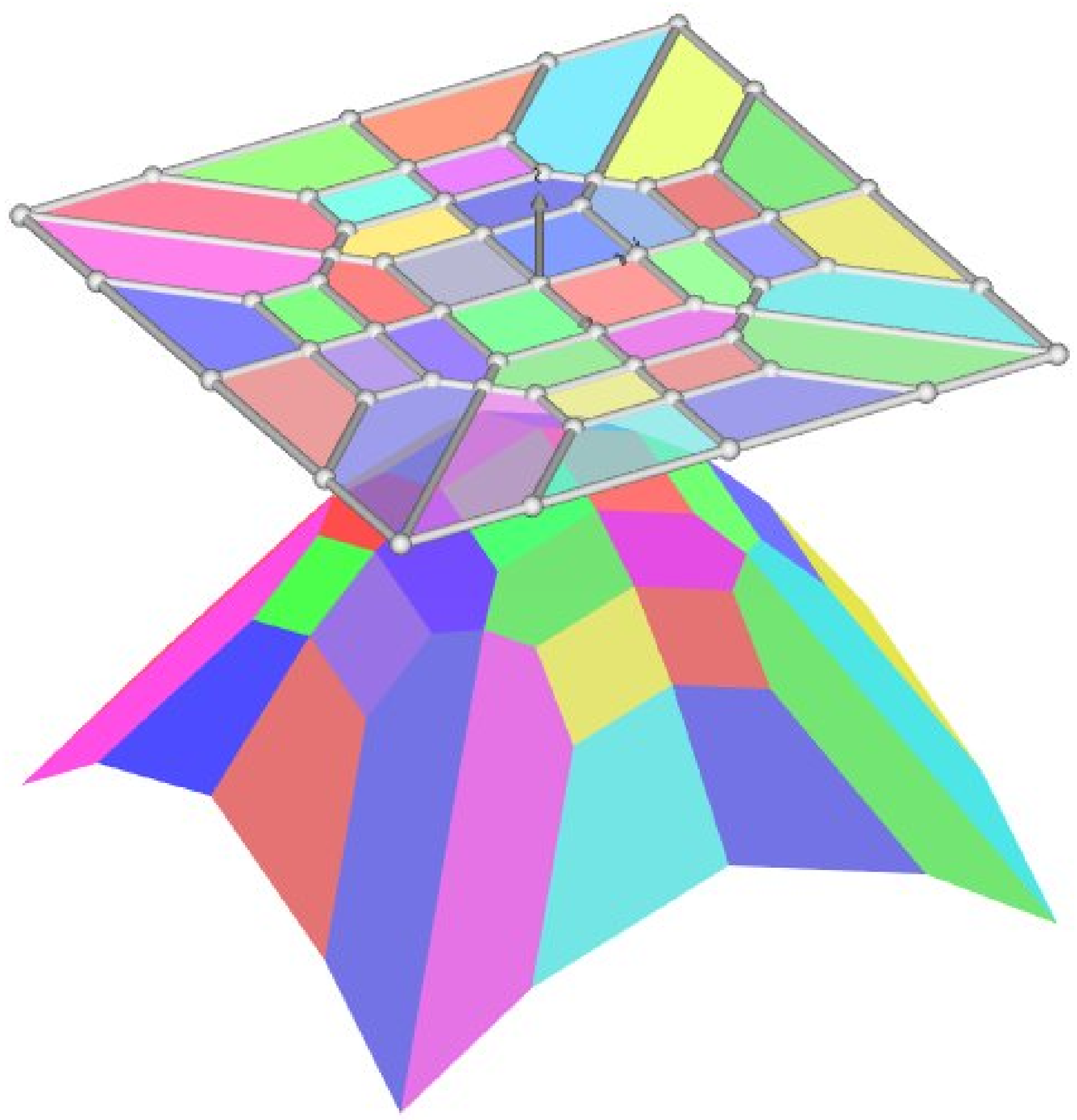,width=5cm,silent=}\\
      (a) & (b) & (c)
  \end{tabular}
  \caption[Lower envelopes of various types of surfaces]
  {\capStyle{Lower envelopes of various types of surfaces
      (a) The lower envelope of triangles.
      (b) The lower envelope of spheres.
      (c) The lower envelope of planes. (The minimization diagrams is
      drawn above the planes for clarity.)}}
  \label{fig:envelope}
\end{figure}
The \cEnvelope{} package of \cgal{}~\cite{m-rgece-06,cgal:mwz-e3-07}
computes the lower (or the upper) envelope of a set of surfaces in $\rrr$.
It is based on the \aos{} package, and like the base package, it handles
degenerate input, and produces exact results. An arrangement data-structure
is used to represent the resulting minimization diagram~\cite{h-a-04}.
The envelope computation is enabled by a traits class --- a model of
the concept \concept{EnvelopeTraits\_3}, which refines the
\concept{ArrangementTraits\_2} concept. The \cEnvelope{} package
currently contains three models of the \concept{EnvelopeTraits\_3}
concept that can be used to compute the envelope of triangles, spheres,
and planes, respectively; see Figure~\ref{fig:envelope} for an
illustration. Other models of the \concept{EnvelopeTraits\_3} concept
have been developed, for example, a traits class that enables the
computation of the envelope of a set of quadric surfaces~\cite{bm-ceq-07}.

\subsection{Voronoi Diagrams}
\label{ssec:aos:applications:voronoi}
Voronoi diagrams were thoroughly investigated, and were used to solve
many geometric problems, since introduced by Shamos and Hoey to the
field of computer science~\cite{shamos-hoey-75} (although their origin
dates back centuries ago; see~\cite{obsc-stcav-00}). The concept of
computing cells of points that are closer to a certain object than to
any other object, among a finite number of objects, was extended to
various kinds of geometric sites, ambient spaces, and distance
functions, e.g., power diagrams of circles in the plane,
multiplicatively weighted Voronoi diagrams, additively weighted
Voronoi diagrams~\cite{ak-vd-00}. This space decomposition is strongly
connected to arrangements~\cite{voronoi-es-86}, a property that yields
a very general approach for computing Voronoi diagrams.

Given a set of $n$ points $P = \{p_1, \ldots, p_n\}$, $p_i \in S$, we
define
$R(P, p_i) = \{x \in S \mid \distance{x}{p_i} < \distance{x}{p_j}, j \neq i\}$,
where $\distance{p_i}{p_j}$ is some given
\Index{distance function}.\footnote{In certain cases, the distance to
a site may depend on various parameters associated with the site,
e.g., in the cases of \mobius{} diagrams or anisotropic diagrams.}
$R(P, p_i)$ is the region of all points that are closer to $p_i$ then to
any other point in $P$. The Voronoi diagram\index{diagram!Voronoi} of $P$
over $S$ is defined to be the regions
$R(P, p_1), R(P, p_2), \ldots, R(P, p_n)$ and their boundaries.
Edelsbrunner and Seidel~\cite{voronoi-es-86} observed the connection
between Voronoi diagrams in $\reals^{d}$ and lower envelopes in
$\reals^{d+1}$ of the corresponding distance functions to the sites. From
the above definitions it is clear that if \mbox{$f_i: S \to \reals$} is set
to be $f_i(x) = \distance{x}{p_i}$, for $i = 1, \ldots, n$, then the
minimization diagram of $\{f_1, \ldots, f_n\}$ over $S$ is exactly the
Voronoi diagram of $P$ over $S$.

\begin{wrapfigure}[13]{r}{7.5cm}
  \vspace{-10pt}
  \begin{tabular}{cc}
    \epsfig{figure=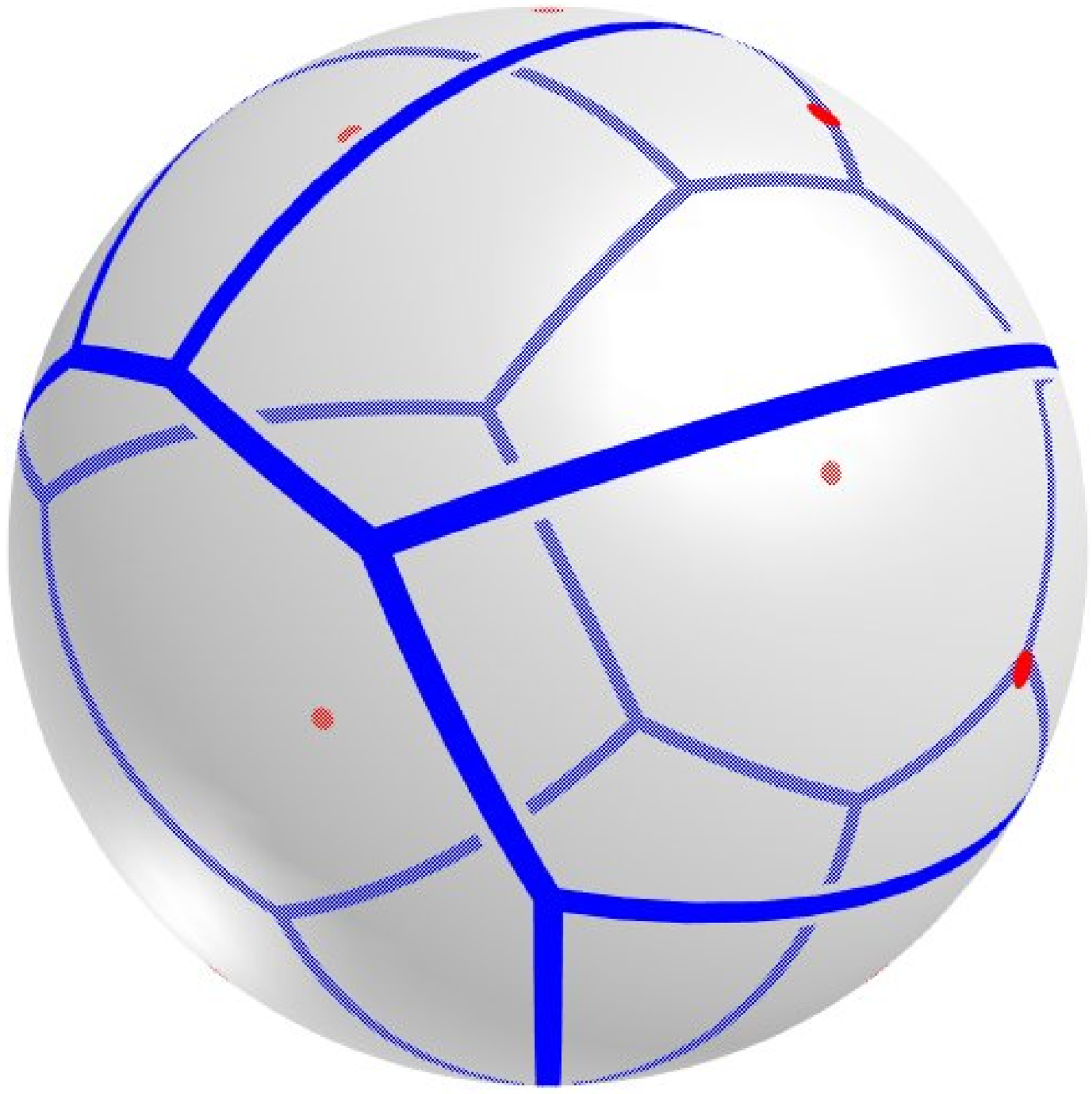,width=3.5cm,silent=} &
    \epsfig{figure=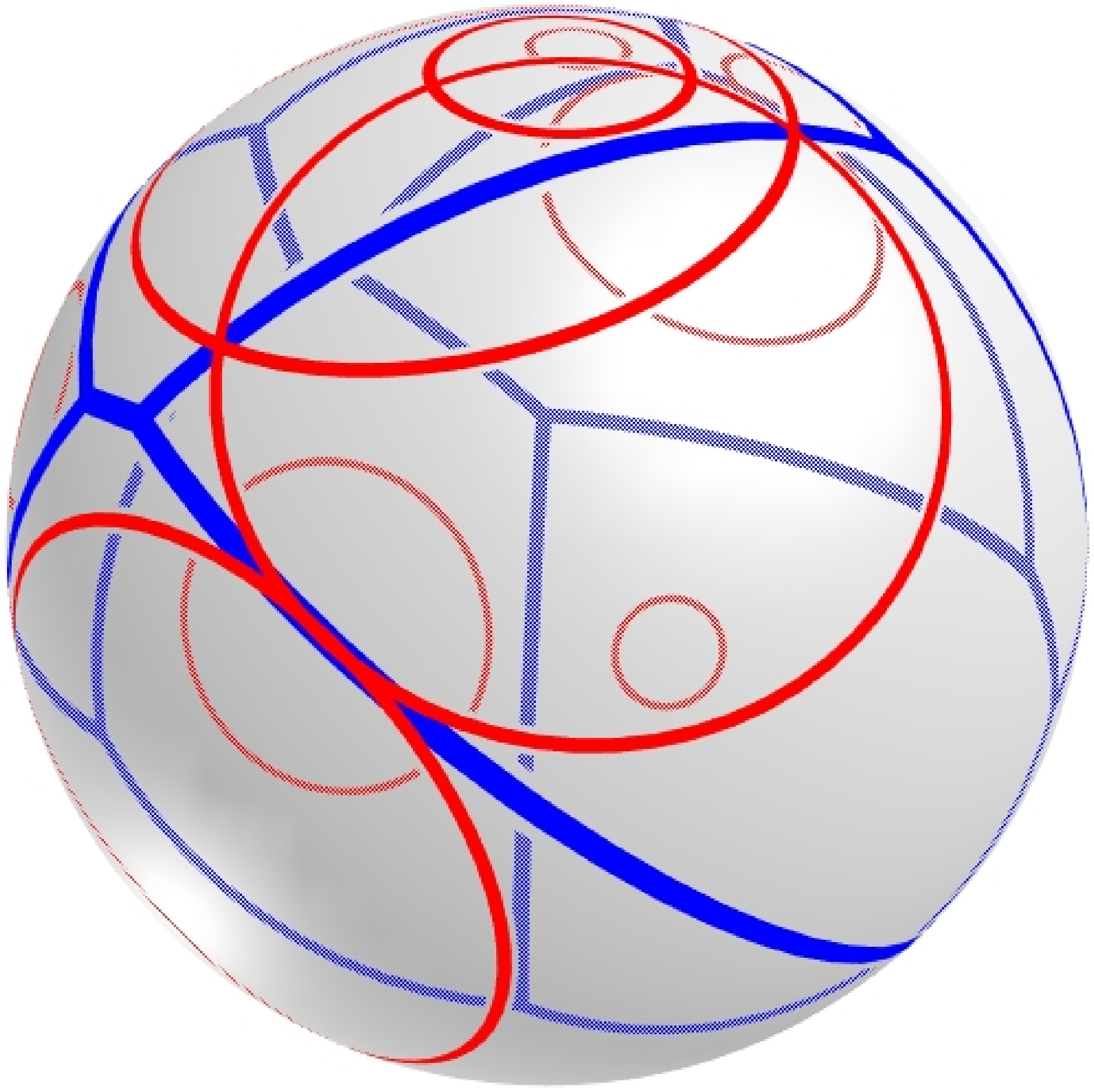,width=3.5cm,silent=}\\
    (a) & (b)
  \end{tabular}
  \caption[Voronoi diagrams on the sphere]{\capStyle{Voronoi diagrams
      on the sphere.
      (a) The Voronoi diagram of 14 random points. (b) The power
      diagram of 10 random circles.}}
  \label{fig:voronoi}
\end{wrapfigure}
K.\ E.\ Hoff \etal\ developed a technique for computing
generalized 2D and 3D Voronoi diagrams using interpolation-based
polygon rasterization hardware~\cite{hklmc-fcgvd-99}. The hardware is
used to draw the discrete and approximate lower envelope of the
site-distance functions. Following similar principles, O.\ Setter
\etal\ developed a new framework to compute different types of
Voronoi diagrams embedded on certain parametric surfaces in an exact
manner~\cite{shs-cdvdd-09}. The framework is based on the exact
computation of the lower envelope of the site-distance functions over
the surface~\cite{m-rgece-06}. It provides a reduced and convenient
interface between the construction of the diagrams and the
construction of envelopes, which in turn are computed using the
\cEnvelope{} package~\cite{cgal:mwz-e3-07}. Obtaining a new type of
Voronoi diagrams only amounts to the provision of a geometry-traits
class that handles the type of bisector curves of the new diagram
type. Essentially, every type of Voronoi diagram, the bisectors of
which can be handled by a geometry traits class, can be implemented
using this framework. In particular the geometry-traits class for
geodesic arcs embedded on the sphere enables Voronoi diagrams of
points on the sphere and their generalization, power diagrams, also
known as Laguerre Voronoi diagrams, on the sphere, as the bisector
curves between point sites on the sphere are great
circles~\cite{na-sphere-voronoi-02,obsc-stcav-00}, and so are the
bisectors between circle sites on the sphere under the Laguerre
distance~\cite{sugihara-sphere-02}; see Figure~\ref{fig:voronoi}.
Power diagrams on the sphere have several applications similar to the 
applications of power diagrams in the plane. For example, determining
whether a point is included in the union of circles on the sphere, and
finding the boundary of the union of circles on the
sphere~\cite{imai-85,sugihara-sphere-02}.

\pagebreak
\begin{wrapfigure}[7]{r}{4.2cm}
  \centerline{
    \epsfig{figure=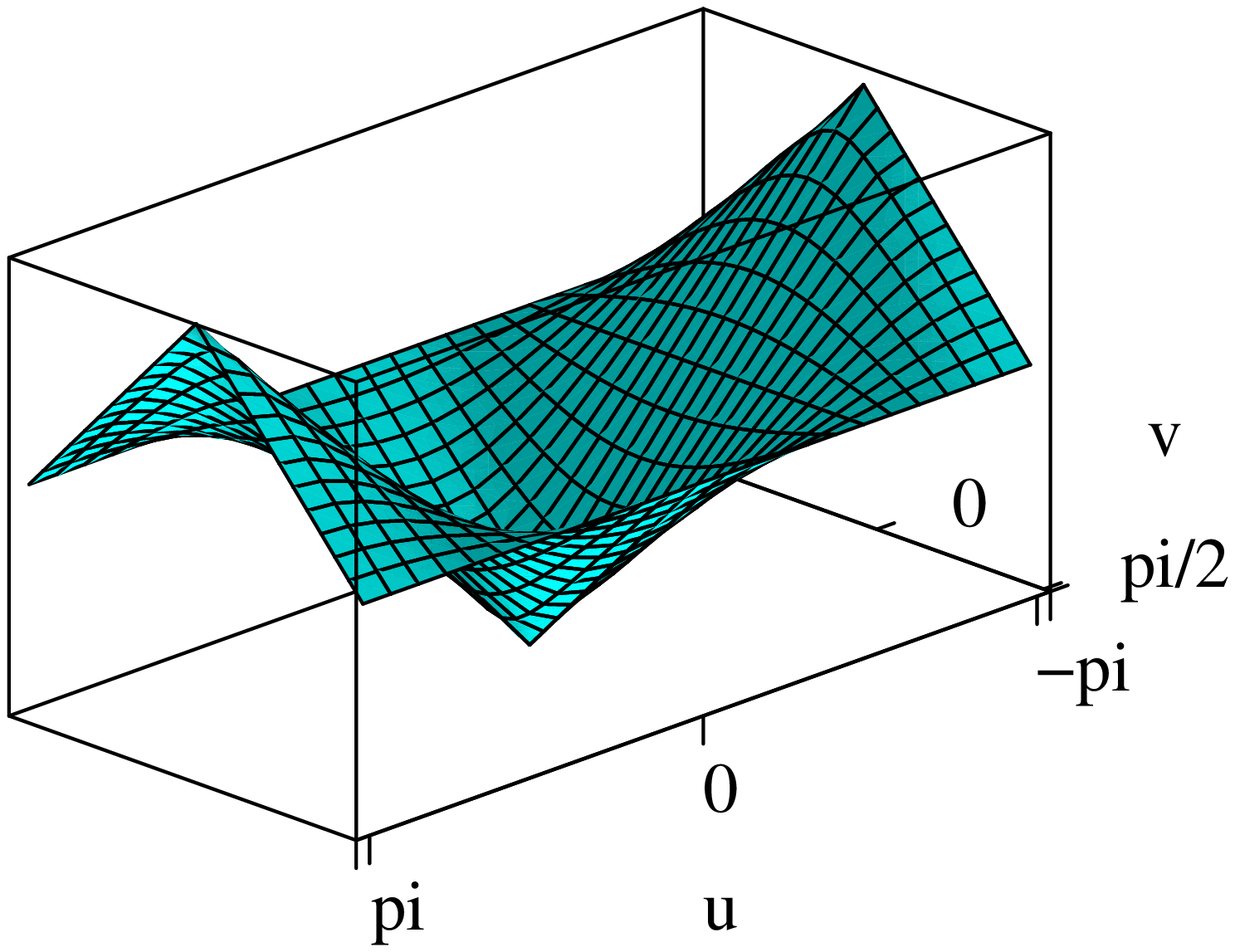,width=4cm,silent=}
  }
\end{wrapfigure}
We implicitly construct envelopes of distance functions defined over the
sphere to compute Voronoi diagrams. The image to the right illustrates
the distance function from
$(0, 0) \in [-\pi, \pi] \times [-\frac{\pi}{2}, \frac{\pi}{2}]$
on the sphere in the parameter space.
The great circle bisector of two point sites on the sphere is the
intersection of the sphere and the bisector plane of the points in
$\rrr$ (imposed by the Euclidean metric). 

If a point on the sphere is given as a general vector in $\rrr$,
it must be normalized. If the normalization results with a point with
irrational coordinates, then it must be approximated to a point that
lies exactly on the sphere~\cite{cdr-rrmrg-92}. Once approximated, all
further computations are carried out using exact rational arithmetic.

\begin{wrapfigure}{l}{7.5cm}
  \begin{tabular}{cc}
    \epsfig{figure=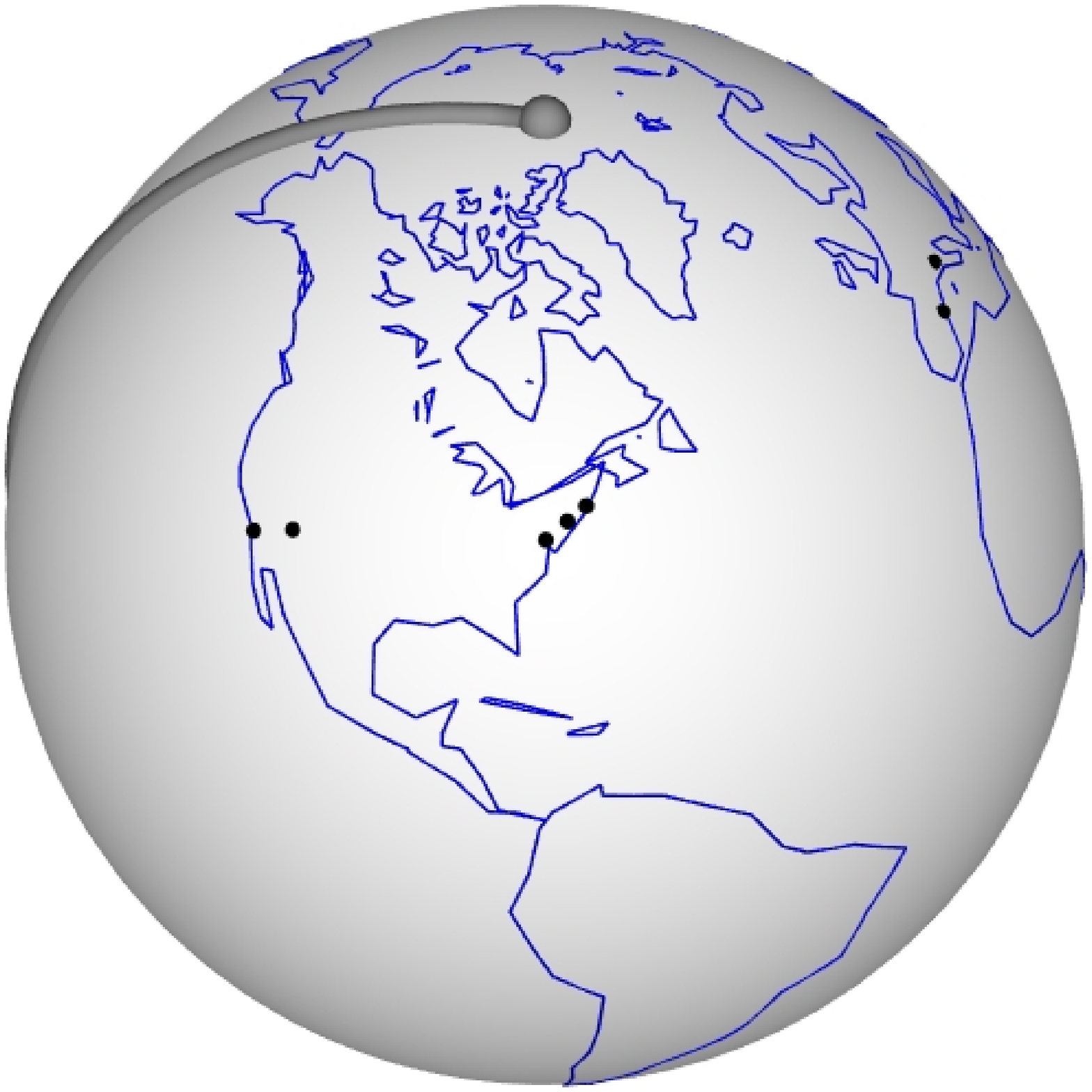,width=3.5cm,silent=} &
    \epsfig{figure=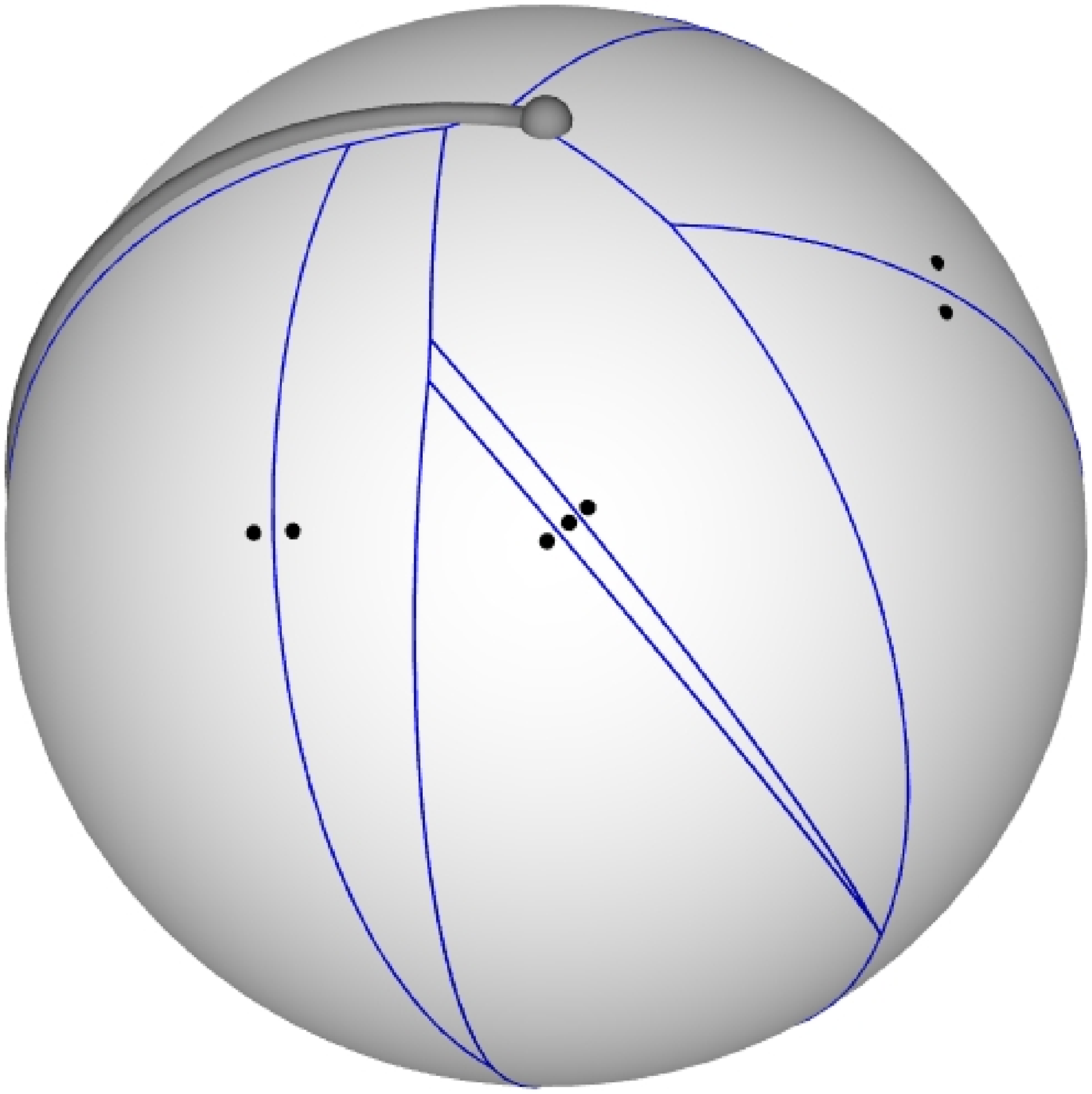,width=3.5cm,silent=}\\
    (a) & (b)
  \end{tabular}
  \caption[Arrangements on the sphere]{\capStyle{Arrangements on the sphere.}}
  \label{fig:world}
\end{wrapfigure}
Figure~\ref{fig:world} (a) on the left shows an arrangement on the
sphere induced by (i) the continents and some of the islands on earth,
and (ii) the institutions that hosted SoCG during this millennium,
which appear as isolated vertices. The sphere is oriented such that
College Park, MD, USA is at the center. The arrangement consists of
1054 vertices, 1081 edges, and 117 faces. The data was taken from
gnuplot~\citelinks{gnuplot} and google
maps~\citelinks{google-maps}. Figure~\ref{fig:world} (b) shows an
arrangement that represents the Voronoi diagram of the nine cities,
the institutions above are located at, namely
College Park,
Gyeongju,
Sedona,
Pisa,
New York, 
San Diego,
Barcelona,
Medford, and
Hong Kong.

\begin{wrapfigure}[8]{r}{3.5cm}
  \vspace{-12pt}
  \centerline{
    \epsfig{figure=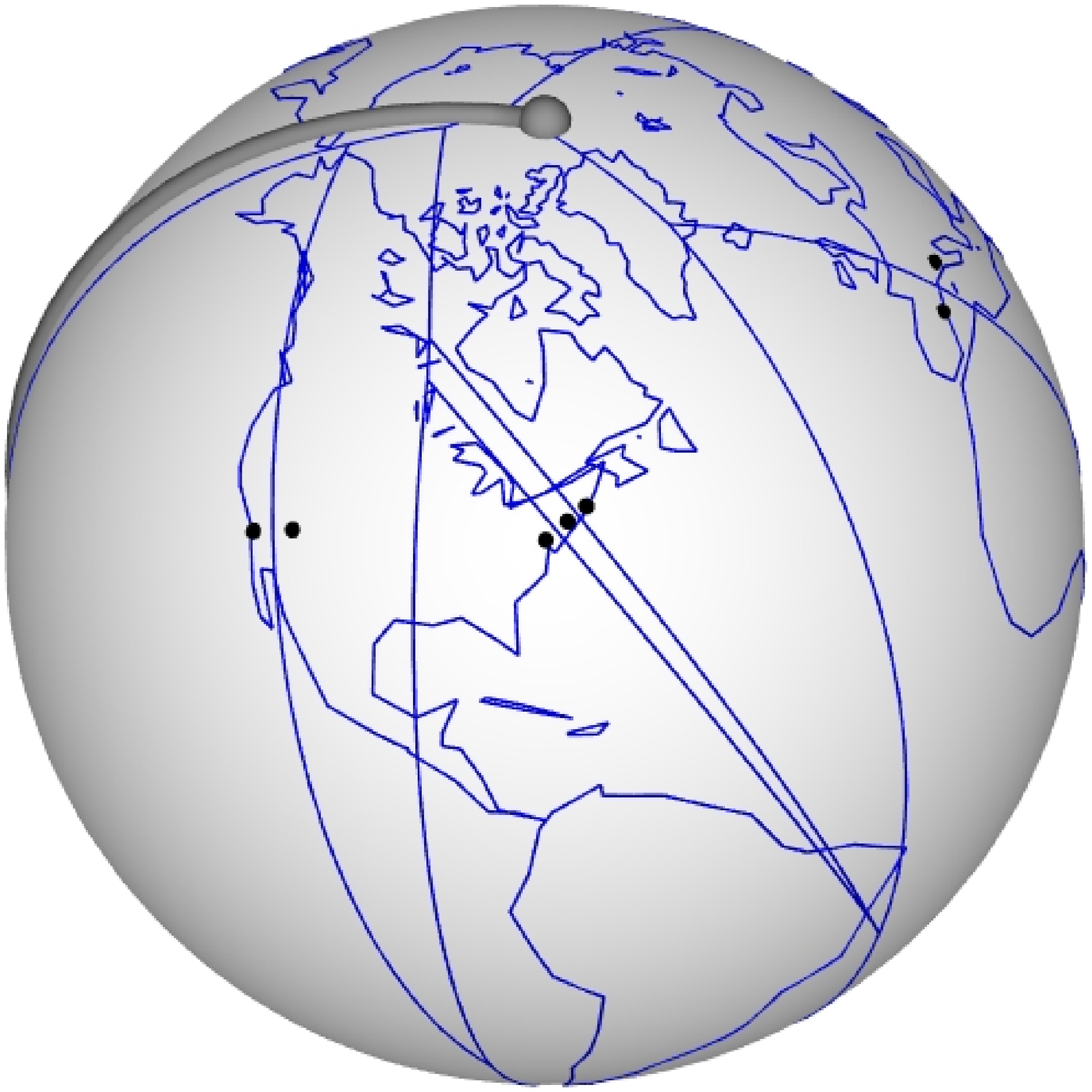,width=3.5cm,silent=}
  }
\end{wrapfigure}
As mention above Voronoi diagrams, among the other, are represented as
arrangements and can be passed as input to consecutive operations
on arrangements supported by the \aos{} package and its derivatives.
The figure on the right shows the \Index{overlay} of the two
arrangements shown in Figure~\ref{fig:world}.

\begin{savequote}[10pc]
\sffamily
Efficiency is just intelligent laziness.
\qauthor{Anonymous}
\end{savequote}
\chapter{Minkowski Sum Construction}
\label{chap:mink-sums-construction}
\newlength{\smallCellWidth}\setlength{\smallCellWidth}{0.7cm}
\newlength{\cellWidth}\setlength{\cellWidth}{0.9cm}

We present two exact and robust implementations of efficient
output-sensitive algorithms\index{algorithm!output sensitive} to
compute the \Index{Minkowski sum} of two polytopes in $\rrr$. We
demonstrate the effectiveness of our Minkowski-sum computations
through simple applications that exploit these operations to detect
collision, and compute the Euclidean \Index{separation distance}
between, and the directional \Index{penetration depth} of, two
polytopes in $\rrr$. In Chapter~\ref{chap:assem_plan} we show a more
involved application of these operations.

Each method we have developed uses a different variant of Gaussian maps
to maintain dual representations of polytopes. Each method employs a
different variant of two-dimensional arrangements\index{arrangement}
to maintain the dual representations, and it makes use of many
operations applied to arrangements in the corresponding
representations. The first method uses the traditional (spherical)
Gaussian map\index{Gaussian map!spherical}. The map is represented as
an arrangement of geodesic arcs embedded on the unit sphere. The
second method uses a data structure called {\em Cubical Gaussian
Map} \index{Gaussian map!cubical}. It consists of six arrangements
induced by linear segments embedded in the plane. The six arrangements
correspond to the six faces of the \Index{unit cube} --- the
parallel-axis cube circumscribing the \Index{unit sphere}. 

\newcounter{ms-const:cntr}
A simple method to compute the Minkowski sum of two polytopes is to
compute the \Index{convex hull} of the pairwise sum of the vertices of
the two polytopes. Although there are many implementations of various
algorithms to compute Minkowski sums and answer proximity queries, we
are unaware of the existence of complete implementations of methods to
compute exact Minkowski sums other than
\setcounter{ms-const:cntr}{1}(\roman{ms-const:cntr}) the naive method above, 
\addtocounter{ms-const:cntr}{1}(\roman{ms-const:cntr}) a method based on Nef
polyhedra embedded on the sphere~\cite{hkm-bosnc-07}, and
\addtocounter{ms-const:cntr}{1}(\roman{ms-const:cntr}) an implementation 
by Weibel~\citelinks{w-ms} of Fukuda's algorithm~\cite{f-zcmac-04}.
Both our methods exhibit much better performance than the other methods
in all cases, as demonstrated by the experiments reported in
Table~\ref{tab:mink-time}. Our methods well handle degenerate cases that
require special treatment when alternative representations are
used. For example, the case of two parallel facets facing the same
direction, one from each polytope, does not bear any burden on our
methods, and neither does the extreme case of two polytopes with
identical sets of facet normals.

The results of experimentation with a broad family of convex polyhedra are
reported. The relevant programs, source code, data sets, and documentation
are available at \url{http://www.cs.tau.ac.il/~efif/CD}.
A short movie~\cite{fh-emscp_05} that describes some of
the concepts of the cubical Gaussian-map method can be downloaded from
\url{http://acg.cs.tau.ac.il/projects/internal-projects/gaussian-map-cubical/Mink3d.avi}. Another short\linebreak
movie~\cite{fsh-agas-08} that describes some of
the concepts of the (spherical) Gaussian map method among the other can be
downloaded from
\url{http://acg.cs.tau.ac.il/projects/internal- }
\url{projects/arr-geodesic-sphere/movie/aos-xvid.avi}.

Both methods are implemented on top of \cgal{}, and are mainly based on
the arrangement package of the library (see Chapter~\ref{chap:aos}
and ~\cite{fwh-cfpeg-04,wfzh-aptac-07}),
although other parts, such as the polyhedral-surface package developed by
L.\ Kettner~\cite{k-ugpdd-99}, are used as well. In some cases it is
sufficient to build only portions of the boundary of the Minkowski sum of
two given polytopes to answer collision and proximity queries efficiently.
This requires locating the corresponding features that contribute to the
sought portion of the boundary. As both methods we have developed employ
two-dimensional arrangements implemented on top of the \cgal{} arrangement
package, we harness the ability to answer point-location queries
efficiently that comes along, to locate corresponding features of two
given polytopes.

The rest of this chapter is organized as follows. The {\em Gaussian
maps} dual representations of polytopes in $\rrr$ are described in
Section~\ref{sec:mscn:gauss_map} along with some of their properties.
In Sections~\ref{sec:mscn:sgm-method}~and~\ref{sec:mscn:cgm-method} we
show how 3D Minkowski sums can be computed efficiently, after the
summands are converted to (spherical) Gaussian maps and cubical
Gaussian maps, respectively. Section~\ref{sec:mscn:3d_col_det} presents 
an exact implementation of an efficient
collision-detection\index{collision detection} algorithm
under translation based on either of the dual representations.
In the last section,
\begin{wrapfigure}{r}{7.6cm}
  \vspace{-10pt}
  \centerline{
    \begin{tabular}{cc}
      \epsfig{figure=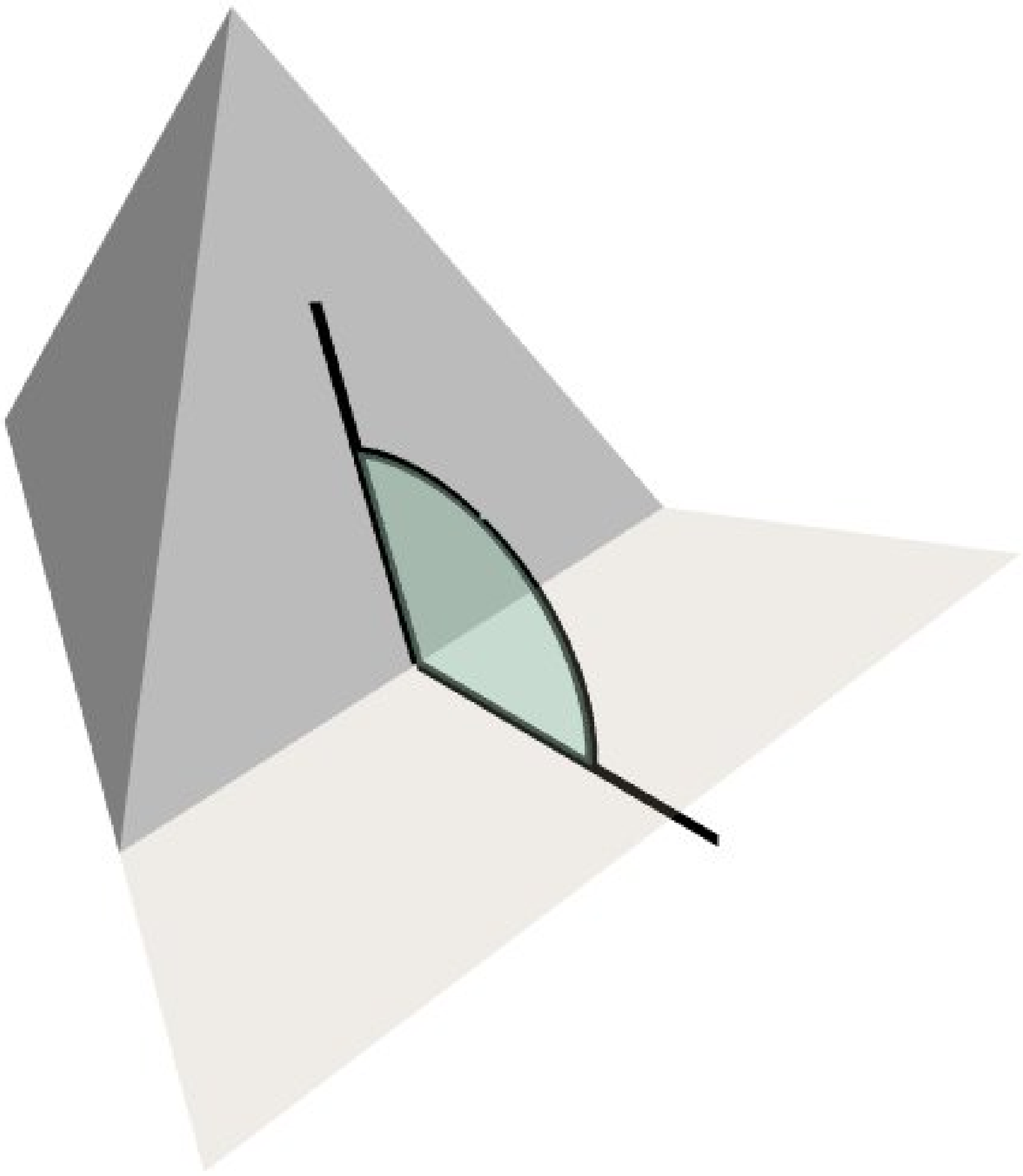,width=3.5cm,silent=} &
      \epsfig{figure=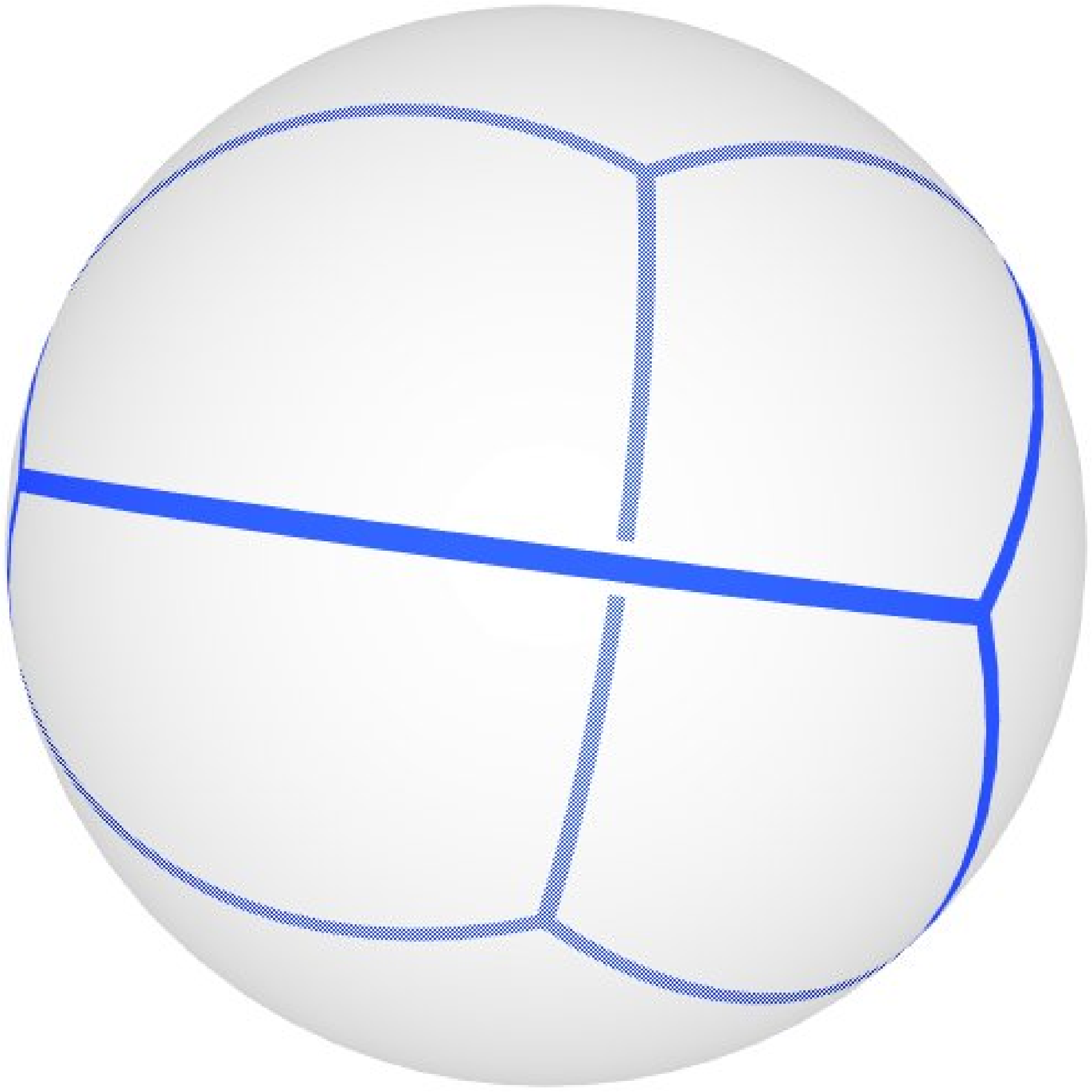,width=3.5cm,silent=}\\
      (a) & (b)\\
      \epsfig{figure=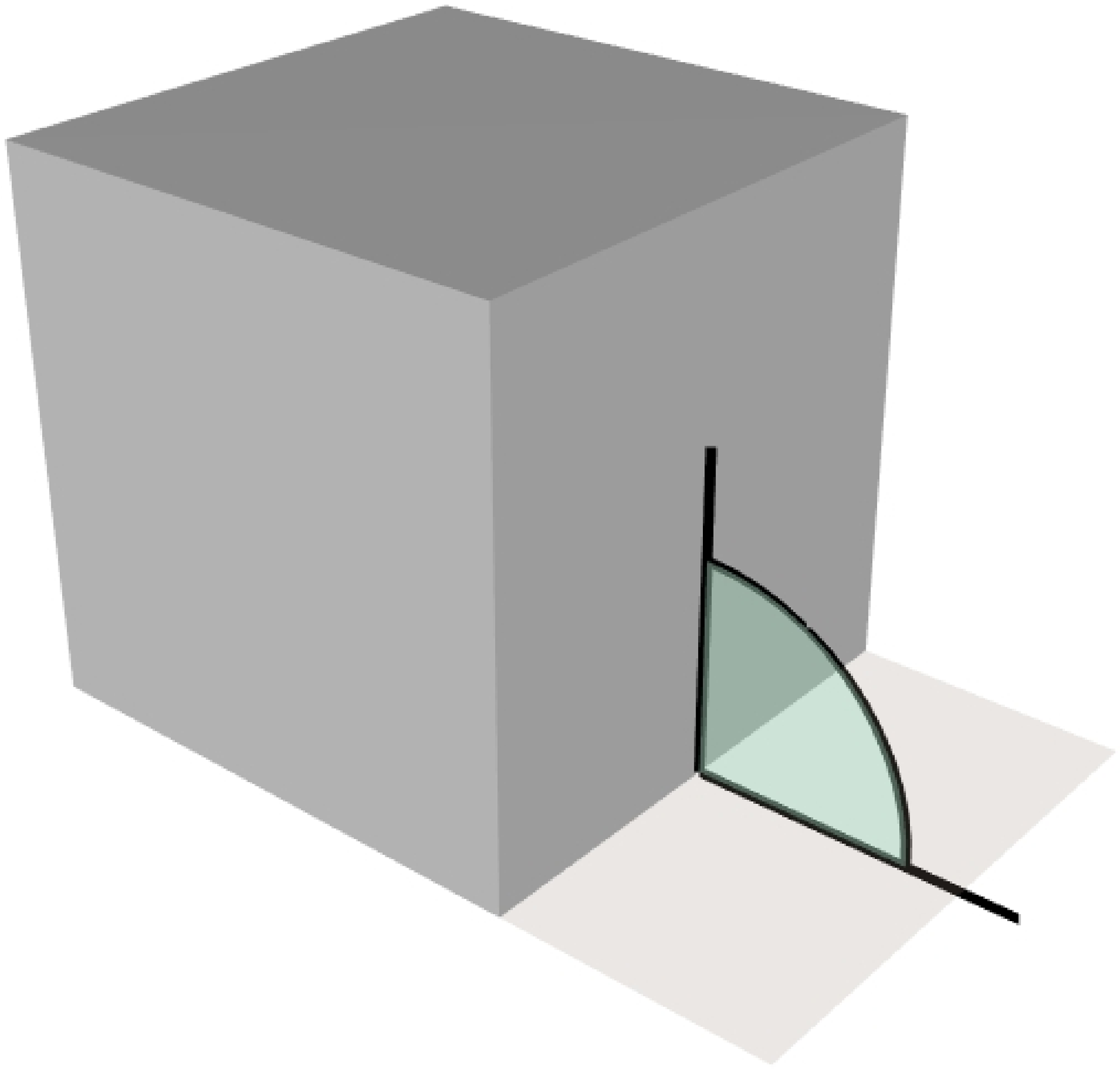,width=3.5cm,silent=} &
      \epsfig{figure=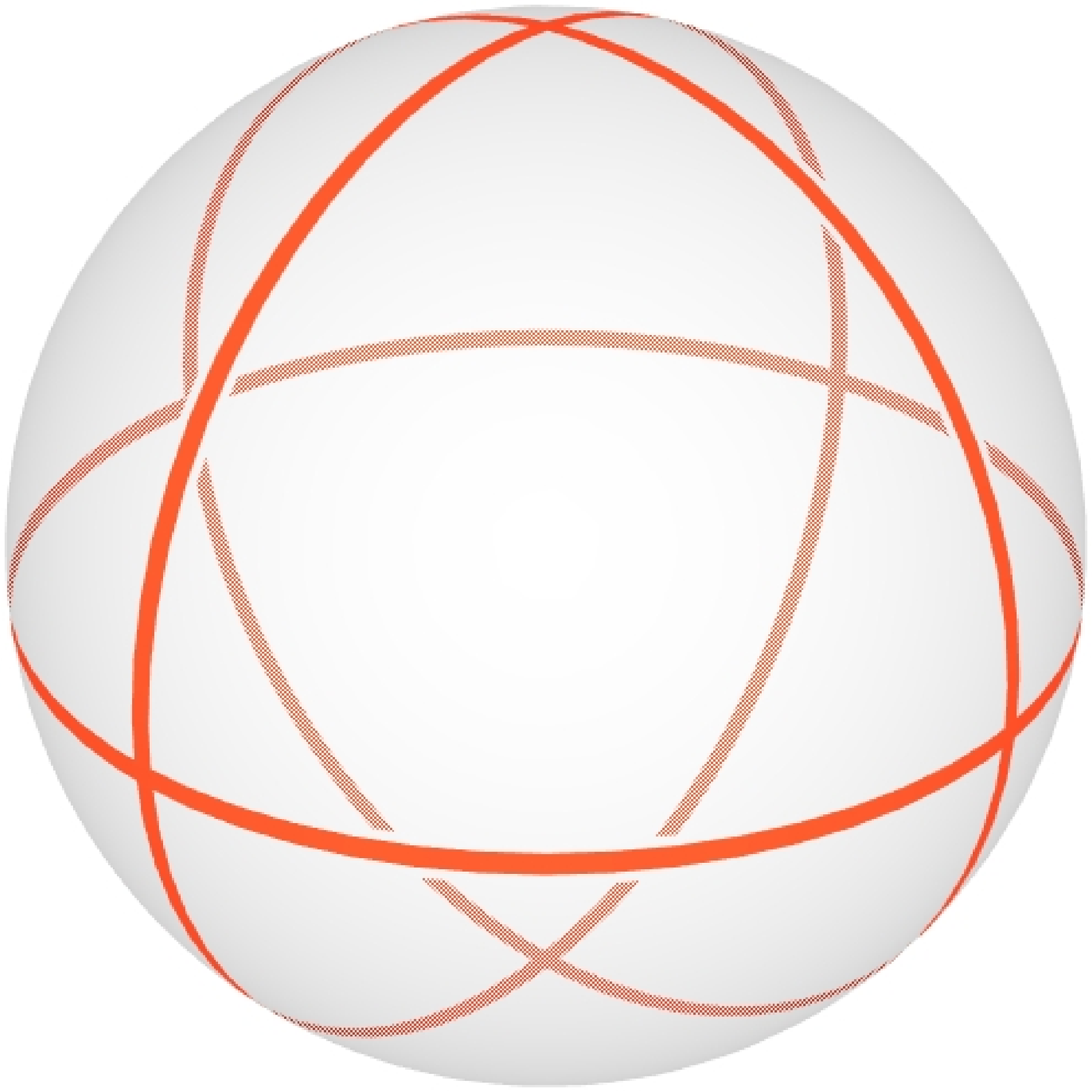,width=3.5cm,silent=}\\
      (c) & (d)\\
    \end{tabular}
  }
  \caption[Gaussian maps of polytopes]%
          {\capStyle{(a) A tetrahedron, (b) the Gaussian map of the
           tetrahedron, (c) a cube, and (d) the Gaussian map of the cube.}}
  \label{fig:gm}
  \vspace{-70pt}
\end{wrapfigure}
dedicated to experimental results, we highlight the
performance of our methods on various 
benchmarks. Suggestions for future directions
are provided in the conclusion chapter in Section~\ref{sec:conclusion:cd}.
The software access-information along with some further design details
are provided in the Appendix.

\section{Gaussian Maps}
\label{sec:mscn:gauss_map}
The {\em \Index{Gaussian map}} $G = G(P)$ of a compact convex
polyhedron $P$ in Euclidean three-dimensional space $\rrr$ is a set-valued
function from $P$ to the unit sphere $\spheretwo$, which assigns to
each point $p$ on the boundary of $P$ the set of outward unit normals
to support planes to $P$ at $p$. Thus, the whole of a facet $f$ of $P$
is mapped under $G$ to a single point, representing the outward unit
normal to $f$. An edge $e$ of $P$ is mapped to a (geodesic) segment $G(e)$
on $\spheretwo$, whose length is easily seen to be the exterior
dihedral angle at $e$. A vertex $v$ of $P$ is mapped by $G$ to a
spherical polygon $G(v)$, whose sides are the image under $G$ of
edges incident to $v$, and whose angles are the angles supplementary
to the planar angles of the facets incident to $v$; that is, $G(e_1)$
and  $G(e_2)$ meet at angle $\pi - \alpha$ whenever $e_1$ and $e_2$
meet at angle $\alpha$. In other words, $G(v)$ is exactly the
``spherical polar'' of the link of $v$ in $P$. (The link of a vertex
is the intersection of an infinitesimal sphere centered at $v$ with
$P$, rescaled, so that the radius is 1.) The above implies that $G(P)$
is combinatorially dual to $P$ and an arrangement embedded on the
unit sphere~\cite{hrs-cchpc-92}. Extending the mapping above,
by marking each face $f = G(v)$ of the arrangement with its dual vertex
$v$, enables a unique inverse Gaussian mapping, denoted by $G^{-1}$,
which maps an extended arrangement embedded on the unit sphere back to a
polytope boundary.

\begin{wrapfigure}[6]{r}{4.4cm}
  \vspace{-10pt}
  \centerline{
    \pspicture[](-2.2,-1.3)(2.2,1.5)
    \psset{unit=1cm,linewidth=0.5pt}
    \psline[linewidth=1pt](-1.4,0.4)(2.2,0.4)
    \pscircle*[linecolor=white](0,0){1.3}
    \pscircle[linewidth=1pt](0,0){1.2}
    \psarc[linecolor=white,linewidth=2pt](0,0){1.2}{-25}{0}
    \psarc[linecolor=white,linewidth=2pt](0,0){1.2}{180}{205}
    \psarcellipse[linewidth=1pt](0,0)(1.2,0.3){180}{0}
    \psline[linewidth=1pt](1.4,-0.4)(2.2,0.4)
    \psline[linewidth=1pt](-2.2,-0.4)(-1.4,0.4)
    \psline[linewidth=1pt](-2.2,-0.4)(1.4,-0.4)
    \rput{0}(0,1){
      \pspolygon*[linecolor=white](1.4,-0.4)(2.2,0.4)(-1.4,0.4)(-2.2,-0.4)
      \pspolygon[linewidth=1pt](1.4,-0.4)(2.2,0.4)(-1.4,0.4)(-2.2,-0.4)
      \rput{0}(1,0){$x_d = 1$}
    }
    \cnode*(0,0){2pt}{o}
    \cnode*(-0.7,0.42){2pt}{u}
    \cnode*(-1.5,0.9){2pt}{uhat}
    \pnode(-2.0,1.2){t}
    \ncline[linestyle=dotted,dotsep=2pt,linewidth=1pt]{o}{u}
    \ncline{u}{uhat}
    \ncline[arrowsize=3pt 3]{->}{uhat}{t}
    \uput[-135]{0}(-0.7,0.42){$u$}
    \uput[45]{0}(-1.5,0.9){$\hat{u_d}$}
    \uput[45]{0}(0,0){$o$}
    \endpspicture
  }
   \label{fig:central_projection}
\end{wrapfigure}
An alternative and practical definition follows. A direction in
$\mathbb{R}^3$ can be represented by a point $u \in \mathbb{S}^2$.
Let $P$ be a polytope in $\rrr$, and let $V$ denote the set of
its boundary vertices. For a direction $u$, we define the
{\em extremal point} in direction $u$ to be
$\lambda_V(u) = \arg \max_{p \in V}\langle u,p\rangle$, where
$\langle\cdot,\cdot\rangle$ denotes the inner product.
The decomposition of $\mathbb{S}^2$ into maximal connected regions, so that
the extremal point is the same for all directions within any region forms
the Gaussian map of $P$.
For some $u \in \mathbb{S}^2$ the intersection point of the ray
$\vec{ou}$ emanating from the origin with one of the planes
listed below is a {\em central projection} of $u$ denoted as $\hat{u_d}$,
and illustrated on the right.
The relevant planes are $x_d = 1,\, d = 1,2,3$, if $u$ lies in 
the positive respective hemisphere, and $x_d = -1,\, d = 1,2,3$ otherwise.

Similar to the Gaussian map, the {\em Cubical Gaussian Map} (\cgm)
$C = C(P)$ of a polytope $P$ in $\rrr$ is a set-valued
function from $P$ to the six faces of the unit cube whose edges are
parallel to the major axes and are of length two. A facet $f$ of $P$
is mapped under $C$ to a central projection of the outward unit normal
to $f$ onto one of the cube faces. Observe that, a single edge $e$ of
$P$ is mapped to a chain of at most four connected segments that lie
in four adjacent cube-faces respectively, and a vertex $v$ of $P$ is
mapped to at most five abutting convex dual faces that lie in five
adjacent cube-faces, respectively. The decomposition of the unit-cube
faces into maximal connected regions, so that the extremal point is
the same for all directions within any region forms the \cgm{} of
$P$. Likewise, the inverse \cgm, denoted by $C^{-1}$, maps the six
extended arrangement embedded on the six faces of the unit cube to the
polytope boundary. Figure~\ref{fig:tet} shows the \cgm{} of a tetrahedron.
\begin{figure*}[!htp]
  \centerline{%
    \begin{tabular}{ccc}
      \epsfig{figure=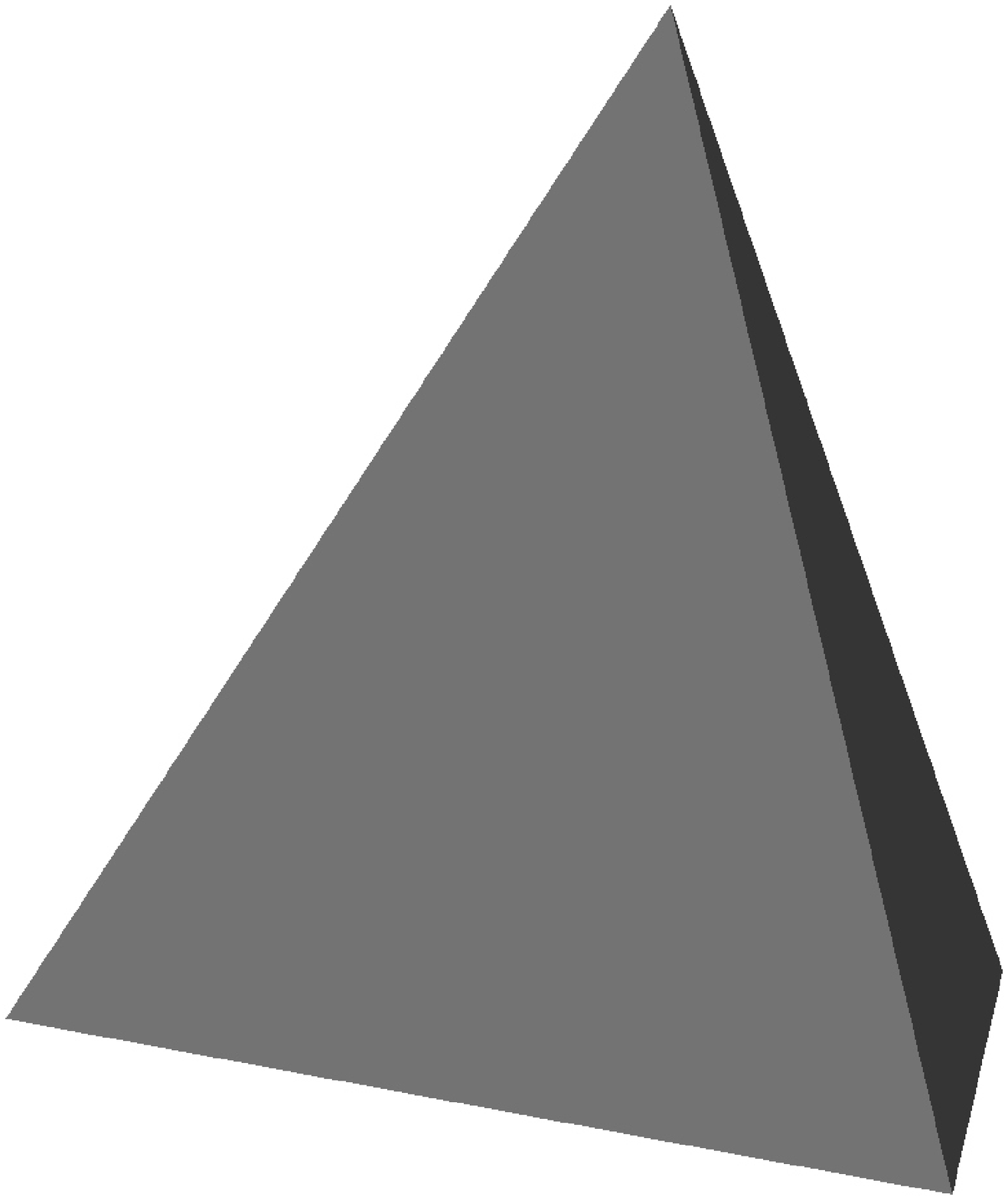,width=4cm,silent=} &
      \epsfig{figure=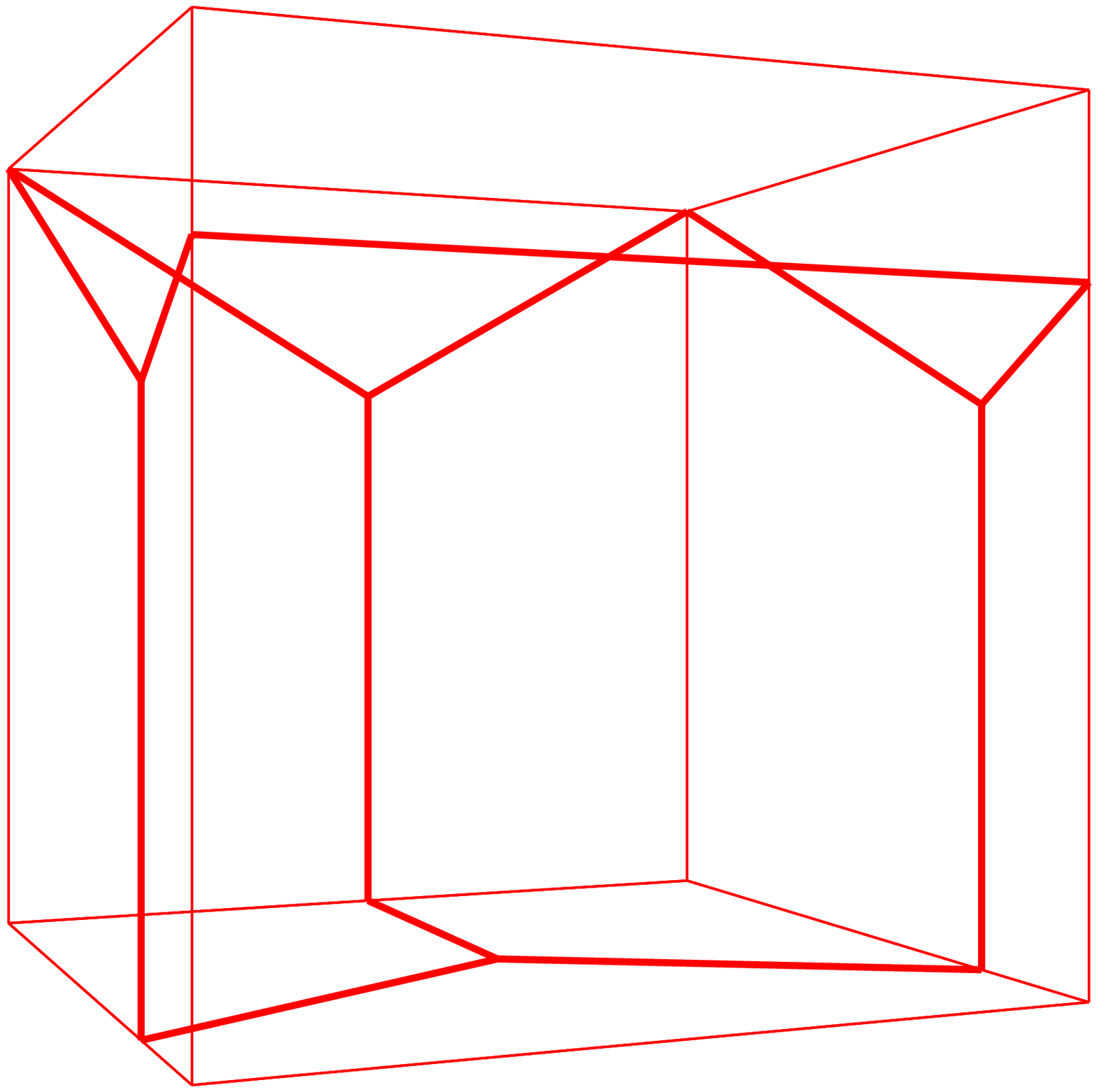,width=4cm,silent=} &
      \epsfig{figure=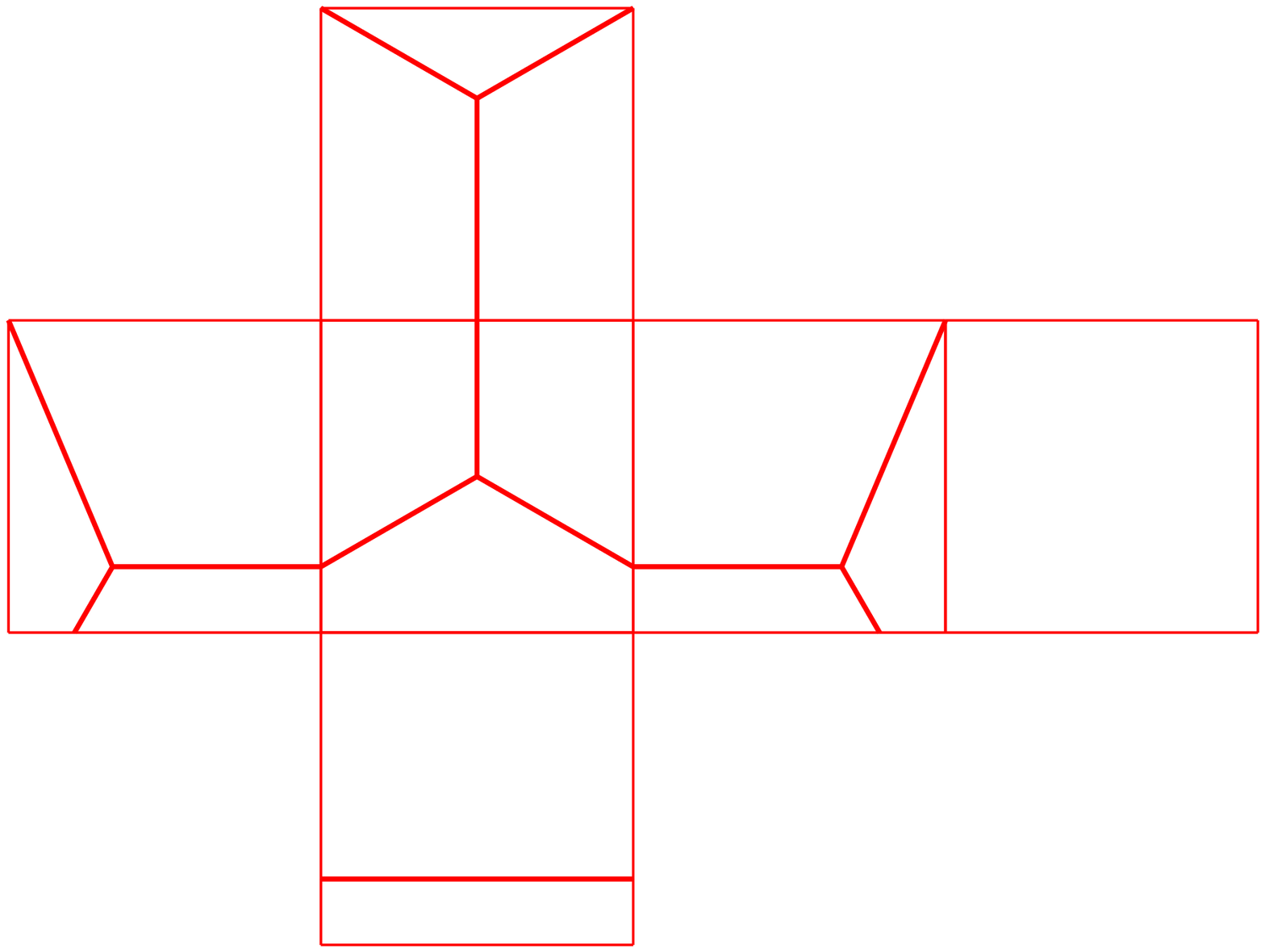,width=5cm,silent=}\\
      	(a) & (b) & (c)
    \end{tabular}
  }
  \caption[The cubical Gaussian map of a tetrahedron]%
           {\capStyle{(a) A tetrahedron, (b) the \cgm{} of the
           tetrahedron, and (c) the \cgm{} unfolded. Thick lines
           indicate real edges.}}
  \label{fig:tet}
\end{figure*}

\section{The (Spherical) Gaussian-Map Method}
\label{sec:mscn:sgm-method}
Armed with the geometry-traits class for geodesic arcs on the sphere
(see Section~\ref{sec:aos:geodesics}), we
compute Minkowski sums of convex polyhedra, by overlaying their
respective Gaussian maps represented as arrangements of geodesics on
the sphere. Each face $f$ of the Gaussian map is extended with
the coordinates of its dual vertex $v = G^{-1}(f)$, resulting with a
unique representation.

\subsection{The Representation}
\label{ssec:mscn:sgm:representation}
An input model of a polytope is provided as a polyhedral mesh in $\rrr$.
A polyhedral mesh representation consists of an array of boundary
vertices and the set of boundary facets, where each facet is described
by an array of indices into the vertex array. Constructing the Gaussian
map of a model given in this representation is done indirectly. First,
the \cgal{} \cPolyhedron~\cite{k-ugpdd-99} data-structure that represents
the model is constructed. This data structure provides quick
access to the incidence relations on the polytope features. Then, the
Gaussian map is constructed from the accessible information stored in
the  \cPolyhedron{} data-structure. Once the construction of the
Gaussian map is complete, the \cPolyhedron{} intermediate
representation is discarded.

\begin{figure}[!htp]
  \centerline{
    \begin{tabular}{cccc}
      \epsfig{figure=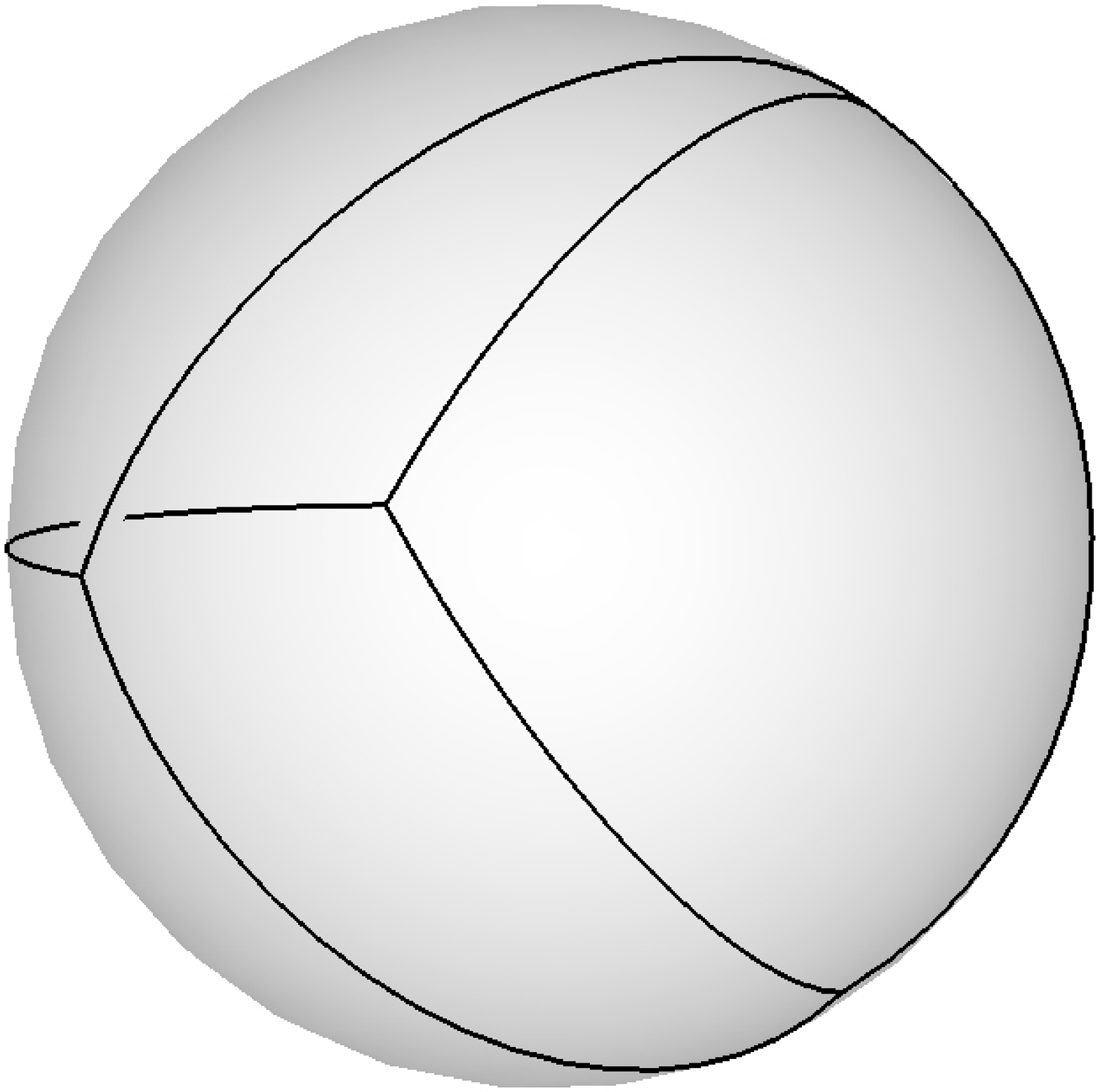,width=3cm,silent=} &
      \epsfig{figure=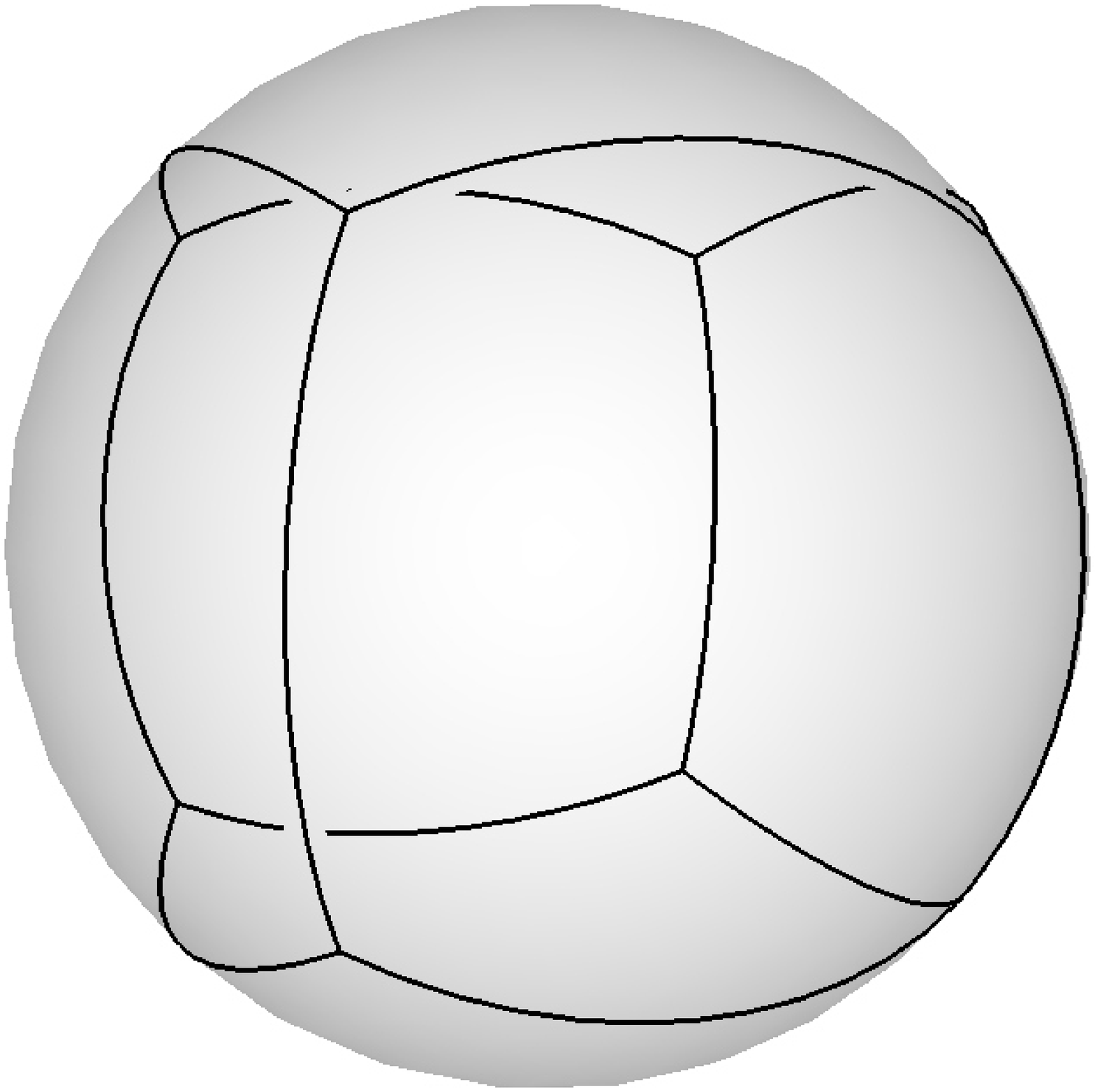,width=3cm,silent=} &
      \epsfig{figure=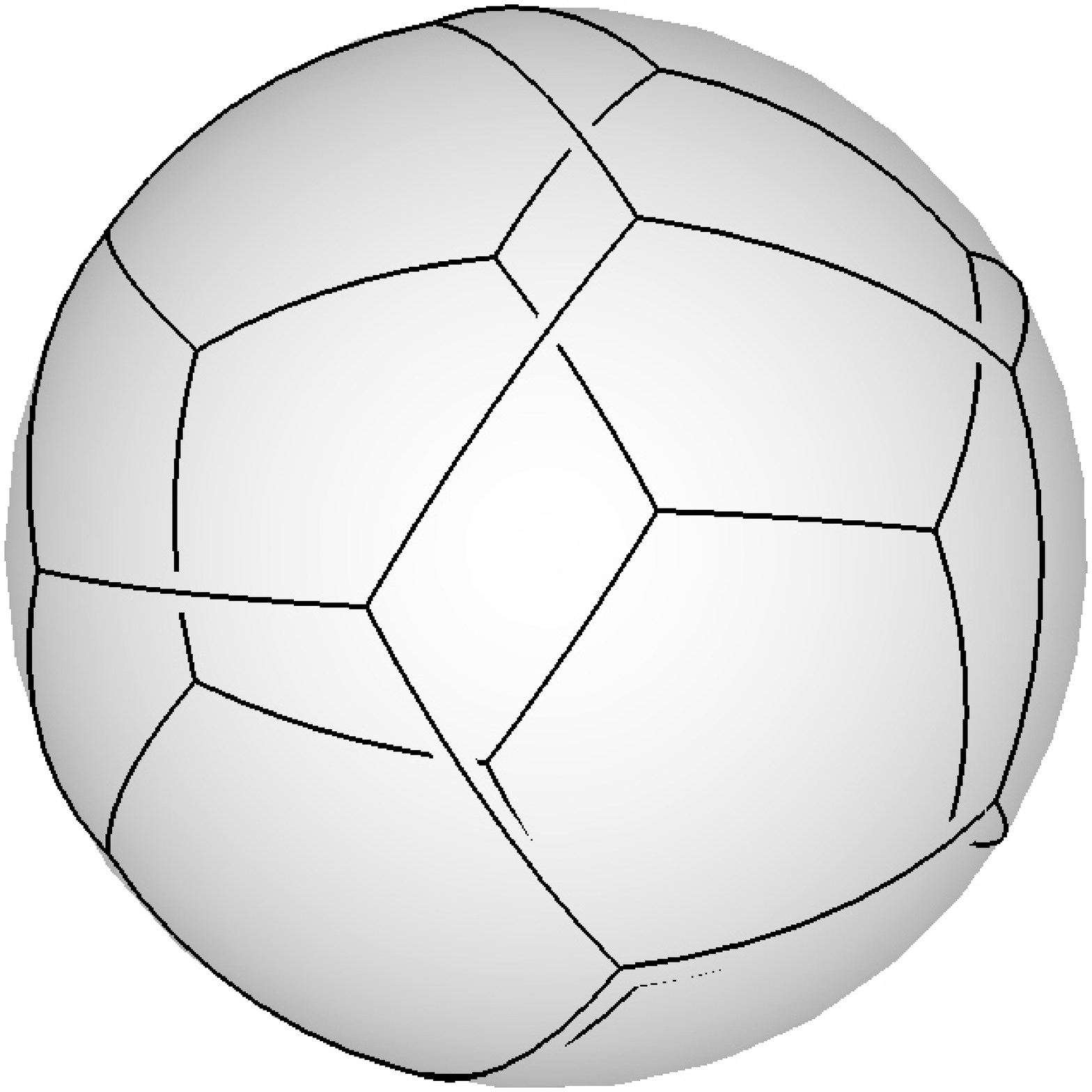,width=3cm,silent=} &
      \epsfig{figure=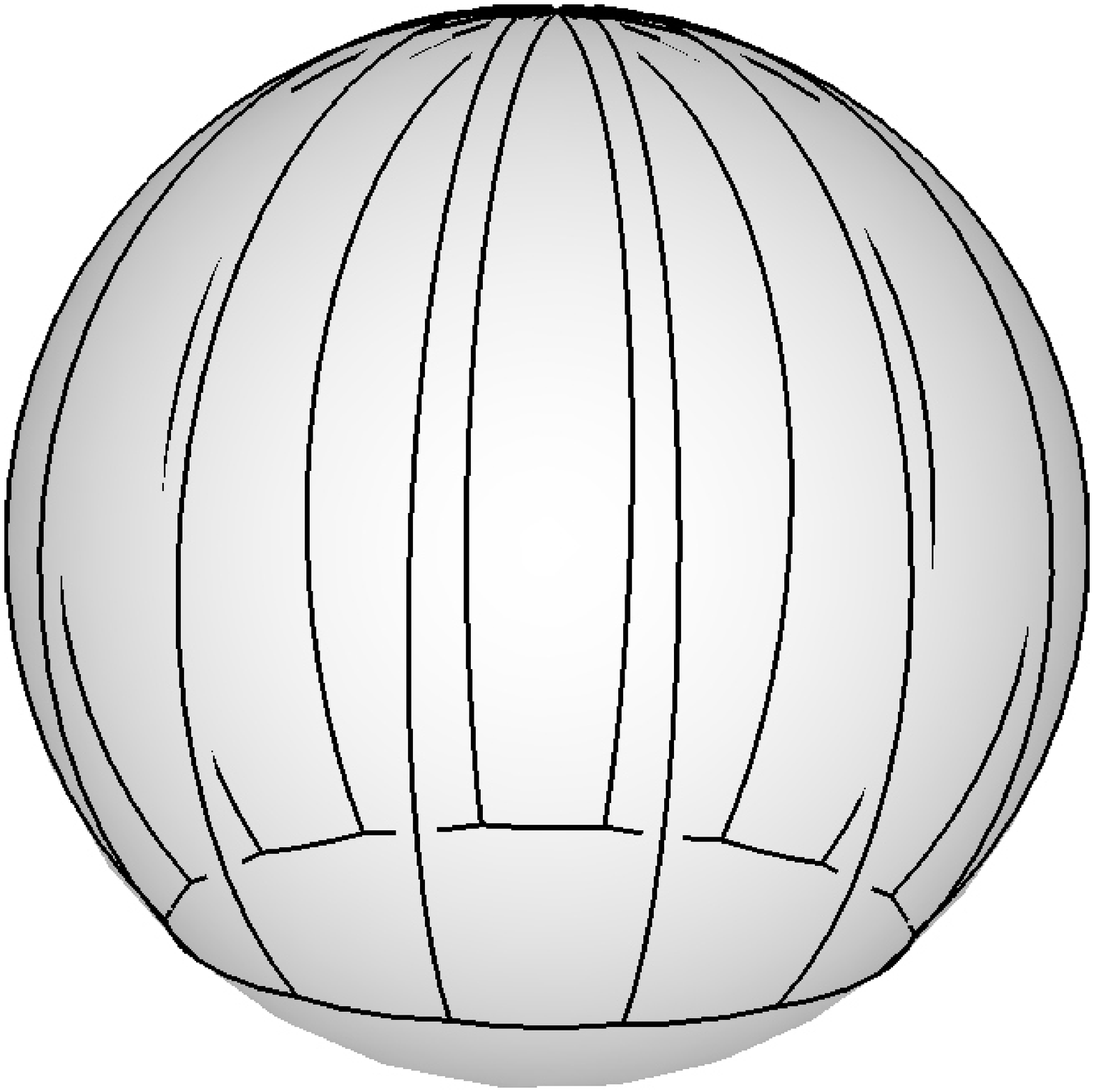,width=3cm,silent=}\\
      Tetrahedron & Octahedron & Icosahedron & Dioctagonal Pyramid\\
      {\small ~} & {\small ~} & {\small ~} & {\small ~} \\
      \epsfig{figure=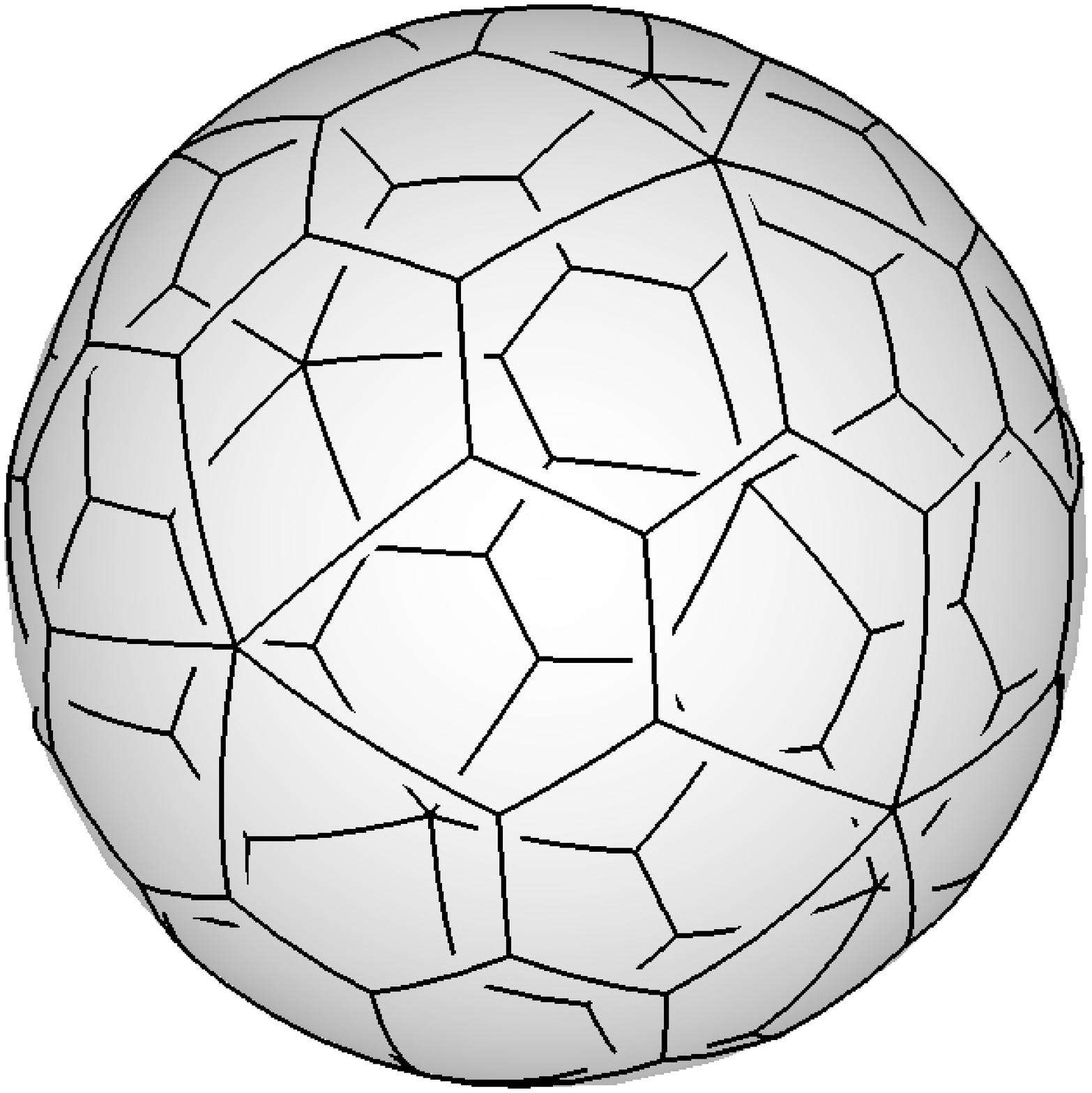,width=3cm,silent=} &
      \epsfig{figure=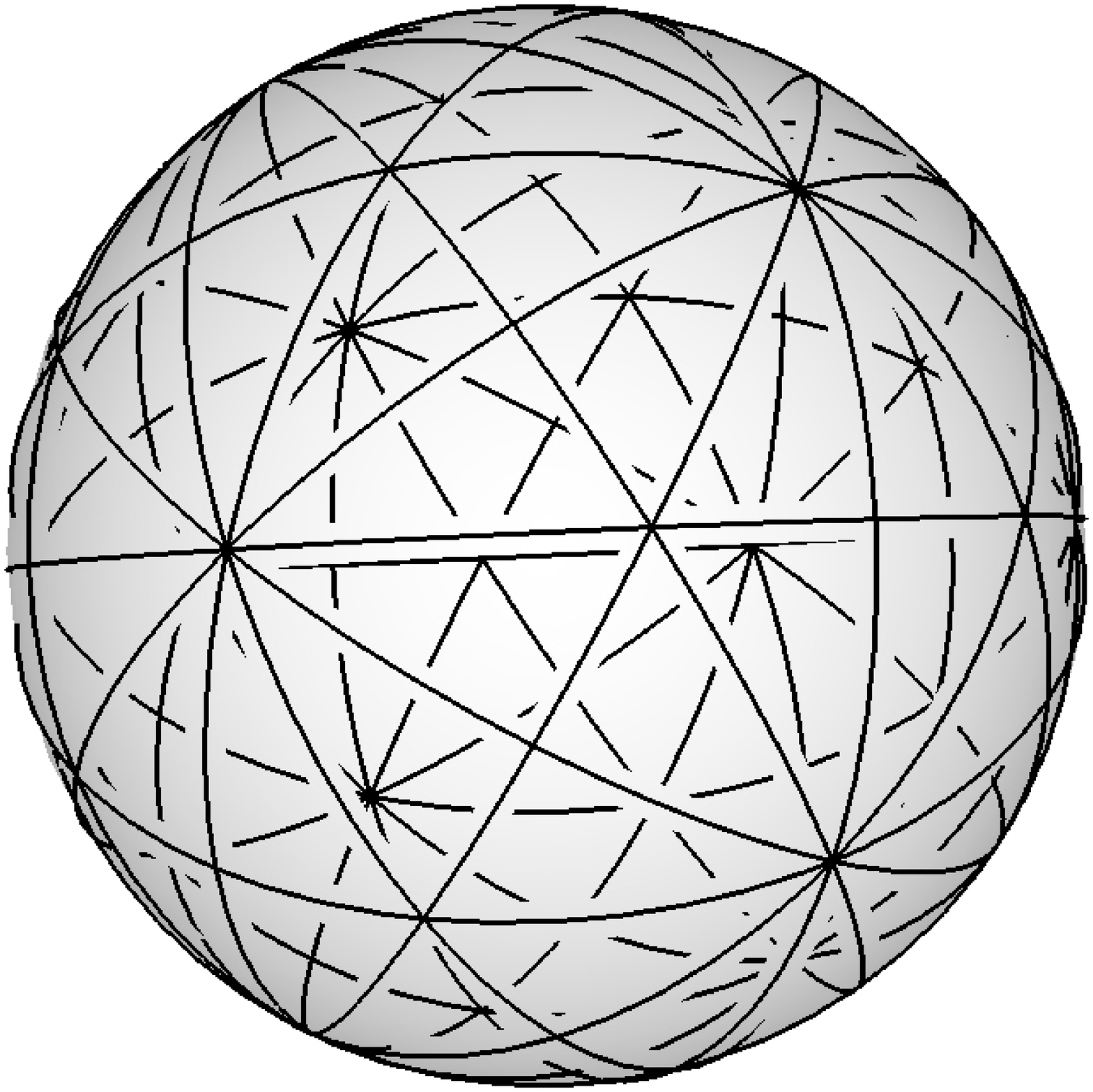,width=3cm,silent=} &
      \epsfig{figure=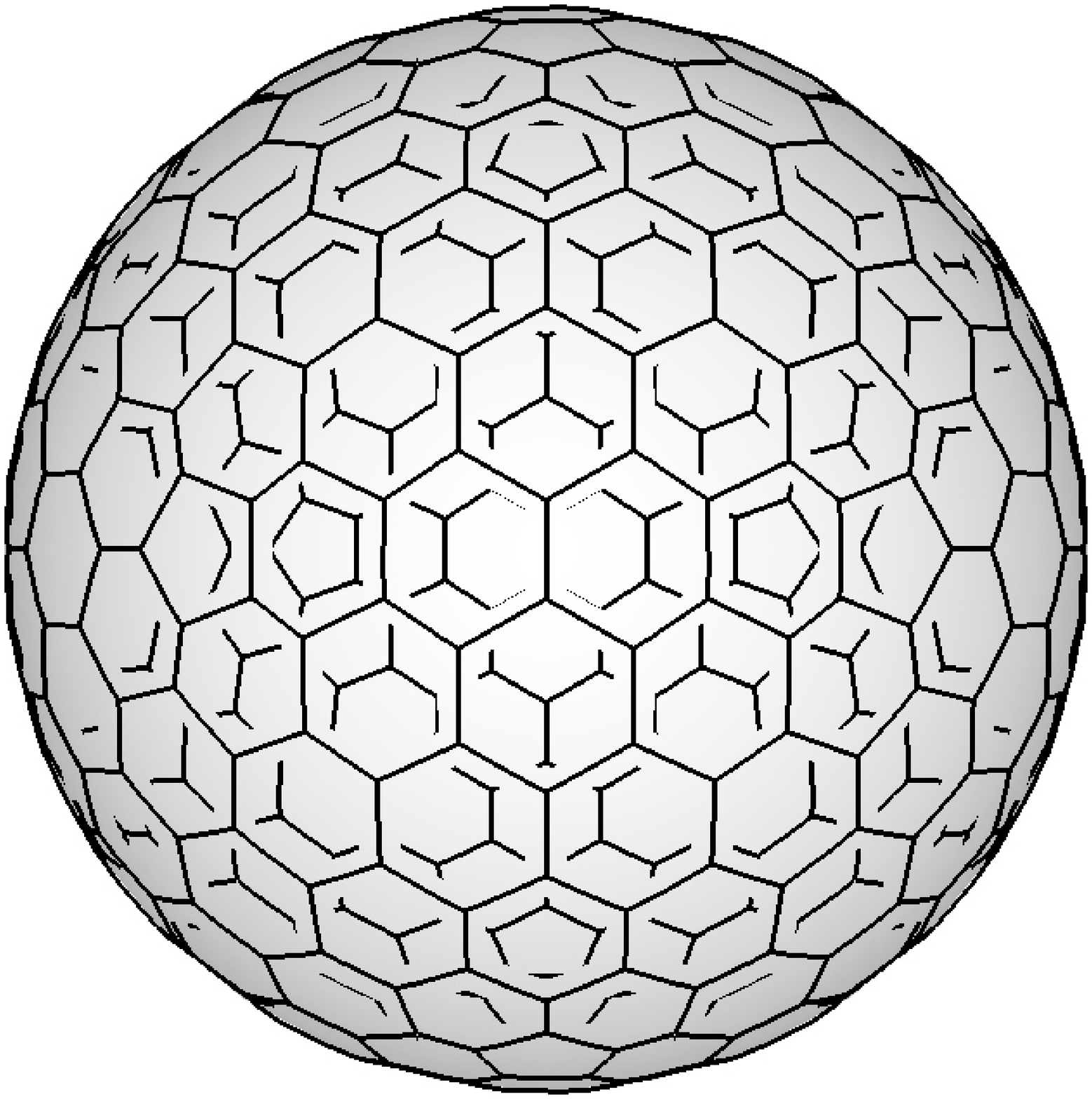,width=3cm,silent=} &
      \epsfig{figure=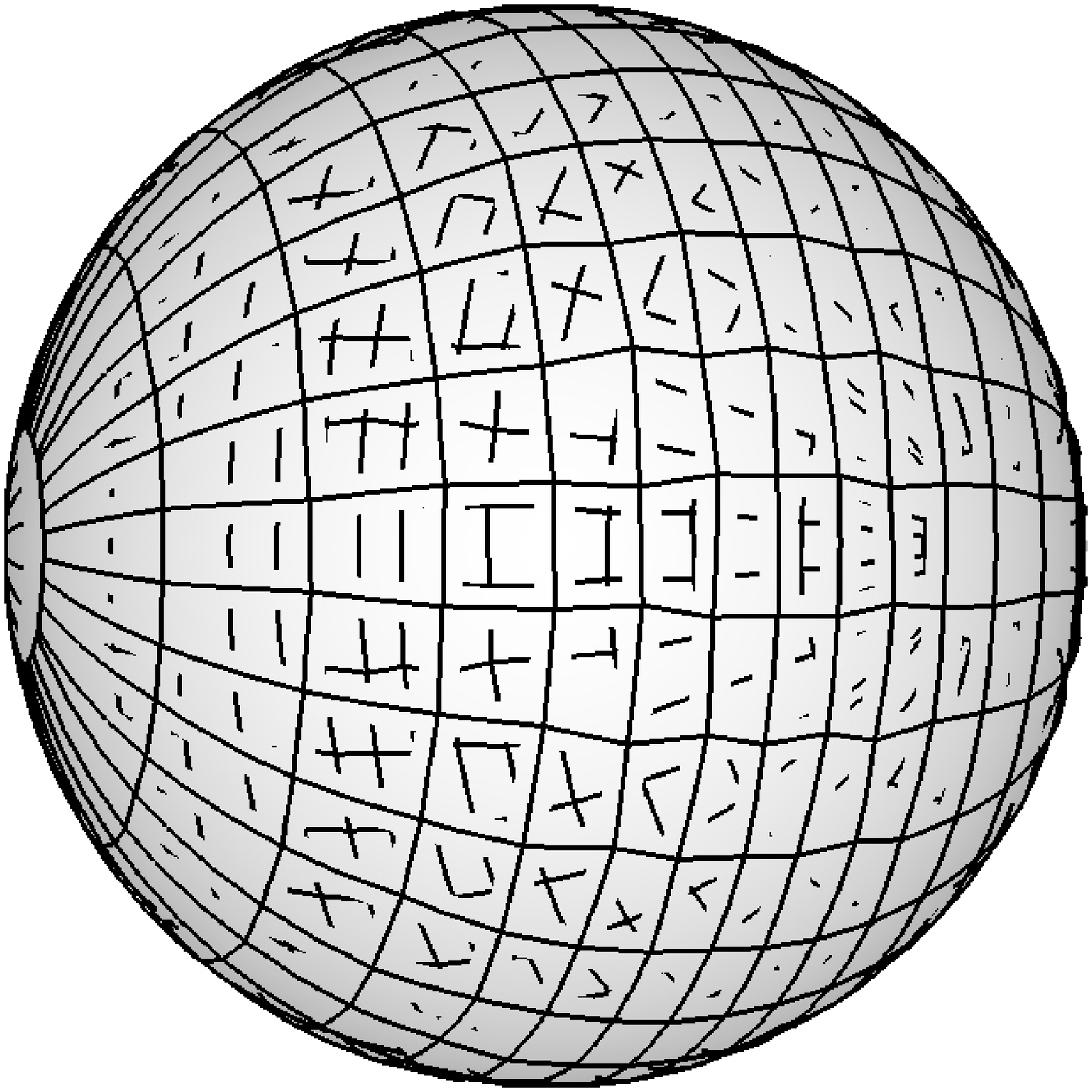,width=3cm,silent=}\\
      Pentagonal      & Truncated         & Geodesic       & Ellipsoid \\
      Hexecontahedron & Icosidodecahedron & Sphere level~4 & like polyhedron
    \end{tabular}
  }
  \caption[Gaussian maps of various polytopes]{\capStyle{Gaussian maps
    of various polytopes.}}
  \label{fig:gaussian-map}
\end{figure}

The \cPolyhedron{} data-structure, like the arrangement
\dcel\index{DCEL@\dcel}, is
based on the implementation of an \hds\index{HDS@\hds}; see
Section~\ref{sec:aos:architecture}. It consists of extendible
vertices, halfedges, and facets and incidence relations on them. It
provides methods to traverse all vertices, halfedges, and facets, and
for local traversals, such as traversing all halfedges incident to a
specific vertex. It also provides quick access from one halfedge
to its twin, and to the incident facet to its left. Each vertex and
halfedge of  the \cPolyhedron{} data-structure is extended with a
Boolean flag that indicates whether the vertex or the halfedge
respectively have already been processed during the construction of
the arrangement that represents the Gaussian map. Each facet is
extended with the handle of an arrangement vertex. The handle
extension of a facet $f$ that has already been processed points to 
the dual vertex $v = G(f)$. The procedure that converts an
(extended) \cPolyhedron{} data-structure $P$ representing a
polytope to an (extended) \aos{} data-structure $G(P)$ consists of
two steps. First all field extensions of all vertices, halfedges, and
facets of $P$ are cleared. Then, a recursive function provided with
an arbitrary vertex $v$ of $P$ as a single parameter is invoked.
This function traverses the halfedges incident to $v$. When it
encounters an unprocessed halfedge $e$, it obtains the normal $n_1$
and the vertex-handle extension $h(v_1)$ of the facet $f_1$ adjacent to
the left of $e$ and the normal $n_2$ and the vertex-handle extension
$h(v_2)$ of the facet $f_2$ adjacent to the left of the next halfedge
in the cyclic chain of halfedges incident to $v$. Finally, the geodesic
short arc between $n_1$ and $n_2$ is constructed and inserted into the
arrangement $G(P)$, as explained below, using one of the efficient
insertion member-functions of \aos{}; see
Section~\ref{ssec:aos:architecture:member-ops}. Once the insertion is
complete, the halfedge $e$ is marked as processed. $v$ is marked as
processed once all its incident halfedges are processed. The function
recursively invokes itself providing an unprocessed vertex adjacent
to $v$, and terminates when no such vertex is found.

Let $C$ indicate the new geodesic arc to be inserted into the
arrangement representing the Gaussian map. Assume that $C$ is $u$-monotone
with respect to the parameterization defined by the geometry-traits class;
see Section~\ref{sec:aos:geodesics}. That is, it does not intersect the
identification arc. Let $v_1$, $v_2$, $f_1$, and $f_2$ be the two vertices
and two facets as described above. There are four cases to handle as follows.
\begin{compactenum}
\item
  If $v_1$ and $v_2$ are both null, it implies that this is the
  first attempt to insert a geodesic arc into the arrangement. In this
  case we call \ccode{insert\_in\_face\_interior($C$,$f$)}, where $f$ is
  the single face the \dcel{} was initialized with; see
  Section~\ref{sec:aos:geodesics}. The handle of the two new vertices
  associated with the endpoints of the newly created geodesic arc $C$
  are stored in the records of the corresponding facets $f_1$ and
  $f_2$ of $P$ respectively for later use.
\item
  If $v_1$ is null but $v_2$ is not, we call either 
  \ccode{insert\_from\_left\_vertex($C$,$v_2$)} or\linebreak
  \ccode{insert\_from\_right\_vertex($C$,$v_2$)} depending on whether
  the existing vertex $v_2$ is to the right or to the left of $C$, and
  update the vertex-handle field of the corresponding facet $f_1$ with
  the new vertex $v_1$.
\item
  We handle the analogous case where $v_2$ is null but $v_1$ is not
  similarly.  
\item
  If both $v_1$ and $v_2$ are not null, we call
  \ccode{insert\_at\_vertices($C$,$v_1$,$v_2$)}. In this case no
  vertex-handle field needs to be updated.
\end{compactenum}
If $C$ intersects the identification arc, it is first split at the
intersection point into two $u$-monotone arcs $C_1$ and $C_2$. Then,
$C_1$ and $C_2$ are inserted according to four cases similar to the
above. The handling is a bit more intricate. For example, consider
the case where $v_1$ is null but $v_2$ already exists. Assume that
$C_1$ reaches the right boundary of the parameter space and $C_2$
reaches the left boundary (they both meet at an identified point).
If $v_2$ is associated with the left endpoint of $C_1$, we first call
\ccode{insert\_from\_left\_vertex($C_1$,$v_2$)}, and then
\ccode{insert\_from\_left\_vertex($C_2$,$v'$)}, where $v'$ is the new
vertex associated with the right endpoint of $C_1$ introduced while
$C_1$ is inserted. Otherwise, $v_2$ must be associated with the right
endpoint of $C_2$. In this case we first
call \ccode{insert\_from\_right\_vertex($C_2$,$v_2$)}, and
then \ccode{insert\_from\_right\_vertex($C_1$,$v'$)}, where $v'$ is the
new vertex associated with the right endpoint of $C_2$ introduced
while $C_2$ is inserted. The other three cases are handled similarly.

We have created a large database of models of polytopes.
Table~\ref{tab:rep} lists, for a small subset of our polytope
collection, the number of features in the arrangement of geodesic arcs
embedded on the sphere that represents the Gaussian map of each
polytope. Recall that the number of faces ({\bf F}) of the Gaussian
map is always equal to the number of vertices of the polytope. However,
the number of vertices ({\bf V}) of the Gaussian map is either equall to,
or greater than, the number of facets of the primal representation due to
intersections between Gaussian-map edges and the identification arc. A
similar argument holds for the edges. That is, the number of halfedges
({\bf HE}) of the Gaussian map is either equall to, or greater than,
twice the number of edges of the primal representation. An edge of the
Gaussian map that intersects the identification arc must be split at the
intersection point into two $u$-monotone geodesic arcs. The table also
lists the time in seconds ({\bf t}) it takes to construct the arrangement
once the intermediate polyhedron is in place, on a Pentium PC clocked at
1.7~GHz.

\subsection{Exact Minkowski Sums}
\label{ssec:mscn:sgm:mink_sum}
\begin{wrapfigure}[12]{r}{7.6cm}
  \vspace{-20pt}
  \centerline{
    \begin{tabular}{cc}
      \epsfig{figure=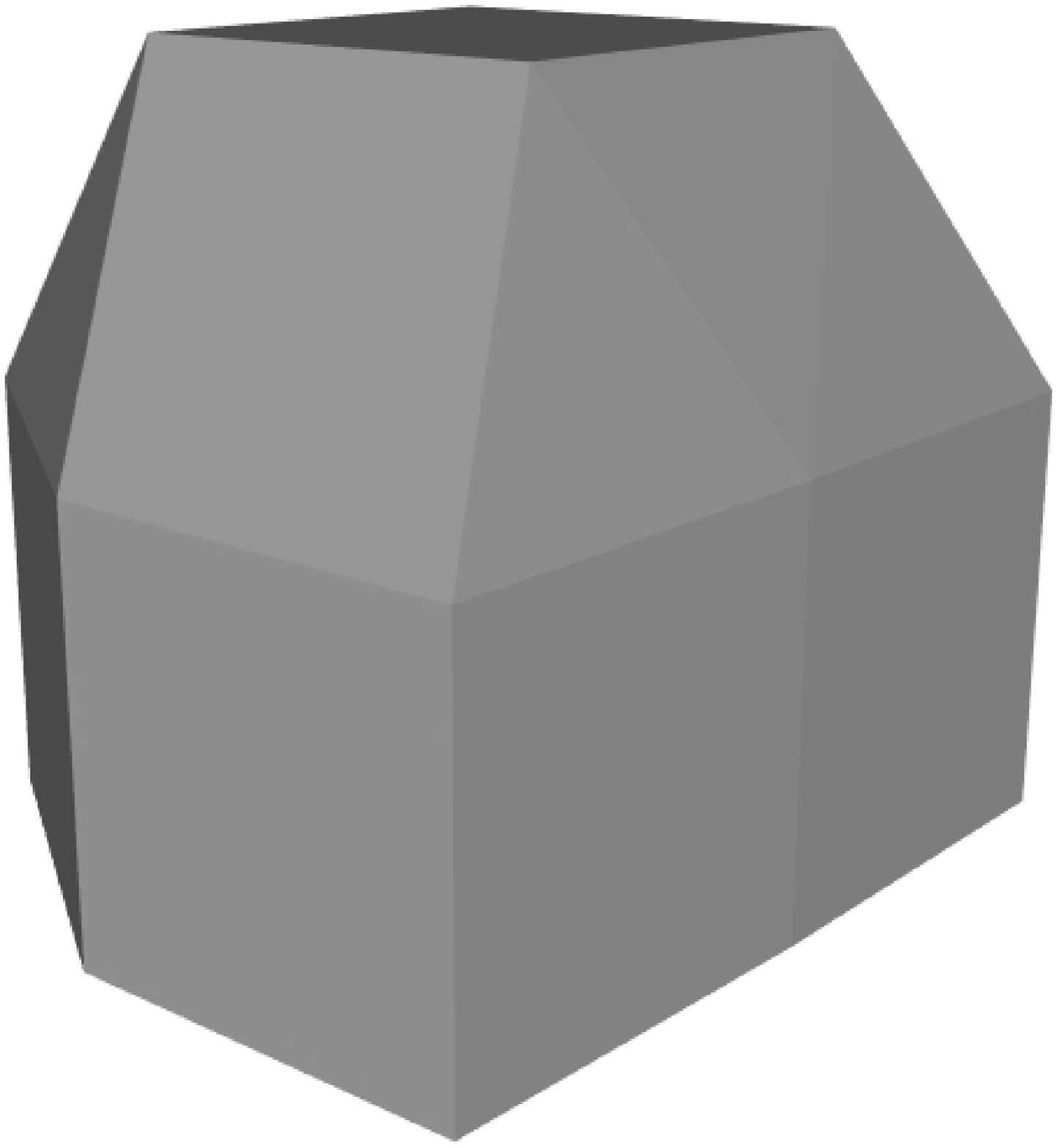,width=3.2cm,silent=} &
      \epsfig{figure=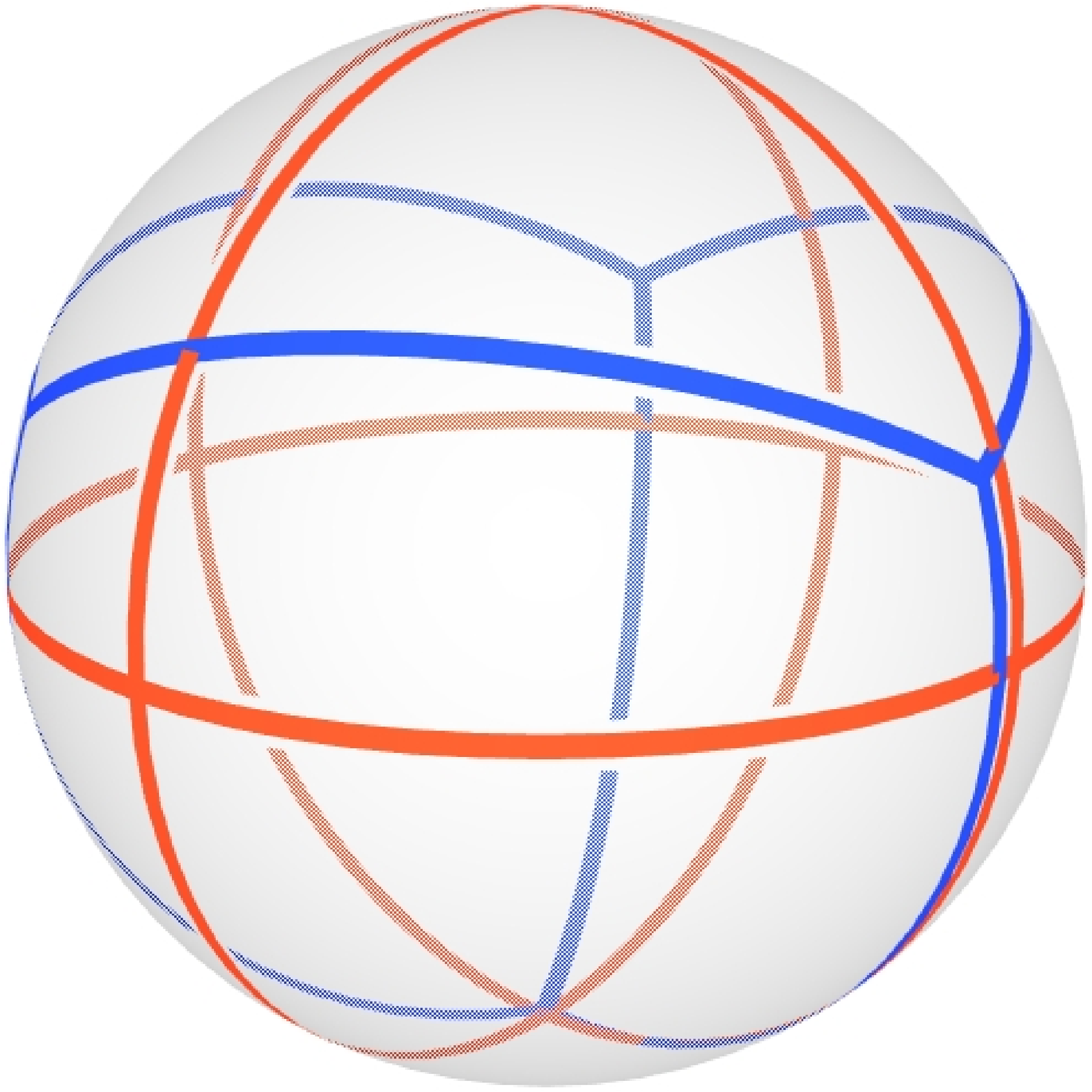,width=3.5cm,silent=}\\
      (a) & (b)
    \end{tabular}
  }
  \caption[The Minkowski sum of two polyhedra]%
           {\capStyle{(a) The Minkowski sum of a tetrahedron and a cube and
            (b) the Gaussian map of the Minkowski sum.}}
  \label{fig:mink-gm}
  \vspace{-20pt}
\end{wrapfigure}
The overlay (see Section~\ref{ssec:aos:facilities:overlay} for the
exact definition) of the Gaussian maps of two polytopes $P$ and $Q$
respectively identifies all pairs of features of $P$ and $Q$ that
have parallel supporting planes, as they occupy the same space on the
unit sphere, thus, identifying all the pairwise features that
contribute to the boundary of the Minkowski sum of $P$ and $Q$. A
facet of the Minkowski sum is either a facet $f$ of $Q$ translated by
a vertex of $P$ supported by a plane parallel to $f$, or vice versa,
or it is a facet parallel to two parallel planes supporting an edge of
$P$ and an edge of $Q$, respectively. A vertex of the Minkowski sum is
the sum of two vertices of $P$ and $Q$ respectively supported by
parallel planes.

\begin{figure}[!htp]
  \centerline{
    \begin{tabular}{cccc}
      \epsfig{figure=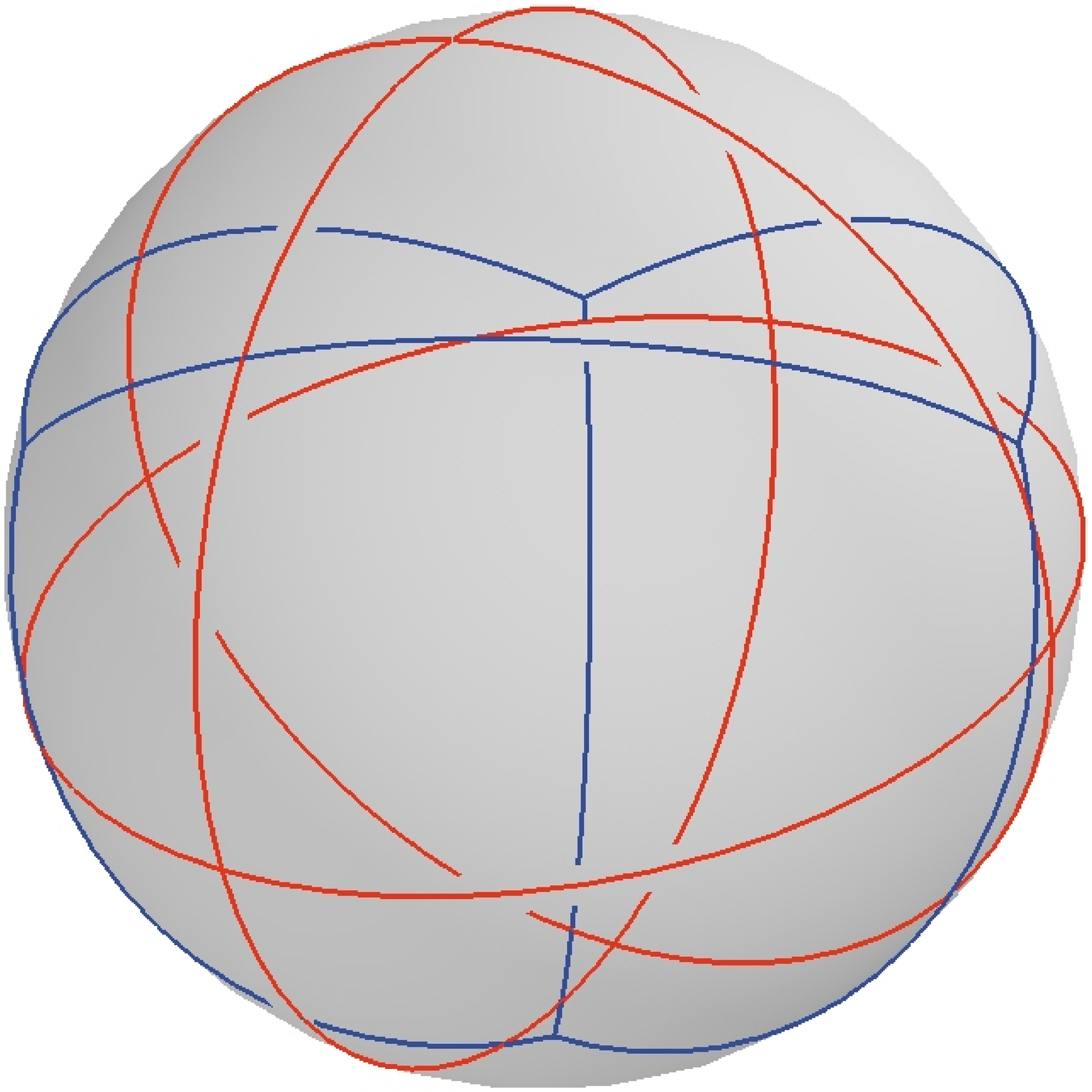,width=3cm,silent=} &
      \epsfig{figure=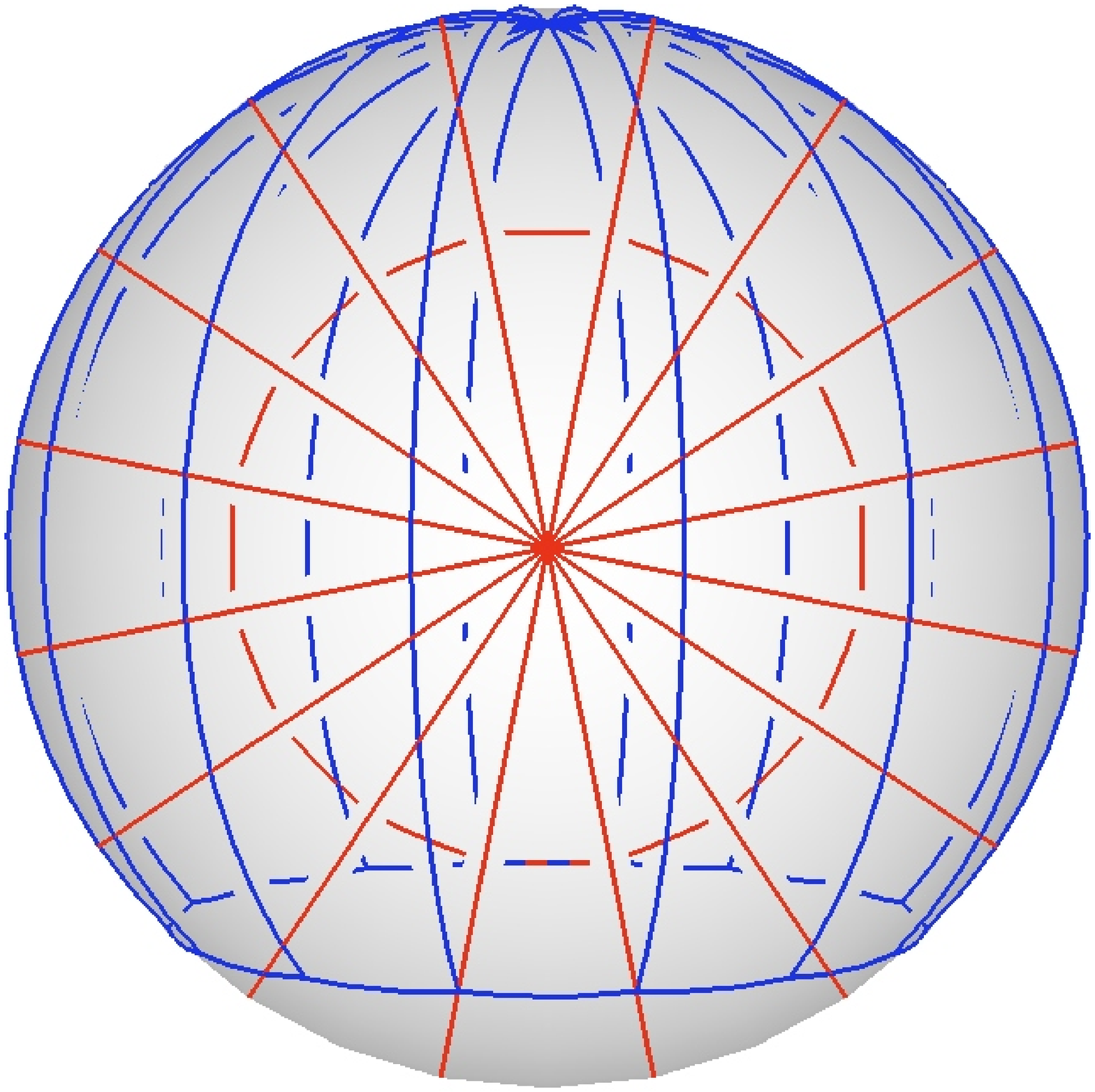,width=3cm,silent=} &
      \epsfig{figure=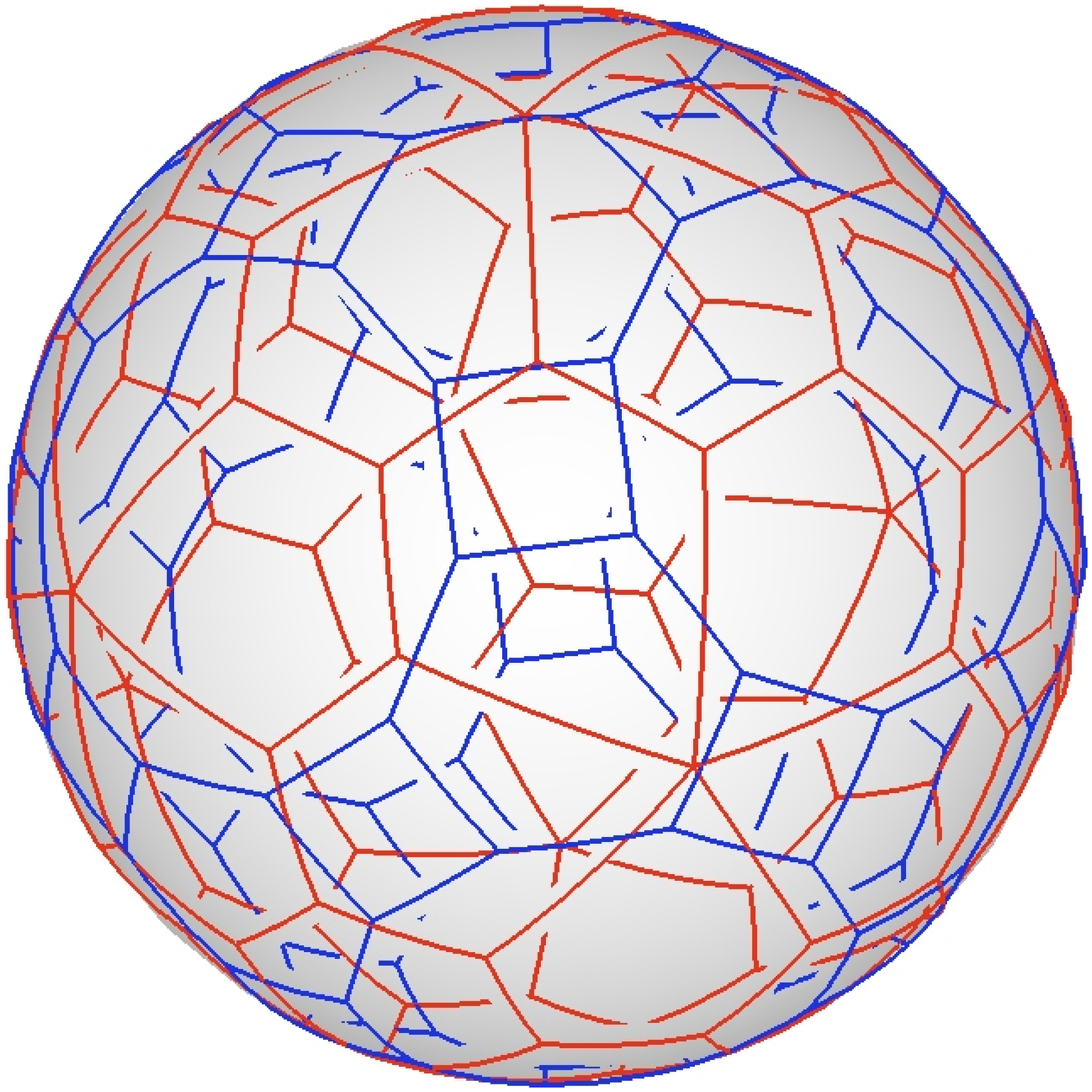,width=3cm,silent=} &
      \epsfig{figure=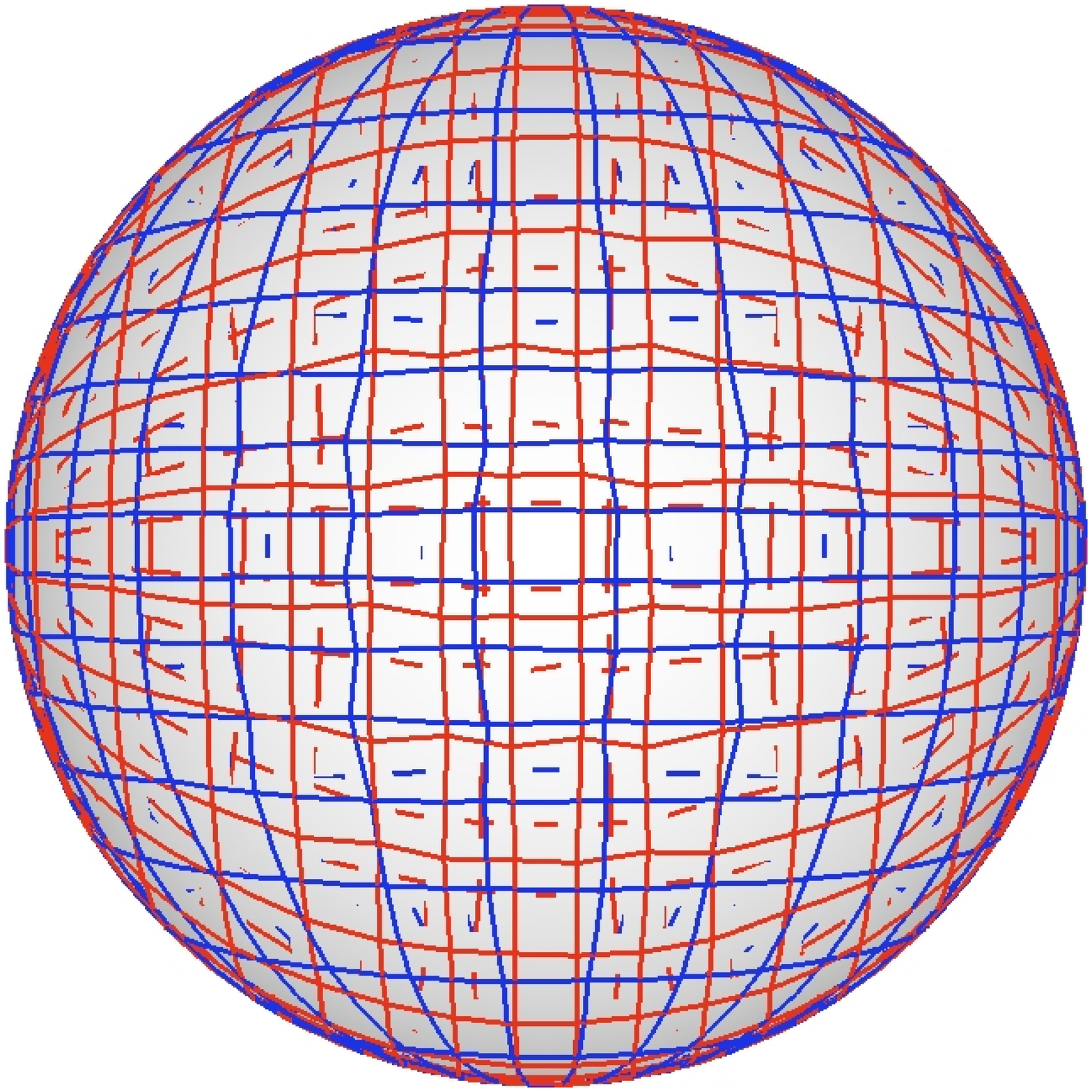,width=3cm,silent=}\\
      Tetra. $\oplus$ Cube &
      DP $\oplus$ ODP &
      PH $\oplus$ TI &
      El16 $\oplus$ OEl16
    \end{tabular}
  }
  \caption[Gaussian maps of Minkowski sums]{\capStyle{Gaussian maps of
    Minkowski sums.
             DP --- Dioctagonal Pyramid,
             PH --- Pentagonal Hexecontahedron,
             TI --- Truncated Icosidodecahedron,
             GS4 --- Geodesic Sphere level~4,
             El16 --- Ellipsoid-like polyhedron made of 16 latitudes and 32
             longitudes.
    }}
  \label{fig:gaussian-map-minkowski-sum}
\end{figure}

When the overlay operation progresses, new vertices, edges, and faces
of the resulting arrangement are created based on features of the two
operands. When a new feature is created its attributes are updated.
There are ten cases that arise and must be handled; see
Section~\ref{ssec:aos:facilities:overlay} for the precise enumeration
of the various cases. For example, a new face $f$ is induced by the
overlap of two faces $f_1$ and $f_2$ of the two summands, respectively.
The primal vertex associated with $f$ is set to be the sum of the primal
vertices associated with $f_1$ and $f_2$, respectively.

Table~\ref{tab:mink-cnt} lists the number of features ({\bf V},
{\bf HE}, {\bf F}) in the arrangement that represents the Gaussian map
of the respective Minkowski sums. Table~\ref{tab:mink-time} shows the time
in seconds ({\bf t}) it takes to construct the arrangement once the
Gaussian maps of the summands are in place.

\section{The Cubical Gaussian-Map Method}
\label{sec:mscn:cgm-method}
While using the \cgm{} increases the overhead of some operations
sixfold, and introduces degeneracies that are not present in the case
of alternative representations, it simplifies the construction and
manipulation of the representation, as the partition of each cube face
is a planar map of segments, a well known concept that has been
intensively experimented with during recent years. Indeed, all the basic
software components that the \cgm{} layer depends on are available in
\cgal{} version~3.3 and higher, while many of the software components
required by the (spherical) Gaussian map method are expected to appear
only in a future release of \cgal{}. The \cgm{} method, being more mature,
exhibits better performance than the (spherical) Gaussian map method. One
of the reasons for the performance gap is the lack of optimized primitives
that operate on unnormalized vectors in $\rrr$ in case of the (spherical)
Gaussian map method. Evidently, most of the methods of the
geometry-traits class that handle geodesic arcs embedded on the sphere
project their input $u$-monotone curves and points onto one of the
axis-aligned planes every time they are invoked, while all geometric
objects in case of the \cgm{} method are projected onto the plane
{\em a priori}. This creates an opportunity for optimization; see
Section~\ref{ssec:conclusion:arr-2d:geometry-traits-models}.
In addition, the \cgm{} data structure, being based on components
confined to the plane, has a broader recognition, as it can be used in
restricted environments, e.g., 3D hardware accelerators; see
Section~\ref{sec:conclusion:reflection-mapping}.

\subsection{The Representation}
\label{ssec:mscn:cgm:representation}
We use the \cgal{} \arr{} data-structure to maintain the planar
maps. The construction of the six planar maps from the polytope
features and their incident relations is similar to the construction
of the arrangement of geodesic arcs that represents the Gaussian map;
see Section~\ref{ssec:mscn:sgm:representation}. As in the case of the
(spherical) Gaussian map, the \cPolyhedron{} intermediate
representation is discarded once the construction of the \cgm{} is
complete. However, while a single edge of $P$ is mapped to at most two
$u$-monotone geodesic arcs in case of the (spherical) Gaussian map, it
can be mapped to a chain of at most four connected segments that lie
in four adjacent cube-faces, respectively. In any case the constructions
of the Gaussian maps of both methods respectively amount to the
insertion of curves that are pairwise disjoint in their interior into
the arrangement, an operation that is carried out efficiently,
especially when one or both endpoints are known. The construction of the
Minkowski sum, described in Section~\ref{ssec:mscn:cgm:mink_sum},
amounts to the computation of the six overlays of six pairs of planar
maps, respectively, an operation well supported by the data structure as
well.

A related dual representation had been considered and discarded before the
\cgm{} representation was chosen. It uses only two planar maps that partition
two parallel planes respectively instead of six, but each planar map 
partitions the entire plane.\footnote{Each planar map that corresponds to one 
of the six unit-cube faces in the \cgm{} representation also partitions the 
entire plane, but only the $[-1,-1] \times [1,1]$ square is relevant. 
The unbounded face, which comprises all the rest, is irrelevant.}
In this 2-map representation facets that are near orthogonal to the parallel
planes are mapped to points that are
far away from the origin. The exact representation of such points requires
coordinates with large bit-lengths, which increases significantly the
time it takes to perform exact arithmetic operations on them. Moreover,
facets exactly orthogonal to the parallel planes are mapped to points at
infinity, and require special handling all together.

Features that are not in general position, such as two parallel facets
facing the same direction, one from each polytope, or worse yet, two
identical polytopes, typically require special treatment. Still, the
handling of many of these problematic cases falls under the
``general'' case, and becomes transparent to the Gaussian-map layer
(either cubical or spherical). Consider for example the case of two
neighboring facets in one polytope that have parallel neighboring
facets in the other. This translates to overlapping segments in case
of the \cgm{} method and overlapping geodesic arcs in case of the
(spherical) Gaussian-map method, one from each Gaussian map of the two
polytopes,\footnote{Other conditions translate to overlapping segments
in case of the \cgm{} or geodesic arcs in case of the (spherical)
Gaussian map method as well.} that appear during the Minkowski sum
computation. The algorithm that computes it is oblivious to this
condition, as the underlying arrangement data structure, either
embedded in the plane or on the sphere, is perfectly capable of
handling overlapping curves. However, as mentioned above, other
degeneracies do emerge, and are handled successfully. One example, in 
case of the \cgm{}, is a facet $f$ mapped to a point that lies on an
edge of the unit cube, or even worse, coincides with one of the eight
corners of the cube. Figure~\ref{fig:models1}(a,b,c) depicts an
extreme degenerate case of an octahedron oriented in such a way that
its eight facets are mapped to the eight vertices of the unit
cube, respectively. The (spherical) Gaussian map method is not free of
degenerate conditions. For example, a facet mapped to a point that
lies on the identification arc, or worse yet, coincides with one of
the two poles.

The dual representation is extended further, in order to handle all
these degeneracies and perform all the necessary operations as
efficiently as possible. Each planar map is initialized with four
edges and four vertices that define the unit square --- the
parallel-axis square circumscribing the unit circle. During
construction, some of these edges or portions of them along with some
of these vertices may turn into real elements of the \cgm. The
introduction of these artificial elements not only expedites the
traversals below, but is also necessary for handling degenerate cases,
such as an empty cube face that appears in the representation of the
tetrahedron and depicted in Figure~\ref{fig:tet}(c). The global data 
consists of the six planar maps and 24 references to the planar
vertices that coincide with the unit-cube corners.

\begin{wrapfigure}{l}{8.6cm}
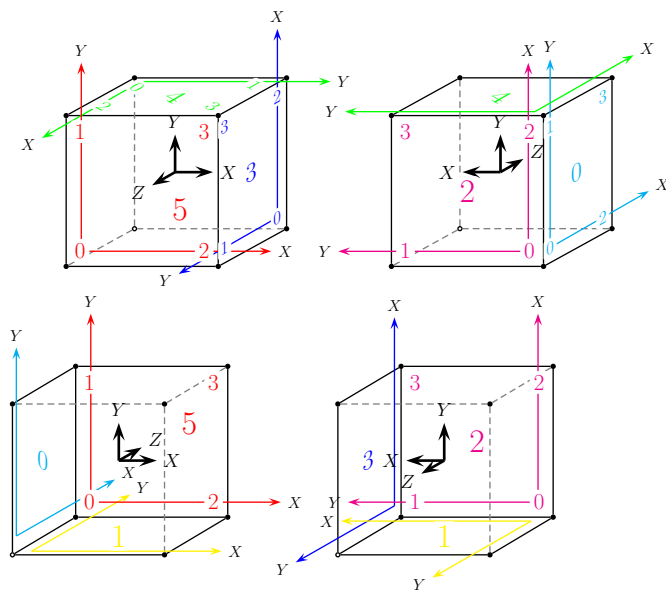

  \vspace{-25pt}
  \begin{center}
    \scalebox{0.5}{\newlength{\displacement}\setlength{\displacement}{0.4cm}
\newlength{\ndisplacement}\setlength{\ndisplacement}{-0.4cm}
\newlength{\hdisplacement}\setlength{\hdisplacement}{0.2cm}
  \begin{tabular}{cc}
    \pspicture[](-3,-2.5)(5,4.5)
    \psset{unit=1cm,subgriddiv=0,arrowsize=3pt 3,linewidth=1pt,
      viewpoint=1 -1 1
    }
    \cnode*(-2,-2){2pt}{0}
    \cnode*(-2, 2){2pt}{1}
    \cnode*( 2, 2){2pt}{2}
    \cnode*( 2,-2){2pt}{3}
    \rput{30}(0,0){
      \rput{-30}(2,0){
	\cnode(-2,-2){2pt}{4}
	\cnode*(-2, 2){2pt}{5}
	\cnode*( 2, 2){2pt}{6}
	\cnode*( 2,-2){2pt}{7}
      }
    }
    \rput{30}(2,-2){
      \rput{0}(2,0){
	\rput{-30}(-\hdisplacement,0){
	  \rput{0}(0,\displacement){
	    \pnode(0,0){a1}
	    \pnode(0,5){b1}
	    \uput[90]{0}(0,5){$X$}
	    \rput{30}(0,0){
	      \pnode(-3,0){c1}
	      \uput[180]{-30}(-3,0){$Y$}
	    }
	  }
	}
      }
    }
    \rput{30}(-2,2){
      \rput{0}(2,0){
	\rput{-30}(-\hdisplacement,0){
	  \rput{0}(\displacement,0){
	    \pnode(0,0){a2}
	    \pnode(5,0){b2}
	    \uput[0]{0}(5,0){$Y$}
	    \rput{30}(0,0){
	      \pnode(-3,0){c2}
	      \uput[180]{-30}(-3,0){$X$}
	    }
	  }
	}
      }
    }
    \ncline{0}{1}
    \ncline{1}{2}
    \ncline{2}{3}
    \ncline{3}{0}
    \ncline[linecolor=gray,linestyle=dashed]{4}{5}
    \ncline{5}{6}
    \ncline{6}{7}
    \ncline[linecolor=gray,linestyle=dashed]{7}{4}
    \ncline[linecolor=gray,linestyle=dashed]{0}{4}
    \ncline{1}{5}
    \ncline{2}{6}
    \ncline{3}{7}
    \ncline[linecolor=blue]{->}{a1}{b1}
    \ncline[linecolor=blue]{->}{a1}{c1}
    \ncline[linecolor=green]{->}{a2}{b2}
    \ncline[linecolor=green]{->}{a2}{c2}
    \rput{0}(\displacement,\displacement){
      \psline[linecolor=red]{->}(-2,-2)(3,-2)
      \psline[linecolor=red]{->}(-2,-2)(-2,3)
      \uput[0]{0}(3,-2){$X$}
      \uput[90]{0}(-2,3){$Y$}
    }
    \rput{30}(0,0){
      \rput{-30}(1,0){
	\my_axis{0}{1}{90}{1}{210}{0.7}
      }
    }
    \rput{0}(1,-0.5){\psframebox*[framearc=.3]{\red \Huge 5}}
    \rput{30}(2,0){
      \rput{-30}(1,0){
	\ThreeDput[normal=1 .0 .0](0,0,0){
	  \psframebox*[framearc=.3]{\blue \Huge 3}
	}
      }
    }
    \rput{0}(2,2){
      \rput{30}(0,-\displacement){
	\rput{0}(2,0){
	  \rput{-30}(-\hdisplacement,0){
	    \ThreeDput[normal=1 .0 .0](0,0,0){
	      \psframebox*[framearc=.3]{\blue \Large 2}
	    }
	  }
	}
      }
    }
    \rput{0}(2,-2){
      \rput{30}(0,\displacement){
	\rput{-30}(\hdisplacement,0){
	  \ThreeDput[normal=1 .0 .0](0,0,0){
	    \psframebox*[framearc=.3]{\blue \Large 1}
	  }
	}
      }
    }
    \rput{0}(2,-2){
      \rput{30}(0,\displacement){
	\rput{0}(2,0){
	  \rput{-30}(-\hdisplacement,0){
	    \ThreeDput[normal=1 .0 .0](0,0,0){
	      \psframebox*[framearc=.3]{\blue \Large 0}
	    }
	  }
	}
      }
    }
    \rput{0}(2,2){
      \rput{30}(0,-\displacement){
	\rput{-30}(\hdisplacement,0){
	  \ThreeDput[normal=1 .0 .0](0,0,0){
	    \psframebox*[framearc=.3]{\blue \Large 3}
	  }
	}
      }
    }
    \rput{30}(0,2){
      \rput{-30}(1,0){
	\ThreeDput[normal=.0 .0 1](0,0,0){
	  \psframebox*[framearc=.3]{\green \Huge 4}
	}
      }
    }
    \rput{0}(2,2){
      \ThreeDput[normal=.0 .0 1](0,0,0){
	\rput{0}(-0.4,0.2){
	  \psframebox*[framearc=.3]{\green \Large 3}
	}
      }
    }
    \rput{0}(-2,2){
      \ThreeDput[normal=.0 .0 1](0,0,0){
	\rput{0}(0.4,1){
	  \psframebox*[framearc=.3]{\green \Large 2}
	}
      }
    }
    \rput{30}(2,2){
      \rput{-30}(2,0){
	\ThreeDput[normal=.0 .0 1](0,0,0){
	  \rput{0}(-0.3,-0.8){
	    \psframebox*[framearc=.3]{\green \Large 1}
	  }
	}
      }
    }
    \rput{30}(-2,2){
      \rput{-30}(2,0){
	\ThreeDput[normal=.0 .0 1](0,0,0){
	  \rput{0}(0.4,-0.2){
	    \psframebox*[framearc=.3]{\green \Large 0}
	  }
	}
      }
    }
    \uput[45]{0}(-2,-2){\psframebox*[framearc=.3]{\red \Large 0}}
    \uput[-45]{0}(-2, 2){\psframebox*[framearc=.3]{\red \Large 1}}
    \uput[-135]{0}( 2, 2){\psframebox*[framearc=.3]{\red \Large 3}}
    \uput[135]{0}( 2,-2){\psframebox*[framearc=.3]{\red \Large 2}}
    \endpspicture &
    \pspicture[](-3,-2.5)(5,4.5)
    \psset{unit=1cm,subgriddiv=0,arrowsize=3pt 3,linewidth=1pt,
      viewpoint=-1 -1 -1
    }
    \cnode*(-2,-2){2pt}{0}
    \cnode*(-2, 2){2pt}{1}
    \cnode*( 2, 2){2pt}{2}
    \cnode*( 2,-2){2pt}{3}
    \rput{30}(0,0){
      \rput{-30}(2,0){
	\cnode(-2,-2){2pt}{4}
	\cnode*(-2, 2){2pt}{5}
	\cnode*( 2, 2){2pt}{6}
	\cnode*( 2,-2){2pt}{7}
      }
    }
    \rput{30}(2,-2){
      \rput{-30}(\hdisplacement,0){
	\rput{0}(0,\displacement){
	  \pnode(0,0){a1}
	  \pnode(0,5){b1}
	  \uput[90]{0}(0,5){$Y$}
	  \rput{30}(0,0){
	    \pnode(3,0){c1}
	    \uput[0]{-30}(3,0){$X$}
	  }
	}
      }
    }
    \rput{0}(2,2){
      \rput{30}(-\displacement,0){
	\rput{-30}(\hdisplacement,0){
	  \pnode(0,0){a2}
	  \pnode(-5,0){b2}
	  \uput[180]{0}(-5,0){$Y$}
	  \rput{30}(0,0){
	    \pnode(3,0){c2}
	    \uput[0]{-30}(3,0){$X$}
	  }
	}
      }
    }
    \ncline{0}{1}
    \ncline{1}{2}
    \ncline{2}{3}
    \ncline{3}{0}
    \ncline[linecolor=gray,linestyle=dashed]{4}{5}
    \ncline{5}{6}
    \ncline{6}{7}
    \ncline[linecolor=gray,linestyle=dashed]{7}{4}
    \ncline[linecolor=gray,linestyle=dashed]{0}{4}
    \ncline{1}{5}
    \ncline{2}{6}
    \ncline{3}{7}
    \ncline[linecolor=cyan]{->}{a1}{b1}
    \ncline[linecolor=cyan]{->}{a1}{c1}
    \ncline[linecolor=green]{->}{a2}{b2}
    \ncline[linecolor=green]{->}{a2}{c2}
    \rput{0}(2,-2){
      \rput{0}(-\displacement,\displacement){
	\psline[linecolor=magenta]{->}(0,0)(-5,0)
	\psline[linecolor=magenta]{->}(0,0)(0,5)
	\uput[180]{0}(-5,0){$Y$}
	\uput[90]{0}(0,5){$X$}
      }
    }
    \rput{30}(0,0){
      \rput{-30}(1,0){
	\my_axis{180}{1}{90}{1}{30}{0.7}
      }
    }
    \rput{0}(0,0){\psframebox*[framearc=.3]{\magenta \Huge 2}}
    \rput{30}(2,0){
      \rput{-30}(1,0){
	\ThreeDput[normal=-1 .0 .0](0,0,0){
	  \psframebox*[framearc=.3]{\cyan \Huge 0}
	}
      }
    }
    \rput{0}(2,-2){
      \rput{30}(0,\displacement){
	\rput{-30}(\hdisplacement,0){
	  \ThreeDput[normal=-1 .0 .0](0,0,0){
	    \psframebox*[framearc=.3]{\cyan \Large 0}
	  }
	}
      }
    }
    \rput{0}(2,-2){
      \rput{30}(0,\displacement){
	\rput{0}(2,0){
	  \rput{-30}(-\hdisplacement,0){
	    \ThreeDput[normal=-1 .0 .0](0,0,0){
	      \psframebox*[framearc=.3]{\cyan \Large 2}
	    }
	  }
	}
      }
    }
    \rput{0}(2,2){
      \rput{30}(0,-\displacement){
	\rput{-30}(\hdisplacement,0){
	  \ThreeDput[normal=-1 .0 .0](0,0,0){
	    \psframebox*[framearc=.3]{\cyan \Large 1}
	  }
	}
      }
    }
    \rput{0}(2,2){
      \rput{30}(0,-\displacement){
	\rput{0}(2,0){
	  \rput{-30}(-\hdisplacement,0){
	    \ThreeDput[normal=-1 .0 .0](0,0,0){
	      \psframebox*[framearc=.3]{\cyan \Large 3}
	    }
	  }
	}
      }
    }
    \rput{30}(0,2){
      \rput{-30}(1,0){
	\ThreeDput[normal=.0 .0 -1](0,0,0){
	  \psframebox*[framearc=.3]{\green \Huge 4}
	}
      }
    }
    \uput[45]{0}(-2,-2){\psframebox*[framearc=.3]{\magenta \Large 1}}
    \uput[-45]{0}(-2, 2){\psframebox*[framearc=.3]{\magenta \Large 3}}
    \uput[-135]{0}( 2, 2){\psframebox*[framearc=.3]{\magenta \Large 2}}
    \uput[135]{0}( 2,-2){\psframebox*[framearc=.3]{\magenta \Large 0}}
    \endpspicture \\
    \pspicture[](-3,-3.5)(4.5,4)
    \psset{unit=1cm,subgriddiv=0,arrowsize=3pt 3,linewidth=1pt,
      viewpoint=-1 1 -1
    }
    \cnode*(-2,-2){2pt}{0}
    \cnode*(-2, 2){2pt}{1}
    \cnode*( 2, 2){2pt}{2}
    \cnode*( 2,-2){2pt}{3}
    \rput{210}(0,0){
      \rput{-210}(2,0){
	\cnode(-2,-2){2pt}{4}
	\cnode*(-2, 2){2pt}{5}
	\cnode*( 2, 2){2pt}{6}
	\cnode*( 2,-2){2pt}{7}
      }
    }
    \rput{210}(-2,-2){
      \rput{0}(2,0){
	\rput{-210}(-\hdisplacement,0){
	  \rput{0}(0,\displacement){
	    \pnode(0,0){a1}
	    \pnode(0,5){b1}
	    \uput[90]{0}(0,5){$Y$}
	    \rput{210}(0,0){
	      \pnode(-3,0){c1}
	      \uput[180]{-210}(-3,0){$X$}
	    }
	  }
	}
      }
    }
    \rput{0}(-2,-2){
      \rput{210}(\displacement,0){
	\rput{0}(2,0){
	  \rput{-210}(-\hdisplacement,0){
	    \pnode(0,0){a2}
	    \pnode(5,0){b2}
	    \uput[0]{0}(5,0){$X$}
	    \rput{210}(0,0){
	      \pnode(-3,0){c2}
	      \uput[180]{-210}(-3,0){$Y$}
	    }
	  }
	}
      }
    }
    \ncline{0}{1}
    \ncline{1}{2}
    \ncline{2}{3}
    \ncline{3}{0}
    \ncline{4}{5}
    \ncline[linecolor=gray,linestyle=dashed]{5}{6}
    \ncline[linecolor=gray,linestyle=dashed]{6}{7}
    \ncline{7}{4}
    \ncline{0}{4}
    \ncline{1}{5}
    \ncline[linecolor=gray,linestyle=dashed]{2}{6}
    \ncline{3}{7}
    \ncline[linecolor=cyan]{->}{a1}{b1}
    \ncline[linecolor=cyan]{->}{a1}{c1}
    \ncline[linecolor=yellow]{->}{a2}{b2}
    \ncline[linecolor=yellow]{->}{a2}{c2}
    \rput{0}(\displacement,\displacement){
      \psline[linecolor=red]{->}(-2,-2)(3,-2)
      \psline[linecolor=red]{->}(-2,-2)(-2,3)
      \uput[0]{0}(3,-2){$X$}
      \uput[90]{0}(-2,3){$Y$}
    }
    \rput{210}(0,0){
      \rput{-210}(1,0){
	\my_axis{0}{1}{90}{1}{30}{0.7}
      }
    }
    \rput{0}(1,0.5){\psframebox*[framearc=.3]{\red \Huge 5}}
    \rput{210}(-2,0){
      \rput{-210}(1,0){
	\ThreeDput[normal=-1 .0 .0](0,0,0){
	  \psframebox*[framearc=.3]{\cyan \Huge 0}
	}
      }
    }
    \rput{210}(0,-2){
      \rput{-210}(1,0){
	\psframebox*[framearc=.3]{\yellow \Huge 1}
      }
    }
    \uput[45]{0}(-2,-2){\psframebox*[framearc=.3]{\red \Large 0}}
    \uput[-45]{0}(-2, 2){\psframebox*[framearc=.3]{\red \Large 1}}
    \uput[-135]{0}( 2, 2){\psframebox*[framearc=.3]{\red \Large 3}}
    \uput[135]{0}( 2,-2){\psframebox*[framearc=.3]{\red \Large 2}}
    \endpspicture &
    \pspicture[](-3,-3.5)(4.5,4)
    \psset{unit=1cm,subgriddiv=0,arrowsize=3pt 3,linewidth=1pt,
      viewpoint=1 1 -1
    }
    \cnode*(-2,-2){2pt}{0}
    \cnode*(-2, 2){2pt}{1}
    \cnode*( 2, 2){2pt}{2}
    \cnode*( 2,-2){2pt}{3}
    \rput{210}(0,0){
      \rput{-210}(2,0){
	\cnode(-2,-2){2pt}{4}
	\cnode*(-2, 2){2pt}{5}
	\cnode*( 2, 2){2pt}{6}
	\cnode*( 2,-2){2pt}{7}
      }
    }
    \rput{0}(-2,-2){
      \rput{210}(0,\displacement){
	\rput{-210}(\hdisplacement,0){
	  \pnode(0,0){a1}
	  \pnode(0,5){b1}
	  \uput[90]{0}(0,5){$X$}
	  \rput{210}(0,0){
	    \pnode(3,0){c1}
	    \uput[0]{-210}(3,0){$Y$}
	  }
	}
      }
    }
    \rput{0}(2,-2){
      \rput{210}(-\displacement,0){
	\rput{-210}(\hdisplacement,0){
	  \pnode(0,0){a2}
	  \pnode(-5,0){b2}
	  \uput[180]{0}(-5,0){$X$}
	  \rput{210}(0,0){
	    \pnode(3,0){c2}
	    \uput[0]{-210}(3,0){$Y$}
	  }
	}
      }
    }
    \ncline{0}{1}
    \ncline{1}{2}
    \ncline{2}{3}
    \ncline{3}{0}
    \ncline{4}{5}
    \ncline[linecolor=gray,linestyle=dashed]{5}{6}
    \ncline[linecolor=gray,linestyle=dashed]{6}{7}
    \ncline{7}{4}
    \ncline{0}{4}
    \ncline{1}{5}
    \ncline[linecolor=gray,linestyle=dashed]{2}{6}
    \ncline{3}{7}
    \ncline[linecolor=blue]{->}{a1}{b1}
    \ncline[linecolor=blue]{->}{a1}{c1}
    \ncline[linecolor=yellow]{->}{a2}{b2}
    \ncline[linecolor=yellow]{->}{a2}{c2}
    \rput{0}(-\displacement,\displacement){
      \psline[linecolor=magenta]{->}(2,-2)(-3,-2)
      \psline[linecolor=magenta]{->}(2,-2)(2,3)
      \uput[180]{0}(-3,-2){$Y$}
      \uput[90]{0}(2,3){$X$}
    }
    \rput{210}(0,0){
      \rput{-210}(1,0){
	\my_axis{180}{1}{90}{1}{210}{0.7}
      }
    }
    \rput{0}(.0,.0){\psframebox*[framearc=.3]{\magenta \Huge 2}}
    \rput{210}(-2,0){
      \rput{-210}(1,0){
	\ThreeDput[normal=1 .0 .0](0,0,0){
	  \psframebox*[framearc=.3]{\blue \Huge 3}
	}
      }
    }
    \rput{210}(0,-2){
      \rput{-210}(1,0){
	\psframebox*[framearc=.3]{\yellow \Huge 1}
      }
    }
   \uput[45]{0}(-2,-2){\psframebox*[framearc=.3]{\magenta \Large 1}}
    \uput[-45]{0}(-2, 2){\psframebox*[framearc=.3]{\magenta \Large 3}}
    \uput[-135]{0}( 2, 2){\psframebox*[framearc=.3]{\magenta \Large 2}}
    \uput[135]{0}( 2,-2){\psframebox*[framearc=.3]{\magenta \Large 0}}
    \endpspicture
  \end{tabular}
}
    \caption[The cubical Gaussian map data structure]%
            {\capStyle{The data structure. Large-font numbers indicate plane
             ids. Small-font numbers indicate corner ids. $X$ and $Y$ axes
             in different 2D coordinate systems are rendered in different
             colors.}}
   \label{fig:data_struct}
  \end{center}
  \vspace{-20pt}
\end{wrapfigure}
The exact mapping from a facet normal in the 3D coordinate-system to a pair
that consists of a planar map and a planar point in the 2D coordinate-system
is defined precisely through the indexing and ordering system, illustrated
in Figure~\ref{fig:data_struct}. Now before your eyes cross permanently, we
advise you to keep reading the next few lines, as they reveal the meaning of
some of the enigmatic numbers that appear in the figure. The six planar maps
are given unique ids from $0$ through $5$. Ids $0$, $1$, and $2$ are
associated with the planes $x = -1$, $y = -1$, and $z = -1$, respectively,
and ids $3$, $4$, and $5$ are associated with the planes $x = 1$, $y = 1$,
and $z = 1$, respectively. The major axes in the 2D Cartesian 
coordinate-system of each planar map are determined by the 3D
coordinate-system. The four corner vertices of each planar map are
also given unique ids from $0$ through $3$ in lexicographic order in
their respective 2D coordinate-system, see
Table~\ref{tab:coord-system} columns titled
\textbf{Underlying Plane} and \textbf{2D Axes}.

\begin{table}[!hbp]
  \caption[The coordinate systems and the cyclic chains of corner vertices]%
          {\capStyle{The coordinate systems and the cyclic chains of corner
           vertices. \textbf{PM} stands for \textbf{Planar Map}, and
           \textbf{Cr} stands for \textbf{Corner}.}}
  \label{tab:coord-system}
  \centerline{
  \begin{tabular}{|cc||c|c||c|c||c|c||c|c||c|c|}
    \hline
    \multicolumn{2}{|c||}{\textbf{Underlying}} &
    \multicolumn{2}{c||}{\multirow{2}*{\textbf{2D Axes}}} &
    \multicolumn{8}{c|}{\textbf{Corner}} \\ \cline{5-12}
    \multicolumn{2}{|c||}{\textbf{ ~~~~ Plane}} & \multicolumn{1}{c}{\ } & &
    \multicolumn{2}{c||}{\textbf{0} (0,0)} &
    \multicolumn{2}{c||}{\textbf{1} (0,1)} &
    \multicolumn{2}{c||}{\textbf{2} (1,0)} &
    \multicolumn{2}{c|}{\textbf{3} (1,1)} \\  \cline{1-2} \cline{3-4} \cline{5-12}
    \multicolumn{1}{|c|}{\textbf{Id}} & \textbf{Eq} &
    \multicolumn{1}{c|}{\boldmath ~ $X$ ~} &
    \multicolumn{1}{c||}{\boldmath $Y$} &
    \textbf{PM} & \textbf{Cr} & \textbf{PM} & \textbf{Cr} &
    \textbf{PM} & \textbf{Cr} & \textbf{PM} & \textbf{Cr} \\
    \hline
    \hline
    \multicolumn{1}{|c|}{0} & $x = -1$ & $Z$ & $Y$ & 1 & 0 & 2 & 2 & 5 & 0 & 4 & 2\\
    \multicolumn{1}{|c|}{1} & $y = -1$ & $X$ & $Z$ & 2 & 0 & 0 & 2 & 3 & 0 & 5 & 2\\
    \multicolumn{1}{|c|}{2} & $z = -1$ & $Y$ & $X$ & 0 & 0 & 1 & 2 & 4 & 0 & 3 & 2\\
    \multicolumn{1}{|c|}{3} & $x =  1$ & $Y$ & $Z$ & 2 & 1 & 1 & 3 & 4 & 1 & 5 & 3\\
    \multicolumn{1}{|c|}{4} & $y =  1$ & $Z$ & $X$ & 0 & 1 & 2 & 3 & 5 & 1 & 3 & 3\\
    \multicolumn{1}{|c|}{5} & $z =  1$ & $X$ & $Y$ & 1 & 1 & 0 & 3 & 3 & 1 & 4 & 3\\
    \hline
    \end{tabular}
  }
\end{table}

Each feature type of the \dcel\index{DCEL@\dcel} used to maintain the incidence
relations of the vertices, halfedges, and faces of the \arr{} data
structure (see Section~\ref{sec:aos:architecture}) is extended to hold
additional attributes. Some of the attributes are introduced only in
order to expedite the computation of certain operations, but most of
them are necessary to handle degenerate cases such as a planar vertex
lying on the unit-square boundary. Each planar-map vertex $v$ is
extended with
\setcounter{ms-const:cntr}{1}(\roman{ms-const:cntr}) the coefficients of the
plane containing the polygonal facet $C^{-1}(v)$ (see
Section~\ref{sec:mscn:gauss_map} for the definition of $C$ and $C^{-1}$),
\addtocounter{ms-const:cntr}{1}(\roman{ms-const:cntr}) the location of the
vertex --- an enumeration indicating whether the vertex coincides with a cube
corner, or lies on a cube edge, or contained in a cube face,
\addtocounter{ms-const:cntr}{1}(\roman{ms-const:cntr}) a Boolean flag
indicating whether it is non-artificial (there exists a facet that maps to
it), and
\addtocounter{ms-const:cntr}{1}(\roman{ms-const:cntr}) a pointer to a vertex
of a planar map associated with an adjacent cube-face that represents the
same central projection for vertices that coincide with a cube corner or lie
on a cube edge. Each planar-map halfedge $e$ is extended with a Boolean flag
indicating whether it is non-artificial (there exists a polytope edge that
maps to it). Each planar-map face $f$ is extended with the polytope vertex
that maps to it $v = C^{-1}(f)$.

\begin{figure*}[!htp]%
  \centerline{
    \begin{tabular}{ccc}
      \epsfig{figure=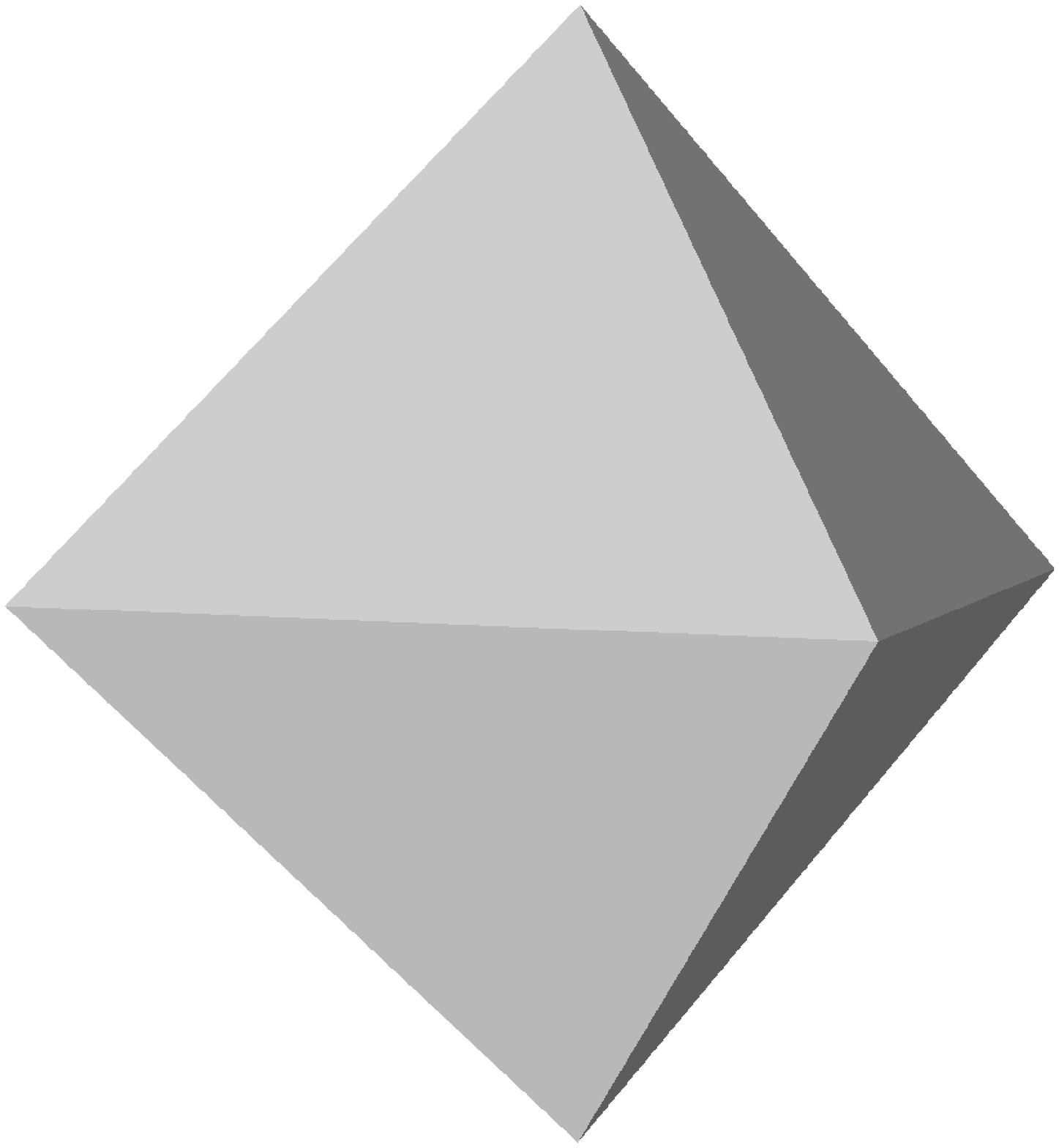,width=4cm,silent=} &
      \epsfig{figure=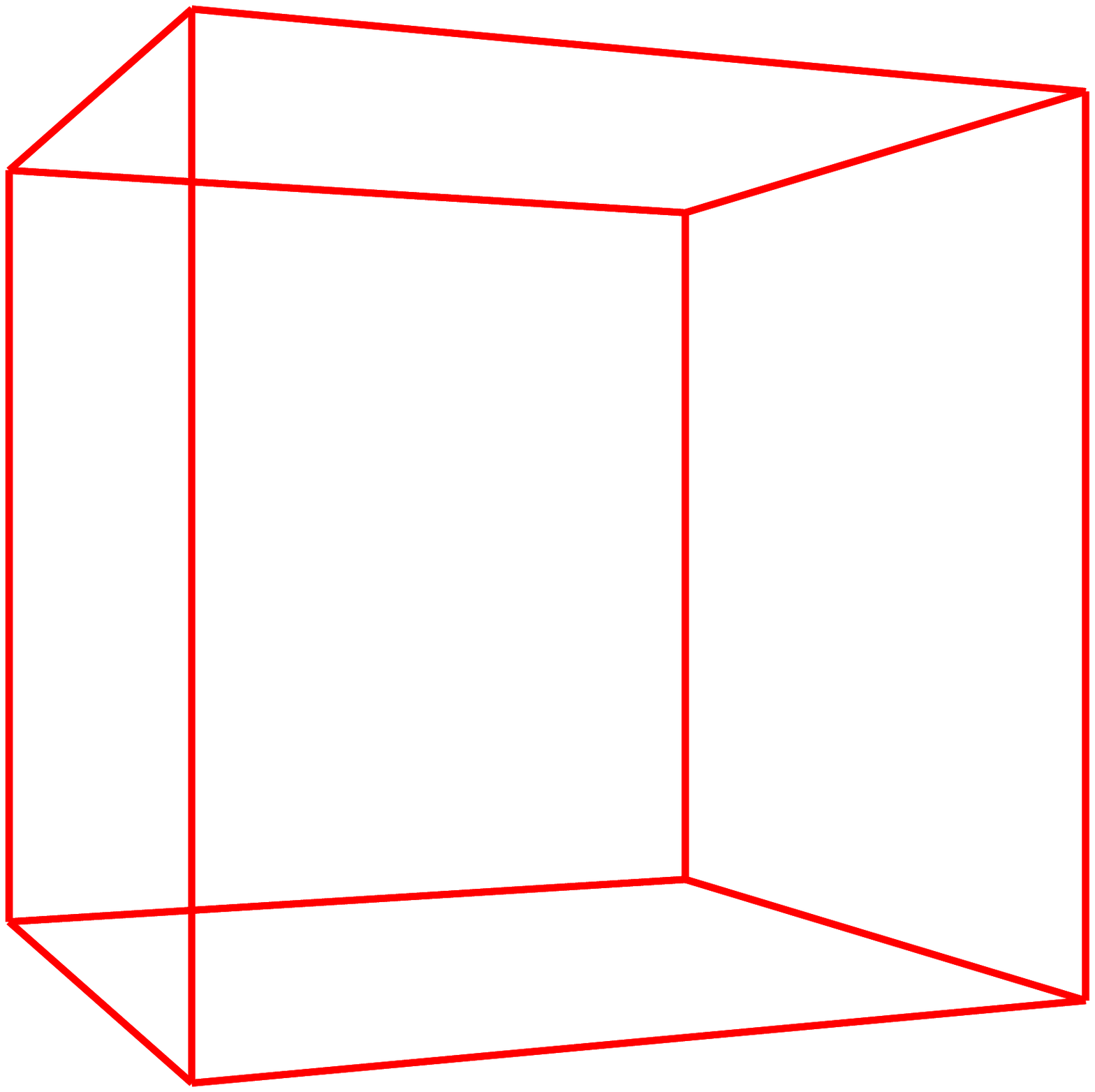,width=4cm,silent=} &
      \epsfig{figure=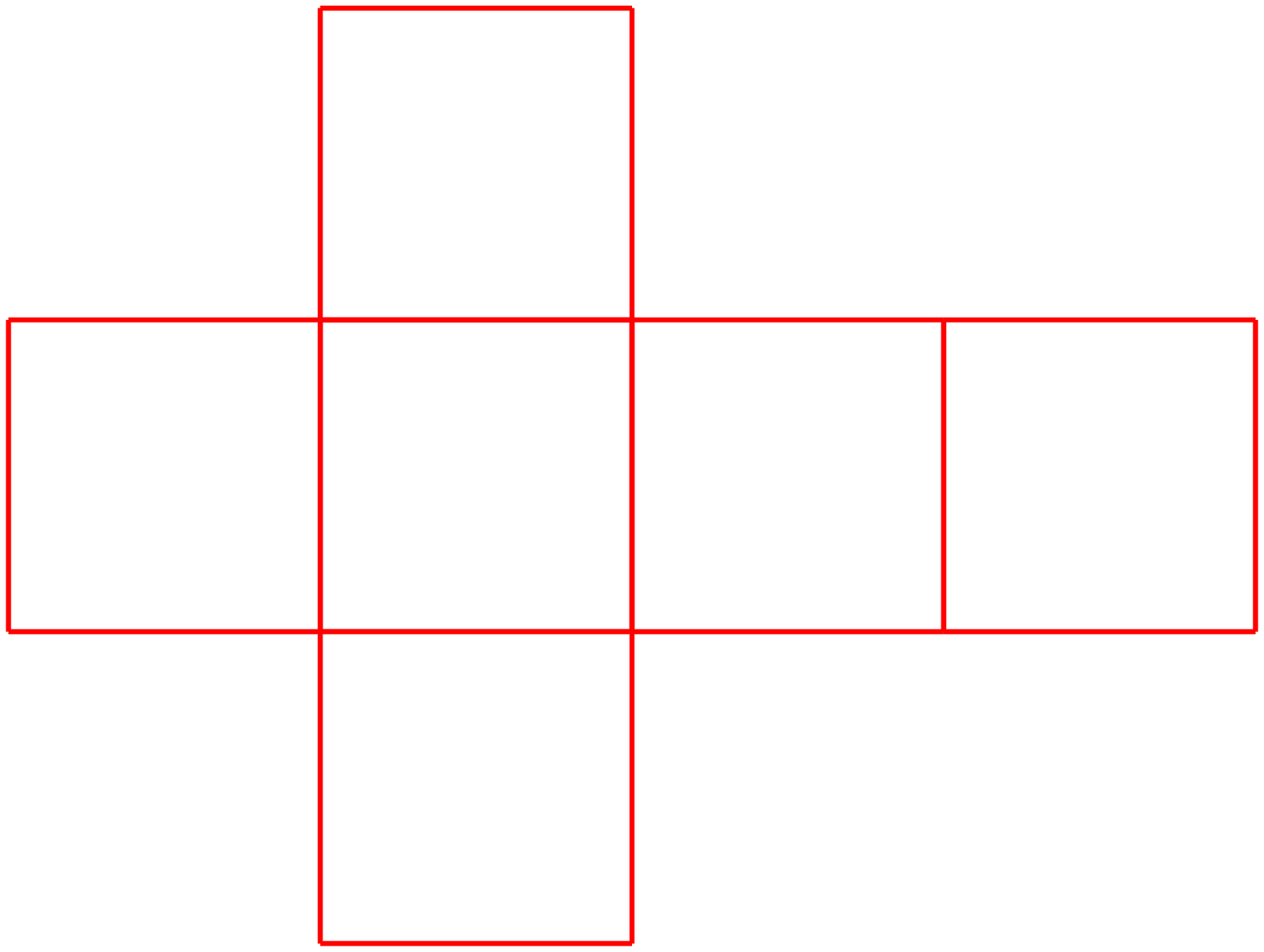,width=5cm,silent=}\\
      (a) & (b) & (c)\\
      \epsfig{figure=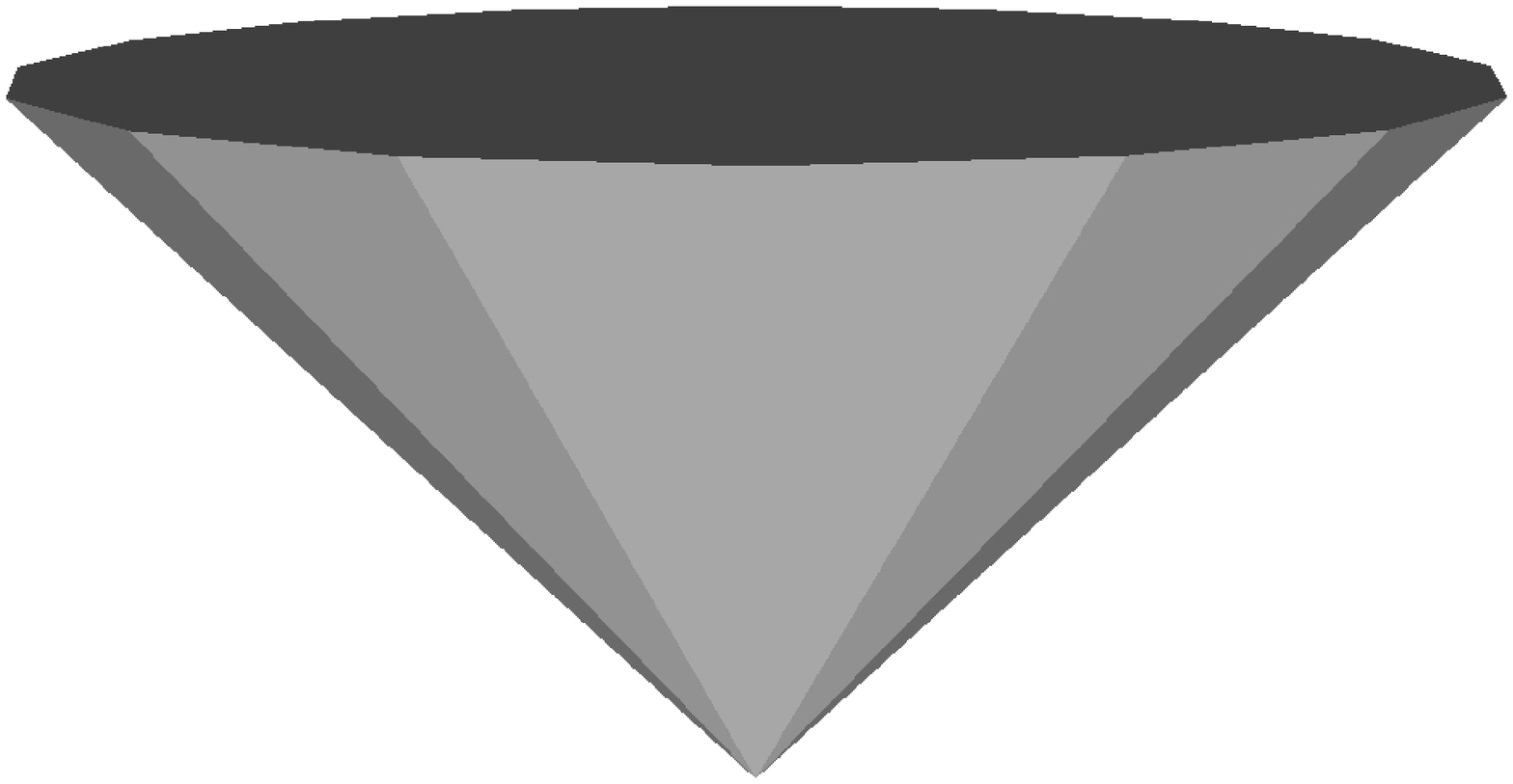,width=4cm,silent=} &
      \epsfig{figure=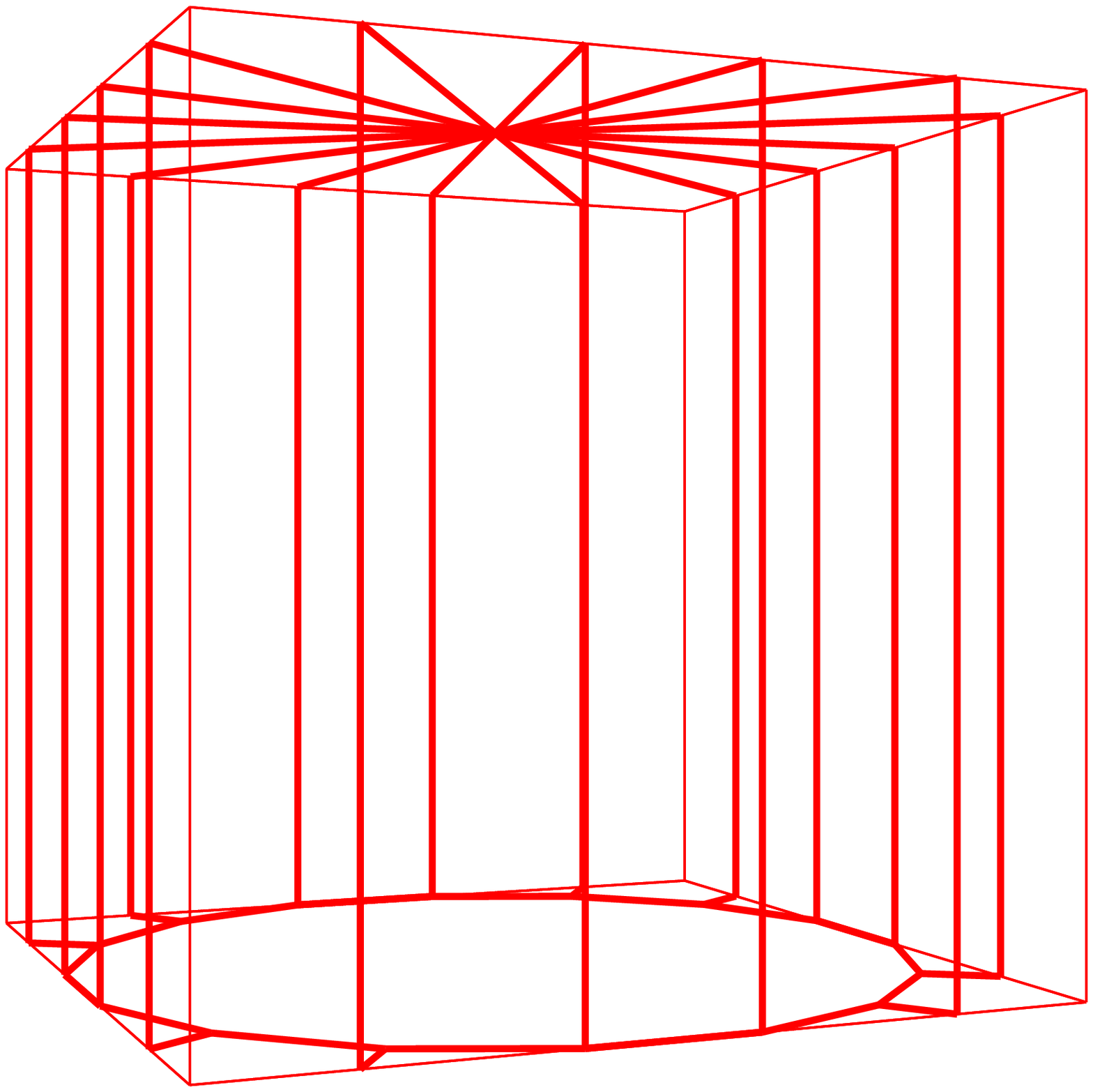,width=4cm,silent=} &
      \epsfig{figure=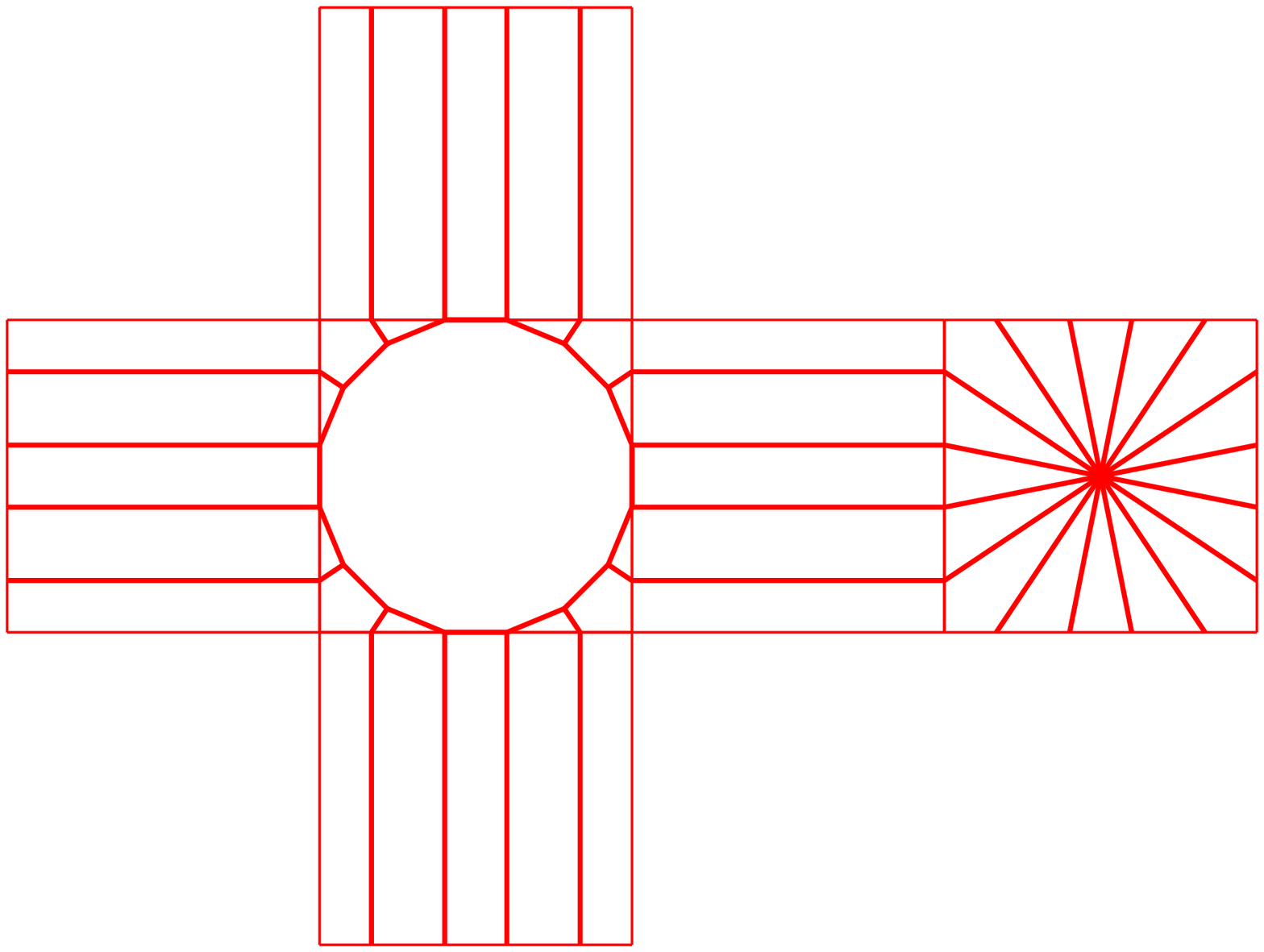,width=5cm,silent=}\\
      (d) & (e) & (f)\\
      \epsfig{figure=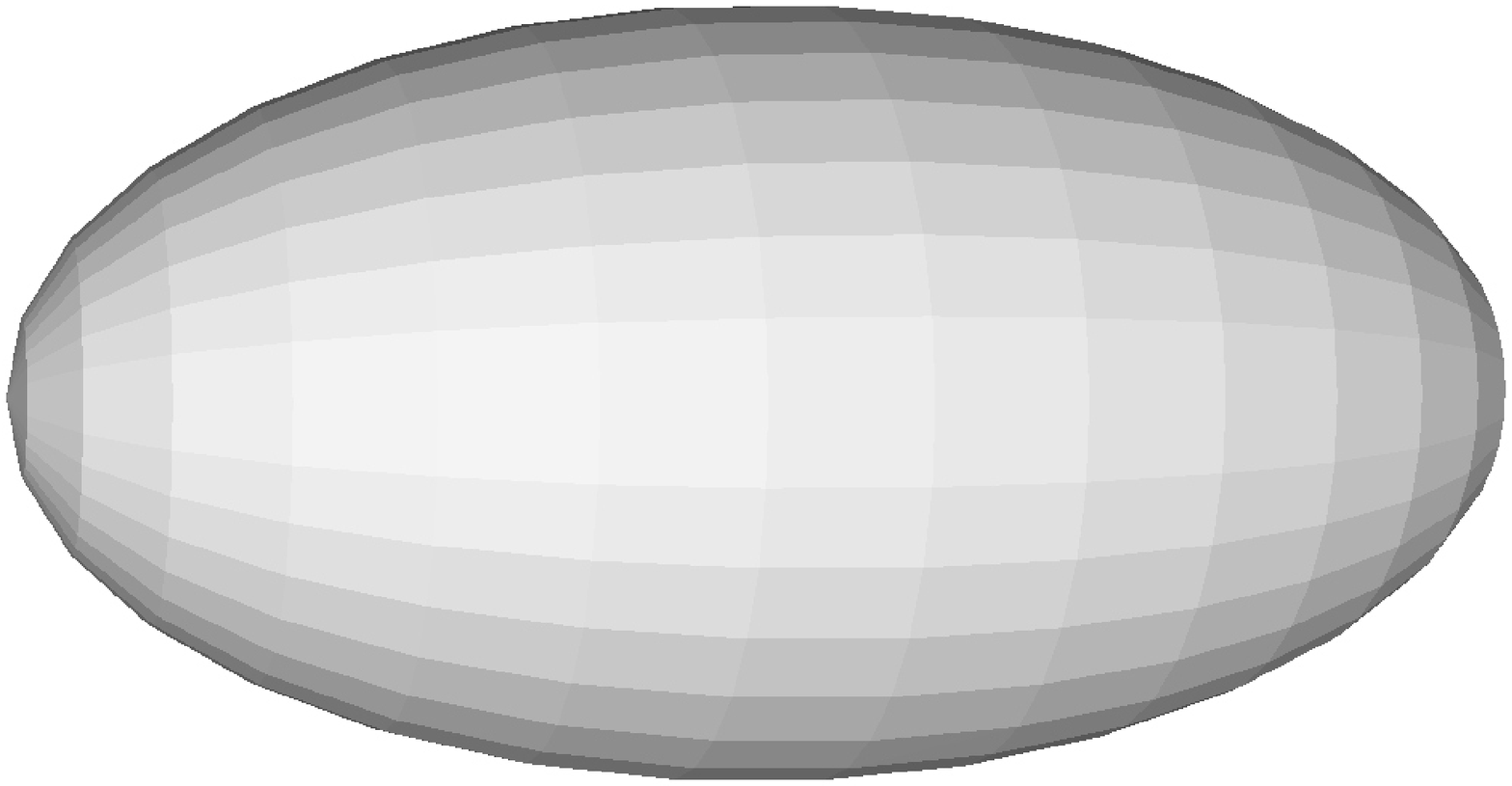,width=4cm,silent=} &
      \epsfig{figure=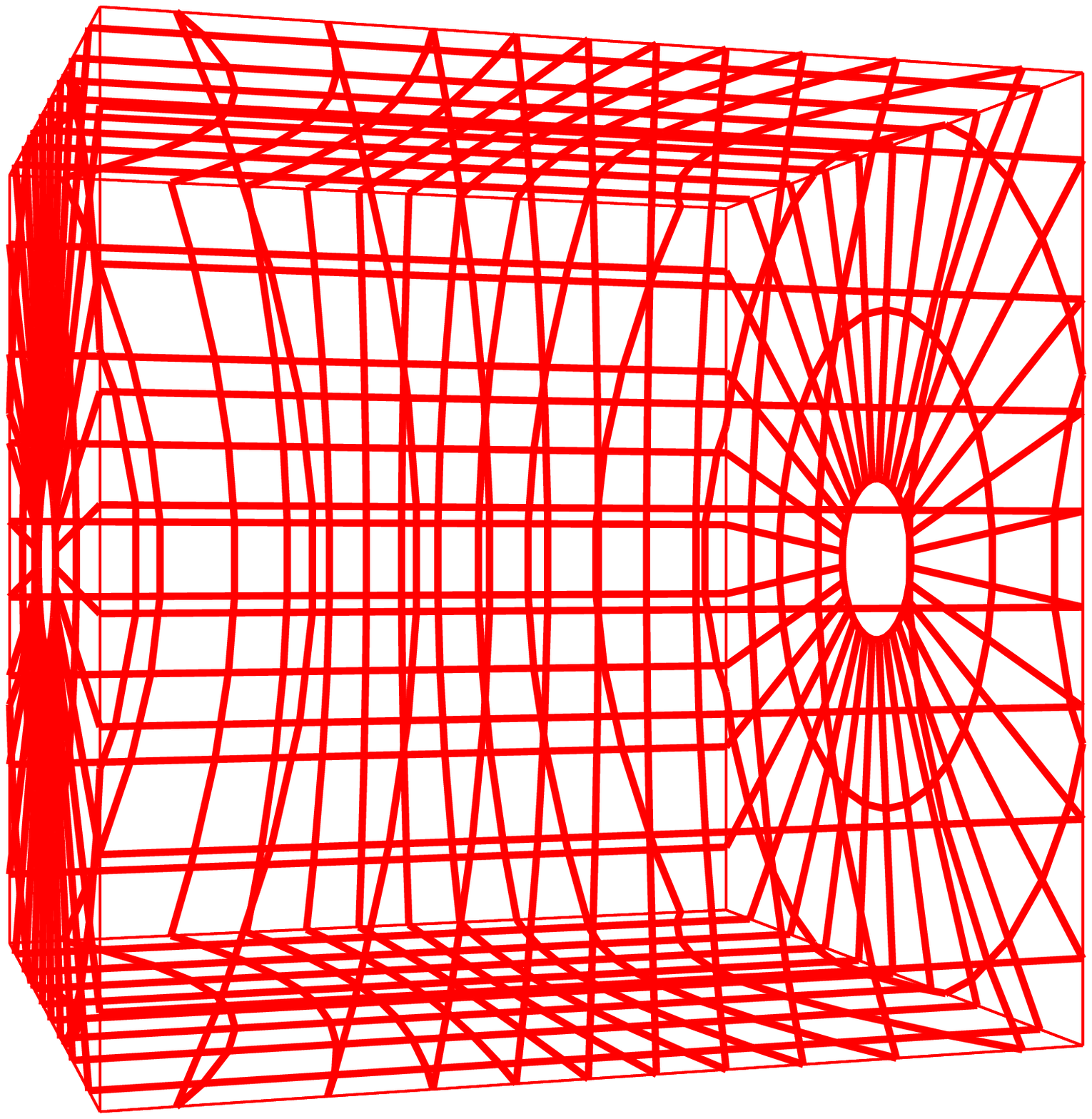,width=4cm,silent=} &
      \epsfig{figure=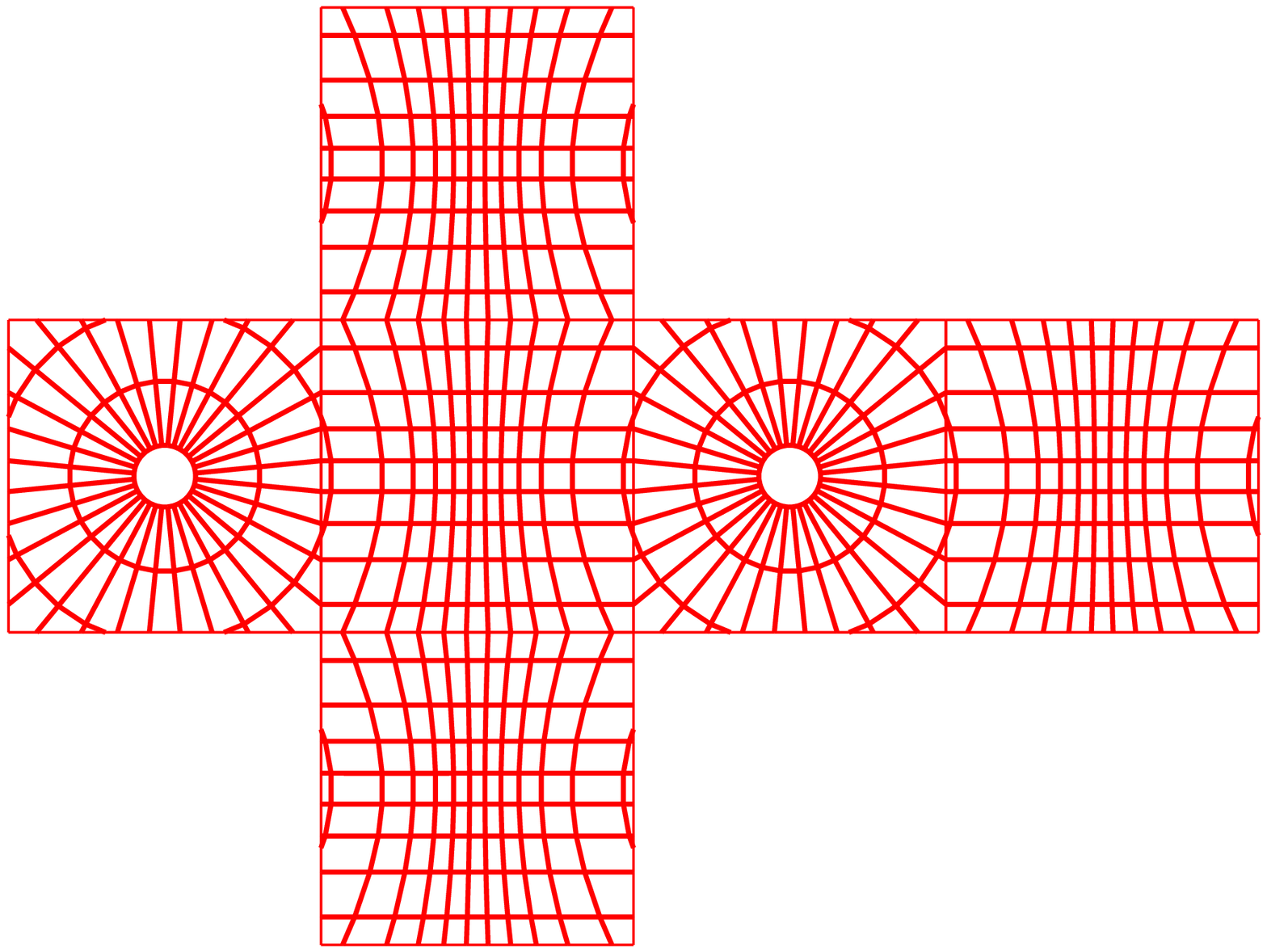,width=5cm,silent=}\\
      (g) & (h) & (i)
    \end{tabular}
  }
  \caption[Cubical Gaussian maps of polyhedra]%
           {\capStyle{(a) An octahedron, (d) a dioctagonal pyramid,
           (g) an ellipsoid-like polyhedron level~16,
           (b,e,h) the \cgm{} of the respective polytope, and (c,f,i)
           the \cgm{} unfolded.}}
  \label{fig:models1}
\end{figure*}

Each vertex that coincides with a unit-cube corner or lies on a
unit-cube edge contains a pointer to a vertex of a planar map
associated with an adjacent cube face that represents the same central
projection. Vertices that lie on a unit-cube edge (but do not coincide
with
\begin{wrapfigure}[5]{r}{3.0cm}
  \centerline{
    \pspicture[](-1.3,-1.1)(1.3,0.9)
    \psset{unit=1cm,linewidth=1pt}
    \rput{ 30}(0,0){\rput{*0}(1,0){\ovalnode{a}{0,0}}}
    \rput{150}(0,0){\rput{*0}(1,0){\ovalnode{b}{1,0}}}
    \rput{270}(0,0){\rput{*0}(1,0){\ovalnode{c}{2,0}}}
    \ncarc{<-}{b}{a}
    \ncarc{<-}{c}{b}
    \ncarc{<-}{a}{c}
    \endpspicture
  }
\end{wrapfigure}
unit-cube corners) come in pairs. Two vertices that form such a
pair lie on the unit-square boundary of planar maps associated with
adjacent cube faces, and they point to each other. Vertices that
coincide with unit-cube corners come in triplets and form cyclic
chains ordered clockwise around the respective vertices.
The diagram on the right specifies one of the eight cyclic chains.
The left and right indices of each pair specify a planar-map id and a
corner id, respectively. All specific connections are listed in
Table~\ref{tab:coord-system}. Recall that all maps are facing outwards.
As a convention, edges incident to a vertex are ordered clockwise around
the vertex, and edges that form the boundary of a face are ordered
counterclockwise. The {\tt Polyhedron\_3} and \arr{}
data structures for example both use a \dcel{} data structure that
follows the convention above.

\begin{wrapfigure}[7]{l}{3.3cm}
  \centerline{
    \pspicture[](-0.7,-0.3)(1.9,1.5)
    \psset{unit=1cm,linewidth=0.5pt}
    \rput{0}(1,1){
      \cnode*( 0, 0){2.5pt}{a0}
      \cnode (-2, 0){2.5pt}{a1}
      \cnode ( 0,-2){2.5pt}{a2}
      \cnode (-2,-2){2.5pt}{a5}
      \pnode(-1,-0.7){a3}
      \pnode(-0.7,-1){a4}
      \psarc[linewidth=0.5pt,arrowsize=3pt 3]{<-}(0,0){1}{200}{250}
      \rput{45}(0,0.2){
	\rput{-45}(0.1,0){
	  \cnode*( 0, 0){2.5pt}{b0}
	  \cnode (-2, 0){2.5pt}{b1}
	  \rput{-45}(0,0){
	    \cnode (0, 1){2.5pt}{b2}
	    \pnode(-0.5,0.3){b3}
	    \pnode(-0.8,-0.1){b4}
	  }
	  \rput{0}(-2,0){
	    \rput{-45}(0,0){
	      \cnode(0,1){2.5pt}{b5}
	    }
	  }
	  \psarcellipse[linewidth=0.5pt,arrowsize=3pt 3](0,0)(1,0.5){80}{160}
	}
      }
      \rput{-45}(0.2,0){
	\rput{45}(0,0.1){
	  \cnode*( 0, 0){2.5pt}{c0}
	  \cnode ( 0,-2){2.5pt}{c1}
	  \rput{-45}(0,0){
	    \cnode (0, 1){2.5pt}{c2}
	    \pnode(0.6,0.2){c3}
	    \pnode(0.9,-0.1){c4}
	  }
	  \rput{0}(0,-2){
	    \rput{-45}(0,0){
	      \cnode(0,1){2.5pt}{c5}
	    }
	  }
	  \psarcellipse[linewidth=0.5pt,arrowsize=3pt 3](0,0)(0.5,1){290}{10}
	}
      }
    }
    \ncline[linewidth=0.5pt]{a1}{a0}
    \ncline[linewidth=0.5pt]{a2}{a0}
    \ncline[linewidth=0.5pt]{b1}{b5}
    \ncline[linewidth=0.5pt]{b5}{b2}
    \ncline[linewidth=1pt]{a3}{a0}
    \ncline[linewidth=1pt]{a4}{a0}
    \ncline[linewidth=0.5pt]{a5}{a1}
    \ncline[linewidth=0.5pt]{a5}{a2}
    \ncline[linewidth=0.5pt]{b1}{b0}
    \ncline[linewidth=0.5pt]{b2}{b0}
    \ncline[linewidth=1pt]{b3}{b0}
    \ncline[linewidth=1pt]{b4}{b0}
    \ncline[linewidth=0.5pt]{c1}{c0}
    \ncline[linewidth=0.5pt]{c2}{c0}
    \ncline[linewidth=0.5pt]{c5}{c1}
    \ncline[linewidth=0.5pt]{c5}{c2}
    \ncline[linewidth=1pt]{c3}{c0}
    \ncline[linewidth=1pt]{c4}{c0}
    \endpspicture
  }
\end{wrapfigure}
We provide a fast clockwise traversal of the faces incident to any
given vertex $v$. Clockwise traversals around internal vertices are immediately
available by the \dcel. Clockwise traversals around boundary vertices
are enabled by the cyclic chains above. This traversal is used to calculate
the normal to the (primary) polytope-facet $f = C^{-1}(v)$ and to
render the facet. Fortunately, rendering systems are capable of
handling a sequence of vertices that define a polygon in clockwise
order as well, an order opposite to the conventional ordering above.

\begin{wrapfigure}[5]{r}{2.4cm}
  \centerline{
    \pspicture[](-1,-0.7)(1,0.7)
    \psset{unit=1cm,linewidth=0.5pt}
    \cnode (-1, -1){2.5pt}{a}
    \cnode (-1, -1){2.5pt}{a}
    \cnode*(-1,  1){2.5pt}{0}
    \cnode*(-0.5,1){2.5pt}{1}
    \cnode*( 0,  1){2.5pt}{2}
    \cnode*( 0.5,1){2.5pt}{3}
    \cnode*( 1,  1){2.5pt}{4}
    \cnode ( 1, -1){2.5pt}{5}
    \pnode(-0.7,-0.5){11}
    \pnode(-0.3, 0){22}
    \pnode(0.6, -0.2){33}
    \ncline[linewidth=0.5pt]{a}{0}
    \ncline[linewidth=1pt]{->}{0}{1}
    \ncline[linewidth=1pt]{->}{1}{2}
    \ncline[linewidth=1pt]{->}{2}{3}
    \ncline[linewidth=1pt]{->}{3}{4}
    \ncline[linewidth=0.5pt]{4}{5}
    \ncline[linewidth=0.5pt]{5}{a}
    \ncline[linewidth=1pt]{1}{11}
    \ncline[linewidth=1pt]{2}{22}
    \ncline[linewidth=1pt]{3}{33}
    \endpspicture
  }
\end{wrapfigure}
The data structure also supports a fast traversal over the
planar-map halfedges that form each one of the four unit-square edges.
This traversal is used during construction to quickly locate a vertex
that coincides with a cube corner or lies on a cube edge. It is also
used to update the cyclic chains of pointers mentioned above; see 
Section~\ref{ssec:mscn:cgm:mink_sum}.

We maintain a flag that indicates whether a planar vertex coincides with a
cube corner, a cube edge, or a cube face. At first glance this looks redundant.
After all, this information could be derived by comparing the $x$ and $y$
coordinates to $-1$ and $+1$. However, it has a good reason as explained next.
Using exact number-types often leads to representations of the geometric
objects with large bit-lengths. Even though we use various techniques
to prevent the length from growing exponentially~\cite{fwh-cfpeg-04}, we
cannot prevent the length from growing at all. Even the computation of a single
intersection requires a few multiplications and additions. 
Cached information computed once and stored at the features of the planar map
avoids unnecessary processing of potentially long representations.

\begin{wraptable}[10]{r}{6.6cm}
  \vspace{-24pt}
  \caption[The complexity of the dioctagonal pyramid \cgm{} planar maps]
          {\capStyle{The number of features of the six planar maps of
           the \cgm{} of the dioctagonal
           pyramid.}}
  \label{tab:pyramid_features}
  \centerline{
  \begin{tabular}{|l||r|r|r|}
    \hline
    \multicolumn{1}{|c||}{\textbf{Planar map}} &
    \multicolumn{1}{p{\smallCellWidth}|}{\textbf{~V~}} &
    \multicolumn{1}{p{\smallCellWidth}|}{\textbf{HE}} &
    \multicolumn{1}{p{\smallCellWidth}|}{\textbf{~F~}} \\
    \hline
    0, ($x = -1$) & 12 &  32 &  6 \\
    1, ($y = -1$) & 28 &  80 & 14 \\
    2, ($z = -1$) & 12 &  32 &  6 \\
    3, ($x = 1$)  & 12 &  32 &  6 \\
    4, ($y = 1$)  & 21 &  72 & 17 \\
    5, ($z = 1$)  & 12 &  32 &  6 \\
    \hline
    \textbf{Total} & 97 & 280 & 55 \\
    \hline
  \end{tabular}}
\end{wraptable}
The table to the right shows the number of vertices ({\bf V}),
halfedges ({\bf HE}), and faces ({\bf F}) of the six planar maps that
comprise the \cgm{} of the \Index{dioctagonal pyramid} shown in
Figure~\ref{fig:models1} (d,e,f). The number of faces of each planar map 
include the unbounded face.
Table \ref{tab:rep} shows the number of features in the primal and
dual representations of a small subset of our polytopes
collection, on which we report the results of experiments below. The
number of planar features is the total number of features of the six
planar maps.

\begin{table}[!hbp]%
  \caption[Complexities of primal and dual representations]%
            {\capStyle{Complexities of the primal and dual representations.
             DP --- Dioctagonal Pyramid,
             PH --- Pentagonal Hexecontahedron,
             TI --- Truncated Icosidodecahedron,
             GS4 --- Geodesic Sphere level~4,
             El16 --- Ellipsoid-like polyhedron made of 16 latitudes and 32 
             longitudes,
             t - time consumption in seconds.}}
  \label{tab:rep}
  \centerline{
    \begin{tabular}{|l||r|r|r||r|r|r|r|r|r|r|r|}
      \hline
      \multirow{2}*{\textbf{Object Type}} &
      \multicolumn{3}{c||}{\textbf{Primal}} &
      \multicolumn{4}{c|}{{\bf SGM}} &
      \multicolumn{4}{c|}{{\bf CGM}}\\\cline{2-12}
      &
      \multicolumn{1}{|c}{\textbf{V}} &
      \multicolumn{1}{|c}{\textbf{E}} &
      \multicolumn{1}{|c||}{\textbf{F}} &
      \multicolumn{1}{c|}{\textbf{V}} &
      \multicolumn{1}{c|}{\textbf{HE}} &
      \multicolumn{1}{c|}{\textbf{F}} &
      \multicolumn{1}{c|}{\textbf{t}} &
      \multicolumn{1}{c|}{\textbf{V}} &
      \multicolumn{1}{c|}{\textbf{HE}} &
      \multicolumn{1}{c|}{\textbf{F}} &
      \multicolumn{1}{c|}{\textbf{t}}\\
      \hline
      Tetrahedron &   4 &   6 &   4 &   4 &   12 &   4 & 0.01 &  42 &  102 &  21 & 0.01\\
      Octahedron  &   6 &  12 &   8 &  10 &   28 &   6 & 0.01 &  24 &   48 &  12 & 0.01\\
      Icosahedron &  12 &  30 &  20 &  21 &   62 &  12 & 0.01 &  72 &  192 &  36 & 0.01\\
      DP          &  17 &  32 &  17 &  25 &   80 &  17 & 0.01 &  97 &  280 &  55 & 0.01\\
      PH          &  60 & 150 &  92 & 101 &  318 &  60 & 0.03 & 200 &  600 & 112 & 0.02\\
      TI          & 120 & 180 &  62 &  77 &  390 & 120 & 0.05 & 230 &  840 & 202 & 0.03\\
      GS4         & 252 & 750 & 500 & 506 & 1512 & 252 & 0.08 & 708 & 2124 & 366 & 0.07\\
      El16        & 482 & 992 & 512 & 528 & 2016 & 482 & 0.11 & 776 & 2752 & 612 & 0.06\\
      \hline
    \end{tabular}
  }
\end{table}

\subsection{Exact Minkowski Sums}
\label{ssec:mscn:cgm:mink_sum}
A similar argument regarding the representation of Minkowski sums
using Gaussian maps mentioned in Section~\ref{sec:mscn:sgm-method}
holds for the cubical Gaussian maps with the unit cube replacing the
unit sphere. More precisely, a single map that subdivides the unit
sphere is replaced by six planar maps, and the computation of a single
overlay is replaced by the computation of six overlays of
corresponding pairs of planar maps. Recall that each (primal) vertex
is associated with a planar-map face, and is the sum of two vertices
associated with the two overlapping faces of the two \cgm's of the two
input polytopes, respectively.

Each planar map in a \cgm{} is a convex subdivision. Finke and
Hinrichs~\cite{fh-oscpl-95} describe how to compute the overlay of such
special subdivisions optimally in linear time. However, a preliminary
investigation shows that a large constant governs the linear complexity,
which renders this choice less attractive.
Instead, we resort to a sweep-line based algorithm that exhibits good
practical performance, and incurs a mere logarithmic factor over the
optimal computing time. In particular we use the overlay operation
supported by the Arrangement package. It requires the
provision of a complementary component that is responsible for
updating the attributes of the \dcel{} features of the resulting six
planar maps; see Section~\ref{ssec:aos:facilities:overlay}.

The overlay operates on two instances of \arr{}.
In the description below $v_1$, $e_1$, and $f_1$ denote a vertex, a
halfedge, and a face of the first operand respectively, and $v_2$,
$e_2$, and $f_2$ denote the same feature types of the second operand,
respectively. When the overlay operation progresses, new vertices,
halfedges, and faces of the resulting planar map are created based on
features of the two operands. Exactly ten cases described below
arise and must be handled. When a new feature is created its attributes
are updated. The updates performed in all cases except for case (1) are
simple and require constant time.
\begin{compactenum}
\item A new vertex $v$ is induced by coinciding vertices $v_1$ and $v_2$.\\
  The location of the vertex $v$ is set to be the same as the
  location of the vertex $v_1$ (the locations of $v_2$ and $v_1$ must
  be identical). The induced vertex is not artificial if and only if 
  \setcounter{ms-const:cntr}{1}(\roman{ms-const:cntr}) at least one of the vertices 
  $v_1$ or $v_2$ is not artificial, or
  \addtocounter{ms-const:cntr}{1}(\roman{ms-const:cntr}) the vertex lies on a cube 
  edge or coincides with a cube corner, and both vertices $v_1$ and 
  $v_2$ have non-artificial incident halfedges that do not overlap.

\item A new vertex $v$ is induced by a vertex $v_1$ that lies on an edge 
  $e_2$.\\
  The location of the vertex $v$ is set to be the same as the
  location of the vertex $v_1$. $v$ is not artificial if and only if 
  $v_1$ is not artificial or $e_2$ is not artificial.

\item An analogous case of a vertex $v_2$ that lies on an edge $e_1$.

\item A new vertex $v$ is induced by a vertex $v_1$ that is contained in a 
  face $f_2$.\\
  The attributes of the vertex $v$ are set to be the same as the
  attributes of the vertex $v_1$.

\item An analogous case of a vertex $v_2$ contained in a face $f_1$.

\item A new vertex $v$ is induced by the intersection of two edges $e_1$ and 
  $e_2$.\\
  The vertex $v$ cannot lie on a cube edge and cannot coincide with a cube
  corner. Thus, it is necessarily not artificial.
  
\item A new edge $e$ is induced by the overlap of two edges $e_1$ and $e_2$.\\
  The edge $e$ is not artificial if at least one of $e_1$ or $e_2$ is not
  artificial.
  
\item A new edge $e$ is induced by an edge $e_1$ that is contained in a 
  face $f_2$.\\
  The edge $e$ is not artificial if $e_1$ is not artificial.

\item An analogous case of an edge $e_2$ contained in a face $f_1$.

\item A new face $f$ is induced by the overlap of two faces $f_1$ and $f_2$.\\
  The primal vertex associated with $f$ is set to be the sum of the primal
  vertices associated with $f_1$ and $f_2$, respectively.
\end{compactenum}

After the six map overlays are computed, some maintenance operations
must be performed to obtain a valid \cgm{} representation. As mentioned
above, the global data consists of the six planar maps and 24
references to vertices that coincide with the unit-cube corners. For
each planar map we traverse its vertices, obtain the four vertices
that coincide with the unit-cube corners, and initialize the global
data. We also update the cyclic chains of pointers to vertices that
represent identical central projections. To this end, we exploit the
fast traversal over the halfedges that coincide with the unit-cube
edges mentioned in Section~\ref{ssec:mscn:cgm:representation}.

The complexity of a single overlay operation is $O(k \log n)$, where 
$n$ is the total number of vertices in the input planar maps, and $k$
is the number of vertices in the resulting planar map. The total number of 
vertices in all the six planar maps in a \cgm{} that represents a polytope
$P$ is of the same order as the number of facets in the primary polytope $P$. 
Thus, the complexity of the entire overlay operation is
$O(F \log (F_1 + F_2))$, where
$F_1$ and $F_2$ are the number of facets in the input polytopes
respectively, and $F$ is the number of facets in the Minkowski sum.

\section{Exact Collision Detection\index{collision detection}}
\label{sec:mscn:3d_col_det}
Computing the separation distance between two polytopes with $m$ and $n$
features respectively can be done in $O(\log m \log n)$ time, after an
investment of at most linear time in preprocessing~\cite{dk-dsppu-90}.
Many practical algorithms that exploit spatial and temporal
coherence between successive queries have been developed, some of which
became classic, such as the GJK algorithm~\cite{gjk-fpcdb-88} and its
improvement~\cite{c-egcmp-97}, and the LC algorithm~\cite{lc-faidc-91} and 
its optimized variations~\cite{el-apqcp-00,ghz-hhdcm-99,m-vcfrp-98}.
Several general-purpose software libraries that offer practical solutions are
available today, such as the SOLID library~\citelinks{solid} based on the
improved GJK algorithm, the SWIFT library~\citelinks{swift} based on an
advanced version of the LC algorithm, the QuickCD
library~\citelinks{quickcd}, and more. For an extensive review of methods
and libraries see the survey~\cite{cpq-lm-04}.

Given two polytopes $P$ and $Q$, detecting collision between them and
computing their relative placement can be conveniently done in the
configuration space, where their Minkowski sum $M = P \oplus (-Q)$ resides.
These problems can be solved in many ways, and not all require the
explicit representation of the Minkowski sum $M$. However, having it
available is attractive, especially when the polytopes are restricted to
translations only, as the combinatorial structure of the Minkowski sum $M$
is invariant to translations of $P$ or $Q$. The algorithms described below 
are based on the following well known observations (see
Chapter~\ref{chap:intro} for definitions):
\begin{align*}
P^u \cap Q^w & \neq \emptyset \Leftrightarrow w - u \in M = P \oplus (-Q)\ ,\\
\pi(P^u,Q^w) & = \min\{\Vert t \Vert \,|\, (w - u + t) \in M, t \in \mathbb{R}^3\}\ ,\\
\delta_r(P^u,Q^w) & = \inf\{\alpha \,|\, (w - u + \alpha \vec{r}) \notin M, \alpha \in \mathbb{R}\}\ .
\end{align*}

Given two polytopes $P$ and $Q$ in either (spherical) Gaussian map or
\cgm{} representation respectively, we reflect $Q$ through the origin
to obtain $-Q$, compute the Minkowski sum $M = P \oplus (-Q)$, and
retain it in the respective Gaussian-map representation $G(M)$. Then,
each time $P$ or $Q$ or both translate by two vectors $u$ and $w$ in
$\rrr$ respectively, we apply a procedure that determines whether the
query point $s = w - u$ is inside, on the boundary of, or outside $M$.
In addition to an enumeration of one of the three conditions above, the
procedure returns a witness of the respective relative placement. Let
$r$ be a ray emanating from an internal point $c \in M$ and going through
$s$. If the (spherical) Gaussian map representation is used, the witness
data is the vertex $v = G(f)$ --- a mapping of a facet $f$ of $M$ embedded
on the sphere and stabbed by the ray $r$. If the \cgm{} representation is
used, the witness data is a pair that consists of a vertex $v = G(f)$ ---
a mapping of a facet $f$ of $M$ embedded in a unit cube face, and the
planar map $\mathcal{P}$ containing $v$. This information is used as a
hint in consecutive invocations. The internal point $c$ could be the
average of all vertices of $M$ computed once and retained along $M$,
or just the midpoint of two vertices that have supporting planes with
opposite normals easily extracted from either map representation. Once
$f$ is obtained, determining whether $P^u$ and $Q^w$ collide is trivial,
according to the first formula (of the three) above. The query point
$s$ is contained in the open half-space defined by the supporting plane to
$f$ if and only if $s$ is outside of $M$, this occurs if and only if $P^u$
does not collide with $Q^w$.

\begin{wrapfigure}[9]{r}{5.7cm}
  \vspace{-12pt}
  \centerline{
    \epsfig{figure=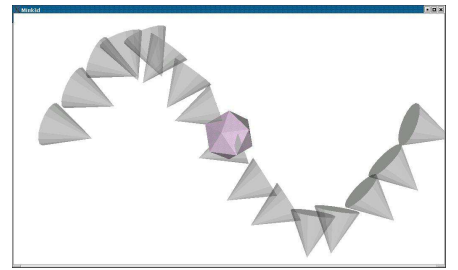,width=\linewidth,silent=}
  }
  \caption[Simulation of motion]{\capStyle{Simulation of motion.}}
  \label{fig:sim}
\end{wrapfigure}
The collision-detection procedure applies a local walk on the
respective Gaussian map faces. It starts with some vertex $v_0$, and
then performs a loop moving from the current vertex to a neighboring
vertex, until it reaches the final vertex. If the \cgm{}
representation is used, the procedure may jump from a planar map
associated with one cube-face to a different one associated with 
an adjacent cube-face. The first time the procedure is invoked, $v_0$ is
chosen to be a vertex that lies on the central projection of the normal
directed in the same direction as the ray $r$. In consecutive calls,
$v_0$ is chosen to be the final vertex of the previous call
exploiting spatial and temporal coherence.
The figure above is a snapshot of a simulation program that
detects collision between a static obstacle and a moving robot, and
draws the obstacle and the trail of the robot; instructions to obtain,
install, and execute the program appear in the Appendix.
The Minkowski sum is recomputed only when the robot is rotated, which
occurs every other frame. The program has the distinctive feature of
being able to identify the case where the robot grazes the obstacle,
but does not penetrate it, since it produces exact results. The
computation takes just a fraction of a second on a Pentium PC clocked
at 1.7~GHz using either representation. Similar procedures that compute
the directional penetration-depth and minimum distance are available as
well.

\section{Minkowski Sum Complexity}
\label{sec:mscn:mink_complexity}
\begin{figure*}[!htp]
  \centerline{
    \begin{tabular}{ccc}
      \epsfig{figure=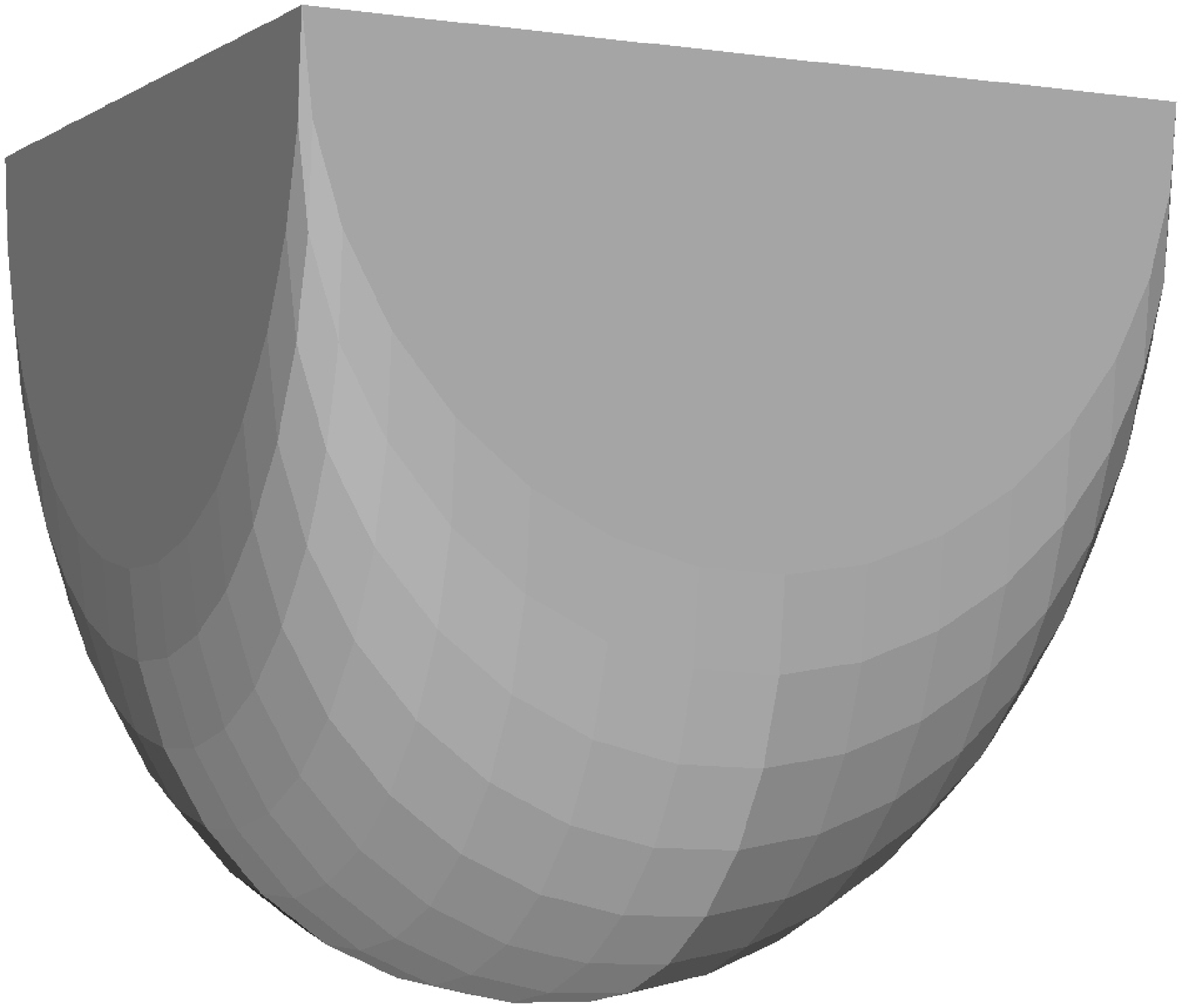,width=4cm,silent=} &
      \epsfig{figure=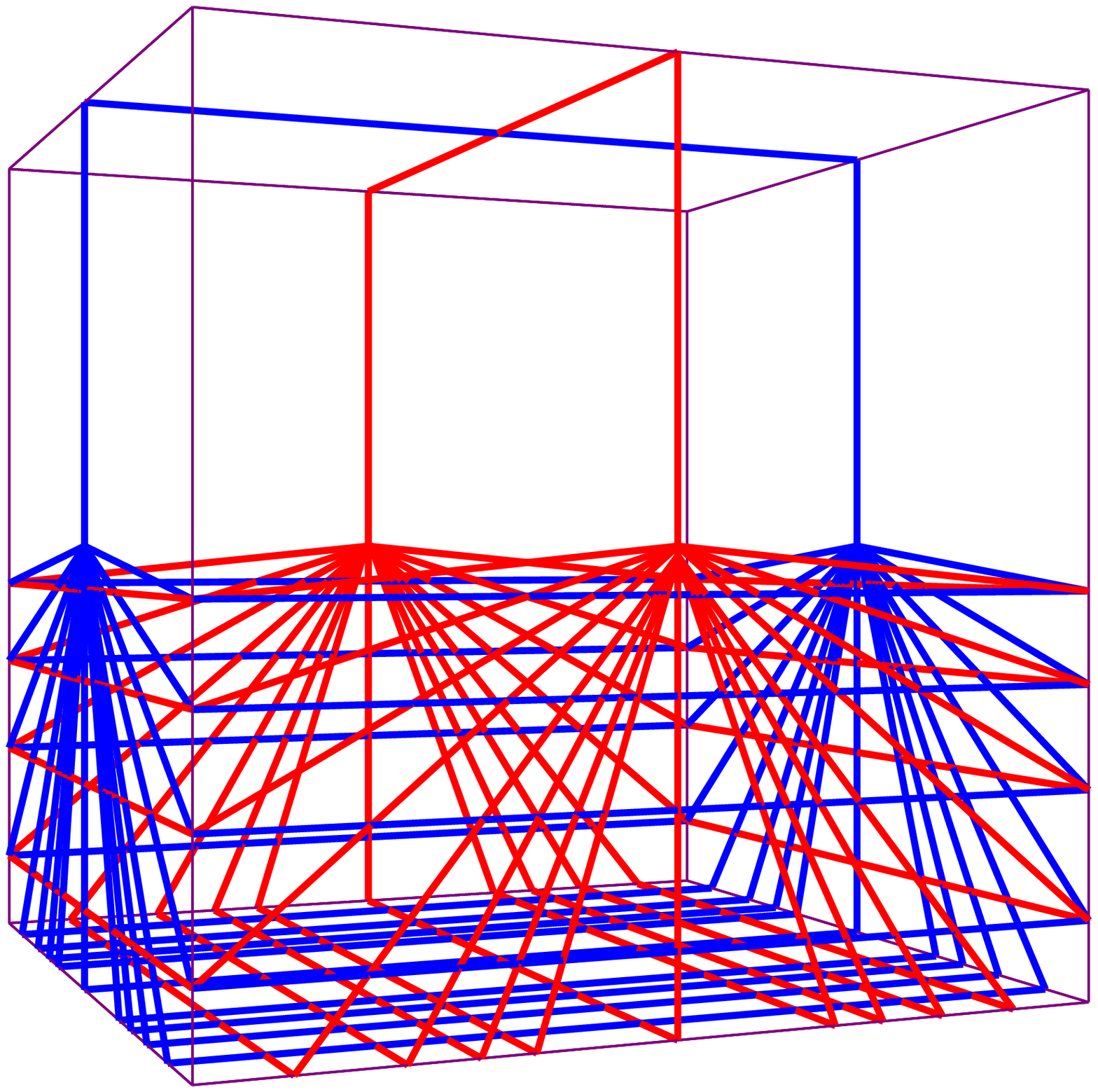,width=4cm,silent=} &
      \epsfig{figure=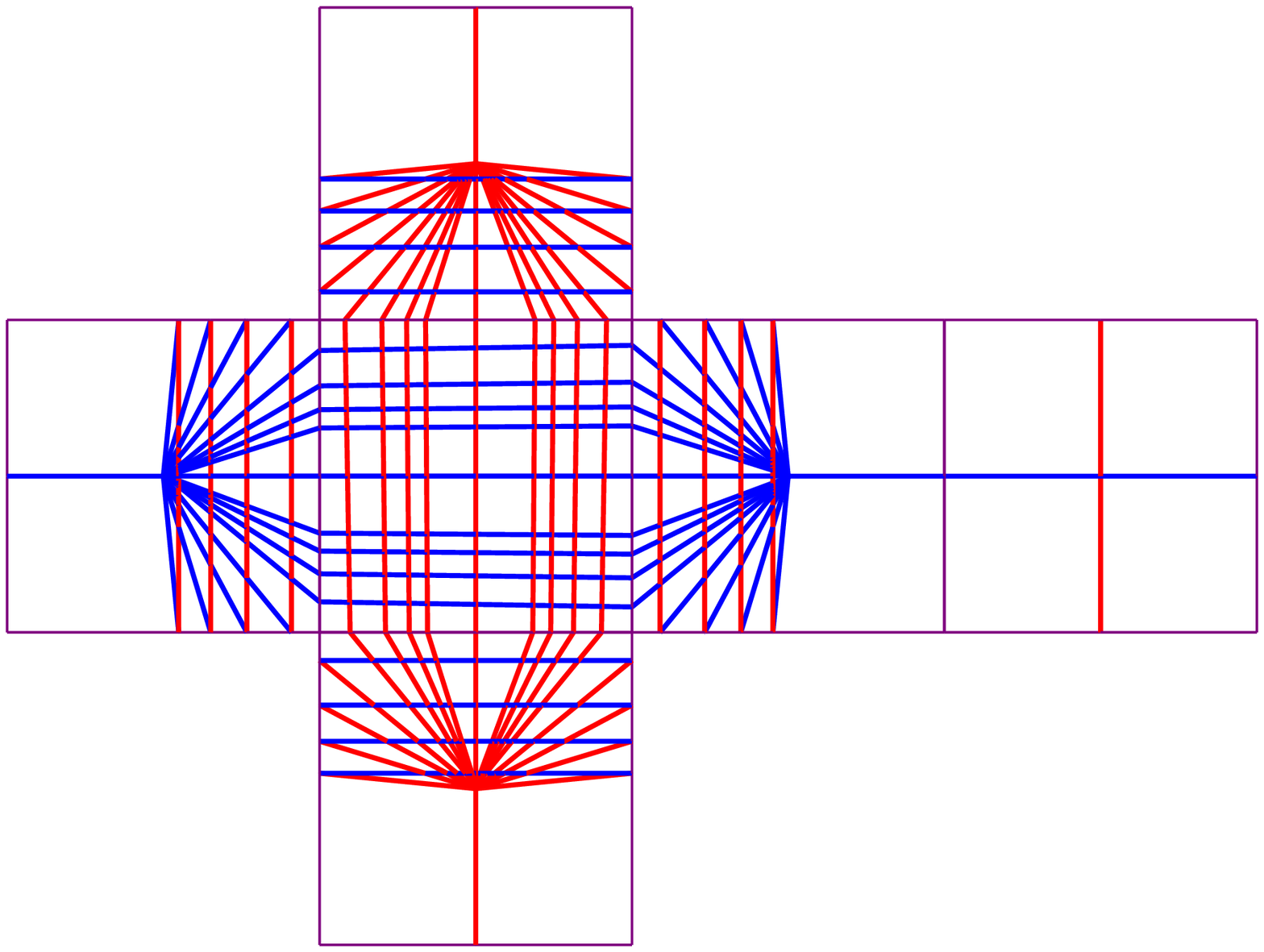,width=5cm,silent=}\\
      	(a) & (b) & (c)
    \end{tabular}
  }
  \caption[The Minkowski sum (of two polytopes) the complexity of which
           is maximal]%
          {\capStyle{(a) The Minkowski sum (of two polytopes) the
           complexity of which is maximal, (b) the \cgm{} of the Minkowski
           sum, and (c) the \cgm{} unfolded. Red lines are graphs of edges
           that originate from one polytope and blue lines are graphs of
           edges that originate from the other.}}
  \label{fig:max11}
\end{figure*}
In Chapter~\ref{chap:mink-sum-complexity} we show that the exact
maximum number of facets of Minkowski sums of two polytopes is
$4mn - 9m - 9n + 26$, where $m$ and $n$ are the number of facets of
the two summands, respectively. This bound is tight~\cite{fhw-emcms-07}.
The example depicted in Figure~\ref{fig:max11} shows a Minkowski sum
that reaches this maximum complexity. It is the sum of two identical
polytopes, each containing $n$ faces ($n = 11$ in
Figure~\ref{fig:max11}), but one is rotated about the vertical axis 
approximately\footnote{The results of all rotations are approximate,
as we have not yet dealt with exact rotation. One of our future goals
is the handling of exact rotations.} $90^\circ$ relative to the
other. The polytopes are specifically shaped to ensure that the number
of intersections between dual edges, which are the mappings of the
polytope edges, is maximal. A careful counting reveals that the number
of vertices in the dual representation excluding the artificial
vertices reaches $4 \cdot 11 \cdot 11 - 9 \cdot 11 - 9 \cdot 11 + 26 =
312$, which is the number of facets of the Minkowski sum.

Not every pair of polytopes yields a Minkowski sum proportional to $mn$.
As a matter of fact, it can be as low as $n$ in the extremely-degenerate
case of two identical polytopes variant under scaling. Even if no
degeneracies exist, the complexity can be proportional to only $m+n$,
as in the case of two geodesic spheres\footnote{An icosahedron, every
triangle of which is divided into $4^l$ triangles using class I
aperture 4 partition method, whose vertices are elevated to the
circumscribing sphere~\cite{swk-gdggs}.} level $l = 2$ slightly
rotated with respect to each other, depicted in
Figure~\ref{fig:geo}. Naturally, an algorithm that accounts for all
pairs of vertices, one from each polytope, is rendered inferior
compared to an output-sensitive
algorithm\index{algorithm!output sensitive} like ours in such cases,
as we demonstrate in the next section.
\begin{figure*}[!htp]
  \centerline{
    \begin{tabular}{ccc}
      \epsfig{figure=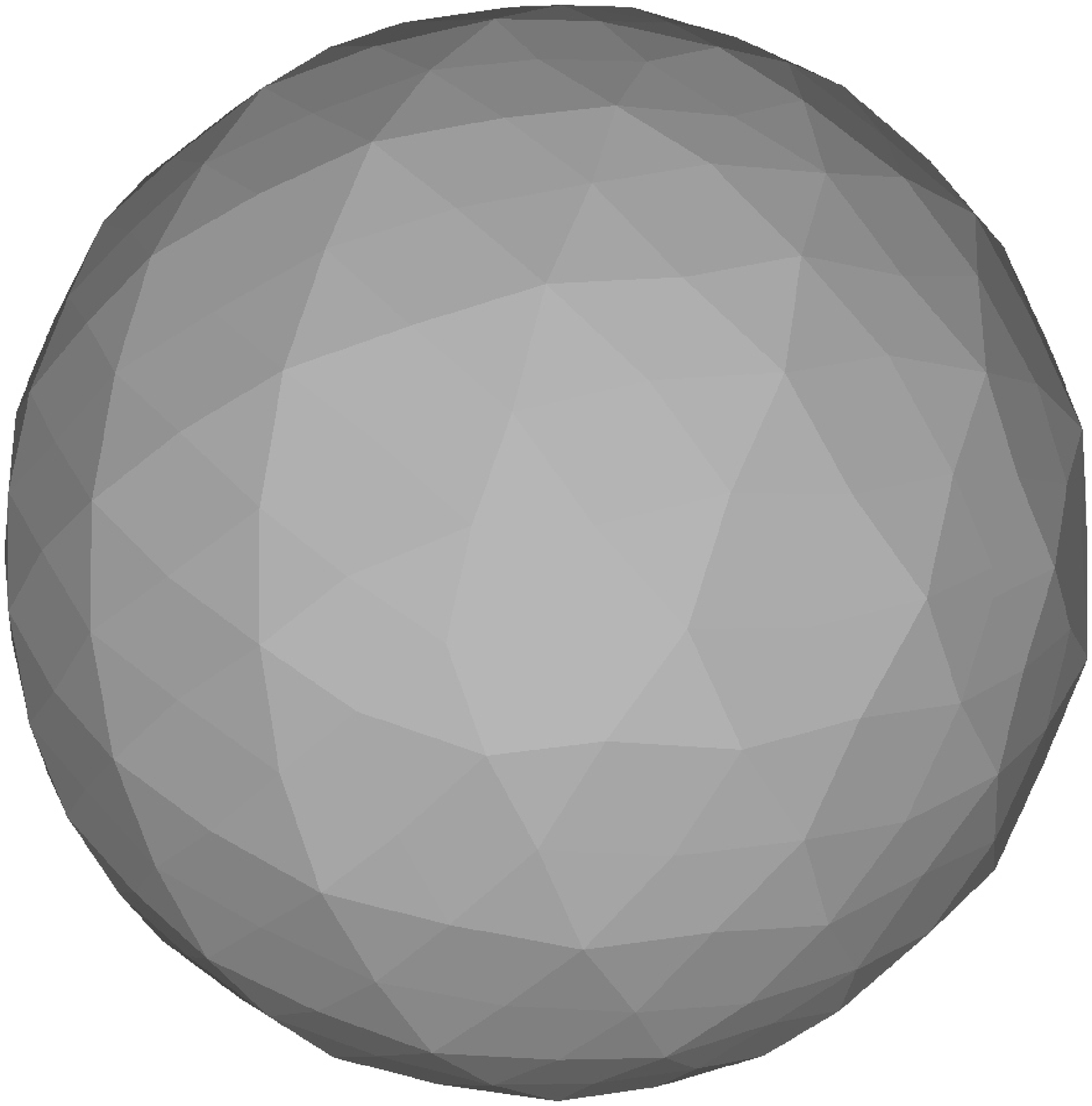,width=4cm,silent=} &
      \epsfig{figure=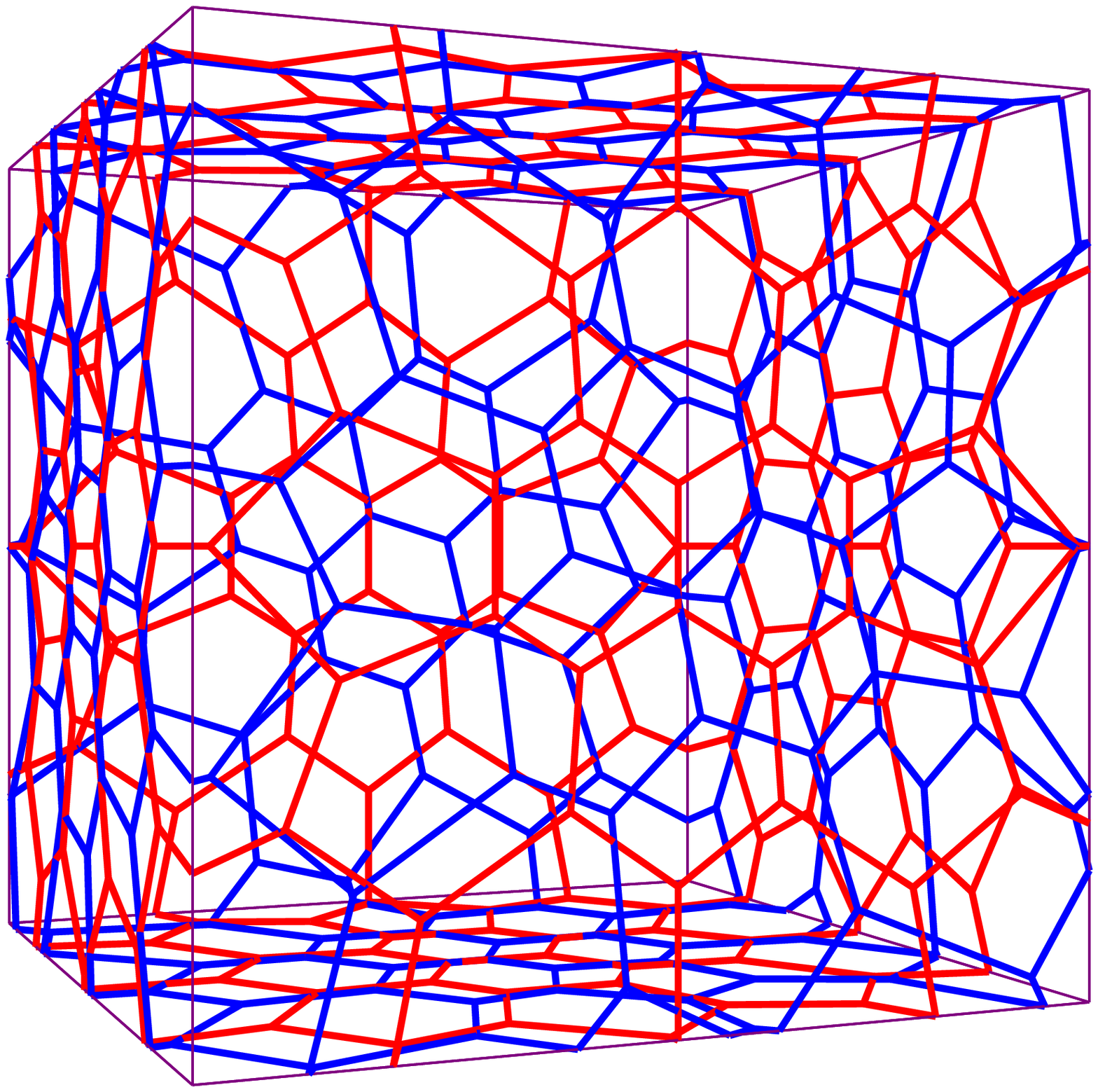,width=4cm,silent=} &
      \epsfig{figure=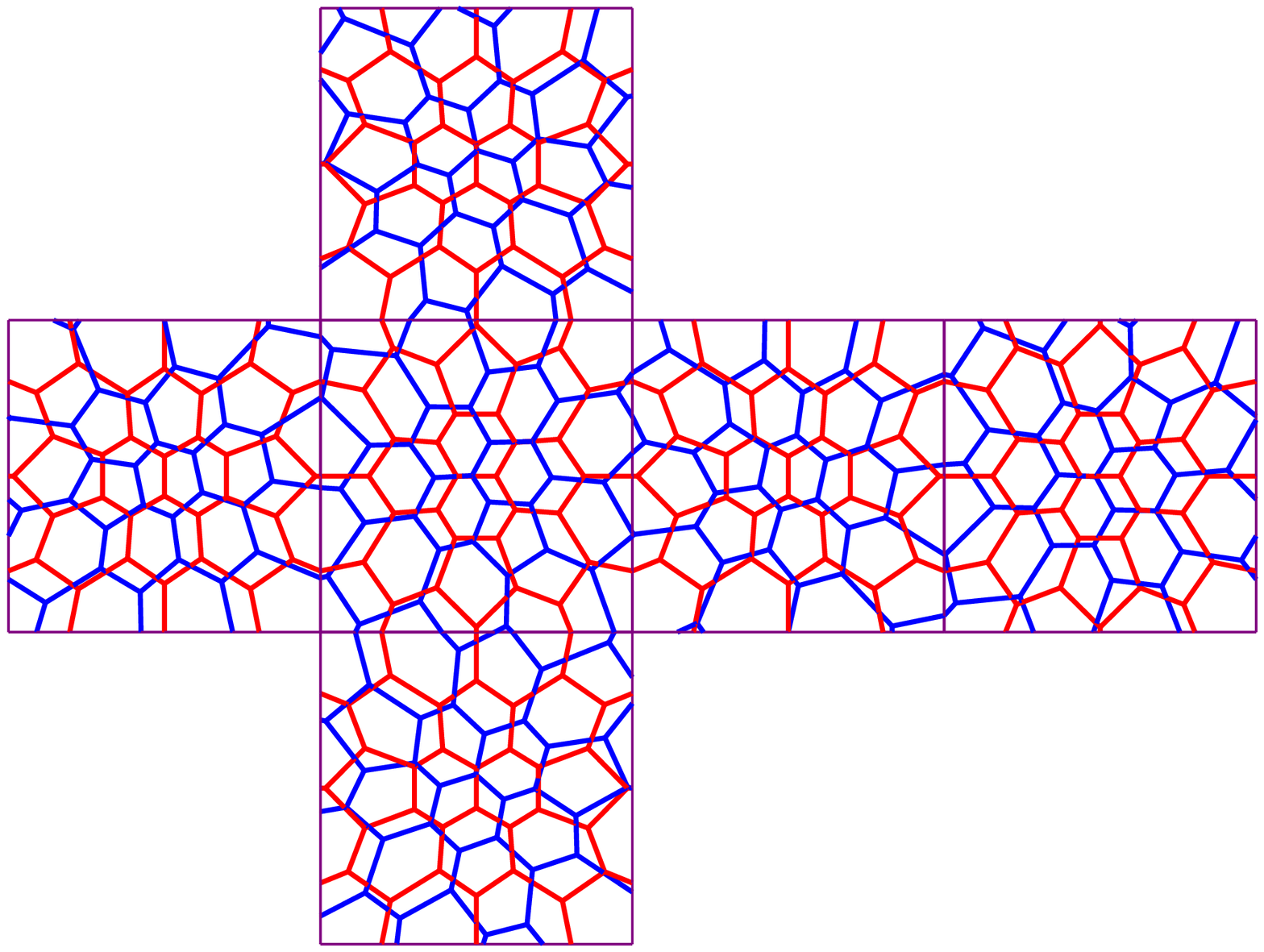,width=5cm,silent=} \\
      	(a) & (b) & (c)
    \end{tabular}
  }
  \caption[The Minkowski sum of two geodesic spheres level~2]%
          {\capStyle{(a) The Minkowski sum of two geodesic spheres
           level~2 slightly rotated with respect to each other, (b)
           the \cgm{} of the Minkowski sum, and (c) the \cgm{}
           unfolded.}}
   \label{fig:geo}
\end{figure*}

\section{Experimental Results}
\label{sec:mscn:exp_res}
\begin{figure*}[!htp]
  \centerline{
    \begin{tabular}{ccc}
      \epsfig{figure=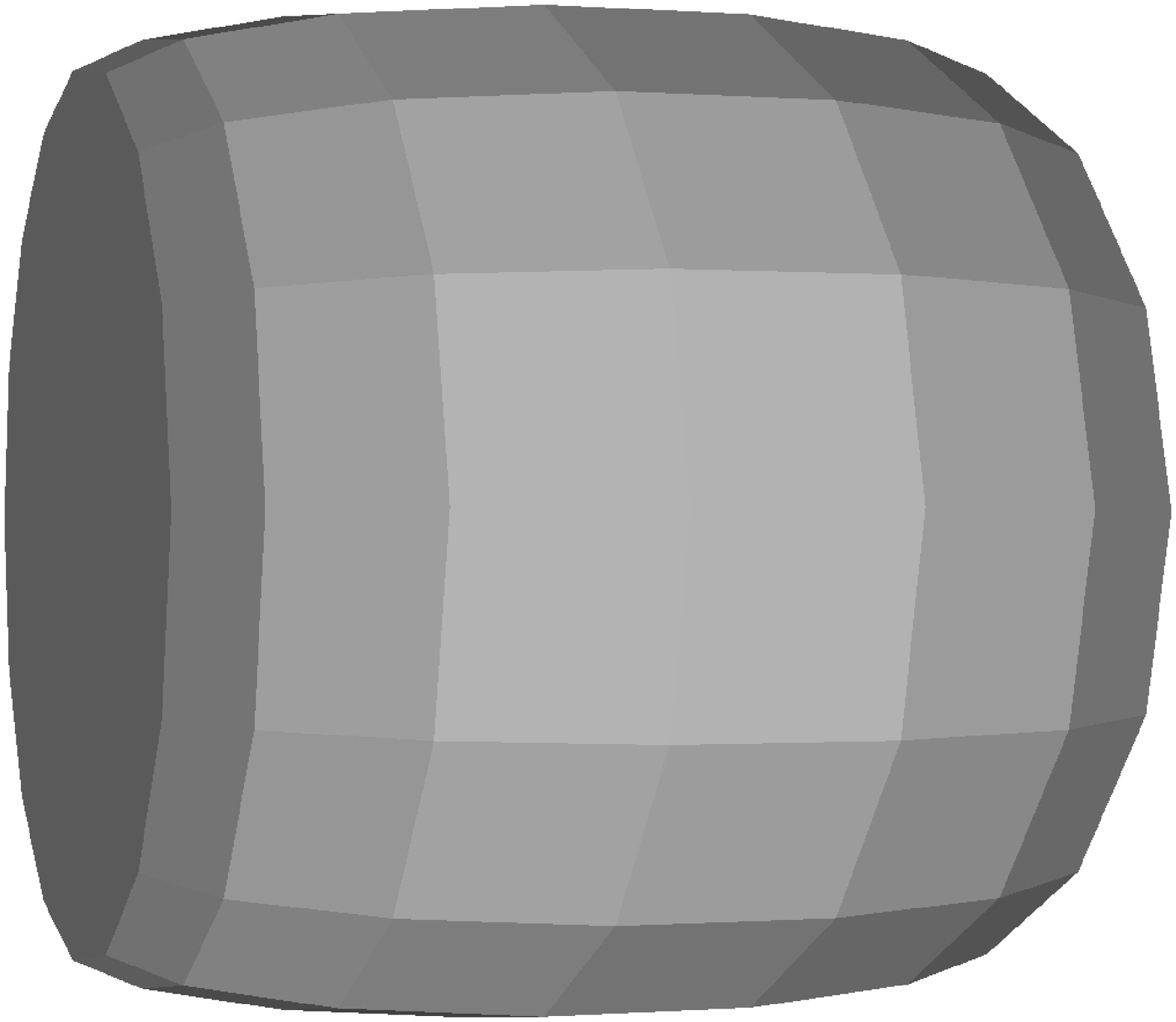,width=4cm,silent=}
      &
      \epsfig{figure=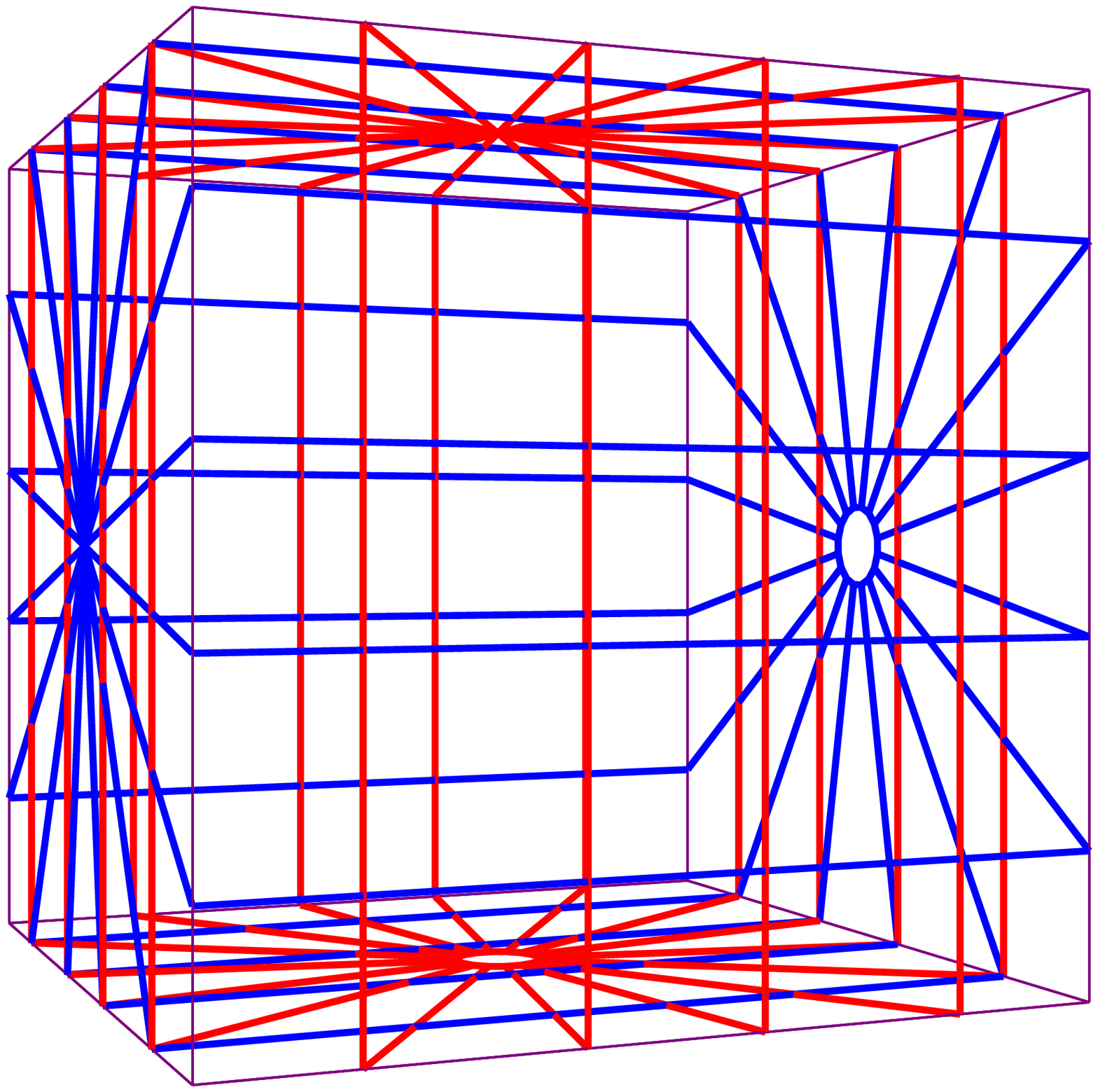,width=4cm,silent=}
      &
      \epsfig{figure=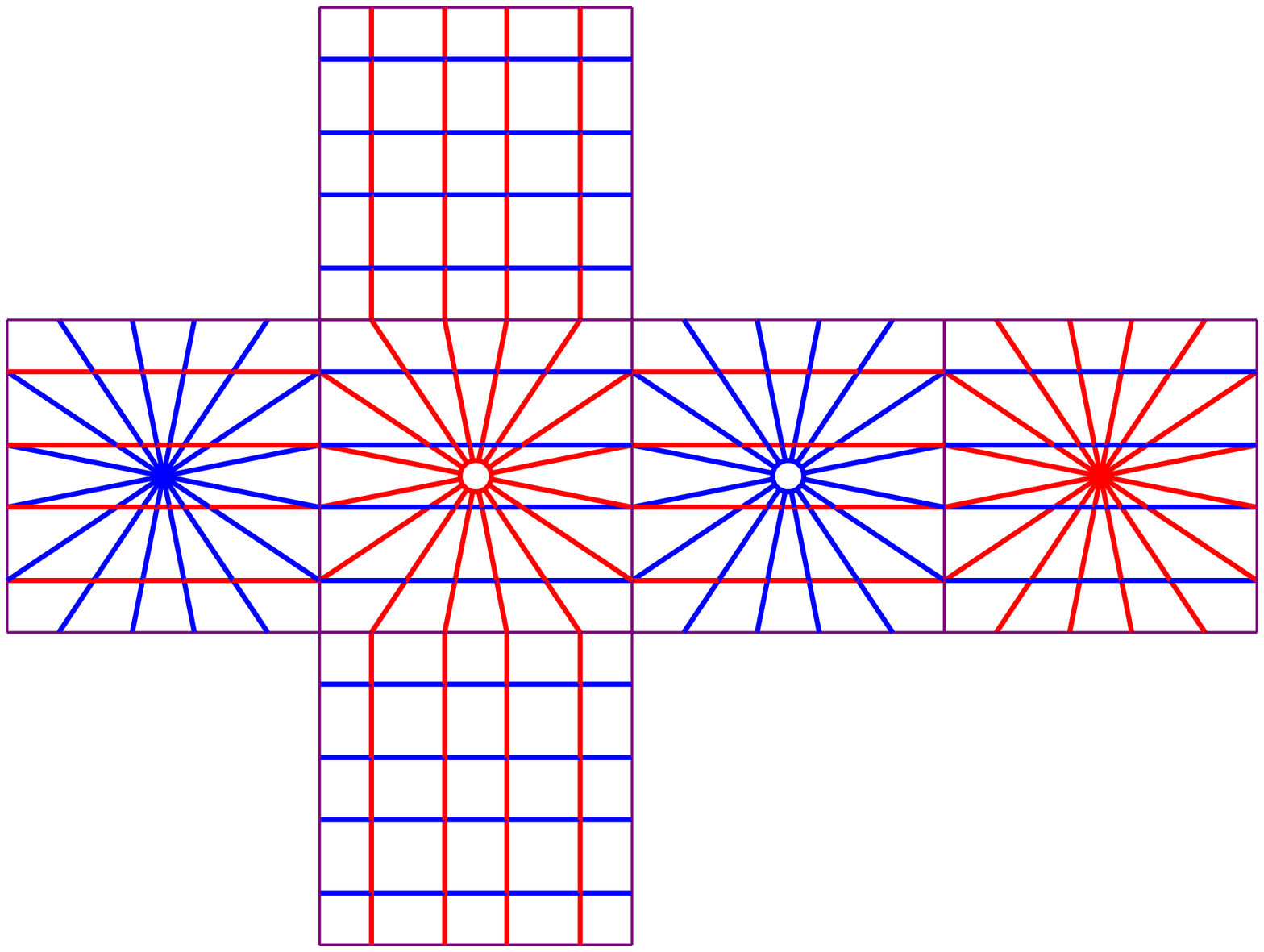,width=5cm,silent=}\\
      (a) & (b) & (c)
    \end{tabular}
  }
  \caption[The Minkowski sum of two orthogonal squashed dioctagonal pyramids]%
          {\capStyle{(a) The Minkowski sum of two approximately orthogonal
           squashed dioctagonal pyramids, (b) the \cgm, and (c) the \cgm{}
           unfolded.}}
  \label{fig:pyr_pyr2}
\end{figure*}
As mentioned above, the Minkowski sum of two polytopes is the
\Index{convex hull} of the pairwise sum of the vertices of the two
polytopes. We have implemented this straightforward method using the
\cgal{} \ccode{convex\_hull\_3} function, which uses the
\cPolyhedron{} data structure to represent the resulting polytope, and
used it to verify the correctness of our two methods. We compared the
time it took to compute exact Minkowski sums using our two methods, a
third method implemented by Hachenberger based on Nef polyhedra
embedded on the sphere~\cite{hkm-bosnc-07}, a fourth method
implemented by Weibel~\citelinks{w-ms}, based on an output-sensitive
algorithm designed by Fukuda~\cite{f-zcmac-04}, and the naive
convex-hull method.

The Nef-based method is not specialized for Minkowski sums. It can
compute the overlay of two arbitrary Nef polyhedra embedded on the
sphere, which can have open and closed boundaries, facets with holes,
and lower dimensional features. The overlay is computed by two separate
hemisphere-sweeps. 

Fukuda's algorithm relies on linear programming. Its complexity is
$O(\delta LP(3,\delta) V)$, where $\delta = \delta_1 + \delta_2$ is the
sum of the maximal degrees of vertices, $\delta_1$ and $\delta_2$, in the
two input polytopes respectively, $V$ is the number of vertices of the
resulting Minkowski sum, and $LP(d,m)$ is the time required to solve a
linear programming in $d$ variables and $m$ inequalities. Note that
Fukuda's algorithm is more general, as it can be used to compute the
Minkowski sum of polytopes in an arbitrary dimension $d$, and as far as
we know, it has not been optimized specifically for $d = 3$.

\begin{table*}[!htp]
  \setlength{\tabcolsep}{4pt}
  \caption[Complexities of primal and dual Minkowski-sum representations]%
           {\capStyle{Complexities of primal and dual Minkowski-sum
           representations.
           DP --- Dioctagonal Pyramid,
           ODP --- Dioctagonal Pyramid orthogonal to DP,
           PH --- Pentagonal Hexecontahedron,
           TI --- Truncated Icosidodecahedron,
           GS4 --- Geodesic Sphere level~4,
           RGS4 --- Rotated Geodesic Sphere level~4,
           El16 --- Ellipsoid-like polyhedron made of 16 latitudes and 32 
             longitudes,
           OEl16 --- Ellipsoid-like polyhedron made of 16 latitudes
             and 32 longitudes orthogonal to El16.}} 
  \label{tab:mink-cnt}
  \centerline{
    \begin{tabular}{|l|l||r|r|r||r|r|r||r|r|r|}
    \hline
    \multirow{3}*{\textbf{Summand 1}} &
    \multirow{3}*{\textbf{Summand 2}} &
    \multicolumn{9}{c|}{\textbf{Minkowski Sum}}\\\cline{3-11}
    & &
    \multicolumn{3}{c||}{\textbf{Primal}} &
    \multicolumn{3}{c||}{\sgmo} &
    \multicolumn{3}{c|}{\cgmo}\\\cline{3-11}
    & &
    \multicolumn{1}{c}{\textbf{V}} &
    \multicolumn{1}{|c}{\textbf{E}} &
    \multicolumn{1}{|c||}{\textbf{F}} &
    \multicolumn{1}{c|}{\textbf{V}} &
    \multicolumn{1}{c|}{\textbf{HE}} &
    \multicolumn{1}{c||}{\textbf{F}} &
    \multicolumn{1}{c|}{\textbf{V}} &
    \multicolumn{1}{c|}{\textbf{HE}} &
    \multicolumn{1}{c|}{\textbf{F}}\\
    \hline
    Icosahedron & Icosahedron &
                    12 &   30 &   20 &   21 &   62 &   12 &   72 &   192 &   36\\
    \hline
    DP   & ODP  &  131 &  261 &  132 &  141 &  540 &  131 &  242 &   832 &  186\\
    \hline
    PH   & TI   &  248 &  586 &  340 &  429 & 1712 &  429 &  514 &  1670 &  333\\
    \hline
    GS4  & RGS4 & 1048 & 2582 & 1536 & 1564 & 5220 & 1048 & 1906 &  6288 & 1250\\
    \hline
    El16 & OEl16 &
                  2260 & 4580 & 2322 & 2354 & 9224 & 2260 & 2826 & 10648 & 2510\\
    \hline  
    \end{tabular}
  }
\end{table*}

\begin{table*}[!htp]
  \setlength{\tabcolsep}{4pt}
  \caption[Time consumption of the Minkowski-sum computation]%
           {\capStyle{Time consumption (in seconds) of the Minkowski-sum
             computation.
             \ch{} --- the Convex Hull method,
             \sgmo --- the (spherical) Gaussian map based method,
             \cgmo --- the cubical Gaussian-map based method,
             \ngmo --- the Nef based method,
             \Fuku --- Fukuda's Linear-Programming based algorithm,
             \boldmath $\frac{F_1 F_2}{F}$ --- the ratio between the product of
             the number of input facets and the number of output facets.}}
  \label{tab:mink-time}
  \centerline{
    \begin{tabular}{|l|l||r|r|r|r|r||r|}
    \hline
    \textbf{Summand 1} & \textbf{Summand 2} & \sgmo & \cgmo & \ngmo & \Fuku & \ch &
    {\boldmath $\frac{F_1 F_2}{F}$}\\
    \hline
    Icosahedron & Icosahedron &
                   0.01 & \textbf{0.01} & 0.12 &  0.01 & 0.01 &  20.0\\
    \hline
    DP   & ODP   & 0.04 & \textbf{0.02} & 0.33 &  0.35 & 0.05 &   2.2\\
    \hline
    PH   & TI    & 0.13 & \textbf{0.03} & 0.84 &  1.55 & 0.20 &  10.9\\
    \hline
    GS4  & RGS4  & 0.71 & \textbf{0.12} & 6.81 &  5.80 & 1.89 & 163.3\\
    \hline
    El16 & OEl16 & 1.01 & \textbf{0.14} & 7.06 & 13.04 & 6.91 & 161.3\\
    \hline  
    \end{tabular}
  }
\end{table*}

The results listed in Table~\ref{tab:mink-time}, produced by
experiments conducted on a Pentium PC clocked at 1.7~GHz, show that
our methods are much more efficient in all cases, and the \cgm{}
method in particular is more than fifty times faster than the
convex-hull\index{convex hull} method in one case. The number of
models used as summands in the listed experiments is just a small
fraction of the total number of models in our collection, which
contains hundreds of models of polytopes; see
Section~\ref{sec:software:availability} for instructions how to download
the corresponding files. The listed experiments are just a small
sample of all the experiments we have conducted. The last column
of the table indicates the ratio $\frac{F_1 F_2}{F}$, where $F_1$ and
$F_2$ are the number of facets of the input polytopes respectively,
and $F$ is the number of facets of the Minkowski sum. As this ratio
increases, the relative performance of the output-sensitive algorithms
compared to the convex-hull method improves as expected.
Figure~\ref{fig:models2} illustrates some of the resulting Minkowski
sums listed in Table~\ref{tab:mink-time}.

The \sgmo{}, \cgmo{}, \ngmo{}, and \ch{} methods are all based on \cgal{}.
As the implementation of these methods and of \cgal{} is generic, it is
possible to instantiate specific components, such as the number type,
the geometric kernel, and the segment-handling traits class with the
model of the respective concept that achieves the best result. The code
of the programs used to obtain the results listed in Table~\ref{tab:mink-time}
are instantiated with the {\tt CGAL::Gmpq} exact rational number-type, the
{\tt Simple\_cartesian} geometric kernel, and the {\tt Lazy\_kernel}
kernel adaptor, which performs lazy exact
computations~\cite{pf-glese-06}. The computation is delayed until a
point where the approximation with interval arithmetic is not
precise enough to perform safe comparisons. In other words, these
programs need only to compute to sufficient precision to evaluate
predicates correctly, exploiting a significant relaxation of the naive
concept of numerical exactness.

We experimented with two different exact number types: One provided
by \leda{} 4.4.1, namely {\tt leda::rational}, and another based on
\gmp{} 4.1.2, namely {\tt CGAL::Gmpq}. The former does not normalize
the rational numbers automatically. Therefore, we had to initiate
normalization operations to contain their bit-length growth. In case
of the \cgm{} method, for example, we chose to do it right after the
central projections of the facet-normals are calculated, and before
the chains of segments, which are the mapping of facet-edges, are
inserted into the planar maps. Our experience shows that indiscriminate
normalization considerably slows down the planar-map construction, and
the choice of number type may have a drastic impact on the performance
of the code overall.

\label{sec:mscn:models}
\begin{figure*}[!hbp]%
  \centerline{
    \begin{tabular}{ccc}
      \epsfig{figure=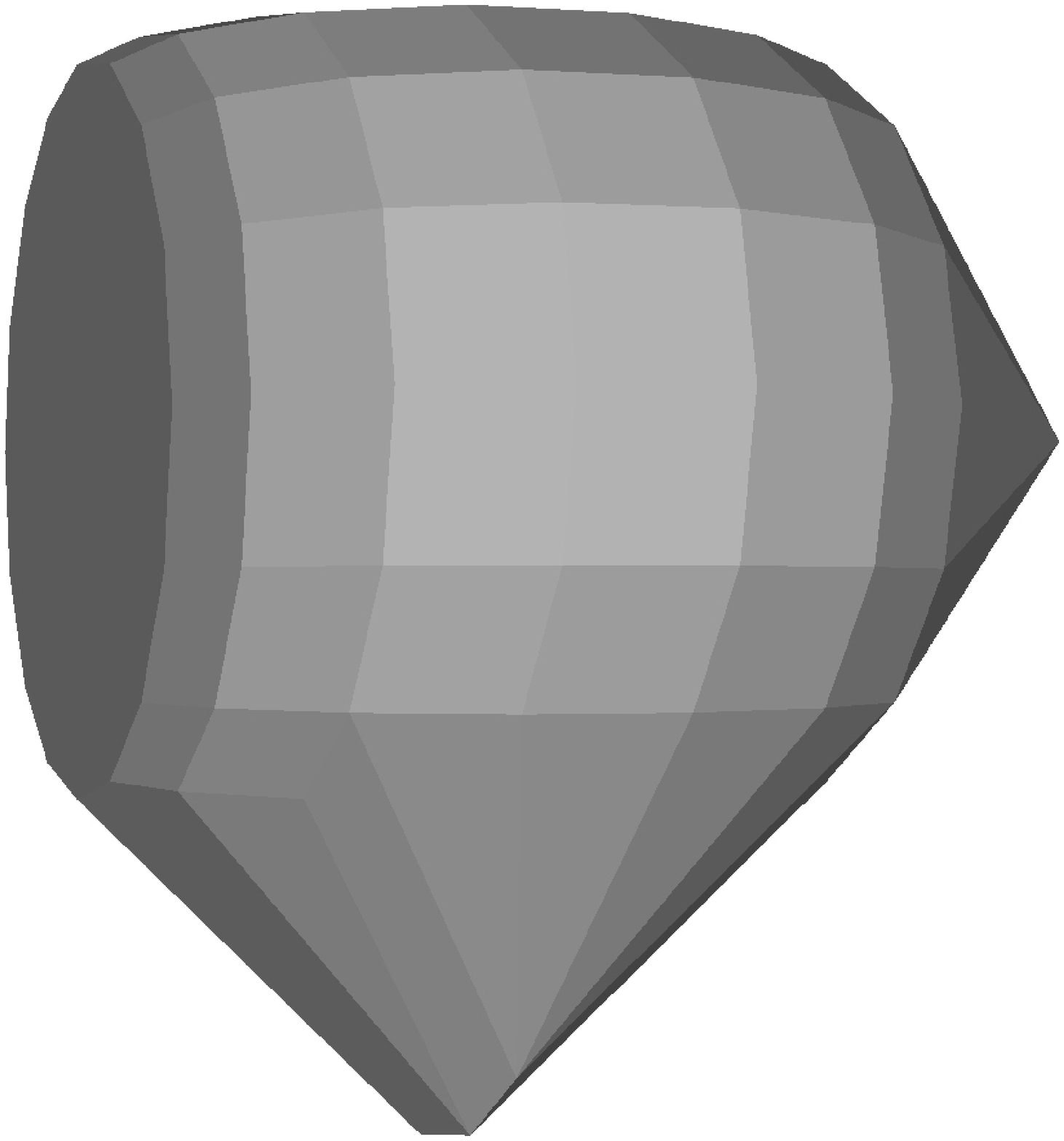,width=4cm,silent=} &
      \epsfig{figure=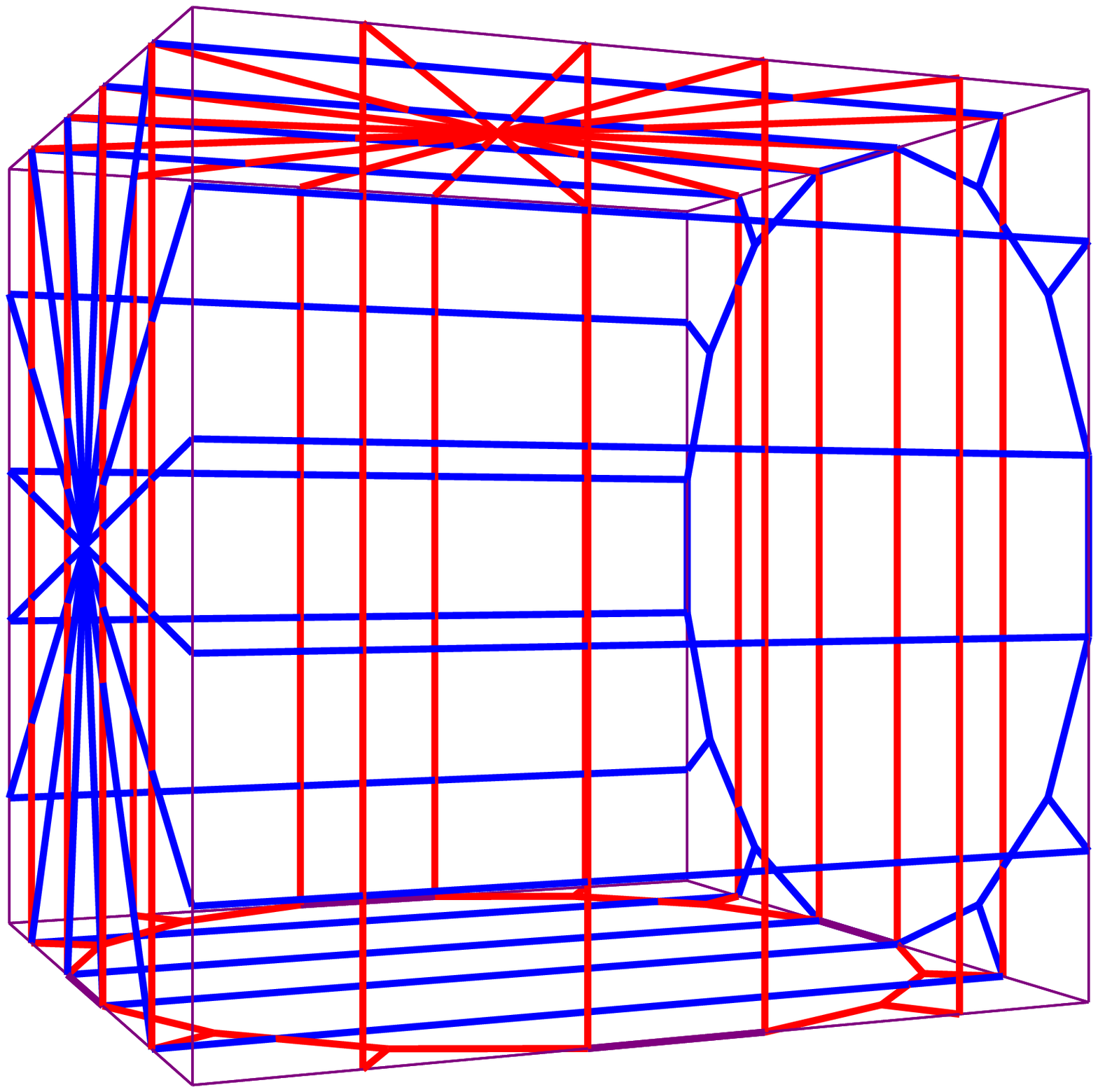,width=4cm,silent=} &
      \epsfig{figure=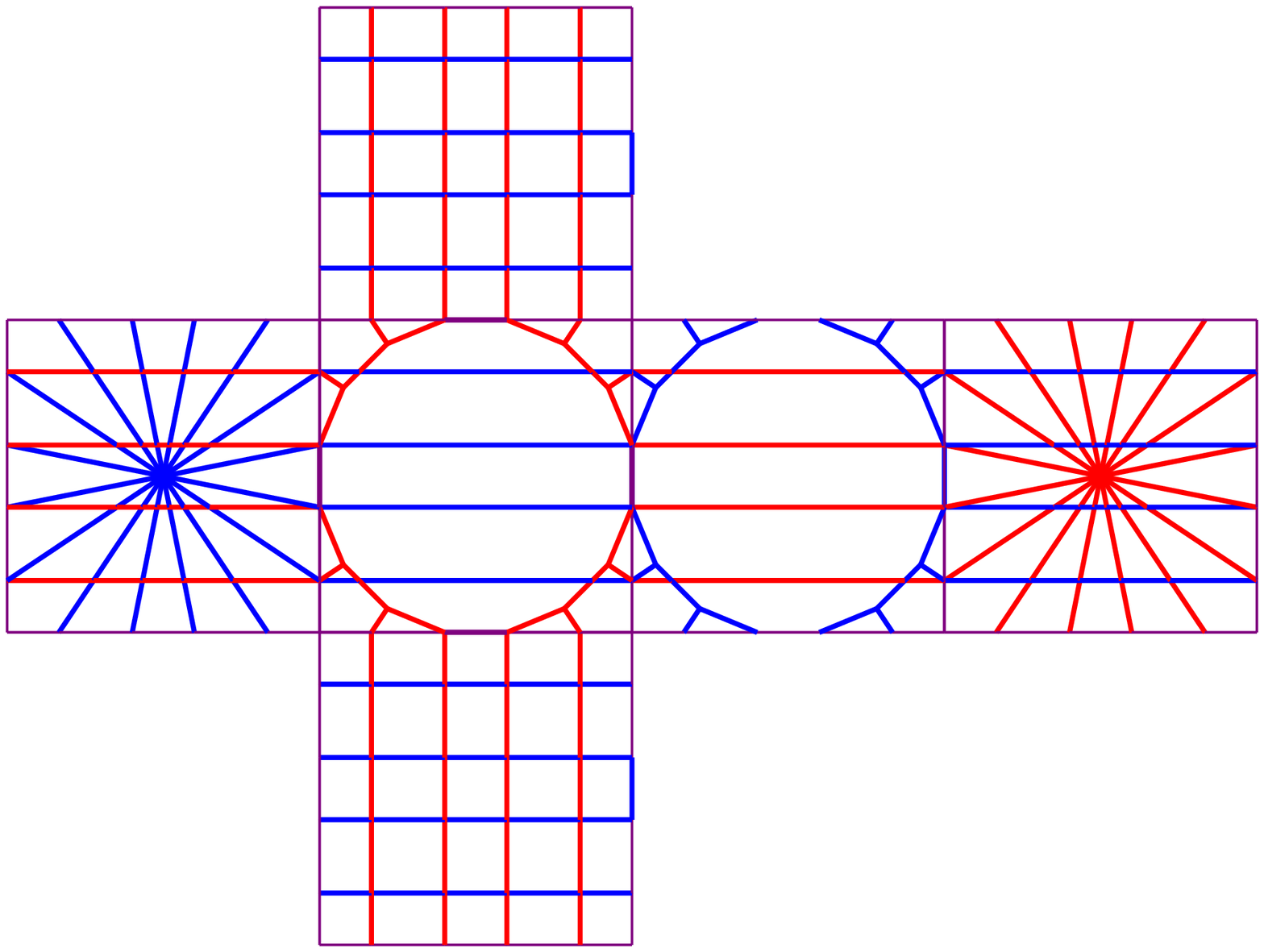,width=5cm,silent=}\\
      (a) & (b) & (c)\\
      \epsfig{figure=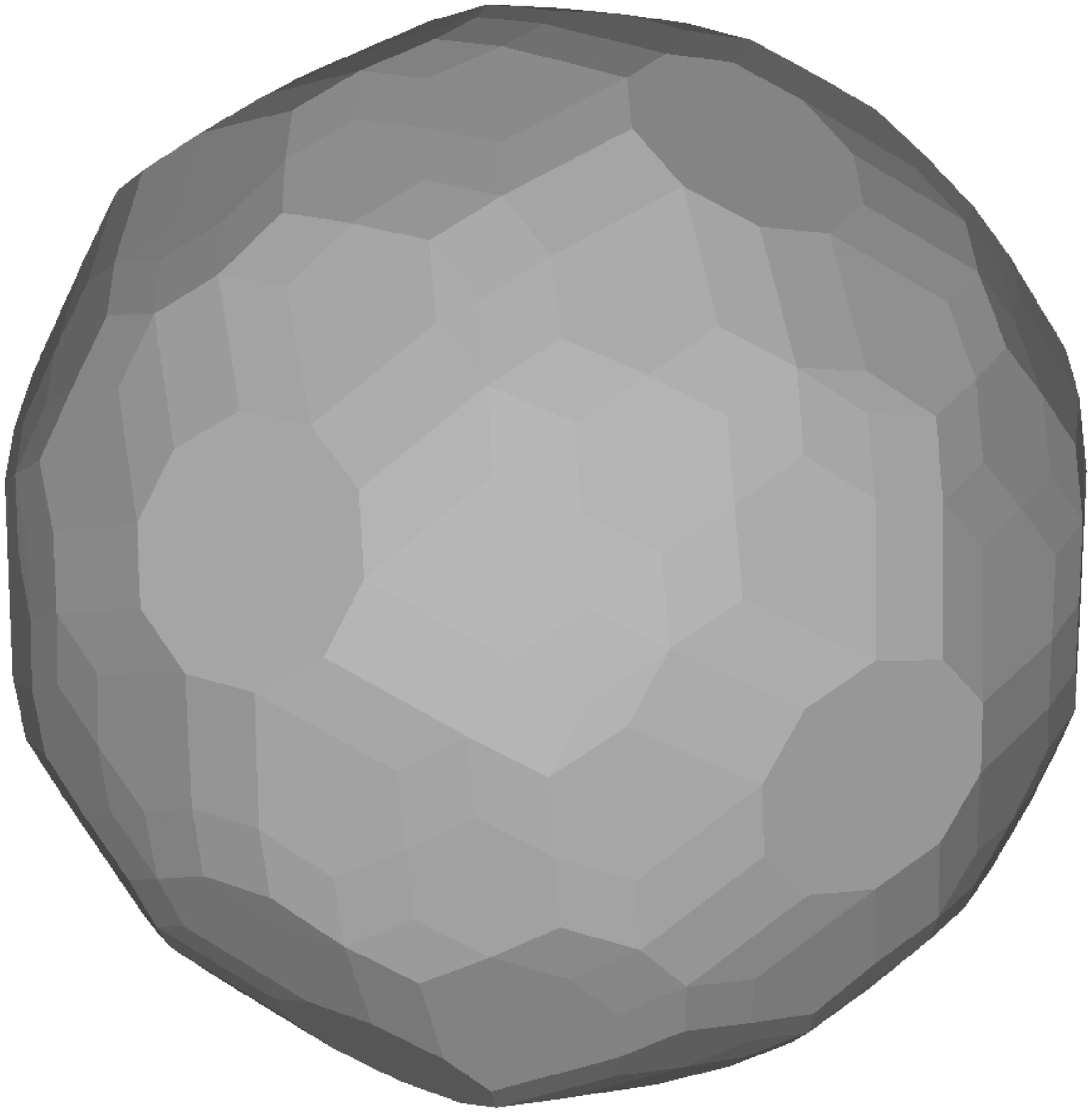,width=4cm,silent=} &
      \epsfig{figure=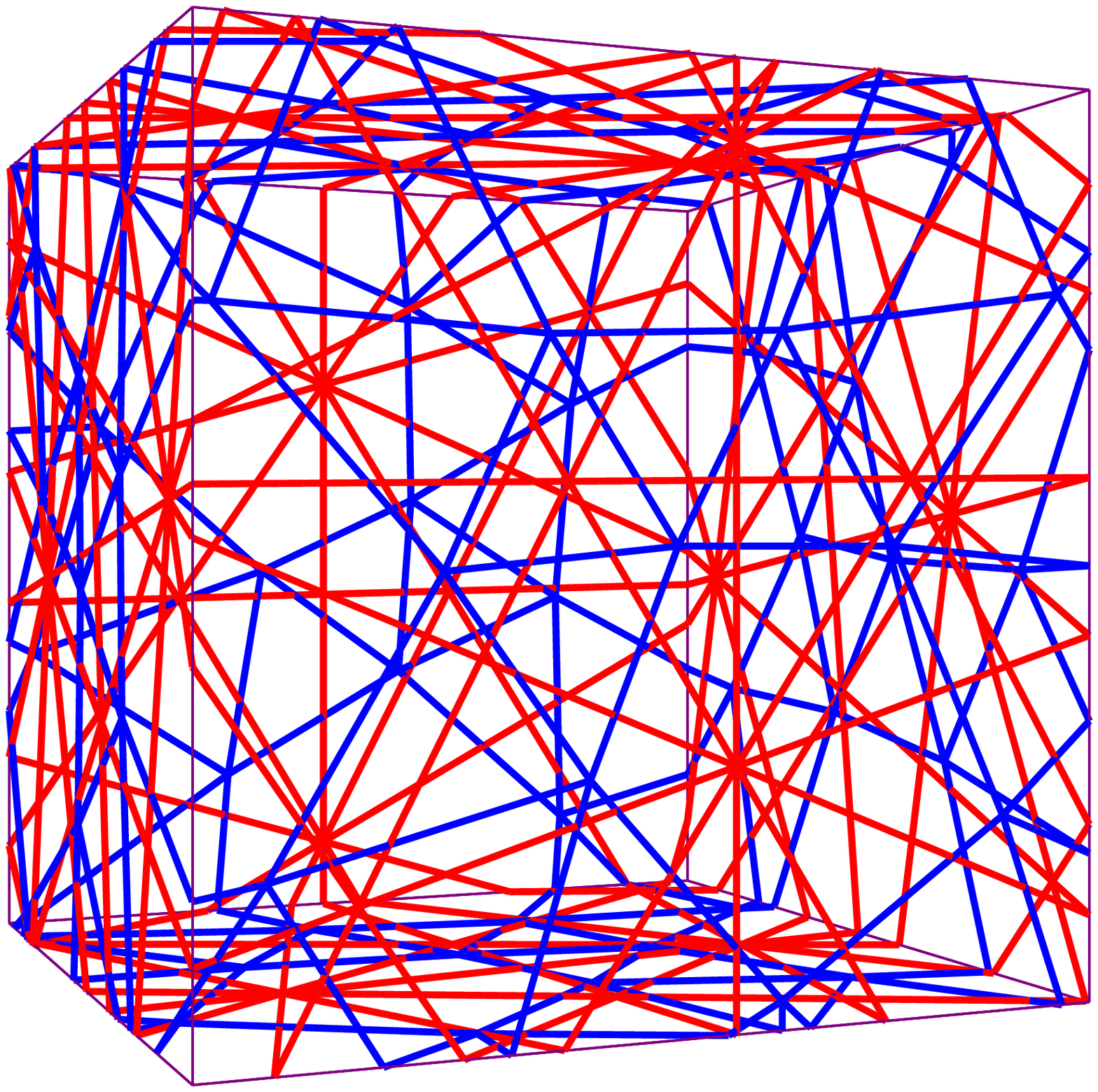,width=4cm,silent=} &
      \epsfig{figure=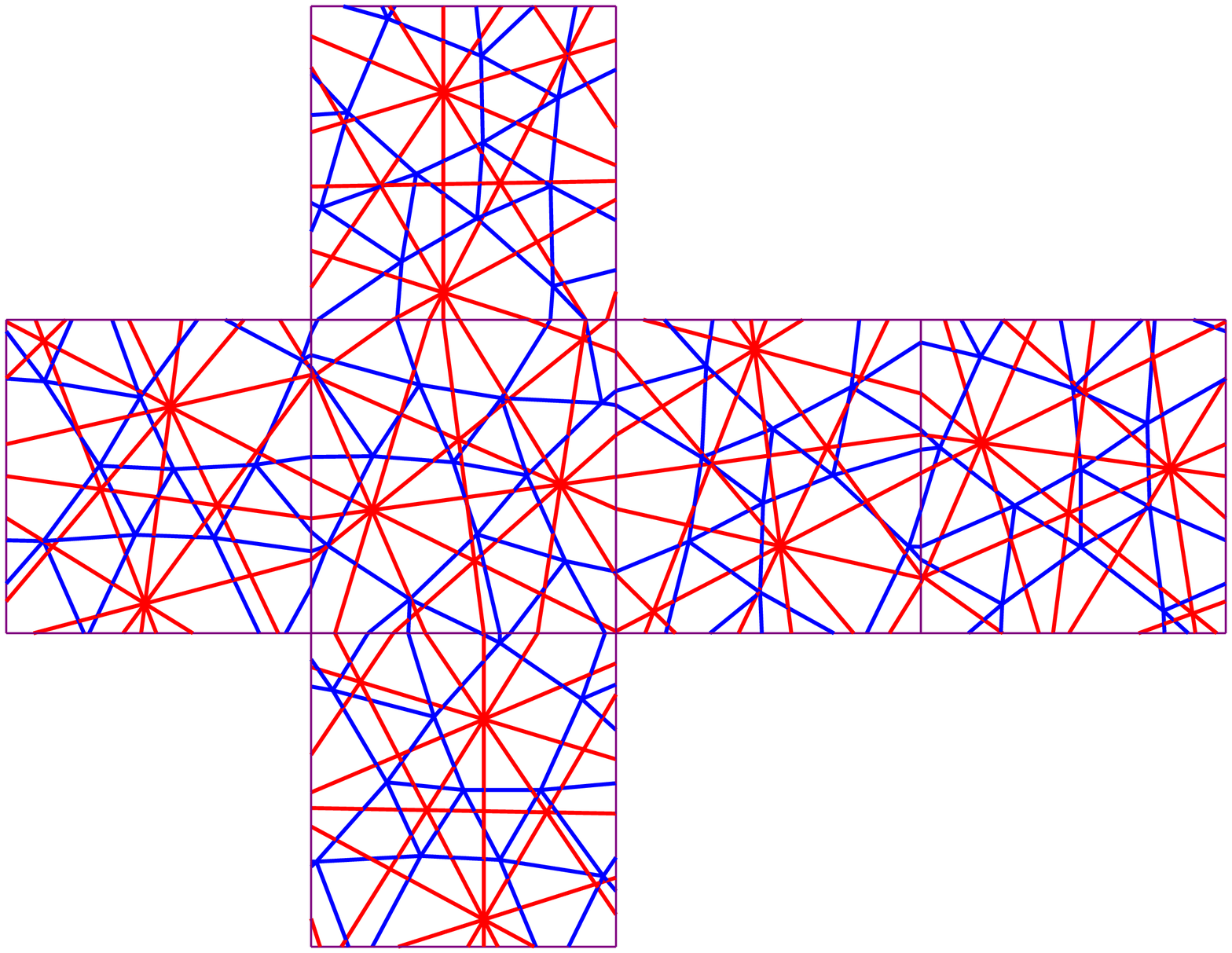,width=5cm,silent=}\\
      (d) & (e) & (f)\\
      \epsfig{figure=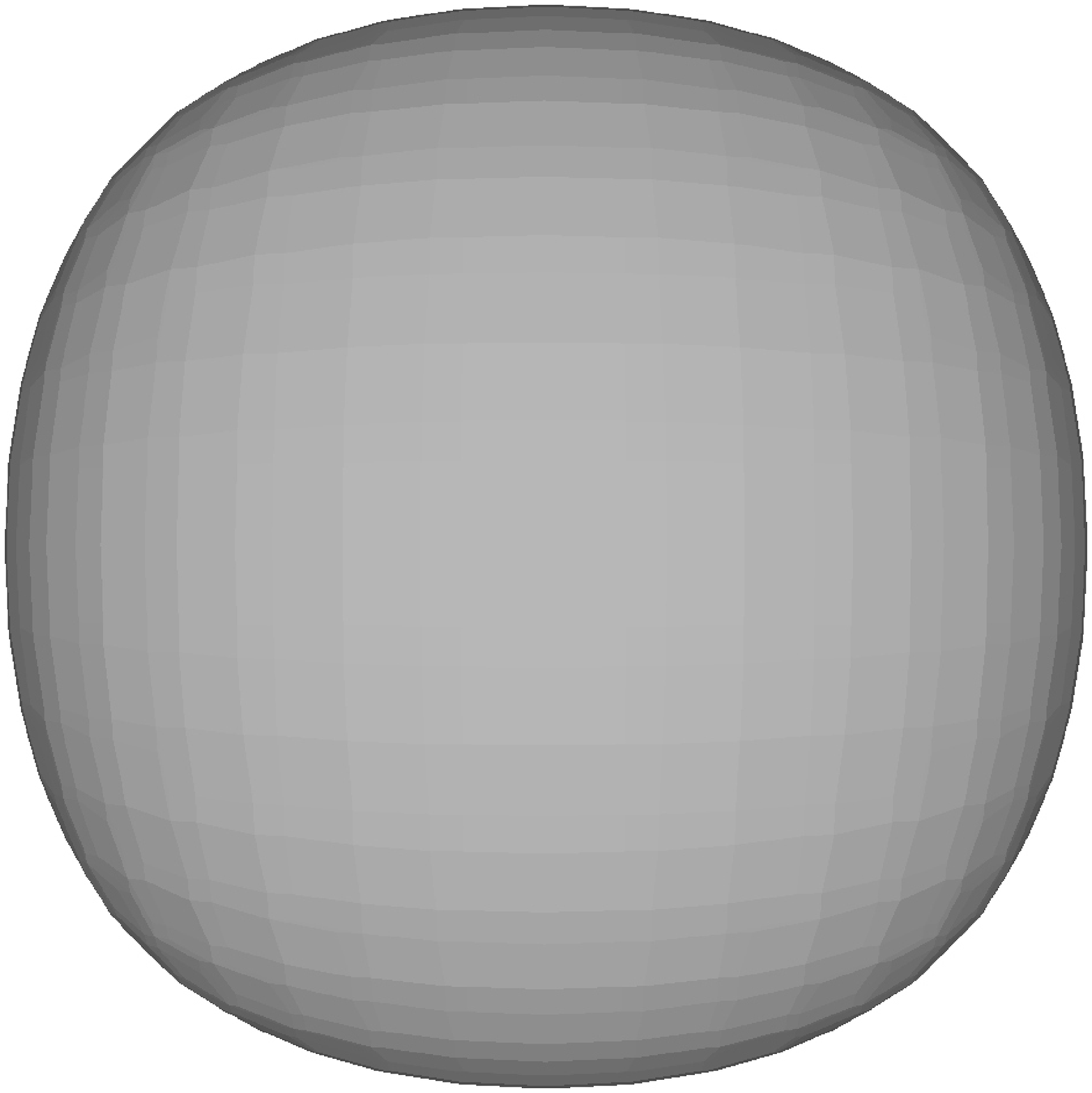,width=4cm,silent=} &
      \epsfig{figure=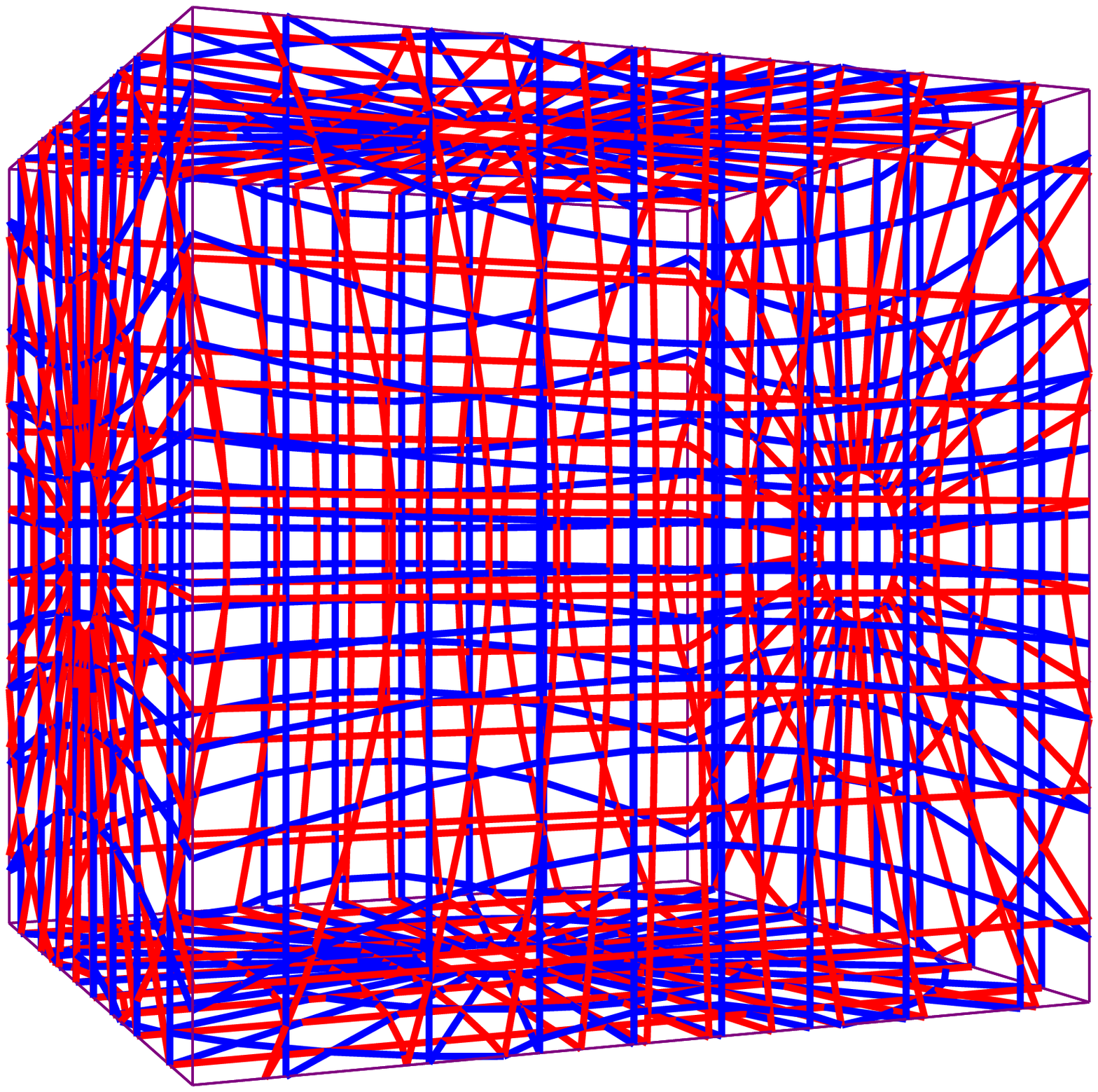,width=4cm,silent=} &
      \epsfig{figure=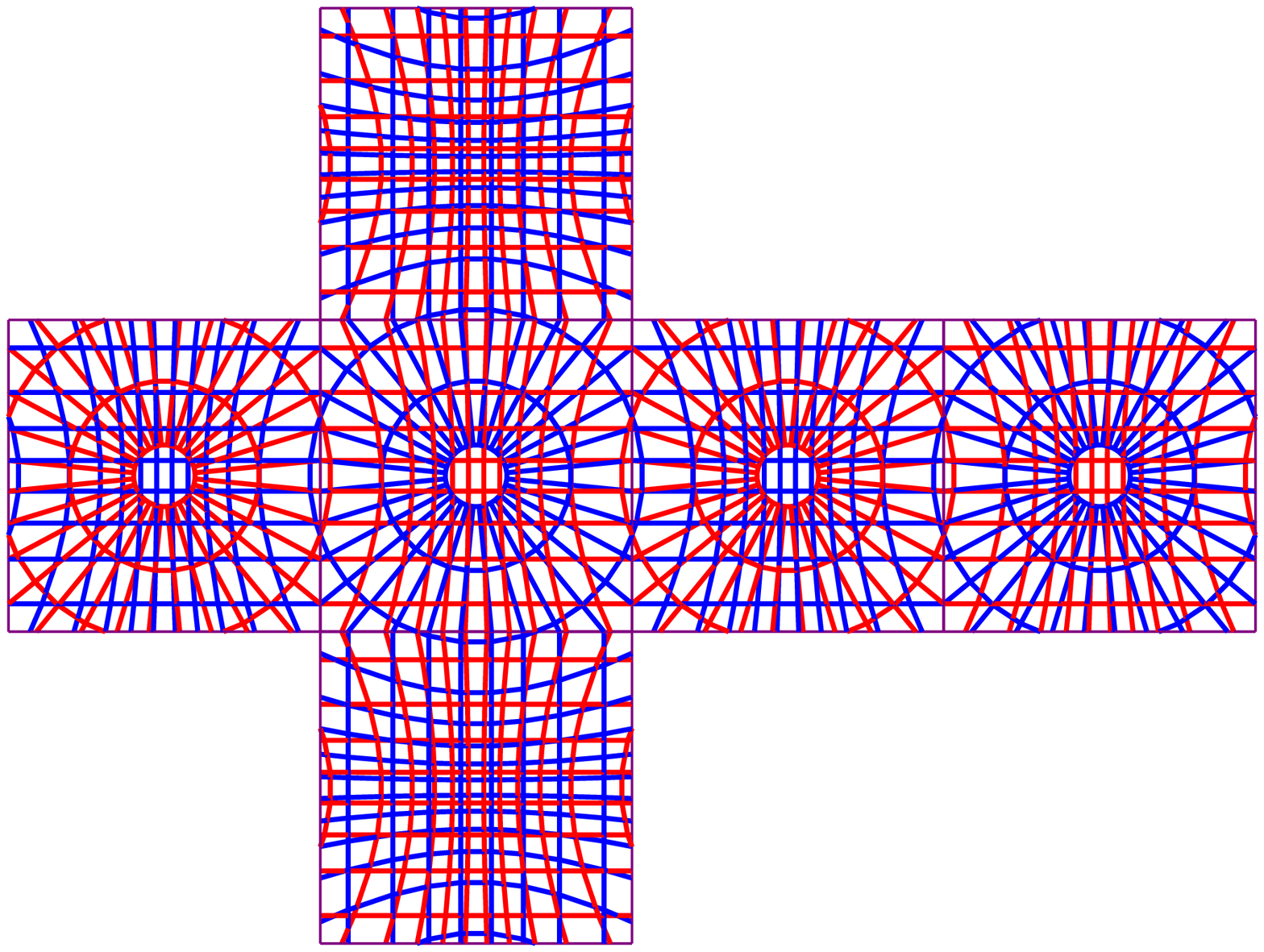,width=5cm,silent=}\\
      (g) & (h) & (i)\\
    \end{tabular}
  }
  \caption[Cubical Gaussian maps and Minkowski sums of polyhedra]%
           {\capStyle{(a) The Minkowski sum of two approximately
            orthogonal dioctagonal pyramids, (d) the Minkowski sum of
            a Pentagonal Hexecontahedron and a Truncated
            Icosidodecahedron, (g) the Minkowski sum of two approximately
            orthogonal ellipsoid-like polyhedra, (b,e,h) the \cgm{} of the  
             respective polytope, and (c,f,i) the \cgm{} unfolded.}}
  \label{fig:models2}
\end{figure*}

\begin{savequote}[10pc]
\sffamily
I was working on the proof of one of my poems all the morning, and
took out a comma. In the afternoon I put it back again.
\qauthor{Oscar Wilde}
\end{savequote}
\chapter{Exact Complexity of Minkowski Sums}
\label{chap:mink-sum-complexity}
\newlength{\radius}\setlength{\radius}{3.2cm}

\begin{figure}[!b]
  \centerline{
  \begin{tabular}{cccc}
  \epsfig{figure=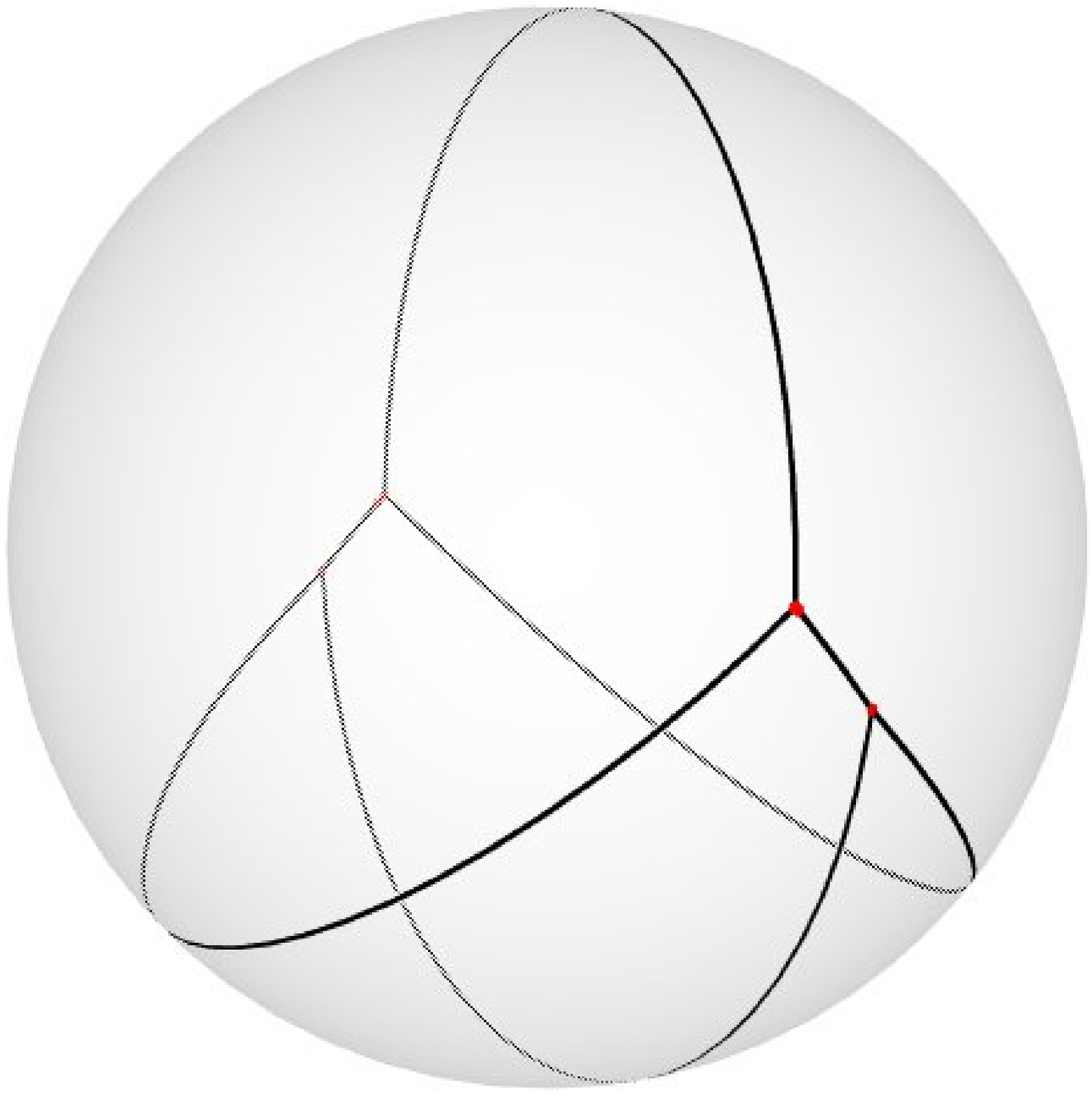,width=3.5cm,silent=} &
  \epsfig{figure=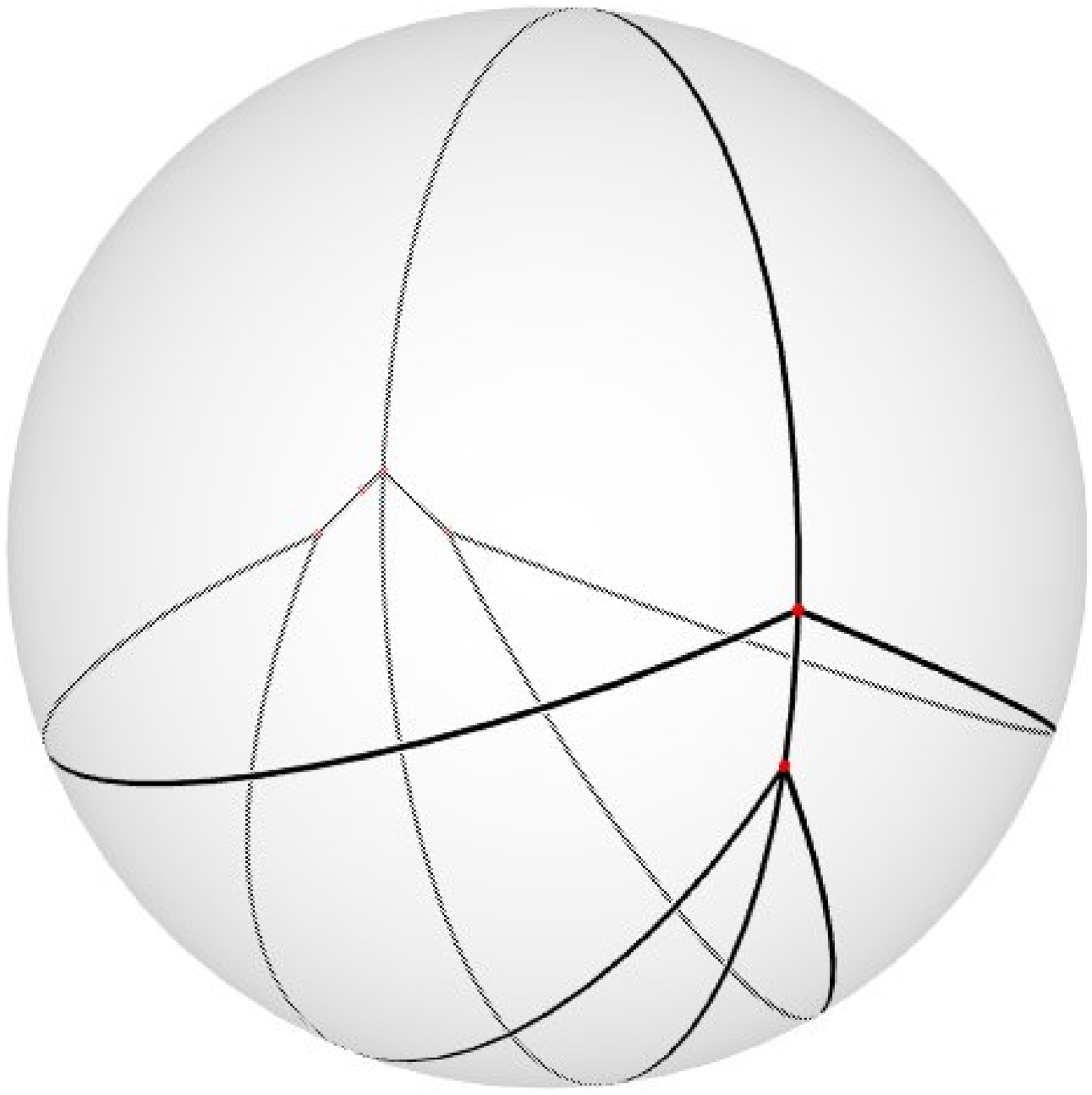,width=3.5cm,silent=} &
  \epsfig{figure=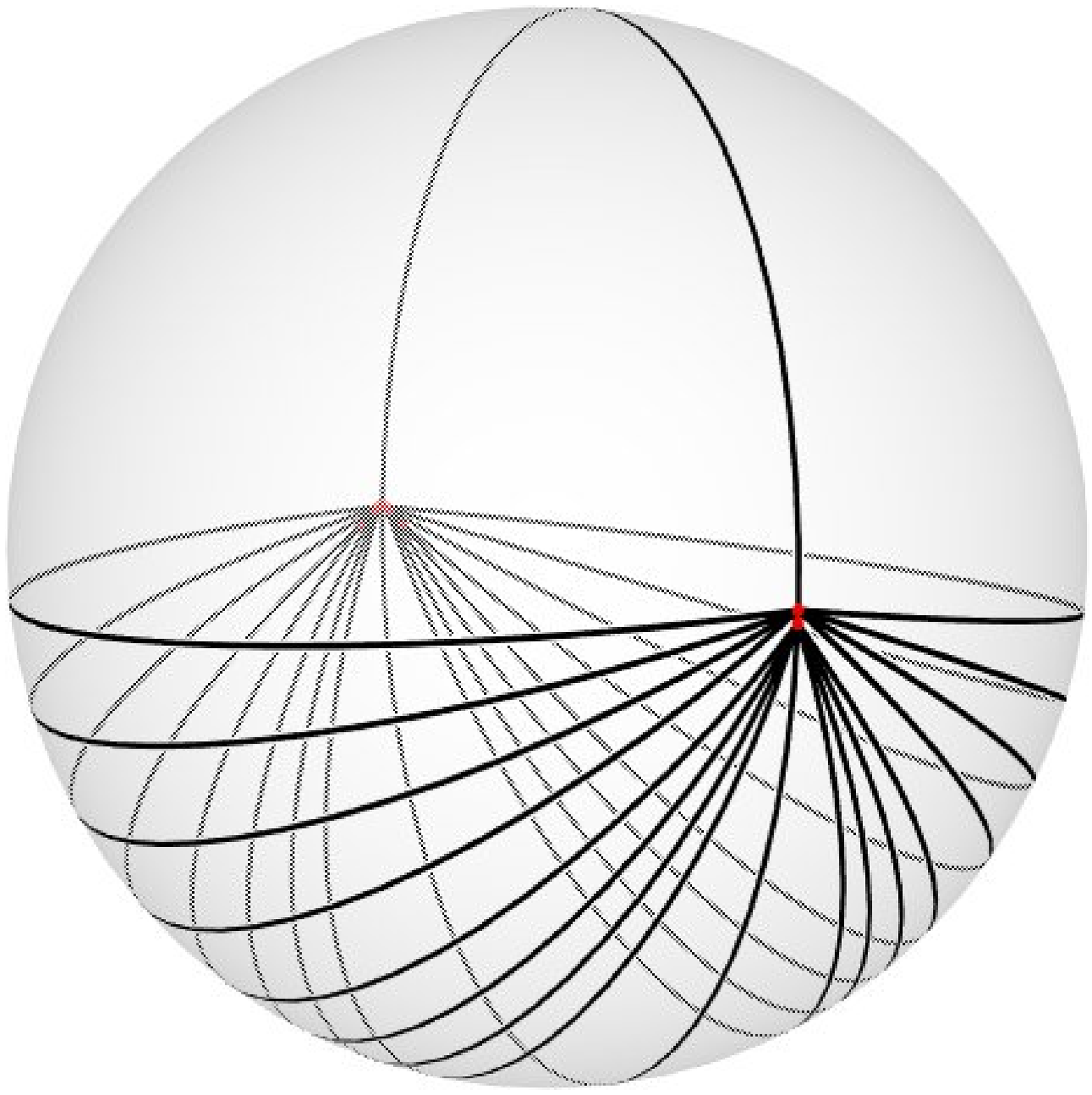,width=3.5cm,silent=} &
  \epsfig{figure=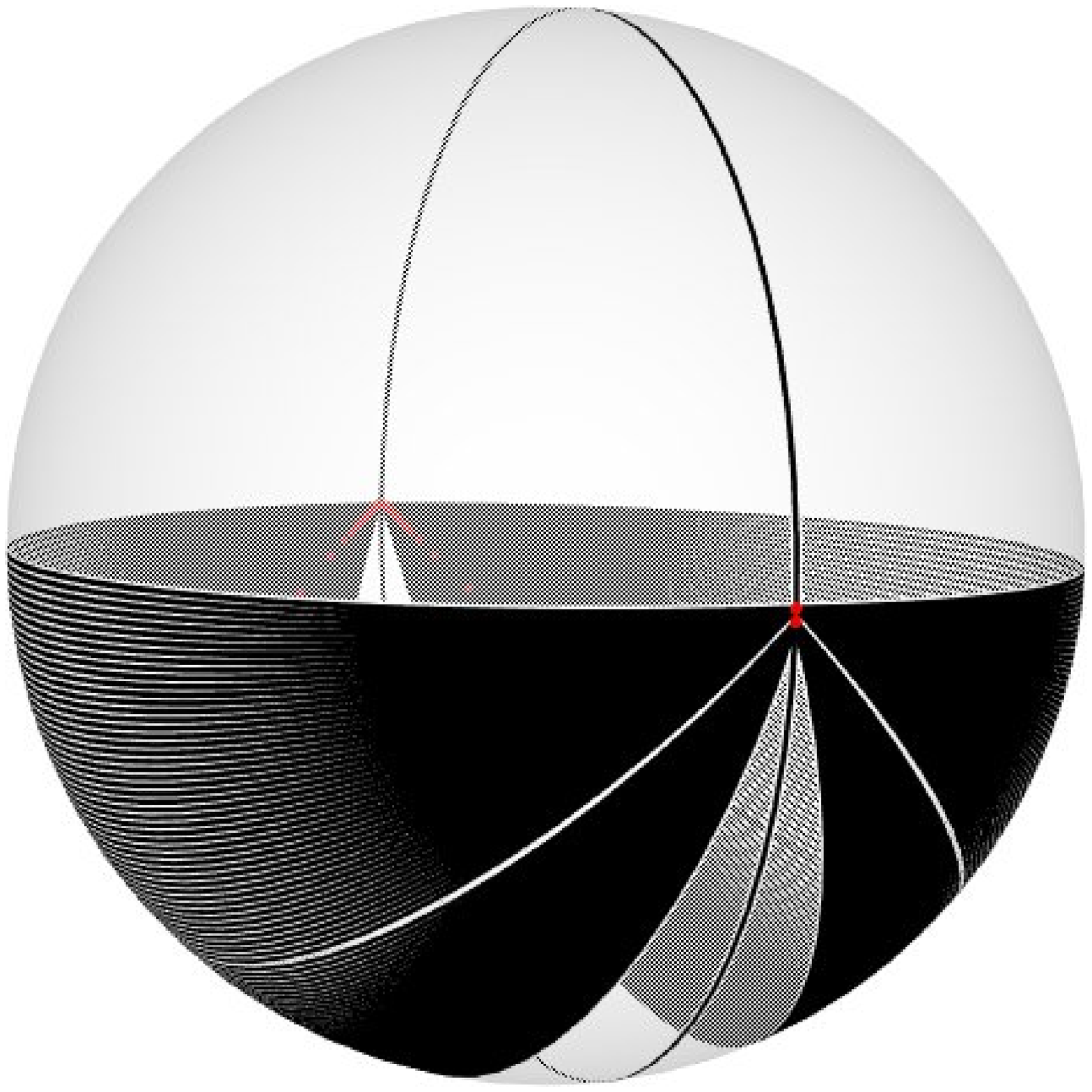,width=3.5cm,silent=}\\
  (a) & (b) & (c) & (d)
  \end{tabular}
  }
  \caption[Gaussian maps of summands of Minkowski sums with maximal complexities]
          {\capStyle{Gaussian maps of summands of Minkowski sums with maximal
          number of facets. The polytopes represented by the Gaussian maps (a),
          (b), (c), and (d) consist of 4, 5, 11, and 101 facets, respectively.}}
  \label{fig:summands-max-ms}  
\end{figure}

We present a tight bound on the exact maximum complexity of Minkowski 
sums of polytopes in $\mathbb{R}^3$. In particular, we prove that the
maximum number of facets of the \Index{Minkowski sum} of $k$ polytopes
with $m_1,m_2,\ldots,m_k$ facets respectively is bounded from above by
$\sum_{1 \leq i < j \leq k}(2m_i - 5)(2m_j - 5) + 
\sum_{1 \leq i \leq k}m_i + \binom{k}{2}$. Given $k$ positive integers
$m_1,m_2,\ldots,m_k$, we describe how to construct two polytopes with
corresponding number of facets, such that the number of facets of their
Minkowski sum is exactly $\sum_{1 \leq i < j \leq k}(2m_i - 5)(2m_j - 5) + 
\sum_{1 \leq i \leq k}m_i + \binom{k}{2}$. When $k = 2$, for example, the
expression above reduces to $4m_1m_2 - 9m_1 - 9m_2 + 26$.
Figure~\ref{fig:summands-max-ms} illustrates some polytopes that when
rotated properly and used as summand, the number of facets of the
resulting Minkowski sums is maximal.

Various methods to compute the Minkowski sum of two polyhedra in $\rrr$
have been proposed; see Section~\ref{sec:intro:ms-bg} for details
about these methods and about the combinatorial complexity of the sum.
Recall, that (i) a common approach for computing Minkowski sums of
non-convex polyhedra decomposes each polyhedron into convex pieces, and
computes pairwise Minkowski sums of the convex pieces, and (ii) all the
efficient methods are output sensitive. Thus, the exact maximum
complexity of the Minkowski sum structure has practical implications.

\newcounter{ms-comp:cntr}
One method to compute the Minkowski sum of two polytopes is to compute 
the \Index{convex hull} of the pairwise sum of the vertices of the two
polytopes. While being simple and easy to implement, the time
complexity of this method is $\Omega(mn)$ regardless of the size of
the resulting sum, which can be as low as $(m+n)$ (counting facets)
for degenerate cases.\footnote{It can be as low as $m (= n)$ in the
extremely-degenerate case of two similar polytopes with parallel facets.}
In Chapter~\ref{chap:mink-sums-construction} we describe several complete
implementations of output-sensitive methods for computing exact
Minkowski sums (beyond the naive method mentioned above), including
two methods that we introduce. These methods exploit efficient innovative
techniques in the area of exact geometric computing to minimize the
time it takes to ensure exact results. However, even with the use of
these techniques, the amortized time of a single arithmetic operation
is larger than the time it takes to carry out a single arithmetic
operation on native number types, such as floating point. Thus, the
constant that scales the dominant element in the expression of the
time complexity of these algorithms increases, which makes the
question this chapter attempts to answer, ``What is the exact maximum
complexity of Minkowski sums of polytopes in $\rrr$?'', even more relevant.

Gritzmann and Sturmfels~\cite{gs-mapca-93} formulated an upper
bound on the number of features $f^d_i$ of any given dimension $i$ of
the Minkowski sum of many polytopes in $d$ dimensions:
$f^d_i(P_1 \oplus P_2 \oplus \ldots \oplus P_k) \leq 2
\binom{j}{i}\sum\limits_{\substack{h=0}}^{d-i-1}\binom{j-i-1}{h}$ 
for $i = 0,1,\ldots d-1$, where $j$ denotes the number of non-parallel 
edges of $P_1,P_2,\ldots,P_k$. According to this expression, the number of 
facets $f^3_2$ of the Minkowski sum of two polytopes in $\mathbb{R}^3$ is
bounded from above by $j(j-1)$. Fukuda and Weibel~\cite{fw-fvmsa-05} 
obtained upper bounds on the number of edges and facets of the Minkowski 
sum of two polytopes in $\mathbb{R}^3$ in terms of the number of vertices 
of the summands:
$f^3_2(P_1 \oplus P_2) \leq f^3_0(P_1)f^3_0(P_2) + f^3_0(P_1) + f^3_0(P_2) - 6$.
They also studied the properties of Minkowski sums of perfectly centered
polytopes and their polars, and provided a tight bound on the number of
vertices of the sum of polytopes in any given dimension.

The main result presented in this chapter concerning two polytopes follows.
\begin{theorem}
\label{the:theorem}
Let $P_1,P_2,\ldots,P_k$ be a set of $k$ polytopes in $\mathbb{R}^3$,
such that the number of facets of $P_i$ is $m_i$ for $i = 1,2,\ldots,k$.
The number of facets of the Minkowski sum
$P_1 \oplus P_2 \oplus\ldots\oplus P_k$ cannot exceed
$\sum_{1 \leq i < j \leq k}(2m_i - 5)(2m_j - 5) +
\sum_{i=1}^k m_i + \binom{k}{2}$. 
This bound is tight. Namely, given $k$ integers
$m_1,m_2,\ldots,m_k$, such that $m_i \geq 4$ for $i = 1,2,\ldots,k$,
it is possible to construct $k$ polytopes in $\mathbb{R}^3$ with
corresponding number of facets, such that the number of facets of their
Minkowski sum is exactly the expression above.
\end{theorem}

The rest of this chapter is organized as follows. The upper bound on
the number of facets of Minkowski sums for the special case of two
polytopes in $\mathbb{R}^3$ is derived in Section~\ref{sec:mscy:upper}.
In Section~\ref{sec:mscy:lower} we describe how to construct two
polytopes, the number of facets of which is given, such that the
number of facets of their Minkowski sum is identical to the bound
derived in Section~\ref{sec:mscy:upper}. The bounds for the general
case of $k$ polytopes is proved in Section~\ref{sec:mscy:many}.
Information about the polyhedra models and the interactive program that 
computes their Minkowski sums and visualizes them, used to verify the
results and generate the figures in this chapter are provided in the Appendix.

\section{The Upper Bound for $k=2$}
\label{sec:mscy:upper}
The \Index{overlay} (see Section~\ref{ssec:aos:facilities:overlay} for the
exact definition) of the Gaussian maps\index{Gaussian map} (see
Section~\ref{sec:mscn:gauss_map} for the exact definition) of two
polytopes $P$ and $Q$ respectively is the Gaussian map of the
Minkowski sum of $P$ and $Q$; see Section~\ref{ssec:mscn:sgm:mink_sum}
for a detailed explanation.

The number of facets of the Minkowski sum $M$ of two polytopes $P$ and
$Q$ with $m$ and $n$ facets respectively is equal to the number of
vertices of the Gaussian map $G(M)$ of $M$. A vertex in $G(M)$ is
either due to a vertex in the Gaussian map of $P$, or due to
a vertex in the Gaussian map of $Q$, or due to an intersection of 
exactly two edges, one of the Gaussian map of $P$ and the other of
the Gaussian map of $Q$. Thus, the number of facets of $M$ cannot
exceed $m + n + g(M)$, where $g(M)$ is the number of intersections of
edges of $G(P)$ with edges of $G(Q)$ in $G(M)$.\footnote{The number of
facets is strictly equal to the given  expression, only when no
degeneracies occur.}

\begin{observation}
\label{obser:max-gauss}
The maximum exact number of edges in a Gaussian map $G(P)$ of a 
polytope $P$ with $m$ facets is $3m - 6$. The maximum exact number of
faces is $2m - 4$. Both maxima occur at the same Gaussian maps.
\end{observation}
The above can be obtained by a simple application of Euler's formula
for planar graphs to the Gaussian map $G(P)$. It can be used to trivially
bound the exact maximum number of facets of the Minkowski sum of two
polytopes defined as $f(m,n) = \max\{f(P \oplus Q)\,|\,f(P) = m, f(Q) = n\}$,
where $f(P)$ is the number of facets of a polytope $P$.
First, we can use the bound on the number of edges to obtain: 
$f(m,n) \leq m + n + (3m-6) \cdot (3n-6) = 9mn - 17m - 17n + 36$. Better
yet, we can plug the bound on the number of dual faces, which is the number 
of primal vertices, in the expression introduced by Fukuda and 
Weibel, see above, to obtain:
$f(m,n) \leq (2m - 4) \cdot (2n - 4) + (2m - 4) + (2n - 4) - 6 = 
4mn - 6m -6n + 2$.
Still, we can improve the bound even further, but first we need to bound
the number of faces in $G(M)$.

\begin{lemma}
\label{lemma:overlay-single}
Let $G_1$ and $G_2$ be two Gaussian maps of convex polytopes, and let $G$
be their overlay. Let $f_1$, $f_2$, and $f$ denote the number of faces of
$G_1$, $G_2$, and $G$, respectively. Then, $f \le f_1 \cdot f_2$.
\end{lemma}

Each face in the overlay is an intersection of a face of each
map. Since these faces are spherically convex (and smaller than
hemispheres), the intersection is also spherically convex (and thus
connected). This lemma is similar to the one where convex planar maps
replace the Gaussian maps. Nevertheless, we provide a formal proof directly
applied to the spherical case.
\begin{proof}
  We label each face in $G$ by a pair of indices of the originating
  overlaid faces in $G_1$ and $G_2$ respectively, and argue that no
  two faces in $G$ can have the same label. Assume to the contrary that
  there exist two faces $h_a$ and $h_b$ in $G$ that have the same label,
  say $\langle i,j \rangle$. That is, the faces $h_1^i$ in $G_1$ and
  $h_2^j$ in $G_2$ induce the two distinct faces $h_a$ and $h_b$. Pick
  two points $a \in h_a$ and $b \in h_b$. There must be a geodesic
  segment between $a$ and $b$ that is entirely contained in $h_1^i$ and
  also in $h_2^j$, as both maps are spherically convex. This implies
  that none of the edges in $G_1$ and $G_2$ split this geodesic segment,
  contradicting the fact that they reside in two different faces of $G$.
\end{proof}

We are ready to tackle the upper bound of
Theorem~\ref{the:theorem} for the special case $k=2$, that is, prove that
the number of facets of the Minkowski sum $P \oplus Q$ of two polytopes
$P$ and $Q$ with $m$ and $n$ facets respectively cannot exceed
$4mn - 9m - 9n + 26$; see Page~\pageref{the:theorem}.
\begin{proof}
Let $v_1, e_1, f_1$ and $v_2, e_2, f_2$ denote the number of vertices,
edges, and faces of $G(P)$ and $G(Q)$, respectively. 
The number of vertices, edges, and faces of $G(M)$ is denoted as $v$, $e$,
and $f$, respectively. Assume that $P$ and $Q$ are two polytopes, such that the
number of facets of their Minkowski sum is maximal.
Recall that the number of facets of a polytope is equal to the number of
vertices of its Gaussian map. Thus, we have $v_1 = m$, $v_2 = n$,
and $v = f(m,n)$.
First, we need to show that vertices of $G(P)$, vertices
of $G(Q)$, and intersections between edges of $G(P)$ and edges of
$G(Q)$ do not coincide. Assume to the contrary that some do. Then, one
of the polytopes $P$ or $Q$ or both can be slightly rotated to escape
this degeneracy, but this would increase the number of vertices
$v = f(m,n)$, contradicting the fact that $f(m,n)$ is maximal.
Therefore, the number of vertices $v$ is exactly equal to
$v_1 + v_2 + v_x$, where $v_x$ denotes the number of intersections of
edges of $G(P)$ and edges of $G(Q)$ in $G(M)$.
Counting the degrees of all vertices in $G(M)$ implies that
$2e_1 + 2e_2 + 4v_x = 2e$. Using Euler's formula, we get
$e_1 + e_2 + 2v_x = f + v_1 + v_2 + v_x - 2$. Applying
Lemma~\ref{lemma:overlay-single}, we can bound $v_x$ from above
$v_x \leq f_1 f_2 + v_1 + v_2 - 2 - e_1 - e_2$.

Observation~\ref{obser:max-gauss} sets an upper bound on the number of
edges $e_1$. Thus, $e_1$ can be expressed in terms of $\ell_1$, a
non-negative integer, as follows: $e_1 = 3v_1 - 6 - \ell_1$. Applying
Euler's formula, the number of facets can be expressed in terms of
$\ell_1$ as well: $f_1 = e_1 + 2 - v_1 = 2v_1 - 4 - \ell_1$.
Similarly, we have $e_2 = 3v_2 - 6 - \ell_2$ and
$f_2 = 2v_2 - 4 - \ell_2$ for some non-negative integer $\ell_2$.
\begin{align}
v_x & \leq (2v_1 - 4 - \ell_1)(2v_2 - 4 - \ell_2) + v_1 + v_2 - 2 -
  (3v_1 - 6 - \ell_1) - (3v_2 - 6 - \ell_2)\notag\\
    & \leq 4v_1 v_2 - 10v_1 - 10v_2 + 26 + h(\ell_1,\ell_2)\label{vx}\ ,
\end{align}
where $h(\ell_1,\ell_2) = \ell_1 \ell_2 + 5 \ell_1 + 5 \ell_2 - 2 v_1 \ell_2 - 2 v_2 \ell_1$.

$G(P)$ consists of a single connected component. Therefore, the number
of edges $e_1$ must be at least $v_1 - 1$. This is used to obtain an
upper bound on $\ell_1$ as follows:
$v_1 - 1 \leq e_1 = 3 v_1 - 6 - \ell_1$, which implies
$\ell_1 \leq 2 v_1 - 5$, and similarly $\ell_2 \leq 2 v_2 - 5$. Thus, we have:
\begin{align*}
h(\ell_1,\ell_2) & = \ell_1 \ell_2 + 5 \ell_1 + 5 \ell_2 - 2 v_1 \ell_2 - 2 v_2 \ell_1\\
 & = \ell_1(\frac{\ell_2}{2} - (2v_2 - 5)) + \ell_2(\frac{\ell_1}{2} - (2v_1 - 5))
 \leq 0\ .
\end{align*}

From Equation \eqref{vx} we get that
$v_x \leq 4v_1 v_2 - 10v_1 - 10v_2 + 26$, and since
$f(m,n) = v_1 + v_2 + v_x$, we conclude that
$f(m,n) \leq 4v_1 v_2 - 9v_1 - 9v_2 + 26$. The maximum number of facets can
be reached when $h(\ell_1,\ell_2)$ vanishes. This occurs when
$\ell_1 = \ell_2 = 0$. That is, when the number of edges of $G(P)$ and
$G(Q)$ is maximal. This concludes the proof of the upper bound of
Theorem~\ref{the:theorem} for the special case $k=2$.
\end{proof}
\begin{corollary}
\label{col:maximal}
The maximum number of facets can be attained only when the number of edges
of each of $P$ and $Q$ is maximal for the given number of facets.
\end{corollary}
\section{The Lower Bound for $k=2$}
\label{sec:mscy:lower}
Given two integers $m \geq 4$ and $n \geq 4$, we describe how to
construct two polytopes in $\mathbb{R}^3$ with $m$ and $n$ facets
respectively, such that the number of facets of their Minkowski sum is
exactly $4mn - 9m - 9n + 26$, establishing the lower bound of
Theorem~\ref{the:theorem} for the special case $k=2$. More precisely,
given $i$, we describe how to construct a skeleton of a polytope $P_i$
with $i$ facets, $3i - 6$ edges, and $2i-4$ vertices, and prove that
the number of facets of the Minkowski sum of $P_m$ and $P_n$, properly
adjusted and oriented, is exactly $4mn - 9m - 9n + 26$. As in the previous
sections we mainly operate in the dual space of Gaussian maps. However,
the construction of the desired Gaussian maps described below is an
involved task, since not every \Index{arrangement} of arcs of great
circles embedded on the unit sphere, the faces of which are convex and the
edges of which are strictly less than $\pi$ long constitutes a valid
Gaussian map.

\begin{wrapfigure}[8]{r}{3.5cm}
  \vspace{-15pt}
  \epsfig{figure=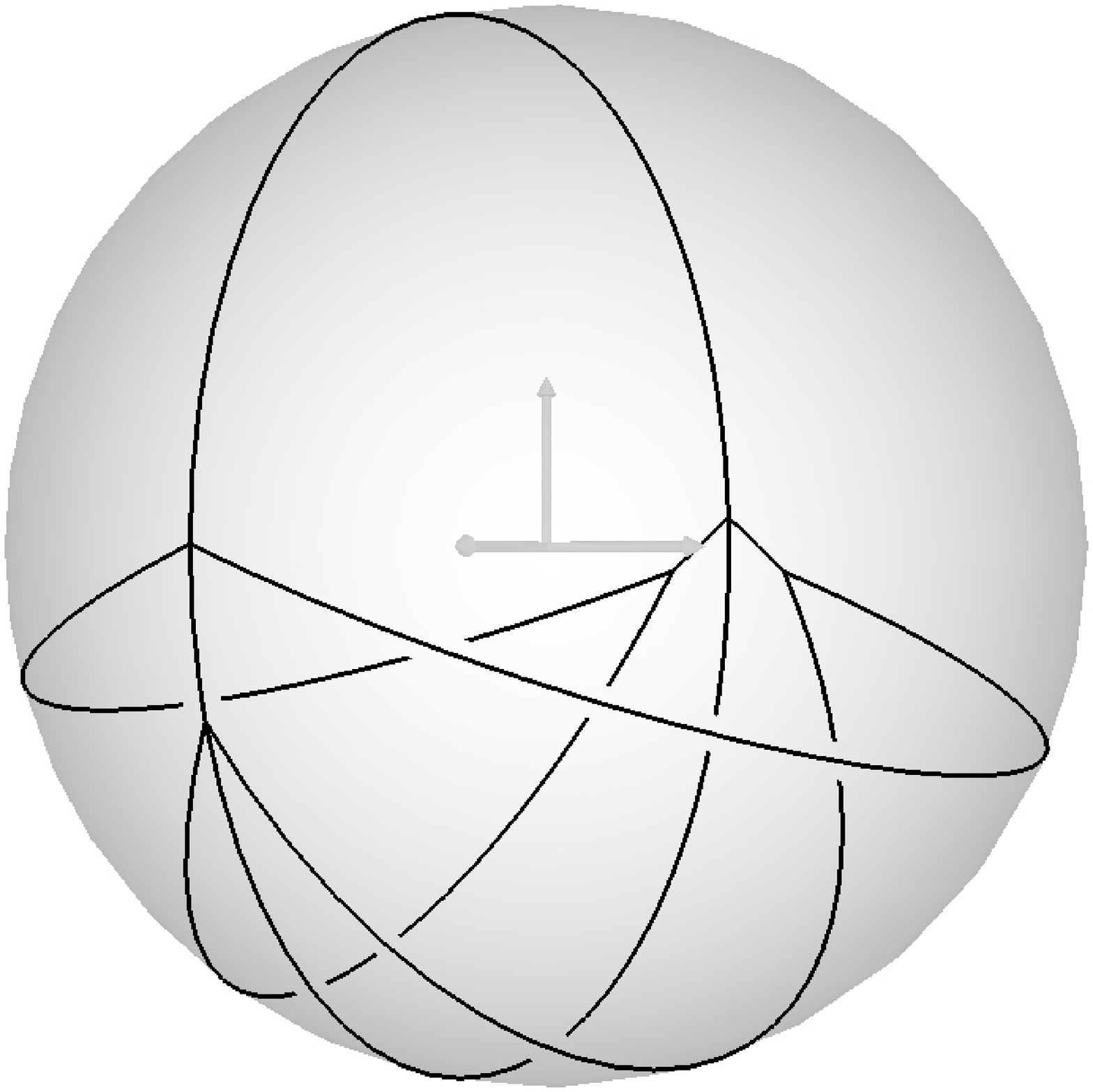,width=3.5cm,silent=}
  \vspace{-24pt}
  \vspace{-3.5cm}
  \pspicture[](0,0)(3.5,3.5)
  \psset{unit=1cm,linewidth=1pt}
  \rput{0}(2.31,1.55){\qdisk(0,0){2pt}\uput[ur]{0}(0,0){$u$}}
  \rput{0}(0.65,1.46){\qdisk(0,0){2pt}\uput[ul]{0}(0,0){$v$}}
  \rput{0}(0.71,0.89){\qdisk(0,0){2pt}\uput[-30]{0}(0,0){$w$}}
  \uput[0]{0}(1.49,2.20){$Y$}
  \pscircle(2.11,1.35){2pt}
  \pscircle(2.50,1.36){2pt}
  \endpspicture
\end{wrapfigure}
We defer the treatment of the special case $i=4$ to the sequel,
and start with the general case $i \geq 5$. The figure to the right
depicts the Gaussian map of $P_5$. We use the subscript letter $i$ in 
all notations $X_i$ to identify some object $X$ with the polytope $P_i$.
For example, we give the Gaussian map $G(P_i)$ of $P_i$ a shorter 
notation $G_i$, but in this paragraph we omit the subscript letter in 
all notations for clarity. First, we examine the structure of the 
Gaussian map $G$ of $P$ to better understand the structure of $P$. Let 
$V$ and $E$ denote the set of vertices and edges of $G$, respectively.
Recall that the number of vertices, edges, and faces of $G$ is $i$, $3i-6$,
and $2i-4$, respectively. The unit sphere, where $G$ is embedded on, is
divided by the plane $y = 0$ into two hemispheres
$H^- \subset \{(x,y,z)\,|\, y \leq 0\}$ and
$H^+ \subset \{(x,y,z)\,|\, y > 0\}$. 
Three vertices, namely $u$, $v$, and $w$, lie in the plane $x = 0$. 
$u$ is located very close to the pole $(0,0,-1)$. It is the only vertex 
(out of the $i$ vertices) that lies in $H^+$. $v$ is located exactly at
the opposite pole $(0,0,1)$, and $w$ lies in $H^-$ very close to $v$. 
None of the remaining $i-3$ vertices in $V \setminus \{u,v,w\}$ lie in 
the plane $x = 0$; they are all concentrated near the pole $(0,0,-1)$
and lie in $H^-$. The edge $\overline{uv}$, which is contained in the 
plane $x = 0$, is the only edge whose interior is entirely contained
in $H^+$. Every vertex in $V \setminus \{u,v,w\}$ is connected by two
edges to $v$ and $w$, respectively. These edges together with the edge
$\overline{uw}$, contained in the plane $x = 0$, form a set of $2i - 5$
edges, denoted as $E' = E \setminus \{\overline{uv}\}$. The length of
each of the edges in $E'$ is almost $\pi$, due to the near proximity of
$u$, $v$, and $w$ to the respective poles.

It is easy to verify that if the polytope $P$ is not degenerate,
namely, its \Index{affine hull} is 3-space, then any edge of $G(P)$ is
strictly less than $\pi$ long. Bearing this in mind, the main
difficulty in arriving at a tight-bound construction is forcing
sufficient edges of the set $E'$ of the Gaussian map of one
polytope to intersect sufficient edges of the set $E'$ of the
Gaussian map of the other polytope. The remaining pair of edges,
one from each Gaussian map, contributes a single intersection to the
total count. As shown below, this is the best one can do in terms of
intersections.

\begin{wrapfigure}[13]{l}{3.7cm}
  \vspace{-15pt}
  \epsfig{figure=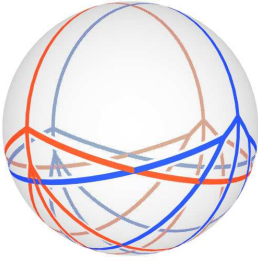,width=3.5cm,silent=}
  \caption[The overlay of $G_5$ and $G_5$ rotated about the $Y$ axis]%
          {\capStyle{The overlay of $G_5$ and $G'_5$, where $G'_5$ is
           $G_5$ rotated $90^{\circ}$ about the $Y$ axis.}}
  \label{fig:lb-5-5}
  \vspace{-10pt}
\end{wrapfigure}
The number of facets of the Minkowski sum of $P_m$ and $P_n$ is
maximal, when the number of vertices in the overlay of $G_m$ and $G_n$
is maximal. This occurs, for example, when one of them is rotated
$90^{\circ}$ about the $Y$ axis, as depicted on the left for the case of
$m = n = 5$. In this configuration, all the $2m - 5$ edges in $E'_m$
intersect all the $2n - 5$ edges in $E'_n$. All intersections occur in
$H^-$. In addition, the edge $\overline{uv}_m$ intersects the edge
$\overline{uv}_n$. The intersection point lies in $H^+$ exactly at the
pole $(0,1,0)$. Counting all these intersections results with  
$(2m - 5)(2n - 5) + 1 = 4mn - 10m -10n +26$. Adding the original
vertices of $G(P_m)$ and $G(P_n)$, yields the desired result.

Next, we explain how $P_i, i \geq 5$ is constructed to match the
description of $G_i$ above. The construction of $P_i$ is guided by a
cylinder.
All the vertices of $P_i$ lie on the boundary of a cylinder the
axis of which coincides with the $Z$ axis.
We start with the case $i=5$, and show how to generalize the
construction for $i > 5$. The special case $i = 4$ is explained last.

\subsection{Constructing $P_5$}
\label{ssec:mscy:lower:p5}
Figure~\ref{fig:lb-5} shows various views of $P_5$. Recall that $P_5$
has 6 vertices, denoted as $v_0$, $v_1,\ldots,v_5$, and 9 edges. We
omit the subscript digit $5$ in all the notations through the rest of
this subsection for clarity. Let $\overline{v_1 v_2 \ldots v_n}$ denote
the face defined by the sequence of vertices $v_1,v_2,\ldots,v_n$ on the
face boundary. The projection of all vertices onto the plane $z = 0$ lie
on the unit circle. As a matter of fact, the entire face
$f^v = \overline{v_0 v_1 v_2 v_3}$ lies in the plane $z=0$. It is mapped
under $G$ to the vertex $v = G(f^v)$. Similarly, the faces
$f^u = \overline{v_5 v_4 v_2 v_1}$ and
$f^w = \overline{v_3 v_4 v_5 v_0}$ are mapped under $G$ to the
vertices $u = G(f^u)$ and $w = G(f^w)$, respectively. Consider
the projection of the vertices onto the plane $z = 0$ best seen in
Figure~\ref{fig:lb-5}(b). Once the projection $v'_5$ of $v_5$ is
determined as explained below, $v_0$ is placed exactly on the bisector
of $\angle{v'_5 o v_1}$. The vertices $v_4$, $v_3$, and $v_2$
are the reflection of the vertices $v_5$, $v_0$, and $v_1$
respectively through the plane $x = 0$.

Two parameters govern the exact placement of $v_5$ (and $v_4$). One is
the size of the exterior-dihedral angle at the edge
$\overline{v_0 v_3}$, denoted as  $\alpha$, that is, the length of the
geodesic-segment that is the mapping of the edge $\overline{vw}$ of
$G$. This angle is best seen in Figure~\ref{fig:lb-5}(c). Notice,
that the $Z$ axis is scaled up for clarity, and the angle in practice
is much smaller. The other parameter is the size of the angle
$\beta = \angle{v'_4 o v'_5}$, where $v'_4$ and $v'_5$ are the
projections of $v_4$ and $v_5$ respectively onto the plane $z = 0$.
This is best seen in Figures~\ref{fig:lb-5}(b) and (e). Given $m$
and $n$, these angles for each of $P_m$ and $P_n$ depend on both $m$
and $n$. For large values of $m$ and $n$ the values of $\alpha$ and
$\beta$ should be small. For example, setting
$\alpha = \beta = 10^{\circ}$ is sufficient for the case $m = n = 5$
depicted in Figure~\ref{fig:lb-5-5}. The actual setting is discussed
below after the description of the general case $i > 5$.

\begin{figure*}[t]
  \centerline{
  \begin{tabular}{c|c|c}
  \multirow{1}*[115pt]{\epsfig{figure=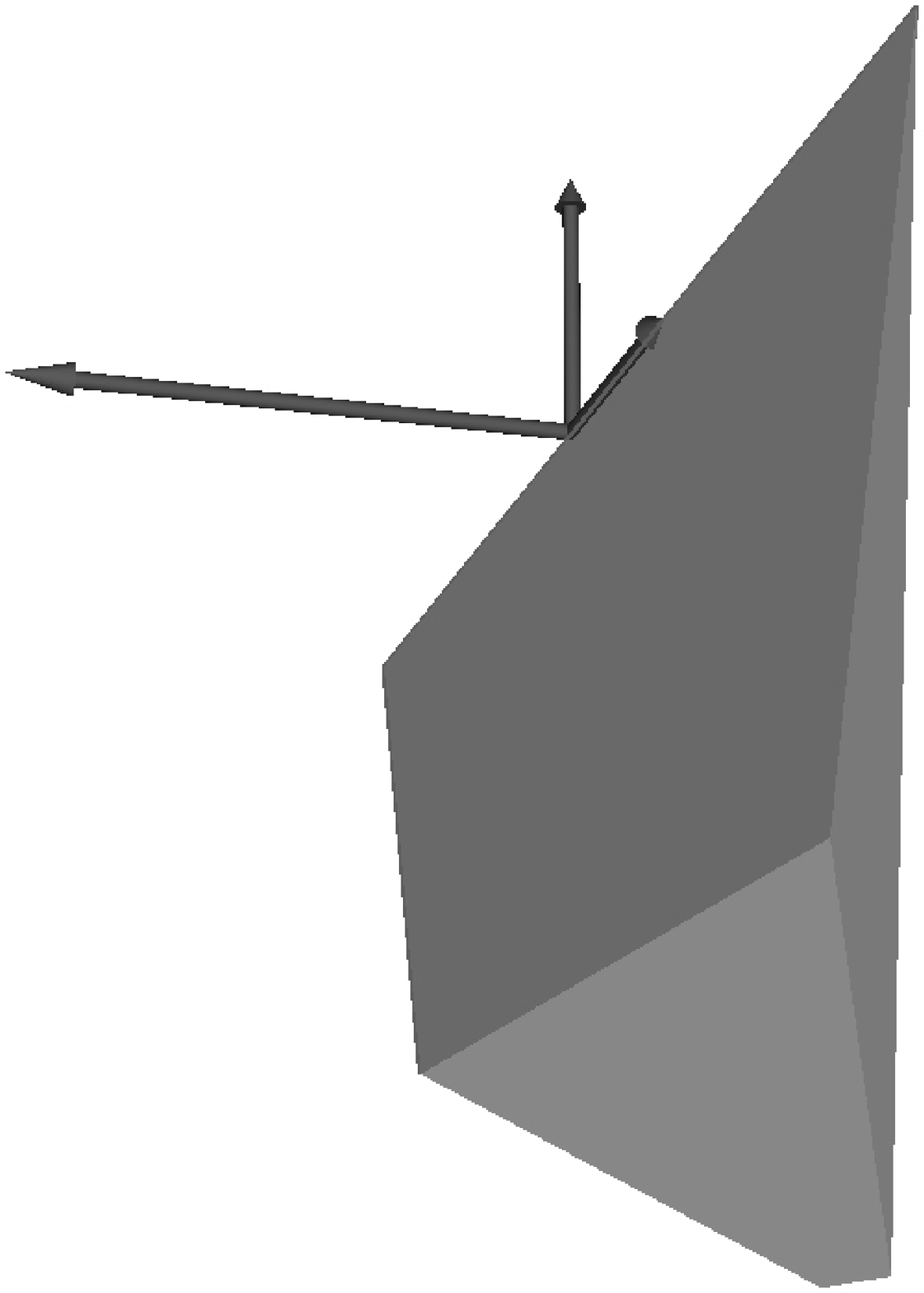,width=4cm,silent=}} &
  \pspicture[](-2.5,-2.5)(2.5,2.5)
  \psset{unit=1cm,linewidth=1pt,framesep=1.5pt}
  \pscircle[linewidth=0.5pt](0,0){2}
  \rput{320}(0,0){\cnode*(2,0){3pt}{0}\uput[0]{-320}(2,0){$v_0$}}
  \cnode*(2,0){3pt}{1}\uput[0]{0}(2,0){$v_1$}
  \cnode*(-2,0){3pt}{2}\uput[180]{0}(-2,0){$v_2$}
  \rput{220}(0,0){\cnode*(2,0){3pt}{3}\uput[0]{-220}(2,0){$v_3$}}
  \rput{260}(0,0){\cnode*(2,0){3pt}{4}\uput[0]{-260}(2,0){$v_4$}}
  \rput{280}(0,0){\cnode*(2,0){3pt}{5}\uput[0]{-280}(2,0){$v_5$}}
  \ncline{0}{1}
  \ncline{1}{2}
  \ncline{2}{3}
  \ncline{3}{4}
  \ncline{4}{5}
  \ncline{5}{0}
  \ncline{0}{3}
  \ncline[linestyle=dotted]{2}{4}
  \ncline[linestyle=dotted]{1}{5}
  \my_axis{0}{1}{90}{1}{210}{1.4}
  \pnode(0,0){c}\uput[ur]{0}(0,0){$o$}
  \endpspicture &
  \pspicture[](-2.5,-2.5)(2.5,2.5)
  \psset{unit=1cm,linewidth=1pt,framesep=1.5pt}
  \psline[linewidth=0.5pt](-2,2)(2,2)
  \psline[linewidth=0.5pt](-2,-2)(2,-2)
  \rput{0}(2,0){
  \psellipse[linewidth=0.5pt](0,0)(0.3,2)}
  \cnode*(0,0){3pt}{c}\uput[ur]{0}(0,0){$o$}
  \cnode*(0,-1.285){3pt}{a}
  \cnode*(0.5,-2){3pt}{b}
  \ncline{c}{a}
  \ncline{a}{b}
  \ncline{b}{c}
  \rput{0}(0,-1.285){
    \psarc[linewidth=0.5pt,arcsepB=2pt]{->}{0.65}{-90}{-50}
    \rput{-75}(0,0){\uput[0]{75}(0.2,0){$\alpha$}}
  }
  \psline[linewidth=0.5pt](0,-1.285)(0,-2)
  \my_axis{210}{0.7}{90}{1}{180}{2}
  \endpspicture\\
  {(a)} & {(b)} & {(c)}\\
  \hline
  \multirow{1}*[115pt]{\epsfig{figure=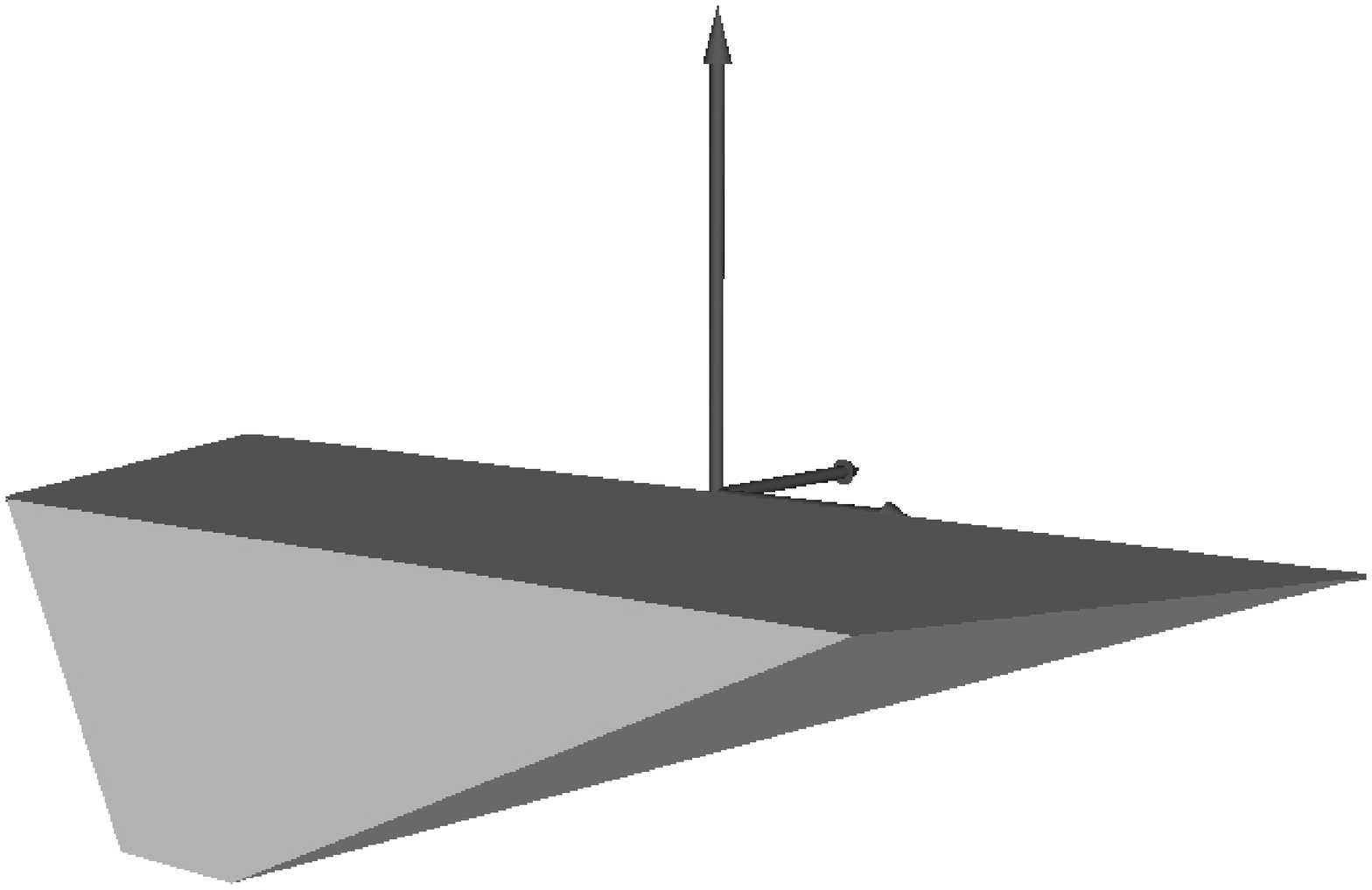,width=4cm,silent=}} &
  \pspicture[](-2.5,-2.5)(2.5,2.5)
  \psset{unit=1cm,linewidth=1pt,framesep=1.5pt}
  \psline[linewidth=0.5pt](2,-2)(2,2)
  \psline[linewidth=0.5pt](-2,-2)(-2,2)
  \rput{0}(0,-2){
  \psellipse[linewidth=0.5pt](0,0)(2,0.3)}
  \pnode(0,0){c}\uput[ul]{0}(0,0){$o$}
  \cnode*(1.532,0){3pt}{0}\uput[u]{0}(1.532,0){$v_0$}
  \cnode*(2,0){3pt}{1}\uput[0]{0}(2,0){$v_1$}
  \cnode*(-2,0){3pt}{2}\uput[180]{0}(-2,0){$v_2$}
  \cnode*(-1.532,0){3pt}{3}\uput[u]{0}(-1.532,0){$v_3$}
  \cnode*(-0.347,-0.5){3pt}{4}\uput[d]{0}(-0.347,-0.5){$v_4$}
  \cnode*(0.347,-0.5){3pt}{5}\uput[d]{0}(0.347,-0.5){$v_5$}
  \ncline{1}{2}
  \ncline{2}{4}
  \ncline{4}{5}
  \ncline{5}{1}
  \ncline{3}{4}
  \ncline{0}{5}
  \my_axis{0}{1}{30}{0.7}{90}{1.9}
  \endpspicture &
  \pspicture[](-2.5,-2.5)(2.5,2.5)
  \psset{unit=1cm,linewidth=1pt,framesep=1.5pt}
  \pscircle[linewidth=0.5pt](0,0){2}
  \rput{320}(0,0){\cnode*(2,0){3pt}{0}\uput[0]{-320}(2,0){$v_3$}}
  \cnode*(2,0){3pt}{1}\uput[0]{0}(2,0){$v_2$}
  \cnode*(-2,0){3pt}{2}\uput[180]{0}(-2,0){$v_1$}
  \rput{220}(0,0){\cnode*(2,0){3pt}{3}\uput[0]{-220}(2,0){$v_0$}}
  \rput{260}(0,0){\cnode*(2,0){3pt}{4}\uput[0]{-260}(2,0){$v_5$}}
  \rput{280}(0,0){\cnode*(2,0){3pt}{5}\uput[0]{-280}(2,0){$v_4$}}
  \ncline{0}{1}
  \ncline{1}{2}
  \ncline{2}{3}
  \ncline{3}{4}
  \ncline{4}{5}
  \ncline{5}{0}
  \ncline{2}{4}
  \ncline{1}{5}
  \ncline[linestyle=dotted]{0}{3}
  \pnode(0,0){c}\uput[dl]{0}(0,0){$o$}
  \my_axis{180}{1}{90}{1}{30}{1.4}
  \endpspicture\\
  {\sf (d)} & {\sf (e)} & {\sf (f)}\\
  \end{tabular}
  }
  \caption[Different views of $P_5$]%
          {\capStyle{Different views of $P_5$. (a) and (d) are
           perspective views, while (b), (c), (e), and (f) are
           orthogonal views. Notice that the $Z$ axis is scaled up for
           clarity.}}
  \label{fig:lb-5}
\end{figure*}

\pagebreak
\subsection{Constructing $P_i, i \geq 5$}
\label{ssec:mscy:lower:pi}
\begin{wrapfigure}[9]{r}{5.5cm}
  \vspace{-15pt}
  \pspicture[](-2.6,-2.5)(2.6,1.6)
  \psset{unit=1cm,linewidth=1pt,framesep=1.5pt}
  \pnode(0,0){c}\uput[ur]{0}(0,0){$o$}
  \psarc[linewidth=0.5pt](0,0){2}{180}{0}
  \rput{320}(0,0){\cnode*(2,0){3pt}{j0}\uput[0]{-320}(2,0){$v_{j_0}$}}
  \rput{350}(0,0){\cnode*[linecolor=white](2,0){3pt}{jx}\cnode(2,0){3pt}{jx}\uput[0]{-350}(2,0){$v_{j_1-1}$}}
  \cnode*(2,0){3pt}{j1}\uput[0]{0}(2,0){$v_{j_1}$}
  \cnode*(-2,0){3pt}{j2}\uput[180]{0}(-2,0){$v_{j_2}$}
  \rput{220}(0,0){\cnode*(2,0){3pt}{j3}\uput[0]{-220}(2,0){$v_{j_3}$}}
  \rput{260}(0,0){\cnode*(2,0){3pt}{j4}\uput[0]{-260}(2,0){$v_{j_4}$}}
  \rput{280}(0,0){\cnode*(2,0){3pt}{j5}\uput[0]{-280}(2,0){$v_{j_5}$}}
  \ncline{j1}{j2}
  \ncline{j4}{j5}
  \ncline{j0}{j3}
  \ncline[linestyle=dotted]{j2}{j4}
  \ncline[linestyle=dotted]{j1}{j5}
  \my_axis{0}{1}{90}{1}{210}{0.7}
  \endpspicture
  \vspace{-10pt}
\end{wrapfigure}
We construct a polytope, such that two facets are visible when looked at
from $z = \infty$, and  $i-2$ facets are visible when looked at from
$z = -\infty$. First, we place the projection of all vertices onto
the plane $z = 0$ along the unit circle, and denote the projection of
a vertex $v$ as $v'$. The projection of the vertices $v_{j_0}$, $v_{j_1}$, 
$v_{j_2}$, $v_{j_3}$, $v_{j_4}$, and $v_{j_5}$, where $j_0 = 0$,
$j_1 = \lfloor (i-2)/2 \rfloor$, $j_2 = \lfloor (i-2)/2 \rfloor+1$,
$j_3 = i - 2$, $j_4 = \lfloor (3i-7)/2 \rfloor$, and
$j_5 = \lfloor (3i-7)/2 \rfloor+1$, are placed at the same locations as
those of the corresponding vertices of $P_5$, as depicted on the
right. The projection of the remaining vertices are placed on the arcs
$\widehat{v'_{j_5},v_{j_0}}$, $\widehat{v_{j_0},v_{j_1}}$,
$\widehat{v_{j_2},v_{j_3}}$, and $\widehat{v_{j_3},v'_{j_4}}$ in cyclic
order.

The angle $\gamma = \angle{v_{j_0} o v_{j_1-1}}$ is another
parameter that governs the final configuration of $P_i$. Once the
placement of the projection of $v_{j_1-1}$ is determined, the
projections of the vertices $v_{j_0+1},v_{j_0+2},\ldots,v_{j_1-2}$
are arbitrarily spread along the open arc $\widehat{v_{j_0},v_{j_1-1}}$. 
The vertex placement along the arc $\widehat{v'_{j_5},v_{j_0}}$ must be
a symmetric reflection of the vertex placement along the arc
$\widehat{v_{j_0},v_{j_1}}$. This guarantees that all the
quadrilateral facets are planar. Similarly, the vertex placement along
the arc $\widehat{v_{j_2},v_{j_3}}$ is a symmetric reflection of the
vertex placement along the arc $\widehat{v_{j_3},v'_{j_4}}$. For large
values of $m$ and $n$ the angle $\gamma$ should be small as explained
below, implying that the projection of the vertices are concentrated
near $v_{j_0}$ and $v_{j_3}$, (which lie on the bisectors of
$\angle{v_{j_1} o v'_{j_5}}$ and $\angle{v_{j_2} o v'_{j_4}}$,
respectively). Figure~\ref{fig:lb-i} depicts the cases
$i = 10$, and $i = 11$. In these examples we force a regular placement,
which is sufficient in many cases. As in the case of $i=5$, the face 
$f^v_i = \overline{v_0,v_1,\ldots,v_{i-2}}$, represented by the vertex
$v_i$ of $G_i$, lies in the plane $z = 0$. The exterior-dihedral angle
$\alpha$ at the edge $\overline{v_{j_0}v_{j_3}}$ is made small, so
that the vertex $w_i$ of $G_i$ representing the adjacent face 
$f^w_i = \overline{v_{j_3}, v_{j_3 + 1},\ldots,v_{2i-5},v_0}$, is kept
in close proximity to $v_i$.

\begin{figure*}[t]
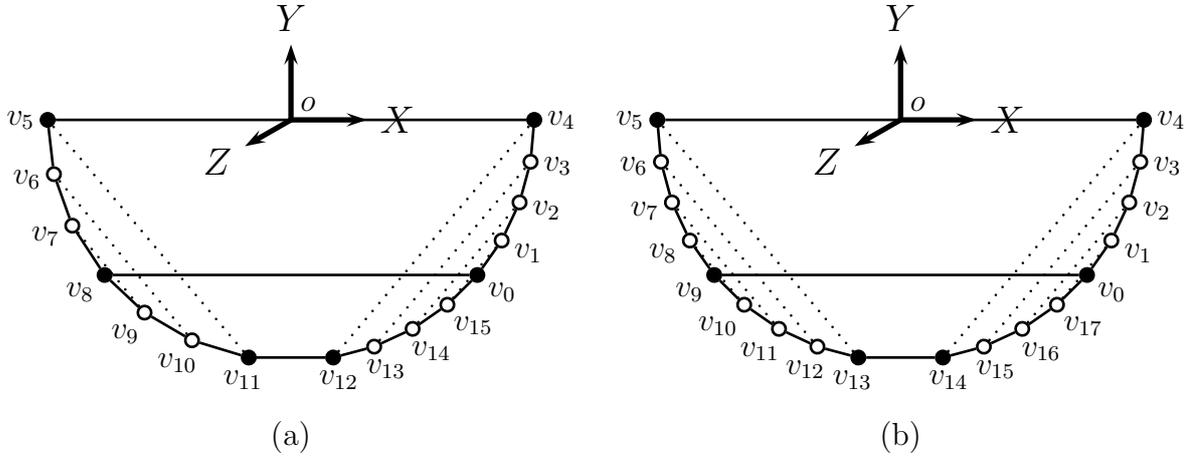

  \centerline{
\begin{tabular}{cc}
\pspicture[](-3.8,-3.8)(3.8,1.5)
\psset{unit=1cm,linewidth=1pt,framesep=1.5pt}
\pnode(0,0){c}\uput[ur]{0}(0,0){$o$}
\rput{320}(0,0){\cnode*(\radius,0){3pt}{0}\uput[0]{-320}(\radius,0){$v_0$}}
\rput{330}(0,0){\cnode(\radius,0){3pt}{1}\uput[0]{-330}(\radius,0){$v_1$}}
\rput{340}(0,0){\cnode(\radius,0){3pt}{2}\uput[0]{-340}(\radius,0){$v_2$}}
\rput{350}(0,0){\cnode(\radius,0){3pt}{3}\uput[0]{-350}(\radius,0){$v_3$}}
\cnode*(\radius,0){3pt}{4}\uput[0]{0}(\radius,0){$v_4$}
\rput{180}(0,0){\cnode*(\radius,0){3pt}{5}\uput[0]{-180}(\radius,0){$v_5$}}
\rput{193}(0,0){\cnode(\radius,0){3pt}{6}\uput[0]{-193}(\radius,0){$v_6$}}
\rput{206}(0,0){\cnode(\radius,0){3pt}{7}\uput[0]{-206}(\radius,0){$v_7$}}
\rput{220}(0,0){\cnode*(\radius,0){3pt}{8}\uput[0]{-220}(\radius,0){$v_8$}}
\rput{233}(0,0){\cnode(\radius,0){3pt}{9}\uput[0]{-233}(\radius,0){$v_9$}}
\rput{246}(0,0){\cnode(\radius,0){3pt}{10}\uput[0]{-246}(\radius,0){$v_{10}$}}
\rput{260}(0,0){\cnode*(\radius,0){3pt}{11}\uput[0]{-260}(\radius,0){$v_{11}$}}
\rput{280}(0,0){\cnode*(\radius,0){3pt}{12}\uput[0]{-280}(\radius,0){$v_{12}$}}
\rput{290}(0,0){\cnode(\radius,0){3pt}{13}\uput[0]{-290}(\radius,0){$v_{13}$}}
\rput{300}(0,0){\cnode(\radius,0){3pt}{14}\uput[0]{-300}(\radius,0){$v_{14}$}}
\rput{310}(0,0){\cnode(\radius,0){3pt}{15}\uput[0]{-310}(\radius,0){$v_{15}$}}
\ncline{0}{1}
\ncline{1}{2}
\ncline{2}{3}
\ncline{3}{4}
\ncline{4}{5}
\ncline{5}{6}
\ncline{6}{7}
\ncline{7}{8}
\ncline{8}{9}
\ncline{9}{10}
\ncline{10}{11}
\ncline{11}{12}
\ncline{12}{13}
\ncline{13}{14}
\ncline{14}{15}
\ncline{15}{0}
\ncline{0}{8}
\ncline[linestyle=dotted]{4}{12}
\ncline[linestyle=dotted]{3}{13}
\ncline[linestyle=dotted]{2}{14}
\ncline[linestyle=dotted]{1}{15}
\ncline[linestyle=dotted]{5}{11}
\ncline[linestyle=dotted]{6}{10}
\ncline[linestyle=dotted]{7}{9}
\my_axis{0}{1}{90}{1}{210}{0.7}
\endpspicture &
\pspicture[](-3.8,-3.8)(3.8,1.5)
\psset{unit=1cm,linewidth=1pt,framesep=1.5pt}
\pnode(0,0){c}\uput[ur]{0}(0,0){$o$}
\rput{320}(0,0){\cnode*(\radius,0){3pt}{0}\uput[0]{-320}(\radius,0){$v_0$}}
\rput{330}(0,0){\cnode(\radius,0){3pt}{1}\uput[0]{-330}(\radius,0){$v_1$}}
\rput{340}(0,0){\cnode(\radius,0){3pt}{2}\uput[0]{-340}(\radius,0){$v_2$}}
\rput{350}(0,0){\cnode(\radius,0){3pt}{3}\uput[0]{-350}(\radius,0){$v_3$}}
\cnode*(\radius,0){3pt}{4}\uput[0]{0}(\radius,0){$v_4$}
\rput{180}(0,0){\cnode*(\radius,0){3pt}{5}\uput[0]{-180}(\radius,0){$v_5$}}
\rput{190}(0,0){\cnode(\radius,0){3pt}{6}\uput[0]{-190}(\radius,0){$v_6$}}
\rput{200}(0,0){\cnode(\radius,0){3pt}{7}\uput[0]{-200}(\radius,0){$v_7$}}
\rput{210}(0,0){\cnode(\radius,0){3pt}{8}\uput[0]{-210}(\radius,0){$v_8$}}
\rput{220}(0,0){\cnode*(\radius,0){3pt}{9}\uput[0]{-220}(\radius,0){$v_9$}}
\rput{230}(0,0){\cnode(\radius,0){3pt}{10}\uput[0]{-230}(\radius,0){$v_{10}$}}
\rput{240}(0,0){\cnode(\radius,0){3pt}{11}\uput[0]{-240}(\radius,0){$v_{11}$}}
\rput{250}(0,0){\cnode(\radius,0){3pt}{12}\uput[0]{-250}(\radius,0){$v_{12}$}}
\rput{260}(0,0){\cnode*(\radius,0){3pt}{13}\uput[0]{-260}(\radius,0){$v_{13}$}}
\rput{280}(0,0){\cnode*(\radius,0){3pt}{14}\uput[0]{-280}(\radius,0){$v_{14}$}}
\rput{290}(0,0){\cnode(\radius,0){3pt}{15}\uput[0]{-290}(\radius,0){$v_{15}$}}
\rput{300}(0,0){\cnode(\radius,0){3pt}{16}\uput[0]{-300}(\radius,0){$v_{16}$}}
\rput{310}(0,0){\cnode(\radius,0){3pt}{17}\uput[0]{-310}(\radius,0){$v_{17}$}}
\ncline{0}{1}
\ncline{1}{2}
\ncline{2}{3}
\ncline{3}{4}
\ncline{4}{5}
\ncline{5}{6}
\ncline{6}{7}
\ncline{7}{8}
\ncline{8}{9}
\ncline{9}{10}
\ncline{10}{11}
\ncline{11}{12}
\ncline{12}{13}
\ncline{13}{14}
\ncline{14}{15}
\ncline{15}{16}
\ncline{16}{17}
\ncline{17}{0}
\ncline{0}{9}
\ncline[linestyle=dotted]{4}{14}
\ncline[linestyle=dotted]{3}{15}
\ncline[linestyle=dotted]{2}{16}
\ncline[linestyle=dotted]{1}{17}
\ncline[linestyle=dotted]{5}{13}
\ncline[linestyle=dotted]{6}{12}
\ncline[linestyle=dotted]{7}{11}
\ncline[linestyle=dotted]{8}{10}
\my_axis{0}{1}{90}{1}{210}{0.7}
\endpspicture\\
  {(a)} & {(b)}
  \end{tabular}
  }
  \caption[Views of $P_{10}$ and $P_{11}$]%
          {\capStyle{(a) An orthogonal view of $P_{10}$. (b) An orthogonal view
           of $P_{11}$.}}
  \label{fig:lb-i}
\end{figure*}

Given $m$ and $n$, three parameters per polytope listed below govern
the final configurations of $P_m$ and $P_n$.
\begin{compactenum}
\item the exterior-dihedral angle $\alpha$ at the edge
$\overline{v_{j_0} v_{j_3}}$,
\item the angle $\beta = \angle{v'_{j_4} o v'_{j_5}}$, and
\item the angle $\gamma = \angle{v_{j_0} o v_{j_1-1}}$.
\end{compactenum}
The settings of these angles must satisfy certain conditions,
which in turn enable all the necessary intersections of edges in the
Gaussian map of the Minkowski sum. We denote the 
face $\overline{v_{j_5+1}v_{j_5}v_{j_1}v_{j_1-1}}$ adjacent to $f^u$
by $f^x$. The vertex $x = G(f^x)$ is the nearest vertex to
$u$. The $y$-coordinate of the vertex $w_n$ must be greater than the
$y$-coordinate of the edge $\overline{xv}_m$ at $z = 0$ in $P_m$'s coordinate
system. Similarly, the $y$-coordinate of the vertex $w_m$ must be 
greater than the $y$-coordinate of the edge $\overline{xv}_n$ at $z = 0$ in
$P_n$'s coordinate system.\footnote{The rotation of, say $P_n$, is
performed about the $Y$ axis. Thus, it has no bearing on
$y$-coordinates.} This is best seen in Figure~\ref{fig:lb-11}(c). The
values of the $y$-coordinates of $w_n$ and $w_m$ are simply
$\sin(\alpha_n)$ and $\sin(\alpha_m)$, respectively. The value of the
$y$-coordinate of the edge $\overline{xv}_m$ at $z = 0$ however depends on all
the three parameters $\alpha_m$, $\beta_m$, and $\gamma_m$. Similarly,
the $y$-coordinate of the edge $\overline{xv}_n$ at $z = 0$ in the respective
coordinate system depends on $\alpha_n$, $\beta_n$, and $\gamma_n$.
Instead of deriving an expression that directly evaluates these
$y$-coordinates, we suggest an iterative procedure that decreases the
angles at every iteration until the conditions above are met, and
argue that this procedure eventually terminates, because at the limit,
we are back at the case where $m = n = 5$, for which valid settings exist.

\begin{figure*}[t]
  \centerline{
    \begin{tabular}{cccc}
    \epsfig{figure=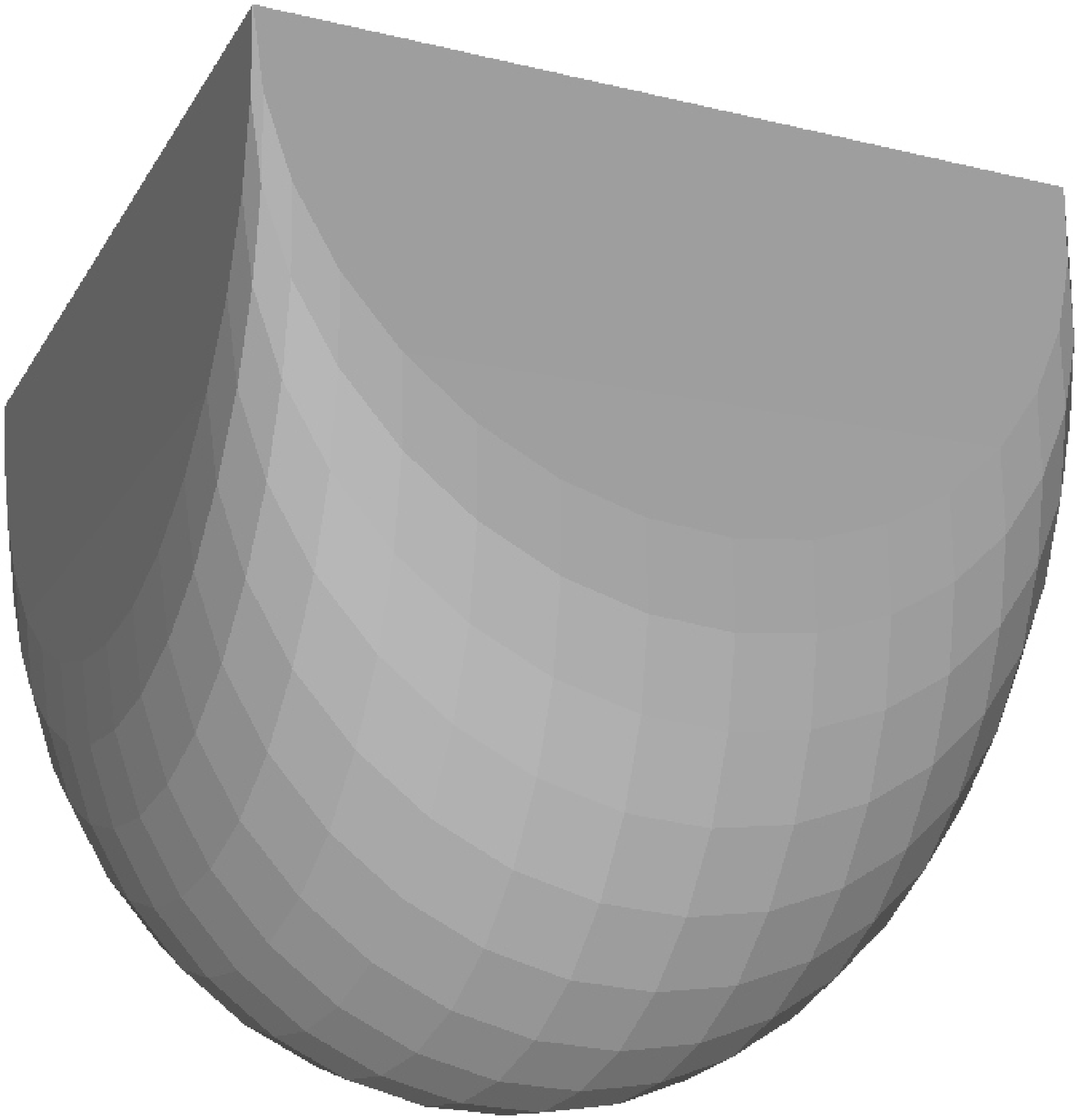,width=3.5cm,silent=} &
    \epsfig{figure=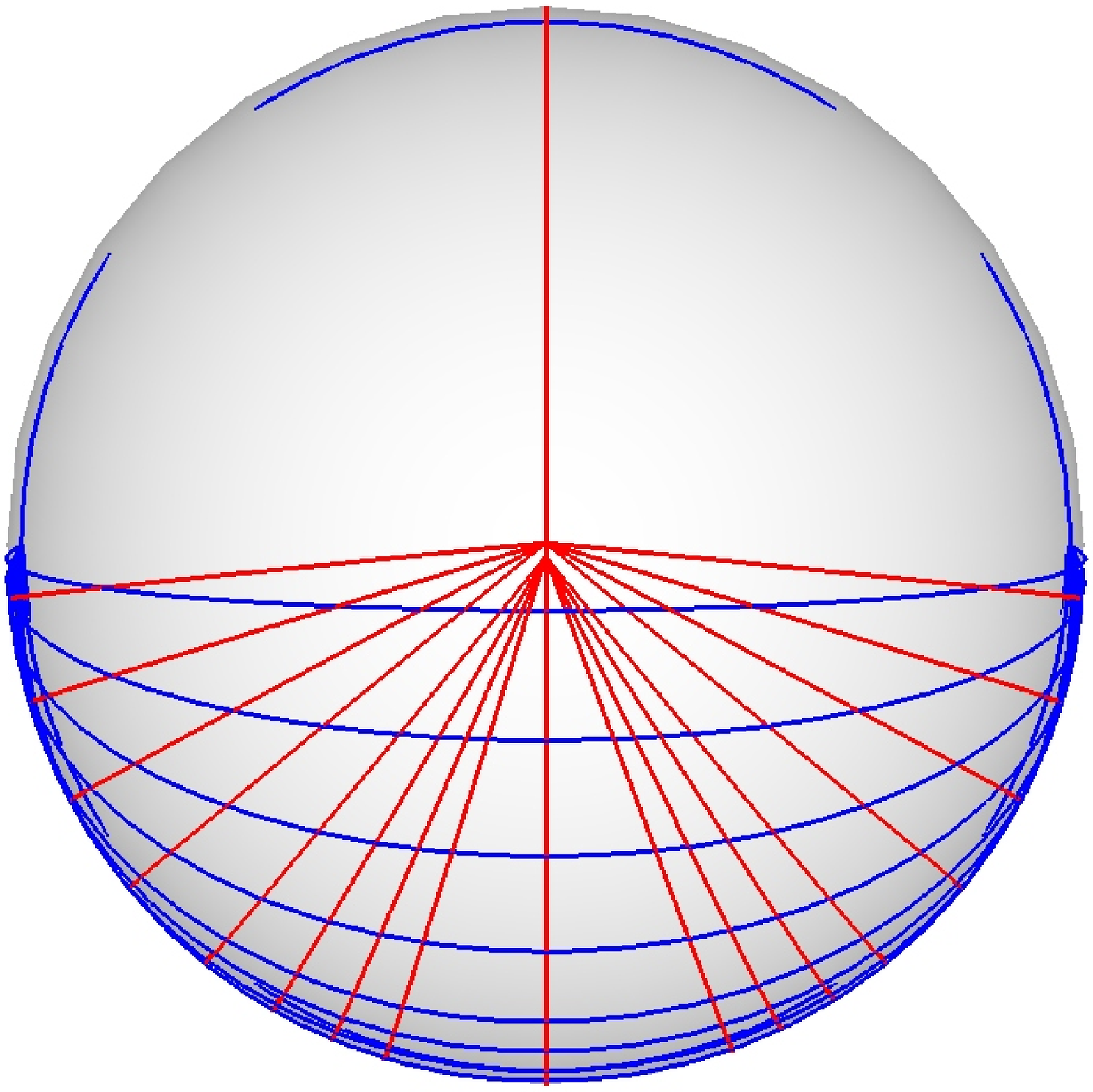,width=3.5cm,silent=} &
    \epsfig{figure=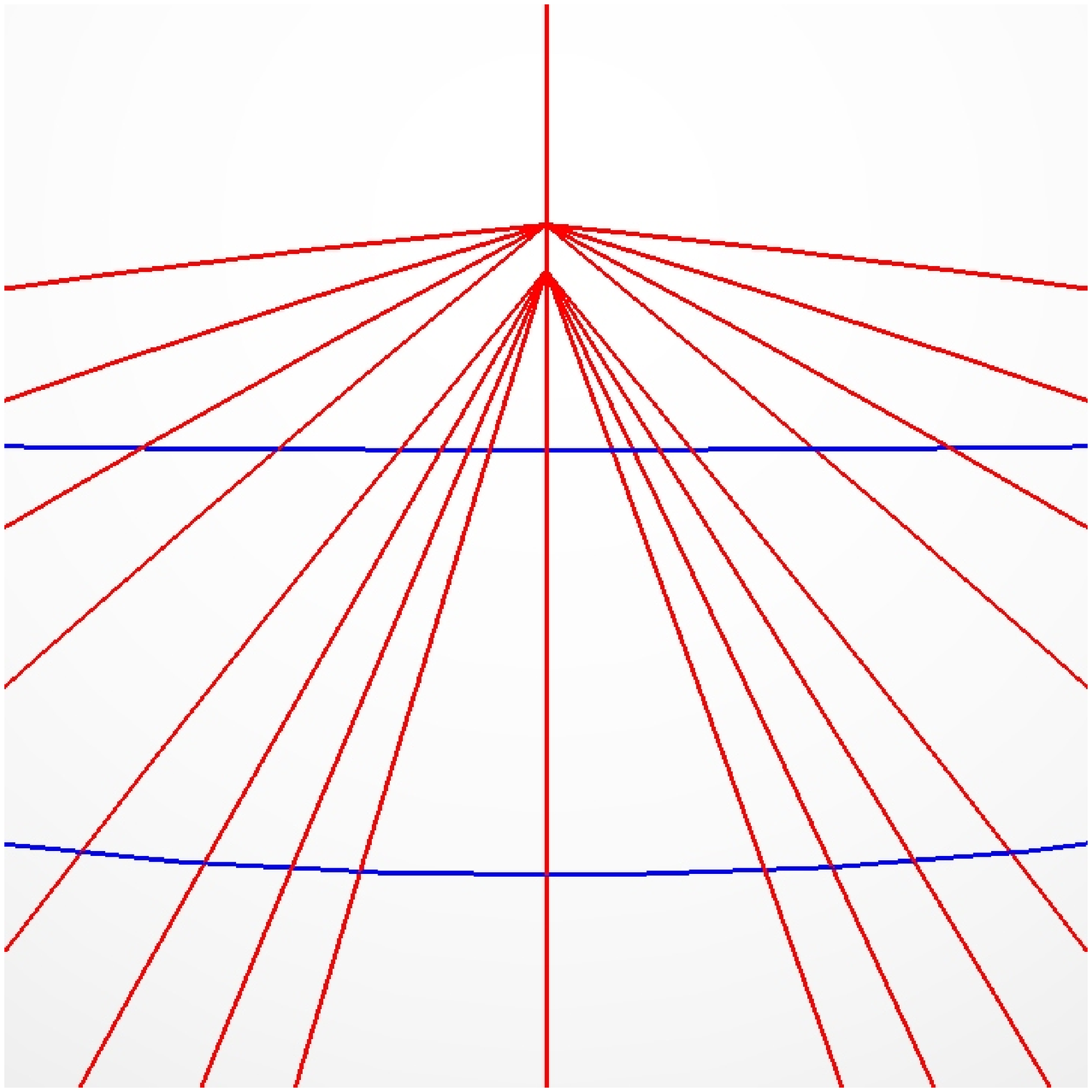,width=3.5cm,silent=} &
    \epsfig{figure=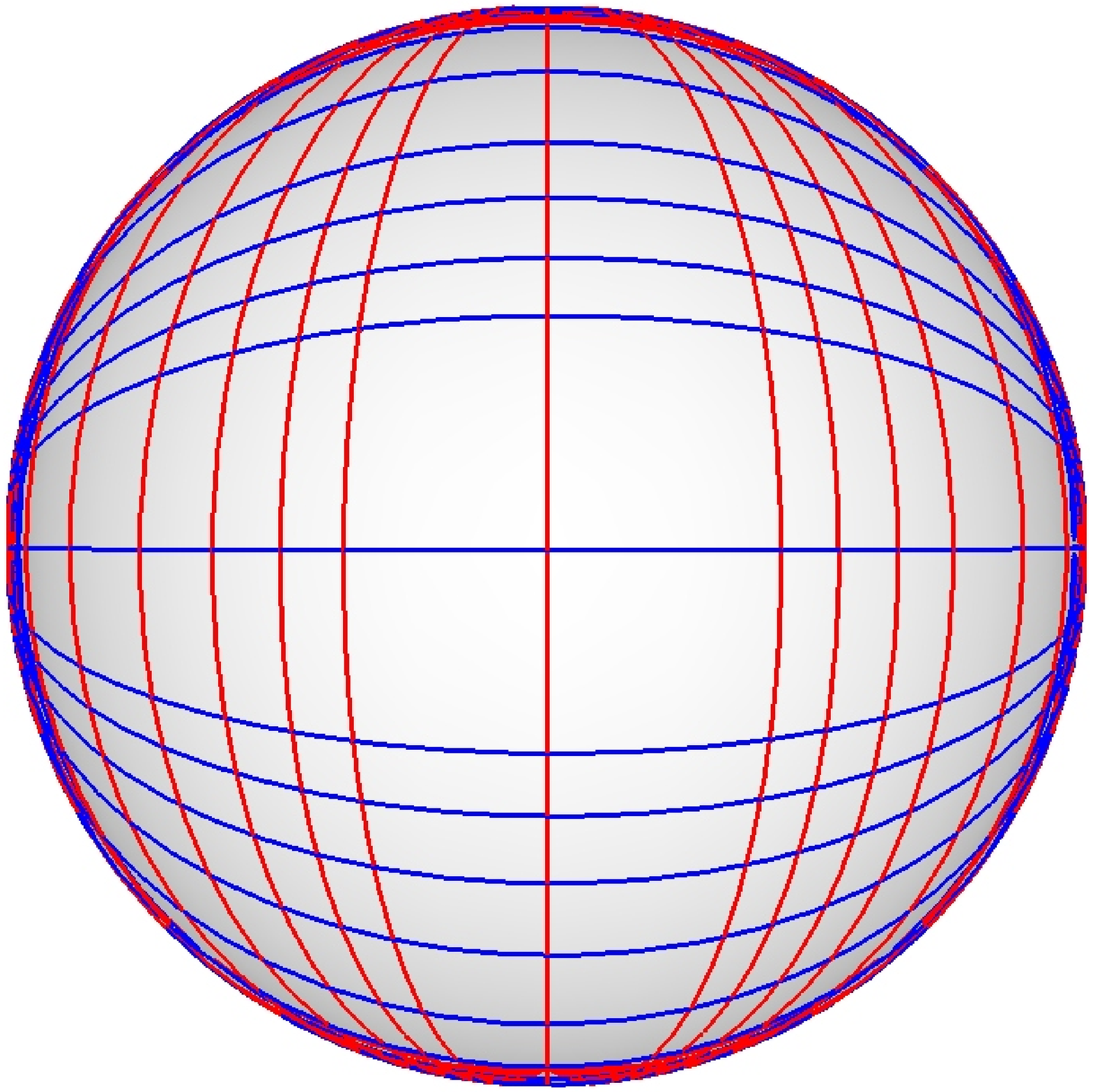,width=3.5cm,silent=}\\
    {(a)} & {(b)} & {(c)} & {(d)}
    \end{tabular}
  }
  \caption[The Minkowski sum of $M_{11,11}$ and $M_{11,11}$ rotated about the
           $Y$ axis]%
          {\capStyle{(a) The Minkowski sum $M_{11,11} = P_{11} \oplus P'_{11}$,
           where $P'_{11}$ is $P_{11}$ rotated $90^{\circ}$ about the $Y$ axis.
           (b) The Gaussian map of $M_{11,11}$ looked at from $z = \infty$.
           (c) A scaled up view of the Gaussian map of $M_{11,11}$ looked at
           from $z = \infty$. (d) The Gaussian map of $M_{11,11}$ looked at
           from $y = -\infty$.}}
  \label{fig:lb-11}
\end{figure*}

\subsection{Constructing $P_4$}
\label{ssec:mscy:lower:p4}
\begin{wrapfigure}[7]{l}{3.5cm}
  \vspace{-15pt}
  \epsfig{figure=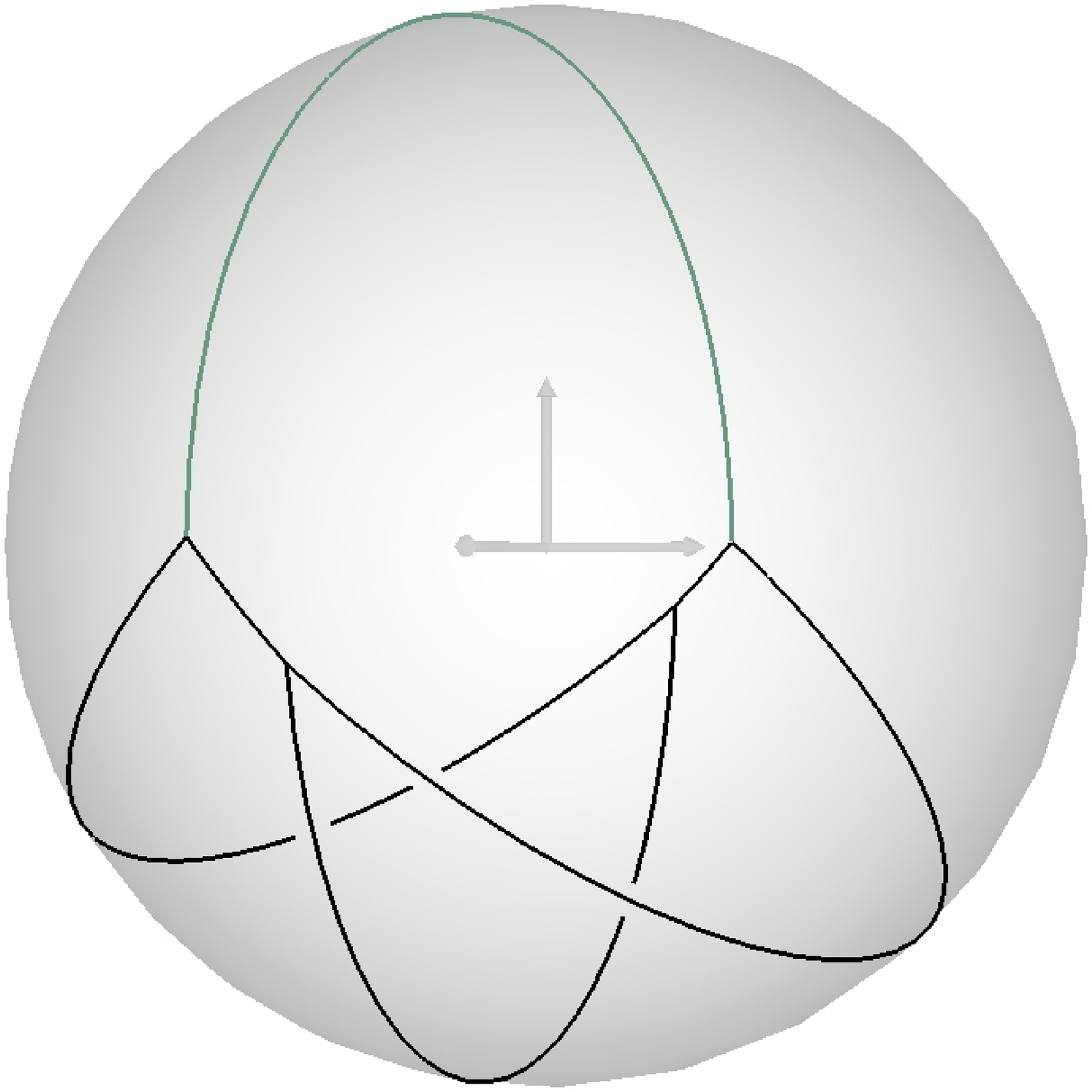,width=3.5cm,silent=}
  \vspace{-8pt}
  \vspace{-3.5cm}
  \pspicture[](0,0)(3.5,3.5)
  \psset{unit=1cm,linewidth=1pt}
  \rput{0}(2.30,1.98){\qdisk(0,0){2pt}}
  \rput{0}(2.12,1.79){\qdisk(0,0){2pt}}
  \rput{0}(0.96,1.64){\qdisk(0,0){2pt}}
  \rput{0}(0.63,2.01){\qdisk(0,0){2pt}}
  \endpspicture
  \vspace{-15pt}
\end{wrapfigure}
Recall that $P_4$ has 4 facets, 6 edges, and 4 vertices. Therefore, it
cannot be constructed according to the prescription provided in the
previous section. Applying the same principles though, we place two
vertices of $G_4$ near the pole $(0,0,-1)$, and two vertices near the
opposite pole $(0,0,1)$. One edge, which
connects a vertex near one pole to a vertex near the other, lightly
shaded in the figure on the left, is entirely contained in $H^+$. The
other three edges that connect vertices near opposite poles mostly lie
in $H^-$. They form a set of $2i - 5 = 3$ edges, denoted as $E'_4$.
The length of every edge in $E'_4$ is almost $\pi$. In contrast to the
case $i \geq 5$, two out of the three edges in $E'_4$ cross the plane
$y = 0$. Namely, small sections of them lie in $H^+$. As in the case
$i > 4$, one edge, the lightly shaded one, is entirely contained in
$H^+$.

\begin{wrapfigure}{r}{5.1cm}
  \vspace{-15pt}
  \pspicture[](-2.5,-2.5)(2.6,1.6)
  \psset{unit=1cm,linewidth=1pt,framesep=1.5pt}
  \pnode(0,0){c}\uput[ur]{0}(0,0){$o$}
  \psarc[linewidth=0.5pt](0,0){2}{180}{0}
  \cnode*(2,0){3pt}{0}\uput[0]{0}(2,0){$v_0$}
  \cnode*(-2,0){3pt}{1}\uput[180]{0}(-2,0){$v_1$}
  \rput{260}(0,0){\cnode*(2,0){3pt}{2}\uput[0]{-260}(2,0){$v_2$}}
  \rput{280}(0,0){\cnode*(2,0){3pt}{3}\uput[0]{-280}(2,0){$v_3$}}
  \ncline{0}{1}
  \ncline{1}{2}
  \ncline{2}{3}
  \ncline{3}{0}
  \ncline{2}{0}
  \ncline[linestyle=dotted]{1}{3}
  \my_axis{0}{1}{90}{1}{210}{0.7}
  \endpspicture
  \vspace{-10pt}
\end{wrapfigure}
We construct $P_4$, such that the two facets
$f^1 = \overline{v_0v_1v_2}$ and $f^2 = \overline{v_0v_2v_3}$ are
visible when looked at from $z = \infty$, and when looked
at from $z = -\infty$, the remaining two facets 
$f^3 = \overline{v_3v_1v_0}$ and $f^4 = \overline{v_3v_2v_1}$ are
visible. As depicted on the right, the projection of all four vertices
onto the plane $z = 0$ lie on the unit circle. The vertices $v_0$
and $v_2$ lie on the plane $z = 0$, and the vertices $v_1$ and $v_3$
lie in a parallel plane. The distance between the planes is small to
form small exterior-dihedral angles at the edges $\overline{v_0v_2}$
and $\overline{v_1 v_3}$.

As in the general case, two parameters govern the exact placement of
$v_1$, $v_2$, and $v_3$. One is the size of the exterior-dihedral
angle at the edge $\overline{v_0v_2}$. The other parameter is
the size of the angle $\angle{v_2 o v'_3}$, where $v'_3$ is the
projection of $v_3$ onto the plane $z = 0$. The sizes of these angles
are determined by the same rationale as in the general case.

This concludes the proof of the lower bound of Theorem~\ref{the:theorem}
for the special case $k=2$.
\hyphenation{Min-kowski}
\section{Maximum Complexity of Minkowski Sums of Many Polytopes}
\label{sec:mscy:many}
Let $P_1,P_2,\ldots,P_k$ be a set of $k$ polytopes in $\mathbb{R}^3$, 
such that the number of facets of $P_i$ is $m_i$ for $i = 1,2,\ldots,k$.
In this section we present a tight bound on the number of facets of the
Minkowski sum $M = P_1 \oplus P_2 \oplus\ldots\oplus P_k$ generalizing
the arguments presented above for $k=2$.

\subsection{The Lower Bound}
\label{ssec:mscy:many:lower-bound}
\begin{wrapfigure}{l}{3.5cm}
  \vspace{-15pt}
  \epsfig{figure=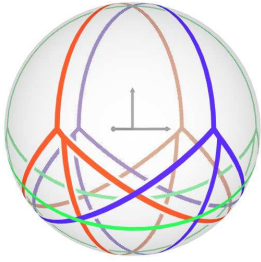,width=3.5cm,silent=}
  \caption[The overlay of the Gaussian maps of three rotated tetrahedra]%
          {\capStyle{The overlay of the Gaussian maps of three
           tetrahedra rotated about the $Y$ axis $0^{\circ}$,
           $60^{\circ}$, and $120^{\circ}$, respectively.}}
  \label{fig:lb-4-4-4}
  \vspace{-10pt}
\end{wrapfigure}
Given $k$ positive integers $m_1,m_2,\ldots,m_k$, such that $m_i \geq 4$,
we describe how to construct $k$ polytopes in $\mathbb{R}^3$ with
corresponding number of facets, such that the number of facets of
their Minkowski sum is exactly
$\sum_{1 \leq i < j \leq k}(2m_i - 5)(2m_j - 5) + \binom{k}{2} +
\sum_{i=1}^k m_i$.
More precisely, given $i$, we describe how to construct a skeleton of
a polytope $P_i$ with $i$ facets, $3i - 6$ edges, and $2i-4$ vertices,
and prove that the number of facets of the Minkowski sum 
$M = P_1 \oplus P_2 \oplus\ldots\oplus P_k$ of the $k$ polytopes
properly adjusted and oriented is exactly the expression above.
We use the same construction described in Section~\ref{sec:mscy:lower}.

The number of facets in the Minkowski sum of $P_i, i = 1,2,\ldots,k$
is maximal, when the number of vertices in the overlay of
$G_i, i = 1,2,\ldots,k$ is maximal. This occurs, for example, when
$G_i$ is rotated $180^{\circ} (i-1)/k$ about the $Y$ axis for
$i = 1,2,\ldots,k$, as depicted on the left for the case of
$m_1 = m_2 = m_3 = 4$. (Recall, that $E'_i$ refers to the set of edges
that span the lowest hemispheres, and its cardinality is smaller than
cardinality of $E$ by one.) In this configuration, all the $2m_i - 5$
edges in $E'_i$ intersect all the $2m_j - 5$ edges in $E'_j$, for
$1 \leq i < j \leq k$. All intersections occur in $H^-$. In addition,
the edge $\overline{uv}_{m_i}$ intersects the edge $\overline{uv}_{m_j}$
for $1 \leq i < j \leq k$. These intersection points lie in $H^+$
near the pole $(0,1,0)$. Counting all these intersections results with  
$\sum_{1 \leq i < j \leq k}(2m_i - 5)(2m_j - 5) + \binom{k}{2}$.
Adding the original vertices of $G(P_i), i = 1,2,\ldots,k$, yields the
bound asserted in Theorem~\ref{the:theorem}.

\subsection{The Upper Bound}
\label{ssec:mscy:upper-bound-many}
We can apply the special case $k=2$ of Theorem~\ref{the:theorem} to obtain
\begin{align*}
f(m_1,m_2,\ldots,m_k) & \leq f(m_1,f(m_2,m_3,\ldots,m_k)) \\
& \leq 4m_1 f(m_2,m_3,\ldots,m_k) - 9m_1 - 9 f(m_2,m_3,\ldots,m_k) + 26\\
& \leq 4^k \prod_{i=1}^k m_i + \ldots
\end{align*}
However, we can apply a technique similar to the one used in
Section~\ref{sec:mscy:upper} and improve this upper bound, but first we
must extend Lemma~\ref{lemma:overlay-single}.

\begin{lemma}
\label{lemma:overlay-many}
Let $G_1,G_2,\ldots,G_k$ be a set of $k$ Gaussian maps of convex polytopes,
and let $G$ be their overlay. Let $f_i$ denote the number of faces of
$G_i$, and let $f$ denote the number of faces of $G$. Then, 
$f \le \sum_{1\leq i<j\leq k}f_if_j-(k-2)\sum_{1\leq i\leq k}f_i+(k - 1)(k - 2)$.
\end{lemma}
\begin{proof}
  Let us choose two antipodal points on the sphere $\mathbb{S}^2$, such
  that no arc of the overlay is aligned with them. In particular, the
  points are in the interior of two distinct faces. We consider these two
  points to be the north pole and south pole of the sphere, and define
  the direction \emph{west} to be the clockwise direction when looking
  from the north pole toward the south pole. We define a
  \emph{western-most corner} to be a pair of a face and one of its
  vertices, which is to the west of all of its other vertices. Apart
  from the two faces, which contain the poles, any face has a unique
  western-most corner, since no edge is aligned with the poles, and all
  faces of any Gaussian map are spherically convex. So for any overlay
  with $f$ faces, there are $f-2$ such western-most corners.
  
  The maximal number of faces is attained when the overlay $G$ is
  non-degenerate. Thus, a vertex of $G$ is either the intersection of
  two edges of some distinct $G_i$ and $G_j$, or a vertex of some
  $G_i$. Therefore, a western-most corner for a face of $G$ is either
  a western-most corner for the overlay of some $G_i$ and $G_j$, or a
  western-most corner for some $G_i$, in which case it also is a
  western-most corner for any overlay involving $G_i$.  The number of
  western-most corners in the Gaussian map $G_i$ is $f_i-2$, and the
  maximal number of western-most corners in the overlay of some $G_i$ and
  $G_j$ is $f_if_j-2$.

  We can therefore write:
\begin{align*}
  f & \leq \sum_{1\leq i<j\leq k}(f_if_j-2)-(k-2)\sum_{i=1}^k(f_i-2) + 2.
\end{align*}
This corresponds to summing the western-most corners appearing in the
overlay of all pairs of Gaussian maps, and subtracting $(k-2)$ times
the western-most corners appearing in all original Gaussian maps, since
each of them appeared $(k-1)$ times in the first sum. Finally, we have:
\begin{align*}
  \sum_{1\leq i<j\leq k}(f_if_j-2)-(k-2)\sum_{i=1}^k(f_i-2) + 2 & =
  \sum_{1\leq i<j\leq k}f_if_j-(k-2)\sum_{i=1}^kf_i + (k - 1)(k - 2)\ .
\end{align*}
\end{proof}

Let $P_1,P_2,\ldots,P_k$ be $k$ polytopes in $\mathbb{R}^3$ with 
$m_1,m_2,\ldots,m_k$ facets, respectively. Let $G(P_i)$ denote the
Gaussian map of $P_i$, and let $v_i$, $e_i$, and $f_i$ denote the
number of vertices, edges, and faces of $G(P_i)$, respectively. Let
$v_x$ denote the number of intersections of edges of $G(P_i)$ and
edges of $G(P_j)$, $i \neq j$ in $G(M)$. Applying the same technique
as in Section~\ref{sec:mscy:upper}, that is, counting the total degrees 
of vertices in $G(M)$ implies that $\sum_{i=1}^k e_i + 2v_x = e$.
Using Euler's formula, we get $\sum_{i=1}^k e_i + 2v_x = f + v - 2$.
Applying Lemma~\ref{lemma:overlay-many} and respecting 
$v = \sum_{1 \leq i \leq k} v_i + v_x$, we can bound $v_x$ from above:
\begin{align}
v_x \leq \sum_{1\leq i<j\leq k}f_if_j-(k-2)\sum_{i=1}^kf_i + (k - 1)(k - 2) + \sum_{i=1}^k (v_i - e_i) - 2\label{vxm1}\ .
\end{align}
According to Corollary~\ref{col:maximal} the maximum number of facets
of the Minkowski sum of two polytopes is attained when the number of
edges of each summand is maximal. We need to establish a similar
property for the general case. Generalizing the derivation procedure in
Section~\ref{sec:mscy:upper}, we introduce $k$ non-negative integers
$\ell_i, i = 1,2,\ldots,k$, such that $e_i = 3v_i - 6 - \ell_i$ and
$f_i = 2v_i - 4 - \ell_i$. Substituting $e_i$ in~\eqref{vxm1} we get:
\begin{align}
v_x & \leq \sum_{1\leq i<j\leq k}f_if_j-(k-2)\sum_{i=1}^kf_i + (k - 1)(k - 2) + \sum_{i=1}^k (v_i - 3v_i + 6 + \ell_i) - 2\notag\\
    & \leq \sum_{1\leq i<j\leq k}f_if_j-(k-2)\sum_{i=1}^kf_i -
\sum_{i=1}^k (2 v_i - 5) + \sum_{i=1}^k \ell_i + k^2 - 2k\label{vxm2}\ .
\end{align}
Substituting $f_i$ in~\eqref{vxm2} we get:
\begin{align}
v_x & \leq \sum_{1\leq i<j\leq k}(2v_i - 4 - \ell_i)(2v_j - 4 -\ell_j) -
           (k-2)\sum_{i=1}^k(2v_i - 4 - \ell_i) -
           \sum_{i=1}^k (2v_i - 5 - \ell_i) + k^2 - 2k\notag\\
    & = \sum_{1\leq i<j\leq k}(2v_i - 5)(2v_j - 5) + \binom{k}{2} + h(\ell_1, \ell_2,\ldots,\ell_k)\notag\ ,
\end{align}
where
\begin{align}
h(\ell_1, \ell_2,\ldots,\ell_k) & =
\sum_{1 \leq i < j \leq k}(\ell_i \ell_j - \ell_j(2v_i - 5) - \ell_i(2v_j - 5))\notag\\
& = \sum_{i=1}^k\ell_{i}(\sum_{j \neq i}(\ell_j/2 - (2v_j - 5)))\notag\ .
\end{align}
Connectivity of $G(P_i)$ implies that $\ell_i \leq 2v_i - 5$, which in
turn implies that $h(\ell_1, \ell_2,\ldots,\ell_k) \leq 0$.
Thus, we have:
\begin{align}
v_x & \leq \sum_{1 \leq i < j \leq k}(2v_i - 5)(2v_j - 5) + \binom{k}{2}\label{vxm3}\ .
\end{align}

We conclude that the exact maximum number of facets of the Minkowski sum of $k$
polytopes cannot exceed $\sum_{1 \leq i < j \leq k}(2m_i - 5)(2m_j - 5) + 
\sum_{1 \leq i \leq k}m_i + \binom{k}{2}$, which completes the proof of
Theorem~\ref{the:theorem}.
For example, the exact maximum number of facets of the Minkowski sum of
$k$ tetrahedra is $5k^2 - k$. For $k = 2$ the expression evaluates to 18.

\begin{savequote}[10pc]
  Divide each difficulty into as many parts as is feasible and
  necessary to resolve it.
\sffamily
\qauthor{Rene Descartes}
\end{savequote}
\chapter{Assembly Planning}
\label{chap:assem_plan}
\newcommand{\redpart}{{\sc \textcolor{red}{$R$}}}
\newcommand{\greenpart}{{\sc \textcolor{green}{$G$}}}
\newcommand{\bluepart}{{\sc \textcolor{blue}{$B$}}}
\newcommand{\purplepart}{{\sc \textcolor{magenta}{$P$}}}
\newcommand{\yellowpart}{{\sc \textcolor{orange}{$Y$}}}
\newcommand{\turquoisepart}{{\sc \textcolor{cyan}{$T$}}}

\begin{figure*}[!b]%
  \setlength{\tabcolsep}{3pt}
  \centerline{%
    \begin{tabular}{cccc}
      \epsfig{figure=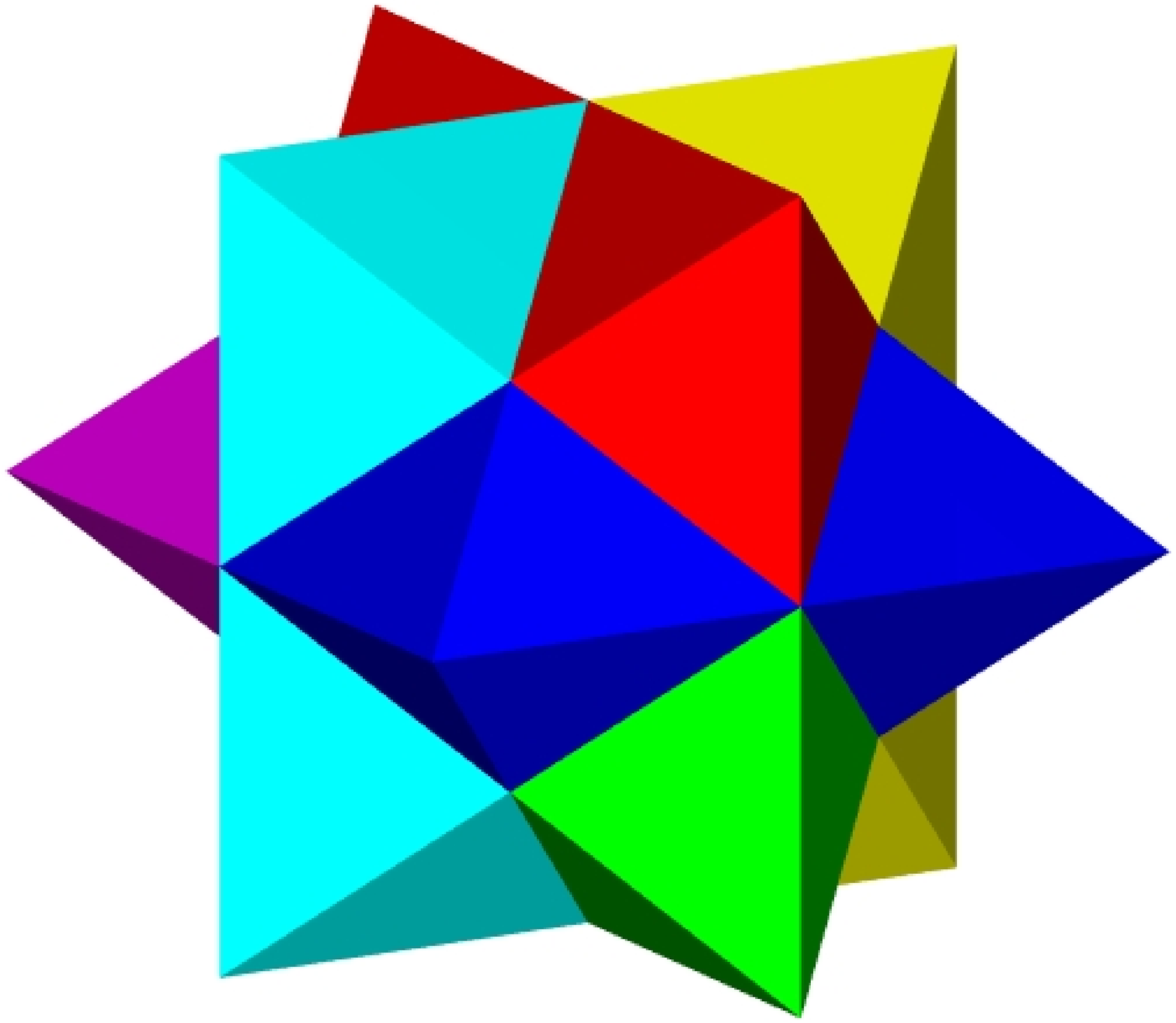,height=3.5cm,silent=} &
      \epsfig{figure=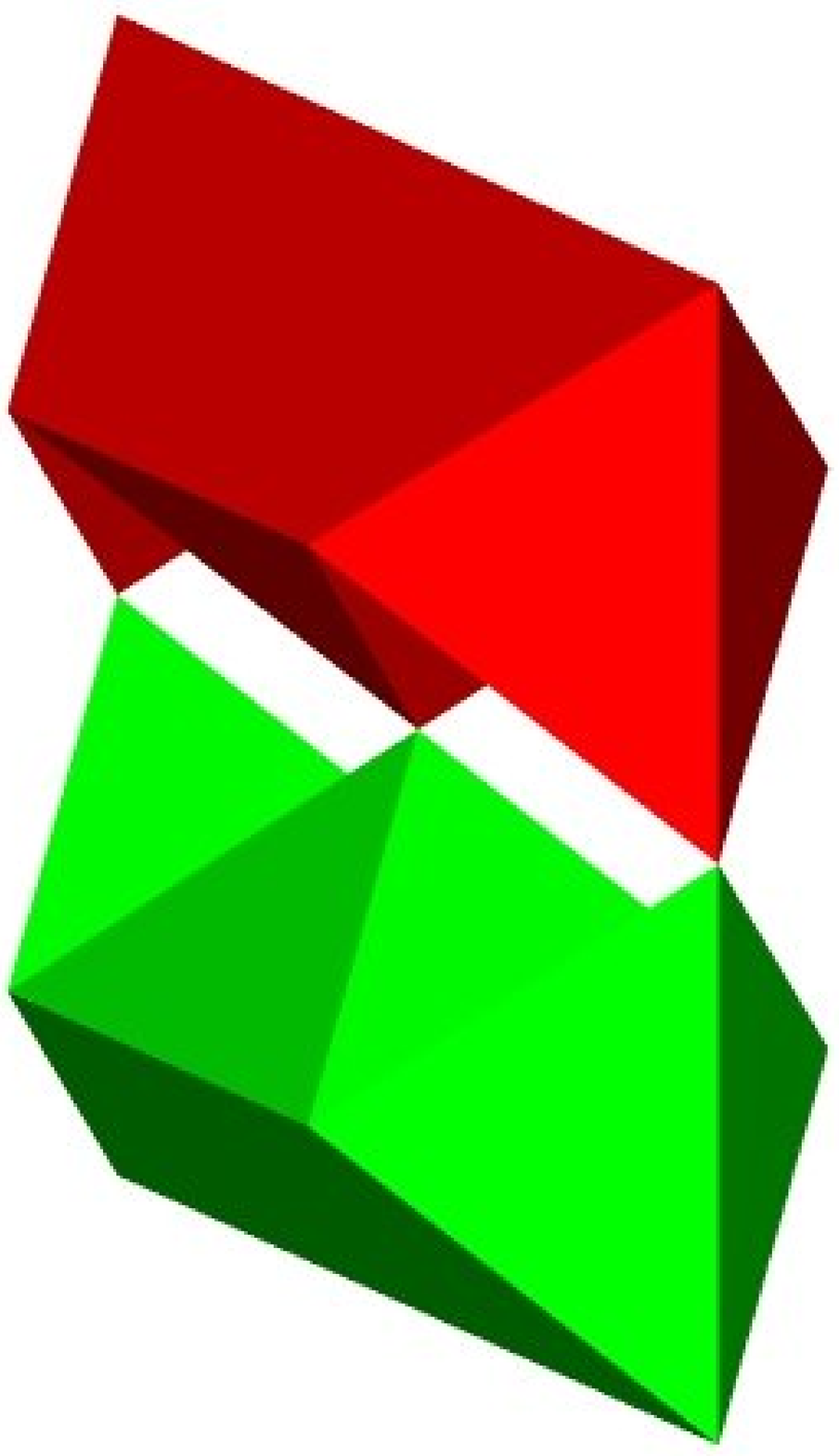,height=3.5cm,silent=} &
      \epsfig{figure=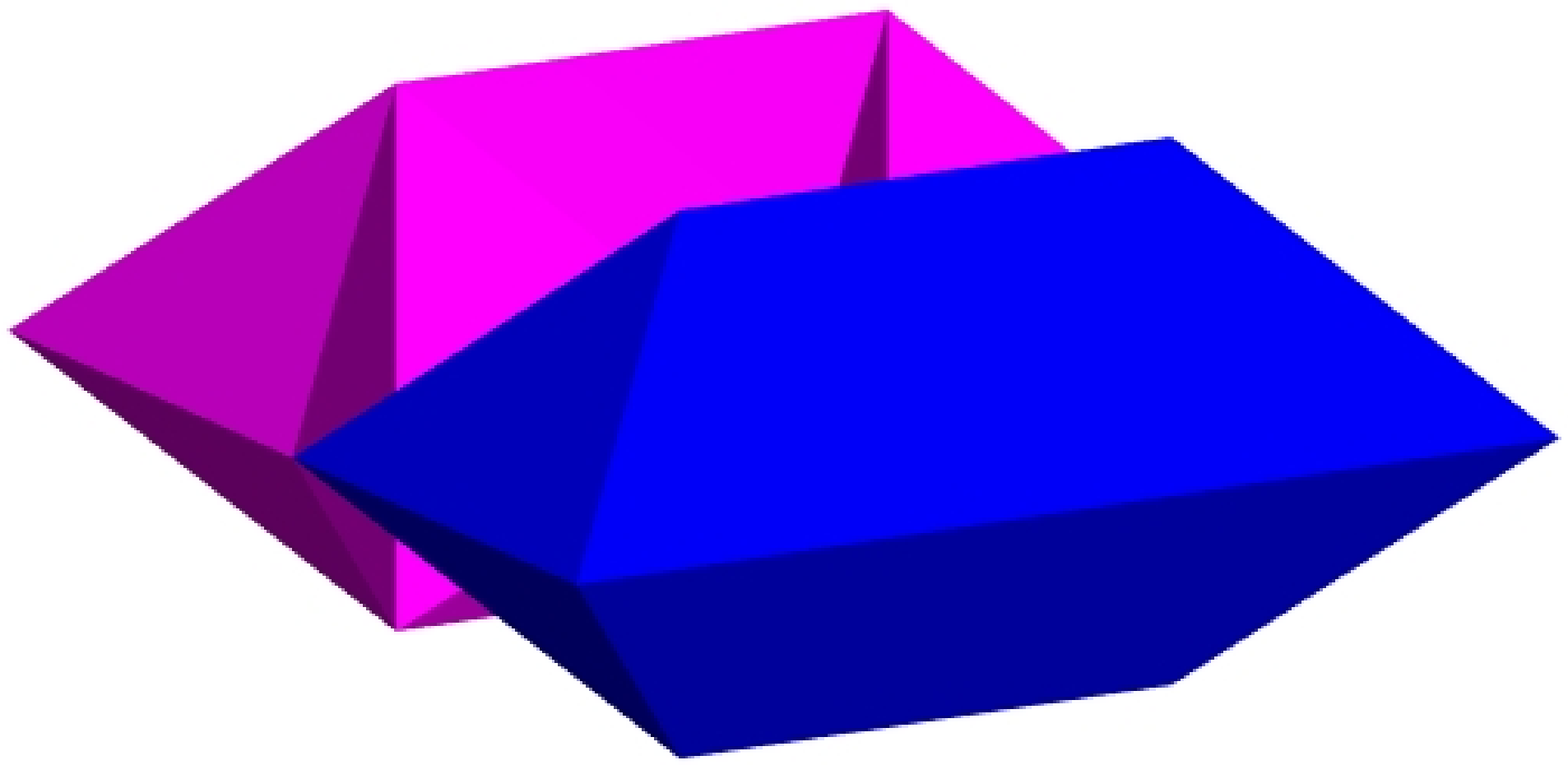,height=3.5cm,silent=} &
      \epsfig{figure=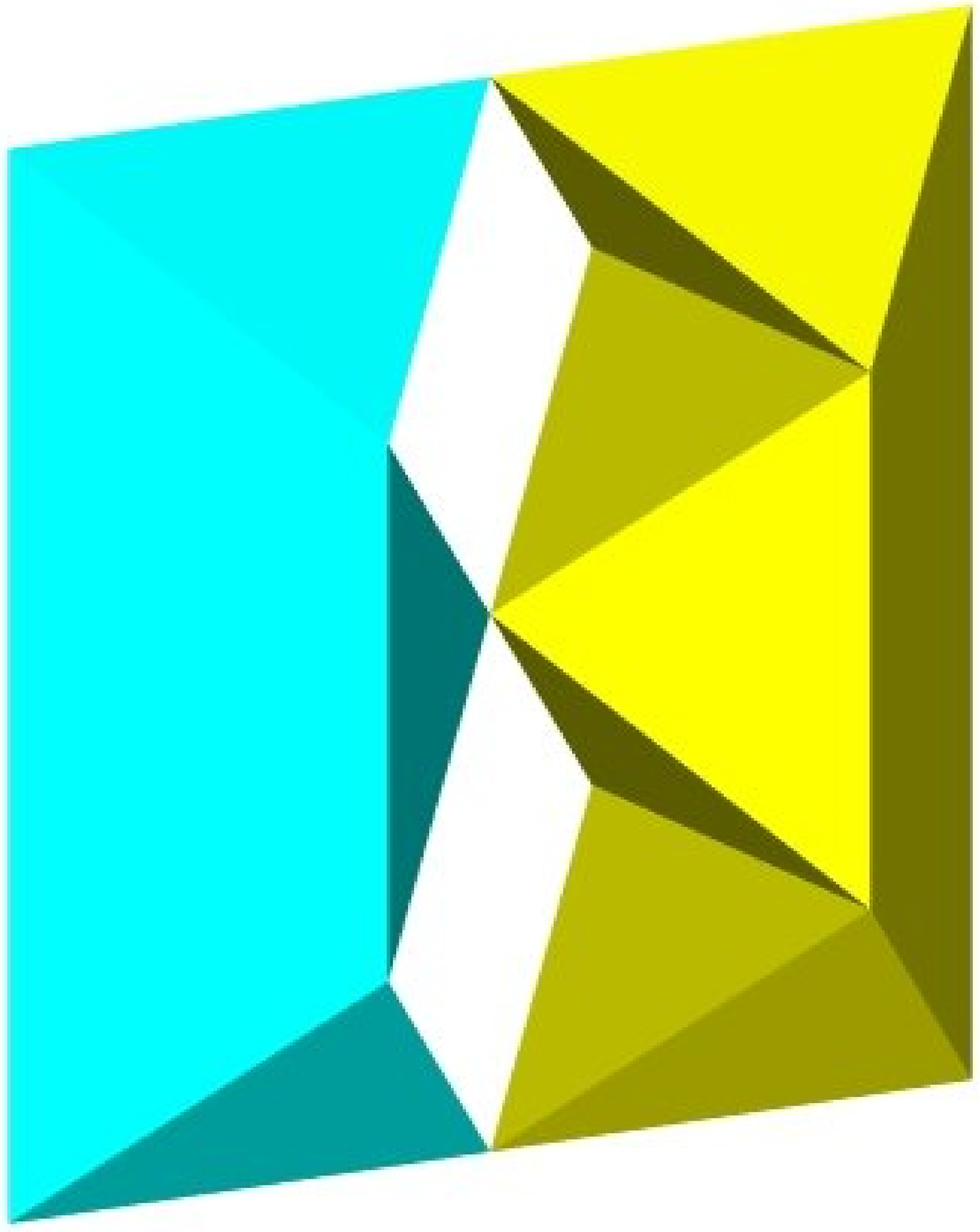,height=3.5cm,silent=}\\
      (a) & (b) & (c) & (d)
    \end{tabular}
  }
  \caption[The Split Star assembly]%
          {\capStyle{(a) The Split Star assembly, and (b),(c), and (d)
              the Split Star  six parts divided into three pairs of 
           symmetric parts. The six parts are named according to their
           color $R$(ed), $G$(reen), $B$(lue), $P$(urple), $Y$(ellow),
           and $T$(urquoise).}}
  \label{fig:split-star}
\end{figure*}

Assembly partitioning\index{assembly partitioning} with an infinite
translation is the application of an infinite translation to partition
an assembled product into two complementing subsets of parts, referred
to as subassemblies, each treated as a rigid body. We present an
exact implementation of an efficient algorithm based on the general
framework described in~\cite{hlw-gfapm-00,wl-grma-94} to obtain such a
motion and subassemblies given an assembly of polyhedra in $\rrr$. As
usual, we do not assume general position. Namely, we handle degenerate
input, and produce exact results. As often occurs, motions that
partition a given assembly or subassembly might be isolated in the
infinite space of motions. Any perturbation of the input or of
intermediate results, caused by, for example, imprecision, might
result with dismissal of valid partitioning-motions. In the extreme
case, there is only a finite number of valid partitioning motions, as
occurs in the assembly shown in Figure~\ref{fig:split-star}, no motion
may be found, even though such exists. Proper handling requires
significant enhancements applied to the original algorithmic
framework. The implementation is based on software components
introduced in Chapters~\ref{chap:aos}
and~\ref{chap:mink-sums-construction}. They paved the way to a
complete, efficient, and concise implementation.
\section{Introduction}
\label{sec:assem_plan:intro}
Assembly planning\index{planning!assembly} is the problem of finding a
sequence of motions that transform the initially separated parts of an
assembly to form the assembled product. The reverse order of
sequenced motions separates an assembled product to its parts. Thus,
for rigid parts, assembly planning and disassembly planning refer to
the same problem, and used interchangeably. In this chapter we
concentrate on the case where the assembly consists of polyhedra in
$\mathbb{R}^3$ and the motions are infinite translations.

The \Index{motion space} is the space of possible motions that
assembly parts may undergo. For each motion in a motion space, a
subset of parts of a given assembly may collide with a different
subassembly, when transformed as a rigid body according to the
motion. Pairs of subassemblies that collide constitute constraints.
The motion space approach dictates the precomputation of a 
decomposition of a motion space into regions, such that the constraints
among the parts in the assembly are the same for all the motions in the
same region. All constraints over a region are represented by a graph,
called the {\em \Index{directional blocking graph}}
(DBG)\index{DBG|see{directional blocking graph}}~\cite{wl-grma-94}.
The collection of all regions in a motion space together with their
associated DBGs, collectively referred to as the
{\em \Index{nondirectional blocking graph}}
(NDBG)\index{NDBG|see{nondirectional blocking graph}}, can be used to
obtain assembly (or disassembly) sequences.

The general framework and some of the techniques presented here have already
been described in a series of papers and reports published in the past
mainly during the late 90's. Halperin, Latombe, and Wilson made the
connection between previously presented techniques that had used the
motion space approach, and introduced a unified general
framework~\cite{hlw-gfapm-00} at the end of the previous millennium. Only
few publications related to this topic appeared ever since, to
the best of our knowledge, which creates a long gap in the time line of the
respective research. We certainly hope that the tools exposed in this chapter
will help revive the research on algorithmic assembly planing, a
research subject of considerable importance. Moreover, we believe that
the machinery presented here, together with other recent advances in
the practice of computational-geometry algorithms, can more generally
support the development of new and improved techniques in
{\em \Index{algorithmic automation}}~\citelinks{algorithmic-automation}.

Solution to the assembly planning problem enables better feedback to
designers. It provides a design team with additional tools to asses a
design, prior to the construction of physical mock-ups, and helps
creating products that are more cost-effective to manufacture and
maintain. This is emphasized in light of the strategy to ``plan
anywhere, build anywhere'' many Computer Aided Design (CAD) and
Computer Aided Manufacturing (CAM) companies are trying to
adopt. Assembly sequences are also useful for selecting assembly tools
and equipment, and for laying out manufacturing facilities.

We restrict ourselves to two-handed partitioning steps,
meaning that we partition the given assembly into two complementing
subsets each treated as a rigid body. Even for two-handed
partitioning, if we allow arbitrary translational motions
(and not restrict ourselves to infinite translations) the
problem is NP-hard~\cite{kk-ppatc-95}. The
more general problem of assembly sequencing, namely
planning a total ordering of assembly operations that merge
the individual parts into the assembled product, is
PSPACE-hard, even when each part is a polygon with a
constant number of vertices~\cite{n-ps-88}.

Notice that the problem that we address in this chapter, namely
partitioning with {\em infinite translations}, is technically
considerably more complex than partitioning with {\em infinitesimal
motions}. Although the latter may sound more general, as it handles
infinitesimal translations and {\em rotations}, it is far simpler to
implement, since it deals only with constraints that can be described
linearly. Thus, the problem can be reduced to
solving linear programs. Indeed, there are several implementations
for partitioning with infinitesimal motions (see, e.g.,
\cite{ghhlw-papum-98,ss-ocbta-94}),
but none that we are aware of dealing robustly with infinite
translations. The shortcoming of using infinitesimal motions only is
that suggested disassembly moves may not be extendible to practical
finite-length separation motions.

Infinite-translation partitioning was not fully robustly implemented
until recently, in spite of being more useful than infinitesimal
partitioning, most probably due to the hardship of accurately
constructing the underlying geometric primitives. What enables the
solution that we present here, is the significant headway in the
development of computational-geometry software over the past decade,
the availability of stable code in the form of the Computational
Geometry Algorithms Library (\cgal{}) in general and code for
Minkowski sums\index{Minkowski sum} and
arrangements\index{arrangement} in particular~\cite{wfzh-aptac-07}.

The implementation presented in this chapter is based on the \aos{}
package of \cgal{}; see Chapter~\ref{chap:aos} for more details. The
implementation uses in particular arrangements of geodesic arcs
embedded on the sphere; see Section~\ref{sec:aos:geodesics}. The
ability to robustly construct such arrangements, and carry out exact
operations on them using only (exact) rational arithmetic is a key
property that enables an efficient implementation. The implementation
exploits supported operations, and requires additional operations,
e.g., \Index{central projection} of polyhedra, which we implemented.
We plan to make these new components available as part of a future
public release of \cgal{} as well. 

\subsection{Split Star Puzzle}
\label{ssec:assem_plan:split-star}
We use the assembly depicted in Figure~\ref{fig:split-star} as a running
example throughout the chapter. The name ``Split Star'' was given to this
shape by Stewart Coffin in one of his Puzzle Craft booklet editions back
in 1985. He uses the term {\em \Index{puzzle}} to refer to any sort of
geometric recreation having pieces that come apart and fit back together.
We use it as an assembly. He describes how to produce the actual pieces
out of wood~\cite{c-gpd-06}, and suggests that they are made very accurately.
He observes that finding the solution requires dexterity and patience,
when the pieces are accurately made with a tight fit. Even though the
assembly seems relatively simple, this should come as little surprise,
since the first partitioning motion is one out of only eight possible
translations of four symmetric pairs of motions in opposite directions
associated with two complementing subassemblies of three parts each.
Evidently, any automatic process that introduces even the slightest error
along the way, will most likely fail to compute the correct first motion
in the sequence, and falsely claim that the assembly is interlocked.

\begin{wrapfigure}[9]{l}{5.5cm}
  \begin{tabular}{cc}
    \multirow{2}*[0pt]{\epsfig{figure=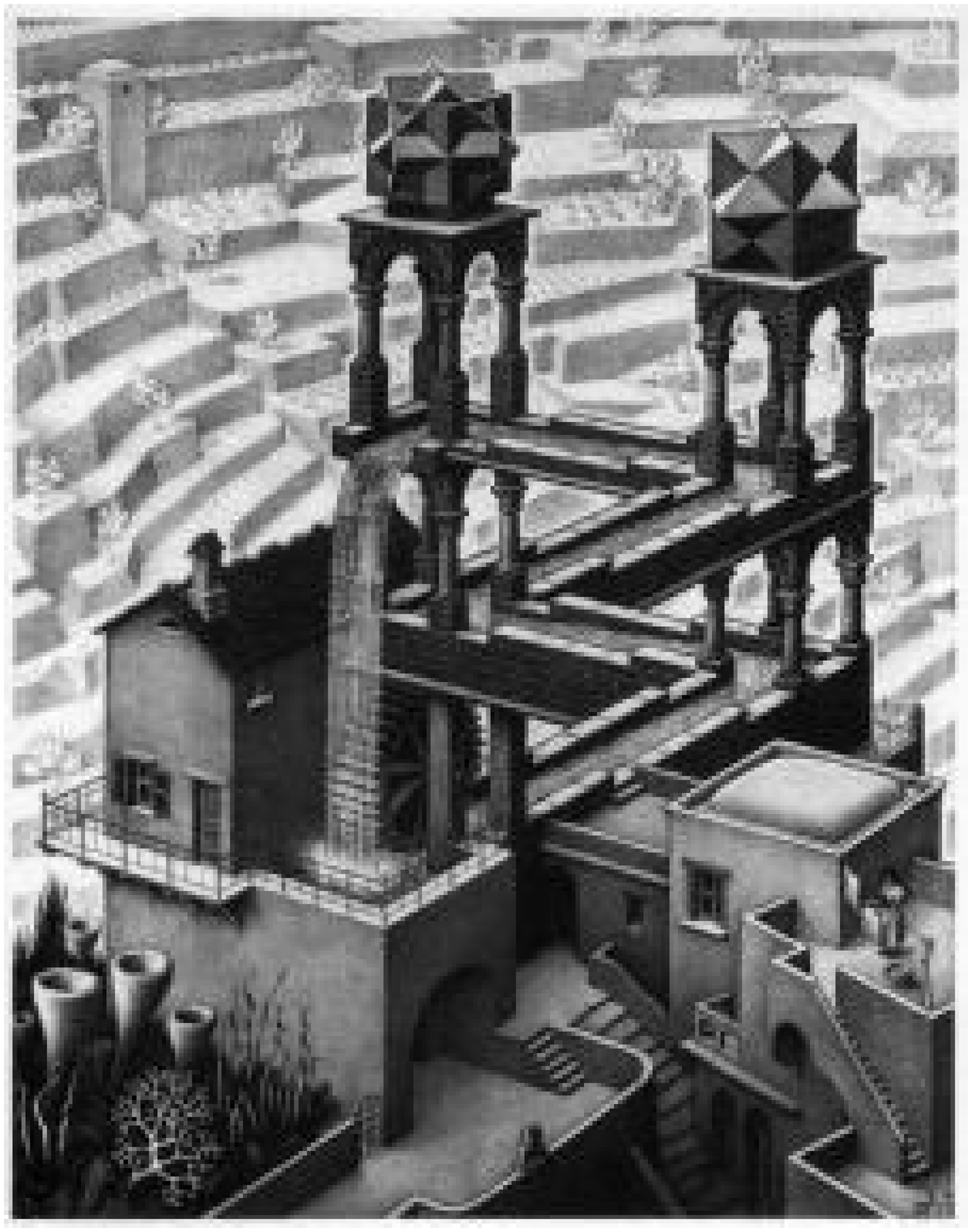,width=2.9cm,silent=}}\vspace{-18pt}\\
    & \epsfig{figure=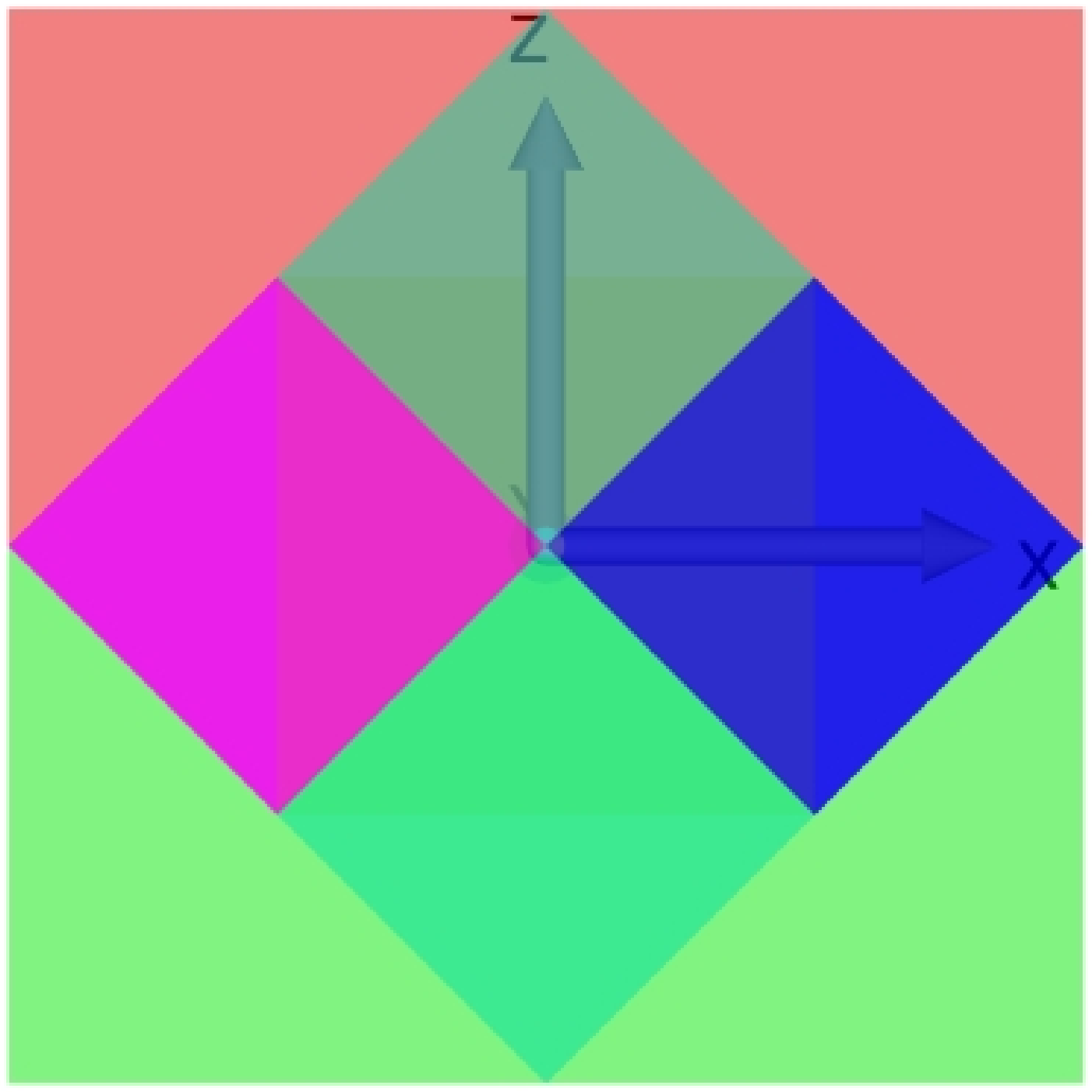,height=2cm,silent=}\vspace{-8pt}\\
    & \epsfig{figure=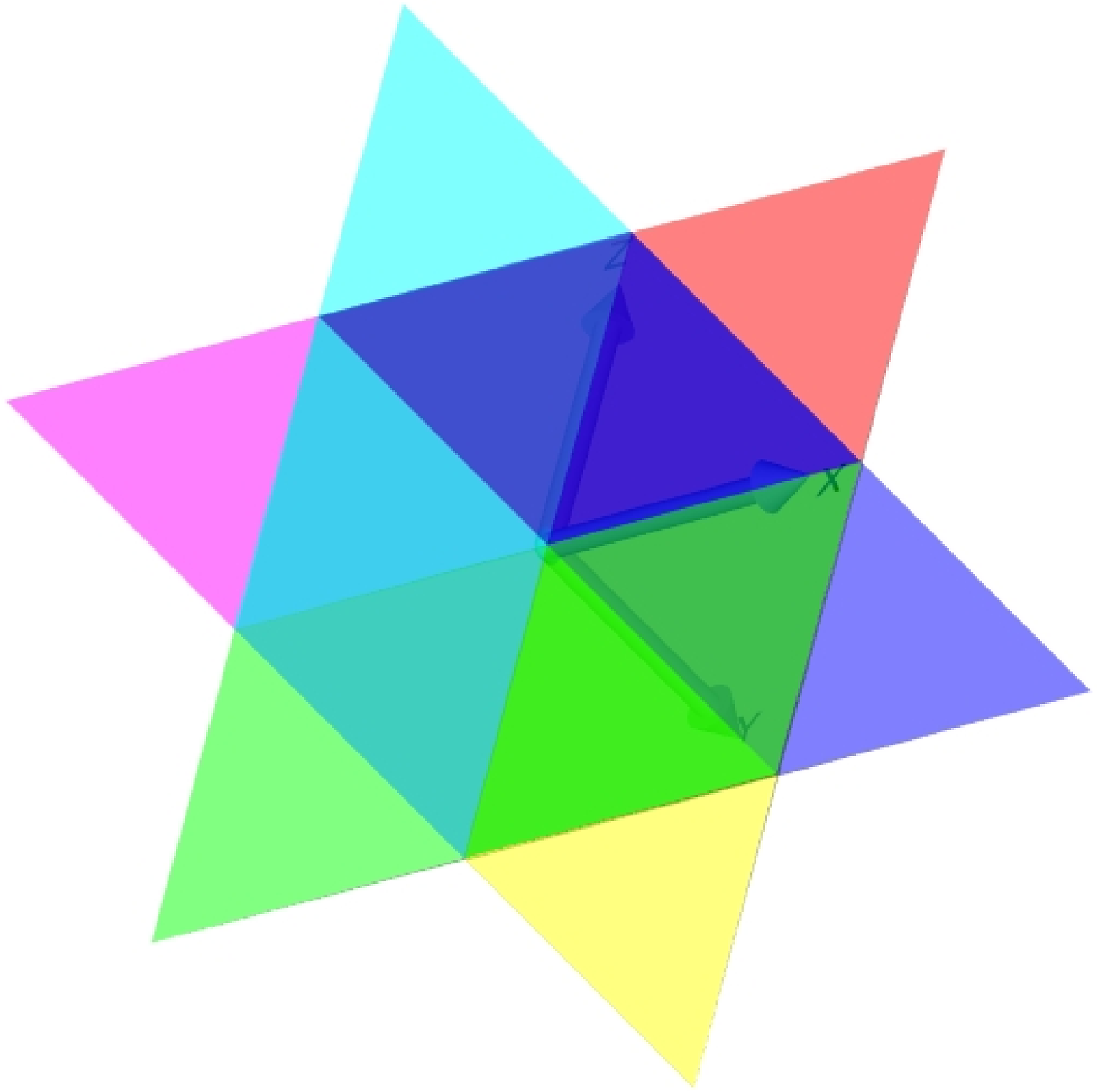,height=2cm,silent=}
  \end{tabular}
\end{wrapfigure}
The \Index{Split Star} assembly has the assembled shape of the first
stellation of the rhombic
dodecahedron\index{dodecahedron!rhombic}~\cite{l-srd-57}, illustrated 
atop the right pedestal in M.~C.~Escher's Waterfall
woodcut~\cite{b-ehlcg-82}. Its orthogonal projection along one of its
fourfold axes of symmetry is a square, while the Star of David is
obtained when it is projected along one of its threefold axes of
symmetry, as seen on the left; for more details see~\cite{c-gpd-06}.
The assembly is a space-filling solid when assembled. It consists of
six identical concave parts. Each part can be decomposed into eight
identical tetrahedra yielding 48 tetrahedra in total. As manufacturing
the pieces requires extreme precision, it is suggested to produce the
48 identical pieces and glue them as necessary. Each part can also be
decomposed into three convex polytopes --- two square pyramids and one
octahedron, yielding 18 polytopes in total. The partitioning described
in this chapter requires the decomposition of the parts into convex
pieces. The choice of decomposition may have a drastic impact on the
time consumption of the entire process, as observed in a different
study in $\rr$~\cite{afh-pdecm-02}, and shown by experiments in
Section~\ref{sec:assem_plan:experiments}.

\subsection{Chapter Outline}
\label{ssec:assem_plan:outline}
The rest of this chapter is organized as follows. The partitioning
algorithm is described in Section~\ref{sec:assem_plan:algorithm} along
with the necessary terms and definitions. In
Section~\ref{sec:assem_plan:implementation} we provide implementation
details. Section~\ref{sec:assem_plan:optimization} presents
optimizations that are not discussed in the preceding sections, some
of which we have already implemented and proved to be useful. We
report on experimental results in Section
\ref{sec:assem_plan:experiments}.

\section{The Partitioning Algorithm}
\label{sec:assem_plan:algorithm}
The main problem we address in this chapter, namely,
{\em polyhedral assembly partitioning with infinite translations},
is formally defined as follows: Let $A = \{P_1, P_2, \ldots, P_n\}$ 
be a set of $n$ pairwise interior disjoint polyhedra in $\rrr$.
$A$ denotes the assembly that we aim to partition. We
look for a proper subset $S\subset A$ and a direction
$\vecd$ in $\rrr$, such that $S$ can be moved as a rigid
body to infinity along $\vecd$ without colliding with the
rest of the parts of the assembly $A\setminus S$. (We allow
sliding motion of one part over the other. We disallow the
intersection of the interior of two polyhedra.)

We follow the NDBG approach~\cite{wl-grma-94}, and describe it
here using the general formulation and notation of
\cite{hlw-gfapm-00}. The {\em \Index{motion space}} in our case,
namely the space of all possible partitioning directions, is
represented by the unit sphere $\spheretwo$. Every point $p$ on
$\spheretwo$ defines the direction from the center of $\spheretwo$
towards $p$. Every direction $\vecd$ is associated with the directed
graph $DBG(\vecd) = (V,E)$ that encodes the blocking relations between
the parts in $A$ when moved along $\vecd$ as follows: The nodes in $V$
correspond to polyhedra in $A$; we denote a node corresponding to the
polyhedron $P_i$ by $v(P_i) \in V$. There is an edge directed from
$v(P_i)$ to $v(P_j)$, denoted $e(P_i,P_j) \in E$, if and only if the
interior of the polyhedron $P_i$ intersects the interior of the
polyhedron $P_j$ when $P_i$ is moved to infinity along the direction
$\vecd$, and $P_j$ remains stationary.

The key idea behind the NDBG approach is that in problems such as ours,
where the number of parts is finite, and any allowable partitioning
motion can be described by a small number of parameters, there is only
a relatively small (polynomial) number of distinct DBGs that need to
be constructed in order to detect a possible partitioning direction.
Stated differently, the motion space can be represented by an arrangement
comprising a finite number of cells each assigned with a fixed DBG.
Once this arrangement is constructed, we construct the DBG over each cell
of the arrangement, and check it for strong connectivity. A DBG that is
not strongly connected is associated with a direction, or a set of
directions in case the cell is not a singular point, that partition the
given assembly. The desired movable subset $S\subset A$ is a byproduct
of the algorithm that checks for strong connectivity. If all the DBGs
over all the cells of the arrangement are strongly connected, 
we conclude that the assembly is {\em interlocked}, as a subset of the
parts in $A$ that can be separated from the rest of the assembly by an
infinite translation does not exist.

Next we show how to construct the motion-space arrangement and compute
the DBG over each one of the arrangement cells. Each ordered pair of
distinct polyhedra $<\!P_i,P_j\!>$ defines a region $Q_{ij}$ on
$\spheretwo$, which is the union of all the directions $\vecd$,
such that when $P_i$ is moved along $\vecd$ its interior will
intersect the interior of $P_j$. How can we effectively compute
this region? Let $M_{ij}$ denote the Minkowski sum
$P_j\oplus(-P_i) = \{b-a\,|\,a\in P_i,\, b\in P_j\}$. We claim that
the \Index{central projection} of $M_{ij}$ onto $\spheretwo$ is exactly $Q_{ij}$.

\begin{lemma}
A direction $\vecd$ is in the interior of the central projection of
$M_{ij}$ onto $\spheretwo$ if and only if when $P_i$ is moved along
$\vecd$ its interior will intersect the interior of $P_j$.
\end{lemma}
\begin{proof}
Let $\vecd$ be some direction in the central projection of $M_{ij}$ onto
$\spheretwo$. In other words, there exists a point $m \in M_{ij}$, such
that $m = s \cdot \vecd$, for some positive scalar $s$. As $m$ is in
$M_{ij}$, there exist two points $p_i \in P_i$ and $p_j \in P_j$, such
that $m = p_j - p_i$. Thus, $p_j = p_i + s \cdot \vecd$, meaning that the
point $p_i$ intersects $p_j$ when moved along $\vecd$. A similar argument
can be used to show the converse.
\end{proof}

Next, we describe how, given two polyhedra $P_i$ and $P_j$, we
compute the region $Q_{ij}$, using robust and efficient
hierarchy of building blocks, which we have developed in recent years.
The existing tools that we use are 
(i) computing the arrangement of spherical polygons
\cite{bfhmw-scmtd-07,fsh-agas-08,fsh-eiaga-08,wfzh-aptac-07}, and
(ii) construction of Minkowski sums of convex polytopes
\cite{bfhmw-scmtd-07,fh-eecms-07,fsh-agas-08,fsh-eiaga-08}.
We also need some extra machinery, as explained below.

Assume $P_i$ is given as the union of a collection of (not
necessarily disjoint) convex polytopes $P^i_1,P^i_2, \ldots,P^i_{m_i}$,
and similarly $P_j$ is given as the collection of convex polytopes
$P^j_1,P^j_2,\ldots,P^j_{m_j}$. It is easily verified that
$M_{ij} = \bigcup_{k=1,\ldots,m_i, \ell = 1,\ldots,m_j}
P^j_\ell \oplus(-P^i_k)$.
So we compute the Minkowski sum of each pair $P^j_\ell\oplus(-P^i_k)$,
and centrally project it onto $\spheretwo$. Finally, we take the union
of all these projections to yield $Q_{ij}$.

There are several ways to effectively compute the central projection
of a convex polyhedron $C$ (one of the polytopes
$M^{ij}_{k\ell} = P^j_\ell\oplus(-P^i_k)$) from the origin onto $\spheretwo$.
We opted for the following. An edge $e$ of $C$ is called a
{\em silhouette edge}, if the plane $\pi$ through the origin and $e$
is tangent to $C$ at $e$; namely, $\pi$ intersects $C$ in
$e$ only. We assume for now that no tangent plane contains a facet of $C$;
we relax this assumption in
Section~\ref{ssec:assem_plan:pairwise-sub-part-ms-projection}, where we
provide a detailed description of the procedure. We traverse the edges of
$C$ till we find a silhouette edge $e_1$. One can verify that the
silhouette edges form a cycle on $C$. We start with $e_1$, and search
for a silhouette edge adjacent to $e_1$. We proceed in the same manner,
till we end up discovering $e_1$ again. Projecting this cycle
of edges onto $\spheretwo$ is straightforward.

All the boundaries of the regions $Q_{ij}$ form an arrangement of
geodesic arcs on the sphere. We traverse the motion-space arrangement
in say a breadth-first fashion. For the first face we check which ones
of the regions $Q_{ij}$ contain it. We construct the corresponding DBG
and check it for strong connectivity. If it is not strongly connected,
we stop and announce a solution as described above. Otherwise we move
to an adjacent feature of the current face. During this move we may
have stepped out from a set of regions $Q_{ij}$, and may have stepped
into a new set of regions $Q_{ij}$. We update the current DBG
according to the regions we left or entered, test the new DBG for
strong connectivity, and so on till the traversal of all the
arrangement cells is completed. Notice that it is important to visit
also vertices and edges of the arrangement, since the solution may not
lie in the interior of a face. Indeed, in our Split Star example,
solutions are on vertices of the arrangement. Without careful exact
constructions, such solutions could easily be missed.

\section{Implementation Details}
\label{sec:assem_plan:implementation}
The implementation of the assembly-partitioning operation consists of
eight phases listed below. They all exploit arrangements of geodesic
arcs embedded on the sphere~\cite{bfhmw-scmtd-07,fsh-eiaga-08} in
various ways. The \aos{} package of \cgal{} already supports the
construction and maintenance of such arrangements, the computation of
union of faces of such arrangements, the construction of Gaussian maps
of polyhedra represented by such arrangements, and the computation of
their Minkowski sums. It provides the ground for efficient and generic
implementation of the remaining required operations, such as central
projection.

\begin{enumerate}
\item \textbf{Convex Decomposition}
\item \textbf{Sub-part Gaussian map construction}
\item \textbf{Sub-part Gaussian map reflection}
\item \textbf{Pairwise sub-part Minkowski sum construction}
\item \textbf{Pairwise sub-part Minkowski sum projection}
\item \textbf{Pairwise Minkowski sum projection}
\item \textbf{Motion-space construction}
\item \textbf{Motion-space processing}
\end{enumerate}

The partitioning process is implemented as a free function that
accepts as input an ordered list of polyhedra in $\rrr$, which
are the parts of the assembly. Each part
is represented as a polyhedral mesh in $\rrr$; see
Section~\ref{ssec:mscn:sgm:representation} for a definition. We proceed
with a detailed discussion of the implementation of each phase.

We deal below with various details that are typically ignored in
reports on geometric algorithms (for example, under the general position
assumption). However, in assembly planning, or more generally in
movable-separability problems~\cite{t-mss-85} in tight scenarios,
much of the difficulty shifts exactly to these technical details and
in particular to handling degeneracies. This is especially emphasized
in Phases~5 and~6
(Subsections~\ref{ssec:assem_plan:pairwise-sub-part-ms-projection}
and~\ref{ssec:assem_plan:pairwise-ms-projection} respectively), but
prevails throughout the entire section.

\subsection{Convex Decomposition}
\label{ssec:assem_plan:convex-decomposition}
We decompose each concave part into convex polyhedra referred to as
sub-parts. The output of this phase is an ordered list of parts, where
each part is an ordered list of convex sub-parts represented as
polyhedral surfaces. Each polyhedral surface is maintained as a 
\cgal{} \ccode{Polyhedron\_3}~\cite{cgal:k-ps-07} data structure,
which consists of vertices, edges, and facets and incidence relations
on them~\cite{k-ugpdd-99}. A part that is convex to start with is
simply converted into an object of type \ccode{Polyhedron\_3}.

\begin{figure*}[t]%
  \setlength{\tabcolsep}{3pt}
  \centerline{%
    \begin{tabular}{ccc}
      \epsfig{figure=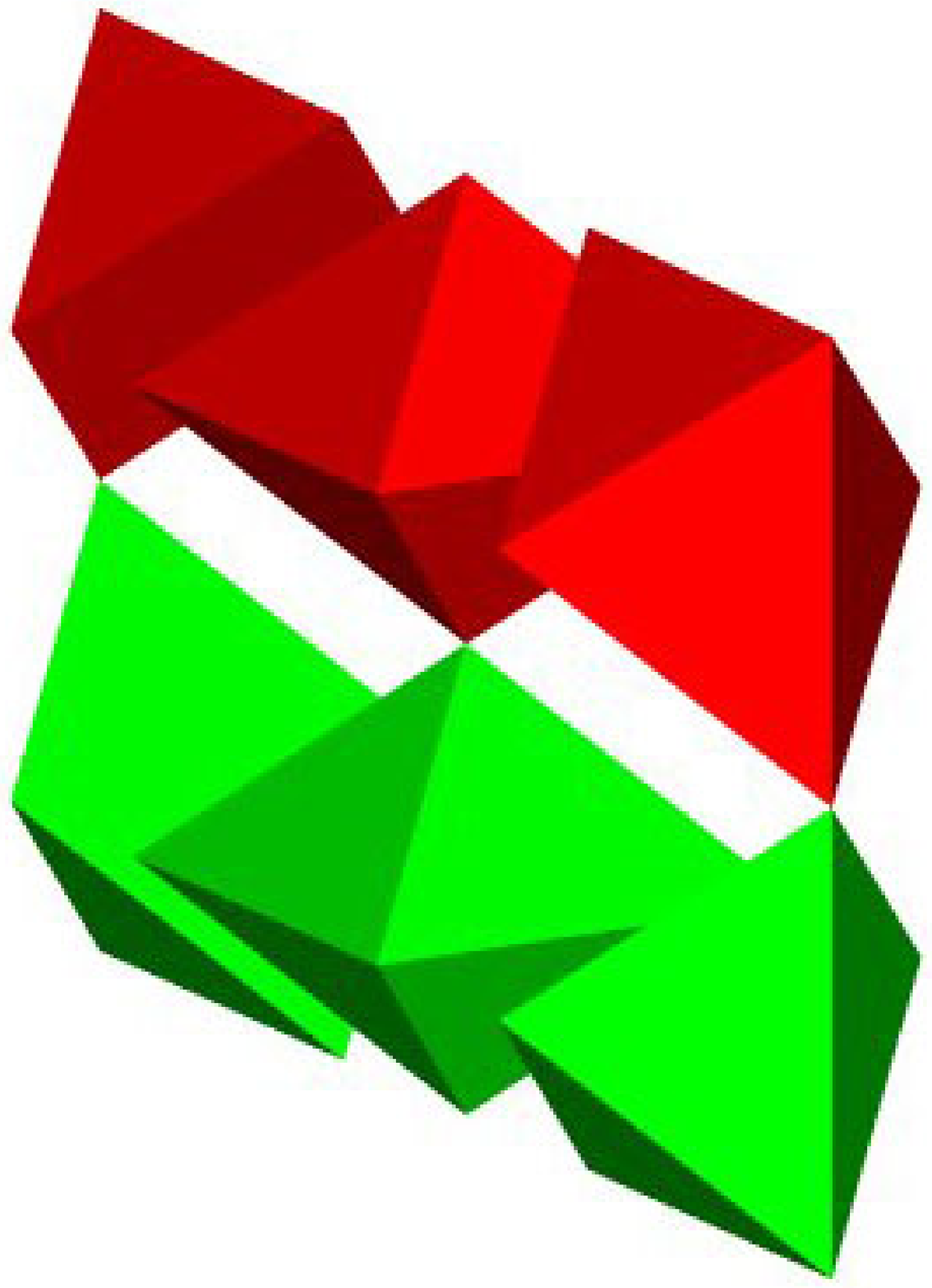,height=5cm,silent=} &
      \epsfig{figure=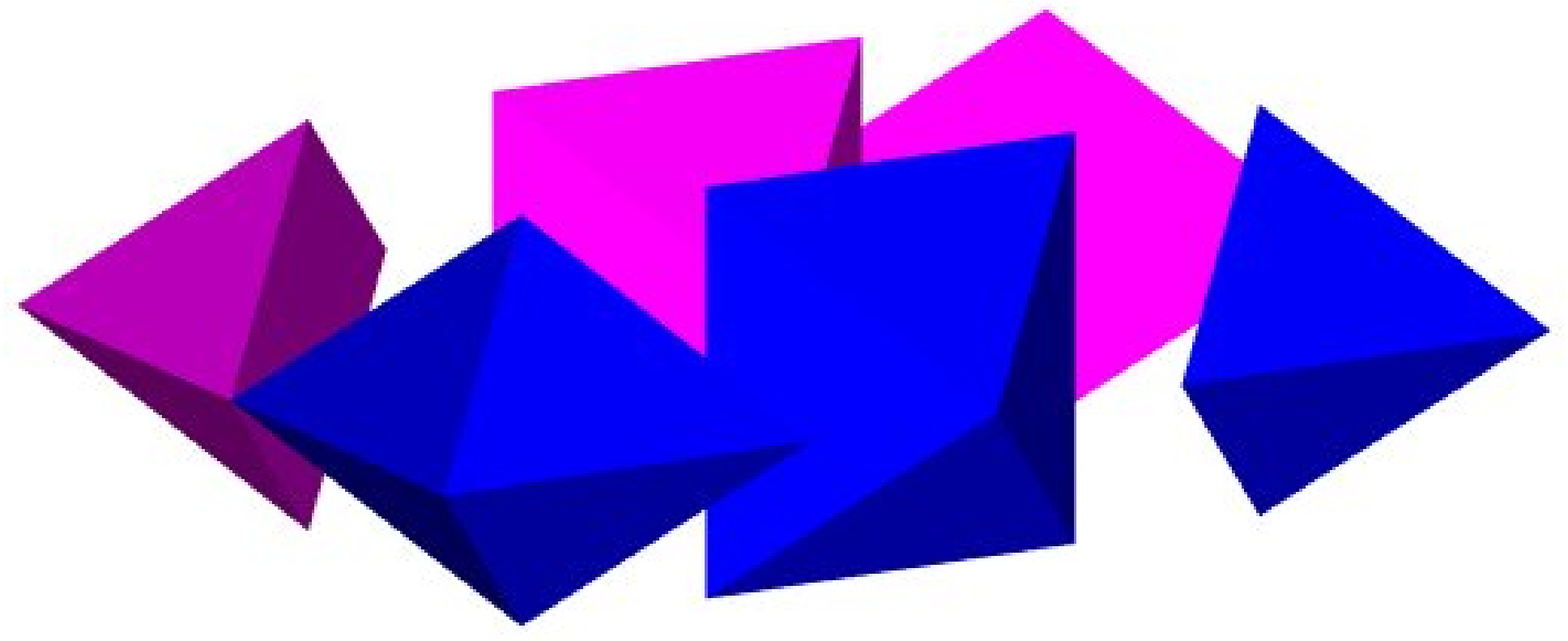,height=5cm,silent=} &
      \epsfig{figure=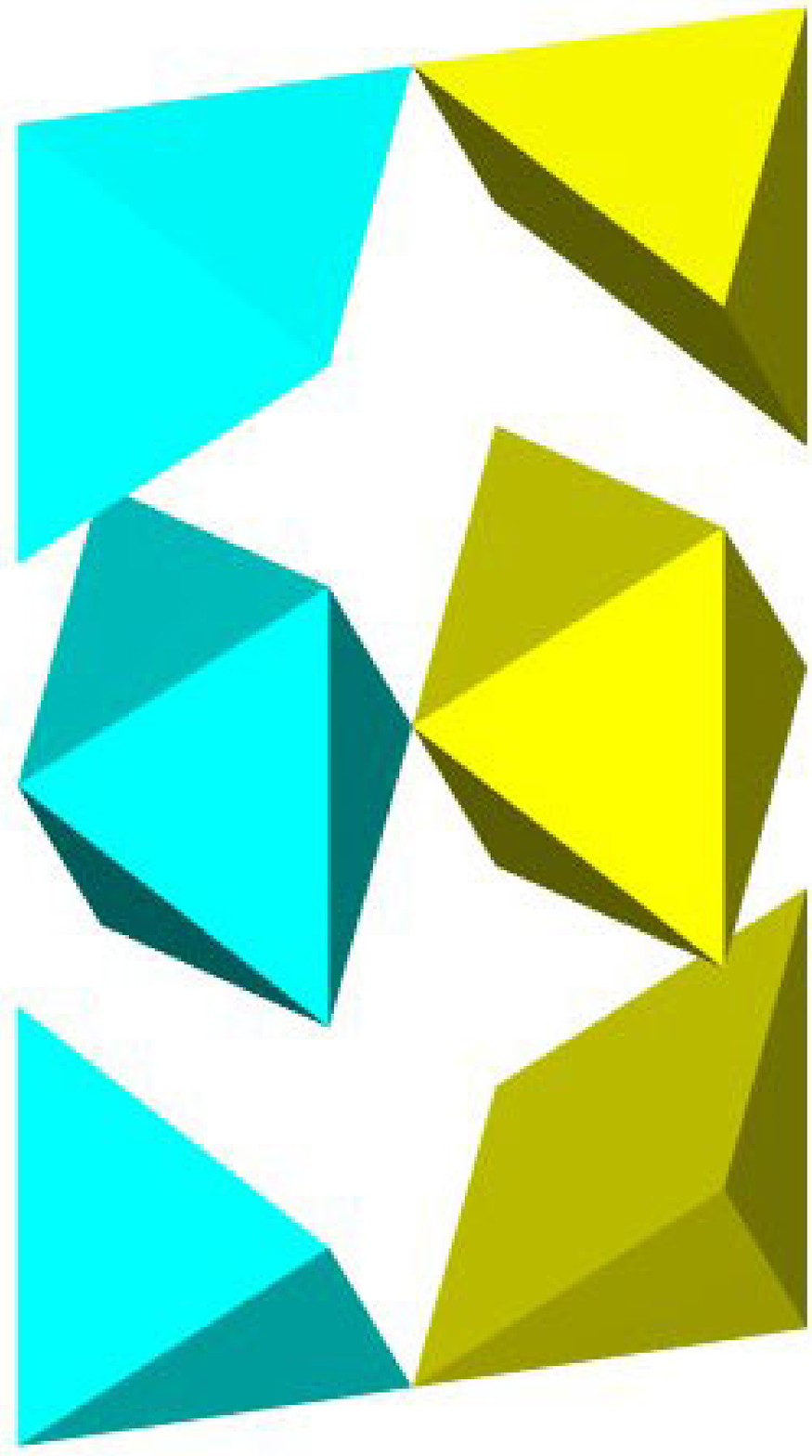,height=5cm,silent=}
    \end{tabular}
  }
  \vspace{-5pt}
  \caption[Decomposition of the Split Star assembly]%
          {\capStyle{The Split Star six parts
           decomposed into three convex sub-parts each.}}
  \label{fig:convex-decomposition}
  \vspace{-5pt}
\end{figure*}
A new package of \cgal{} that supports convex decomposition of
polyhedra has been recently introduced~\cite{h-emspe-07}, but has not
become publicly available yet. As we aim for a fully automatic
process, we intend to exploit such components, once they become
available, and study their impact. For the time being we resorted to a
manual procedure. A simple decomposition of the Split Star parts used
in the running example is illustrated in
Figure~\ref{fig:convex-decomposition}.

\subsection{Sub-part Gaussian Map Construction}
\label{ssec:assem_plan:construction}
We convert each sub-part represented as a polyhedral surface into a
Gaussian map represented as an arrangement of geodesic arcs embedded
on the sphere, where each face $f$ of the arrangement is extended with
the coordinates of its associated primal vertex $v = G^{-1}(f)$,
resulting with a unique representation; see
Section~\ref{sec:mscn:gauss_map} for the exact definition of
Gaussian maps and see Section~\ref{sec:mscn:sgm-method} the exact
procedure to construct one from a polytope.

\begin{figure*}[t]%
  \setlength{\tabcolsep}{3pt}
  \centerline{%
    \begin{tabular}{cc}
      \begin{tabular}{ccc}
        \epsfig{figure=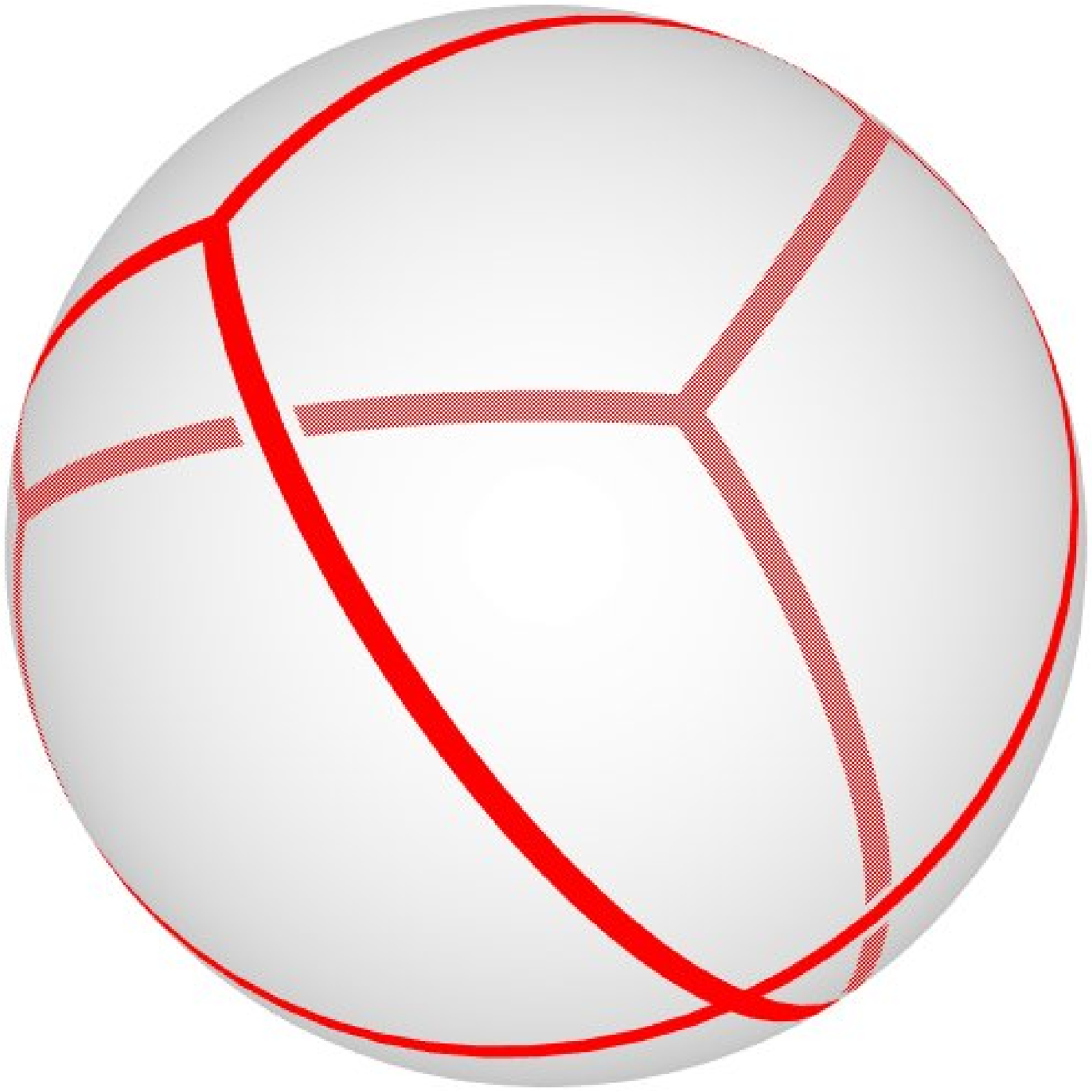,height=1.8cm,silent=} &
        \epsfig{figure=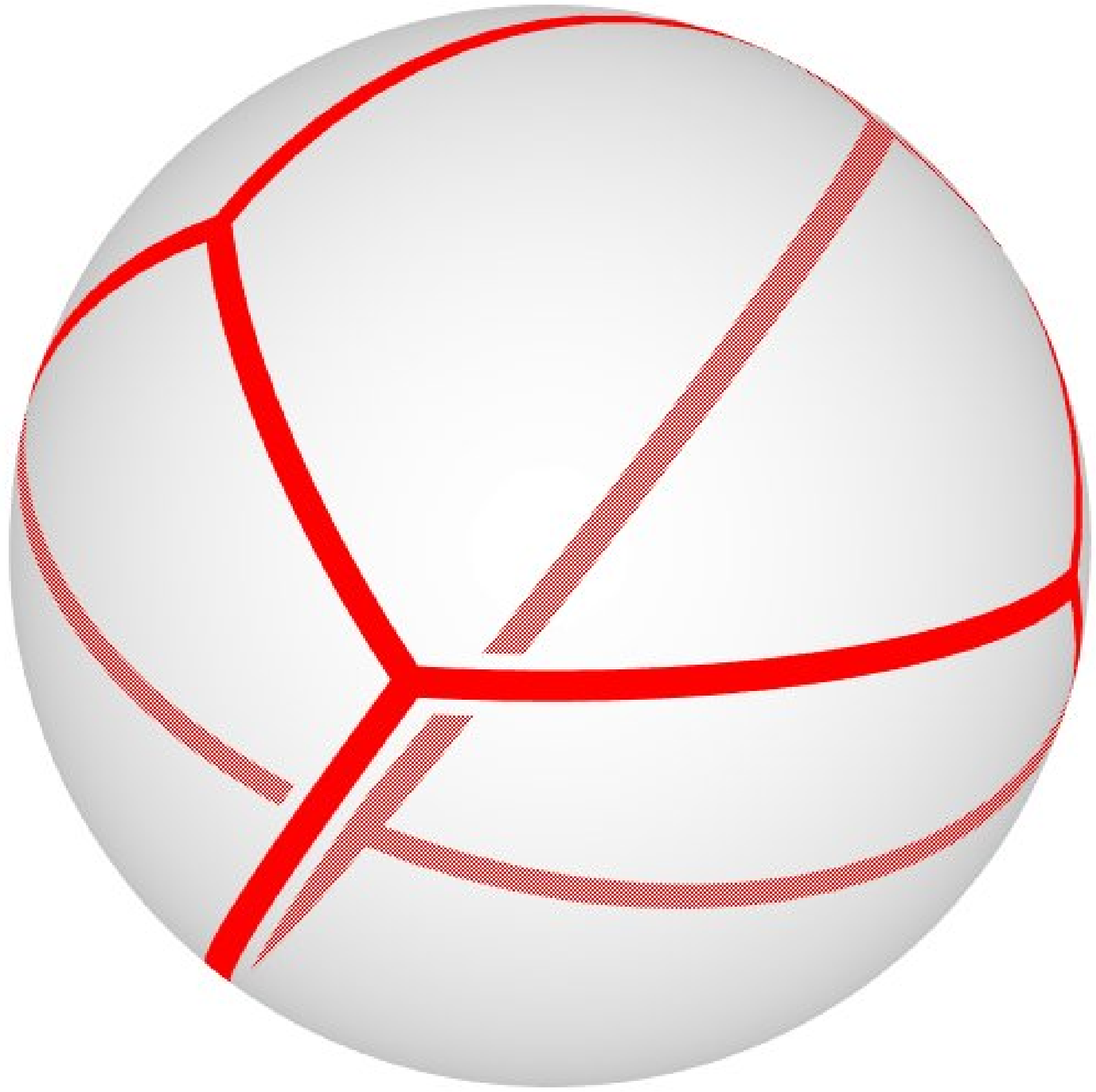,height=1.8cm,silent=} &
        \epsfig{figure=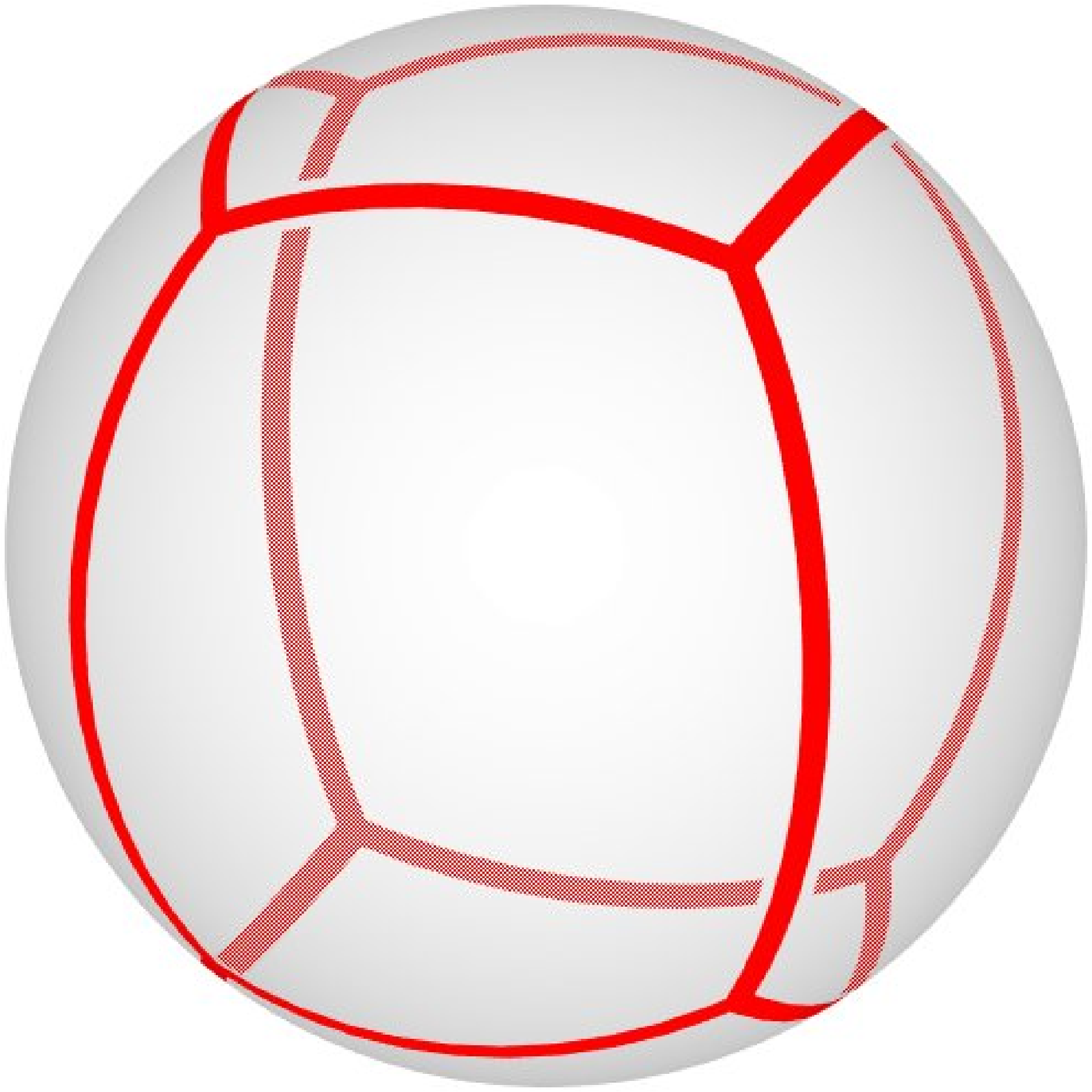,height=1.8cm,silent=}\\
        $R_1$ & $R_2$ & $R_3$
      \end{tabular} &
      \begin{tabular}{ccc}
        \epsfig{figure=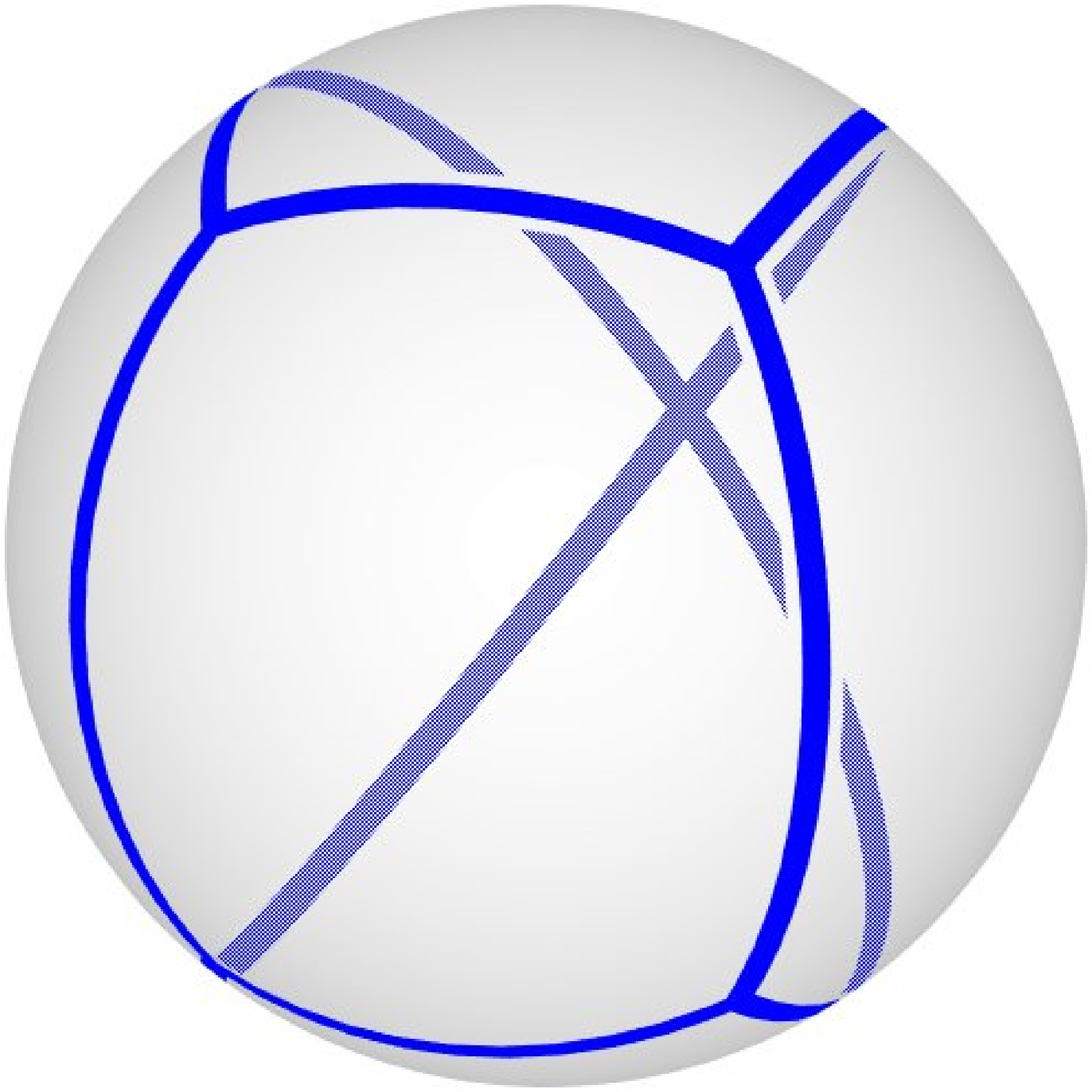,height=1.8cm,silent=} &
        \epsfig{figure=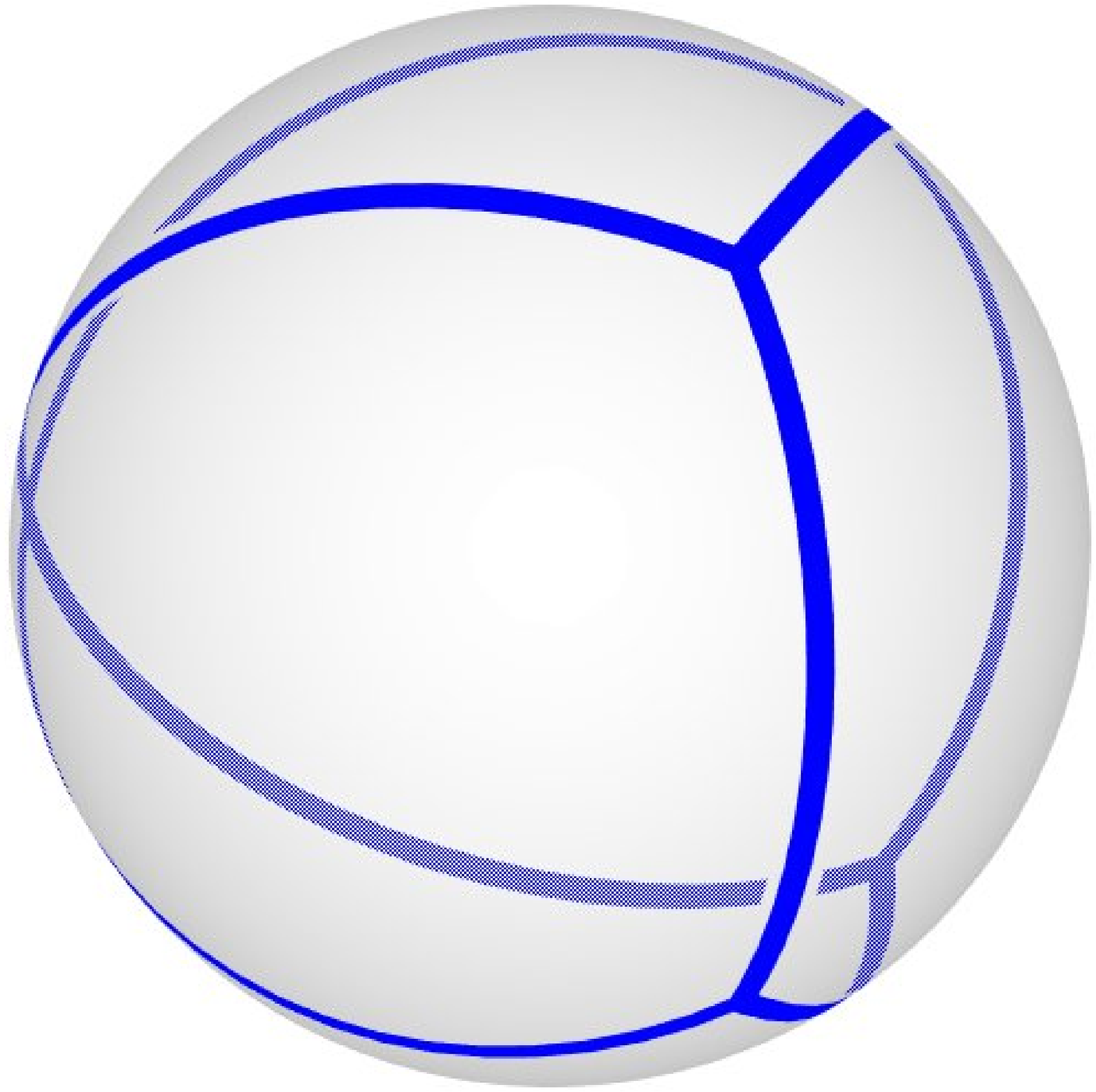,height=1.8cm,silent=} &
        \epsfig{figure=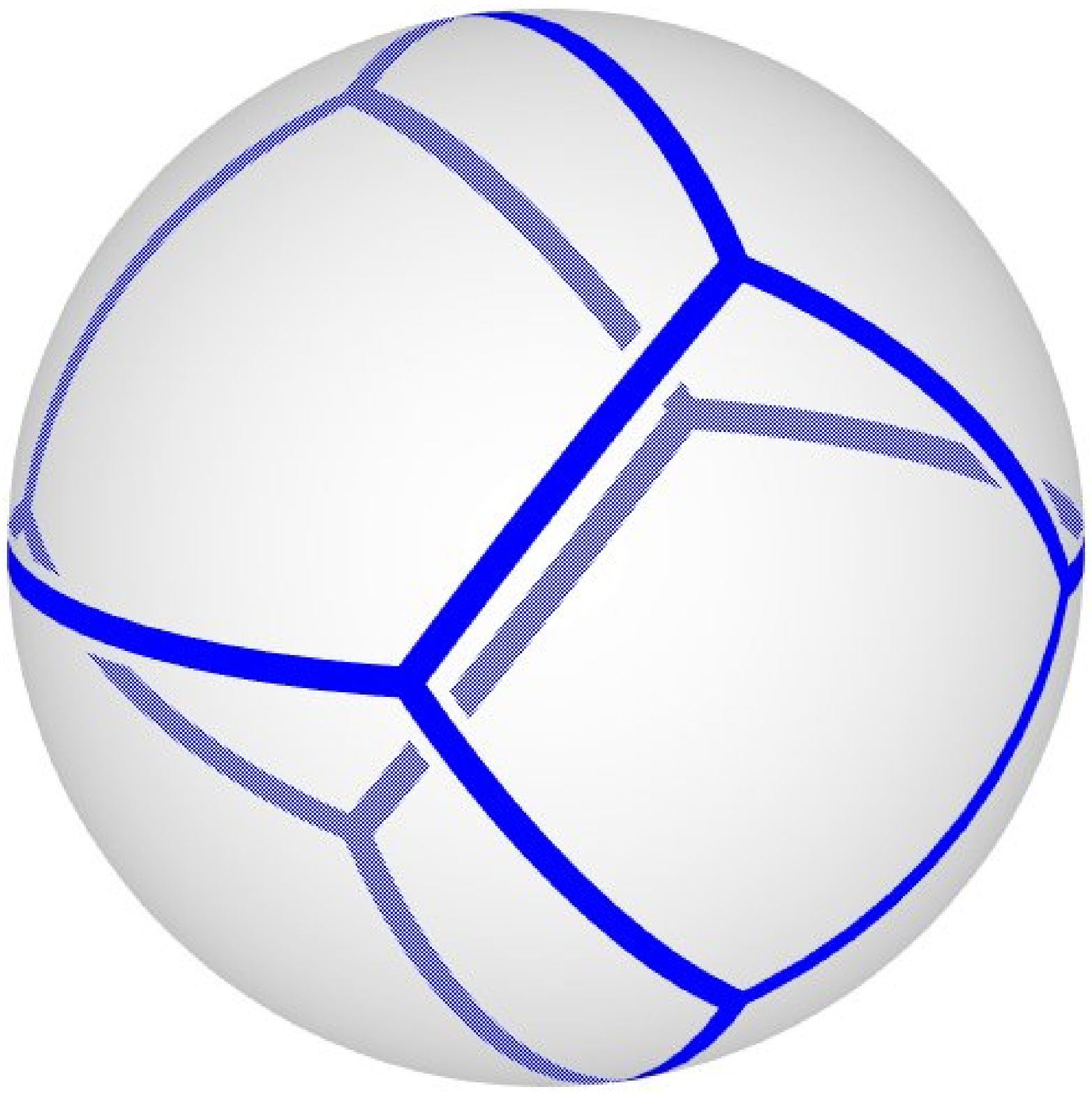,height=1.8cm,silent=}\\
        $B_1$ & $B_2$ & $B_3$
      \end{tabular}\\
      \begin{tabular}{ccc}
        \epsfig{figure=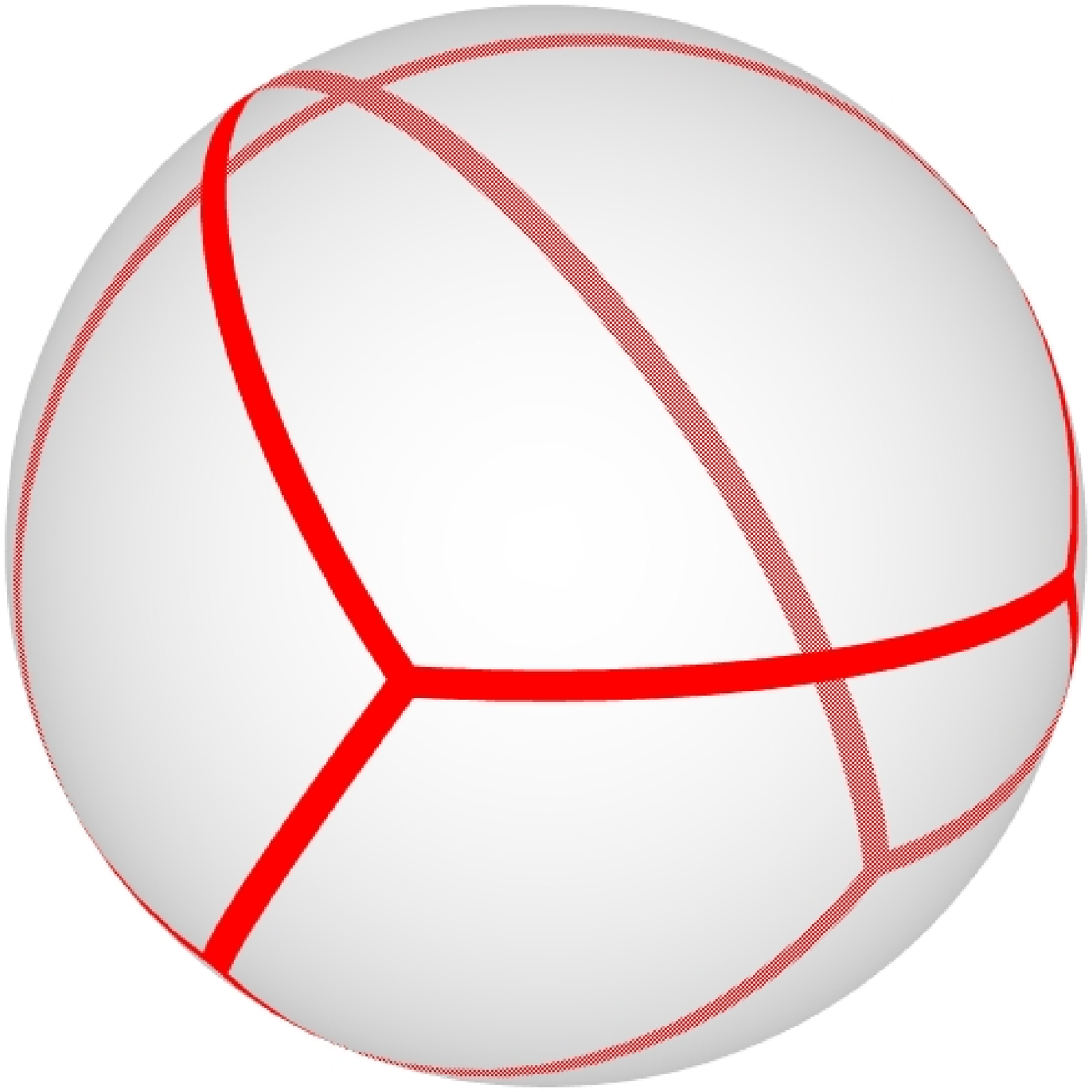,height=1.8cm,silent=} &
        \epsfig{figure=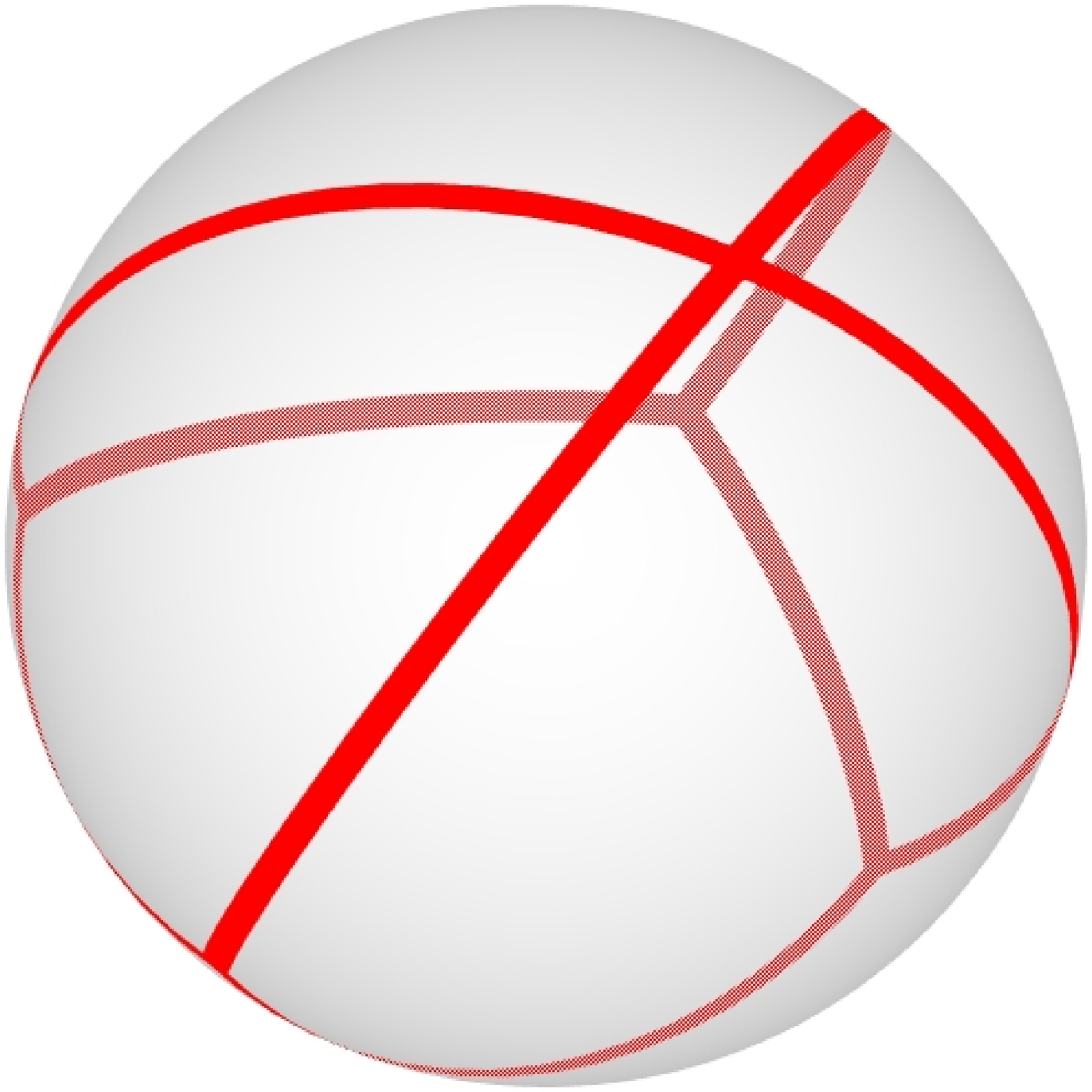,height=1.8cm,silent=} &
        \epsfig{figure=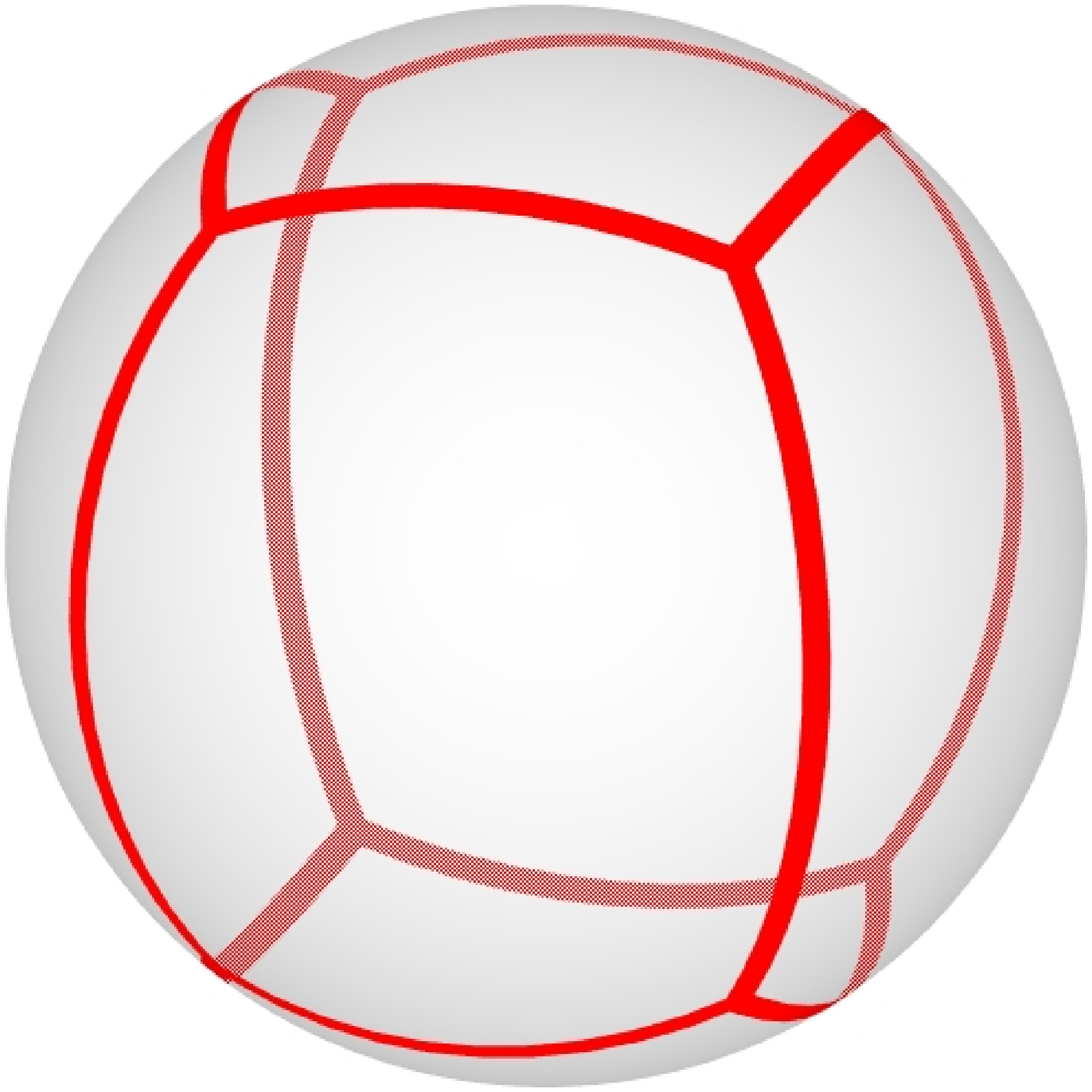,height=1.8cm,silent=}\\
        $-R_1$ & $-R_2$ & $-R_3$
      \end{tabular} &
      \begin{tabular}{ccc}
        \epsfig{figure=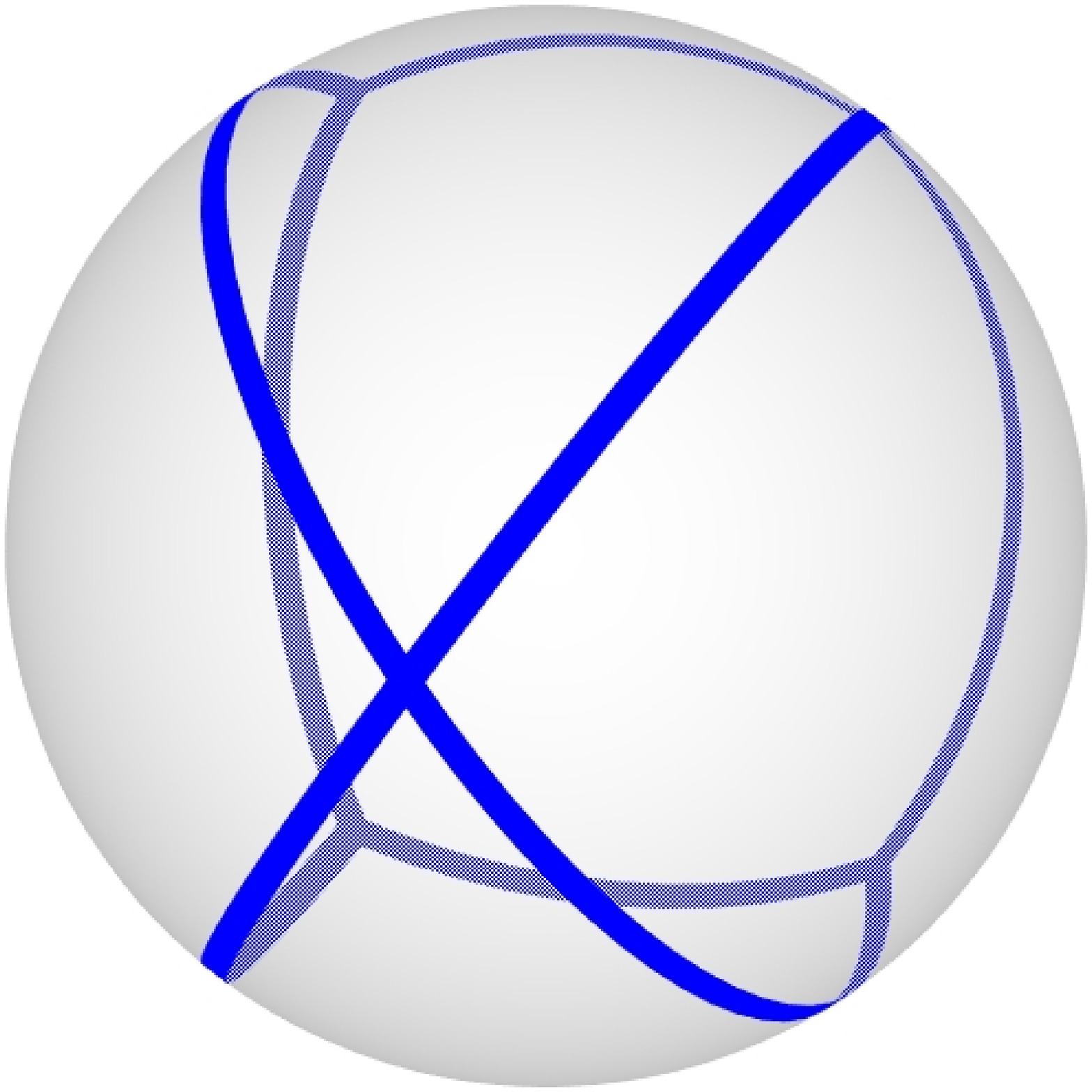,height=1.8cm,silent=} &
        \epsfig{figure=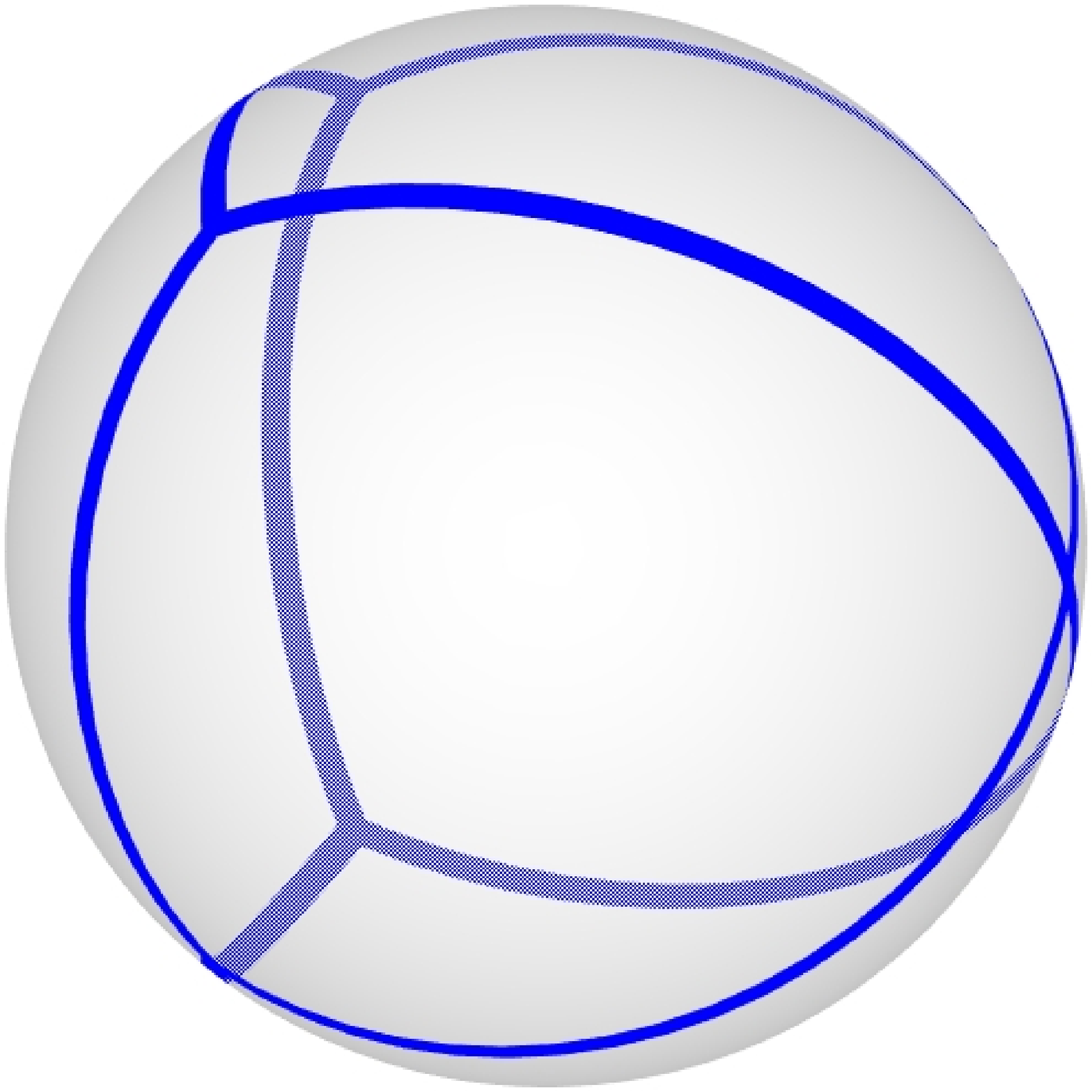,height=1.8cm,silent=} &
        \epsfig{figure=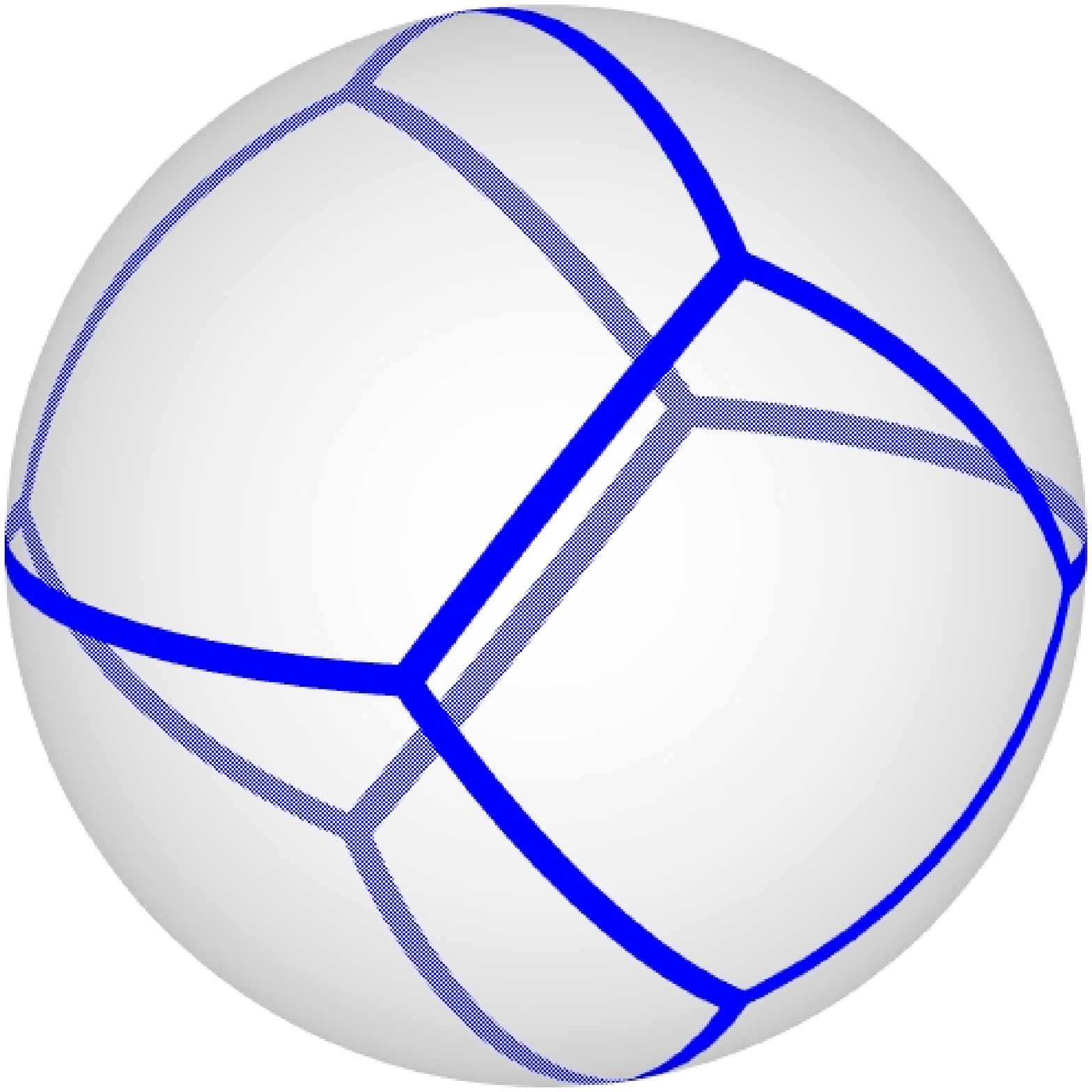,height=1.8cm,silent=}\\
        $-B_1$ & $-B_2$ & $-B_3$
      \end{tabular}
    \end{tabular}
  }
  \vspace{-5pt}
  \caption[Samples of the Gaussian maps of sub-parts of the Split Star
    assembly]{\capStyle{Samples of the Gaussian maps of sub-parts of
      the Split Star assembly. The bottom row contains the reflections
      of the Gaussian maps at the top row.}}
  \label{fig:assem-gaussian-map}
  \vspace{-5pt}
\end{figure*}
The output of this phase is an ordered list of parts, where each part
is an ordered list of the Gaussian maps of the convex
sub-parts. Figure~\ref{fig:assem-gaussian-map} depicts the Gaussian maps of
six of the 18 polytopes that comprise the set of sub-parts of our
Split Star assembly.

\subsection{Sub-part Gaussian Map Reflection}
\label{ssec:assem_plan:sub-part-gaussian-map-reflection}
We reflect each sub-part $P^i_k$ through the origin to obtain $-P^i_k$.
This operation can be performed directly on the output of the previous
phase, reflecting the underlying arrangements of geodesic arcs embedded
on the sphere, which represent the Gaussian maps, while handling the
additional data attached to the arrangement faces. As a matter of fact,
this phase can be reduced as part of an optimization
discussed in Section~\ref{sec:assem_plan:optimization}.

The output of this phase is an ordered list of parts, where each part is
an ordered list of Gaussian maps of the reflected convex sub-parts.
Figure~\ref{fig:assem-gaussian-map} depicts the Gaussian maps of six of the
18 polytopes that comprise the set of reflected sub-parts of the Split
Star example.
\subsection{Pairwise Sub-part Minkowski Sum Construction}
\label{ssec:assem_plan:pairwise-sub-part-ms-construction}
\begin{wraptable}{l}{6.8cm}
  \vspace{-15pt}
  \begin{tabularx}{6.5cm}{l@{}p{2ex}p{2ex}p{2ex}p{2ex}X}
    \hline
    \hline
    \multicolumn{6}{l}{Construct Pairwise Sub-part}\\
    \multicolumn{6}{r}{Minkowski Sums}\\
    \hline
    & \multicolumn{5}{l}{\textbf{for} $i$ \textbf{in} $\{1,2,\ldots,n\}$}\\
    & & \multicolumn{4}{l}{\textbf{for} $j$ \textbf{in} $\{1,2,\ldots,n\}$}\\
    & & & \multicolumn{3}{l}{\textbf{if} $i == j$ \textbf{continue}}\\
    & & & \multicolumn{3}{l}{\textbf{for} $k$ \textbf{in} $\{1,2,\ldots,m_i\}$}\\
    & & & & \multicolumn{2}{l}{\textbf{for} $\ell$ \textbf{in} $\{1,2,\ldots,m_j\}$}\\
    & & & & & \multicolumn{1}{l}{$M^{ij}_{k\ell} = P^j_\ell \oplus (-P^i_k)$}\\
    \hline
  \end{tabularx}
  \vspace{-15pt}
\end{wraptable}
We compute the Minkowski sums of the pairwise sub-parts and reflected
sub-parts. Aiming for an efficient output sensitive algorithm, the
construction of an individual Minkowski sum from two Gaussian maps
represented as two arrangements respectively is performed by
overlaying the two arrangements. When the \Index{overlay} operation progresses,
new vertices, edges, and faces of the resulting arrangement are created
based on features of the two operands. When a new feature is created its
attributes are updated. There are ten cases that must be handled; see
Sections~\ref{ssec:mscn:sgm:mink_sum} for details. The \aos{} package
conveniently supports the overlay operation allowing users to provide
their own version of these ten operations. The overlay operation is
exploited below in several different variants of arrangements of geodesic
arcs embedded on the sphere. Each application requires the provision of a
different set of those ten operations.

The output of this phase is a map from ordered pairs of distinct indices
into lists of Minkowski sums represented as Gaussian maps. Each ordered
pair $<\!i,j\!>, i \neq j$ is associated with the list of Minkowski
sums $\{M^{ij}_{k\ell}\,|\,k = 1,2,\ldots,m_i,\,\ell = 1,2,\ldots,m_j\}$.
In case of our Split Star the map consists of 30 entries that correspond
to all configurations of ordered distinct pairs of parts.
Each entry consists of a list of nine Minkowski sums, that is, 270
Minkowski sums in total.

\subsection{Pairwise Sub-part Minkowski Sum Projection}
\label{ssec:assem_plan:pairwise-sub-part-ms-projection}
\begin{wraptable}{l}{6.8cm}
  \vspace{-15pt}
  \begin{tabularx}{6.5cm}{l@{}p{2ex}p{2ex}p{2ex}p{2ex}X}
    \hline
    \hline
    \multicolumn{6}{l}{Project Pairwise Sub-part}\\
    \multicolumn{6}{r}{Minkowski sums}\\
    \hline
    & \multicolumn{5}{l}{\textbf{for} $i$ \textbf{in} $\{1,2,\ldots,n\}$}\\
    & & \multicolumn{4}{l}{\textbf{for} $j$ \textbf{in} $\{1,2,\ldots,n\}$}\\
    & & & \multicolumn{3}{l}{\textbf{if} $i == j$ \textbf{continue}}\\
    & & & \multicolumn{3}{l}{\textbf{for} $k$ \textbf{in} $\{1,2,\ldots,m_i\}$}\\
    & & & & \multicolumn{2}{l}{\textbf{for} $\ell$ \textbf{in} $\{1,2,\ldots,m_j\}$}\\
    & & & & & \multicolumn{1}{l}{$Q^{ij}_{k\ell} = \text{project}(M^{ij}_{k\ell})$}\\
    \hline
  \end{tabularx}
  \vspace{-15pt}
\end{wraptable}
We centrally project all pairwise sub-part Minkowski sums computed
in the previous phase onto the sphere. Each projection is represented
as an arrangement of geodesic arcs on the sphere, where each cell $c$
of the arrangement is extended with a Boolean flag that indicates whether
all infinite rays emanating from the origin in all directions
$\vecd \in c$ pierce the corresponding Minkowski sum. As the Minkowski
sums are convex, their spherical projections are spherically convex.

\begin{figure*}[t]%
  \setlength{\tabcolsep}{3pt}
  \centerline{%
    \begin{tabular}{cccccc}
      \epsfig{figure=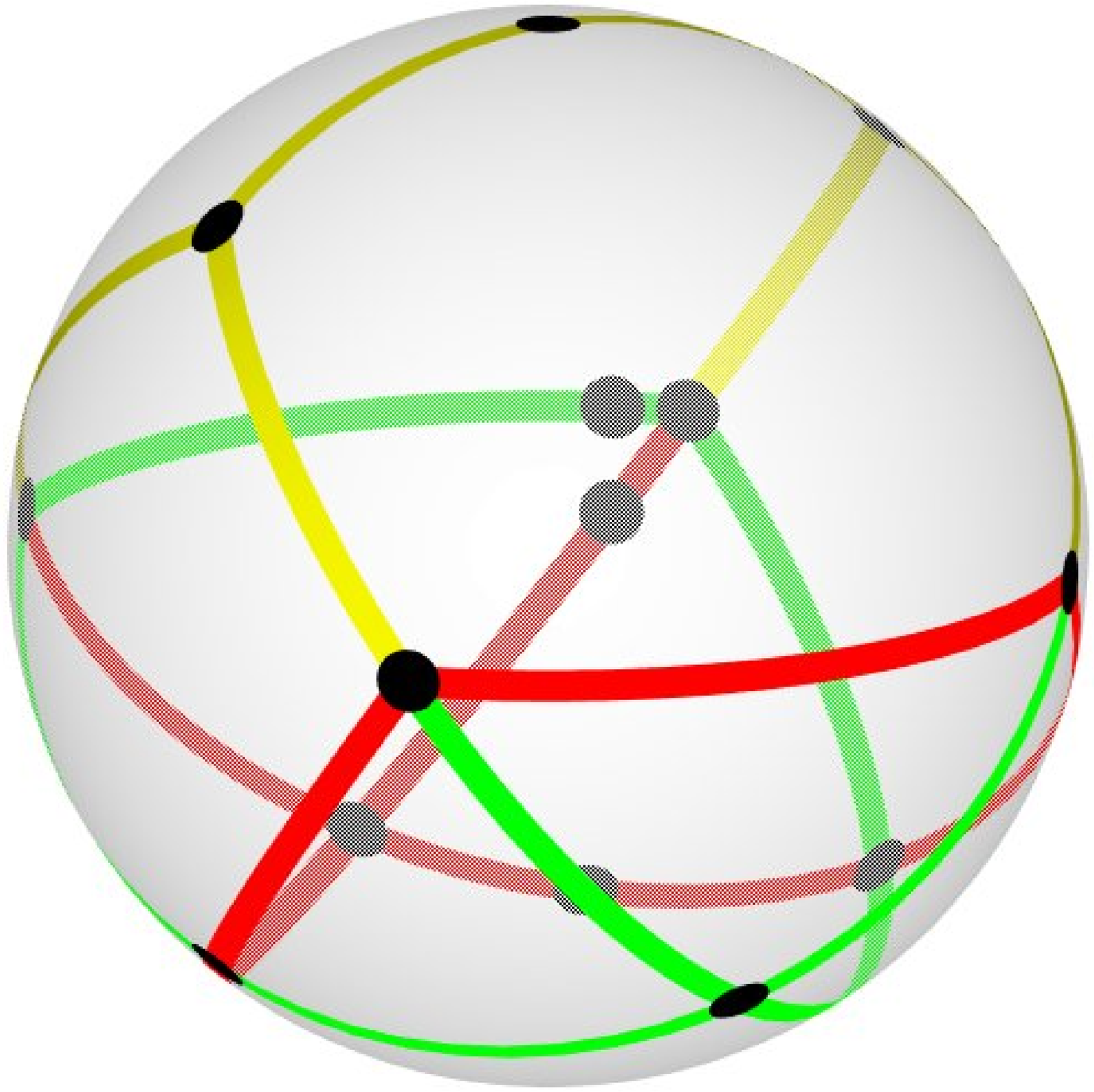,height=1.85cm,silent=} &
      \epsfig{figure=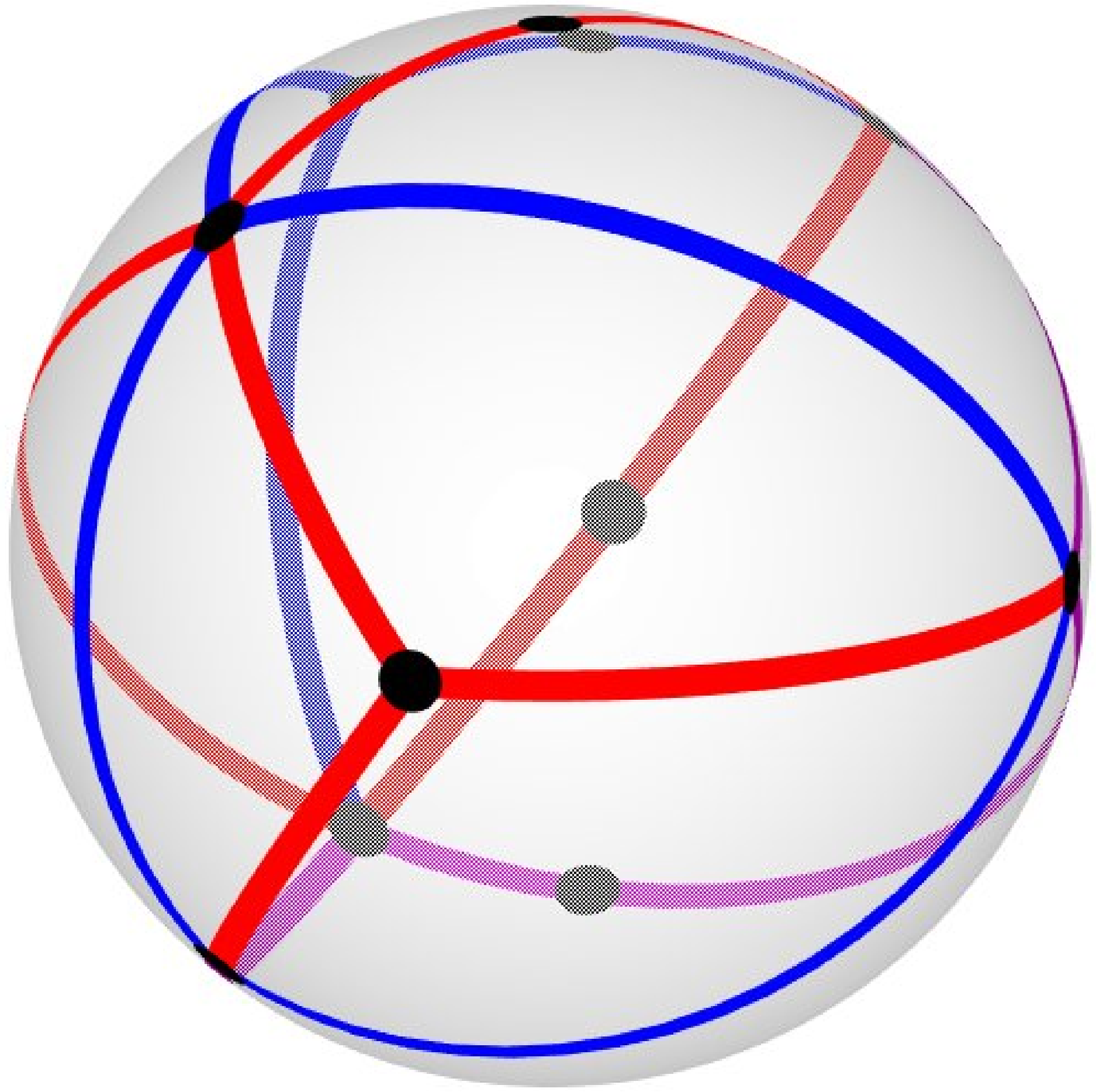,height=1.85cm,silent=} &
      \epsfig{figure=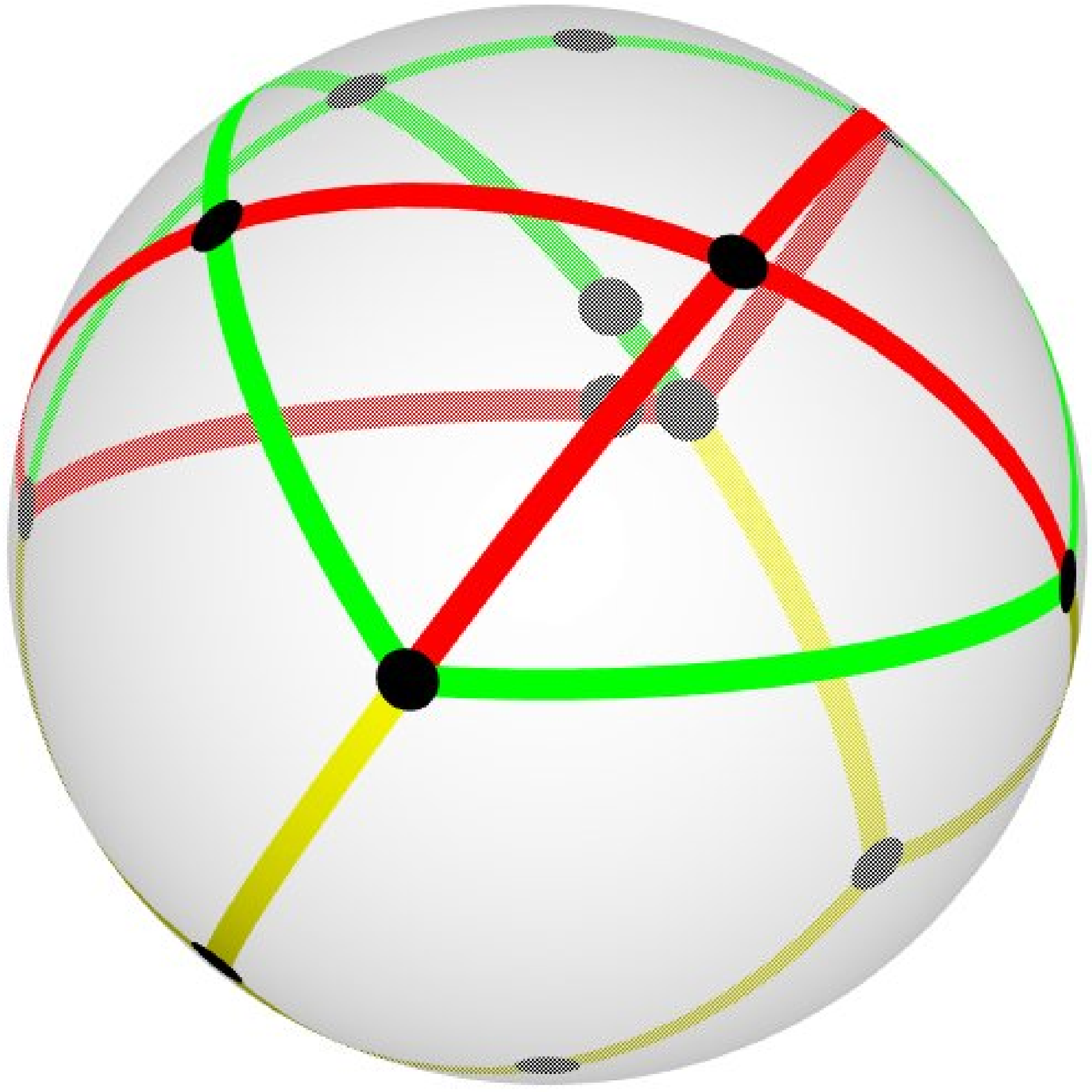,height=1.85cm,silent=} &
      \epsfig{figure=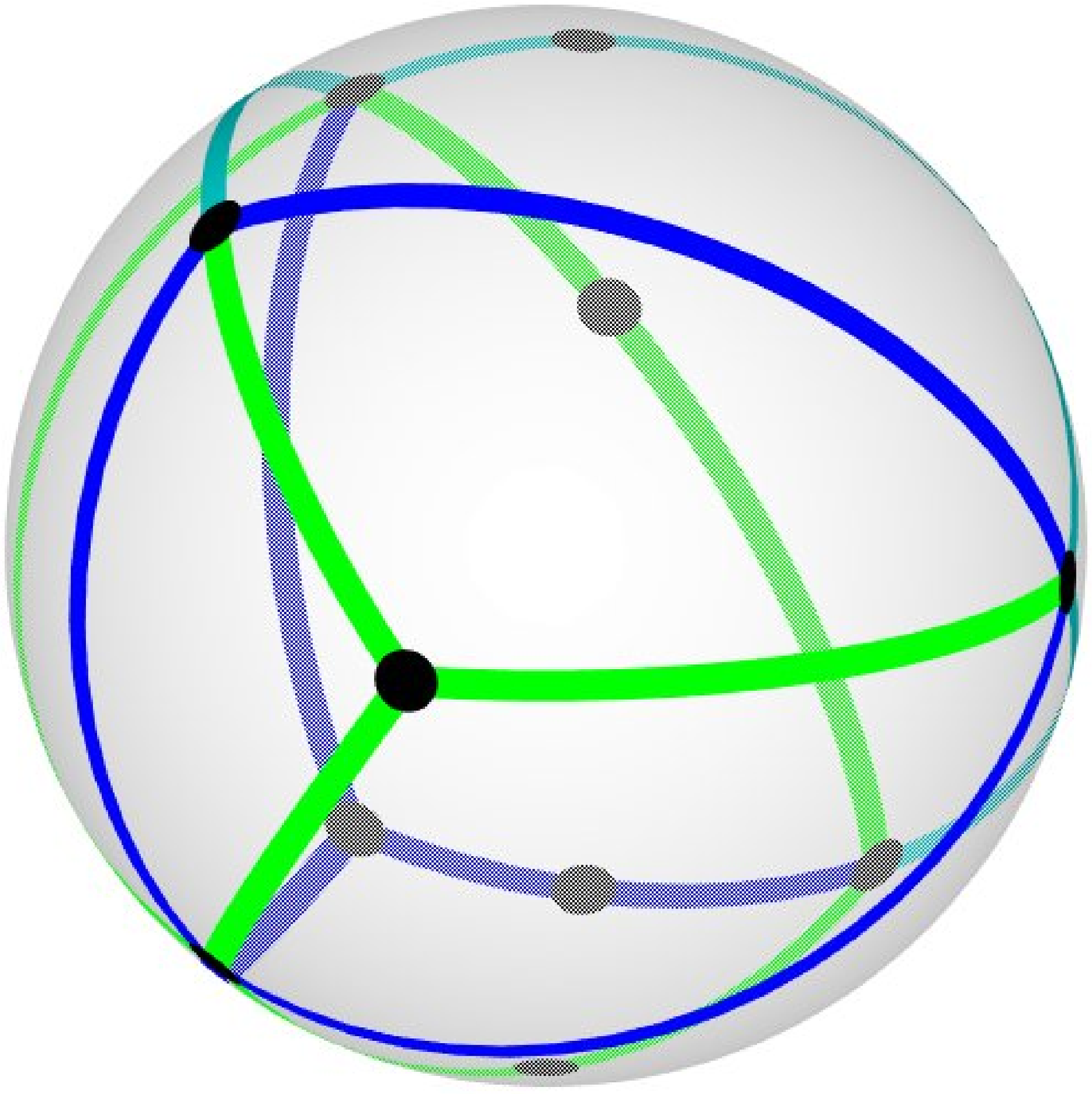,height=1.85cm,silent=} &
      \epsfig{figure=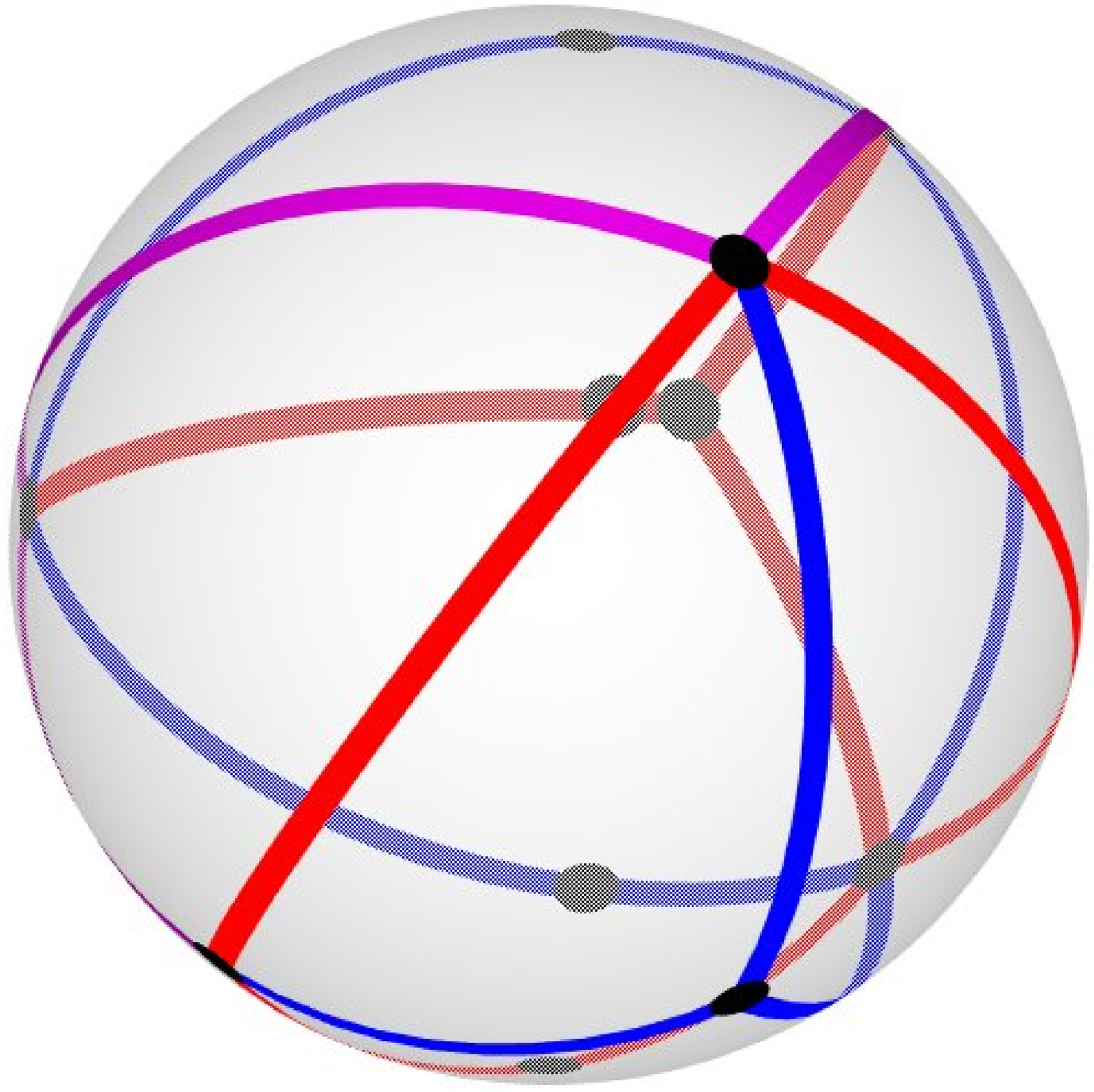,height=1.85cm,silent=} &
      \epsfig{figure=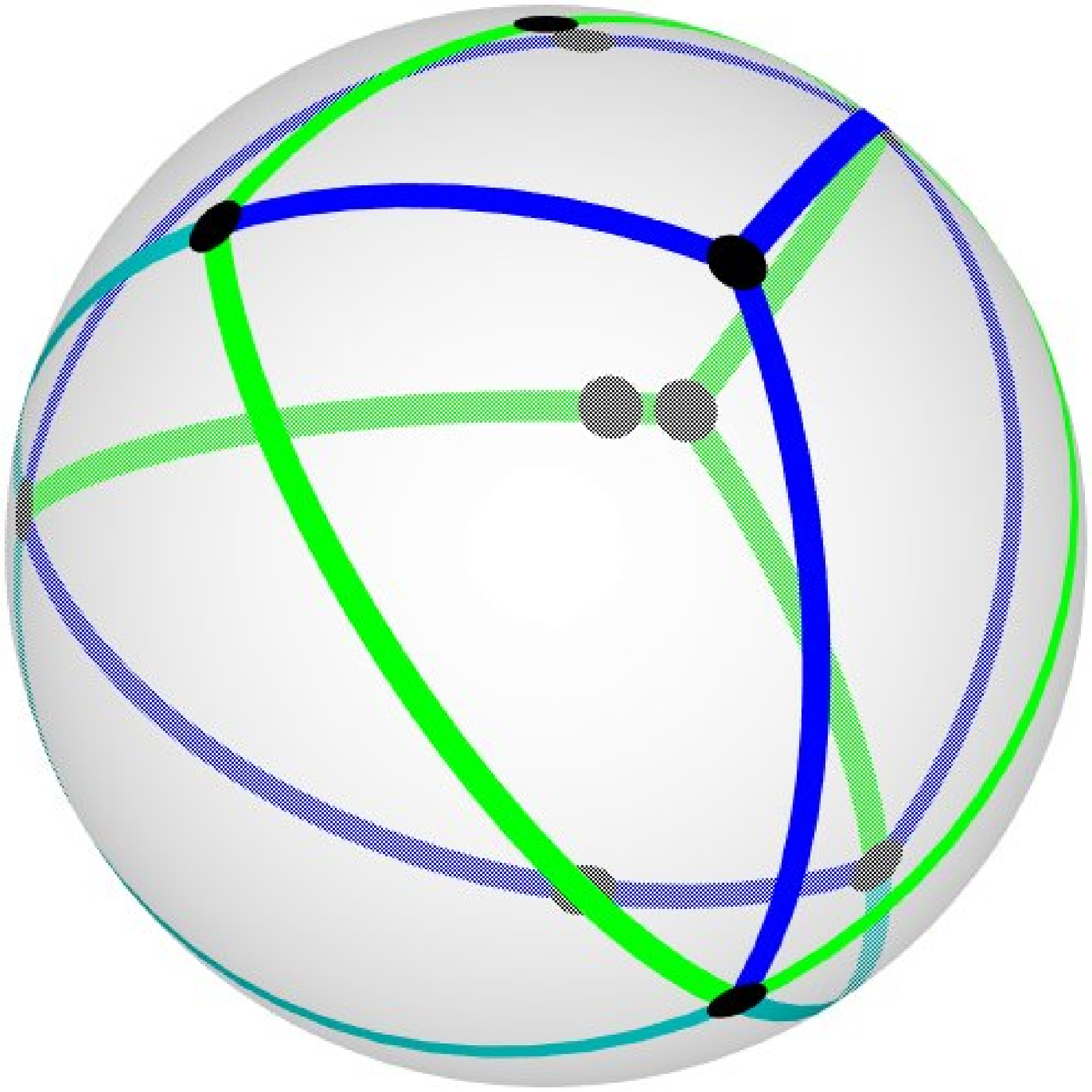,height=1.85cm,silent=}\\
      \epsfig{figure=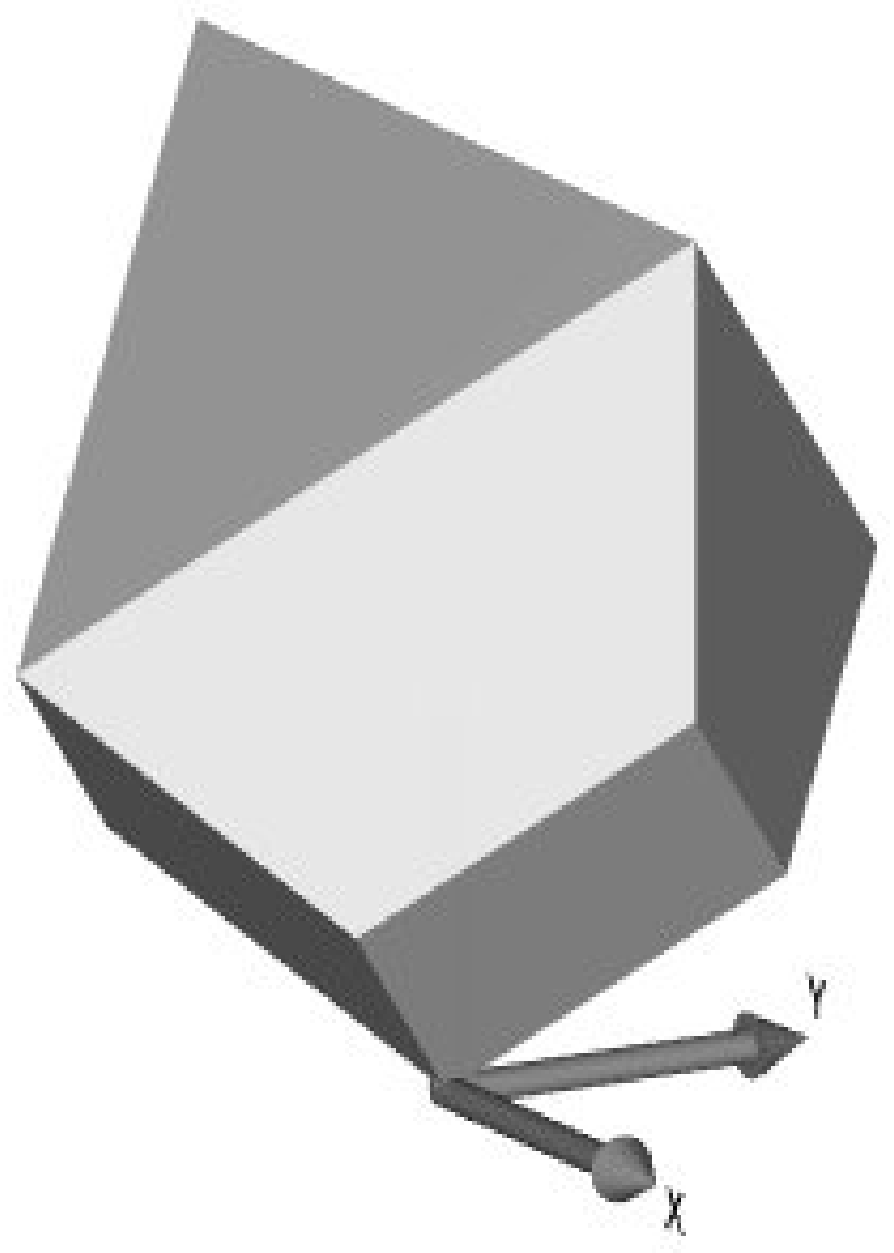,height=2.1cm,silent=} &
      \epsfig{figure=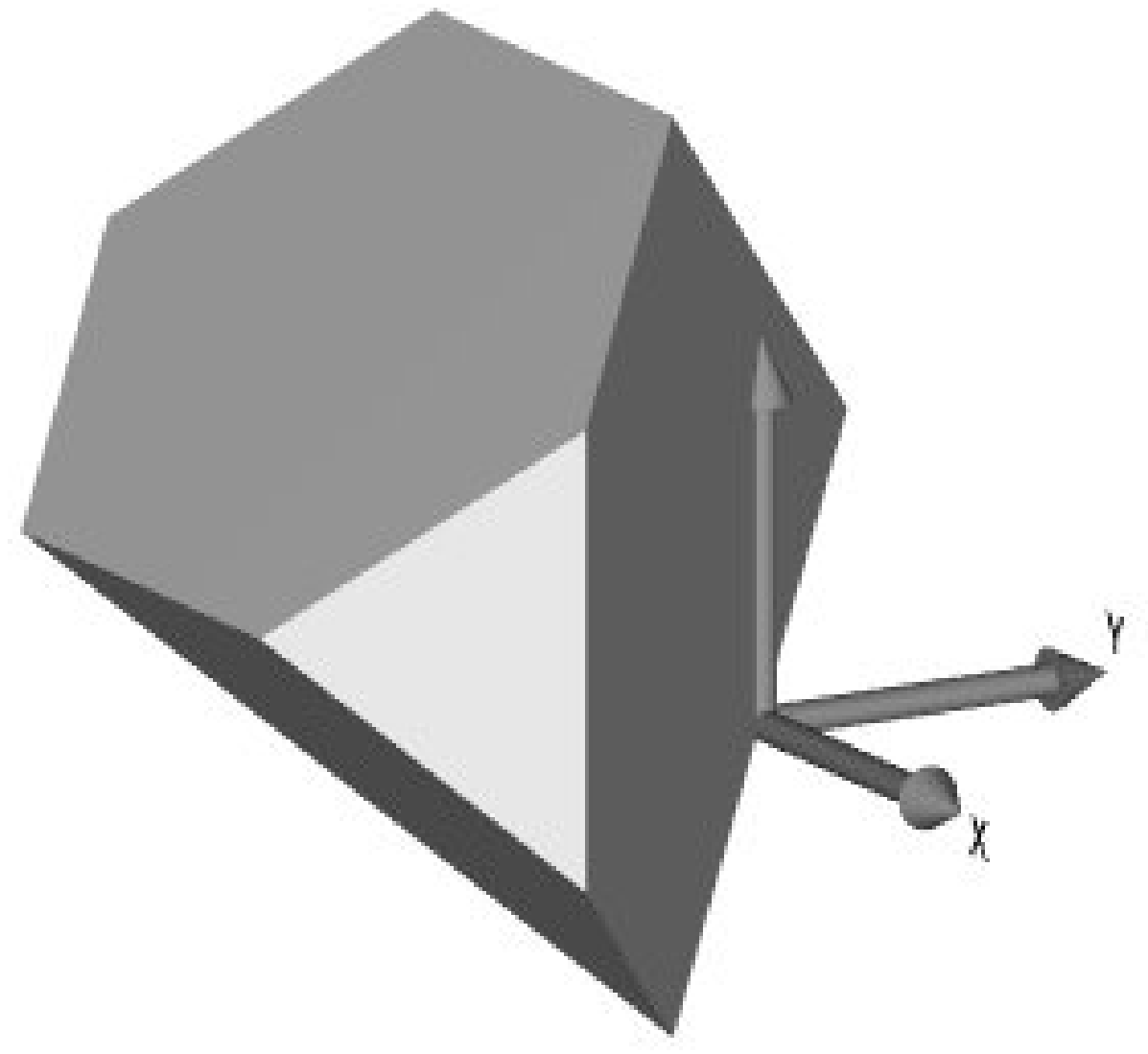,height=2.1cm,silent=} &
      \epsfig{figure=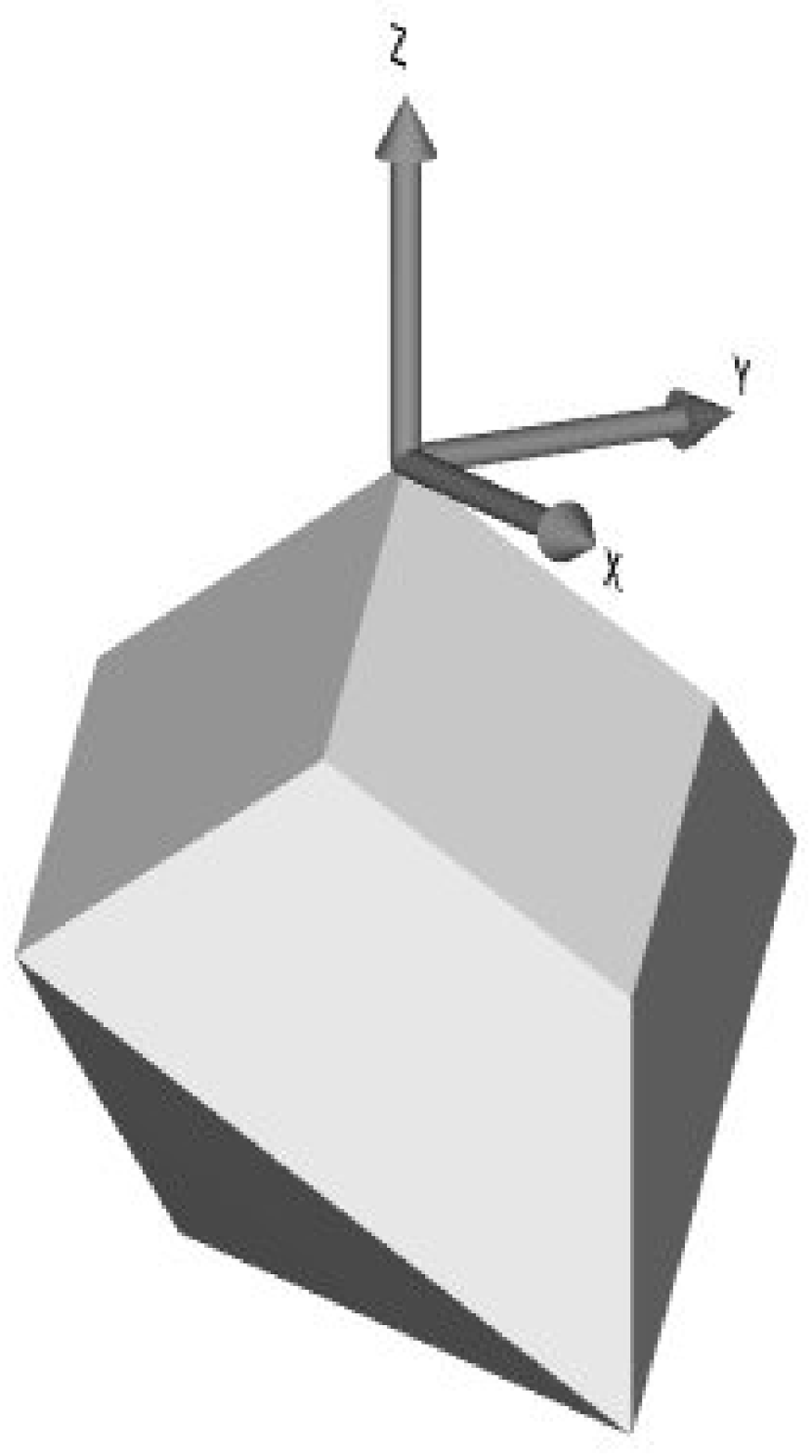,height=2.1cm,silent=} &
      \epsfig{figure=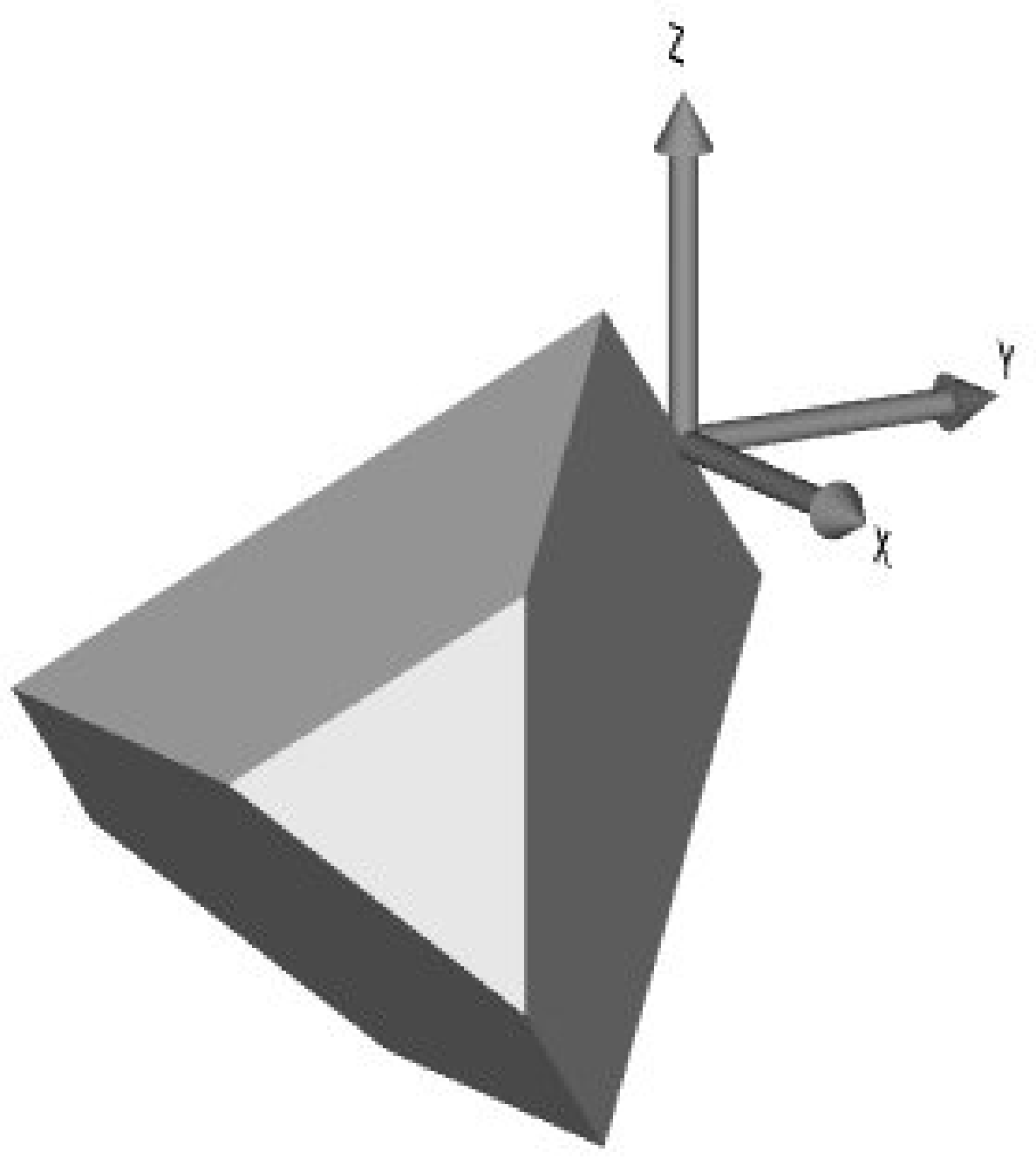,height=2.1cm,silent=} &
      \epsfig{figure=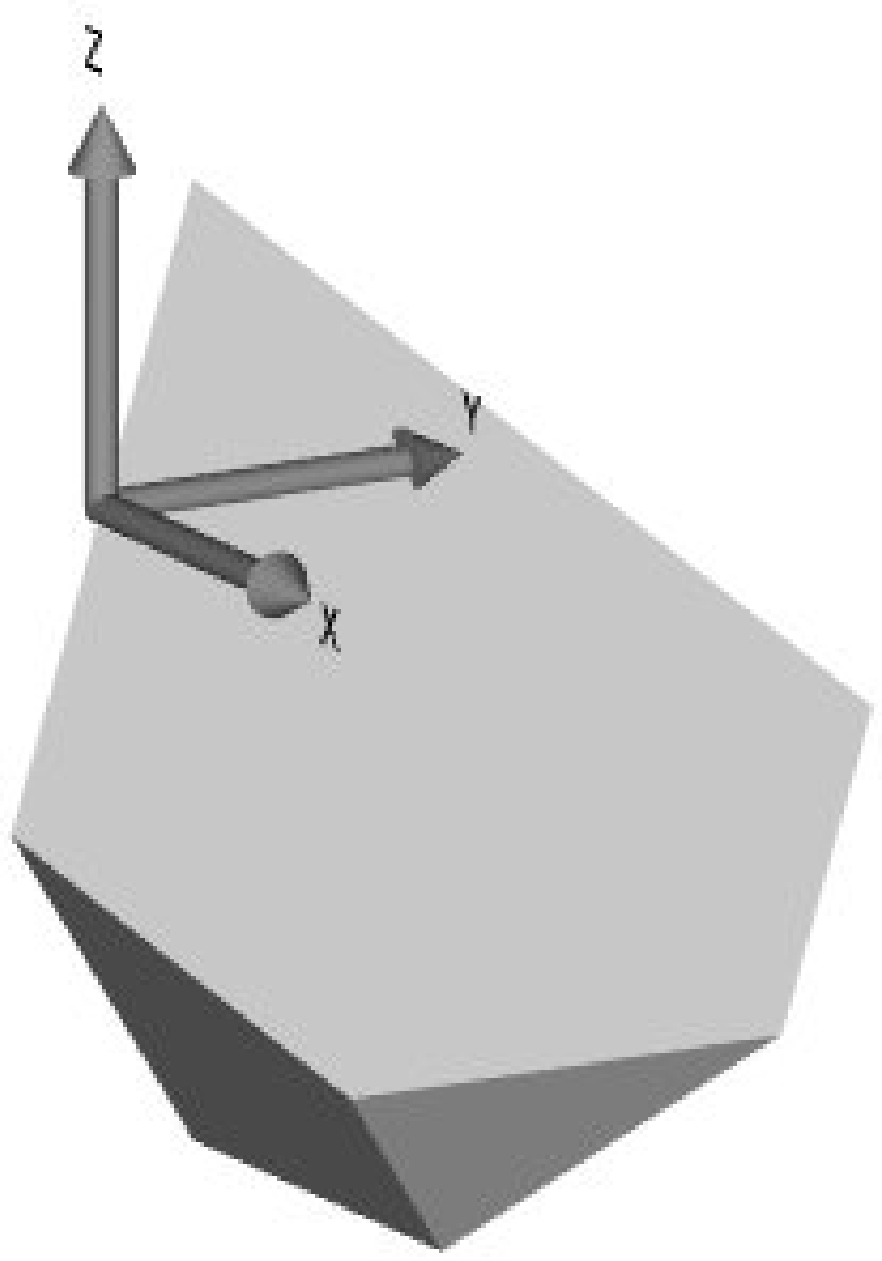,height=2.1cm,silent=} &
      \epsfig{figure=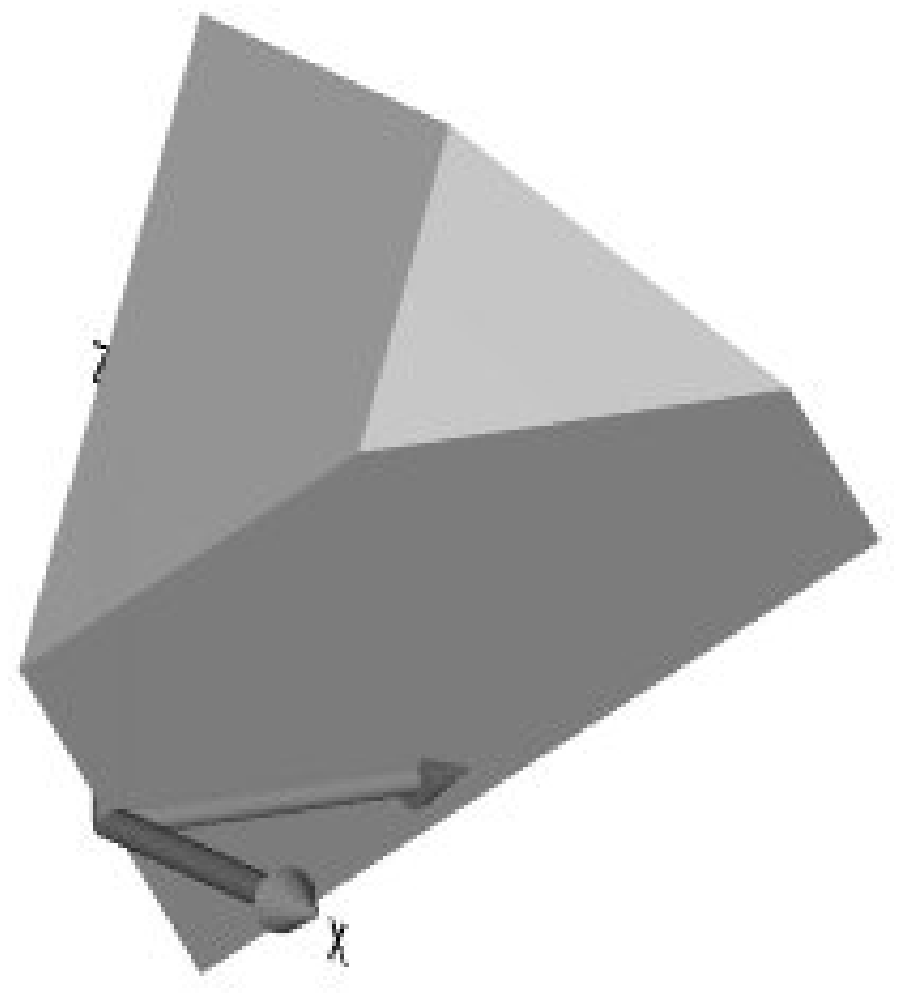,height=2.1cm,silent=}\\
      \epsfig{figure=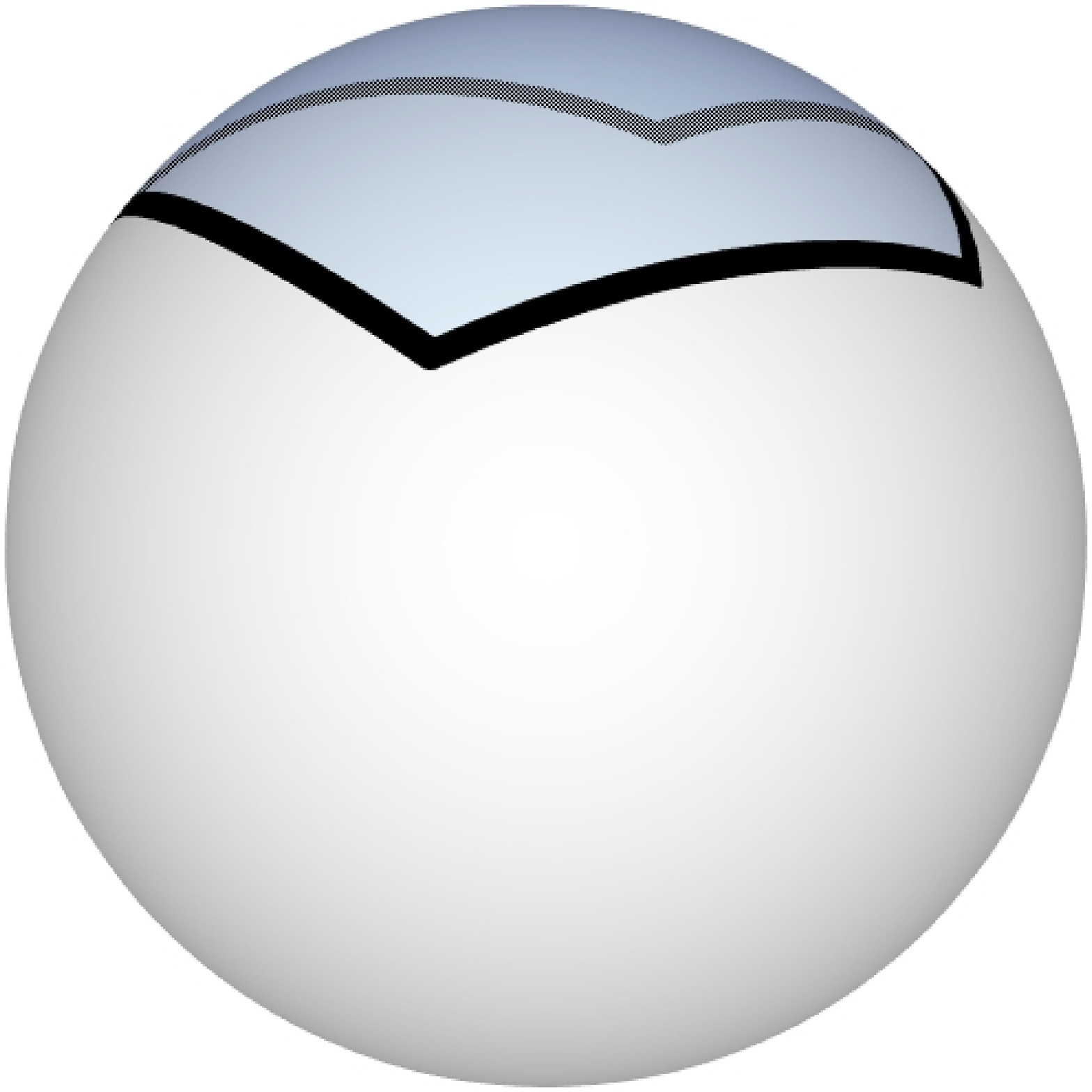,height=1.85cm,silent=} &
      \epsfig{figure=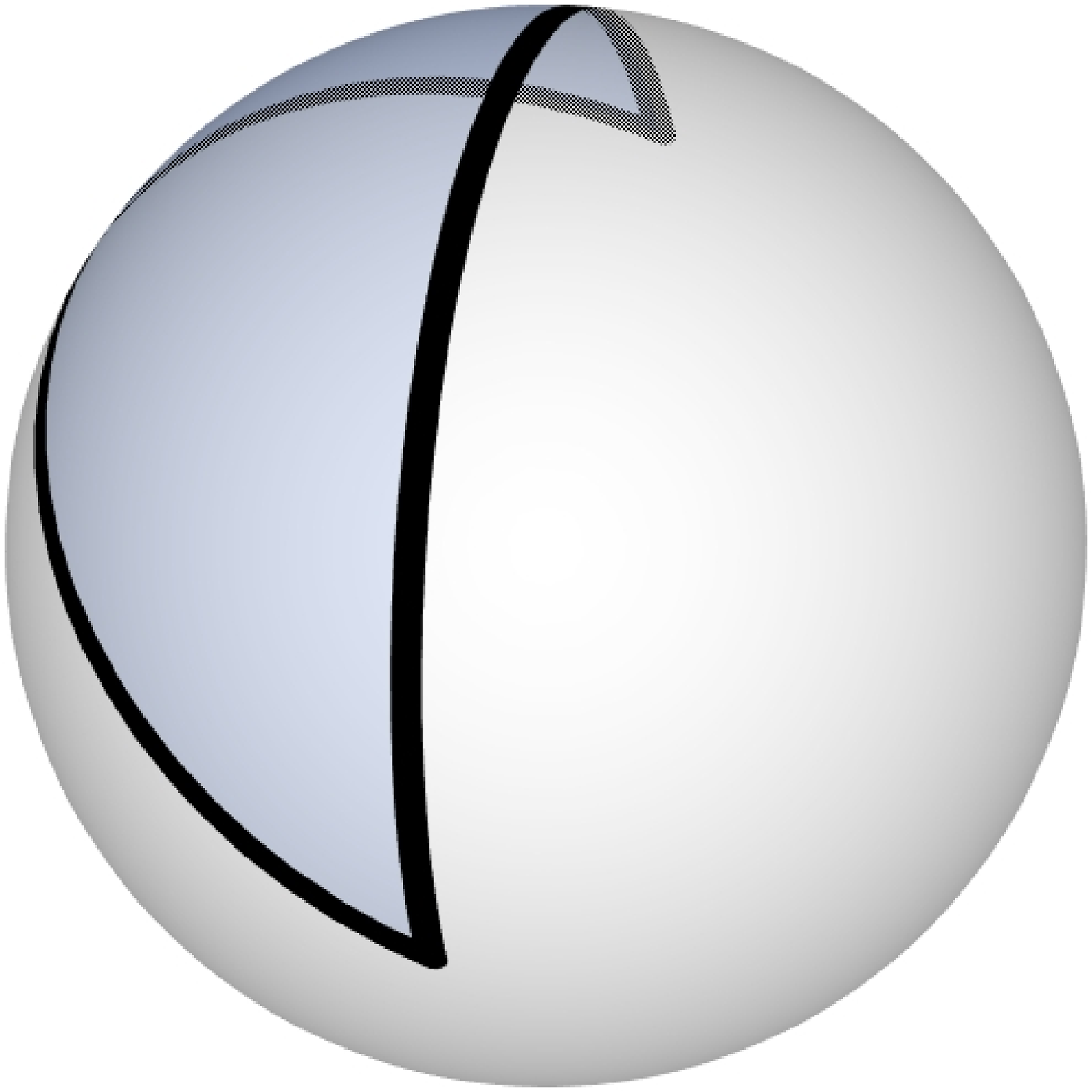,height=1.85cm,silent=} &
      \epsfig{figure=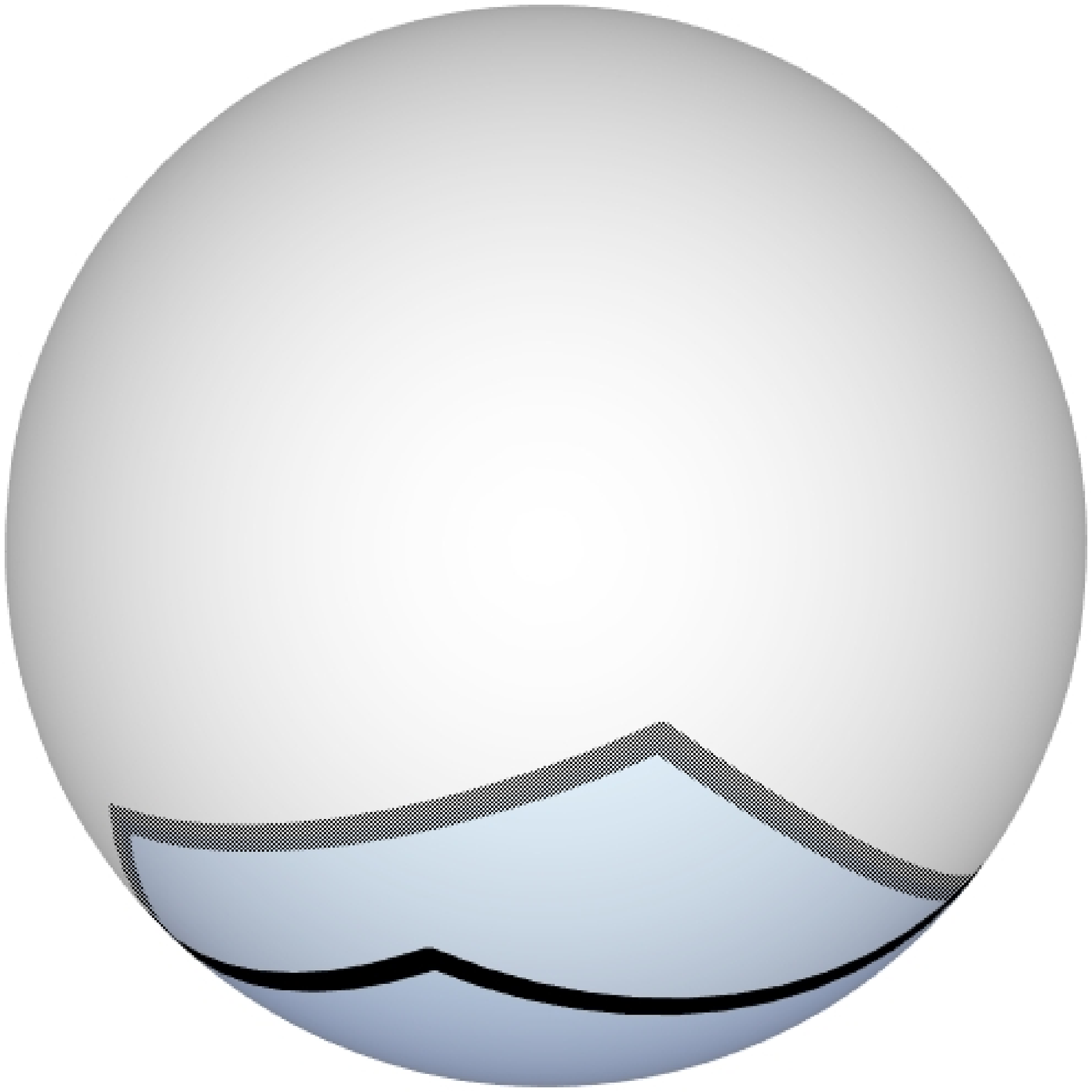,height=1.85cm,silent=} &
      \epsfig{figure=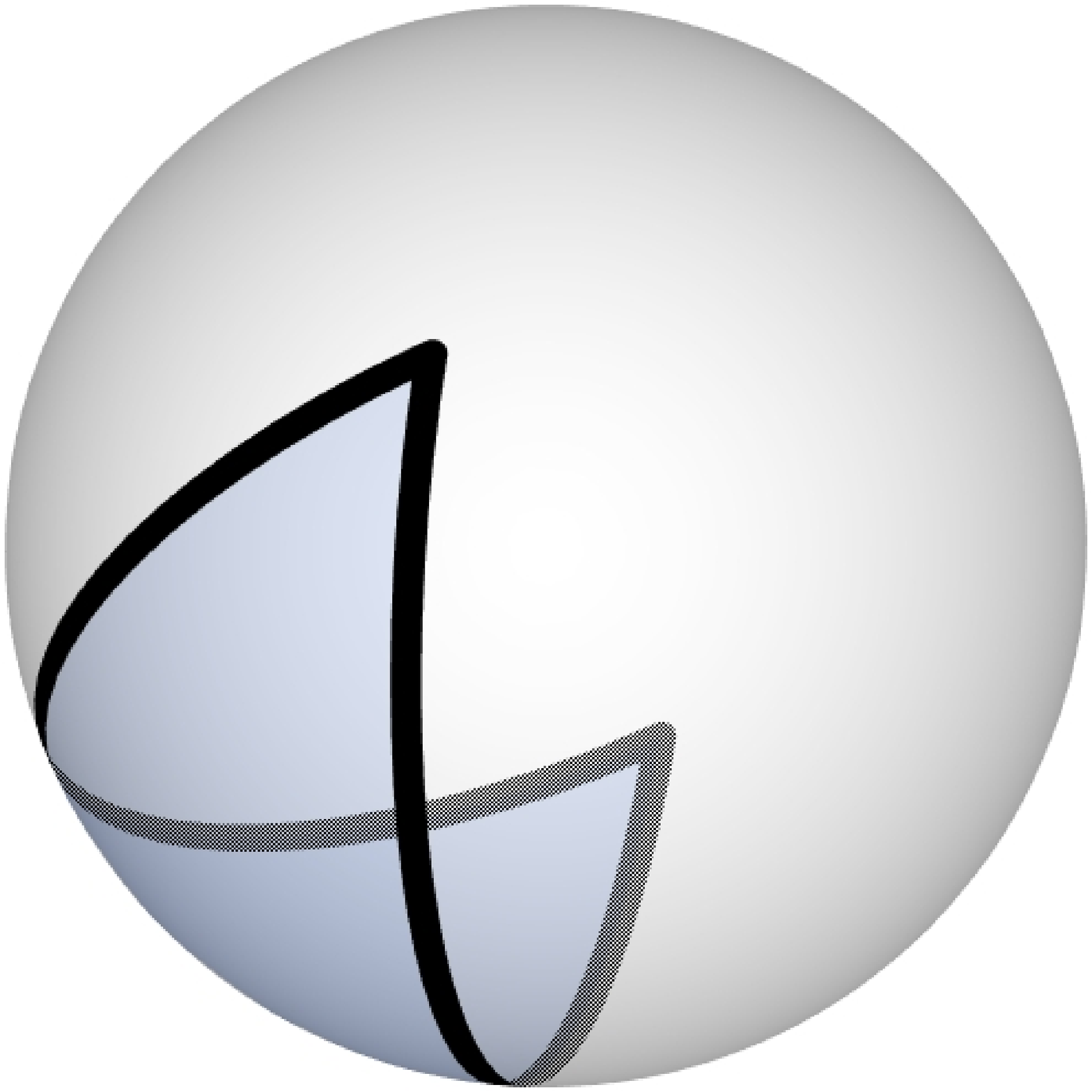,height=1.85cm,silent=} &
      \epsfig{figure=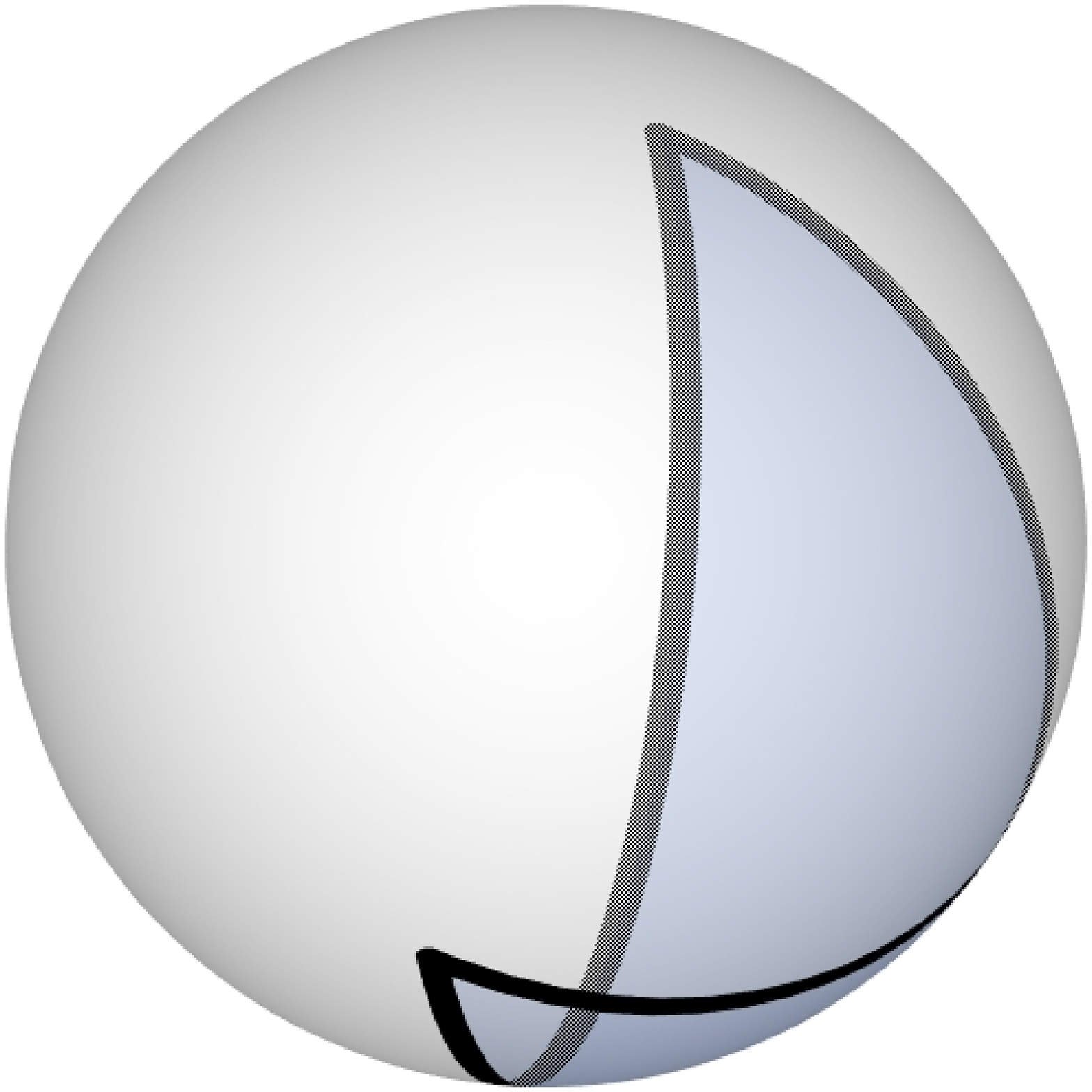,height=1.85cm,silent=} &
      \epsfig{figure=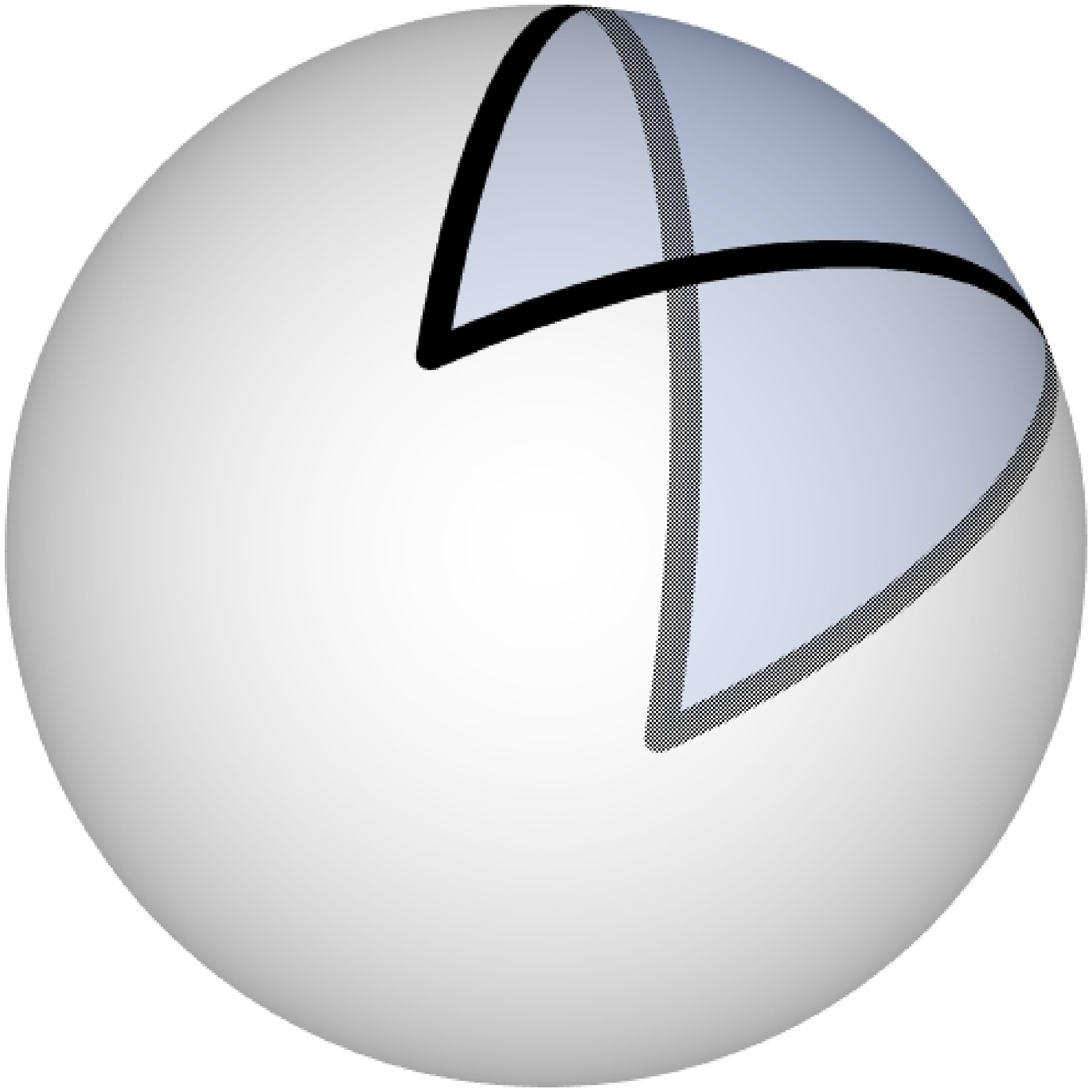,height=1.85cm,silent=}\\
      $R_1 \oplus (-G_1)$ &
      $R_1 \oplus (-B_1)$ &
      $G_1 \oplus (-R_1)$ &
      $G_1 \oplus (-B_1)$ &
      $B_1 \oplus (-R_1)$ &
      $B_1 \oplus (-G_1)$
    \end{tabular}}
  \caption[Samples of the pairwise Minkowski sums of the Split Star assembly
           sub-parts]%
          {\capStyle{Samples of the pairwise Minkowski sums
           of sub-parts of the Split Star assembly. The middle row contains
           six Minkowski sums. The top row contains the corresponding Gaussian
           maps. The bottom row contains the corresponding central projection
           of the Minkowski sums on $\spheretwo$.}}
  \label{fig:mink-sum}
  \vspace{-10pt}
\end{figure*}

Given a convex Minkowski sum $C$, we distinguish between four different
cases as follows:
\begin{enumerate}
\item The origin is contained in the interior of a facet of $C$.
\item The origin lies in the interior of an edge of $C$.
\item The origin coincides with a vertex of $C$.
\item The origin is separated from $C$.
\end{enumerate}
Computing the projection of a convex polytope $C$ can be done efficiently
using dedicated procedures that handle the four cases, respectively. Recall
that $C$ is represented as a Gaussian map, which is internally represented
as an arrangement of geodesic arcs embedded on the sphere. We traverse the
vertices of the arrangement. For each vertex $v$ we extract its associated
primal facet $f = G^{-1}(v)$. We dispatch the appropriate computation
based on the relative position of the origin with respect to the supporting
plane to $f$, and the supporting plane to adjacent facets of $f$.

{\bf If the origin is contained in the interior of a facet $f$ of $C$}, the
projection of the silhouette of $C$ is a great (full) circle that
divides the sphere into two hemispheres. The normal to the plane that
contains the great circle is identical to the normal to the supporting
plane to $f$, easily extracted from the arrangement representing the
Gaussian map of $C$. The \aos{} package conveniently supports the
insertion of a great circle, provided by the normal to the plane that
contains it, into an arrangement of geodesic arcs embedded on the
sphere.

We omit the implementation details of the following two cases, and proceed
to the general case.
{\bf If the origin is separated from $C$}, we traverse all edges of
$C$ until we find a silhouette edge characterized as follows: Let
$v_{s}$ and $v_{t}$ be the source and target vertices of some edge $e$
in the arrangement representing the Gaussian map of $C$, and let
$f_s = G^{-1}(v_s)$ and $f_t = G^{-1}(v_t)$ be their associated primal
facets, respectively. $e$ is a silhouette edge, if and only if, the
origin is not in the negative side of the supporting plane to $f_s$
and not in the positive side of the supporting plane to $f_t$.
We start with the first silhouette edge we find, and search for an
\begin{wrapfigure}[8]{l}{3cm}
  \vspace{-10pt}
  \epsfig{figure=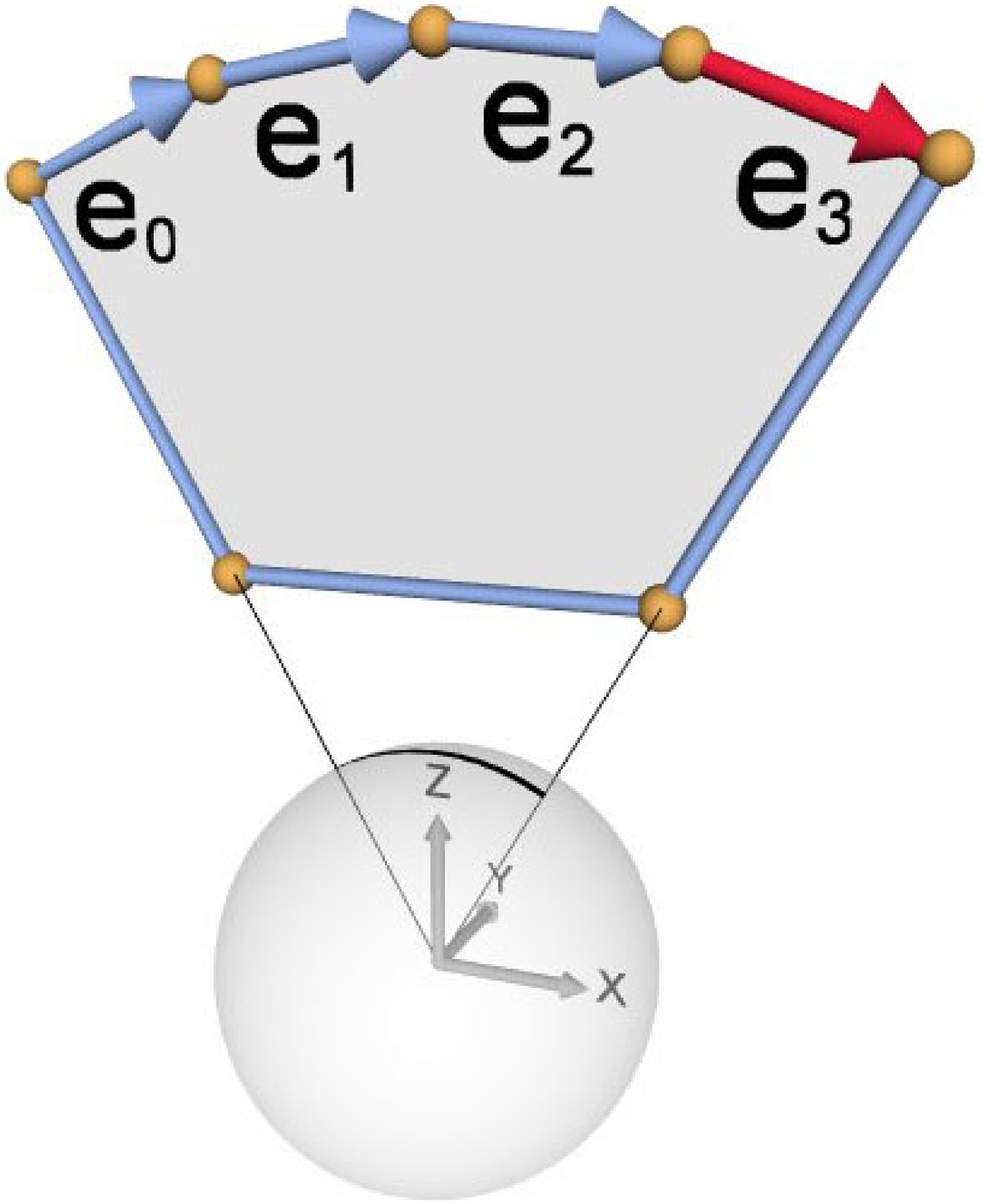,width=3cm,silent=}
\end{wrapfigure}
adjacent silhouette edge in a loop, until we rediscover the first
one. We project only the target vertices of significant silhouette
edges, and connect consecutive projections using arcs of great
circle. Let $e$ and $e'$ be adjacent silhouette edges. We skip $e$,
if the projections of $e$ and $e'$ lie on the same great circle.
For example, all but the last adjacent silhouette edges incident
to a facet supported by a plane that contains the origin are
redundant, as illustrated in the figure above. Here we skip $e_0$,
$e_1$, and $e_2$, and project the target vertex of $e_3$.

The output of this phase is a map from ordered pairs of distinct indices
into lists of arrangements as described above. Each ordered pair
$<\!i,j\!>, i \neq j$ is associated with the list of central
projections of the pairwise Minkowski sums of $P_j$'s sub-parts and
the reflection through the origin of $P_i$'s sub-parts.

\subsection{Pairwise Minkowski Sum Projection}
\label{ssec:assem_plan:pairwise-ms-projection}
For each pair of distinct parts $P_i$ and $P_j$ we compute the
union of projections of the pairwise Minkowski sums of all
sub-parts of part $P_j$ and reflections of all sub-parts of part $P_i$.

The output of this phase is a map from ordered pairs of distinct indices
into arrangements. Each ordered pair $<\!i,j\!>, i \neq j$ is
associated with a single arrangement extended as described above, that
represents the central projection $Q_{ij}$ of $M_{ij} = P_j \oplus (-P_i)$.

\begin{wraptable}{l}{6.8cm}
  \vspace{-15pt}
  \begin{tabularx}{6.5cm}{l@{}p{2ex}p{2ex}p{2ex}p{2ex}X}
    \hline
    \hline
    \multicolumn{6}{l}{Unite Pairwise Sub-part}\\
    \multicolumn{6}{r}{Minkowski sums Projections}\\
    \hline
    & \multicolumn{5}{l}{\textbf{for} $i$ \textbf{in} $\{1,2,\ldots,n\}$}\\
    & & \multicolumn{4}{l}{\textbf{for} $j$ \textbf{in} $\{1,2,\ldots,n\}$}\\
    & & & \multicolumn{3}{l}{\textbf{if} $i == j$ \textbf{continue}}\\
    & & & \multicolumn{3}{l}{$Q_{ij} = \emptyset$}\\
    & & & \multicolumn{3}{l}{\textbf{for} $k$ \textbf{in} $\{1,2,\ldots,m_i\}$}\\
    & & & & \multicolumn{2}{l}{\textbf{for} $\ell$ \textbf{in} $\{1,2,\ldots,m_j\}$}\\
    & & & & & \multicolumn{1}{l}{$Q_{ij} = Q_{ij} \cup Q^{ij}_{k\ell}$}\\
    \hline
  \end{tabularx}
  \vspace{-15pt}
\end{wraptable}
We exploit the \Index{overlay} operation in this phase the second time
throughout this process, this time in a loop. Given two distinct parts
$P_i$ and $P_j$ we traverse all projections in the set
$\{Q^{ij}_{k\ell}\,|\,k = 1,2,\ldots,m_i, \ell = 1,2,\ldots,m_j\}$,
and accumulate the result in the arrangement $Q_{ij}$. As mentioned
in Section~\ref{ssec:assem_plan:pairwise-sub-part-ms-construction}, when
the overlay operation progresses, new vertices, edges, and faces of the
resulting arrangement are created. When a new face $f$ is created as a
result of the overlay of a face $g$ in some projection $Q^{ij}_{k\ell}$, and
a face in the accumulating arrangement, the Boolean flag associated with
$f$, which indicates whether all directions $\vecd \in f$ pierce $M_{ij}$,
is turned on, if $\vecd$ pierces $M^{ij}_{k\ell}$, that is, if the flag
associated with the face of $g$ is on.

The intermediate result of this step are arrangements with potentially
redundant edges and vertices. It is desired (but not necessary) to
remove these cells, as it reduces the time consumption of the
succeeding operations, which is directly related to the complexity of
the arrangements. It has even a larger impact when the optimization
described in Section~\ref{sec:assem_plan:optimization} is applied, as the
optimization decreases the number of preceding operation at the
account of slightly increasing the number of succeeding operations. We
remove all edges and vertices that are in the interior of the projection,
that is all marked edges and vertices. We also remove spherically collinear
vertices on the boundary of the projection, the degree of which decreased
below three, as a result of the redundant-edge removal.

\begin{figure*}[!htp]%
  \centerline{%
    \begin{tabular}{cccccc}
      \epsfig{figure=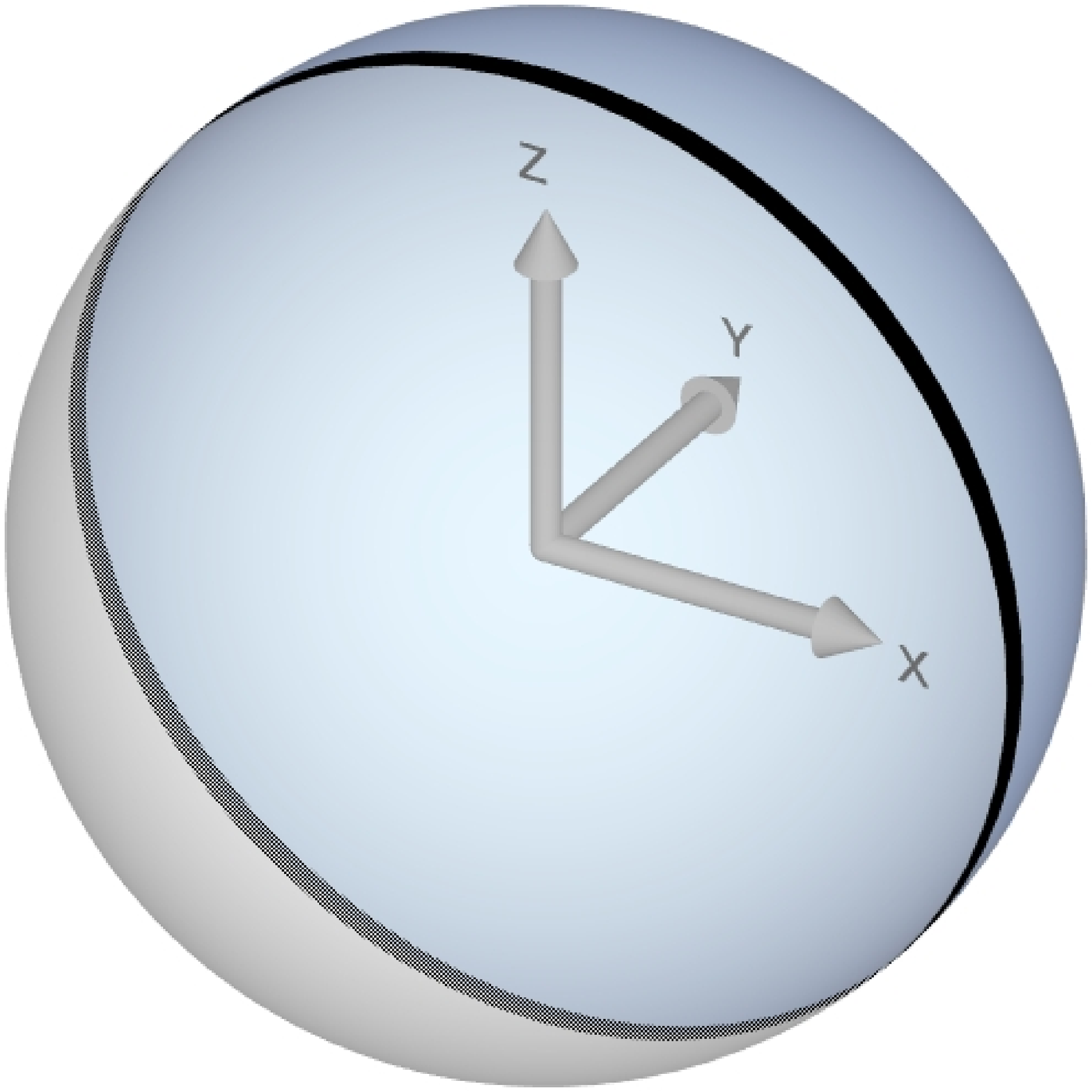,height=1.85cm,silent=} &
      \epsfig{figure=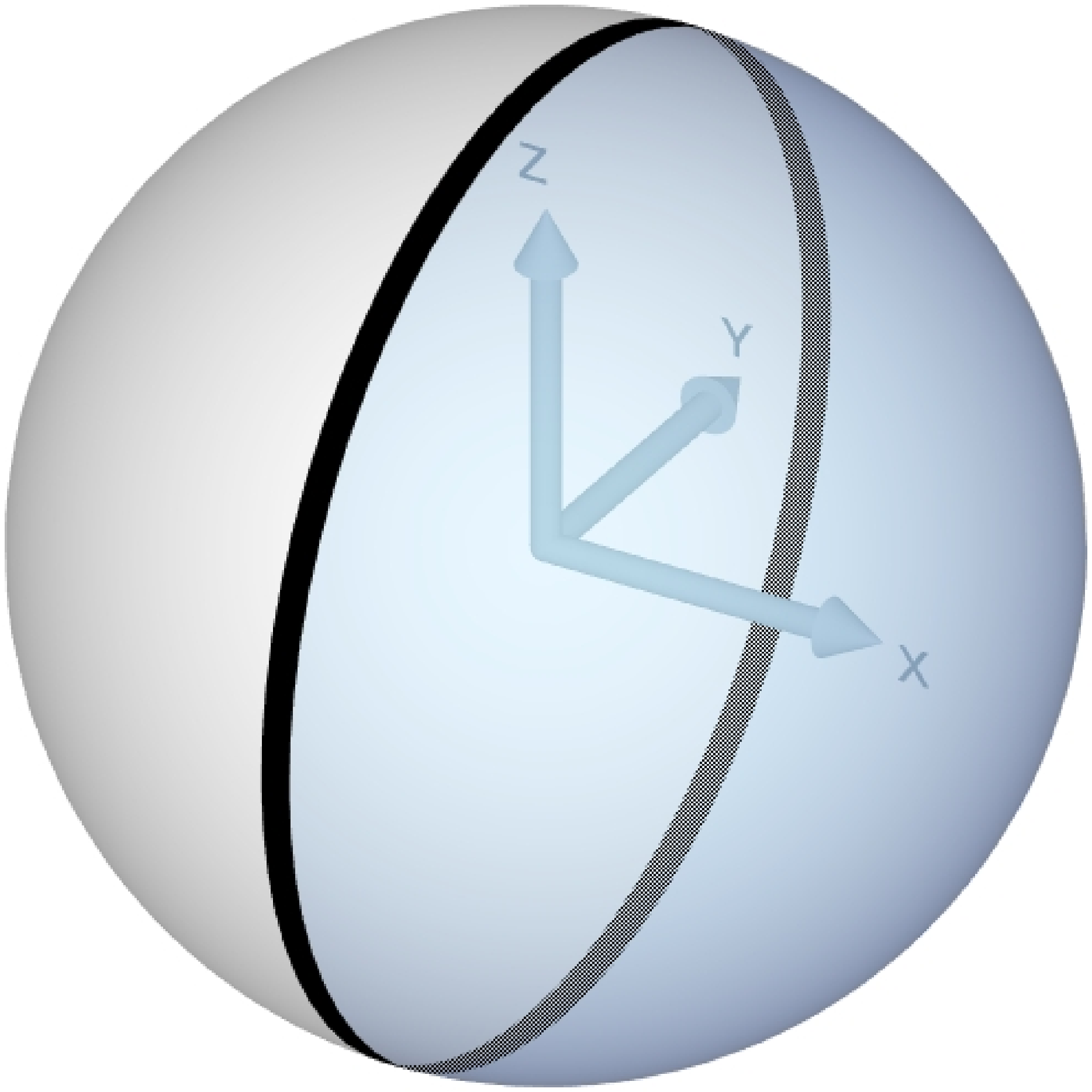,height=1.85cm,silent=} &
      \epsfig{figure=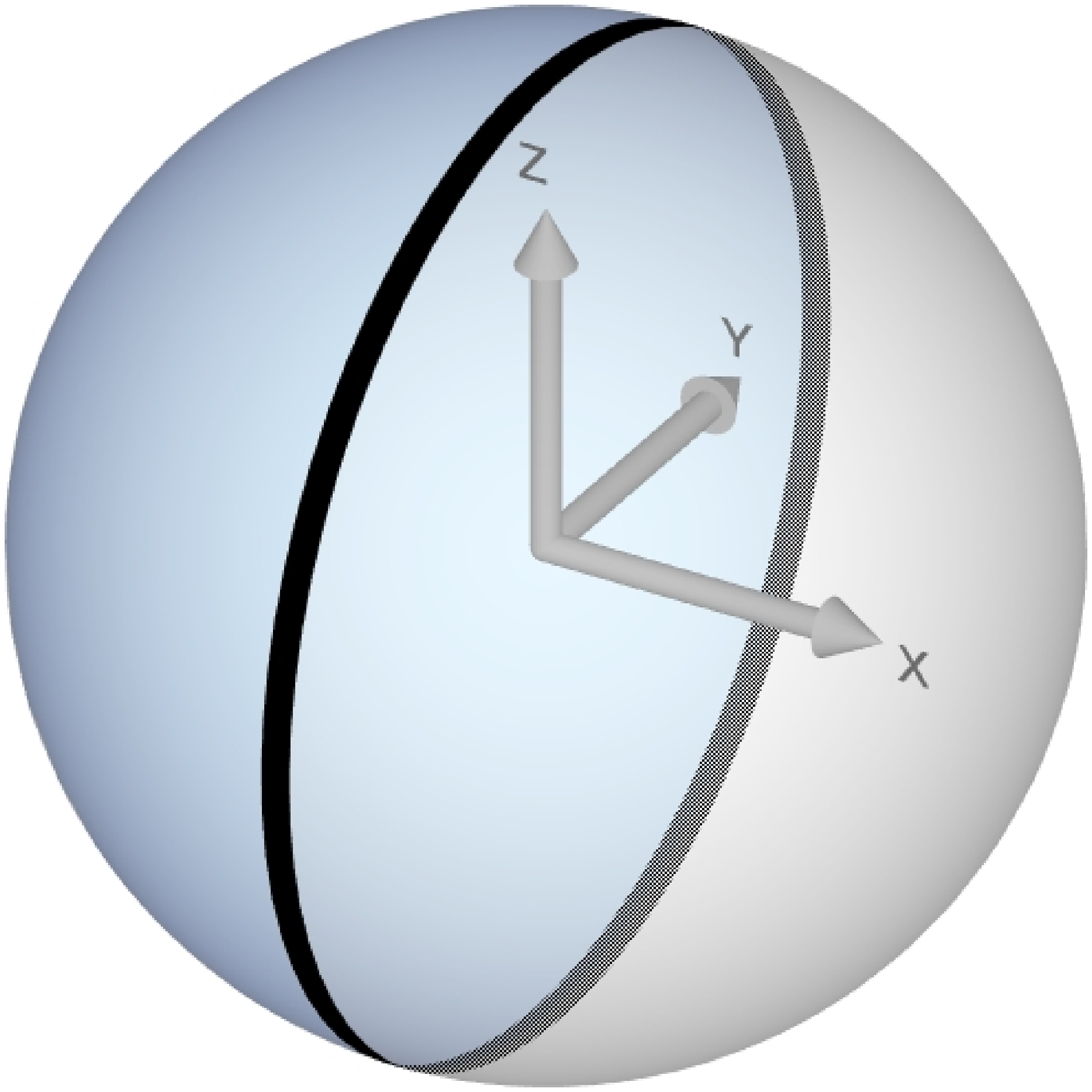,height=1.85cm,silent=} &
      \epsfig{figure=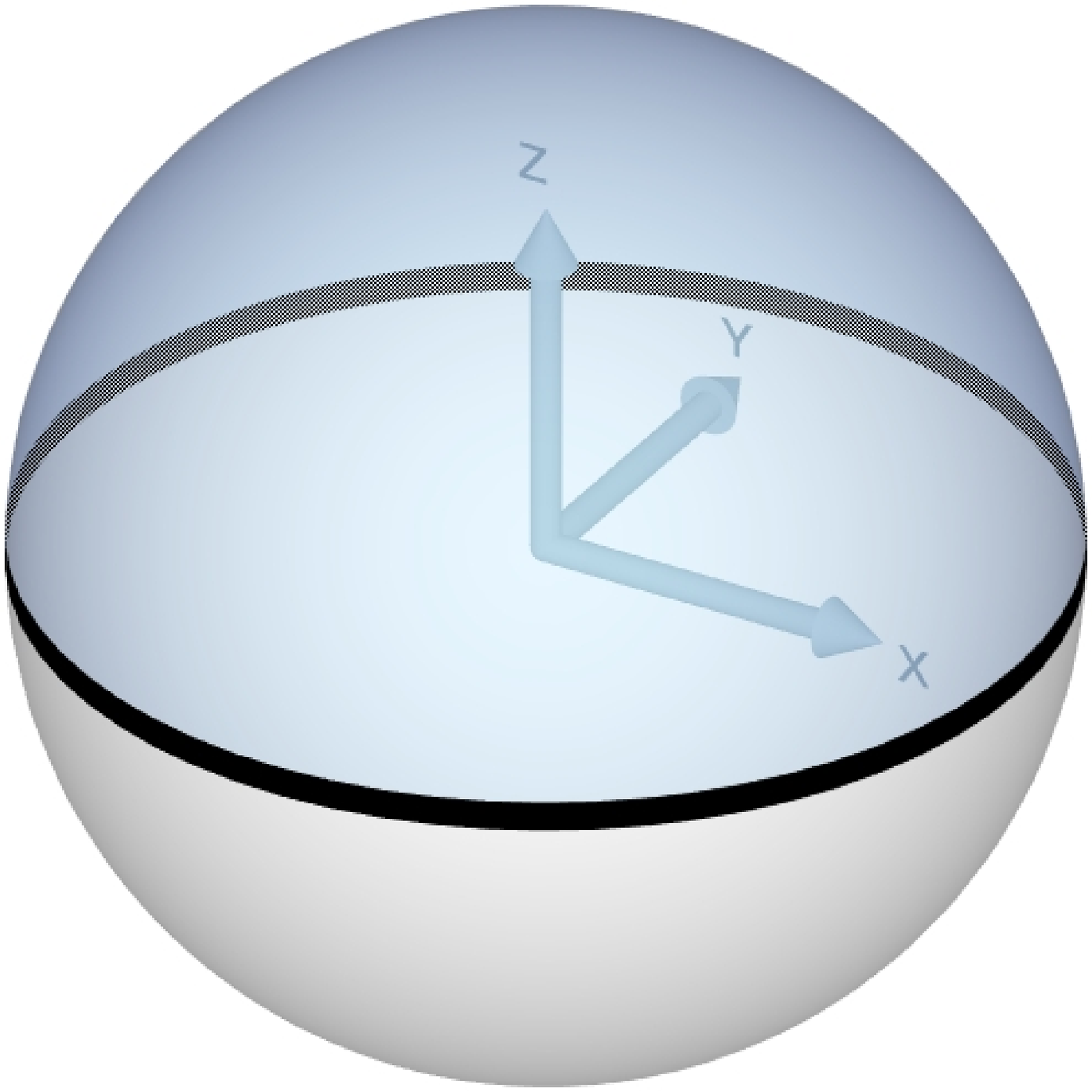,height=1.85cm,silent=} &
      \epsfig{figure=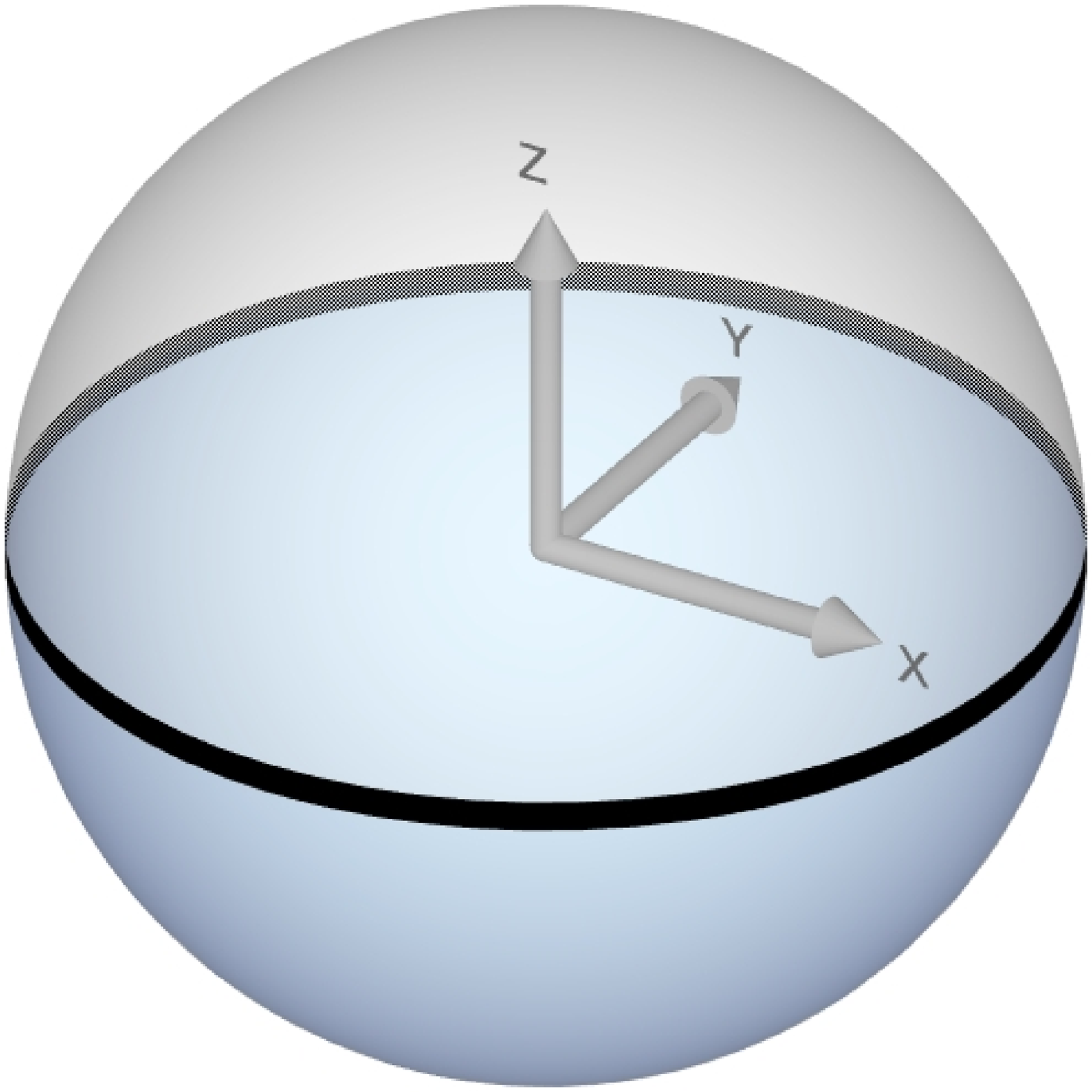,height=1.85cm,silent=} &
      \epsfig{figure=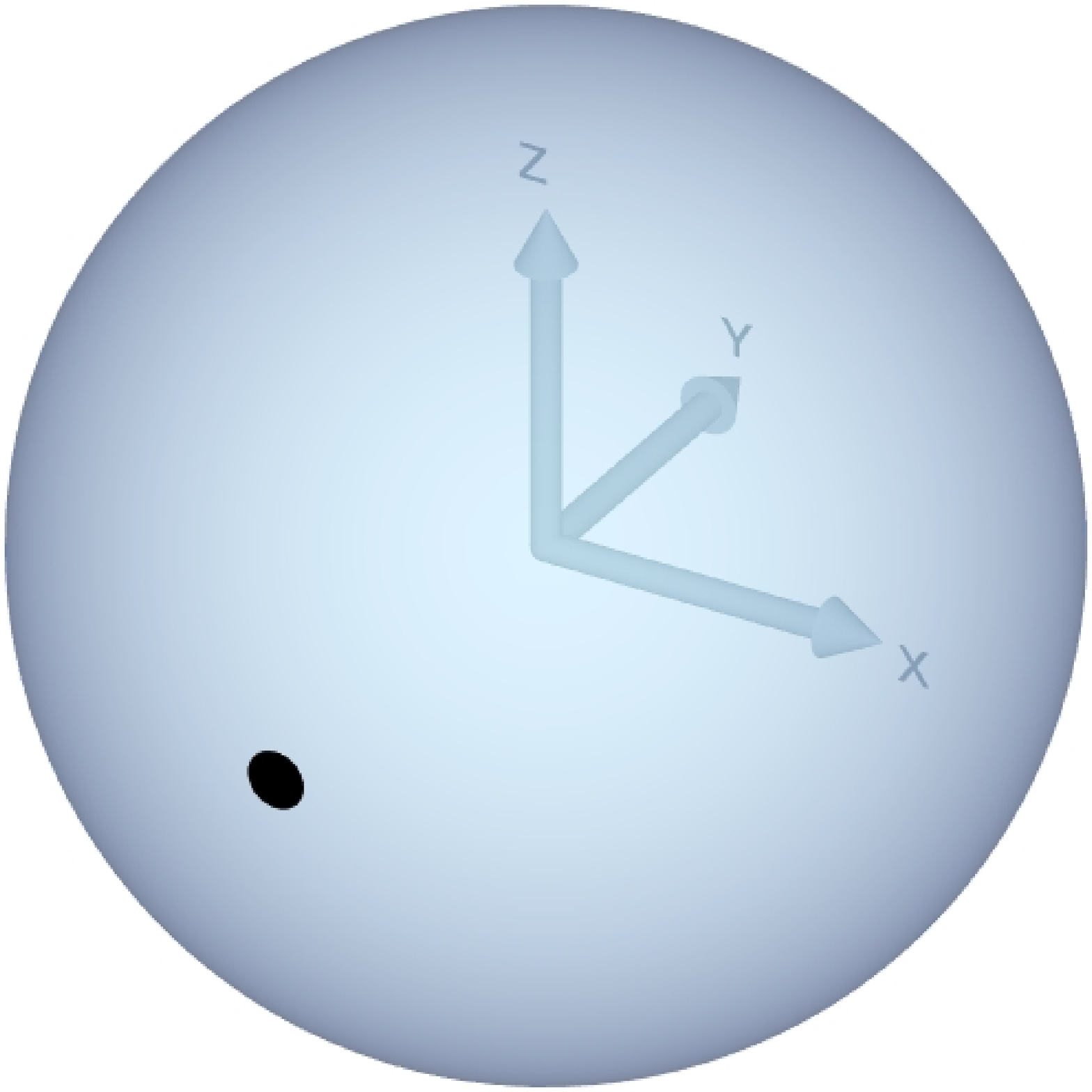,height=1.85cm,silent=}\\
      (a) & (b) & (c) & (d) & (e) & (f)
    \end{tabular}
  }
  \caption[Peg-in-the-hole Minkowski sum projections]%
          {\capStyle{Peg-in-the-hole Minkowski sum
           projections. (a), (b), (c), (d), and (e) are the sub-part
           projection. (f) is the union of the former.}}
  \label{fig:peg-in-the-hole}
  \vspace{-10pt}
\end{figure*}

\cgal{} also supports Boolean operations applied to general
polygons\footnote{The generic code supports point sets bounded by
algebraic curves embedded on parametric surfaces referred to as general
polygons.} and in particular the union operation. However, it consumes
and produces {\em regularized} general polygons; see
Section~\ref{ssec:aos:applications:bso}. This \Index{regularization}
operation is harmful in the realm of assembly planning.
\begin{wrapfigure}[5]{r}{4.6cm}
  \vspace{-10pt}
  \begin{tabular}{cc}
    \epsfig{figure=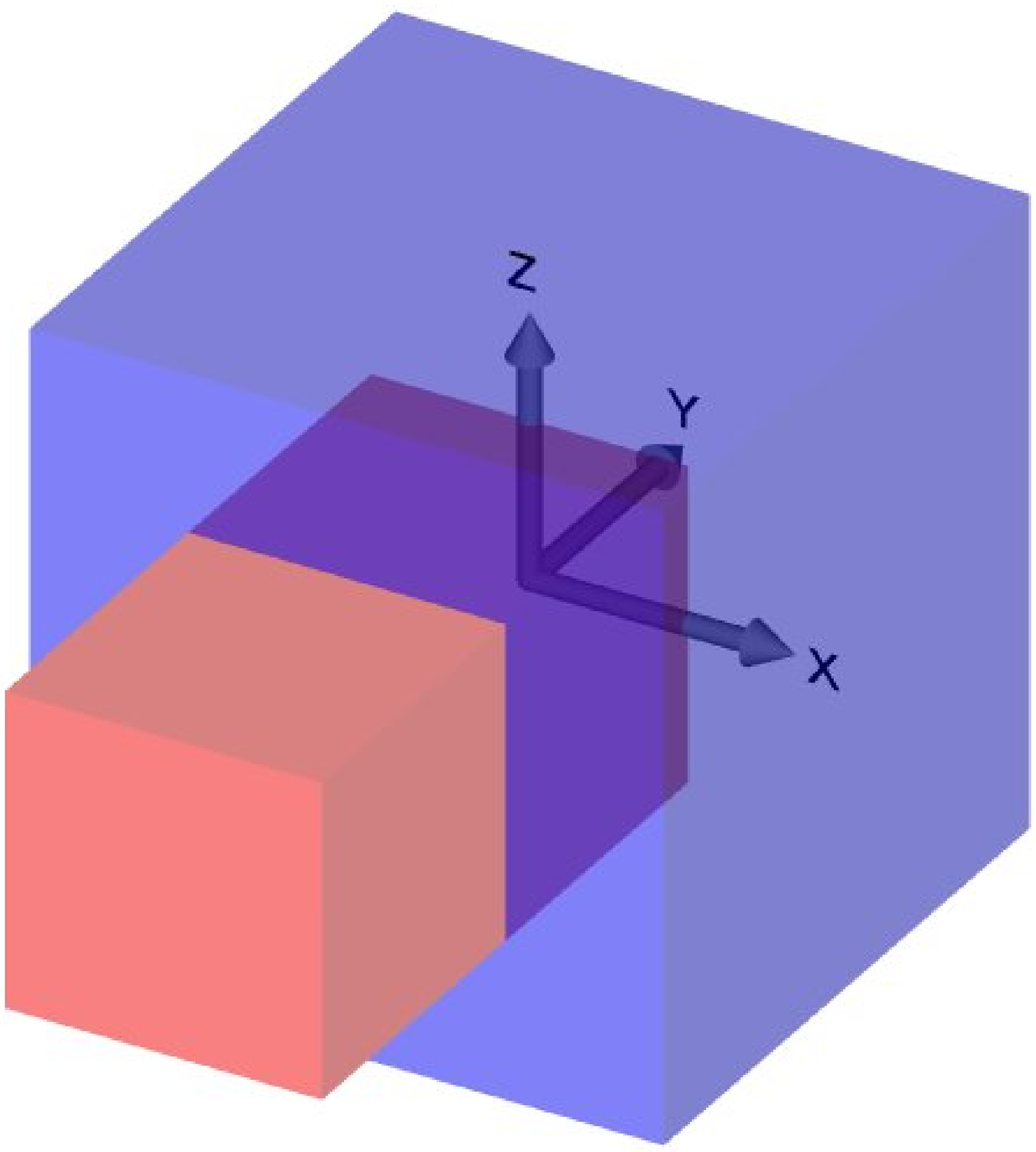,height=2cm,silent=} &
    \epsfig{figure=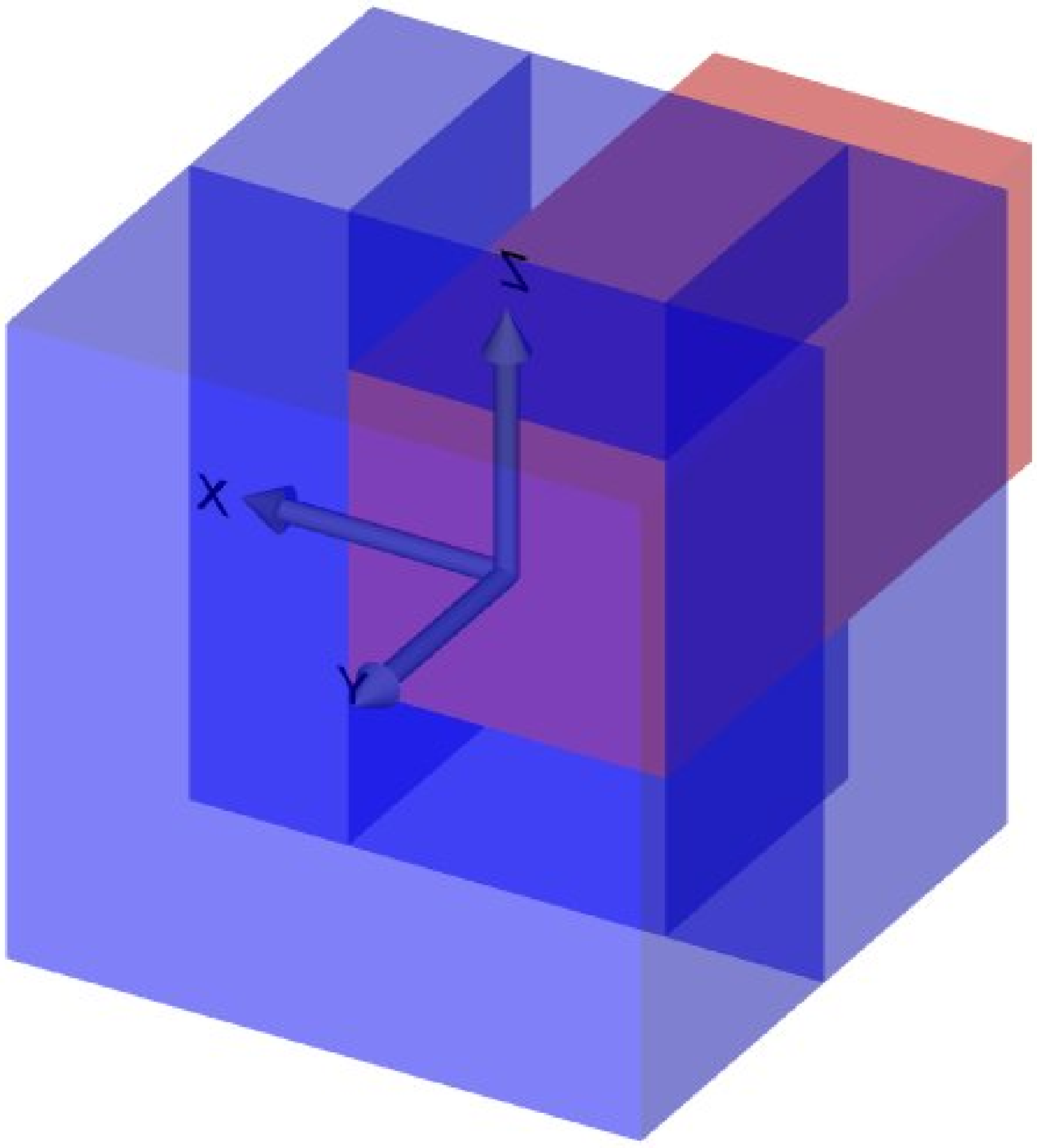,height=2cm,silent=}
  \end{tabular}
\end{wrapfigure}
Therefore, we work 
directly on the
cells of the arrangements $Q_{ij}$. Quite often the projection
contains isolated vertices and edges, as occurs in
the example depicted on the right, referred to as ``peg-in-the-hole''.
Here the assembled product is translucently viewed from two opposite
directions. The blue part is stationary and is decomposed into five
sub-parts. Figure~\ref{fig:peg-in-the-hole}
illustrates the corresponding five pairwise Minkowski sum projections, and
their union. The complement of the union consists of a single isolated vertex.

\begin{wrapfigure}[9]{l}{4.2cm}
  \begin{tabular}{cc}
    \multicolumn{2}{c}{\epsfig{figure=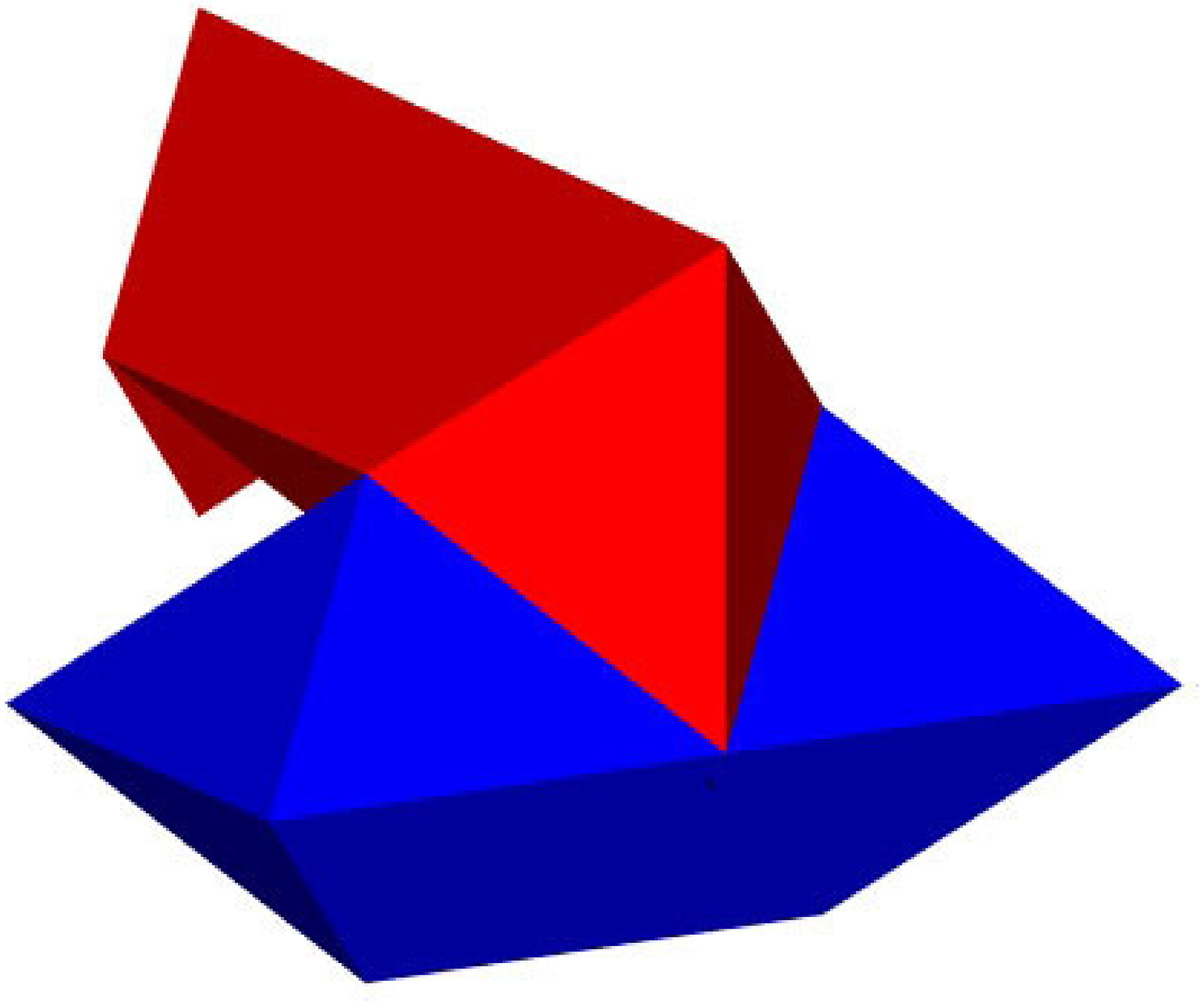,width=3.5cm,silent=}}\\   
    \epsfig{figure=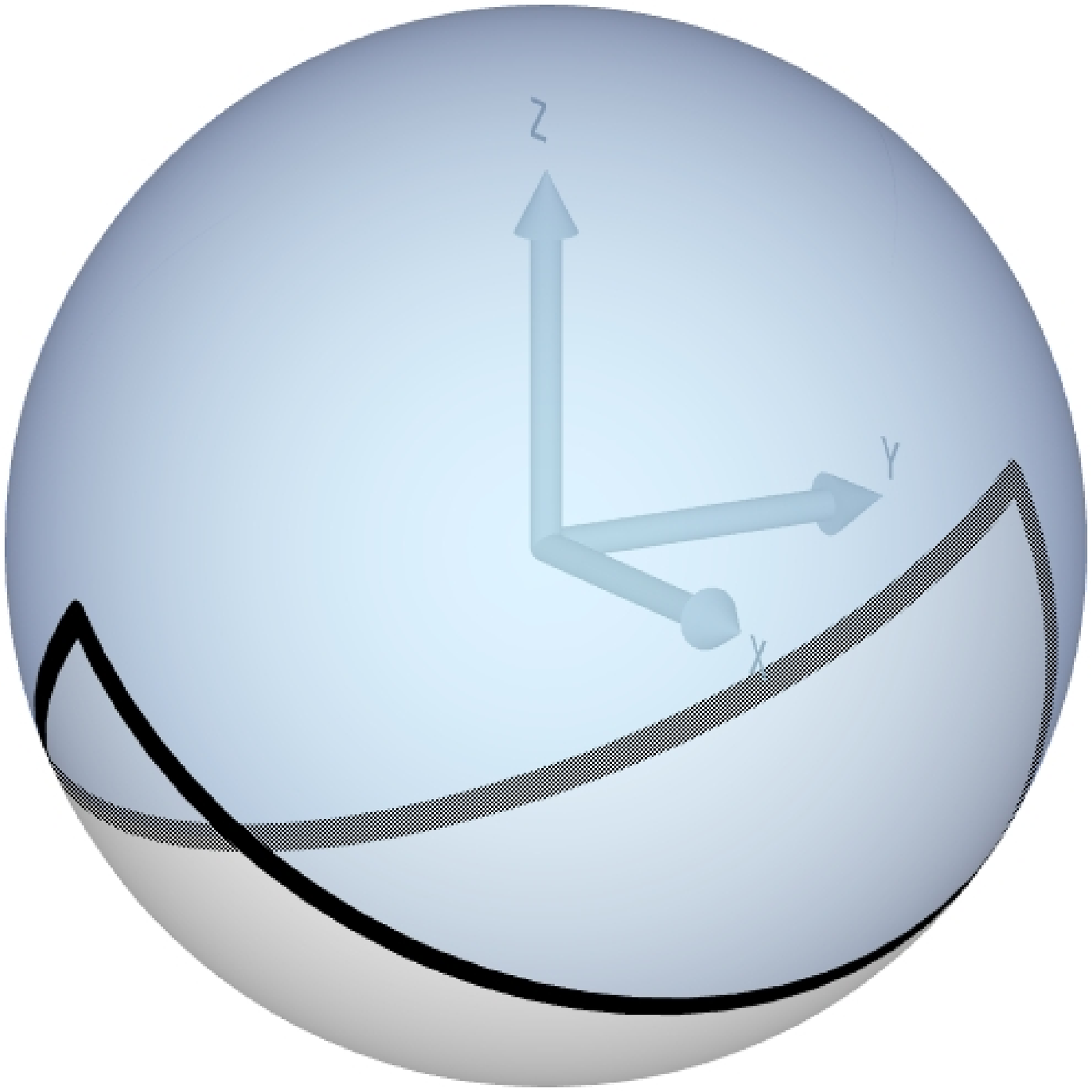,width=1.85cm,silent=} &
    \epsfig{figure=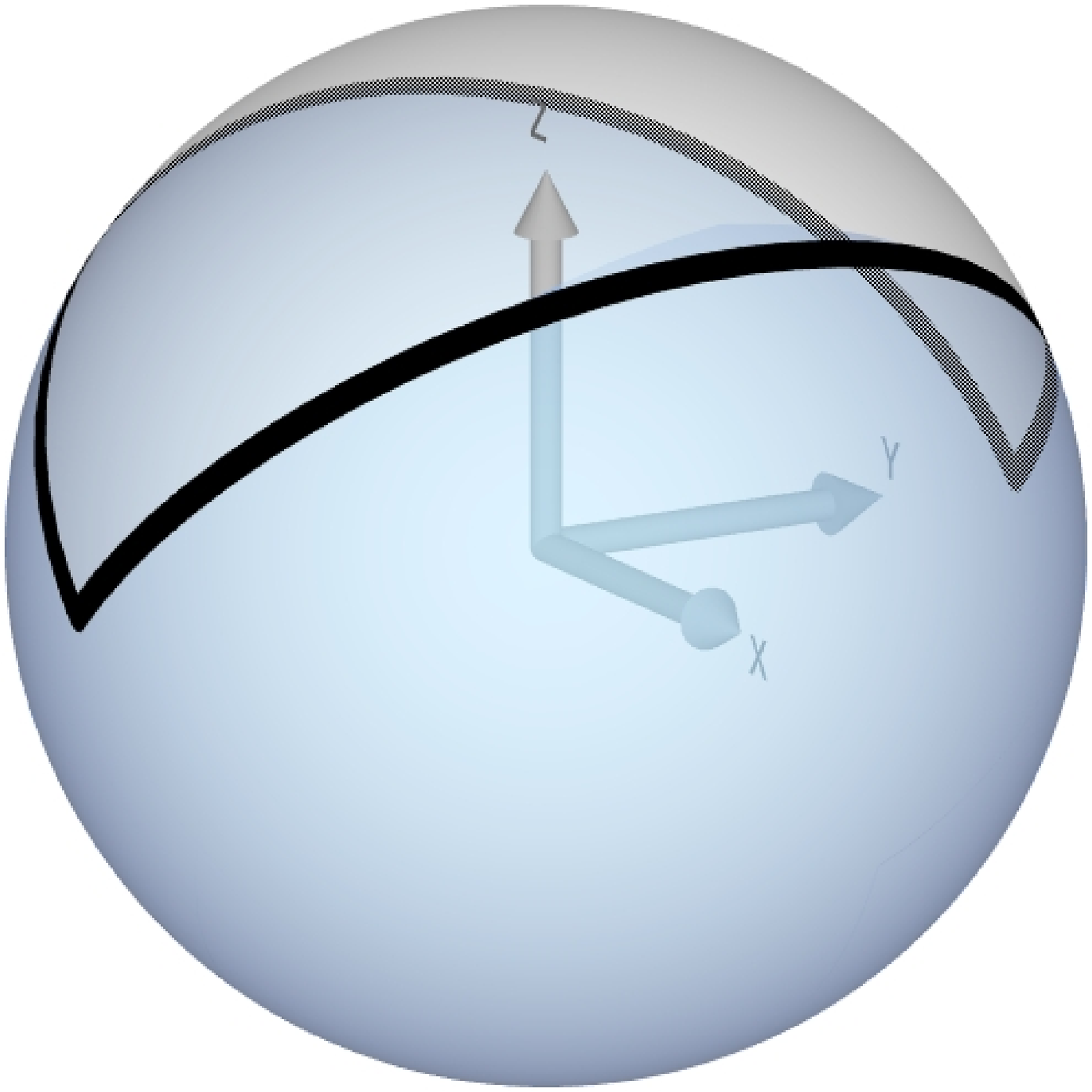,width=1.85cm,silent=}\\
    $B \oplus (-R)$ & $R \oplus (-B)$
  \end{tabular}
\end{wrapfigure}
\noindent
Recalling our Split Star assembly, the projection of the Minkowski sum
of the red part and the reflection of the blue part, and its
reflection, that is, the projection of the Minkowski sum of the blue
part and the reflection of the red part are depicted on the left.

\subsection{Motion-Space Construction}
\label{ssec:assem_plan:motion-space-construction}
We compute a single arrangement that represents the motion space, where
each cell $c$ of the arrangement is extended with a DBG. We use the
adjacency-matrix storage format provided by \boost~\citelinks{boost}
to represent each DBG. Recall that for a graph with $n$ vertices such
as ours, an $n \times n$ matrix is used, where each element $a^c_{ij}$
of a DBG associated with cell $c$ is a Boolean flag that indicates
whether part $P_i$ collides with part $P_j$ when moved along any
direction $\vecd \in c$.
In particular we use the \ccode{adjacency\_matrix} class. It
implements the Boost Graph Library (\bgl)~\cite{sll-bgl-02} interface,
which supports, among the other, easy insertions of new edges into
existing graphs. Handling large assemblies with sparse blocking
relations may require different representations of DBGs to reduce
memory consumption.

We exploit the \Index{overlay} operation in this phase the third time
similar to its application in
Section~\ref{ssec:assem_plan:pairwise-ms-projection}. We traverse all
central projections in the set $\{Q_{ij}\,|\, 1 \leq i < j \leq n\}$,
and accumulate the result in the final motion-space arrangement. As
mentioned above in
Section~\ref{ssec:assem_plan:pairwise-sub-part-ms-construction}, when
the overlay operation progresses, new vertices, edges, and faces of the
\begin{wrapfigure}[7]{r}{3.5cm}
  \vspace{-15pt}
  \epsfig{figure=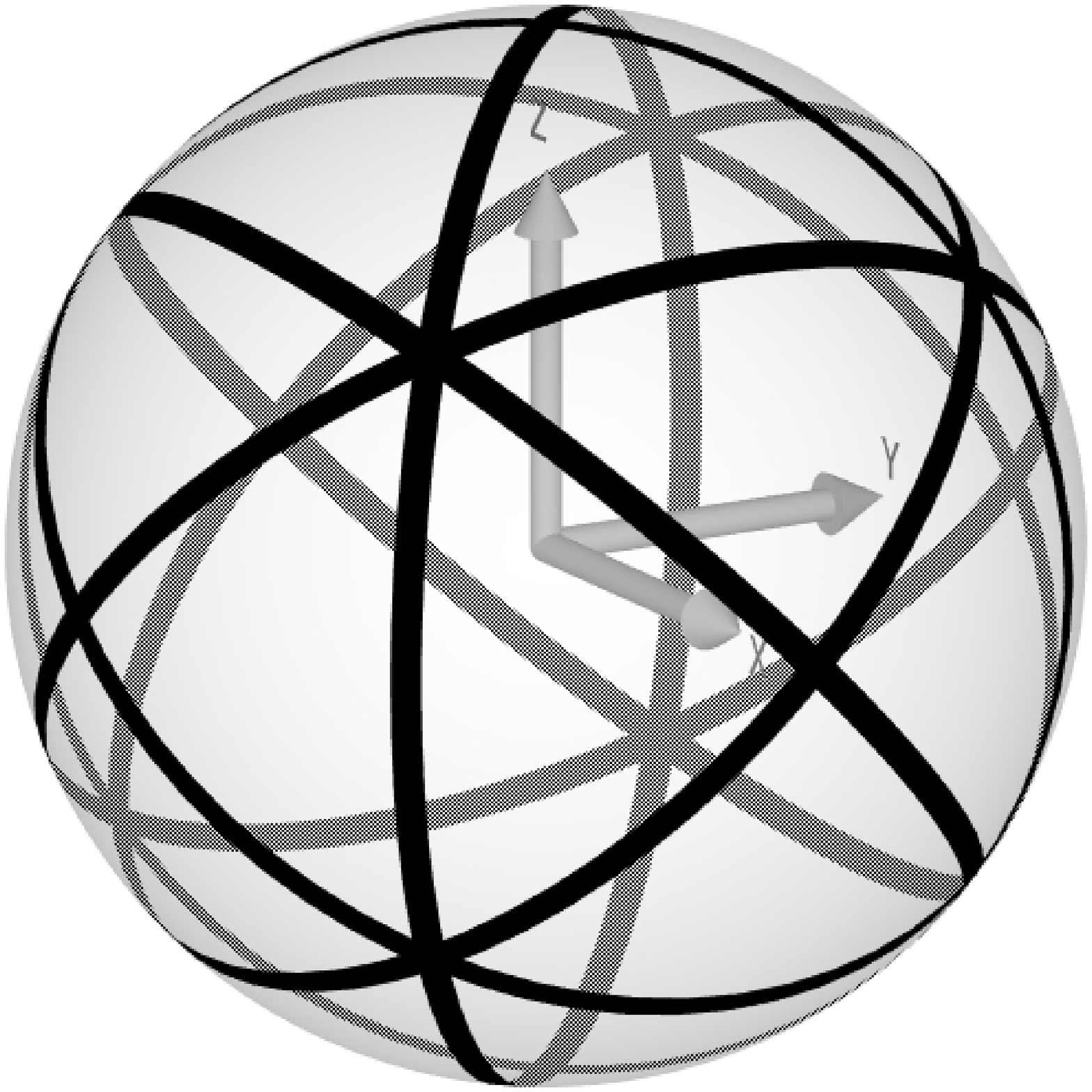,width=3.5cm,silent=}
\end{wrapfigure}
resulting arrangement are created. When a new cell $c$ is created as a
result of the overlay of a face $g$ in some projection $Q_{ij}$, and a
cell in the accumulating arrangement, the DBG associated with $c$ is
updated. That is, if the flag associated with $g$ is turned on, we
insert an edge between vertex $i$ and vertex $j$ into the DBG
associated with $c$.

Depicted on the left is the motion-space arrangement
computed by our program for the Split Star assembly.

\subsection{Motion-Space Processing}
\label{ssec:assem_plan:processing}
\begin{wraptable}[12]{l}{5cm}
  \vspace{-25pt}
  \caption[Split Star partitioning directions and corresponding subassemblies]%
            {\capStyle{The Split Star valid partitioning directions and
                corresponding subassemblies.}}
  \label{tab:split-start-result}
  \begin{tabular}{c|r@{,}r@{,}r|c}
    & \multicolumn{3}{c|}{\bf Direction} & {\bf Subset}\\
    \hline
    1. & $-1$ & $-1$ & $-1$ & \greenpart \bluepart\turquoisepart\\
    2. & $-1$ & $-1$ & $ 1$ & \redpart \bluepart\turquoisepart\\
    3. & $-1$ & $ 1$ & $-1$ & \greenpart \purplepart\turquoisepart\\
    4. & $-1$ & $ 1$ & $ 1$ & \redpart\purplepart\turquoisepart\\
    5. & $ 1$ & $-1$ & $-1$ & \greenpart\bluepart\yellowpart\\
    6. & $ 1$ & $-1$ & $ 1$ & \redpart\bluepart\yellowpart\\
    7. & $ 1$ & $ 1$ & $-1$ & \greenpart\purplepart\yellowpart\\
    8. & $ 1$ & $ 1$ & $ 1$ & \redpart\purplepart\yellowpart
  \end{tabular}
\end{wraptable}
We traverse all vertices, edges, and faces of the motion-space
arrangement in this order, and test the DBG associated with each cell
for strong connectivity using the \boost{} global function
\ccode{strong\_components()}. This function computes the strongly
connected components of a directed graph using Tarjan's algorithm
based on depth-first search\index{depth-first search|see{DFS}}
(\Index{DFS})~\cite{t-dfslg-72}. The set of constraints associated
with a vertex $v$ is a proper subset of the constraints associated with
the edges incident to $v$. Similarly, the set of constraints associated
with an edge $e$ is a proper subset of the constraints associated with
the two faces incident to $e$. Therefore, if the DBGs of all vertices
are strongly connected, we terminate with the conclusion that the

\begin{wrapfigure}[6]{r}{3.5cm}
  \epsfig{figure=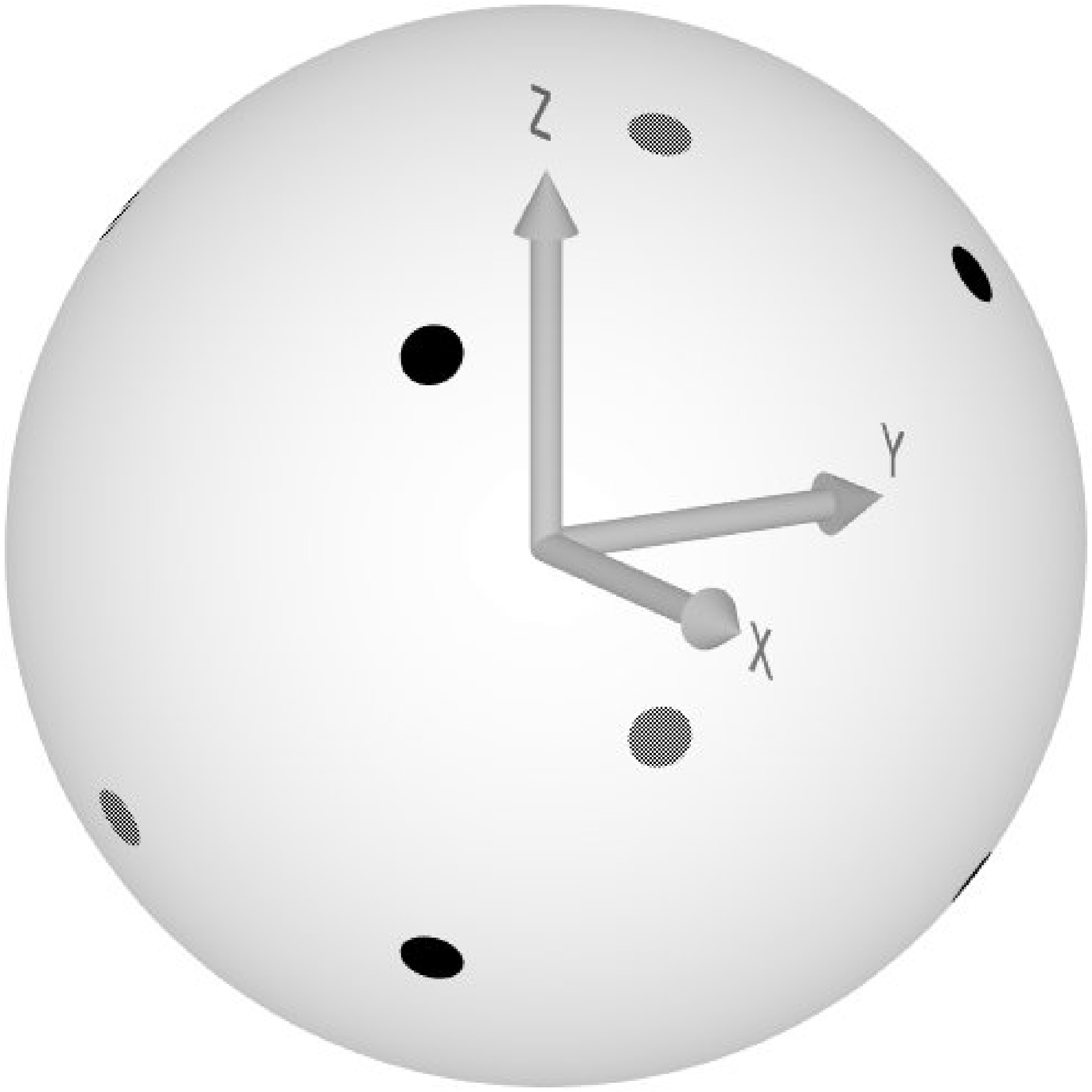,width=3.5cm,silent=}
\end{wrapfigure}
\noindent assembly
is interlocked. Similarly, if we are interested in finding
all solutions, and the DBGs of all edges are strongly connected, we
terminate, as no further solutions on faces exist.

For the Split Star assembly, our program successfully identifies all
the eight partitioning directions depicted on the right along with the
corresponding subset of parts listed in Table~\ref{tab:split-start-result}.

\section{Additional Optimization}
\label{sec:assem_plan:optimization}

The reflection of the sub-parts through the origin as described in
Section~\ref{ssec:assem_plan:sub-part-gaussian-map-reflection} has been
naively implemented. The computation is applied to the polyhedral-mesh
representation of each sub-part. An immediate optimization calls for
an application of the reflection operation directly on the
arrangements that represent the Gaussian map. We are planning to
introduce a generic implementation of the reflection operation that
operates on any applicable arrangement. This operation alters the
incidence relations between the arrangement features and their
geometric embeddings.
For each vertex, it negates its associated point, and inverts the
order of the halfedges incident to it. For each edge, it negates the
associated curve. For each face, it inverts the order of the halfedges
along its outer boundary. Similar to the \Index{overlay} operation (see
Section~\ref{ssec:aos:facilities:overlay}), where the user can provide
a set of ten functions, which are invoked when new vertices, edges,
and faces of the resulting arrangement are created, while the overlay
operation progresses, the user can provide a set of three functions
that are invoked when a new vertex, halfedge, and face are created,
while the reflection operation progresses. Extended data associated
with these types, such as a primal vertex associated with an
arrangement face as in the case of an arrangement representing a
Gaussian map, can easily be updated with the provision of an
appropriate function.

The trivial observation that $P \oplus (-Q) = -((-P) \oplus Q)$ leads to
another optimization. Instead of reflecting all sub-parts in the set
$\{P^i_k\,|\, i=1,2,\ldots,n, k=1,2,\ldots,m_i\}$, we reflect only the
sub-parts in the set $\{P^i_k\,|\, i=2,\ldots,n, k=1,2,\ldots,m_i\}$,
and compute only the pairwise sub-parts Minkowski sums in the set
$\{M^{ij}_{k\ell}\,|\, 1 \leq i < j \leq n, k=1,2,\ldots,m_j,
\ell=1,2,\ldots,m_j\}$, their central projection, and the union of
the appropriate projections to yield the set
$\{Q_{ij}\,|\, 1 \leq i < j \leq n\}$. Then, we apply the reflection
operation described above on each member of this set, and obtain the full
set of projections $\{Q_{ij}\,|\, i=1,2,\ldots,n, j=1,2,\ldots,n\}$.
The Boolean flag associated with a face of an arrangement
that represents a central projection is equal to the flag associated
with its reflection. In other words, a face of $Q_{ij}$ consists of 
directions that pierce $M_{ij}$, if and only if, its reflection in
$Q_{ji}$ consists of directions that pierce $M_{ji}$.

Phase~8 is purely topological. Thus, we do not expect the time
consumption of this phase to dominate the time consumption of the
entire process for any input. Nevertheless, it might be possible to
reduce its contribution to the total time consumption through
efficient testing for strong connectivity applied to all the
DBGs~\cite{kmw-cldgp-98}, exploiting the similarity between DBGs
associated with incident cells. Recall, that the set of arcs in a DBG
associated with a vertex $v$ is a subset of the set of arcs associated
with an edge incident to $v$. Similarly, the set of arcs in a DBG
associated with an edge $e$ is a subset of the set of arcs associated
with a face incident to $e$. The proposed technique reduces the
cost from $O(n^2)$ per DBG to an amortized cost of $O(n^{1.376})$,
where $n$ is the maximum number of arcs in any blocking graph. 

\section{Experimental Results}
\label{sec:assem_plan:experiments}
\begin{table}[t]
  \newlength{\cellwidth}\setlength{\cellwidth}{0.8cm}
  \begin{center}
    \caption[Time consumption of assembly partitioning of the Split Star]%
            {\capStyle{Time consumption (in seconds) of
             the execution of the eight phases applied to the Split Star
             assembly as input. Each one of the three rows refers to a
             different decomposition of the assembly.
             {\bf A} --- number of convex sub-parts per part.
             {\bf B} --- number of sub-part vertices per part.
             {\bf C} --- total number of convex sub-parts.
             {\bf D} --- total number of Minkowski sums.
             {\bf E} --- total number of arrangements of geodesic arcs embedded
             on the sphere constructed throughout the process.}}
    \label{tab:experiments}
    \begin{tabular}{r|r|r|r|r|c|r|r|r|r|r|r|r}
      \multicolumn{1}{p{\cellwidth}|}{{\bf ~~A}} &
      \multicolumn{1}{p{\cellwidth}|}{{\bf ~~B}} &
      \multicolumn{1}{p{\cellwidth}|}{{\bf ~~C}} &
      \multicolumn{1}{p{\cellwidth}|}{{\bf ~~D}} &
      \multicolumn{1}{p{\cellwidth}|}{{\bf ~~E}} &
      \multicolumn{1}{p{\cellwidth}|}{{\bf ~~~1}} &
      \multicolumn{1}{p{\cellwidth}|}{{\bf ~~~2}} &
      \multicolumn{1}{p{\cellwidth}|}{{\bf ~~~3}} &
      \multicolumn{1}{p{\cellwidth}|}{{\bf ~~~4}} &
      \multicolumn{1}{p{\cellwidth}|}{{\bf ~~~5}} &
      \multicolumn{1}{p{\cellwidth}|}{{\bf ~~~6}} &
      \multicolumn{1}{p{\cellwidth}|}{{\bf ~~~7}} &
      \multicolumn{1}{p{\cellwidth}}{{\bf ~~~8}}\\
      \hline
      \hline
      3 & 16 & 18 & 270 & 607 & NA & 0.01 & 0.04 & 2.38 & 0.41 & 2.05 &
      \multirow{3}*[0pt]{0.36} & \multirow{3}*[0pt]{0.01}\\
      5 & 22 & 30 & 750 & 1591 & NA & 0.01 & 0.05 & 5.03 & 1.09 & 7.07 & &\\
      8 & 32 & 48 & 1920 & 3967 & NA & 0.01 & 0.06 & 11.12 & 2.41 & 27.99 & &\\
      \hline
    \end{tabular}
  \end{center}
\end{table}
Our program can handle all inputs. However, we limit ourselves to a
representative set of test cases, where we compare the impact of
different decompositions on the process time consumption. The results
listed in Table~\ref{tab:experiments} were produced by experiments
conducted on a Pentium PC clocked at 1.7~GHz. In all three test cases
we use the Split Star assembly as input. Naturally, in all three cases
identical projections are obtained as the intermediate results of
Phase~6, hence the identical time consumption of the succeeding last
two phases. Evidently, it is desired to decompose each part into as
few as possible sub-parts with as small as possible number of
features. However, an automatic decomposition operation may require
large amount of resources to arrive at optimal or near optimal
decompositions. Notice that Phases~4 and~6 dominate the time
complexity. This is due to the large number of geometric predicates
that must be evaluated during the execution of the overlay operation.

\begin{savequote}[10pc]
  Seriousness is the only refuge of the shallow.
\sffamily
\qauthor{Oscar Wilde}
\end{savequote}
\chapter{Conclusion and Future Work}
\label{chap:conclusion}
In this thesis we show how a complete implementation of extendible 
arrangements\index{arrangement} with a rich set of operations enables
a broad spectrum of robust applications that solve problems arising in
domains such as motion planning, assembly planning, and solid
modeling. For example, we describe how arrangements embedded on
two-dimensional surfaces can be efficiently used to compute Minkowski
sums\index{Minkowski sum} of two polytopes in $\rrr$, which in turn,
and in conjunction with several other operations based on such
arrangements, can be used to partition an assembly with an infinite
translation motion. The rest of this chapter is devoted to future
prospects related to our research topics.

\section{Arrangements on Two-Dimensional Surfaces}
\label{sec:conclusion:arr-2d}
Constructing Minkowski Sums of polytopes in $\rrr$ has been
successfully attempted in the past. We introduce a robust, yet
efficient method. Table~\ref{tab:mink-time} shows that both our exact
methods outperform the other exact methods. However, we believe that
both of our methods, and in fact all \cgal{} based methods have great
potential for further improvements through future optimizations applied
to the infrastructure of \cgal{}, as \cgal{} is an evolving project.
While the space consumption of the \cgm{} method is greater than the
space consumption of the spherical Gaussian-map method, the table also
reveals that the \cgm{} method is currently more efficient than its
rival. We estimate that the gap will decrease, if not vanish, once all
optimizations for the \aos{} data-structure that are still pending are
implemented and enabled.

We are constantly striving to improve the quality of our infrastructure,
that is the \aos{} package. We have already identified few weak spots.
Eliminating them will increase the genericity, extendibility,
efficiency, and functionality of the package. We provide one example in
each category.

\subsection{Generic Observers}
\label{ssec:conclusion:arr-2d:Generic-observers}
There is a certain similarity between observers\index{observer}
(see Section~\ref{ssec:aos:facilities:notification}) and
visitors\index{visitor} (see
Section~\ref{ssec:aos:facilities:sweep-line}), as typically each of
their methods is triggered as a response to a certain event --- a
member of a pre-determined list of events. Technically, the main
difference between them is that observers define a one-to-many mapping
between objects, while visitors define a one-to-one
mapping.\footnote{They also differ in essence.
While an observer typically implements a notifier, a visitor is usually
a coherent part of an algorithm based on a fundamental and flexible
framework.~\cite{ghjv-dp-95}} Recall, for example, that a single arrangement
may register many observers, but it is only natural to relate a single
visitor to a specific algorithmic framework in order to realize a certain
concrete algorithm. Consequently, arrangement observers are derived from a
common base class, and their methods must be virtual. This is how
modules, which are closed for modification, are extended using
\Index{object-oriented programming}. However, composability of such
modules is limited, since independently produced modules generally do
not agree on common abstract interfaces from which supplied types must
inherit. In addition, when techniques from the object-oriented
programming and the generic programming\index{generic programming}
paradigms are mixed, they often clash. There are known methods to
replace lists of objects, derived from a common base class, and linked
during run time, with a list of syntactically unrelated objects
concatenated during compile time (coded using a generic programming
technique)~\cite[Chapter~3]{a-mcd-01}.
Nevertheless, we would like to simultaneously enjoy the benefits of both
the object-oriented and the generic programming paradigms, that is,
to enable the immediate production of composable modules that support
dynamic polymorphism. A very important research direction in our context
is to explore these possibilities, perhaps pushing the limits of the \Cpp{}
programming language along the way.

\subsection{Property Maps}
\label{ssec:conclusion:arr-2d:attribute-maps}
In many cases we need to associate values (called ``properties'') with
the vertices, the halfedges, and the faces of the arrangement. In
addition, it is often necessary to associate multiple properties
with each vertex, edge, or face; this is what \boost{} literature
refers to as multi-parameterization. \bgl{}~\cite{sll-bgl-02}
graph classes have template parameters for vertex and edge
``properties''. A property specifies the parameterized type of the
property and also assigns an identifying tag to the property.

There are various ways to associate properties with arrangement cells.
One option is to extend the geometric types of the kernel, as the
kernel is fully adaptable and extensible~\cite{hhkps-aegk-07}. However,
this indiscriminating extension may lead to an undue
space-consumption, as every geometric object is extended, regardless
of its use.\footnote{It also requires nontrivial knowledge about the
kernel structure and the techniques to extend it.} Another option, is
to extend the vertex, halfedge, or face records of the
\dcel\index{DCEL@\dcel};
see Section~\ref{ssec:aos:facilities:extension}. This may also lead to
excessive space-consumption, for example, when the data associated
with a halfedge is in fact tied to the embedded geometric curve. In
this case the data, or at least a handle to the data, must be stored
twice in both twin halfedges. A third option is to extend the curve
(or point) types defined by the geometry-traits class; see
Section~\ref{ssec:aos:architecture:data-structure}. But even this
option leads to unjustified space-consumption, when only a limited
number of arrangement features are associated with real data. In such
cases it is advantageous to use external search structures that map
individual arrangement features with their data.

Designing a useful and convenient interface, while taking all
considerations above into account, is a research topic on its own, which
may further push the limits of good usage of the \Cpp{} programming
language.

\subsection{Point Location for Surfaces}
\label{ssec:conclusion:arr-2d:Point Location}
Point location is one of the most fundamental operations applied to
arrangements; see Section~\ref{ssec:aos:facilities:point-location}.
The contest between the different point-location strategies for
arrangement embedded in the plane was settled in favor of the
``landmark'' variants for many types of
arrangements~\cite{hh-esplp-08}. Unfortunately, at this point, these
strategies cannot be applied on arrangements embedded on 
surfaces other than the plane. This is due to limitations that arise
from the specific implementations of geometry traits that support
these types of embedded surfaces, e.g., the geometry traits that
handles geodesic arcs embedded on the sphere; see
Section~\ref{sec:aos:geodesics}. The problem should be attacked from
both its ends. That is, we can try to enhance the geometry traits
implementations, and add all ingredients required by the concept
\concept{ArrangementLandmarksTraits\_2} as defined today, and at the
same time, try to come up with alternative, perhaps similar,
strategies that induce different requirements that are easier to
satisfy by the geometry-traits classes.

\subsection{Geometry-Traits Models}
\label{ssec:conclusion:arr-2d:geometry-traits-models}
Continuous and steady effort is made to further extend the arsenal of
geometry-traits models, or simply to improve the existing ones; see
Section~\ref{ssec:aos:geometry-traits:models}. Supporting arrangements
induced by rich families of curves opens the door for numerous
applications.

The dominant bottleneck of all applications mentioned in this thesis
is the application of the geometry operations implemented in the
geometry-traits classes. Expediting their performance, while containing
the growth of their memory footprint is always desired. For example,
the arrangement package provides two traits classes that handle line
segments. Both are parameterized by a geometric kernel; see
Section~\ref{sssec:intro:cgal:content}. Segments defined by most
\cgal{} kernels are represented only by their two endpoints. When a
segment is split several times, the bit-length needed to represent the
coordinates of its endpoints may grow exponentially
(see~\cite{fwh-cfpeg-04} for a discussion), which may significantly
slow down the computation. Therefore, one of the two traits classes
represents a segment by its supporting line in addition to its two
endpoints. When the traits class computes an intersection point of two
line segments, it uses the coefficients of their supporting lines.
When a segment is split at an intersection point, the underlying line
of the two resulting sub-segments remains the same, and only its
endpoints are updated. This traits class thus overcomes the undesired
effect of cascading intersection-point representation, as described
above, at the account of a larger memory footprint. A similar idea can
be applied to the traits class that handles geodesic arcs embedded on
the sphere. An implementation of a geometry-traits that handles such
arcs that stores the projections of all the arrangement geometric
features once calculated, and retrieves them when subsequently needed,
has great potential to reduce time consumption at the price of growth
in space consumption.

Another direction is to expand existing implementations to meet the
requirements of the various concepts in the hierarchy described in
Chapter~\ref{chap:aos}. Consider, for example, the geometry traits
class that handles geodesic arcs embedded on the sphere. Currently,
it supports the basic operations required to construct and maintain
arrangements induced by such arcs. As mentioned in the previous section,
it might be possible to enhance it to enable the use of one or more
of the variants of the ``landmark'' point-location strategies.
Similarly, enabling Boolean set operations of spherical patches bounded
by geodesic arcs embedded on the sphere requires the provision of few
additional operations by the traits class.

\section{Three-Dimensional Arrangements}
\label{sec:conclusion:arr-3d}
Consider the following task: Given a set
$\calS = \{S_1,S_2,\ldots,S_n\}$ of two-dimensional surfaces
in $\rrr$, construct the three-dimensional arrangement $\calA(\calS)$
induced by $\calS$. Fulfilling this task in an efficient, complete,
and robust manner has not been attempted yet, and is considered challenging.
Implementing various strategies of point-location that operate on
arrangements in $\rrr$ and a plane-sweep and zone-construction frameworks
for such a data structure is greatly desired, but extremely ambitious.
In analogy to two-dimensional arrangements, a generic implementation of
a plane-sweep framework can enable various operations, such as the overlay
of spatial subdivisions and ordinary and regularized Boolean set
operations of point sets bounded by general algebraic surfaces. These
operations, in turn, can enable the implementation of a multitude of
applications.

Arrangements embedded on two-dimensional surfaces can be used as 
building blocks in the implementation of a data structure that
represents a three-dimensional arrangement~\cite{w-ieaem-07}. We can
consider each surface $S_k$ separately, and construct the arrangement
$\calA_k = \calA_k(\calS)$ induced by intersection curves between
$S_k$ and $\calS \setminus S_k$ embedded on $\calS$. The arrangements
$\calA_1$,$\calA_2$,\ldots,$\calA_n$ can subsequently be connected
together to properly represent the spatial subdivision $\calA(\calS)$.

\section{Boolean Set-Operations\index{Boolean set-operations}}
\label{sec:conclusion:bso}
In some sense and to some extent this thesis attempts to close gaps
between theoretical results and practical needs. It is not accidental
that great parts of the thesis are closely related to \cgal, as one of
the goals of \cgal{} stated in the Introduction chapter of the thesis
is to translate (theoretical) results into useful, reliable, and
efficient programs for industrial and academic applications. Evidently,
the Boolean set operations package of \cgal{}, which is based on the
\aos{} package (see Section~\ref{ssec:aos:applications:bso})  is one
of the most popular packages among \cgal{} packages in the commercial
market. Naturally, we would like to continue improving this package.
The problems addressed in the next two subsections were raised during
the $3^{rd}$ \cgal{} User Workshop~\cite{w-ucrpg-08}~\citelinks{agilent}.

\subsection{Fixing the Data}
\label{ssec:conclusion:bso:data}
Input data of Boolean set operations, namely, a set of one or more
polygons, used in real-world applications is occasionally corrupted, as
it originates from measuring devices that are susceptive to noise and
physical disturbances. In some other cases, it contains many
degeneracies, which either disable computations based on
fixed-precision arithmetic, or slow down further computation using
exact geometric computation.

\subsubsection{Invalid Data}
\label{ssec:conclusion:bso:data:invalid}
\begin{wrapfigure}[10]{r}{4cm}
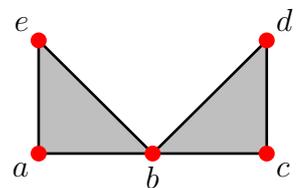

  \vspace{-10pt}
  \centerline{
    \pspicture[](-1.7,0)(1.7,1.7)
    \psset{unit=1cm,linewidth=1pt}
    \pspolygon*[linecolor=lightgray](-1.5,0)(1.5,0)(1.5,1.5)(0,0)(-1.5,1.5)
    \pspolygon(-1.5,0)(1.5,0)(1.5,1.5)(0,0)(-1.5,1.5)
    \pscircle*[linecolor=red](-1.5,0){3pt}
    \pscircle*[linecolor=red](1.5,0){3pt}
    \pscircle*[linecolor=red](1.5,1.5){3pt}
    \pscircle*[linecolor=red](-1.5,1.5){3pt}
    \pscircle*[linecolor=red](0,0){3pt}
    \uput[-135]{0}(-1.5,0){$a$}
    \uput[-90]{0}(0,0){$b$}
    \uput[-45]{0}(1.5,0){$c$}
    \uput[45]{0}(1.5,1.5){$d$}
    \uput[135]{0}(-1.5,1.5){$e$}
    \endpspicture
  }
  \caption[A relatively simple polygon]{\capStyle{A {\em relatively simple}
  polygon that is not simple, given by its boundary $\{a,b,c,d,b,e\}$.}}
\end{wrapfigure}
A polygon $P$ is said to be {\em simple}\index{polygon!simple}
(or Jordan)\index{polygon!Jordan}, if the only points of the plane
belonging to two edges of $P$ are vertices of consecutive edges
$P$~\citelinks{wolfram-simple-polygon}. Namely, no two edges intersect,
except for every two consecutive edges, which share one endpoint. A
simple polygon is topologically equivalent to a disk. A polygon $P$ is
said to be {\em weakly simple}\index{polygon!weakly simple}, if the
chain of the edges of $P$ does not cross itself. A polygon $P$ is said
to be {\em relatively simple}\index{polygon!relatively simple}, if it
is weakly simple and the edges of $P$ do not intersect in their
relative interior. Observe, that a relatively simple polygon, the
vertices of which appear only once in the boundary (the degree of
each vertex is two), is simple.

\begin{wrapfigure}[14]{r}{4cm}
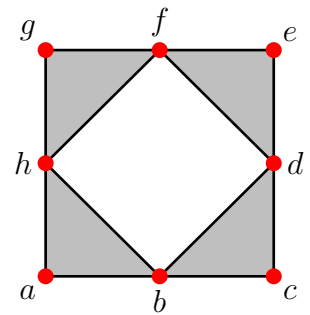

  \vspace{-10pt}
  \centerline{
    \pspicture[](-1.9,-1.9)(1.9,1.9)
    \psset{unit=1cm,linewidth=1pt}
    \pspolygon*[linecolor=lightgray](-1.5,-1.5)(1.5,-1.5)(1.5,1.5)(-1.5,1.5)
    \pspolygon*[linecolor=white](0,-1.5)(1.5,0)(0,1.5)(-1.5,0)
    \pspolygon(-1.5,-1.5)(1.5,-1.5)(1.5,1.5)(-1.5,1.5)
    \pspolygon(0,-1.5)(1.5,0)(0,1.5)(-1.5,0)
    \pscircle*[linecolor=red](-1.5,-1.5){3pt}
    \pscircle*[linecolor=red](0,-1.5){3pt}
    \pscircle*[linecolor=red](1.5,-1.5){3pt}
    \pscircle*[linecolor=red](1.5,0){3pt}
    \pscircle*[linecolor=red](1.5,1.5){3pt}
    \pscircle*[linecolor=red](0,1.5){3pt}
    \pscircle*[linecolor=red](-1.5,1.5){3pt}
    \pscircle*[linecolor=red](-1.5,0){3pt}
    \uput[-135]{0}(-1.5,-1.5){$a$}
    \uput[-90]{0}(0,-1.5){$b$}
    \uput[-45]{0}(1.5,-1.5){$c$}
    \uput[0]{0}(1.5,0){$d$}
    \uput[45]{0}(1.5,1.5){$e$}
    \uput[90]{0}(0,1.5){$f$}
    \uput[135]{0}(-1.5,1.5){$g$}
    \uput[180]{0}(-1.5,0){$h$}
    \endpspicture
  }
  \caption[A polygon with holes]{\capStyle{A polygon with a hole
  given by its outer boundary $\{a,b,c,d,e,f,g,h\}$ and its hole $\{h,f,d,b\}$.}}
\end{wrapfigure}
Input data for any Boolean set operation represents points set that
may be bounded or unbounded, and may have holes. Items of such input
take the form of general polygons or general polygons with holes with
well-defined interiors and exteriors. A valid polygon must be weakly
simple (but not necessarily simple) and its vertices must be ordered
in counterclockwise direction around its interior. A valid polygon
with holes that represents a bounded point set, has an outer boundary
represented as a weakly simple (but not necessarily simple) polygon,
the vertices of which are oriented in counterclockwise order around
its interior. In addition, the set may contain holes, where each hole
is represented as a simple polygon, the vertices of which are oriented
in clockwise order around the interior of the hole.
Note that an unbounded polygon without holes spans the entire plane.
Vertices of holes may coincide with vertices of the boundary.

\begin{figure}[!htp]
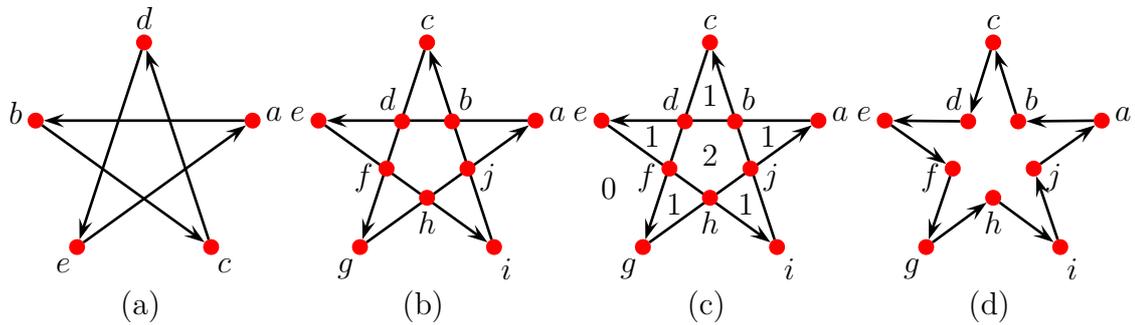

  \centerline{
  \begin{tabular}{cccc}
    \pspicture[](-1.7,-1.6)(1.6,2)
    \psset{unit=1cm,linewidth=1pt}
    \rput{18}(0,0){\cnode*[linecolor=red](1.5,0){3pt}{a}}
    \rput{90}(0,0){\cnode*[linecolor=red](1.5,0){3pt}{d}}
    \rput{162}(0,0){\cnode*[linecolor=red](1.5,0){3pt}{b}}
    \rput{234}(0,0){\cnode*[linecolor=red](1.5,0){3pt}{e}}
    \rput{306}(0,0){\cnode*[linecolor=red](1.5,0){3pt}{c}}
    \ncline{->}{a}{b}
    \ncline{->}{b}{c}
    \ncline{->}{c}{d}
    \ncline{->}{d}{e}
    \ncline{->}{e}{a}
    \rput{18}(0,0){\uput[0]{-18}(1.5,0){$a$}}
    \rput{90}(0,0){\uput[0]{-90}(1.5,0){$d$}}
    \rput{162}(0,0){\uput[0]{-162}(1.5,0){$b$}}
    \rput{234}(0,0){\uput[0]{-234}(1.5,0){$e$}}
    \rput{306}(0,0){\uput[0]{-306}(1.5,0){$c$}}
    \endpspicture &
    \pspicture[](-1.7,-1.6)(1.6,2)
    \psset{unit=1cm,linewidth=1pt}
    \rput{18}(0,0){\cnode*[linecolor=red](1.5,0){3pt}{a}}
    \rput{90}(0,0){\cnode*[linecolor=red](1.5,0){3pt}{b}}
    \rput{162}(0,0){\cnode*[linecolor=red](1.5,0){3pt}{c}}
    \rput{234}(0,0){\cnode*[linecolor=red](1.5,0){3pt}{d}}
    \rput{306}(0,0){\cnode*[linecolor=red](1.5,0){3pt}{e}}
    \ncline{->}{a}{c}
    \ncline{->}{c}{e}
    \ncline{->}{e}{b}
    \ncline{->}{b}{d}
    \ncline{->}{d}{a}
    \rput{18}(0,0){\uput[0]{-18}(1.5,0){$a$}}
    \rput{90}(0,0){\uput[0]{-90}(1.5,0){$c$}}
    \rput{162}(0,0){\uput[0]{-162}(1.5,0){$e$}}
    \rput{234}(0,0){\uput[0]{-234}(1.5,0){$g$}}
    \rput{306}(0,0){\uput[0]{-306}(1.5,0){$i$}}
    \rput{54}(0,0){\pscircle*[linecolor=red](0.56,0){3pt}\uput[0]{-54}(0.56,0){$b$}}
    \rput{126}(0,0){\pscircle*[linecolor=red](0.56,0){3pt}\uput[0]{-126}(0.56,0){$d$}}
    \rput{198}(0,0){\pscircle*[linecolor=red](0.56,0){3pt}\uput[0]{-198}(0.56,0){$f$}}
    \rput{270}(0,0){\pscircle*[linecolor=red](0.56,0){3pt}\uput[0]{-270}(0.56,0){$h$}}
    \rput{342}(0,0){\pscircle*[linecolor=red](0.56,0){3pt}\uput[0]{-342}(0.56,0){$j$}}
    \endpspicture &
    \pspicture[](-1.7,-1.6)(1.6,2)
    \psset{unit=1cm,linewidth=1pt}
    \rput{18}(0,0){\cnode*[linecolor=red](1.5,0){3pt}{a}}
    \rput{90}(0,0){\cnode*[linecolor=red](1.5,0){3pt}{b}}
    \rput{162}(0,0){\cnode*[linecolor=red](1.5,0){3pt}{c}}
    \rput{234}(0,0){\cnode*[linecolor=red](1.5,0){3pt}{d}}
    \rput{306}(0,0){\cnode*[linecolor=red](1.5,0){3pt}{e}}
    \ncline{->}{a}{c}
    \ncline{->}{c}{e}
    \ncline{->}{e}{b}
    \ncline{->}{b}{d}
    \ncline{->}{d}{a}
    \rput{18}(0,0){\uput[0]{-18}(1.5,0){$a$}}
    \rput{90}(0,0){\uput[0]{-90}(1.5,0){$c$}}
    \rput{162}(0,0){\uput[0]{-162}(1.5,0){$e$}}
    \rput{234}(0,0){\uput[0]{-234}(1.5,0){$g$}}
    \rput{306}(0,0){\uput[0]{-306}(1.5,0){$i$}}
    \rput{54}(0,0){\pscircle*[linecolor=red](0.56,0){3pt}\uput[0]{-54}(0.56,0){$b$}}
    \rput{126}(0,0){\pscircle*[linecolor=red](0.56,0){3pt}\uput[0]{-126}(0.56,0){$d$}}
    \rput{198}(0,0){\pscircle*[linecolor=red](0.56,0){3pt}\uput[0]{-198}(0.56,0){$f$}}
    \rput{270}(0,0){\pscircle*[linecolor=red](0.56,0){3pt}\uput[0]{-270}(0.56,0){$h$}}
    \rput{342}(0,0){\pscircle*[linecolor=red](0.56,0){3pt}\uput[0]{-342}(0.56,0){$j$}}
    \rput{0}(0,0){$2$}
    \rput{18}(0,0){\rput{-18}(0.8,0){$1$}}
    \rput{90}(0,0){\rput{-90}(0.8,0){$1$}}
    \rput{162}(0,0){\rput{-162}(0.8,0){$1$}}
    \rput{234}(0,0){\rput{-234}(0.8,0){$1$}}
    \rput{306}(0,0){\rput{-306}(0.8,0){$1$}}
    \rput{198}(0,0){\rput{-198}(1.4,0){$0$}}
    \endpspicture &
    \pspicture[](-1.7,-1.6)(1.6,2)
    \psset{unit=1cm,linewidth=1pt}
    \rput{18}(0,0){\cnode*[linecolor=red](1.5,0){3pt}{a}}
    \rput{54}(0,0){\cnode*[linecolor=red](0.56,0){3pt}{b}}
    \rput{90}(0,0){\cnode*[linecolor=red](1.5,0){3pt}{c}}
    \rput{126}(0,0){\cnode*[linecolor=red](0.56,0){3pt}{d}}
    \rput{162}(0,0){\cnode*[linecolor=red](1.5,0){3pt}{e}}
    \rput{198}(0,0){\cnode*[linecolor=red](0.56,0){3pt}{f}}
    \rput{234}(0,0){\cnode*[linecolor=red](1.5,0){3pt}{g}}
    \rput{270}(0,0){\cnode*[linecolor=red](0.56,0){3pt}{h}}
    \rput{306}(0,0){\cnode*[linecolor=red](1.5,0){3pt}{i}}
    \rput{342}(0,0){\cnode*[linecolor=red](0.56,0){3pt}{j}}
    \ncline{->}{a}{b}
    \ncline{->}{b}{c}
    \ncline{->}{c}{d}
    \ncline{->}{d}{e}
    \ncline{->}{e}{f}
    \ncline{->}{f}{g}
    \ncline{->}{g}{h}
    \ncline{->}{h}{i}
    \ncline{->}{i}{j}
    \ncline{->}{j}{a}
    \rput{18}(0,0){\uput[0]{-18}(1.5,0){$a$}}
    \rput{90}(0,0){\uput[0]{-90}(1.5,0){$c$}}
    \rput{162}(0,0){\uput[0]{-162}(1.5,0){$e$}}
    \rput{234}(0,0){\uput[0]{-234}(1.5,0){$g$}}
    \rput{306}(0,0){\uput[0]{-306}(1.5,0){$i$}}
    \rput{54}(0,0){\uput[0]{-54}(0.56,0){$b$}}
    \rput{126}(0,0){\uput[0]{-126}(0.56,0){$d$}}
    \rput{198}(0,0){\uput[0]{-198}(0.56,0){$f$}}
    \rput{270}(0,0){\uput[0]{-270}(0.56,0){$h$}}
    \rput{342}(0,0){\uput[0]{-342}(0.56,0){$j$}}
    \endpspicture\\
    (a) & (b) & (c) & (d)
  \end{tabular}
  }
  \caption[A self crossing polygon]{\capStyle{(a) A self crossing polygon
  given by $\{a,b,c,d,e\}$. (b) The \arr{} data structure constructed
  from the polygon edges. (c) The \arr{} data structure with updated
  face winding-numbers. (d) The \arr{} data structure with internal
  edges removed.}}
  \label{fig:winding}
\end{figure}

As mentioned above, real-world data is often corrupted. Naturally, passing
invalid polygons (polygons with holes, respectively) as input to a
Boolean set operation must be avoided. Apparently, automatically
``fixing'' corrupted data, that is, converting invalid input polygons
or polygons with holes to valid ones, is not a simple task. Consider,
for example, the self intersecting star depicted in
Figure~\ref{fig:winding}(a). A point is considered inside the point
set, if and only if the number of counterclockwise turns the oriented
boundary makes around the point, also called the winding number, is
greater than zero. It can be efficiently calculated using an \arr{}
data structure as follows. We extend each halfedge $h$ with a Boolean
flag that indicates whether the winding number increases or decreases
when we cross $h$, that is, when we move from the face incident to $h$
to the face incident to the twin of $h$. Observe that twin halfedges
always have opposite flag values. We extend each face with an integer
that counts the winding number of every point in the face. Finally, we
apply a Breath-First Search\index{breath-first search|see{BFS}}
(\Index{BFS}) on all the arrangement faces starting from the unbounded
face and updating the face counters as we cross halfedges. The \bgl{},
for example, can be employed for this task. Figure~\ref{fig:winding}
illustrates the process. Once the winding number of every face is
updated, we remove internal edges, as they are redundant, and convert
the arrangement back into a valid polygon or polygon with holes.

The problem becomes more complicated when the input polygons or
polygons with holes violate several validity properties at the same
time. Naturally, converting corrupted data into valid data consumes
time. The challenge is to perform this task flawlessly and efficiently,
while presenting a convenient interface to the user.

\subsubsection{Degenerate Data}
\label{ssec:conclusion:bso:data:degenerate}
In computational geometry there are two main techniques to eliminate
degeneracies and near-degeneracies. One is
\Index{snap rounding}~\cite{gm-rad-95,gght-srlse-97,h-psifp-99} and
the other is \Index{controlled
perturbation}~\cite{hl-cpac-04,mo-recgc-06}. Both techniques aim at
processing geometric data, e.g., curves of the boundaries of general
polygons, to yield new data that can be further robustly and more
efficiently processed, perhaps using only limited precision.
Traditionally, snap rounding has been applied to linear objects
embedded in the plane~\cite{hp-srr-01,hp-isr-02}, where it replaces
sets of linear segments with sets of polylines. It can be extended
though to other types of curves in the plane, such as \bez{}
curves~\cite{ekw-srbc-07}, or even other curve types embedded on other
surfaces that have a well defined grid (and perhaps other properties),
such as geodesic arcs embedded on the sphere. Controlled perturbation
has even a larger spectrum of applicable platforms. With respect to our
polygons, applying any one of the two techniques above may result with
(partially) overlapping curves, which belong to two different polygons,
respectively. In this case, we need to merge the incremental winding
contributions of the original curves mentioned above.

\subsection{Improving the Efficiency}
\label{ssec:conclusion:bso:efficiency}
The \cBso{} package provides efficient operations that compute the
regularized union or the regularized intersection of a set of input
polygons. There is no restriction on the polygons in the set;
naturally, they may intersect each other. The package also provides an
efficient predicate that determines whether all polygons in a given
set intersect.

\begin{wrapfigure}[12]{r}{4cm}
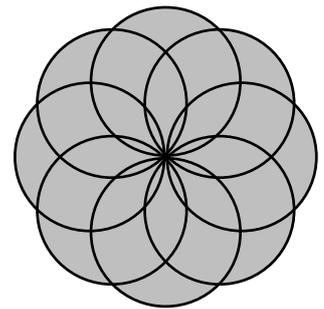

  \centerline{
    \pspicture[](-2,-2)(2,2)
    \psset{unit=1cm,linewidth=1pt}
    \rput{0}(0,0){\pscircle*[linecolor=lightgray](1,0){1cm}}
    \rput{45}(0,0){\pscircle*[linecolor=lightgray](1,0){1cm}}
    \rput{90}(0,0){\pscircle*[linecolor=lightgray](1,0){1cm}}
    \rput{135}(0,0){\pscircle*[linecolor=lightgray](1,0){1cm}}
    \rput{180}(0,0){\pscircle*[linecolor=lightgray](1,0){1cm}}
    \rput{225}(0,0){\pscircle*[linecolor=lightgray](1,0){1cm}}
    \rput{270}(0,0){\pscircle*[linecolor=lightgray](1,0){1cm}}
    \rput{315}(0,0){\pscircle*[linecolor=lightgray](1,0){1cm}}
    \rput{0}(0,0){\pscircle(1,0){1cm}}
    \rput{45}(0,0){\pscircle(1,0){1cm}}
    \rput{90}(0,0){\pscircle(1,0){1cm}}
    \rput{135}(0,0){\pscircle(1,0){1cm}}
    \rput{180}(0,0){\pscircle(1,0){1cm}}
    \rput{225}(0,0){\pscircle(1,0){1cm}}
    \rput{270}(0,0){\pscircle(1,0){1cm}}
    \rput{315}(0,0){\pscircle(1,0){1cm}}
    \endpspicture
  }
  \caption[The union of eight discs]{\capStyle{The union of eight discs.}}
\end{wrapfigure}
There are at least three different methods to compute the union of a set
of polygons $P_1, \ldots P_m$. We can do it incrementally as follows.
At each step we compute the union of $S_{k-1} = \bigcup_{i=1}^{k-1}{P_i}$
with $P_{k}$ and obtain $S_k$. A second option is to use a
divide-and-conquer approach. First, we divide the set of polygons into
two subsets. Then, we compute the union of each subset recursively, and
obtain the partial results in $S_1$ and $S_2$, respectively. Finally, we
compute the union $S_1 \cup S_2$. A third option aggregately computes the
union of all polygons. We construct an arrangement inserting the polygon
edges at once, utilizing the sweep-line\index{sweep line} framework,
(see Section~\ref{ssec:aos:facilities:sweep-line}) and extract the result
from the arrangement. Similarly, it is also possible to aggregately compute
the intersection $\bigcap_{i=1}^{m}{P_i}$ of a set of input polygons.

The incremental method is more efficient for a small (constant) size of
input polygons, and the aggregate method is more efficient for sparse
polygons with a relatively small number of intersections. It is also
possible to mix between the three methods, reaping the benefits of them
all. We would like to figure out what are the exact conditions that
should be used to determine when to use each method, or when to switch
from one to another.

\subsection{Non Regularized Operations}
\label{ssec:conclusion:bso:nonreg-opts}
The \cgal{} package \cNefii{} supports ordinary set-operations on point
sets in $\rr$~\cite{cgal:s-bonp2-07}. The point-set operands and results
are rectilinear polygonal model. Such a point set can be
defined by a finite set of open halfspaces, or obtained by set complement
and set intersection operations. The package supports operations that
consume and produce linear polygons defined by linear edges. The Boolean set
operations package, on the other hand, supports only regularized
set-operations, but the operations consume and produce general polygons.
Recall, that a general polygon is a point set in $\rr$ that has a
topology of a polygon, but its boundary edges map to arcs of curves,
which are not necessarily linear. Extending the package to support not
only regularized operations, but also ordinary ones, will make it useful
for more applications; see, for example,
Section~\ref{ssec:assem_plan:pairwise-ms-projection}.

\subsection{Operating in 3-Space}
\label{ssec:conclusion:bso:space}
Boolean set operations are intuitive and therefore popular in many
fields. For example, \Index{CSG} is a representation model for solids
based on Boolean set operations. Solids represented using CSG result
from Boolean set operations applied to elementary solids called
primitives, e.g., cubes, spheres, cones, and cylinders. A CSG solid
is represented in a tree structure, where the leaves represent
primitives, and internal nodes represent Boolean operations.

The \cgal{} package \cNefiii{} supports ordinary set-operations on point
sets in $\rrr$~\cite{cgal:hk-bonp3-07}. Similar to the planar case, the
package supports operations that consume and produce (linear)
polyhedra defined by (linear) halfspaces. Having the ability to
construct and maintain arrangements in $\rrr$, will enable the
development of a new package, that will support either regularized, or
even non-regularized, robust Boolean set-operations that consume and
produce general (curved) polyhedra in $\rrr$, the boundaries of which
are general surfaces. Many fundamental problems in solid modeling, motion
planning, and other domains, can benefit from such a package.

\section{Collision Detection}
\label{sec:conclusion:cd}
One possible progression of the collision detection algorithm and its
implementation described in Section~\ref{sec:mscn:3d_col_det} is
a complete integrated framework that answers proximity queries about the
relative placement of polytopes that undergo rigid motions including
{\em rotation}. The framework may use either the spherical or the
cubical Gaussian-map to represent polytopes. The interface of these
two data structures should be consolidated to allow rapid interchanging.

Some of the methods we foresee compute only those portions of the
Minkowski sum that are absolutely necessary, making our approach even
more competitive. Briefly, instead of computing the Minkowski sum of
$P$ and $-Q$, we walk simultaneously on the two respective \cgm's,
producing one feature of the Minkowski sum at each step of the
walk. Such a strategy could be adapted to the case of rotation by
rotating the trajectory of the walk, keeping the \cgm{} of $-Q$
intact, instead of rotating the \cgm{} itself.

\section{Reflection Mapping and GIS}
\label{sec:conclusion:reflection-mapping}
\begin{wrapfigure}{r}{5cm}
  \vspace{-12pt}
  \epsfig{figure=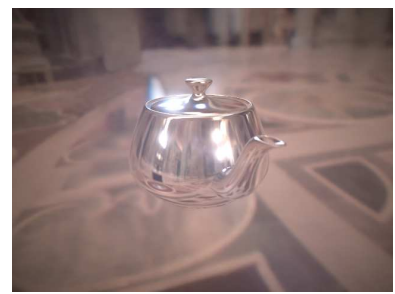,width=5cm,silent=}
    \caption[Environment mapping]{\capStyle{Environment of the St.\ Peters
    Cathedral mapped on a teapot with a silver material applied. Taken in
    RTHDRIBLE~\citelinks{rthdribl}.}}
  \label{fig:env-map}
\end{wrapfigure}
We have developed a new data structure that can be used to construct
and maintain cubical Gaussian-maps and compute Minkowski sums of pairs
of polytopes represented by the new data structure; see
Chapter~\ref{chap:mink-sums-construction}. The name of the data
structure is \cCgm{}, and we are considering introducing a package by
the same name to a prospective release of \cgal. The implementation is
generic and can be used for other purposes, where six planar
subdivisions embedded on a unit cube and stitched properly at the
edges of the cube is useful, for example Cubical Environment Mapping;
see, e.g.,~\cite[Section~16.6]{fvfh-cgppc-95}.

In computer graphics, {\em \Index{reflection mapping}} is an efficient
method of simulating a complex mirroring surface by means of a
precomputed texture image. The texture is used to store the image of the
environment surrounding the rendered object. The surrounding environment
can be represented, constructed, maintained, and stored in several ways;
the most common ones are the Spherical Environment Mapping in which a
single texture contains the image of the surrounding as reflected on a
mirror ball, or the Cubical Environment Mapping in which the environment
is unfolded onto the six faces of a cube and stored therefore as six
square textures.

Reflection Mapping can be categorized as some sort of a
\Index{geographic information system} (GIS). There are two broad
methods used to store data in a GIS: Raster and Vector. In a GIS 
data is often related from different sources possibly of different
storing types. Regardless of whether a single arrangement embedded on
a sphere, or six arrangements embedded on the cube, are concerned, the
connection to GIS is clear --- both data-structures can accommodate
geographic vector data in a natural way.

\section{Exact Complexity of Minkowski Sums}
\label{sec:conclusion:exact-complexity}
\begin{wraptable}[4]{r}{11.8cm}
  \vspace{-12pt}
  \begin{tabular}{c||c}
    \textbf{Dimension} & \textbf{Exact Maximum Complexity}\\
    \hline
    $d = 2$ & $m_1 + m_2 + \ldots + m_k$ \\
    \hline
    $d = 3$ &
    $\sum_{1 \leq i < j \leq k}(2m_i - 5)(2m_j - 5) + \sum_{1 \leq i \leq k}m_i + \binom{k}{2}$\\
  \end{tabular}
\end{wraptable}
The table to the right summarizes the known exact bounds on the maximum
complexity of Minkowski sums of polytopes in terms of number of facets
($(d-1)$-faces) derived in Chapter~\ref{chap:mink-sum-complexity}. The
exact bounds are unknown for higher Dimensions as far as we know.

It is known that the exact complexity (counting faces) of the
Minkowski sum of two polytopes with $m$ and $n$ facets can be as low
as $m$, when the two polytopes have the same number of facets $m$ and
parallel features, but it is unknown what is the minimum exact
complexity of Minkowski sums of polytopes that have only a limited
number of parallel features, or none at all.

\begin{appendix}
\renewcommand{\chaptermark}[1]{\markboth{{\it \appendixname\ \thechapter.\ #1}}{}}

\chapter{Software Components, Libraries, and Packages}
\label{app:software}
\newcounter{sw:cntr}

\section{Visual Simulation}
\label{sec:software:visual-simulation}
We have developed a toolkit, called \sgal{} (Scene Graph Algorithm
Library),\footnote{We plan to offer \sgal{} with an open-source
license in the future, making it available to the public.} that
supports the construction and maintenance of directed acyclic graphs
that represent scenes and models in $\rrr$. The toolkit includes,
among the other, two interactive 3D applications. The first detects
collisions and answers proximity queries for polytopes that undergo
translation and rotation. The second enables users to visualize a
scene in an interactive manner. It parses input files that describe
the scene in a degenerate yet extended \vrml{} format~\citelinks{web3d}.
The format is degenerate, as not all \vrml{} features are supported
(yet). However, it has been significantly extended as described below.

Both applications are linked with
\setcounter{sw:cntr}{1}(\roman{sw:cntr}) \cgal,
\addtocounter{sw:cntr}{1}(\roman{sw:cntr}) a library that provides
the exact rational number-type,
\addtocounter{sw:cntr}{1}(\roman{sw:cntr}) several \boost{} libraries,
\addtocounter{sw:cntr}{1}(\roman{sw:cntr})
imagemagick~\citelinks{imagemagic} --- a library that creates, edits,
and composes bitmap images, and
\addtocounter{sw:cntr}{1}(\roman{sw:cntr}) internal libraries that
construct and maintain 3D scene-graphs, written in C++, and built on
top of {\tt OpenGL}. The internal code is divided into two libraries;
{\tt SGAL} --- The main 3D scene-graph library and
{\tt SCGAL} --- Extensions that depend on \cgal.

We added several geometry nodes that represent polytopes using various
sub-representations, such as Gaussian maps, a few other related nodes
that handle coordinates, and many other miscellaneous nodes that
provide services, such as antialiasing and snapshoting. The descriptions
of some of the geometry nodes follows.
\begin{compactdesc}
\item[{\tt ArrangementOnSphere}]
  This node represents models as arrangements of geodesic arcs
  embedded on the sphere.
\item[{\tt ExactPolyhedron}]
  This node represents polyhedra as meshes using the \cgal{}
  \cPolyhedron{} data structure.
\item[{\tt SphericalGaussianMap}]
  This node represents polytopes as spherical Gaussian maps using 
  the \aos{} data structure instantiated as an arrangement embedded on
  the sphere.
\item[{\tt CubicalGaussianMap}]
  This node represents polytopes as cubical Gaussian maps using the
  \cCgm{} data structure.
\end{compactdesc}

The implementation relies on inputing {\em exact} coordinates. To this
end, the format was further extended with a node called
{\tt ExactCoordinate} that represents exact coordinates, and enables
the provision of exact rational coordinates as input.

The entire partitioning process described in
Chapter~\ref{chap:assem_plan} is realized within \sgal{}. The library
contains all the necessary ingredients required to represent and
visualize the input and the output, and to simulate the process. In
particular, it has been extended with two geometry node types: the
{\tt Assembly} node type represents assemblies or subassemblies, and
the {\tt AssemblyPart} node type represents parts of assemblies.
Notice, that each node object of the three types {\tt AssemblyPart},
{\tt SphericalGaussianMap}, and {\tt ArrangementOnSphere} internally
maintains the \cgal{} data structure that represents an arrangement of
geodesic arcs embedded on the sphere~\cite{bfhmw-scmtd-07}, an
instance of the \aos{} class template.

\section{Software Availability}
\label{sec:software:availability}
Compiling and executing the programs require the latest internal
release of \cgal{} (post version~3.31) and the components listed above
including an internal package of \cgal{} that supports the \cCgm{} data
structure. The source code is available upon
request.\footnote{Send email to \url{efifogel@gmail.com}.}
Precompiled executables compiled either with {\tt g++}~4.2.3 on Linux
or with {\tt VC~8} on Windows, data sets, and documentation can be
downloaded from \url{http://www.cs.tau.ac.il/~efif/CD/3d}.

\end{appendix}

\bibliographystyle{alpha}
\bibliography{abrev,thesis}

\bibliographystylelinks{plain}
\label{bib:links}
\bibliographylinks{links}
\printindex
\end{document}